\documentclass[12pt]{article}
\usepackage{amsmath}
\usepackage{amsthm}
\usepackage{amsfonts}
\usepackage{graphicx,psfrag,epsf}
\usepackage{enumerate}
\usepackage{times}
\usepackage{caption}
\usepackage{subcaption}
\usepackage{natbib}
\usepackage{bm,algorithm} %!!check it see whether it can be used
\usepackage{url} % not crucial - just used below for the URL 
\usepackage{comment}
\newcommand{\blind}{0}
\usepackage{setspace}

\usepackage[pagewise]{lineno}
\def\spacingset#1{\renewcommand{\baselinestretch}%
{#1}\small\normalsize} \spacingset{1}

\usepackage{geometry}
\usepackage{tikz}
\usetikzlibrary{arrows,shapes.arrows,shapes.geometric,shapes.multipart,backgrounds,decorations.pathmorphing,positioning,fit,automata}

\usepackage{epsf,pstricks,pst-node}
\usepackage{float}
\usepackage[outdir=./]{epstopdf}
\usepackage{verbatim} % comment
\DeclareMathOperator*{\argmin}{arg\,min}

\newcommand{\bea}{\begin{eqnarray*}}
\newcommand{\eea}{\end{eqnarray*}}
\newcommand{\be}{\begin{eqnarray}}
\newcommand{\ee}{\end{eqnarray}}
\newcommand{\bay}{\begin{array}}
\newcommand{\eay}{\end{array}}
\newcommand{\bi}{\begin{itemize}}
\newcommand{\ei}{\end{itemize}}
\newcommand{\ben}{\begin{enumerate}}
\newcommand{\een}{\end{enumerate}}
\newcommand{\bcen}{\begin{center}}
\newcommand{\ecen}{\end{center}}

\usepackage[symbol]{footmisc}

\newcommand{\T}{\mathrm{\scriptscriptstyle T}}

\DeclareOldFontCommand{\bf}{\normalfont\bfseries}{\mathbf}

\newtheorem{thm}{Theorem}%[section]
\newtheorem{lem}{Lemma}%[section]
\newtheorem{prop}{Proposition}%[section]
\newtheorem{corr}{Corollary}%[section]
\allowdisplaybreaks

\date{} % Today's date or a custom date

\begin{document}
 \if0\blind
 {
  \title{\bf  Mapping the Genetic-Imaging-Clinical   Pathway with Applications to Alzheimer's Disease}

  \maketitle
\begin{center}
  \author{\large Dengdeng Yu \\
  \vspace{10pt}
  Department of Mathematics, University of Texas at Arlington}\\ \vspace{10pt}
    \author{\large Linbo Wang \\
    \vspace{10pt}
  Department of Statistical Sciences, University of Toronto}\\ \vspace{10pt}
    \author{\large Dehan Kong \\
    \vspace{10pt}
  Department of Statistical Sciences, University of Toronto}\\ \vspace{10pt}
  \author{\large Hongtu Zhu \\
  \vspace{10pt}
  Department of Biostatistics, University of North Carolina, Chapel Hill }\\
  \vspace{10pt}
 {\large for the Alzheimer's Disease Neuroimaging Initiative
\footnote[1]{
Data used in preparation of this article were obtained from the Alzheimer's Disease
Neuroimaging Initiative (ADNI) database (adni.loni.usc.edu). As such, the investigators within the ADNI contributed to the design and implementation of ADNI and/or provided data but did not participate in analysis or writing of this report. A complete listing of ADNI investigators can be found at: \url{http://adni.loni.usc.edu/wp-content/uploads/how_to_apply/ADNI_Acknowledgement_List.pdf}.} }
\end{center}
\newpage
} \fi

 \if1\blind
 {
  \title{\bf Mapping the Genetic-Imaging-Clinical Pathway with Applications to Alzheimer's Disease}
  \maketitle
} \fi

\vspace{-50pt}
\spacingset{1.7} % DON'T change the spacing!

\begin{center} \textbf{Abstract}
\end{center}

\begin{quote}
Alzheimer's disease is a progressive form of dementia that results in problems with memory, thinking,  and behavior. It often starts with abnormal aggregation and deposition of $\beta$ amyloid and tau,  followed by neuronal damage such as atrophy of the hippocampi,  leading to Alzheimer’s Disease (AD). The aim of this paper is to map the genetic-imaging-clinical pathway for  AD 
in order to  delineate the genetically-regulated   brain changes that drive disease progression based on the Alzheimer’s Disease Neuroimaging Initiative (ADNI) dataset.
We develop a novel two-step approach to delineate the   association between high-dimensional 2D hippocampal surface exposures and the Alzheimer’s Disease 
Assessment Scale (ADAS)
 cognitive score,
while taking into account the ultra-high dimensional clinical and genetic  covariates at baseline. 
 Analysis results suggest that the radial distance of each pixel of both hippocampi is negatively associated with the severity of behavioral deficits conditional on observed clinical and genetic covariates. These associations are stronger in Cornu Ammonis region 1 (CA1) and subiculum subregions compared to Cornu Ammonis region 2 (CA2) and Cornu Ammonis region 3 (CA3) subregions. 
Supplementary materials for this article, including a standardized description of the materials available for reproducing the work, are available as an online supplement.
\end{quote}

%keyword
\begin{quote}
\end{quote}

\newpage
%\spacingset{1.5} % DON'T change the spacing!
\section{Introduction}
\label{Introduction}
Alzheimer's disease (AD) is an irreversible brain disorder that slowly destroys memory and thinking skills. According to World Alzheimer Reports \citep{gaugler20192019}, there are around $55$ million people worldwide living with Alzheimer's disease and related dementia. The total global cost of Alzheimer's disease and related dementia was estimated to be a trillion US dollars, equivalent to $1.1\%$ of global gross domestic product.  Alzheimer's patients often suffer from behavioral deficits including memory loss and difficulty of thinking, reasoning and decision making. 

In the current model of AD pathogenesis, it is well established that deposition of amyloid plaques is an early event that, in conjunction with tau pathology, causes neuronal damage. Scientists have identified risk genes that may cause the abnormal aggregation and deposition of the amyloid plaques \citep[e.g.][]{morishima2002alzheimer}. The neuronal damage typically starts from the hippocampus and results in the first clinical manifestations of
the disease in the form of episodic memory deficits \citep{weiner2013alzheimer}. Specifically, \citet{jack2010hypothetical} presented a hypothetical model for biomarker dynamics in AD pathogenesis, which has been empirically and collectively supported by many works in the literature. The model begins with the abnormal deposition of $\beta$ amyloid (A$\beta$) fibrils, as evidenced by a corresponding drop in the levels of soluble A$\beta$-42 in cerebrospinal fluid (CSF) \citep{ aizenstein2008frequent}. After that, neuronal damage begins to occur, as evidenced by increased levels of CSF tau protein \citep{hesse2001transient}. Numerous studies have investigated how A$\beta$ and tau impact the hippocampus \citep[e.g.][]{ferreira2011abeta}, known to be fundamentally involved in acquisition, consolidation, and recollection of new episodic memories \citep{frozza2018challenges}. In particular, as neuronal degeneration progresses, brain atrophy, which starts with hippocampal atrophy \citep{fox1996presymptomatic}, becomes detectable by magnetic resonance imaging (MRI).  Studies from recent years also found other important CSF proteins that may be related to hippocampal atrophy. For instance, the low levels of chromogranin A (CgA) and trefoil factor 3 (TFF3), and high level of cystatin C (CysC) are evidently associated with hippocampal atrophy \citep{khan2015subset,paterson2014cerebrospinal}. Indeed, the impact of protein concentration on behavior can also be through atrophy of other brain regions. For example, there exists potential entorhinal tau pathology on episodic memory decline \citep{maass2018entorhinal}. As sufficient brain atrophy accumulates, it results in cognitive symptoms and impairment. This process of AD  pathogenesis is summarized by the flow chart in Figure \ref{fig: flowchart}.
Note that it is still debatable how A$\beta$ and tau interact with each other as mentioned by \citet{jack2013tracking, majdi2020amyloid}; however, it is evident that A$\beta$ may still hit a biomarker detection threshold earlier than tau \citep{jack2013tracking}. In addition, as noted by \citet{hampel2018cholinergic}, it is likely that highly complex interactions exist between A$\beta$ as well as tau, and the cholinergic system. For instance, the association has been found between CSF biomarkers of amyloid and tau pathology in AD \citep{remnestaal2021association}.
It has also been found that other factors, such as dysregulation and dysfunction of the Wnt signaling pathway, may also contribute to A$\beta$ and tau pathologies \citep{ferrari2014wnt}. In addition, the M1 and M3 subtypes of muscarinic receptors increase amyloid precursor protein production via the induction of the phospholipase C/protein kinase C pathway and increase BACE expression in AD brains \citep{nitsch1992release}.
\begin{figure}[htbp]
\centering
\includegraphics[height=2.2in,width=6.5in]{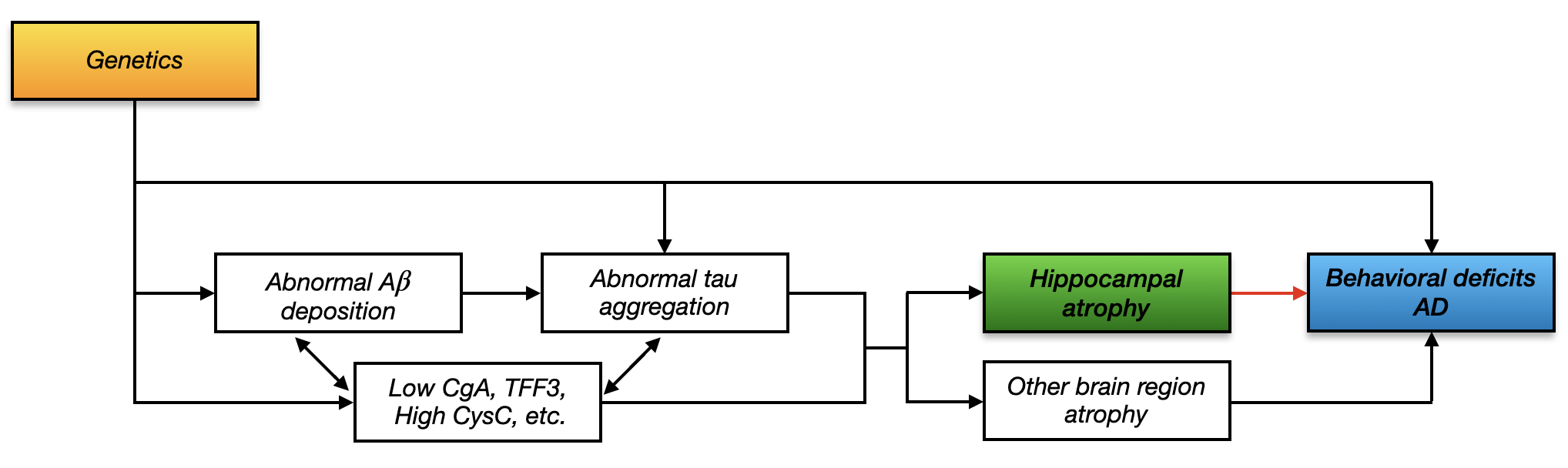}
\caption{A hypothetical model of AD pathogenesis based on \cite{selkoe2016amyloid}. The double arrows represent the possible interactions that exist between A$\beta$ as well as tau, and the cholinergic system. The red arrow denotes the conditional association we are interested in estimating.}
\label{fig: flowchart}
\end{figure}

The aim of this paper is to  
map the genetic-imaging-clinical (GIC) pathway for  AD, which is  
  the most important part of the hypothetical model of AD pathogenesis  in Figure \ref{fig: flowchart}.  Histological studies have shown that the hippocampus is particularly vulnerable to Alzheimer's disease pathology and has already been considerably damaged at the first occurrence of clinical symptoms  \citep{braak1998evolution}. Therefore, the hippocampus has become a major focus in Alzheimer's studies \citep{de1989early}.
 Some neuroscientists even conjecture that the association between hippocampal atrophy and behavioral deficits may be causal, because the former destroys the connections that help the neuron communicate and results in a loss of function \citep{BrainAtrophy}. 
We are interested in  delineating the genetically-regulated   hippocampal shape that drives  AD related behavioral deficits and disease progression.

To map the GIC pathway, we extract clinical, imaging, and genetic variables  from  the Alzheimer's Disease Neuroimaging Initiative (ADNI) study as follows. First, we use the Alzheimer's Disease Assessment Scale (ADAS) cognitive score to quantify behavioral deficits, for which a higher score indicates more severe behavioral deficits. Second, we characterize the exposure of interest, hippocampal shape, by the hippocampal morphometry surface measure, summarized as two $100  \times 150$ matrices corresponding to the left/right hippocampi.  Each element of the matrices is a continuous-valued variable, representing the radial distance from the corresponding coordinate on the hipppocampal surface to the medial core of the hippocampus. Compared with the conventional scalar measure of hippocampus shape \citep{jack2003mri}, recent studies show that  the additional information contained in the hippocampal morphometry surface measure is valuable for Alzheimer's diagnosis \citep{thompson2004mapping}. For example, \citet{li2007hippocampal} showed that the surface measures of the hippocampus could provide more subtle indexes compared with the volume differences in discriminating between patients with Alzheimer's and healthy control subjects. In our case, with the 2D matrix radial distance measure, one may investigate how local shapes of hippocampal subfields are associated with the behavioral deficits. Third, the ADNI study measures ultra-high dimensional genetic covariates and other demongraphic covariates at baseline. 
There are more than $6$ million genetic variants per subject.

The special data structure of the ADNI data application presents  new challenges for statistically mapping the GIC pathway. First, unlike conventional statistical analysis that deals with scalar exposure, our exposure of interest is represented by  high-dimensional 2D hippocampal imaging measures. Second, the dimension of baseline covariates, which are also potential confounders, is much larger than the sample size.  
Recently there have been many developments for confounder selection, most of which are in the causal inference literature. 
Studies show inclusion of the variables only associated with exposure but not directly with the outcome except through the exposure (known as instrumental variables) may result in loss of efficiency  in the causal effect estimate \citep[e.g.][]{schisterman2009overadjustment}, while inclusion of variables only related to the outcome but not the exposure (known as precision variables) may provide efficiency gains \citep[e.g.][]{brookhart2006variable}; see \cite{shortreed2017outcome, richardson2018discussion,tang2021outcome} and references therein for an overview.

When a large number of covariates are available, the primary difficulty  for mapping the GIC pathway is how to include all the confounders and precision variables, while excluding all the instrumental variables and irrelevant variables (not related to either outcome or exposure). We develop a novel two-step approach to estimate the conditional association between the high-dimensional 2D hippocampal surface exposure and the Alzheimer's behaviorial score, while taking into account the ultra-high dimensional baseline covariates. The first step is  a fast screening procedure based on both the outcome and exposure  models to rule out most of the irrelevant variables. The use of both models in screening is crucial for both computational efficiency and selection accuracy, as we will show in detail in Sections \ref{jointscreening} and \ref{Blockwise joint screening}. The second step is  a penalized regression procedure for the outcome generating model to further exclude instrumental and irrelevant variables, and simultaneously estimate the conditional association. Our simulations and ADNI data application demonstrate the effectiveness of the proposed procedure.

Our analysis represents a novel inferential target compared to recent developments in imaging genetics mediation analysis \citep{bi2017genome}. Although we consider a similar set of variables and structure among these variables as illustrated later in Figure \ref{fig: DAG}, our goal is to estimate the conditional association of hippocampal shape with behavioral deficits. In contrast, in mediation analysis,  
researchers are often interested in the effects of genetic factors on behavioral deficits, and how those are mediated through hippocampus. Direct application of methods developed for imaging genetics mediation analysis to our problem may select genetic factors that are confounders affecting both hippocampal shape and behavioral deficits. In comparison, we aim to include precision variables in the adjustment set as they may improve efficiency.

The rest of the article is organized as follows. Section \ref{description} includes a detailed data and problem description. We introduce our models and a two-step variable selection procedure in Section \ref{Setup}. We analyze the ADNI data and estimate the association between hippocampal shape and behavioral deficits conditional on observed clinical and genetic covariates in Section \ref{imgenedatastudy}. Simulations are conducted in Section \ref{simulation} to evaluate the finite-sample performance of the proposed method. We finish with a discussion in Section \ref{sec:discussion}. The theoretical properties of our procedure are included in Section \ref{Theoretical guarantees} in the supplementary material.

\section{Data and problem description}
\label{description}

Understanding how human brains work and how they connect to human behavior is a central goal in medical studies. 
In this paper, we are interested in studying whether and how hippocampal shape is associated with behavioral deficits in Alzheimer's studies. We consider the clinical, genetic, imaging and behavioral measures in the ADNI dataset. The outcome of interest is the Alzheimer's Disease Assessment Scale cognitive score observed at 24th month after baseline measurements. The Alzheimer's Disease Assessment Scale cognitive 13 items score (ADAS-13) \citep{mohs1997development} includes 13 items: word recall task, naming objects and fingers, following commands, constructional praxis, ideational praxis, orientation, word recognition task, remembering test directions, spoken language, comprehension, and word-finding difficulty, delayed word recall and a number cancellation or maze task. A higher ADAS score indicates more severe behavioral deficits.

The exposure of interest is the baseline 2D surface data obtained from the left/right hippocampi. The hippocampus surface data were preprocessed from the raw MRI data, where the detailed preprocessing steps are included in Section \ref{image preprocessing} of the supplementary material.  After preprocessing, we obtained left and right hippocampal shape representations as two $100\times 150$ matrices. The imaging measurement at each pixel  is an absolute metric, representing the radial distance from the pixel to the medial core of the hippocampus. The unit for the measurement is in millimeters.

In the ADNI data, there are millions of observed covariates that one may need to adjust for, including the whole genome sequencing data from all of the 22 autosomes. We have included detailed genetic preprocessing techniques in Section \ref{genetic preprocessing} of the supplementary material. After preprocessing, $6,087,205$ bi-allelic markers (including SNPs and indels) of $756$ subjects were retained in the data analysis.

We excluded those subjects with missing hippocampus shape representations, baseline intracranial volume (ICV) information or ADAS-13 score observed at Month 24, after which there are $566$ subjects left. Our aim is to estimate the association between the hippocampal surface exposure and the ADAS-13 score conditional on clinical measures including age, gender and length of education, ICV, diagnosis status, and $6,087,205$ bi-allelic markers. 

\section{Methodology}
\label{Setup}
\subsection{Basic set-up}
\label{Setup:Notation}
Suppose we observe independent and identically distributed samples $ \{L_i=(X_i, \bm{Z}_i, Y_i), 1\leq i\leq n\} $ generated from $ L $, where $ L=(X, \bm{Z}, Y) $ has support $ \mathcal{L}=(\mathcal{X}\times \mathcal{Z}\times \mathcal{Y})$. Here $ \bm{Z}\in \mathcal{Z}\subseteq \mathbb{R}^{p\times q}$ is a 2D-image continuous exposure, $ Y\in \mathcal{Y} $ is a continuous outcome of interest, and $ X\in \mathcal{X}\subseteq \mathbb{R}^{s} $ denotes a vector of ultra-high dimensional genetic (and clinical) covariates, where we assume $ s>>n$. We are interested in characterizing the association between the 2D exposure $ \bm{Z} $ and outcome $ Y $ conditional on the observed covariates $ X$.

\begin{figure}[htbp]
\centering
\includegraphics[height=3in,width=4in]{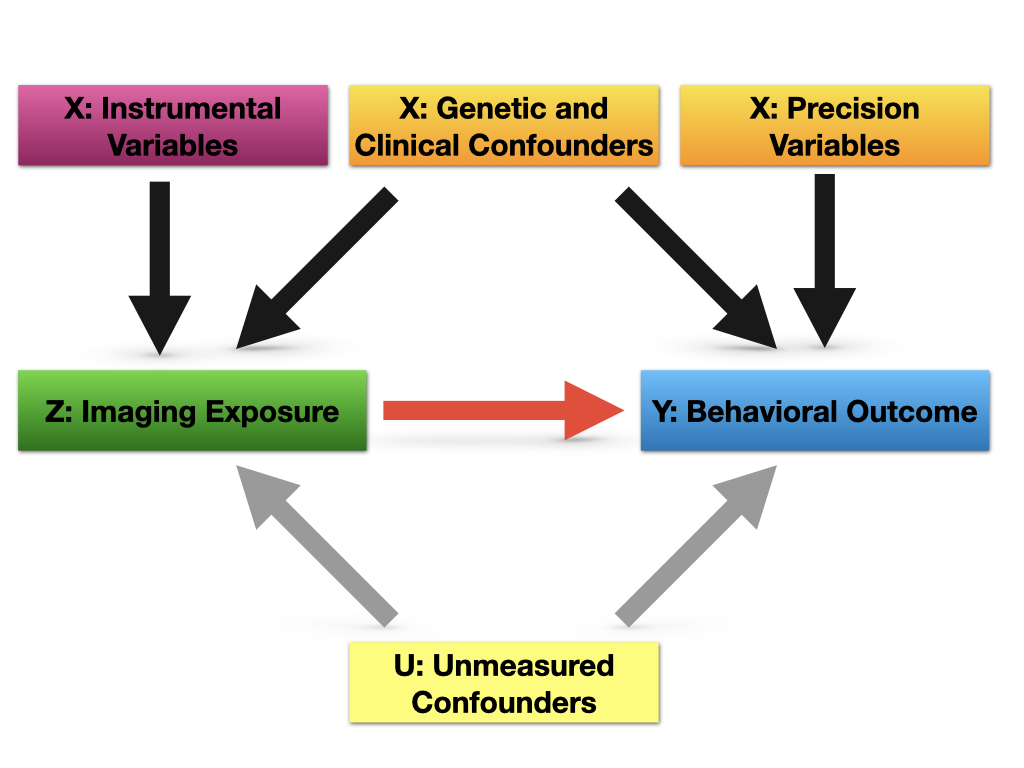}
\caption{Directed acyclic graph showing potential high dimensional confounder and precision variables $X$ (gold),  the possible unmeasured confounders $U$ (light yellow), the 2D imaging exposure $Z$ (green), the instrumental variables $X$ (purple) and the outcome of interest $Y$ (blue). The red arrow denotes the association of interest. }\label{fig: DAG}
\end{figure}

Denote $ X_i=(X_{i1}, \ldots, X_{is})^{\T} $. Without loss of generality, we assume that $X_{il}$ has been standardized for every $1 \leq l \leq s$, and $\bm{Z}_i$ and $Y_i$ have been centered. To map the GIC pathway, we assume the following linear equation models:
\begin{eqnarray}
\label{outcomemodel}
Y_i &=& \sum_{l=1}^s X_{il}\beta_l+ \langle \bm{Z}_i,\bm{B} \rangle + \epsilon_{i}  \quad \textrm{(outcome  model)}; \\
\label{treatmentmodel}
\bm{Z}_i &=& \sum_{l=1}^s X_{il}*\bm{C}_{l} + \bm{E}_{i}  \quad \textrm{(exposure  model)}.
\end{eqnarray}
In \eqref{outcomemodel}, the matrix $ \bm{B}\in \mathbb{R}^{p\times q} $ is the main parameter of interest, representing the association between the 2D imaging treatment $\bm{Z}_i$ and the behavioral outcome $ Y_i $, $\beta_l$ represents the association between the $l$-th observed covariate $X_{il} $ and the behavioral outcome $ Y_i $, and  $ \epsilon_i $ and $\bm E_i$ are  random errors that may be correlated. 
The inner product between two matrices is defined as $ \langle\bm{Z}_i, \bm{B} \rangle =\langle \mathrm{vec}(\bm{Z}_i), \mathrm{vec}(\bm{B})\rangle$, where $ \mathrm{vec}(\cdot) $ is a vectorization operator that stacks the columns of a matrix into a vector. Model \eqref{treatmentmodel},  previously introduced in \citet{kong2019l2rm},  specifies the relationship between the 2D imaging exposure and the observed covariates. 
The $ \bm{C}_{l} $  is a $ p\times q $ coefficient matrix characterizing the association between the $l$th covariate $X_{il} $ and the 2D imaging exposure $ \bm{Z}_i $, and $ \bm{E}_i $ is a $ p\times q $ matrix of random errors with mean $ 0 $. The symbol ``$*$" denotes element-wise multiplication.   Define $\mathcal{M}_1 = \{ 1 \leq l \leq s: \beta_l \neq 0\}$ and $\mathcal{M}_2 = \{  1 \leq l \leq s : \bm{C}_l \neq \bm{0} \}$, where we assume $|\mathcal{M}_{1}| << n$ and $|\mathcal{M}_{2}| << n$; here $|\mathcal{M}_1|$ and  $|\mathcal{M}_2|$ represent the number of elements in $\mathcal{M}_1$ and $\mathcal{M}_2$ respectively. 

 To estimate $ \bm{B}$, the first step is to perform variable selection in models \eqref{outcomemodel} and \eqref{treatmentmodel}. For all the covariates $ X_{l}$, we  group them into four categories. Let $ \mathcal{A} =  \{ 1,\ldots, s\}$, and denote $\mathcal{C}$ the indices of confounders, i.e. variables associated with both the outcome and the exposure; $\mathcal{P}$ denotes the indices of precision variables, i.e. predictors of the outcome, but not the exposure; $\mathcal{I}$  denotes the indices of instrumental variables, i.e. covariates  that are only associated with the exposure but not directly with the outcome except through the exposure;  $\mathcal{S}$ denotes the indices of irrelevant variables, i.e. covariates that are not related to  the outcome or the exposure. Mathematically speaking, $\mathcal{C}=  \{l \in \mathcal{A}| \beta_l \neq 0 \textrm{ and } \bm{C}_l \neq 0\}$, $\mathcal{P}=  \{l \in \mathcal{A}| \beta_l \neq 0 \textrm{ and } \bm{C}_l = 0\}$, $\mathcal{I} =  \{l \in \mathcal{A}| \beta_l = 0 \textrm{ and } \bm{C}_l \neq 0\}$ and $\mathcal{S} =  \{l \in \mathcal{A}| \beta_l = 0 \textrm{ and } \bm{C}_l = 0\}$. The relationships among different types of $ X $, $ \bm{Z}$ and $ Y $ are shown in Figure \ref{fig: DAG}, where $U$ denotes possible unmeasured confounders. Since we are interested in characterizing the association between   $ \bm{Z} $ and   $ Y $ conditional on   $ X$, further discussion of $U$ will be omitted for the remainder of the paper. 

When there are no unobserved confounders $U$, the estimate of $ \bm{B}$ has underlying causal interpretations. In this case, the ideal adjustment set includes all confounders to avoid bias and all precision variables to increase statistical efficiency, while excluding instrumental variables and irrelevant variables \citep{brookhart2006variable,shortreed2017outcome}. Although we are studying the conditional association rather than the causal relationship due to the possible unobserved confounding, our target adjustment set remains the same.  In other words,  we aim to retain all covariates from $\mathcal{M}_1=\mathcal{C} \cup \mathcal{P}=\{l \in \mathcal{A}| \beta_l \neq 0\}$, while excluding covariates from $\mathcal{I} \cup \mathcal{S}=\{l \in \mathcal{A}| \beta_l=0\}$.

\subsection{Naive screening methods}
\label{Naive screening methods}
To find the nonzero $ \beta_l$'s, a straightforward idea is to consider a penalized estimator obtained from the outcome generating model \eqref{outcomemodel}, where one imposes, say Lasso penalties, on $ \beta_l $'s. However, this is computationally infeasible in our ADNI data application as the number of baseline covariates $ s $ is over $ 6 $ million. Consequently, it is important to employ a screening procedure \citep[e.g.][]{fan2008sure} to reduce the model size. To find covariates $X_l$'s that are associated with the outcome $Y$ conditional on the exposure $\bm Z$, one might consider a   conditional screening procedure for model \eqref{outcomemodel}  \citep{barut2016conditional}. Specifically, one can fit the model $ Y_i = X_{il}\beta_l+ \langle \bm{Z}_i, \bm{B} \rangle + \epsilon_{i} $ for each $ 1\leq l\leq s$, obtain marginal estimates of $ \widehat{\beta}^{MZ}_l$'s and then sort the $|\widehat{\beta}^{MZ}_l|$'s for screening. This procedure works well if the exposure variable $ \bm{Z} $ is of low dimension as one only needs to fit low dimensional ordinary least squares (OLS)  $s$ times. However, in our ADNI data application, the imaging exposure $ \bm{Z} $ is of dimension $ pq=15,000 $. As a result,  one cannot obtain an OLS estimate since $ n<pq$. Thus, to apply the conditional sure independence screening procedure to our application, one may need to solve a penalized regression problem for each $ 1\leq l\leq s$, such as $\arg\min_{\bm{B},\beta_l} \left[ \frac{1}{2 n} \sum_{i=1}^n \left( Y_i -\langle \bm{Z}_i, \bm{B} \rangle -  {X}_{il} \beta_l \right)^2 + P_\lambda (\bm{B}) \right]$, 
where $P_\lambda (\bm{B})$ is a penalty of $\bm{B}$. In theory, for each $l\in \{1,\ldots, s\},$ one can  obtain the estimates $\widehat{\beta}^{MZ}_{l,\lambda} $, and then rank the $|\widehat{\beta}^{MZ}_{l,\lambda}| $'s. However,  this is  computationally prohibitive in the ADNI data with $ s> 6,000,000 $. First, the penalized regression problem is much slower to solve compared to the OLS. Second, selection of the tuning parameter $ \lambda $ based on grid search substantially increases the computational burden.

Alternatively one may apply the marginal screening procedure of \cite{fan2008sure} to model \eqref{outcomemodel}. Specifically, one may solve the following marginal OLS on each $ X_{il} $ by ignoring the exposure $ \bm{Z}_i$: $\arg\min_{\beta_l} \left[ \frac{1}{2 n} \sum_{i=1}^n \left( Y_i -{X}_{il} \beta_l \right)^2 \right]$.
The marginal OLS estimate has a closed form $ \widehat{\beta}_l^{M}=n^{-1} \sum_{i=1}^n X_{il} Y_i $, and one can rank $|\widehat{\beta}_l^{M}|$'s for screening. Specifically, the selected sub-model is defined as ${\widehat{\mathcal{M}}_{1}^*} = \{1 \leq l \leq s: |\widehat{\beta}_l^{M} | \geq \gamma_{1,n}\}$,
where $ \gamma_{1,n} $ is a threshold. Computationally, it is much faster than conditional screening for model \eqref{outcomemodel} as we only need to fit one dimensional OLS for $ s> 6,000,000 $ times. However, this procedure is likely to miss some important confounders. To see this, plugging model \eqref{treatmentmodel} into \eqref{outcomemodel} yields
\begin{eqnarray*}
Y_i &=& \sum_{l=1}^s X_{il}(\beta_l+\langle \bm{C}_{l},\bm{B} \rangle)+ \langle \bm{E}_{i},\bm{B} \rangle + \epsilon_{i}.
\end{eqnarray*}
Even in the ideal case when $X_{il} $'s are orthogonal for $ 1\leq l\leq s$, $ \widehat{\beta}_l^{M} $ is not a good estimate of $ \beta_l $ because of the bias term $\langle \bm{C}_{l},\bm{B} \rangle$. Thus, we may miss some nonzero $ \beta_l $'s in the screening step if $ \beta_l $ and $\langle \bm{C}_{l},\bm{B} \rangle$ are of similar magnitudes but opposite signs. We illustrate this point in  Figures \ref{sim1step1n200sigma1} 
in Section \ref{simulation}, in which cases the conventional marginal screening on \eqref{outcomemodel} fails to capture some of the important confounders.  

\subsection{Joint screening}
\label{jointscreening}

To overcome the drawbacks of the estimation methods discussed in Section \ref{Naive screening methods}, we develop a joint screening procedure, specifically for our ADNI data application. The procedure is not only computationally efficient, but can also select all the important confounders and precision variables with high probability. The key insight here is that although we are interested in selecting important variables in the outcome generating model, this can be done much more efficiently by incorporating information from 
the exposure model. Specifically, let $\widehat{\bm{C}}_l^M = n^{-1} \sum_{i=1}^n X_{il} * \bm{Z}_i \in \mathbb{R}^{p \times q}$ be  the marginal OLS estimate in model \eqref{treatmentmodel} for $l = 1,\ldots,s$. Following \citet{kong2019l2rm}, the important covariates in model \eqref{treatmentmodel} can be selected by ${\widehat{\mathcal{M}}_{2}} = \{1 \leq l \leq s: \|\widehat{\bm{C}}_l^M\|_{op} \geq \gamma_{2,n}\}$, 
where $||\cdot||_{op}$ is the operator norm of a matrix and $ \gamma_{2,n} $ is a threshold.

We define our joint screening set as $\widehat{\mathcal{M}} = \widehat{\mathcal{M}}_1^* \cup \widehat{\mathcal{M}}_2$. 
Intuitively, most important confounders and precision variables are contained in the set ${\widehat{\mathcal{M}}_{1}}^*$. The only exceptions are the covariates  $X_l$  for which both $\beta_l$ and $\langle \bm{C}_{l},\bm{B} \rangle$ are of similar magnitudes but  opposite signs. On the other hand, these  $X_l$ will be included in ${\widehat{\mathcal{M}}_{2}}$ and hence,   $ \widehat{\mathcal{M}}_{} $ along with instrumental variables with large $ ||\bm{C}_l||_{op} $.
In Section \ref{Theoretical guarantees} of the supplementary material, we show that with properly chosen $\gamma_{1,n}$ and $\gamma_{2,n}$, the joint screening set includes the confounders and precision variables with high probability: $P(\mathcal{M}_1 \subset \widehat{\mathcal{M}} ) \rightarrow 1$ as $ n\rightarrow \infty$. In practice, we recommend choosing $ \gamma_{1,n}$ and $\gamma_{2,n}$ such that $ |\widehat{\mathcal{M}}_1^*|=|\widehat{\mathcal{M}}_2|=k $, where $ k $ is the smallest integer such that $|\widehat{\mathcal{M}}| \geq \lfloor n/\log(n) \rfloor $. We set them to be of equal sizes following the convention that the size of screening set is determined only by the sample size \citep{fan2008sure}, which is the same for $\widehat{\mathcal{M}}_{1}^{*}$ and $\widehat{\mathcal{M}}_{2}$. 
Depending on the prior knowledge about the sizes and signal strengths of  confounding, precision and instrumental variables, the sizes of $|\widehat{\mathcal{M}}_1^*|$ and $|\widehat{\mathcal{M}}_2|$ may be chosen differently. In the simulations and real data analyses, we conduct sensitivity analyses by varying the relative sizes of $\widehat{\mathcal{M}}_{1}^{*}$ and $\widehat{\mathcal{M}}_{2}$.

In general,  the set $ \widehat{\mathcal{M}}_{}$ includes not only confounders and precision variables in $ \mathcal{M}_1=\mathcal{C}\bigcup \mathcal{P}$, but also instrumental variables in $ \mathcal{I}$ and a small subset of the irrelevant variables $\mathcal{S}$. Nevertheless, the size of $ |\widehat{\mathcal{M}}_{}| $ is greatly reduced compared to that of all the observed covariates. 
This makes it feasible to perform the second step procedure, a  refined penalized estimation of ${\bm B}$ based on the covariates $ \{X_{l}: l\in \widehat{\mathcal{M}}_{}\} $.

\subsection{Blockwise joint screening}
\label{Blockwise joint screening}
Linkage disequilibrium (LD) is a ubiquitous biological phenomenon where genetic variants present a strong blockwise correlation (LD) structure \citep{wall2003haplotype}.  If all the SNPs of a particular LD block are important but with relatively weak signals, they may be missed by the screening procedure described in Section \ref{jointscreening}. To appropriately utilize LD blocks' structural information to select those missed SNPs, we develop a modified screening procedure described below.

Suppose that $X=(X_1,\ldots, X_s)^T$ can be divided into $b$ discrete haplotype blocks: regions of high LD that are separated from other haplotype blocks by many historical recombination events \citep{wall2003haplotype}.
Let the index set of each $b$ non-overlapping block be $\mathcal{B}_1, \ldots, \mathcal{B}_b$ with $\cup_{j=1}^b \mathcal{B}_j = \{ 1\ldots, s\}$. 
For $ l = 1,\ldots, s,$ we define 
$${\beta^{block,M}_{l}} = \sum_{j=1}^b  \frac{1(l \in \mathcal{B}_j)}{|\mathcal{B}_{j}|}\sum_{i \in \mathcal{B}_{j}} |{\beta_i^{M}}| 
\quad \textrm{and} \quad {C^{block,M}_{l}} =  \sum_{j=1}^{b} \frac{1(l \in \mathcal{B}_j)}{|\mathcal{B}_{j}|}\sum_{i \in \mathcal{B}_{j}}  \|{\bm{C}}_i^M\|_{op}, $$
$$\widehat{\beta}^{block,M}_{l} = \sum_{j=1}^b  \frac{1(l \in \mathcal{B}_j)}{|\mathcal{B}_{j}|}\sum_{i \in \mathcal{B}_{j}} |\widehat{\beta}_i^{M}| 
\quad \textrm{and} \quad \widehat{C}^{block,M}_{l} =  \sum_{j=1}^{b} \frac{1(l \in \mathcal{B}_j)}{|\mathcal{B}_{j}|}\sum_{i \in \mathcal{B}_{j}}  \|\widehat{\bm{C}}_i^M\|_{op}, $$
where $1(\cdot)$ is the indicator function of an event.  We also define 
\begin{equation*}
{\widehat{\mathcal{M}}_{1}^{block,*}} = \{1 \leq l \leq s: \widehat{\beta}^{block,M}_{l} \geq \gamma_{3,n}\} \quad \textrm{and} \quad {\widehat{\mathcal{M}}_{2}^{block}} = \{1 \leq l \leq s: \widehat{C}_l^{block,M} \geq \gamma_{4,n}\}. 
\end{equation*}

We propose to use the new set $\widehat{\mathcal{M}}^{block} = \widehat{\mathcal{M}}_1^* \cup \widehat{\mathcal{M}}_2  \cup {\widehat{\mathcal{M}}_{1}^{block,*}} \cup \widehat{\mathcal{M}}_{2}^{block}$, rather than $\widehat{\mathcal{M}} = \widehat{\mathcal{M}}_1^* \cup \widehat{\mathcal{M}}_2$, to select   important covariates. 
Intuitively, when $|\beta_{l_1}| > |\beta_{l_2}|$, $X_{l_1}$ is much more easily selected compared with $X_{l_2}$ by  $\widehat{\mathcal{M}}_1^*$. However, suppose that $l_1 \in \mathcal{B}_1$ and $l_2 \in \mathcal{B}_2$, with only a small proportion of $X_l$ in $\mathcal{B}_1$ having $|\beta_l|>0$, whereas a large proportion of $X_l$ in $\mathcal{B}_2$ has $|\beta_l|>0$. It may well be the case that ${\beta^{block,M}_{l_1}}< {\beta^{block,M}_{l_2}} $,  meaning that  $X_{l_2}$ can be selected more easily than $X_{l_1}$ by $\widehat{\mathcal{M}}_1^{block,*}$. In addition, as $\widehat{\beta}_l^{M}$ is not a good estimate of $\beta_l$ due to  the bias term $\langle \bm{C}_{l},\bm{B} \rangle$, $\widehat{\beta}^{block,M}_{l}$ is not a good estimate of ${\beta^{block,M}_{l}}$ either. Therefore, some $X_l$ with nonzero ${\beta^{block,M}_{l}}$ may not be included in ${\widehat{\mathcal{M}}_{1}^{block,*}}$. Nevertheless, they will be included in ${\widehat{\mathcal{M}}_{2}^{block}}$ and hence $\widehat{\mathcal{M}}^{block}$.

Theoretically, when $\gamma_{1,n}$, $\gamma_{2,n}$, $\gamma_{3,n}$ and $\gamma_{4,n}$ are chosen properly, $P(\mathcal{M}_1 \subset \widehat{\mathcal{M}}^{block} ) \rightarrow 1$ as $ n\rightarrow \infty$; see Theorem \ref{thm1a} in Section \ref{Theoretical guarantees} of the supplementary material. In practice, we recommend choosing  $\gamma_{1,n}$, $\gamma_{2,n}$, $\gamma_{3,n}$ and $\gamma_{4,n}$, such that $ |\widehat{\mathcal{M}}_1^*|=|\widehat{\mathcal{M}}_2|= |{\widehat{\mathcal{M}}_{1}^{block,*}}| = |{\widehat{\mathcal{M}}_{2}^{block}}|= k $, where $ k $ is the smallest integer such that $|\widehat{\mathcal{M}}^{block}| \geq 2\lfloor n/\log(n) \rfloor $. The threshold here is twice the threshold what we suggested in Section \ref{jointscreening} since we unionize two additional sets.

\subsection{Second-step estimation}

In this step, we aim to estimate $\bm{B}$ by excluding the instrument variables $ \mathcal{I}$ and irrelevant variables in $\mathcal{S}$ from $ \widehat{\mathcal{M}}$ (or $\widehat{\mathcal{M}}^{block}_{}$) and keeping the other covariates. This can be done by solving the following optimization problem:
\begin{equation}
\label{min1}
\argmin_{\bm{B},\{\beta_l,l \in \widehat{\mathcal{M}}_{}\}} \left[ \frac{1}{2 n} \sum_{i=1}^n \left( Y_i -\langle \bm{Z}_i, \bm{B} \rangle -  \sum_{l \in \widehat{\mathcal{M}}_{}} {X}_{il} \beta_l \right)^2 + \lambda_{1}  \sum_{l \in \widehat{\mathcal{M}}_{}}  |\beta_l|  + \lambda_{2} || \bm{B}||_* \right].
\end{equation}
Denote $ (\widehat{\bm{B}}, \widehat{\boldsymbol\beta})$ the solution to the above optimization problem. 
The Lasso penalty on $\beta_l$ is used to exclude instrumental and irrelevant variables in $\widehat{\mathcal{M}}_{}$, whose corresponding coefficients $ \beta_l$'s are zero.  The nuclear norm penalty $|| \cdot ||_*$, defined as the sum of all the singular values of a matrix, is used to achieve a low-rank estimate of $\bm{B}$, where the low-rank assumption in estimating 2D structural coefficients is commonly used in the literature  \citep{zhou2014regularized, kong2019l2rm}. For  tuning parameters, we use five-fold cross validation based on two-dimensional grid search, and select $\lambda_{1}$ and $\lambda_{2}$ using the one standard error rule \citep{hastie2009}.  

\section{ADNI data applications}
\label{imgenedatastudy}
We use the data obtained from the ADNI study (adni.loni.usc.edu). The  data usage acknowledgement is included in Section \ref{usage} of the supplement material. As described in Section \ref{description}, we include $566$ subjects from the ADNI1 study. The exposure of interest is the baseline 2D hippocampal surface radial distance measures, which can be represented as a $ 100\times 150 $ matrix for each part of the hippocampus. The outcome of interest is the ADAS-13 score observed at Month 24. The average ADAS-13 score is  $20.8$ with standard deviation $14.1$.
The covariates to adjust for include $6,087,205$ bi-allelic markers as well as clinical covariates, including age, gender and education length, baseline intracranial volume (ICV),  and baseline diagnosis status. The average age is $75.5$ years old with standard deviation $6.6$ years, and the average education length is $15.6$ years with standard deviation $2.9$ years. Among all the $566$ subjects, $58.1\%$ were female. The average ICV was $1.28\times 10^6$ $\rm{mm}^3$ with standard deviation $1.35\times 10^5$ $\rm{mm}^3$. There were $175$ ($184$ at  Month 24) cognitive normal patients, $268$ ($157$ at Month 24) patients with mild cognitive impairment (MCI), and $123$ ($225$ at  Month 24) patients with AD at the baseline. Studies have shown that age and gender are the main risk factors for Alzheimer's disease \citep{vina2010women, guerreiro2015age} with older people and females more likely to develop Alzheimer's disease. Multiple studies have also shown that prevalence of dementia is greater among those with low or no education \citep{zhang1990prevalence}. On the other hand, age, gender and length of education have been found to be strongly associated with the hippocampus \citep{van2006hippocampal, jack2000rates, noble2012hippocampal}. Previous studies \citep{sargolzaei2015practical} suggest that the ICV is an important measure that needs to be adjusted for in studies of brain change and AD. In addition, the baseline diagnosis status may help explain the baseline hippocampal shape and the AD status at Month 24. Therefore, we consider age, gender, education length, baseline ICV and baseline diagnosis status as part of the confounders, and adjust for them in our analysis. In addition, we also adjust for population stratification, for which we use the top five principal components of the whole genome data.  As both left and right hippocampi have 2D radial distance measures and the two parts of hippocampi have been found to be asymmetric \citep{pedraza2004asymmetry}, we apply our method to the left and right hippocampi separately.

We use the default method \citep{gabriel2002structure} of Haploview \citep{barrett2005haploview} and PLINK \citep{purcell2007plink} to form linkage disequilibrium (LD) blocks.
Previous studies report that about  50 genetic variants are associated with AD; see the review in  \cite{sims2020multiplex} for details.  This provides support for our assumption that $|\mathcal{M}_{1}| < n$ ($n=566$). On the other hand,   a genome-wide association analysis of 19,629 individuals by \cite{zhao2019genome} shows that $57$ genetic variants are associated with the left hippocampal volumes and $54$ are associated with the right hippocampal volumes. This provides support for our assumption that $|\mathcal{M}_{2}| < n$. Therefore, we apply our blockwise joint screening procedure on those SNPs on each part of hippocampal outcome $ Y_i $ and exposure $ \bm{Z}_i $ marginally. 
We choose the thresholds $\gamma_{1,n}$, $\gamma_{2,n}$,$ \gamma_{3,n}$ and $\gamma_{4,n}$ such that  $|\widehat{\mathcal{M}}^{block}|=2 \lfloor n/\log(n) \rfloor=178$. In Table \ref{imgenet1} of the supplementary material, we list the top $20$ SNPs corresponding to left and right hippocampi, respectively. As suggested by one referee, we plot similar figures as the Manhattan plot for $\widehat{\mathcal{M}}^{*}_1$, $\widehat{\mathcal{M}}^{block,*}_1$, $\widehat{\mathcal{M}}^{\textbf{}}_2$ and $\widehat{\mathcal{M}}^{block}_2$ in Figure \ref{manhattanPlots} of the supplementary material, where genomic coordinates are displayed along the x-axis, the y-axis represents the magnitude of $| \widehat{\beta}^{M}_{l} |$, $ \widehat{\beta}^{block,M}_{l} $,  $ \| \widehat{\bm{C}}_l^M \|_{op}$ and $\widehat{C}_l^{block,M}$ and the horizontal dashed line represents the threshold values $\gamma_{1,n}$, $\gamma_{2,n}$, $\gamma_{3,n}$, and $\gamma_{4,n}$, respectively.

From Table \ref{imgenet1} and Figure \ref{manhattanPlots}, one can see that there are quite a few important SNPs for both hippocampi. For example, the top SNP is the rs429358 from the 19th chromosome. This SNP is a C/T single-nucleotide variant (snv) variation in the APOE gene. It is also one of the two SNPs that define the well-known APOE alleles,  the major genetic risk factor for Alzheimer's disease \citep{kim2009role}. 
In addition, a great portion of the SNPs in Table \ref{imgenet1} has  been found to be strongly associated with Alzheimer's. These include  rs10414043 \citep{du2018fast}, an A/G snv variation in the APOC1 gene; rs7256200 \citep{takei2009genetic}, an A/G snv variation in the APOC1 gene; rs73052335 \citep{zhou2018identification}, an A/C snv variation in the APOC1 gene; rs769449 \citep{chung2014exome}, an A/G snv variation in the APOE gene; rs157594 \citep{hao2017mining}, a G/T snv variation; rs56131196 \citep{gao2016shared}, an A/G snv variation in the APOC1 gene; rs111789331 \citep{gao2016shared}, an A/T snv variation; and rs4420638 \citep{coon2007high}, an A/G snv variation in the APOC1 gene. 

Among those SNPs that have been found to be associated with Alzheimer's, some of them are also directly associated with hippocampi. For example, \citet{zhou2019analysis} revealed that the SNPs rs10414043, rs73052335 and rs769449 are among the top SNPs that have significant genetic effects on the volumes of both left and right hippocampi. \cite{guo2019genome} identified the SNP rs56131196 to be associated with hippocampal shape.

We then perform our second-step estimation procedure for each part of the hippocampi. Here $ X_{\widehat{\mathcal{M}}}$ denotes  the SNPs selected in the screening step, the population stratification (top five principal components of the whole genome data) and the 
five clinical measures (age, gender, education length, baseline ICV and baseline diagnosis status), 
and $ {\bm Z} $ denotes the left/right hippocampal surface image matrix. 
To visualize the results, we map the estimates $\widehat{\bm B}$ corresponding to each part of the hippocampus onto a representative hippocampal surface and plot it in Figure \ref{dataPlots}(a).   We have also plotted the hippocampal subfield \citep{apostolova20063d} in Figure \ref{dataPlots}(b). Here Cornu Ammonis region 1 (CA1), Cornu Ammonis region 2 (CA2) and Cornu Ammonis region 3 (CA3) are a strip of pyramidal neurons within the hippocampus proper. CA1 is the top portion, named as ``regio superior of Cajal'' \citep{blackstad1970distribution}, which consists of small pyramidal neurons. Within the lower portion (regio inferior of Cajal), which consists of larger pyramidal neurons, there are a smaller area called CA2 and a larger area called CA3. The cytoarchitecture and connectivity of CA2 and CA3 are different. The subiculum is a pivotal structure of the hippocampal formation, positioned between the entorhinal cortex and the CA1 subfield of the hippocampus proper (for a complete review, see \citealt{dudek2016rediscovering}).

From the plots, we can see that   all the $15,000$ entries of $\widehat{\bm{B}}$ corresponding to both hippocampi are negative. This implies that the radial distance of each pixel of both hippocampi is negatively associated with the ADAS-13 score, which depicts the severity of behavioral deficits. The subfields with strongest associations are mostly CA1 and subiculum. 
Existing literature \citep{apostolova2010subregional} has found that as Alzheimer's disease progresses, it first affects CA1 and subiculum subregions and later CA2 and CA3 subregions. This can partially explain why the shapes of CA1 and subiculum may have stronger associations with ADAS-13 scores compared to CA2 and CA3 subregions.

\begin{figure}[htbp]
\captionsetup[subfigure]{justification=centering}
\centering
 \subcaptionbox{}[0.45\linewidth]
{\includegraphics[height=3in,width=4in]{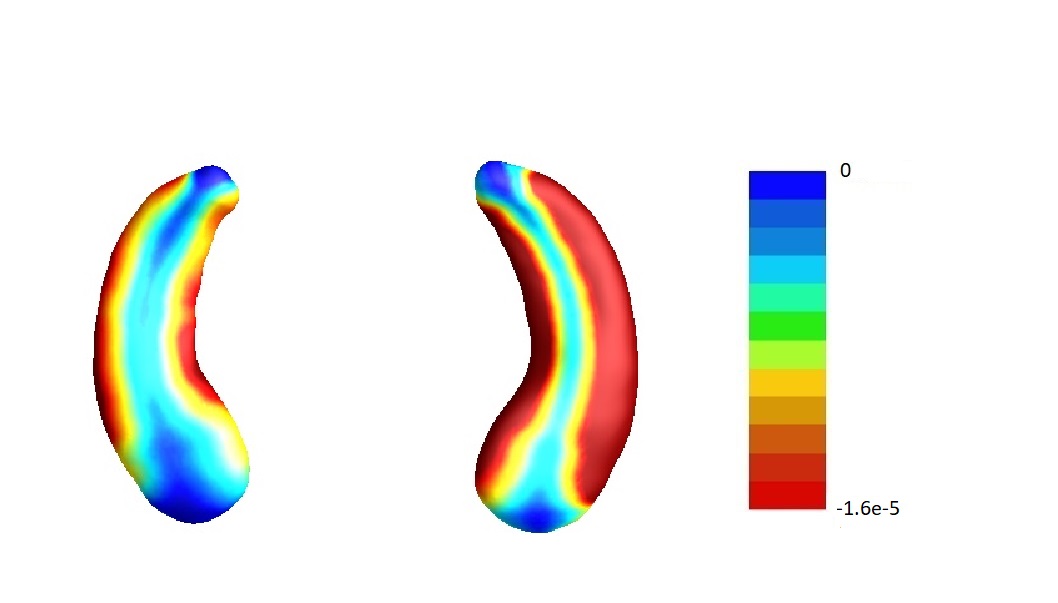}}
 \hfill
 \subcaptionbox{}[0.45\linewidth]
 {\includegraphics[height=2.4in,width=2in]{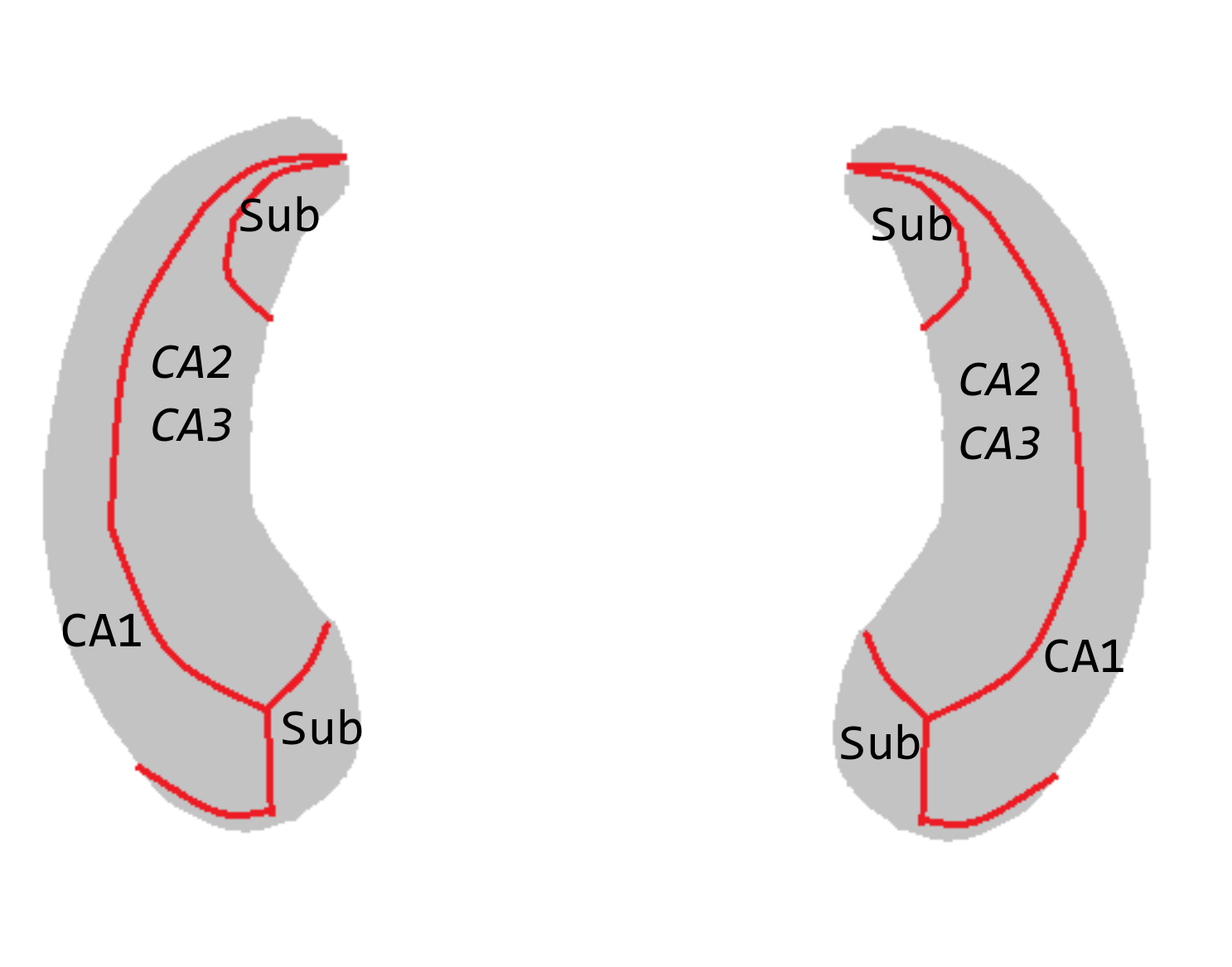}}
 \hfill
\caption{Real data results: Panel (a) plots the estimate $\widehat{\bm B}$ corresponding to the left hippocampi (left part) and  the right hippocampi (right part). Panel (b) plots the hippocampal subfield.}
\label{dataPlots}
\end{figure}
We examine the effect size of the whole hippocampal shape by evaluating the term $ \langle \bm{Z}_{i}, \widehat{\bm{B}} \rangle $. 
Specifically, we calculate the proportion of variance explained by  imaging covariates as follows: 
\begin{eqnarray*}
R^2 = \frac{\sum_{i=1}^{n} (y_{i}-\bar{y})^{2} - \sum_{i=1}^{n} (y_{i}-\bar{y} -\langle \bm{Z}_{i}, \widehat{\bm{B}} \rangle )^2}{\sum_{i=1}^{n} (y_{i}-\bar{y})^{2}}.
\end{eqnarray*}
Our results show that the shape of the left hippocampi and that of the right one account for 5.83\% and 4.71\% of the total variations in behavior deficits, respectively. Such effect sizes are quite large compared with those for  polygenic risk scores in genetics. In addition, we perform permutation test to test whether the $R^2$ statistic is significant. In particular, we randomly permutate the $\{Y_1,\ldots,Y_n\}$, denoted by $ \{Y_i^*,\ldots,Y_n^*\}$, and we then apply our estimation procedure based on $(X_i, \bm{Z}_i, Y_i^*)$, obtain $ \widehat{\bm{B}}^* $, and calculate $(R^2)^*$. We repeat this for 1,000 times and  and we obtain the $\{(R_{(k)}^2)^*, 1\leq k\leq 1000 \}$, which mimics the null
distribution. Finally, the p-value can be calculated as $\frac{1}{1000}\sum_{k=1}^{1000} 1\{(R_{(k)}^2)^*\geq R^2 \}$. The p-values for both hippocampi are less than $0.001$, suggesting that the $R^2$'s  are significantly different from zero.

We also conduct sensitivity analysis by varying the relative sizes of $\widehat{\mathcal{M}}_{1}^{*}$ and $\widehat{\mathcal{M}}_{2}$ in the joint screening procedure. The estimates $\widehat{\bm B}$s are similar among different choices of $|\widehat{\mathcal{M}}_{1}^{*}|$ and $|\widehat{\mathcal{M}}_{2}|$; see supplementary material Section \ref{Sensitivity analysis} for details. In addition, we repeated our analysis on
the $391$ MCI and AD subjects. We have similar findings that the radial distances of each pixel of both hippocampi are mostly negatively associated with the ADAS-13 score. And the subfields with strongest associations are mostly CA1 and subiculum; see supplementary material Section  \ref{Subgroup analysis ADNI data applications} for details. As suggested by one referee, we have performed the SNP-imaging-outcome mediation analysis proposed by \citet{bi2017genome}; see Section \ref{Results for mediation analyses} in the supplementary material for the detailed procedure. There is no evidence for the mediating relationship of  SNP-imaging-outcome from our analysis.

\section{Simulation studies}
\label{simulation}
In this section, we perform simulation studies to evaluate the finite sample performance of the proposed method. The dimension of covariates is set as $s=5000$, and the exposure is a  $64\times 64$ matrix. The $X_i \in \mathbb{R}^s$ is independently generated from $N(\bm{0},\bm\Sigma_x)$, where $\bm\Sigma_x = (\sigma_{x,ll^\prime})$ has an autoregressive structure such that $\sigma_{x,ll^\prime} = \rho_1^{|l - l^\prime|}$ holds for $1 \leq l,l^\prime \leq s$ with $\rho_1 = 0.5$. Define  ${\bm B}$ as a $64 \times 64$ image shown in Figure \ref{figureTCross}(a), and ${\bm C}$ a $64 \times 64$ image shown in Figure \ref{figureTCross}(b). For ${\bm B}$, the black regions of interest (ROIs) are assigned value $0.0408$ and white ROIs are assigned value $0$. For ${\bm C}$, the black ROIs are assigned value $0.0335$ and white ROIs are assigned value $0$. Further 
we set ${\bm C}_l=v_l*\bm{C}$, where $ v_1=-1/3$, $v_2=-1$, $v_3=-3$, $ v_{207}=-3$, $v_{208}=-1$, $v_{209}=-1/3$, and $ v_l=0 $ for $ 4\leq l\leq 206$ and $ 210\leq l \leq s $. 
We set the parameters $\beta_1=3$, $\beta_2=1$, $\beta_3=1/3$, $ \beta_{104}=3$, $\beta_{105}=1$, $\beta_{106}=1/3$, and $ \beta_l=0$ for $4 \leq l \leq 103$ and $ 107\leq l\leq s$. In this setting, we have $\mathcal{C} = \{1, 2, 3\}$, $\mathcal{P} = \{ 104, 105, 106\}$, $\mathcal{I} =\{ 207, 208, 209\}$ and $\mathcal{S} = \{1 \ldots, 5000\} \backslash \{1,2,3,104,105,106,207,208,209\}$.

\begin{figure}[htbp]
\captionsetup[subfigure]{justification=centering}
\centering
 \subcaptionbox{}[0.35\linewidth]
 {\frame{\includegraphics[scale = 2]{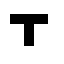}}}
 \subcaptionbox{}[0.35\linewidth]
 {\frame{\includegraphics[scale = 2]{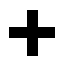}}}
 \hfill
\caption{Panels (a) and (b) plot ${\bm B}$ and ${\bm C}$ respectively. In Panels (a), the value at each pixel is either 0 (white) or 0.0408 (black). In Panels (b), the value at each pixel is either 0 (white) or 0.0335 (black).} 
\label{figureTCross}
\end{figure}

The random error $\mathrm{vec}(\bm{E}_i)$ is independently generated from $N(\bm{0},\bm\Sigma_e)$, where we set the standard deviations of all elements in  $\bm{E}_i$ to be $\sigma_e = 0.2$ and the correlation between $\bm{E}_{i,jk}$ and $\bm{E}_{i,j^\prime k^\prime}$ to be $\rho_2^{|j - j^\prime| + |k-k^\prime|}$ for $1 \leq j,k,j^\prime, k^\prime \leq 64$ with $\rho_2 =0.5$. The random error $\epsilon_{i}$ is generated independently from $N(0,\sigma^2)$, where we consider $ \sigma=1$ or $0.5$. The $Y_i$'s and $ {\bm Z}_i $'s are generated from models \eqref{outcomemodel} and \eqref{treatmentmodel}. We consider three different sample sizes $ n=200, 500$ and $ 1000$.

\subsection{Simulation for screening}
\label{Simulation for screening}
We perform our screening procedure (denoted by ``joint") and report the coverage proportion of $ \mathcal{M}_1 $, which is defined as $\frac{|\widehat{\mathcal{M}} \cap \mathcal{M}_1|}{ |\mathcal{M}_1|}$, where the size of the selected set $ |\widehat{\mathcal{M}}|$ changes from $ 1 $ to $ 100$. In addition, we report the coverage proportion for each of  the confounding and precision variables, i.e. each of the  $j$'s in the set ${\mathcal{M}}_1=\{ 1,2,3,104,105,106\}$. All the coverage proportions are averaged over $ 100 $ Monte Carlo runs. 

To control the changing size of the $|\widehat{\mathcal{M}}|$, we first set $|\widehat{\mathcal{M}}_1^*|= |\widehat{\mathcal{M}}_2|=1 $ by specifying appropriate $ \widehat{\gamma}_{1,n}$ and $ \widehat{\gamma}_{2,n}$. Then we sequentially add two variables, one to $\widehat{\mathcal{M}}_1^*$ by increasing $ \widehat{\gamma}_{1,n}$ and one to $ \widehat{\mathcal{M}}_2$ by increasing $ \widehat{\gamma}_{2,n}$, until $|\widehat{\mathcal{M}}|$ reaches $ 100 $. Note that we always keep $|\widehat{\mathcal{M}}_1^*|= |\widehat{\mathcal{M}}_2|$ in the procedure. We may not obtain all the sizes between $ 1 $ and $ 100 $ because $|\widehat{\mathcal{M}}|$ may increase by at most $ 2 $. Therefore, for those sizes that cannot be reached, we use a linear interpolation to estimate the coverage proportion of $ \mathcal{M}_1 $ by using the closest two end points. 

We  compare the proposed joint screening procedure to two competing procedures. The first is an outcome screening procedure that selects set $\widehat{\mathcal{M}}_1^*.$ For fair comparison, we let $ |\widehat{\mathcal{M}}_1^*| $ range from $ 1 $ to $ 100 $. The second is an intersection screening procedure, that selects set $\widehat{\mathcal{M}}_\cap = \widehat{\mathcal{M}}_1^* \cap \widehat{\mathcal{M}}_2$. We let $ |\widehat{\mathcal{M}}_\cap| $ range from $ 1 $ to $ 100 $, while keeping $|\widehat{\mathcal{M}}_1^*|= |\widehat{\mathcal{M}}_2|$. Similarly, for those specific sizes that $|\widehat{\mathcal{M}}^*| $ cannot reach, we use linear interpolation to estimate the coverage proportions. We plot the results for  $(n, s, \sigma) = (200, 5000, 1)$ in Figure \ref{sim1step1n200sigma1}. The remaining results for $(n,s,\sigma) = (200, 5000, 0.5)$, $ (500,5000,1)$, $(500,5000,0.5)$,  $(1000,5000,1)$ and $(1000,5000,0.5)$ can be found in Figures \ref{sim1step1n200sigma025} -- \ref{sim1step1n1000sigma025} of the supplementary material.

\begin{figure}[htbp]
\captionsetup[subfigure]{justification=centering}
\centering
 \subcaptionbox{Confounder: strong \\ outcome, weak exposure}[0.45\linewidth]
 {\includegraphics[width=6cm,height=3.5cm]{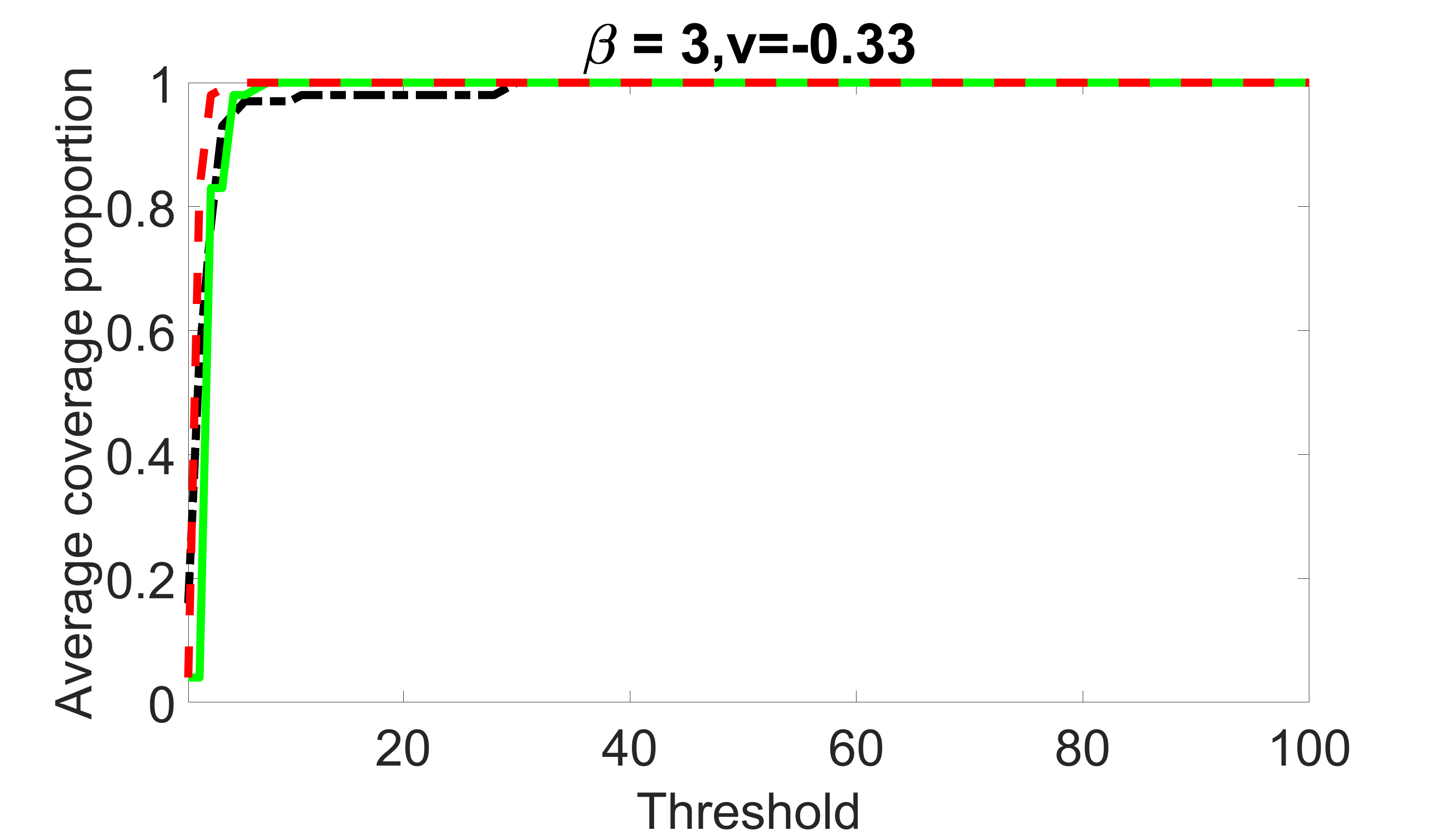}}
 \subcaptionbox{Confounder: medium \\ outcome, medium exposure}[0.45\linewidth]
 {\includegraphics[width=6cm,height=3.5cm]{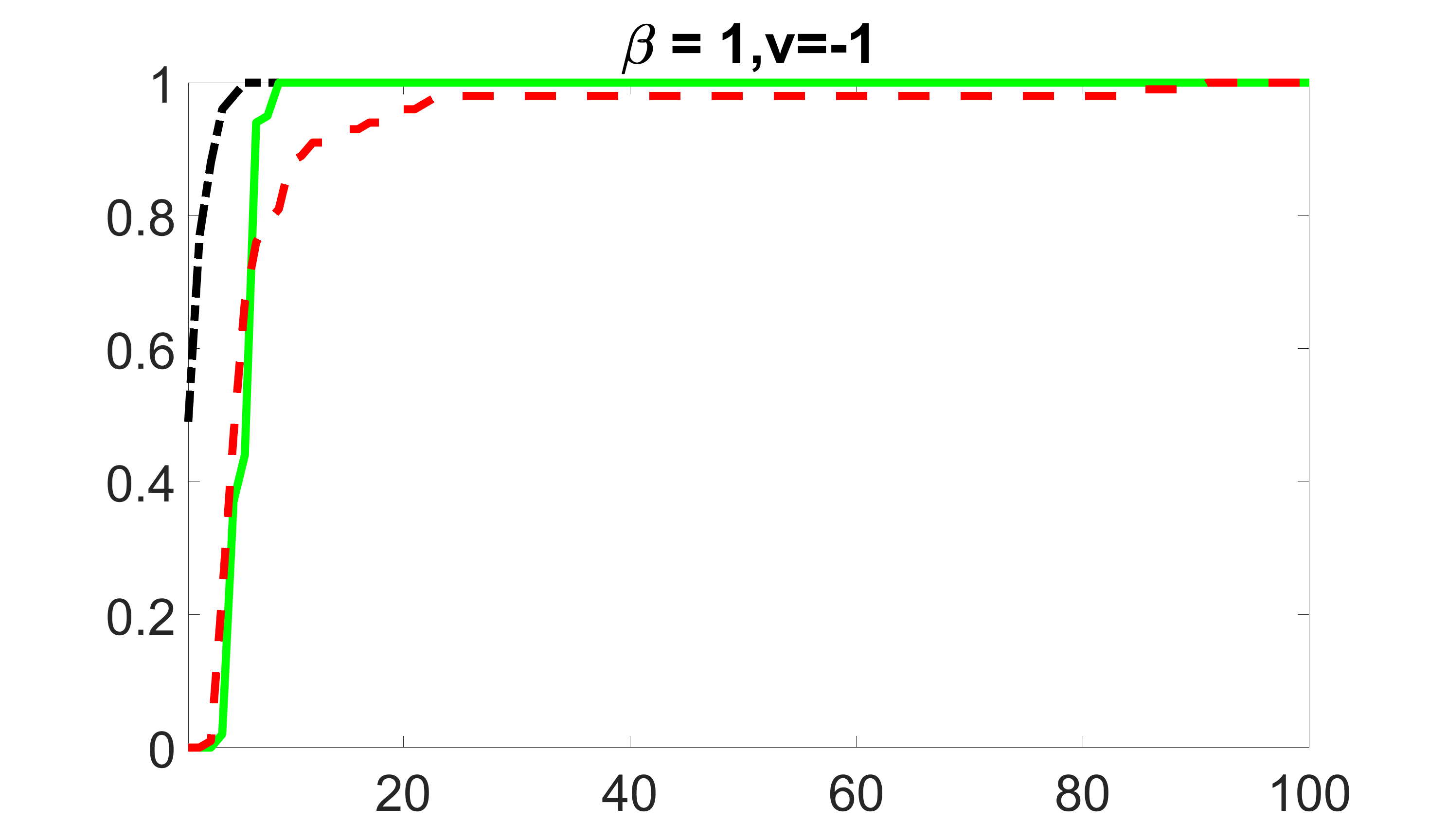}}
  \subcaptionbox{Confounder: weak \\ outcome, strong exposure}[0.45\linewidth]
 {\includegraphics[width=6cm,height=3.5cm]{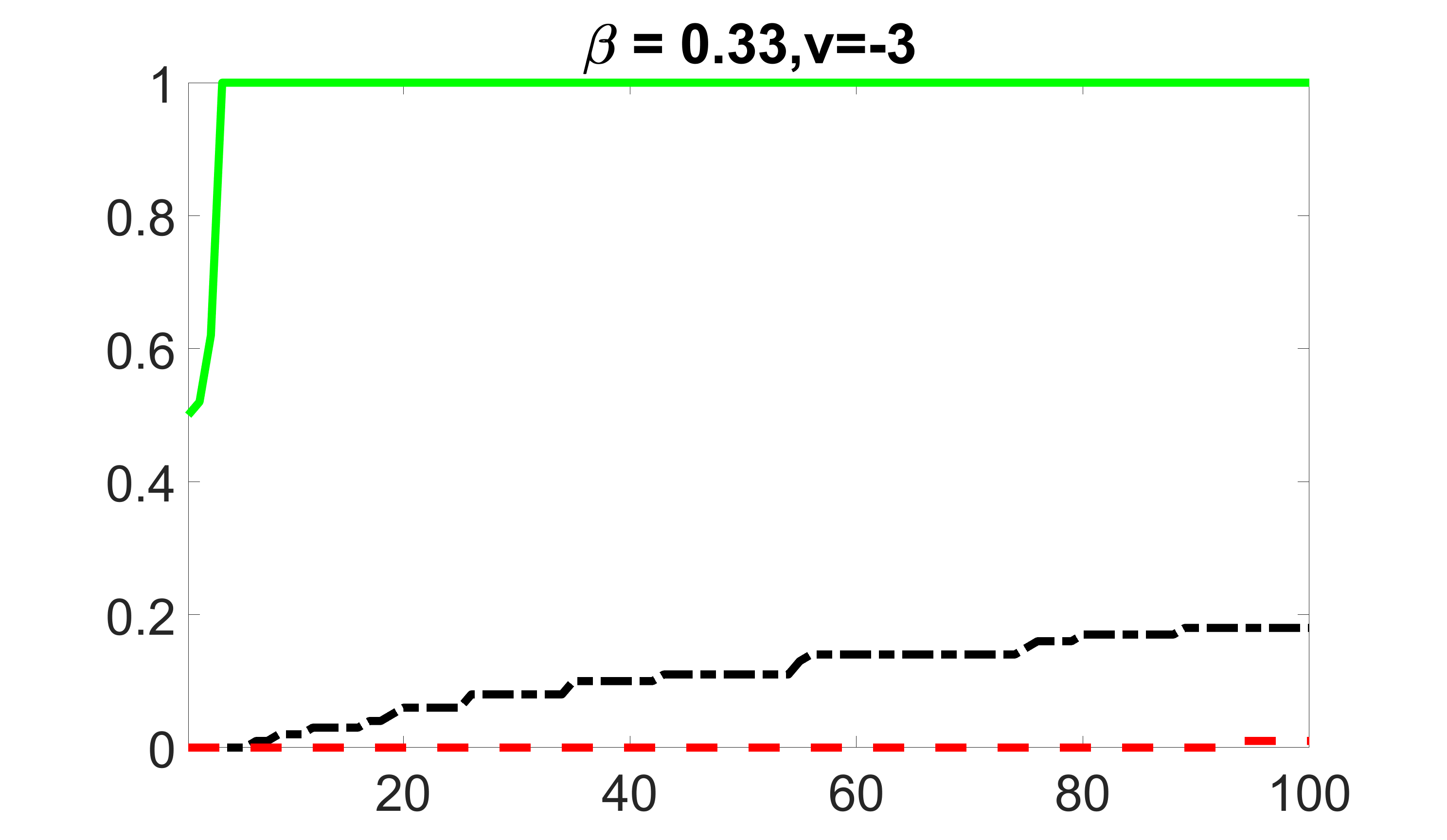}}
  \subcaptionbox{Precision: strong \\ outcome, zero exposure}[0.45\linewidth]
 {\includegraphics[width=6cm,height=3.5cm]{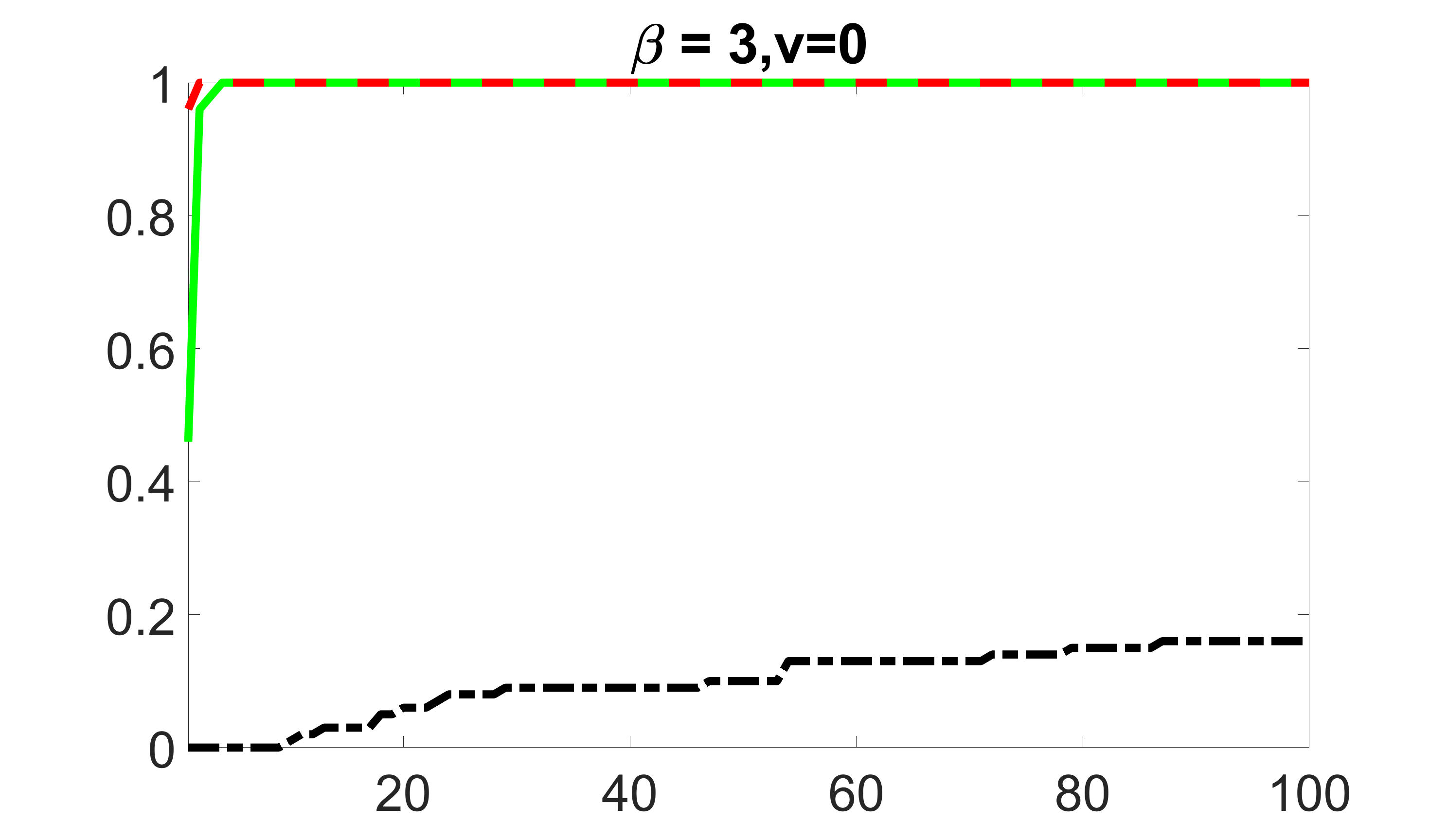}}
  \subcaptionbox{Precision: medium \\ outcome, zero exposure}[0.45\linewidth]
 {\includegraphics[width=6cm,height=3.5cm]{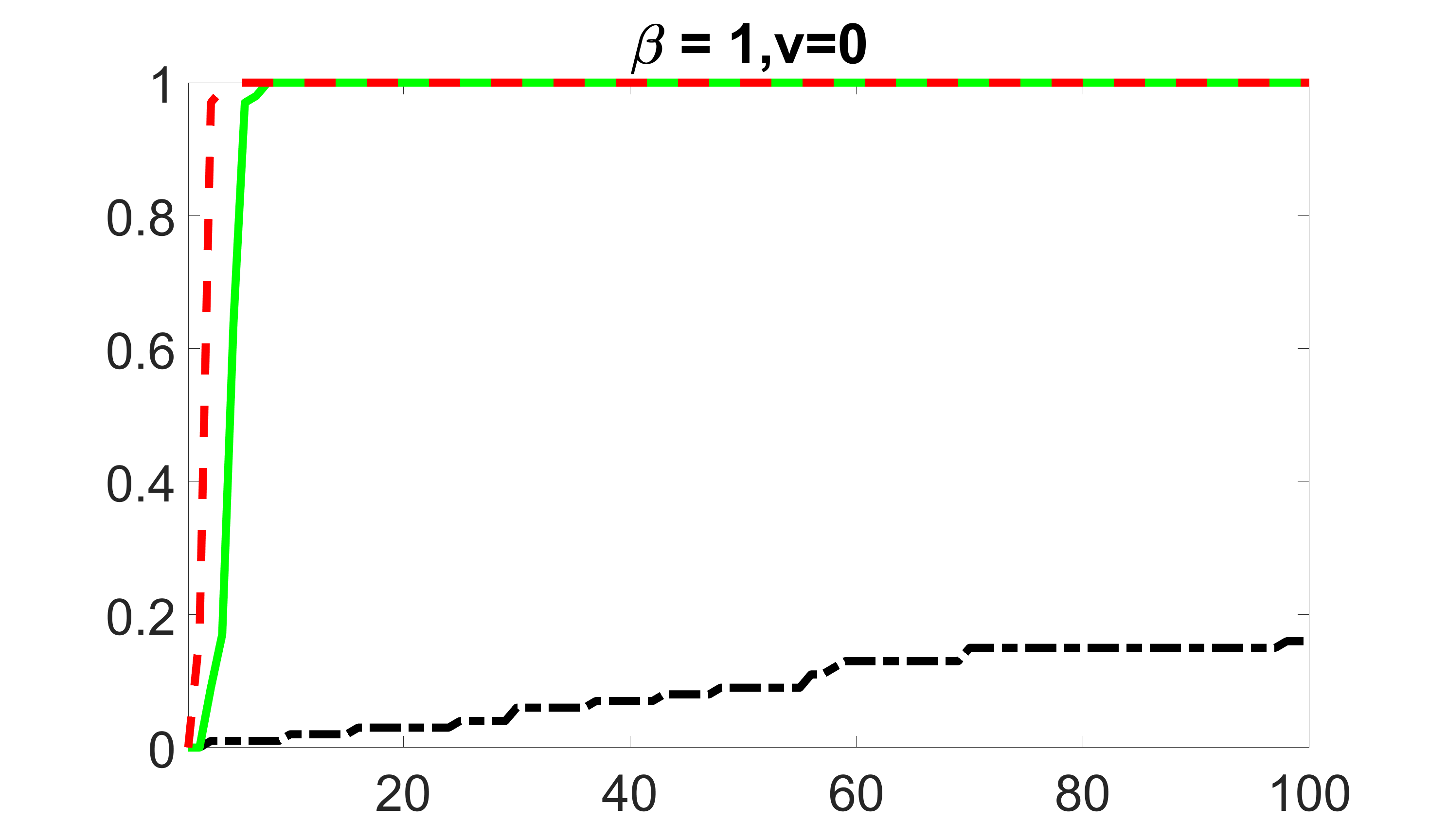}}
  \subcaptionbox{Precision: weak \\ outcome, zero exposure}[0.45\linewidth]
 {\includegraphics[width=6cm,height=3.5cm]{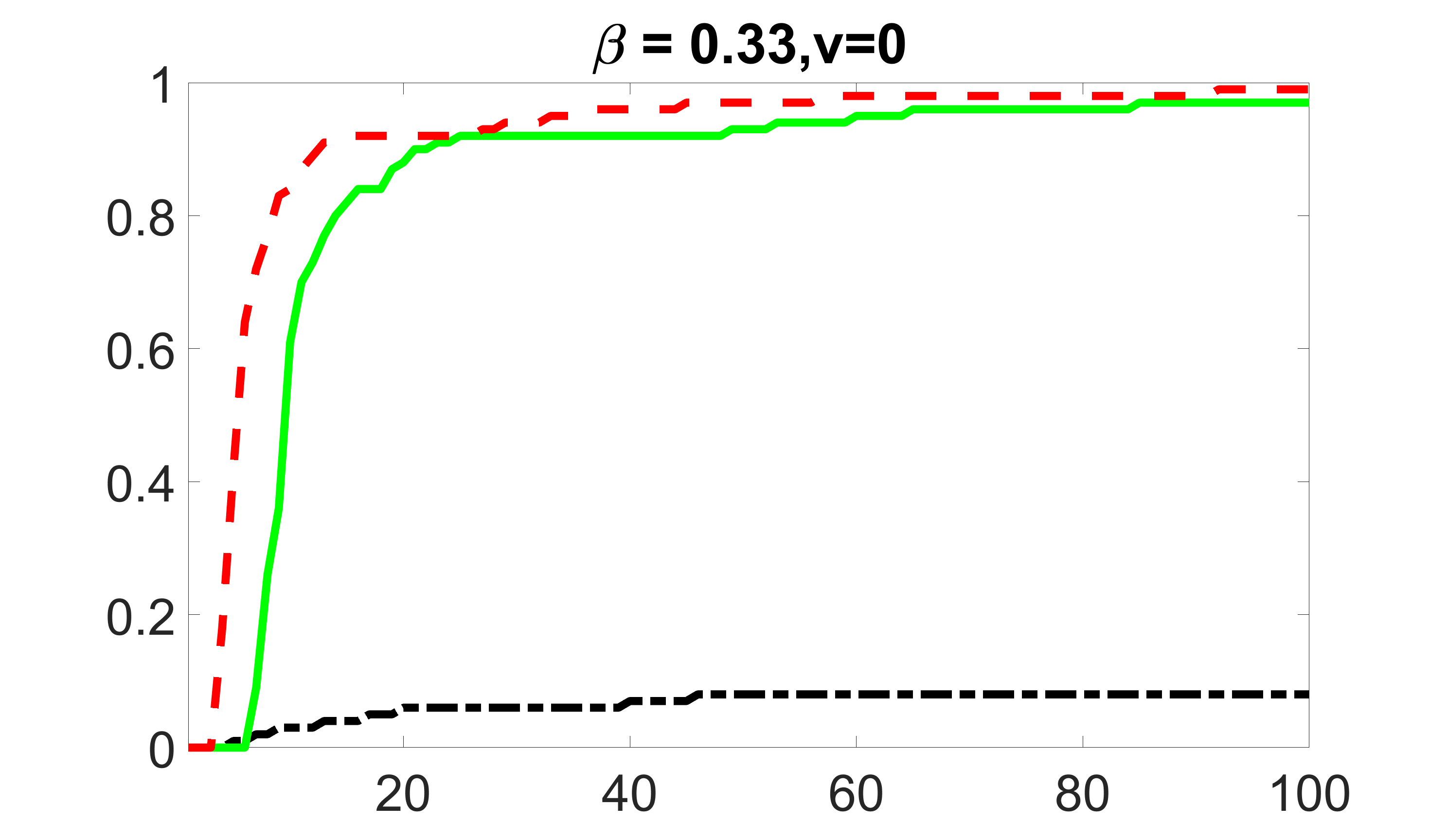}}
  \subcaptionbox{Overall coverage of $\mathcal{M}_1$}[0.45\linewidth]
 {\includegraphics[width=6cm,height=3.5cm]{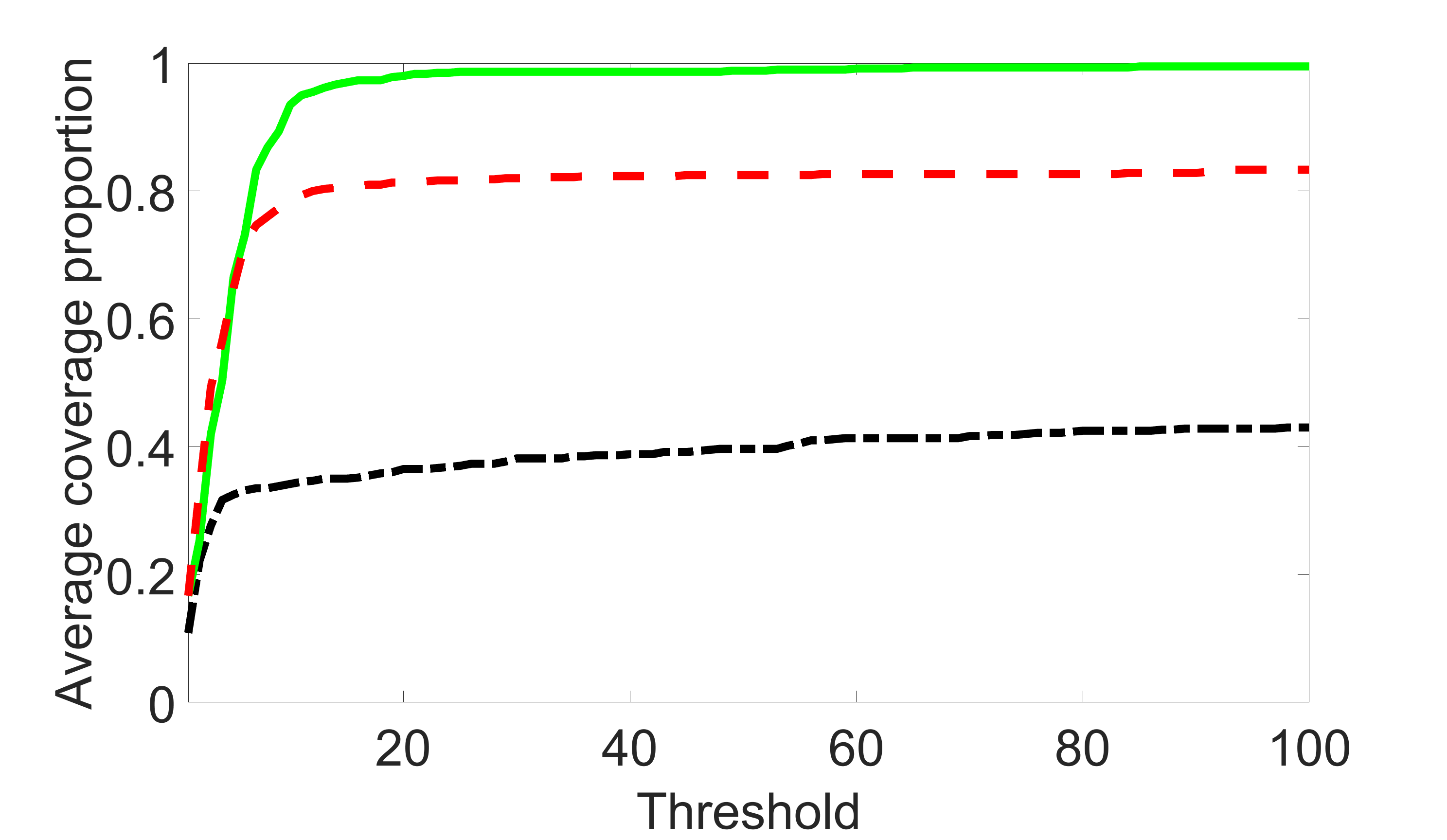}}
\caption{Simulation results for the case $(n,s,\sigma) = (200,5000,1)$: Panels (a) -- (f) plot the average coverage proportion for $X_l$, where $l=1,2,3,104,105$ and $106$. Panels (a) -- (c) correspond to strong outcome and weak exposure predictor, moderate outcome and moderate exposure predictor and weak outcome and strong exposure predictor; Panels (d) -- (f) correspond to strong, moderate and weak predictors of outcome only. Panel (g) plots the average coverage proportion for the index set $\mathcal{M}_1 = \{1,2,3,104,105,106\}$. The x-axis represents the size of $\widehat{\mathcal{M}} $, while
y-axis denotes the average proportion. The green solid, the red dashed and the black dash dotted lines denote our joint screening method, the outcome screening method, and the intersection screening method, respectively.}
\label{sim1step1n200sigma1}
\end{figure}

From the plots, one can see that both the ``intersection'' and ``outcome" screening methods miss the confounder $ X_3 $ with a very high probability even as the size of the selected set approaches $100$. In contrast,  our method can select $ X_3 $ with high probability when $ |\widehat{\mathcal{M}}| $ is relatively small. For confounders $ X_1$ and $ X_2 $, all three methods perform similarly. For the precision variables, the ``outcome" method and our ``joint" method perform similarly in covering these variables, while the ``intersection'' performs badly. Combining the results, one can see that our method performs 
the best as our method selects all the confounders and precision variables with high probabilities. In addition, we find that the coverage proportion of our method increases when the sample size increases, which validates the sure independence screening property developed in Section \ref{Theoretical guarantees} of the supplementary material.

\subsection{Simulation for estimation}
\label{Simulation for estimation}
In this part, we evaluate the performance of our estimation procedure after the first-step screening. For the size of $ \widehat{\mathcal{M}} $ in the screening step, we set $|\widehat{\mathcal{M}}| = \lfloor n / \log(n) \rfloor$. 
We compare the proposed estimate with the oracle estimate, which is calculated by adjusting for the ideal adjustment set including only confounders and precision variables as $X$ and then estimate $\bm{B}$ by using the optimization (\ref{min1}) without imposing the $l_{1}$-regularization. We report the mean squared errors (MSEs) for $\bm{\beta}$ and $\bm{B}$ defined as
$||{\bm\beta}_{}-\widehat{\bm\beta}||_2^2$ and $ \|{\bm{B}}-\widehat{\bm{B}}\|_F^2$, respectively. Table $\ref{sim1t1oracle}$ summarizes the average MSEs of the proposed and oracle estimates for $\bm\beta$ and $\bm{B}$ among 100 Monte Carlo runs when $n = 200$, $500$ and $1000$. We can see that the MSE decreases as the sample size increases. In terms of the primary parameter of interest $\bm{B}$, the proposed estimate is close to the oracle estimate.

\begin{table}[htbp]
\centering
\caption{Simulation results of the proposed joint screening method and oracle estimates for $ \sigma=1 $ and $ \sigma = 0.5 $, when $n=200$, $500$ and $1000$: the average MSEs for $\bm\beta$ and $ {\bm B}$, and their associated standard errors in the parentheses are reported. The results are based on 100 Monte Carlo repetitions.}
\begin{tabular}{ rcc rcc }
$\sigma = 1.0$ & MSE $\bm\beta$ & MSE ${\bm{B}}$ &$\sigma = 0.5$  &MSE $\bm\beta$ & MSE ${\bm{B}}$  \\
\hline
\multicolumn{6}{c}{n = 200}\\
Proposed   &0.496(0.021)&0.667(0.005) & Proposed    &0.276(0.009)&0.528(0.005)\\
Oracle   &0.086(0.005)&0.624(0.004)&Oracle   &0.021(0.001)&0.501(0.004)\\
\multicolumn{6}{c}{n = 500}\\
Proposed   &0.303(0.008)&0.574(0.006)& Proposed    &0.191(0.005)&0.345(0.004)\\
Oracle   &0.036(0.002)&0.553(0.005)&Oracle   &0.006(0.000)&0.340(0.004)\\
\multicolumn{6}{c}{n = 1000}\\
Proposed   &0.217(0.004)&0.449(0.004)& Proposed    &0.128(0.006)&0.234(0.002)\\
Oracle   &0.013(0.001)&0.460(0.005)&Oracle   &0.003(0.000)&0.233(0.002)\\
\end{tabular}
\label{sim1t1oracle}
\end{table}

We plot the 2D map of  $\widehat{\bm B}$ 
based on the average of 100 Monte Carlo runs in Figure \ref{sim1Btn1000sigma025}(c). For comparison, we also plot the corresponding average oracle estimate $\widehat{{\bm B}}_{\rm{oracle}}$ in Figure \ref{sim1Btn1000sigma025}(b) and the true ${\bm B}$ in Figure \ref{sim1Btn1000sigma025}(a). 
One can see that the proposed method recovers the signal pattern reasonably well.

\begin{figure}[htbp]
\captionsetup[subfigure]{justification=centering}
\centering
 \subcaptionbox{Truth}[0.30\linewidth]
 {\frame{\includegraphics[scale = 2]{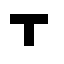}}}
  \subcaptionbox{Oracle }[0.30\linewidth]
 {\frame{\includegraphics[scale = 2]{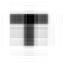}}}
  \subcaptionbox{Proposed }[0.30\linewidth]
 {\frame{\includegraphics[scale = 2]{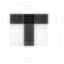}}}
\caption{Panel (a) plots the true ${\bm B}$ (Truth), Panel (b) 
plots the average of $\widehat{\bm B}_{\rm{oracle}}$ (Oracle), Panel (c) 
plots the average of $\widehat{\bm B}$ (proposed).
Here $(n,s,\sigma) = (1000,5000,0.5)$. The value at each pixel is a gray scale with 0 as white and 0.0408 as black.}
\label{sim1Btn1000sigma025}
\end{figure}

We also report the sensitivity and specificity of the estimates in Section \ref{Sensitivity and specificity analysis of simulation} of the supplementary material. We have found that although the proposed method may not remove all of the instrumental variables, eliminating even just some of the instruments greatly reduces the MSEs of both $\bm\beta$ and ${\bm{B}}$, compared to the method where we do not impose $l_{1}$-regularization on $\bm\beta$ in the second-step estimation.

\subsection{Screening and estimation using blockwise joint  screening}
\label{Group Selection Using Linkage Disequilibrium Information main}
Linkage disequilibrium (LD) is a ubiquitous biological phenomenon where genetic variants present a strong blockwise correlation (LD) structure \citep{wall2003haplotype}. In Section \ref{Blockwise joint screening}, we propose the blockwise joint screening procedure to appropriately utilize LD blocks' structural information. The performance of this procedure is illustrated in this section using
an adapted simulation based on the settings of \cite{dehman2015performance}.

For $i=1,\ldots,n$,  $X_i \in \mathbb{R}^s$ is independently generated from an $s$-dimensional multivariate distribution  $N(\bm{0},\bm\Sigma_x)$, where $\bm\Sigma_x = (\sigma_{x,ll^\prime})$ is block-diagonal. If $l \neq l^\prime$ are in the same block, the covariance $\sigma_{x,ll^\prime} = 0.4$, else $\sigma_{x,ll^\prime} = 0$, and the diagonal elements $\sigma_{x,ll}$s are all set to $1$.
We set $X_{ij}$ to $0$, $1$ or $2$ according to whether $X_{ij} < d_1$, $d_1 \leq X_{ij} \leq d_2$, or $X_{ij} > d_2$, where $d_1$ and $d_2$ are thresholds determined for producing a given minor allel frequency (MAF). For instance, choosing $d_1 = \Phi^{-1}(1-6\mathrm{MAF}/5)$ and $d_2 = \Phi^{-1}(1-2\mathrm{MAF}/5)$, where $\Phi$ is the c.d.f. of standard normal distribution, corresponds to a given fixed MAF. 
In order to generate more realistic MAF distributions, we simulate genotype $X_{ij}$, where the MAF for each $j$ is uniformly sampled between 0.05 and 0.5 \citep{dehman2015performance}.

Adapting the simulation setting of Section \ref{Simulation for screening} according to \cite{dehman2015performance}, we set $s=5000$, with $300$ blocks of covariates of size $2, 4, 6, 12, 24, 52$ replicated $50$ times. We perform 100 Monte Carlo runs, and the ordering of the block is drawn at random for each run. The settings for $ {\bm B}$ and $ {\bm C}$ remain the same as before: $ {\bm B} $ is as in Figure 5(a), and $ {\bm C} $ is as in Figure 5(b). Further 
we set $ {\bm C}_l=v_l*\bm{C}$, where $ v_1=-1/3$, $v_2=-1$, $v_3=-3$, $ v_{207}=-3$, $v_{208}=-1$, $v_{209}=-1/3$, and $ v_l=0 $ for $ 4\leq l\leq 206$ and $ 210\leq l \leq s $. 
We set $ \beta_1=3$, $\beta_2=1$, $\beta_3=1/3$, $ \beta_{104}=3$, $\beta_{105}=1$, $\beta_{106}=1/3$,  
$\beta_{j}$ = $1/4$ for $j \in \mathcal{P}_{LD}$, and $ \beta_l=0$ otherwise. 
Here $\mathcal{P}_{LD}$ is a randomly selected block consisting of $K$ consecutive indices from $\{210,211,\ldots,s\}$, where $K \in \{2,4,6,12,24, 52\}$. 
We have $\mathcal{C} = \{1, 2, 3\}$, $\mathcal{P} = \{ 104, 105, 106\} \cup \mathcal{P}_{LD}$, $\mathcal{I} =\{ 207, 208, 209\}$ and $\mathcal{S} = \{1 \ldots, 5000\} \backslash  ( \mathcal{C} \cup \mathcal{P} \cup \mathcal{I})$. The other settings are the same as the ones in Section \ref{Simulation for screening}.

We consider three different sample sizes $n= 200$, $500$ and $1000$. We first perform the screening procedure and report the coverage proportion of $\mathcal{M}_1 = \mathcal{C} \cup \mathcal{P}$. We also report the coverage proportion for each of the confounding and precision variables. 
In particular, we include the screening results, 
in which $s=5000$, $\sigma=1$, $n=200,500,1000$, and $K =2,4,6,12,24,52$, can be found in Figures \ref{sim3step1n200sizesig2sigma1} -- \ref{sim3step1n1000sizesig52sigma1} of the supplementary material. 
 
 From the plots, 
 one can see that the blockwise joint screening method (blue dotted line) selects $\mathcal{P}_{LD}$ and $\mathcal{M}_1$ with higher probability compared with the original joint screening method (green solid line). Based on the results, the blockwise joint screening method can better utilize precision variables with block structures to select $\mathcal{M}_1$.

 In addition, we evaluate the performance of the two proposed estimation procedure after the first-step screening. For the sizes of $ \widehat{\mathcal{M}} $ and $ \widehat{\mathcal{M}} ^{block} $ in the screening step, we set $|\widehat{\mathcal{M}}| = \lfloor n / \log(n) \rfloor$ for the original joint screening procedure, and $|\widehat{\mathcal{M}}^{block}| = 2 \lfloor n / \log(n) \rfloor$ for the blockwise joint screening procedure. We report the average MSEs for $\bm{\beta}$ and $\bm{B}$ when $(n, s, \sigma) = (200,5000, 1)$ in Table \ref{sim1t2main}. The complete results, in which $s=5000$, $\sigma=1$, $n=200, 500, 1000$, and $K =2,4,6,12,24,52$, can be found in Table \ref{sim1t2} of the supplementary material.

In summary, the blockwise joint screening estimate outperforms the original joint screening estimate when the sample size $n$ is small or block size of precision variables $K$ is large. For the rest of the scenarios, there are no significant differences between the two methods.

\begin{table}[htbp]
\centering
\caption{Simulation results for $ (n,s,\sigma) =(200,5000, 1) $: the average MSEs for $\bm\beta$ and $ {\bm B}$, and their associated standard errors in the parentheses are reported. The left panel summarizes the results from the joint screening method; the right panel summarizes the results from the blockwise joint screening method. The results are based on 100 Monte Carlo repetitions. }
\begin{tabular}{ crr | crrrr }
Proposed & MSE $\bm\beta$ & MSE ${\bm{B}}$ & Proposed (block) &MSE $\bm\beta$ & MSE ${\bm{B}}$  \\
\hline
K=2&1.423(0.096)&0.785(0.009)&K=2&1.390(0.090)&0.793(0.010)\\
K=4&1.667(0.096)&0.815(0.011)&K=4&1.548(0.088)&0.805(0.010)\\
K=6&1.955(0.101)&0.826(0.010)&K=6&1.701(0.084)&0.816(0.009)\\
K=12&2.466(0.096)&0.890(0.039)&K=12&2.223(0.129)&0.838(0.011)\\
K=24&2.533(0.164)&0.847(0.014)&K=24&2.136(0.138)&0.821(0.010)\\
K=52&14.650(0.815)&2.034(0.487)&K=52&13.693(0.728)&1.870(0.459)\\
\end{tabular}
\label{sim1t2main}
\end{table}

In the supplementary material, 
we assess the variable screening results for various sparsity levels of instrumental variables in Section \ref{Screening under different sparsity levels}, evaluate the performance of our estimation procedure for different covariances of exposure errors in Section \ref{Screening under different covariance structures of exposure errors}, and assess the sensitivity of the choices for different sizes of $\widehat{\mathcal{M}}_{1}^{*}$ and $\widehat{\mathcal{M}}_{2}$ in Section \ref{Screening and estimation under different sizes}.

\section{Discussion}
\label{sec:discussion}
This paper aims at mapping the complex GIC pathway for AD.  The unique features of the hippocampal morphometry surface measure data motivate us to develop a computationally efficient two-step screening and estimation procedure, which can select biomarkers among more than $6$ million observed covariates and estimate the conditional association simultaneously. If there was no unmeasured confounding, then the conditional association we estimate corresponds to the causal effect. This is,  however, not the case in the ADNI study because we have unmeasured confounders such as A$\beta$ and tau protein levels. To control for unmeasured confounding and estimate the causal effect, one possible approach is to use the generic variants as potential instrumental variables \citep[e.g.][]{lin2015regularization}.

There are a number of other important directions for future work. Firstly, the vast majority of AD, known as ``polygenic AD'', is influenced by the actions and interactions of multiple genetic variants simultaneously, likely in concert with non-genetic factors, such as environmental exposures and lifestyle choices among many others \citep{bertram2020genomic}. Therefore, various types of interaction effects, such as genetic-genetic and imaging-genetic, could be incorporated into the outcome generating model (\ref{outcomemodel}). However, this may significantly increase the computation as the dimension of genetic relevant covariates will increase from $6,087,205$ to more than $ 90$ billion covariates, if we add all the possible imaging-genetics interaction terms. 
One may consider interaction screening procedures \citep{hao2014interaction} as the first-step.
Secondly, this study simply removes observations with missingness. Accommodation of missing exposure, confounders and outcome under the proposed model framework is of great practical value and worth further investigation. 
Thirdly, baseline diagnosis status is an important effect modifier, as the effect of hippocampus shape on behavioral measures can be different across the CN/MCI/AD groups. 
However, the relatively small sample size in the ADNI study does not allow us to conduct a reliable subgroup analysis.
The subgroup analyses are   pertinent for further exploration when a larger sample size is available.
Fourthly, in the ADNI dataset, there are longitudinal ADAS-13 scores observed at different months as well as other longitudinal behavioral scores obtained from Mini-Mental State Examination and Rey Auditory Verbal Learning Test, which can provide a more comprehensive characterization of the behavioral deficits. Integrating these different scores as a multivariate longitudinal outcome to improve the estimation of the conditional association requires substantial effort for future research.
Lastly,  one could consider incorporating information from other brain regions. For instance, an entorhinal tau may exist on episodic memory decline through other brain regions, such as the medial temporal lobe \citep{maass2018entorhinal}.

\section*{Supplementary Material}
\label{supplements}
Supplementary material available online contains detailed derivations and explanations of the main algorithm, ADNI data usage acknowledgement, image and genetic data preprocessing steps, screening results and sensitivity analyses of the ADNI data application with a subgroup analysis including only MCI and AD patients, detailed procedure and results for the SNP-imaging-outcome mediation analyses, additional simulation results, theoretical properties of the proposed procedure including the main theorems, assumptions needed for our main theorems, and proofs of auxiliary lemmas and main theorems.

\section*{Acknowledgement}
The authors thank the editor, associate editor and referees for their constructive comments, which have substantially improved the paper. Yu was partially supported by the Canadian Statistical Sciences Institute (CANSSI) postdoctoral fellowship and the start-up fund of University of Texas at Arlington. Wang and Kong were partially supported by the Natural Science and Engineering Research Council of Canada and the CANSSI Collaborative Research Team Grant. Zhu was partially supported by the NIH-R01-MH116527.

\newpage

 \if0\blind
 {\centering
  \title{\bf \LARGE Supplementary Material for \\``Mapping the Genetic-Imaging-Clinical Pathway with Applications to Alzheimer's Disease''}

  \maketitle
\begin{center}
  \author{\large Dengdeng Yu \\
  \vspace{10pt}
  Department of Mathematics, University of Texas at Arlington}\\ \vspace{10pt}
    \author{\large Linbo Wang \\
    \vspace{10pt}
  Department of Statistical Sciences, University of Toronto}\\ \vspace{10pt}
    \author{\large Dehan Kong \\
    \vspace{10pt}
  Department of Statistical Sciences, University of Toronto}\\ \vspace{10pt}
  \author{\large Hongtu Zhu \\
  \vspace{10pt}
  Department of Biostatistics, University of North Carolina, Chapel Hill }\\
  \vspace{10pt}
 {\large for the Alzheimer's Disease Neuroimaging Initiative
\footnote[1]{
Data used in preparation of this article were obtained from the Alzheimer's Disease
Neuroimaging Initiative (ADNI) database (adni.loni.usc.edu). As such, the investigators within the ADNI contributed to the design and implementation of ADNI and/or provided data but did not participate in analysis or writing of this report. A complete listing of ADNI investigators can be found at: \url{http://adni.loni.usc.edu/wp-content/uploads/how_to_apply/ADNI_Acknowledgement_List.pdf}.} }
\end{center}
\newpage
} \fi

 \if1\blind
 {
  \title{\bf Supplementary Material for ``Mapping the Genetic-Imaging-Clinical   Pathway with Applications to Alzheimer's Disease''}
  \maketitle
} \fi

\spacingset{1.7} % DON'T change the spacing!

The supplementary file is organized as follows. The detailed description of the main algorithm is included in Section \ref{commentforalgorithm}. We include the ADNI data usage acknowledgement in Section \ref{usage} and image and genetic data preprocessing steps in Section \ref{image genetic preprocessing}. In Section \ref{Screening results of ADNI data applications}, we list the screening results of ADNI data applications.
Section \ref{Sensitivity analysis}  examines the sensitivity of the estimate $\widehat{\bm B}$ from the ADNI data application by varying the relative sizes of $\widehat{\mathcal{M}}_{1}^{*}$ and $\widehat{\mathcal{M}}_{2}$.   
Section \ref{Subgroup analysis ADNI data applications} includes a  subgroup analysis including only the 
mild cognitive impairment (MCI) and  Alzheimer's disease (AD) patients.
We include the detailed SNP-imaging-outcome mediation analysis procedure and results in Section \ref{Results for mediation analyses}. 
In Section \ref{addsimulation}, we list additional simulation results. 
The theoretical properties including the main theorems of our procedure are included in Section \ref{Theoretical guarantees}.   
We state the assumptions needed for the main theorems in Section \ref{assumption}. In Section \ref{auxlemma}, we include the auxiliary lemmas needed for the theorems and their proofs. We give the detailed proofs of our main theorems in Section \ref{proofthm}.

\section{Description and derivation of Algorithm 1}\label{commentforalgorithm}
To solve the minimization problem (\ref{min1}) of the main paper, we utilize the Nesterov optimal gradient method \citep{nesterov1998introductory}, which has been widely used in solving optimization problems for non-smooth and non-convex objective functions \citep{beck2009fast, zhou2014regularized}. 

Before we introduce Nesterov's gradient algorithm, we first state two propositions on shrinkage thresholding formulas \citep{beck2009fast, cai2010singular}. 

\begin{prop}
\label{prop1}
For a given matrix $\bm{A}$ with singular value decomposition $\bm{A} = \bm{U} \mathrm{diag}(\bm{a}) \bm{V}^{\T}$, where $\bm{a}=(a_1, \ldots, a_m)^{\T} $ with $\{a_k\}_{1 \leq k \leq r}$ 
being $\bm{A}$'s singular values, the optimal solution to
\[
\min_{\bm{B}} \left\{ \frac{1}{2} \| \bm{B} - \bm{A} \|_F^2 +  \lambda \| \bm{B}\|_*\right\}
\]
share the same singular vectors as $\bm{A}$ and its singular values are $b_k = (a_k - \lambda)_+$ for $k =1,\ldots, r$, where $ (a_k)_+=\max(0, a_k) $. 
\end{prop}

\begin{prop}
\label{prop2}
For vectors $\bm{a}=(a_1, \ldots, a_r)^{\T} $ and $\bm{b}=(b_1, \ldots, b_r)^{\T}$, the optimal solution to 
\[
\min_{\bm b} \left\{  \frac{1}{2}\| \bm{b} - \bm{a} \|_2^2 +  \lambda \| \bm b \|_1\right\}
\]
is $b_k =  \textrm{sgn} (a_k) (|a_k| - \lambda)_+$ for $k =1,\ldots, r$, where $ \textrm{sgn}(\cdot) $ denotes the sign of the number. 
\end{prop}

The Nestrov's gradient method utilizes the first-order gradient of the objective function to produce the next iterate based on the current search point. Differed from the standard gradient descent algorithm, the Nesterov's gradient algorithm uses two previous iterates to generate the next search point by extrapolating, which can dramatically improve the convergence rate. The Nesterov's gradient algorithm for problem (\ref{min1}) is presented as follows.  Denote $l(\bm\beta, \bm B) = \frac{1}{2 n}\sum_{i=1}^n\left( Y_i -\langle \bm{\beta}, X_i\rangle- \langle  \bm{Z}_i, \bm{B}\rangle\right)^2$ and $P(\bm{\beta},\bm{B}) = P_1(\bm{\beta}) + P_2(\bm{B})$, where $  P_1(\bm{\beta}) = \lambda_{1} \sum_l |\beta_l|  $ and $P_2(\bm{B})=  \lambda_{2} || \bm{B}||_*$. We also define
\begin{eqnarray*}
g(\bm\beta,\bm{B}| \bm{s}^{(t)},\bm{S}^{(t)} , \delta) &=& l(\bm{s}^{(t)},\bm{S}^{(t)}) + \left\langle \nabla l(\bm{s}^{(t)}, \bm{S}^{(t)}),\left[ (\bm\beta-\bm{s}^{(t)})^{\T}, \{\mathrm{vec}(\bm{B} - \bm{S}^{(t)})\}^{\T} \right]^{\T} \right\rangle\\
& &+ (2 \delta)^{-1}\left( \left\|\bm\beta-\bm{s}^{(t)} \right\|_2^2+ \left\| \bm{B} - \bm{S}^{(t)}\right\|_F^2 \right)+ P(\bm\beta,\bm{B})\\
 &=& (2 \delta)^{-1}\left[ \left\|\bm\beta - \left\{ \bm{s}^{(t)} - \delta \partial_{\bm{\beta}} l(\bm{s}^{(t)},\bm{S}^{(t)}) \right\} \right\|_2^2  \right.\\
 & & +\left.  \left\|  \mathrm{vec}(\bm{B}) - \left\{ \mathrm{vec}(\bm{S}^{(t)}) - \delta \partial_{\mathrm{vec}(\bm{B})} l(\bm{s}^{(t)}, \bm{S}^{(t)}) \right\} \right\|_2^2  \right] \\
 & &+ P(\bm\beta,\bm{B}) + c^{(t)},
\end{eqnarray*}
where $\nabla l(\bm{\beta},\bm{B}) = [(\partial_{\bm\beta}l)^{\T}, \{ \partial_{\mathrm{vec}(\bm{B})}l\}^{\T}]^{\T} \in \mathbb{R}^{|\widehat{\mathcal{M}}|+pq}$ denotes the first-order gradient of $l(\bm\beta,\bm{B})$ with respect to $\left[ \bm{\beta}^{\T},\{\mathrm{vec}(\bm{B})\}^{\T}\right]^{\T} \in \mathbb{R}^{|\widehat{\mathcal{M}}|+pq}$. We define 
\begin{eqnarray*}
\frac{\partial}{\partial_{\bm\beta}}l(\bm{\beta},\bm{B}) &=& n^{-1} \sum_{i=1}^{n}  X_i \left(\langle \bm{\beta}, X_i\rangle + \langle \bm{B}, \bm{Z}_i\rangle - Y_i\right),\\
\frac{\partial}{\partial_{\mathrm{vec}(\bm{B})}}l(\bm{\beta},\bm{B}) &=& \mathrm{vec} \left\{ n^{-1} \sum_{i=1}^{n}  \bm{Z}_i \left(\langle \bm{\beta}, X_i\rangle + \langle \bm{B}, \bm{Z}_i\rangle - Y_i\right) \right\},
\end{eqnarray*}
with $\partial_{\bm\beta}l(\bm{\beta},\bm{B}) \in \mathbb{R}^{|\widehat{\mathcal{M}}|}$, $\partial_{\mathrm{vec}(\bm{B})}l(\bm{\beta},\bm{B}) \in \mathbb{R}^{pq}$. 
Here $s^{(t)}$ and $S^{(t)}$ are interpolations between $\bm\beta^{(t-1)}$ and $\bm\beta^{(t)}$, and $\bm{B}^{(t-1)}$ and $\bm{B}^{(t)}$ respectively, which will be defined below; $c^{(t)}$ denotes all terms that are irrelevant to $\bm{B}$, and $\delta >0$ is a suitable step size. 
Given previous search points $\bm{s}^{(t)}$ and $\bm{S}^{(t)}$, the next search points $\bm{s}^{(t+1)}$ and $\bm{S}^{(t+1)}$ would be the minimizer of $g(\bm\beta,\bm B| \bm{s}^{(t)},\bm{S}^{(t)} , \delta)$. 
For the search points $\bm{s}^{(t)}$ and $\bm{S}^{(t)}$, they can be generated by linearly extrapolating two previous algorithmic iterates. 
A key advantage of using the Nestrov gradient method is that it has an explicit solution at each iteration.
In fact, minimizing $g(\bm\beta,\bm B| \bm{s}^{(t)},\bm{S}^{(t)} , \delta) $ can be divided into two sub-problems, minimizing $(2 \delta)^{-1} \left\|\bm\beta - \left( \bm{s}^{(t)} - \right. \right.$ $\left. \left. \delta \partial_{\bm\beta} l(\bm{s}^{(t)},\bm{S}^{(t)}) \right) \right\|_2^2 + \lambda_{1} \sum_l |\beta_l| $ and 
$(2 \delta)^{-1}\left\|  \mathrm{vec}(\bm{B}) - \left\{ \mathrm{vec}(\bm{S}^{(t)}) - \right. \right.$ $ \left. \left.  \delta \partial_{\mathrm{vec}(\bm{B})} l(\bm{s}^{(t)}, \bm{S}^{(t)}) \right\} \right\|_2^2 +  \lambda_{2} || \bm{B}||_*$, respectively. These sub-problems can be solved by the shrinkage thresholding formulas in Propositions \ref{prop1} and \ref{prop2}, respectively.

Let $ \bm{X}^{\widehat{\mathcal{M}}}=(X^{\widehat{\mathcal{M}}}_1, \ldots, X^{\widehat{\mathcal{M}}}_n)^{\T}\in \mathbb{R}^{n\times|\widehat{\mathcal{M}}|} $ where $X^{\widehat{\mathcal{M}}}_i$ is $\{ X_{ij}\}^{\T}_{j \in \widehat{\mathcal{M}}} \in \mathbb{R}^{|\widehat{\mathcal{M}}|}$ for $i = 1, \ldots, n$. Define $\bm{Z}_{new} = (\mathrm{vec}(\bm{Z}_1),\ldots,\mathrm{vec}(\bm{Z}_n))^{\T} \in \mathbb{R}^{n \times pq}$ and $\bm{X}_{new} = (\bm{X}^{\widehat{\mathcal{M}}},\bm{Z}_{new}) \in \mathbb{R}^{n \times (|\widehat{\mathcal{M}}|+pq)}$.
For a given vector $\bm{a} = (a_1,\ldots, a_r)^{\T} \in \mathbb{R}^{r}$, $(\bm{a})_+$ is defined as $\{(a_1)_+,\ldots, (a_r)_+\}^{\T} \in \mathbb{R}^{p}$, where $ (a)_+=\max(0, a) $. Similarly,  $\textrm{sgn}(\bm{a})$ is obtained by taking the sign of $\bm{a}$ componentwisely. 
For a given pair of tuning parameters $\lambda_{1}$ and $\lambda_{2}$, \eqref{min1} can be solved by Algorithm \ref{algorithm:solve8}.
\begin{algorithm}[!htb]
\caption{ Shrinkage thresholding algorithm to solve  \eqref{min1}}
\label{algorithm:solve8}
\begin{enumerate}
\item  Initialize: $\bm{\beta}^{(0)} = \bm{\beta}^{(1)}$, $\bm{B}^{(0)} = \bm{B}^{(1)}$,  $\alpha^{(0)} = 0$ and $\alpha^{(1)} = 1$,\\
$\delta = n/\lambda_{\textrm{max}}( \bm{X}_{new}^{\T} \bm{X}_{new})$. 
\item Repeat (a) to (f) until the objective function $Q(\bm{\beta},\bm{B})$ converges:\\
\begin{enumerate}
\item $\bm{s}^{(t)} = \bm{\beta}^{(t)} + \frac{\alpha^{(t-1)}-1}{\alpha^{(t)}} (\bm{\beta}^{(t)} - \bm{\beta}^{(t-1)})$,\\
$\bm{S}^{(t)} = \bm{B}^{(t)} + \frac{\alpha^{(t-1)}-1}{\alpha^{(t)}} (\bm{B}^{(t)} - \bm{B}^{(t-1)})$;
\item $\bm{\beta}_{\textrm{temp}} = \bm{s}^{(t)}  - \delta \frac{\partial l(\bm{s}^{(t)},\bm{B}^{(t)})}{\partial \bm{\beta}} $;\\
$\mathrm{vec}(\bm{B}_{\textrm{temp}}) = \mathrm{vec}(\bm{S}^{(t)})  - \delta \frac{\partial l(\bm{\beta}^{(t)},\bm{S}^{(t)})}{\partial \mathrm{vec}(\bm{B})} $;
\item 
Singular value decomposition: ${\bm{B}}_{\textrm{temp}} = \bm{U} \mathrm{diag}(\bm{B}) \bm{V}^{\T}$;
\item $\bm{a}_{new} = \textrm{sgn}(\bm{\beta}_{\textrm{temp}})\cdot (|\bm{\beta}_{temp}| - \lambda_1 \delta\cdot \bm{1})_{+}$,\\
$\bm{b}_{new} = (\bm{b} - \lambda_2 \delta \cdot \bm{1})_{+}$;\\
\item $\bm{\beta}^{(t+1)} = \bm{a}_{new}$, \\
$\bm{B}^{(t+1)} = \bm{U} \mathrm{diag}(\bm{b}_{new}) \bm{V}^{\T}$;\\
\item $\alpha^{(t+1)} = \left[1 + \sqrt{1+(2 \alpha^{(t)})^2}\right]/2$.\\
\end{enumerate}
\end{enumerate}
\end{algorithm}

In particular, step 2(a) predicts search points $\bm{s}^{(t)}$ and $\bm{S}^{(t)}$ by linear extrapolations from the solutions of previous two iterates, 
where $\alpha^{(t)}$ is a scalar sequence that plays a critical role in the extrapolation. This sequence is updated in step 2(f) as in the original Nesterov method. 
Next, steps 2(b) -- 2(d) perform gradient descent from the current search points to obtain the optimal solutions at current iteration. 
Specifically, the gradient descent is based on minimizing $g(\bm\beta,\bm{B}|$ $ \bm{s}^{(t)},\bm{S}^{(t)} , \delta)$, the first-order approximation to the loss function, at the current search points $\bm{s}^{(t)}$ and $\bm{S}^{(t)}$. This minimization problem is tackled by minimizing two sub-problems by the shrinkage thresholding formulas in Propositions \ref{prop1} and \ref{prop2} respectively, as mentioned above.  
Finally, step 2(e) forces the descent property of the next iterate.

A sufficient condition for the convergence of $\{ \bm\beta^{(t)} \}_{t \geq 1}$ and $\{ \bm{B}^{(t)} \}_{t \geq 1}$ is that the step size $\delta$ should be smaller than or equal to $1/{L_f}$, where $L_f$ is the smallest Lipschitz constant of the function $l(\bm{\beta},\bm{B})$  \citep{beck2009fast}. 
In our case, $L_f$ is equal to $\lambda_{\textrm{max}}( \bm{X}_{new}^{\T} \bm{X}_{new})/n$, where $ \lambda_{\textrm{max}}(\cdot) $ denotes the largest eigenvalue of a matrix. 

\section{Data usage acknowledgement}\label{usage}
Data used in the preparation of this article were obtained from the Alzheimer's Disease
Neuroimaging Initiative (ADNI) database (adni.loni.usc.edu). The ADNI was launched in
2003 as a public-private partnership, led by Principal Investigator Michael W. Weiner,
MD. The primary goal of ADNI has been to test whether serial magnetic resonance imaging
(MRI), positron emission tomography (PET), other biological markers, and clinical and
neuropsychological assessment can be combined to measure the progression of mild
cognitive impairment and early Alzheimer's disease. For up-to-date information,
see www.adni-info.org.

Data collection and sharing for this project was funded by the Alzheimer's Disease
Neuroimaging Initiative (ADNI) (National Institutes of Health Grant U01 AG024904) and
DOD ADNI (Department of Defense award number W81XWH-12-2-0012). ADNI is funded
by the National Institute on Aging, the National Institute of Biomedical Imaging and
Bioengineering, and through generous contributions from the following: AbbVie, Alzheimer's
Association; Alzheimer's Drug Discovery Foundation; Araclon Biotech; BioClinica, Inc.;
Biogen; Bristol-Myers Squibb Company; CereSpir, Inc.; Cogstate; Eisai Inc.; Elan
Pharmaceuticals, Inc.; Eli Lilly and Company; EuroImmun; F. Hoffmann-La Roche Ltd and
its affiliated company Genentech, Inc.; Fujirebio; GE Healthcare; IXICO Ltd.; Janssen
Alzheimer Immunotherapy Research \& Development, LLC.; Johnson \& Johnson
Pharmaceutical Research \& Development LLC.; Lumosity; Lundbeck; Merck \& Co., Inc.;
Meso Scale Diagnostics, LLC.; NeuroRx Research; Neurotrack Technologies; Novartis
Pharmaceuticals Corporation; Pfizer Inc.; Piramal Imaging; Servier; Takeda Pharmaceutical
Company; and Transition Therapeutics. The Canadian Institutes of Health Research is
providing funds to support ADNI clinical sites in Canada. Private sector contributions are
facilitated by the Foundation for the National Institutes of Health (www.fnih.org). The grantee
organization is the Northern California Institute for Research and Education, and the study is
coordinated by the Alzheimer's Therapeutic Research Institute at the University of Southern
California. ADNI data are disseminated by the Laboratory for Neuro Imaging at the
University of Southern California.

\section{Image and genetic data preprocessing}
\label{image genetic preprocessing}

\subsection{Image data preprocessing}
\label{image preprocessing}
The hippocampus surface data were preprocessed from the raw MRI data, which were collected across a variety of 1.5 Tesla MRI scanners with protocols individualized for each scanner. Standard T1-weighted images were obtained by  using volumetric 3-dimensional sagittal MPRAGE or equivalent protocols with varying resolutions. The typical protocol includes: inversion time (TI) = 1000 ms, flip angle = 8$^o$,  repetition time (TR) = 2400 ms, and  field of view (FOV) = 24 cm  with a $256\times 256\times 170$ acquisition matrix in the $x-$, $y-$, and $z-$dimensions yielding a voxel size of $1.25\times 1.26\times 1.2$ mm$^3$. We adopted a surface fluid registration based hippocampal subregional analysis package \citep{shi:nimg13}, which uses isothermal coordinates and fluid registration to generate one-to-one hippocampal surface registration for surface statistics computation. It introduced two cuts on a hippocampal surface to convert it into a genus zero surface with two open boundaries. The locations of the two cuts were at the front and back of the hippocampal surface. By using conformal parameterization, it essentially converts a 3D surface registration problem into a 2D image registration problem. The flow induced in the parameter domain establishes high-order correspondences between 3D surfaces. Finally, the radial distance was computed on the registered surface. This software package and associated image processing methods have been adopted and described in \citet{Wang2011}.

\subsection{Genetic data preprocessing}
\label{genetic preprocessing}
For the genetic data, we applied the following preprocessing technique to the $756$ subjects in ADNI1 study. The first line quality control steps include (i) call rate check per subject and per single nucleotide polymorphism (SNP) marker, 
(ii) gender check, (iii) sibling pair identification, (iv) the Hardy-Weinberg equilibrium test, (v) marker removal by the minor allele frequency, and (vi) population stratification. The second line preprocessing steps include removal of SNPs with (i) more than 5$\%$ missing values, (ii) minor allele frequency smaller than 10$\%$, and (iii) Hardy-Weinberg equilibrium $p$-value $< 10^{-6}$. The 503,892 SNPs obtained from 22 autosomes were included for further processing. MACH-Admix software (http://www.unc.edu/~yunmli/MaCH-Admix/) \citep{LiuLi2013} is applied on all the subjects to perform genotype imputation, using 1000G Phase I Integrated Release Version 3 haplotypes (http://www.1000genomes.org) \citep{GPC1000} as reference panel. Quality control was also conducted after imputation, excluding markers with (i) low imputation accuracy (based on imputation output $R^2$), (ii) Hardy-Weinberg equilibrium $p$-value $10^{-6}$, and (iii) minor allele frequency $<5\%$.  

\section{Screening results of ADNI data applications} 
\label{Screening results of ADNI data applications}
In Table \ref{imgenet1}, we list the top $20$ SNPs selected through the blockwise joint screening procedure corresponding to left and right hippocampi respectively. 
\begin{table}[htbp]
\centering
\begin{tabular}{ c c | c c}
\multicolumn{2}{c}{Left hippocampi}  & \multicolumn{2}{c}{Right hippocampi} \\
\hline
Chromesome number & SNP name & Chromesome number & SNP name \\
\hline
19 & rs429358 & 19 & rs429358\\
7 & rs1016394 & 19 & rs10414043\\
19 & rs10414043 & 14 & 14:25618120:G\_GC\\
7 & rs1181947 & 19 & rs7256200\\
19 & rs7256200 & 14 & rs41470748\\
22 & rs134828 & 19 & rs73052335\\
19 & rs73052335 & 14 & 14:25613747:G\_GT\\
7 & 7:101403195:C\_CA & 19 & rs157594\\
19 & rs157594 & 14 & rs72684825\\
13 & rs12864178 & 19 & rs769449\\
19 & rs769449 & 6 & rs9386934\\
2 & rs13030626 & 6 & rs9374191\\
19 & rs56131196 & 19 & rs56131196\\
2 & rs13030634 & 6 & rs9372261\\
19 & rs4420638 & 19 & rs4420638\\
2 & rs11694935 & 6 & rs73526504\\
19 & rs111789331 & 19 & rs111789331\\
2 & rs11696076 & 14 & rs187421061\\
19 & rs66626994 & 19 & rs66626994\\
2 & rs11692218 & 13 & rs342709\\
\end{tabular}
\caption{The top 20 SNPs selected through the blockwise joint screening procedure. The left two columns correspond to results from the left hippocampi, and the right two columns correspond to results from the right hipppocampus.}
\label{imgenet1}
\end{table}

We plot similar figures as the Manhattan plot for $\widehat{\mathcal{M}}^{*}_1$, $\widehat{\mathcal{M}}^{block,*}_1$, $\widehat{\mathcal{M}}^{\textbf{}}_2$ and $\widehat{\mathcal{M}}^{block}_2$ in Figure \ref{manhattanPlots}. Unlike the conventional Manhattan plots, where genomic coordinates are displayed along the x-axis, with the negative logarithm of the association p-value for each SNP displayed on the y-axis, in our analysis, we do not have the p-values. So in these figures, the y-axis represents the magnitude of $| \widehat{\beta}^{M}_{l} |$, $ \widehat{\beta}^{block,M}_{l} $,  $ \| \widehat{\bm{C}}_l^M \|_{op}$ and $\widehat{C}_l^{block,M}$ and the horizontal dashed line represents the threshold values $\gamma_{1,n}$ (Panel (a)), $\gamma_{2,n}$((Panel (b)), $\gamma_{3,n}$ ((Panel (c)) and $\gamma_{4,n}$ ((Panel (d)). In Panels (c) and (d), the left and right figures represent the left and right hippocampi, respectively. The SNPs with $| \widehat{\beta}^{M}_{l} |$, $ \widehat{\beta}^{block,M}_{l} $,  $ \| \widehat{\bm{C}}_l^M \|_{op}$ and $\widehat{C}_l^{block,M}$ greater than or equal to $\gamma_{1,n}$, $\gamma_{2,n}$, $\gamma_{3,n}$ and $\gamma_{4,n}$, hence being selected by $\widehat{\mathcal{M}}^{*}_1$, $\widehat{\mathcal{M}}^{block,*}_1$, $\widehat{\mathcal{M}}^{\textbf{}}_2$ and $\widehat{\mathcal{M}}^{block}_2$ respectively, are highlighted with red diamond symbols.

\begin{figure}[htbp]
\captionsetup[subfigure]{justification=centering}
\centering
 \subcaptionbox{$| \widehat{\beta}^{M}_{l} |$}[0.45\linewidth]
{\includegraphics[height=2in,width=3in]{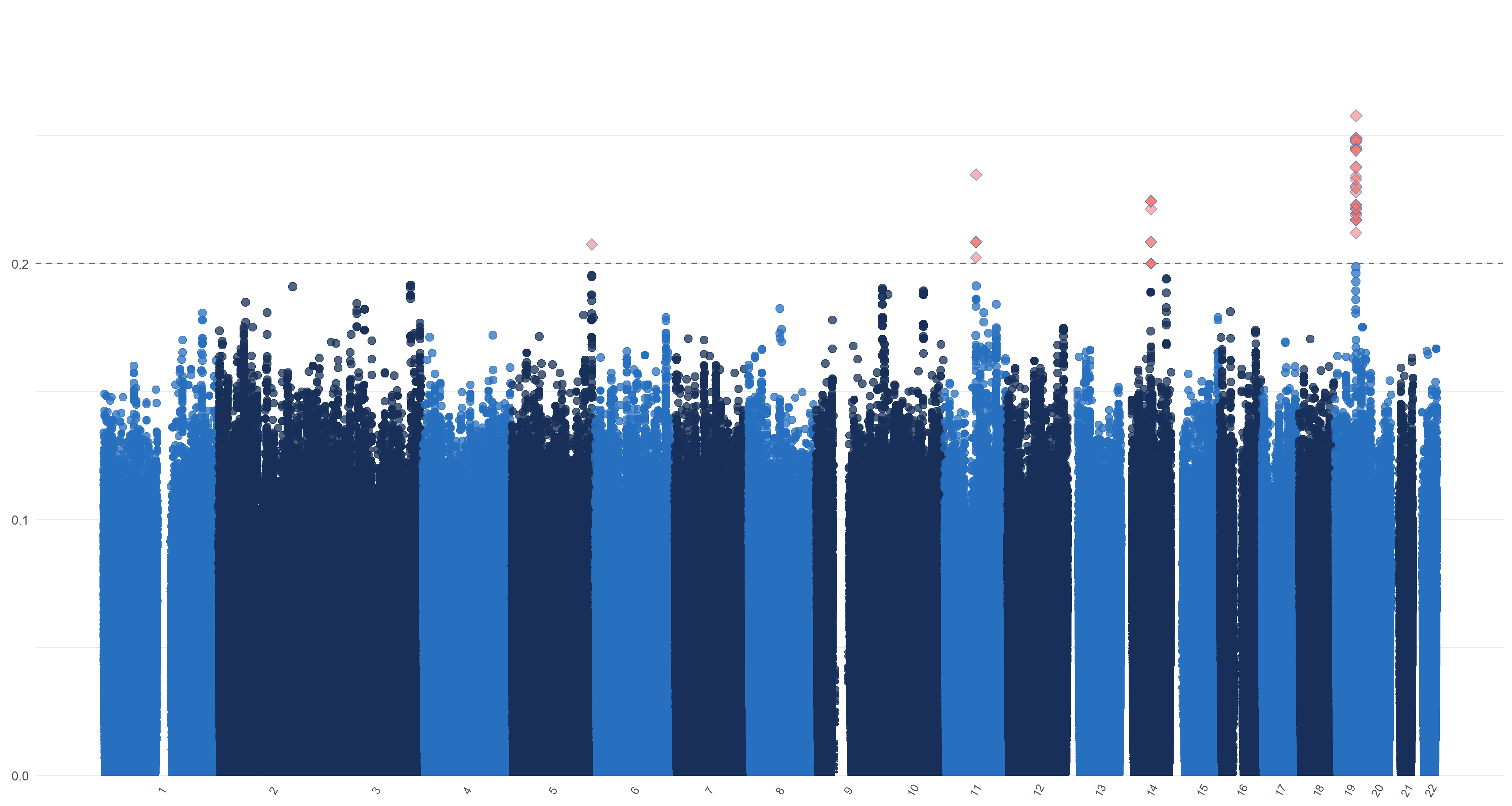}}
 \hfill
 \subcaptionbox{$ \widehat{\beta}^{block,M}_{l} $}[0.45\linewidth]
 {\includegraphics[height=2in,width=3in]{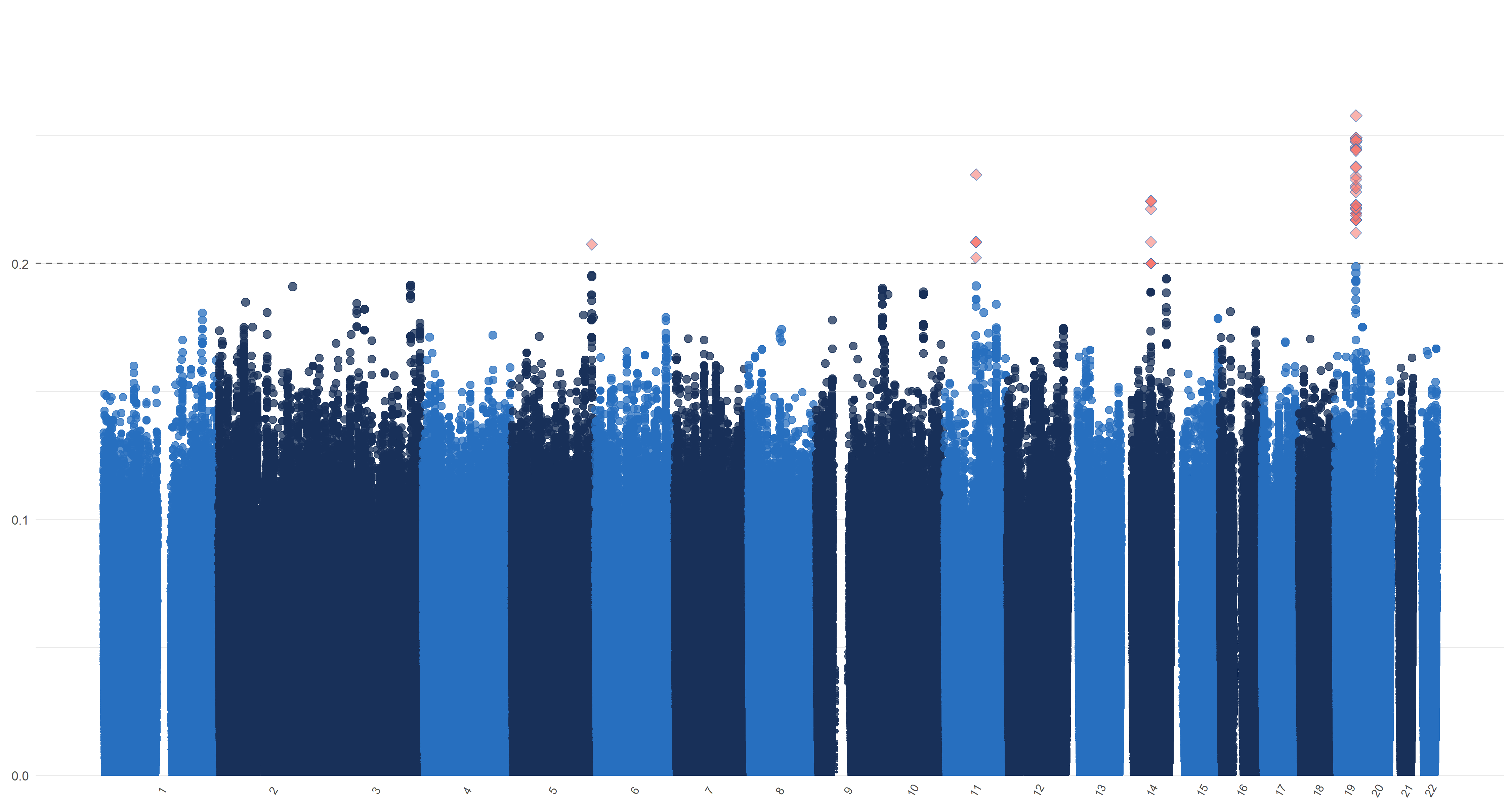}}
 \hfill
  \subcaptionbox{$ \| \widehat{\bm{C}}_l^M \|_{op}$}[0.95\linewidth]
 {\includegraphics[height=2in,width=3in]{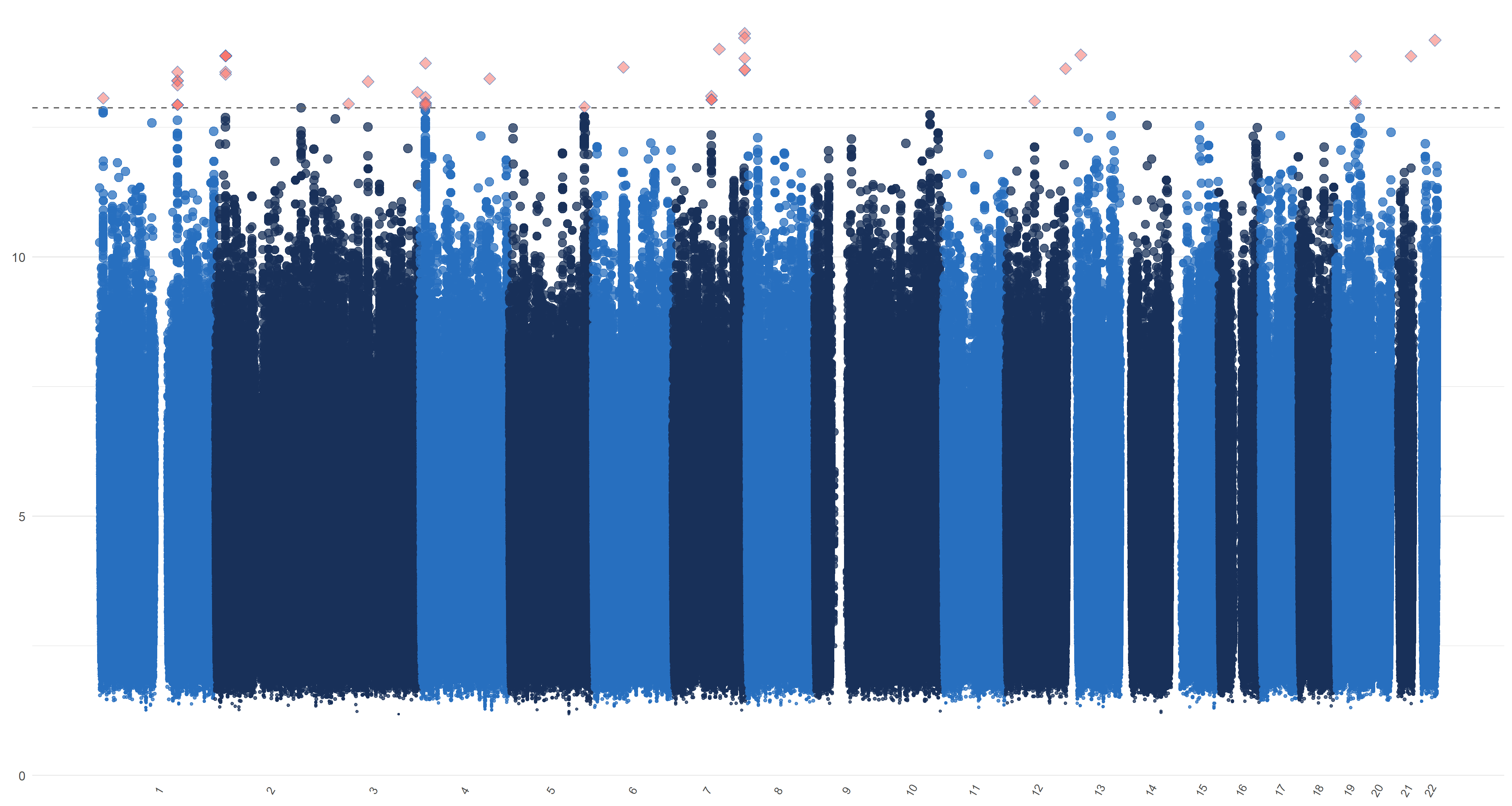}
 \includegraphics[height=2in,width=3in]{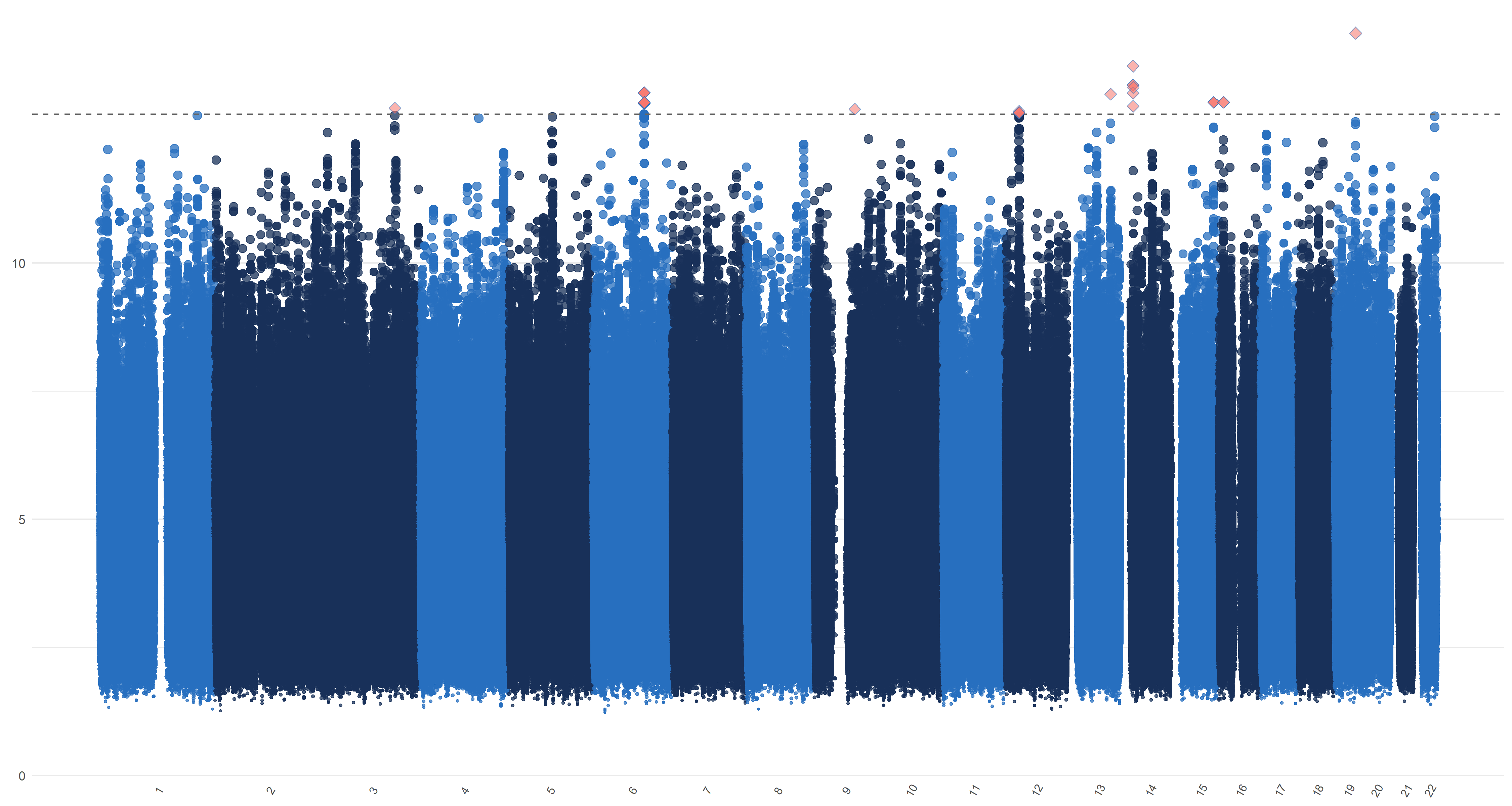}}
 \hfill
  \subcaptionbox{$\widehat{C}_l^{block,M}$}[0.95\linewidth]
 {\includegraphics[height=2in,width=3in]{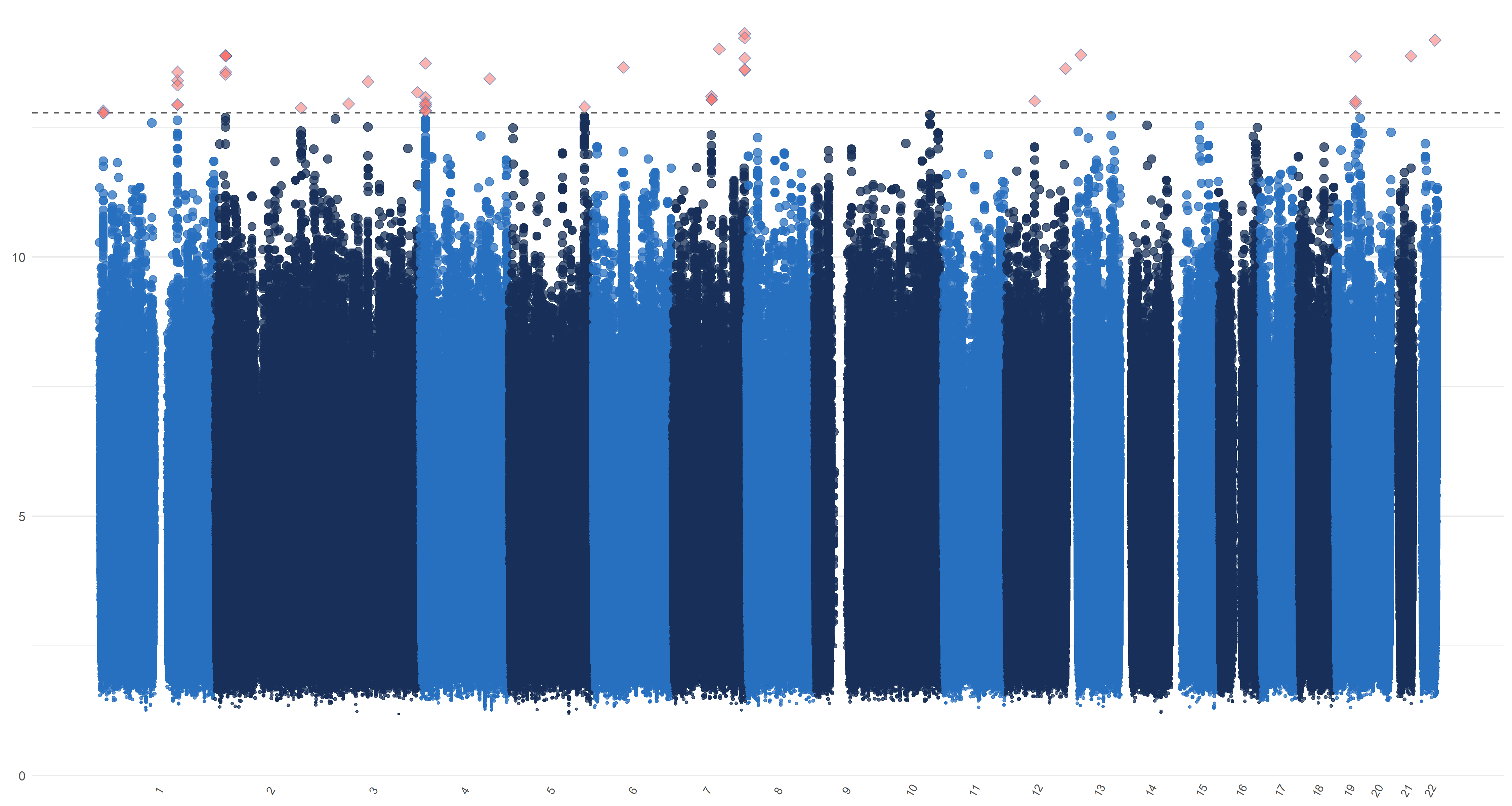}
 \includegraphics[height=2in,width=3in]{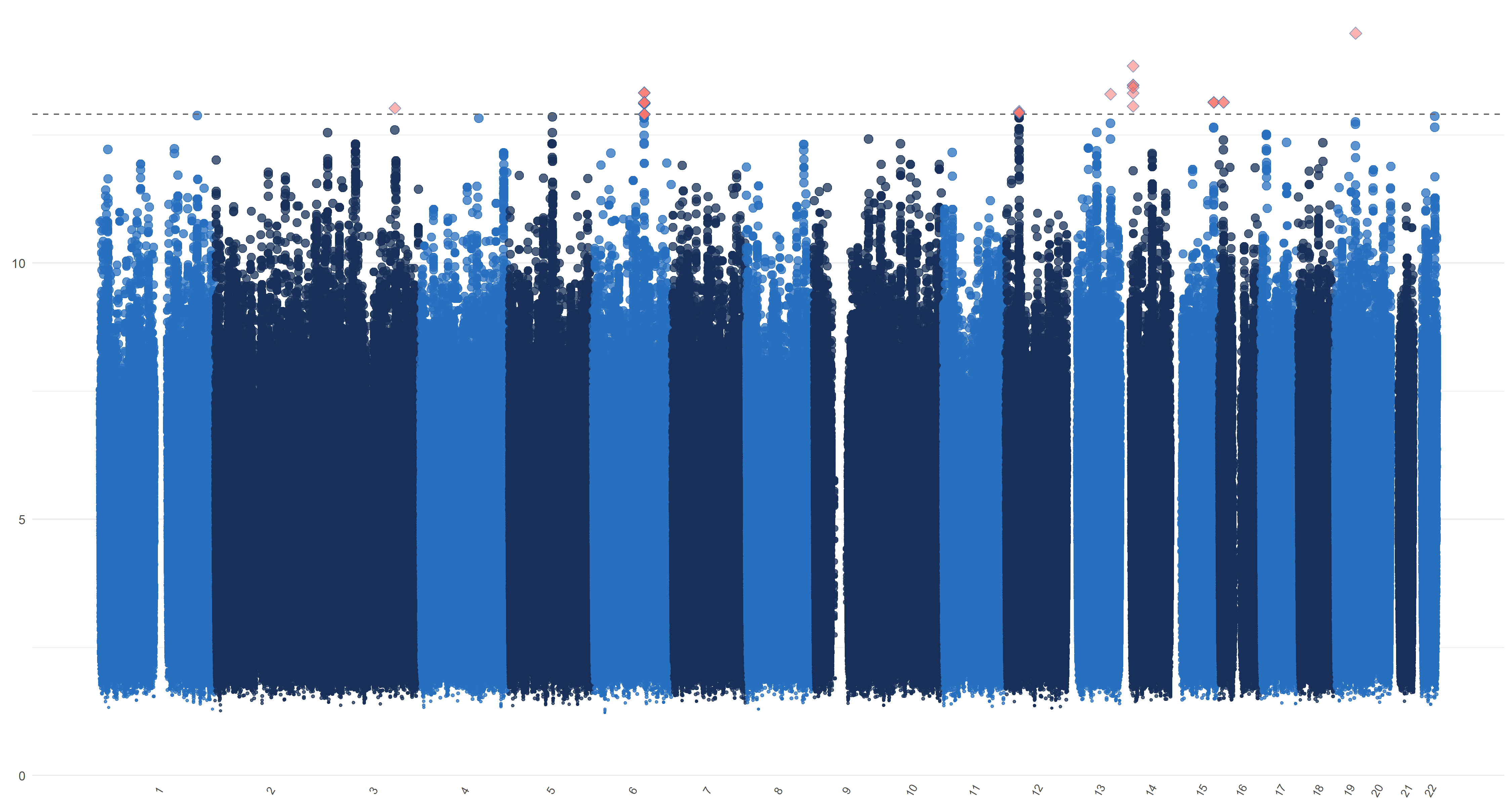}}
 \hfill
    \caption{Real data results: Panels (a) -- (d) present the results for $| \widehat{\beta}^{M}_{l} |$, $ \widehat{\beta}^{block,M}_{l} $,  $ \| \widehat{\bm{C}}_l^M \|_{op}$ and $\widehat{C}_l^{block,M}$, where genomic coordinates are displayed along the x-axis, y-axis represents the magnitude of $| \widehat{\beta}^{M}_{l} |$, $ \widehat{\beta}^{block,M}_{l} $,  $ \| \widehat{\bm{C}}_l^M \|_{op}$ and $\widehat{C}_l^{block,M}$ and the horizontal dashed line represents the threshold values $\gamma_{1,n}$ (Panel (a)), $\gamma_{2,n}$((Panel (b)), $\gamma_{3,n}$ ((Panel (c)) and $\gamma_{4,n}$ ((Panel (d)).
    In Panels (c) and (d), the left and right figures represent the left and right hippocampi, respectively. The SNPs with $| \widehat{\beta}^{M}_{l} |$, $ \widehat{\beta}^{block,M}_{l} $,  $ \| \widehat{\bm{C}}_l^M \|_{op}$ and $\widehat{C}_l^{block,M}$ greater than or equal to $\gamma_{1,n}$, $\gamma_{2,n}$, $\gamma_{3,n}$ and $\gamma_{4,n}$, hence being selected by $\widehat{\mathcal{M}}^{*}_1$, $\widehat{\mathcal{M}}^{block,*}_1$, $\widehat{\mathcal{M}}^{\textbf{}}_2$ and $\widehat{\mathcal{M}}^{block}_2$ respectively, are highlighted with red diamond symbols.}
\label{manhattanPlots}
\end{figure}

\section{Sensitivity analysis of ADNI data applications} 
\label{Sensitivity analysis}

In our analysis, we set $\widehat{\mathcal{M}}_{1}^{*}$ and $\widehat{\mathcal{M}}_{2}$ the same sizes, following the convention that the size of screening set is determined only by the sample size \citep{fan2008sure}, which is the same for $\widehat{\mathcal{M}}_{1}^{*}$ and $\widehat{\mathcal{M}}_{2}$.
To assess the sensitivity of our results to this choice, we conduct sensitivity analyses varying the relative sizes of $\widehat{\mathcal{M}}_{1}^{*}$ and $\widehat{\mathcal{M}}_{2}$ in the joint screening procedure. For simplicity, we only consider the joint screening procedure proposed in Section \ref{jointscreening}. Figure \ref{dataPlots2} lists the estimate $\widehat{\bm B}$ corresponding to the left hippocampi (left part) and the right hippocampi (right part) using $\widehat{\mathcal{M}} = \widehat{\mathcal{M}}_{1}^{*} \cup \widehat{\mathcal{M}}_{2}$. Denote the estimates corresponding to $|\widehat{\mathcal{M}}_{2}|/|\widehat{\mathcal{M}}_{1}^{*}| = 1/2, 1, 2$ by $\widehat{\bm B}^{(0.5)}$, $\widehat{\bm B}^{(1)}$, $\widehat{\bm B}^{(2)}$ respectively. We set $|\widehat{\mathcal{M}}|=\lfloor n/\log(n) \rfloor = 89$. The estimates $\widehat{\bm B}^{(0.5)}$, $\widehat{\bm B}^{(1)}$ and $\widehat{\bm B}^{(2)}$ are plotted in Figure \ref{dataPlots2} (a), (b) and (c) correspondingly. In addition, we consider $|\widehat{\mathcal{M}}_{2}|/|\widehat{\mathcal{M}}_{1}^{*}| =1$ but $|\widehat{\mathcal{M}}|=2\lfloor n/\log(n) \rfloor = 178$. Denote the corresponding estimate by $\widetilde{\bm B}^{(1)}$. We plot  $\widetilde{\bm B}^{(1)}$ in Figure \ref{dataPlots2} (d).

\begin{figure}[htbp]
\captionsetup[subfigure]{justification=centering}
\centering
  \subcaptionbox{}[0.4\linewidth]
 {\includegraphics[height=2.25in,width=3in]{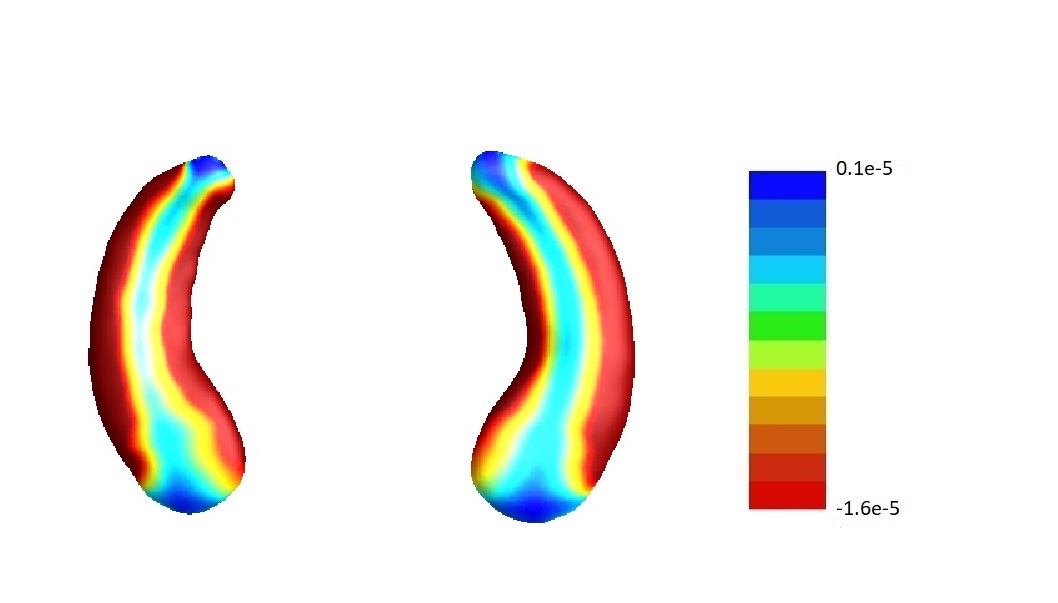}}
 \hfill
 \subcaptionbox{}[0.4\linewidth]
 {\includegraphics[height=2.25in,width=3in]{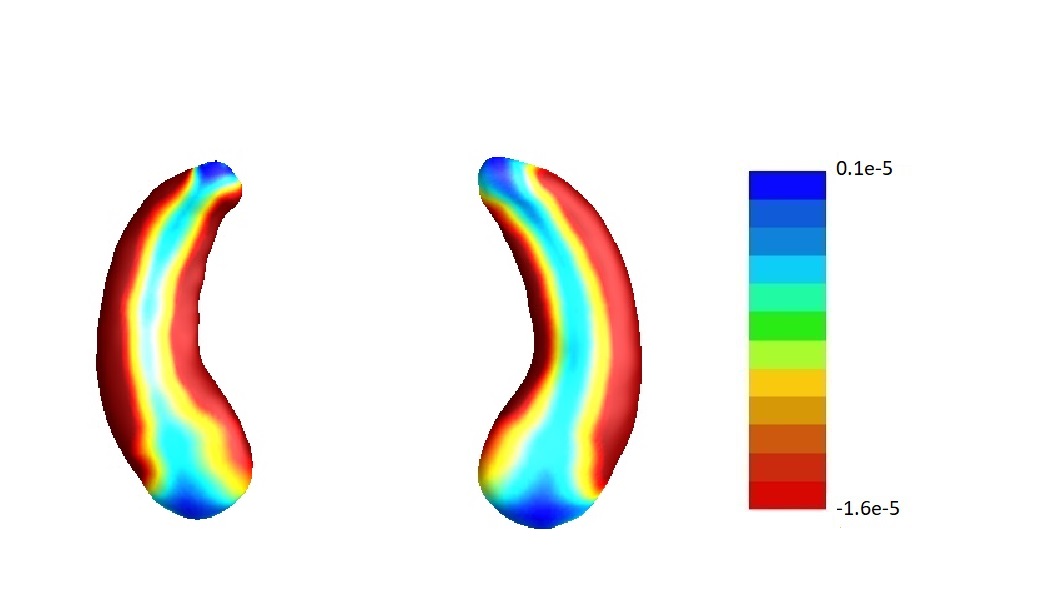}}
 \hfill
 \subcaptionbox{}[0.4\linewidth]
 {\includegraphics[height=2.25in,width=3in]{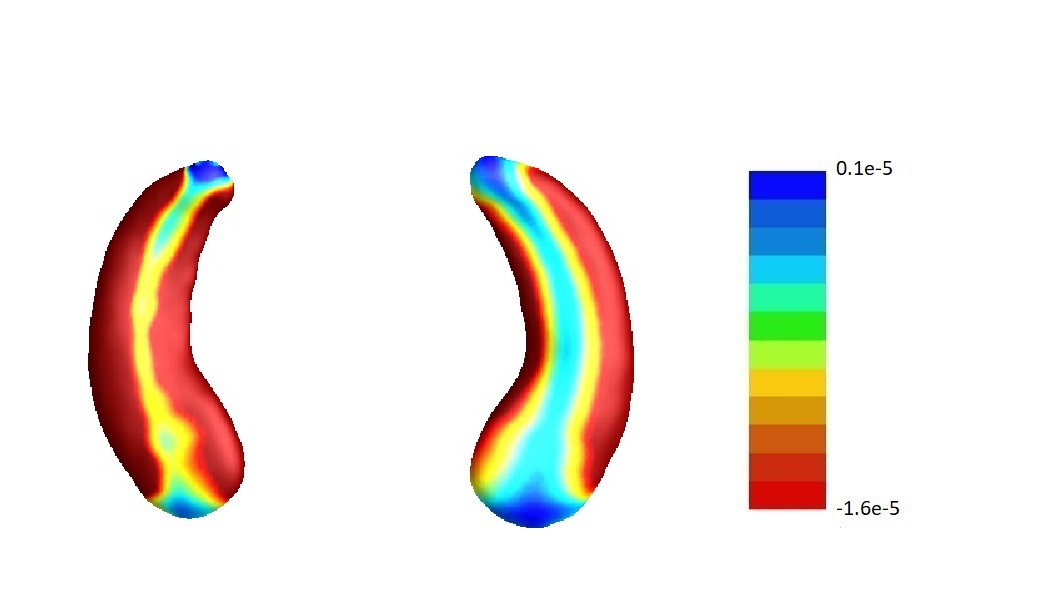}}
  \hfill
  \subcaptionbox{}[0.4\linewidth]
 {\includegraphics[height=2.25in,width=3in]{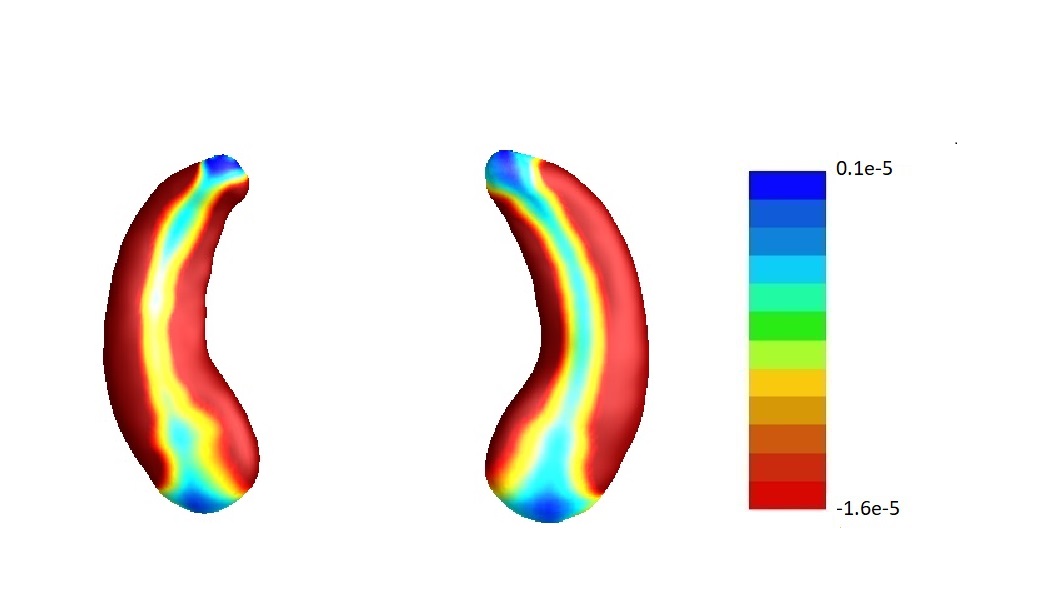}}
 \hfill
\caption{Real Data Results: Panels (a), (b), (c) and (d) plot the estimate $\widehat{\bm B}^{(0.5)}$, $\widehat{\bm B}^{(1)}$, $\widehat{\bm B}^{(2)}$ and $\widetilde{\bm B}^{(1)}$ corresponding to the left hippocampi (left part) and  the right hippocampi (right part).}
\label{dataPlots2}
\end{figure}

Furthermore, by defining the relative risk of an estimate $\widehat{\bm{B}}$ as $\mathrm{RR}(\widehat{\bm{B}}) =  \frac{\| \widehat{\bm{B}}-\widehat{\bm B}^{(1)}  \|_F^2}{\| \widehat{\bm B}^{(1)}  \|_F^2}$, we report the relative risks of three estimates $\widehat{\bm B}^{(0.5)}$, 
$\widehat{\bm B}^{(2)}$ and $\widetilde{\bm B}^{(1)}$ in Table \ref{sensana1}.

\begin{table}[htbp]
\centering
\begin{tabular}{ c c c r}
 & Left hippocampi & Right  hippocampi\\
\hline
$\mathrm{RR}(\widehat{\bm B}^{(0.5)})$ & 0.0022 & 0.1074\\
$\mathrm{RR}(\widehat{\bm B}^{(2)})$ & 0.2938 & 0.0907\\
$\mathrm{RR}(\widetilde{\bm B}^{(1)})$& 0.0611 & 0.0927\\
\end{tabular}
\caption{The relative risks of  $\widehat{\bm B}^{(0.5)}$, 
$\widehat{\bm B}^{(2)}$ and $\widetilde{\bm B}^{(1)}$ for left and right hippocampi. }
\label{sensana1}
\end{table}

\begin{table}[htbp]
\centering
\begin{tabular}{ c c c r}
Number of negative entries  & Left hippocampi & Right  hippocampi\\
\hline
$\widehat{\bm{B}}^{(0.5)}$ & 15,000 & 15,000\\
$\widehat{\bm{B}}^{(1)}$ & 15,000 & 15,000\\
$\widehat{\bm{B}}^{(2)}$ & 14,600 & 15,000\\
$\widetilde{\bm B}^{(1)}$  & 15,000 & 15,000\\
\end{tabular}
\caption{ Number of negative entries of $\widehat{\bm{B}}$ for left and right hippocampi. }
\label{sensana2}
\end{table}

To summarize, the estimate $\widehat{\bm B}$ is not very sensitive against the choices of $|\widehat{\mathcal{M}}_{1}^{*}|$ and $|\widehat{\mathcal{M}}_{2}|$ except for the left hippocampi when $|\widehat{\mathcal{M}}_{2}|/|\widehat{\mathcal{M}}_{1}^{*}|=2$ and $|\widehat{\mathcal{M}}| = 89$. In fact, as shown in Table \ref{sensana2}, when $|\widehat{\mathcal{M}}_{2}|/|\widehat{\mathcal{M}}_{1}^{*}|=2$, for the left hippocampi, there are $400$ entries are non-negative. We believe it may be due to some confounder variables being missed in the screening step. For instance, we find that rs157582, a previously identified  risk loci for Alzheimer's disease \citep{guo2019genome} is adjusted in estimating $\widehat{\bm B}^{(0.5)}$, $\widehat{\bm B}^{(1)}$ and $\widetilde{\bm B}^{(1)}$ except for $\widehat{\bm B}^{(2)}$ . But in general, as demonstrated in Figure \ref{dataPlots2}, the estimate $\widehat{\bm B}$s are similar among different choices of $|\widehat{\mathcal{M}}_{1}^{*}|$ and $|\widehat{\mathcal{M}}_{2}|$ .

\section{Subgroup analysis ADNI data applications}
\label{Subgroup analysis ADNI data applications}
We repeat the analysis on the $391$ MCI and AD subjects. The estimates $\widehat{\bm B}$ corresponding to each part of the hippocampus onto a representative hippocampal surface are plotted in Figure \ref{dataPlotsv12}(a). We have also plotted the hippocampal subfield \citep{apostolova20063d} in Figure \ref{dataPlotsv12}(b).
The results are similar to the complete data analysis including all the $566$ subjects. For example, from these plots, we can see that $13,700$ entries of $\widehat{\bm{B}}$ corresponding to left and all the $15,000$ entries of $\widehat{\bm{B}}$ corresponding to right hippocampi are negative. This implies that the radial distances of each pixel of both hippocampi are mostly negatively associated with the ADAS-13 score, which depicts the severity of behavioral deficits. Furthermore, the subfields with the strongest associations are still mostly CA1 and subiculum.

\begin{figure}[htbp]
\captionsetup[subfigure]{justification=centering}
\centering
 \subcaptionbox{}[0.45\linewidth]
{\includegraphics[height=3in,width=4in]{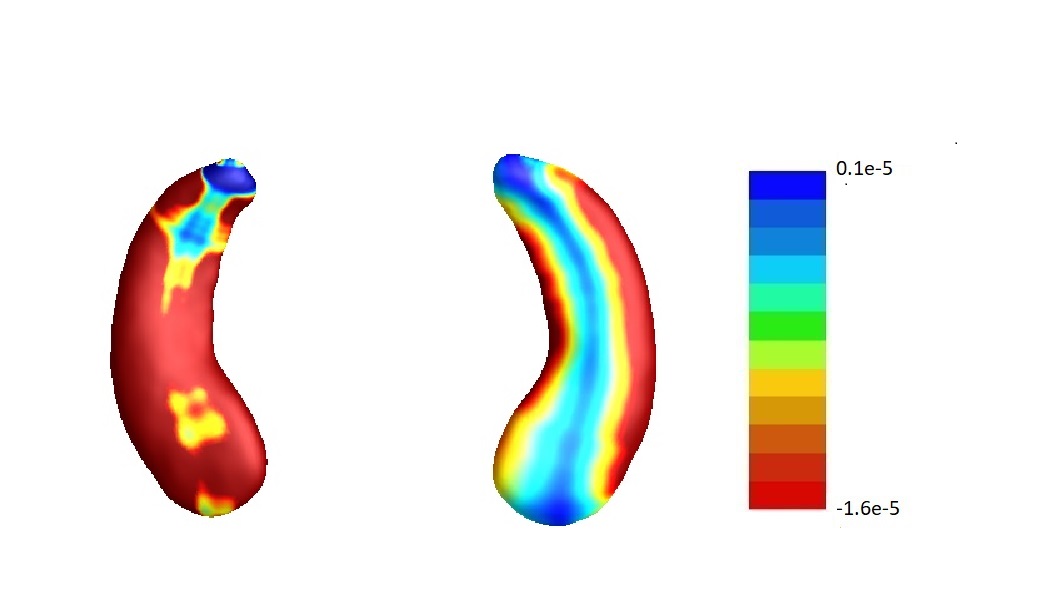}}
 \hfill
 \subcaptionbox{}[0.45\linewidth]
 {\includegraphics[height=2.4in,width=2in]{./plotsMainArkSupp/hippocampalsubfield.pdf}}
 \hfill
\caption{Real data results for MCI and AD subgroup: Panel (a) plots the estimate $\widehat{\bm B}$ corresponding to the left hippocampi (left part) and  the right hippocampi (right part). Panel (b) plots the hippocampal subfield.}
\label{dataPlotsv12}
\end{figure}

\section{Results for mediation analyses}
\label{Results for mediation analyses}
We perform the SNP-imaging-outcome mediation analyses following the same procedure as in \cite{bi2017genome}. In specific, we regress the $30,000$ imaging measures against $6,087,205$ SNPs in the first step to search for the pairs of intermediate imaging measures and genetic variants. Then the behaviour outcome is fit against each candidate genetic variant to identify direct and significant influence. In  
the last step, the behaviour outcome is fit against identified genetic variant and its associated intermediate imaging measure simultaneously. A mediation relationship is built if a) the genetic variant is significant in both first and second steps, b) the intermediate imaging measure is significant in the last step, and c) the genetic variant has a smaller coefficient in the last step compared with the second step. Note that the total effect of the genetic variant in the second step should be the summation of direct and indirect effects which motivates the criterion c) of coefficient comparison. Note that, the total effect may not always be greater than the direct effect in the last step when the direct and indirect effects have opposite signs, while the causal inference tool proposed in this paper does not have this problem. 
Similar to \citet{bi2017genome}, we try to identify the pairs of SNP and imaging measure, for which the direct effect of SNP on behavioral outcome, the effect of SNP on imaging measure and the effect of imaging measure on the behavioral outcome are all significant. However, there is no SNP with at least one paired imaging measure (i.e. hippocampal imaging pixel) being significant. Therefore,
there is no evidence for the mediating relationship of  SNP-imaging-outcome from our analysis.

\section{Additional results for simulation studies}
\label{addsimulation}
In this section, we list additional simulation results. In particular, Figures \ref{sim1step1n200sigma025} -- \ref{sim1step1n1000sigma025} present the screening results for Section \ref{Simulation for screening} with $(n,s,\sigma) = (200,5000,0.5)$, $(500,5000,1)$, $(500,5000,0.5)$, $(1000,5000,1)$ and $(1000,5000,0.5)$, respectively. 
Section \ref{Sensitivity and specificity analysis of simulation} presents the sensitivity and specificity analyses for Section \ref{Simulation for estimation}, where the detailed definitions of sensitivity and specificity can be found here.
Section \ref{Screening under different sparsity levels} presents an additional simulation study considering various sparsity levels of instrumental variables.
Section \ref{Screening under different covariance structures of exposure errors} presents an additional simulation study considering different covariances of exposure errors. 
 Section \ref{Screening and estimation under different sizes} conducts an additional simulation study by varying the relative sizes of $\widehat{\mathcal{M}}_{1}^{*}$ and $\widehat{\mathcal{M}}_{2}$. 
 Section \ref{Group Selection Using Linkage Disequilibrium Information} lists additional screening and estimation results for Section \ref{Group Selection Using Linkage Disequilibrium Information main} of the main article.

\begin{figure}[htbp]
\captionsetup[subfigure]{justification=centering}
\centering
 \subcaptionbox{Confounder: strong \\ outcome, weak exposure}[0.45\linewidth]
 {\includegraphics[width=6cm,height=3.5cm]{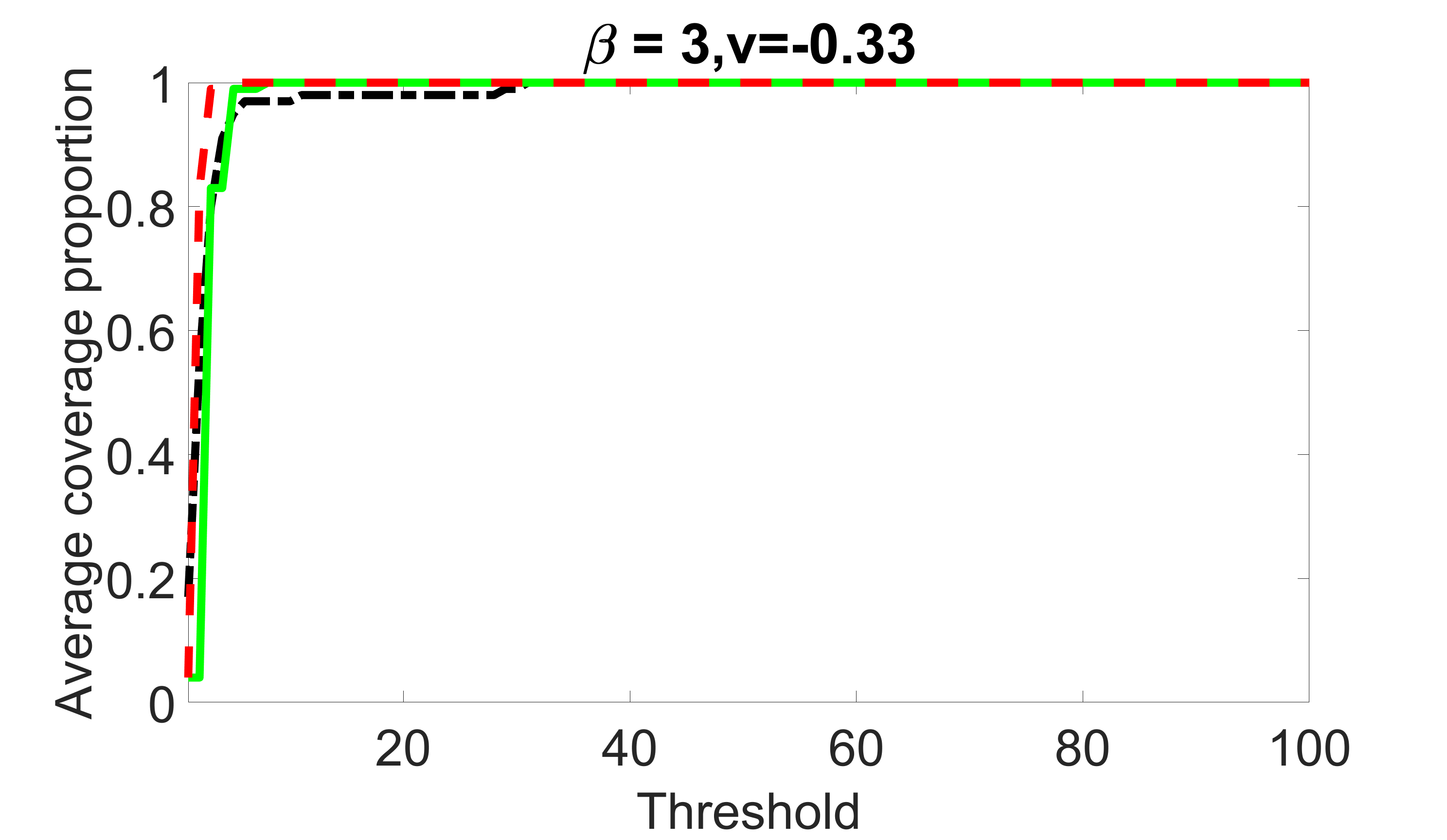}}
 \subcaptionbox{Confounder: medium \\ outcome, medium exposure}[0.45\linewidth]
 {\includegraphics[width=6cm,height=3.5cm]{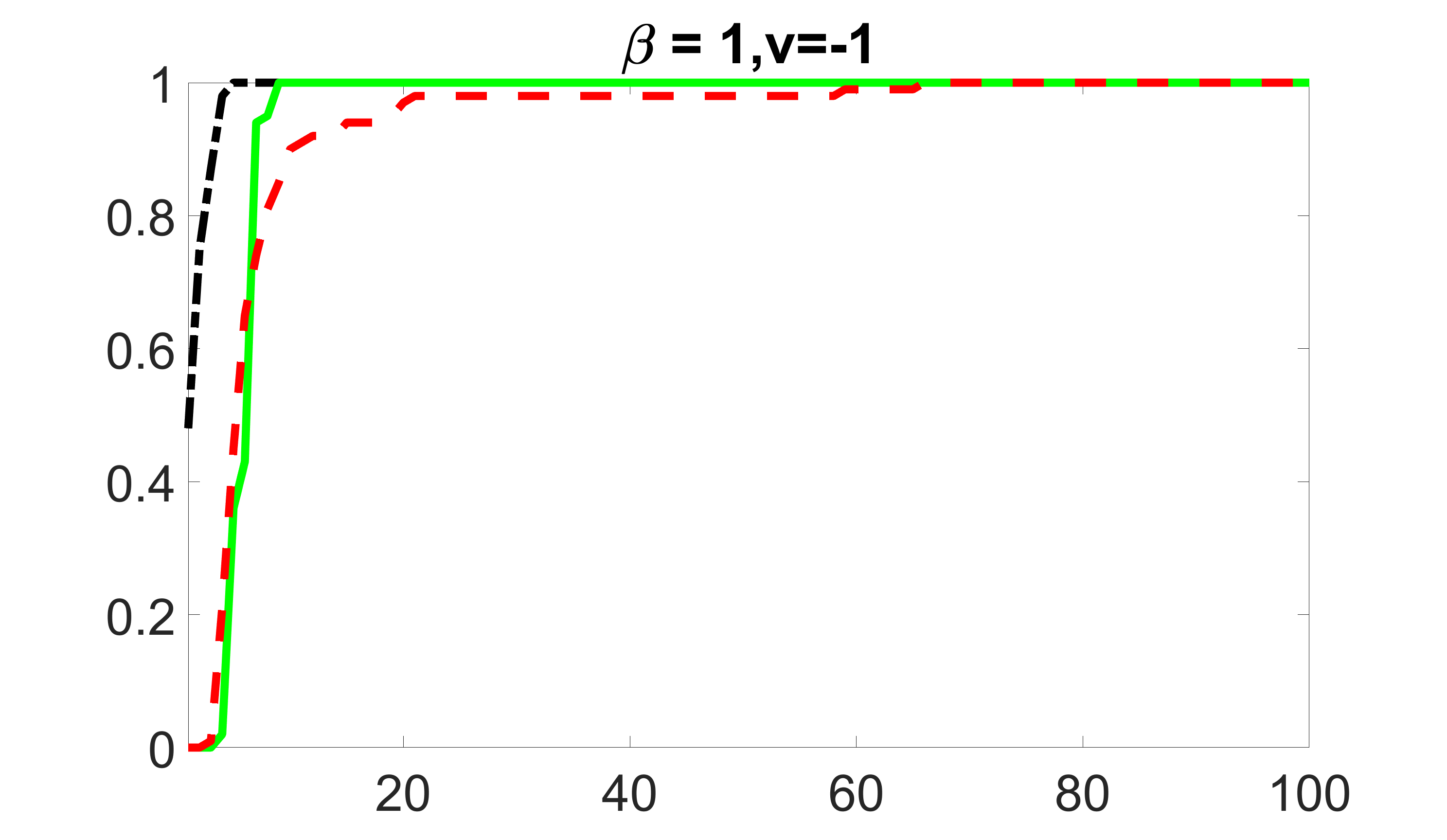}}
  \subcaptionbox{Confounder: weak \\ outcome, strong exposure}[0.45\linewidth]
 {\includegraphics[width=6cm,height=3.5cm]{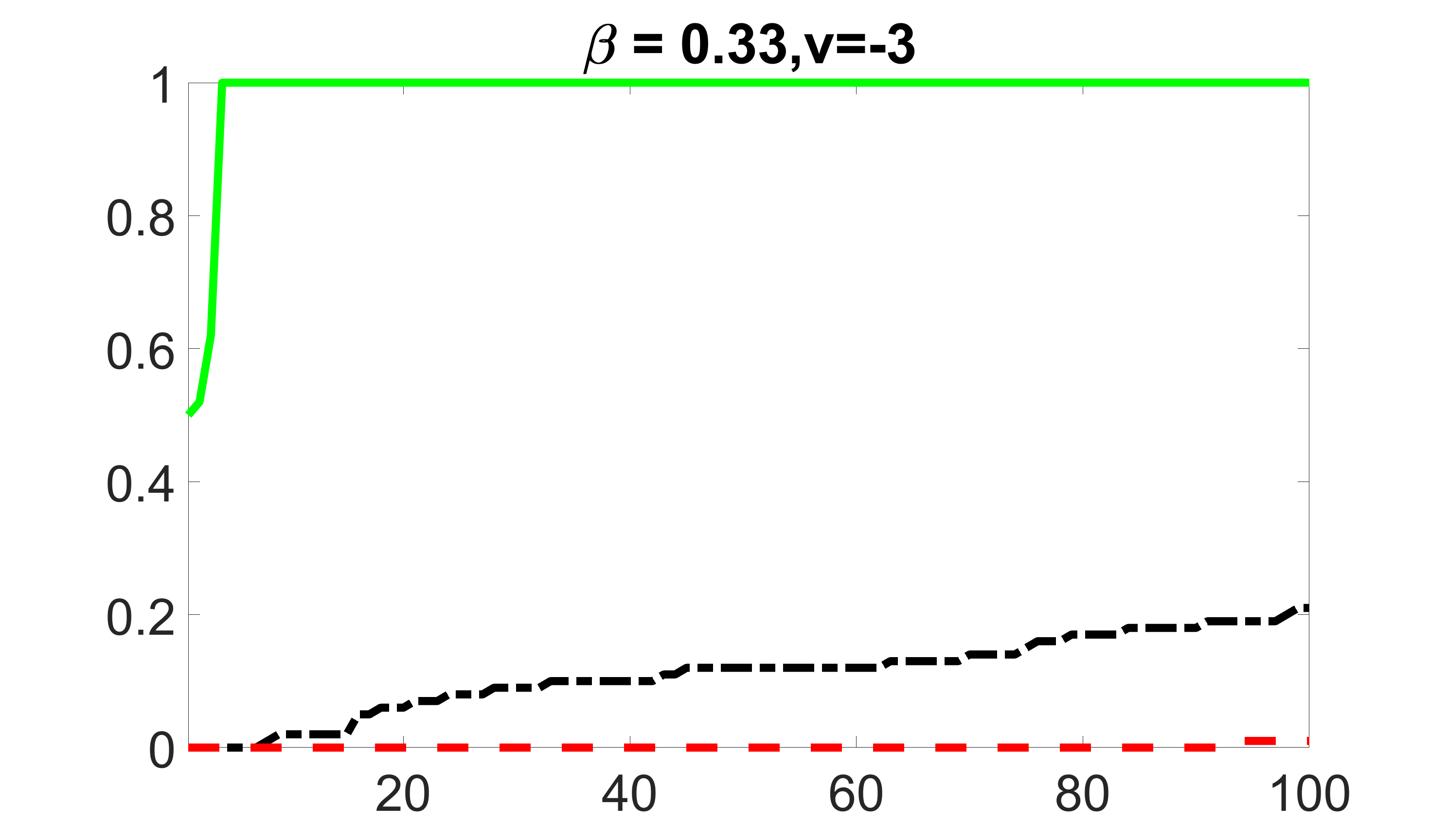}}
  \subcaptionbox{Precision: strong \\ outcome, zero exposure}[0.45\linewidth]
 {\includegraphics[width=6cm,height=3.5cm]{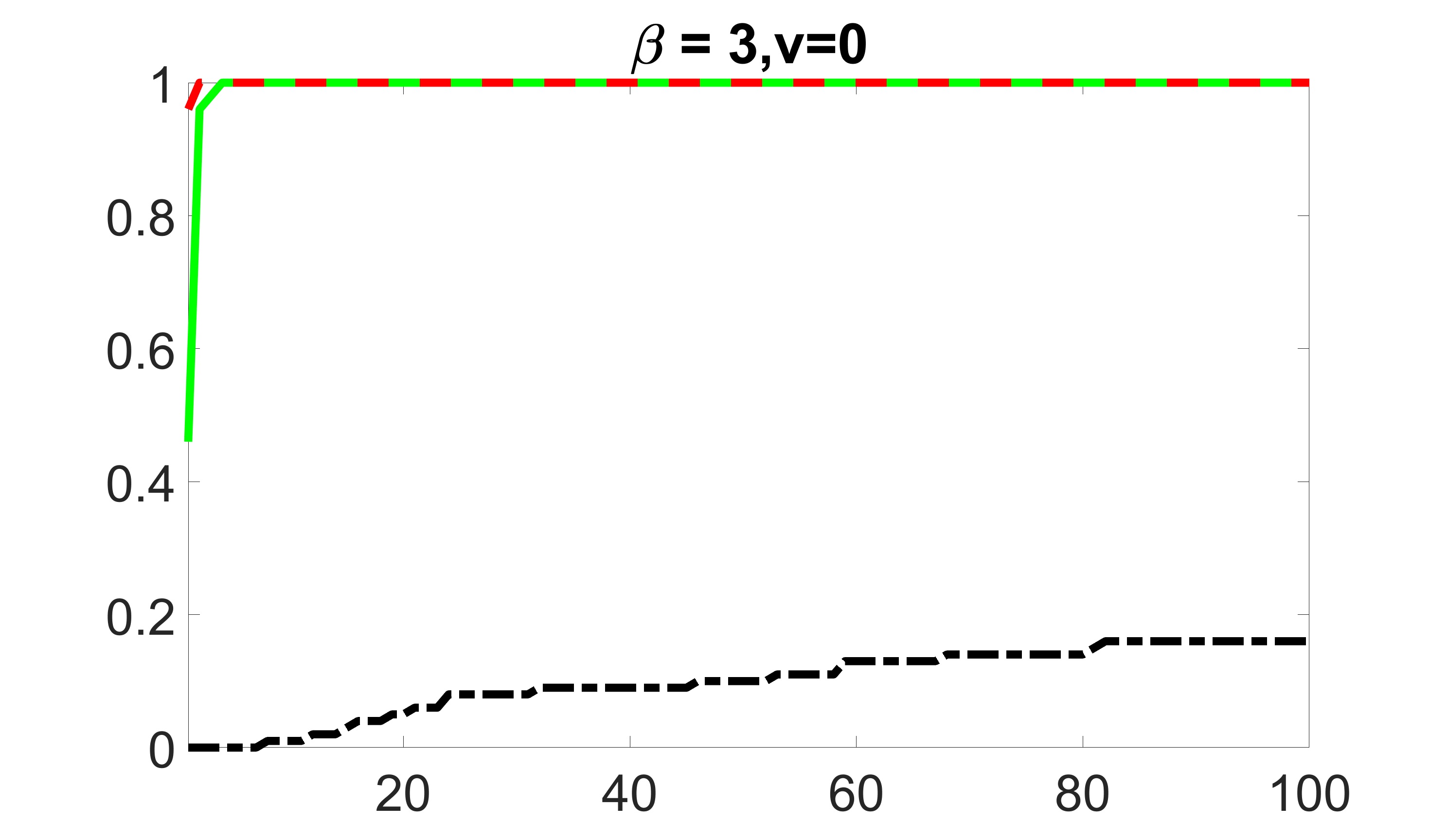}}
  \subcaptionbox{Precision: medium \\ outcome, zero exposure}[0.45\linewidth]
 {\includegraphics[width=6cm,height=3.5cm]{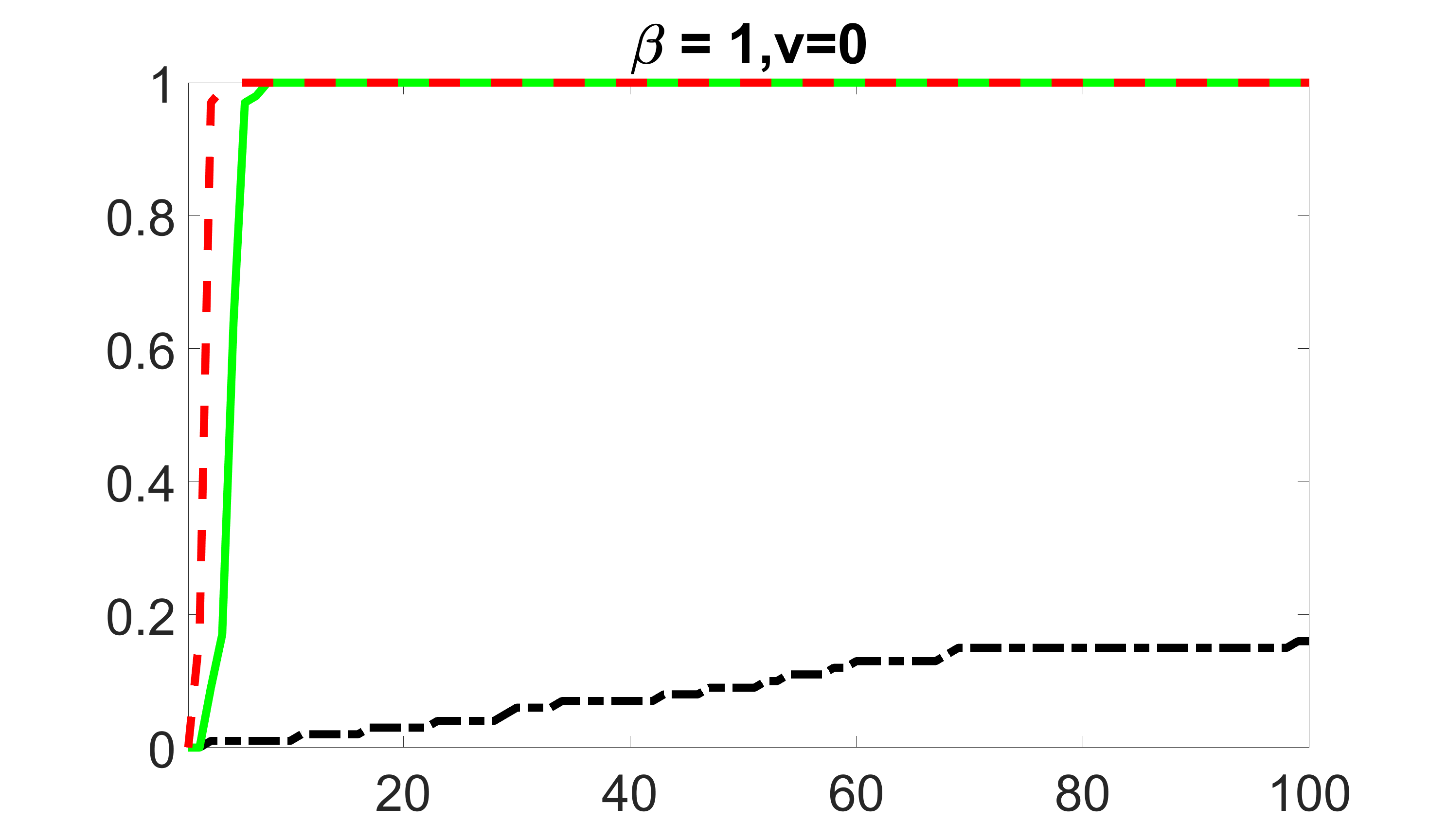}}
  \subcaptionbox{Precision: weak \\ outcome, zero exposure}[0.45\linewidth]
 {\includegraphics[width=6cm,height=3.5cm]{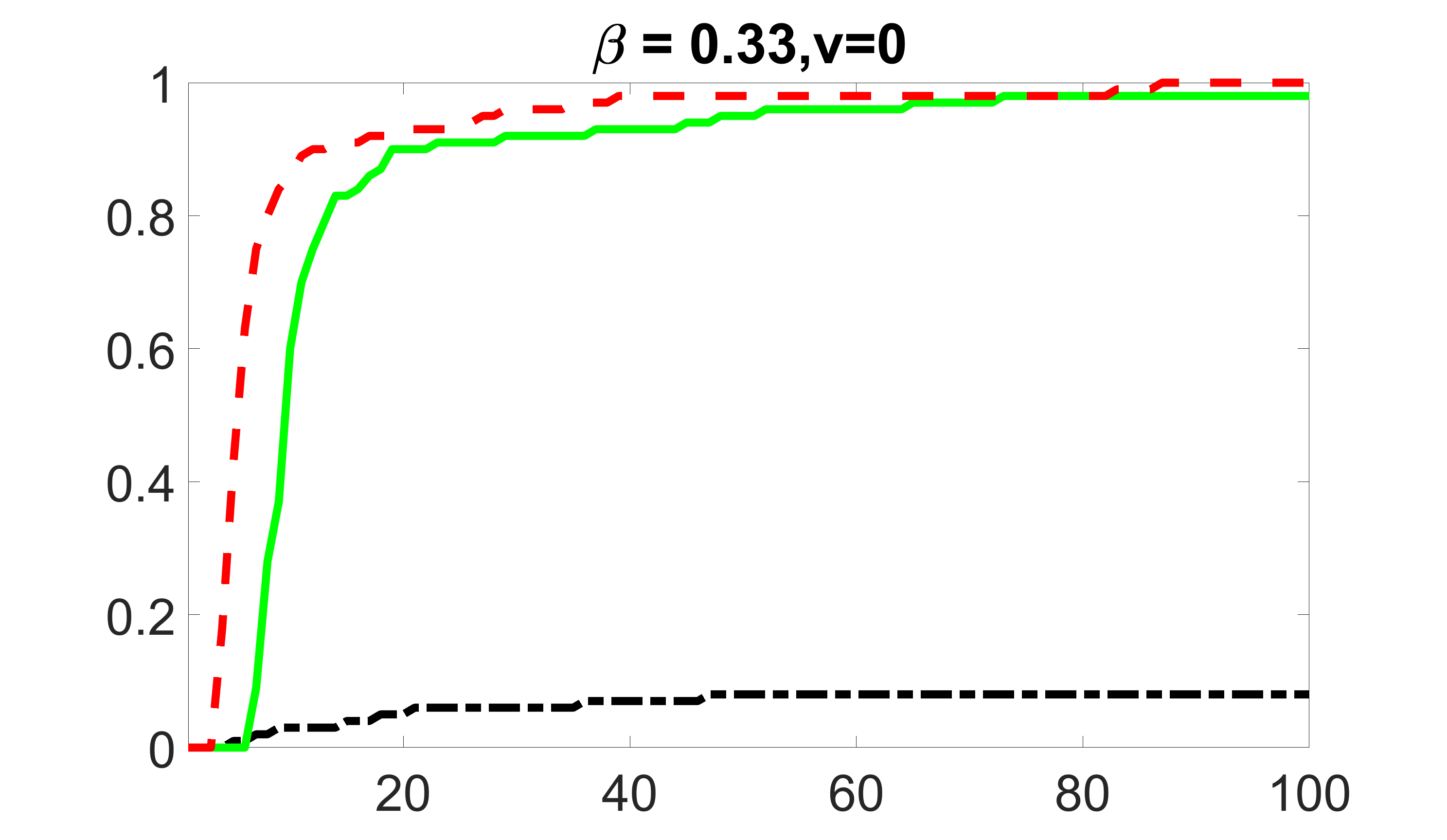}}
  \subcaptionbox{Overall coverage of $\mathcal{M}_1$}[0.45\linewidth]
 {\includegraphics[width=6cm,height=3.5cm]{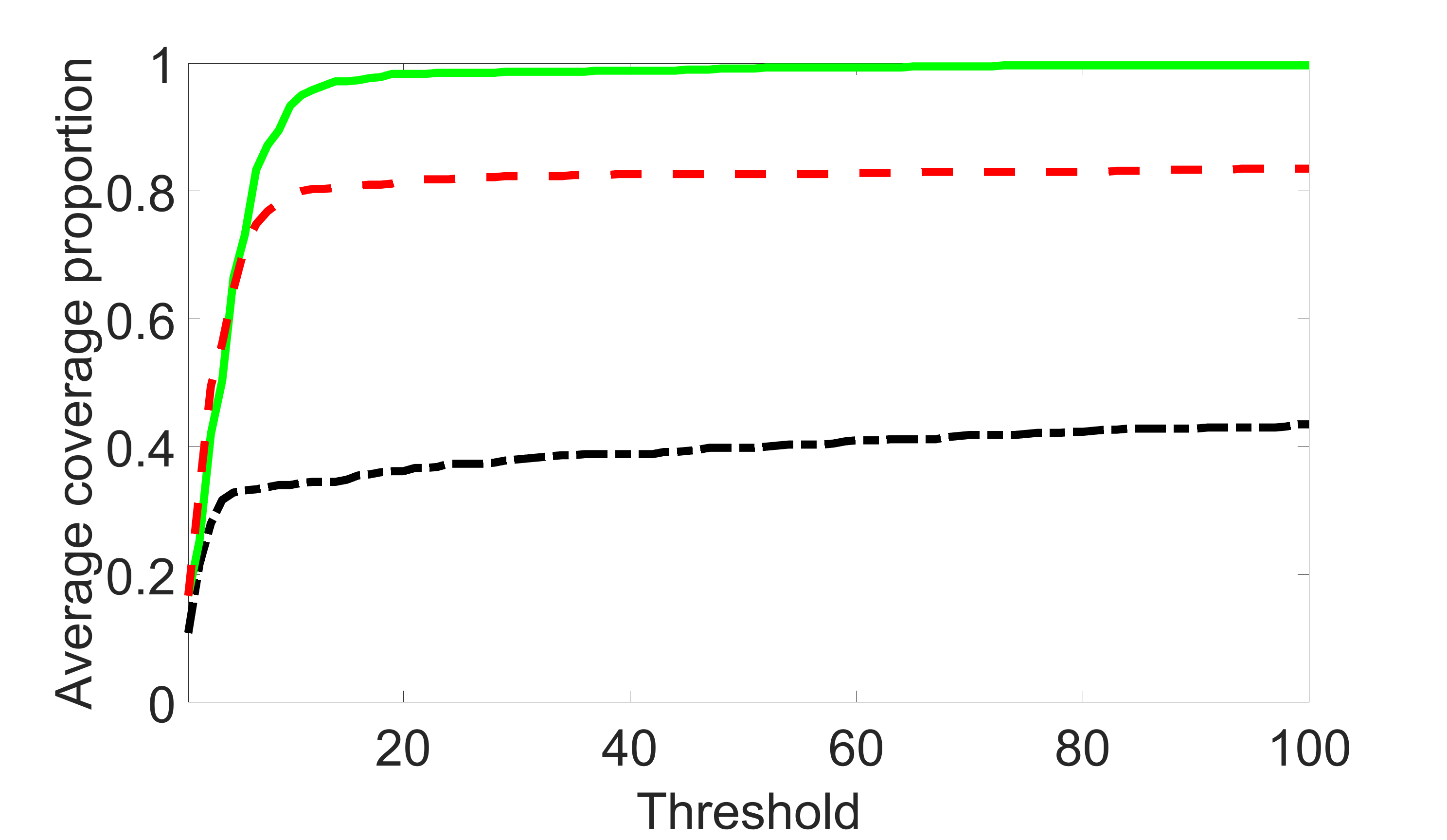}}
\caption{Simulation results for the case $(n,s,\sigma) = (200,5000,0.5)$: Panels (a) -- (f) plot the average coverage proportion for $X_l$, where $l=1,2,3,104,105$ and $106$. Panels (a) -- (c) correspond to strong outcome and weak exposure predictor, moderate outcome and moderate exposure predictor and weak outcome and strong exposure predictor; Panels (d) -- (f) correspond to strong, moderate and weak predictors of outcome only. Panel (g) plots the average coverage proportion for the index set $\mathcal{M}_1 = \{1,2,3,104,105,106\}$. The x-axis represents the size of $\widehat{\mathcal{M}} $, while
y-axis denotes the average proportion. The green solid, the red dashed and the black dash dotted lines denote our joint screening method, the outcome screening method, and the intersection screening method, respectively.}
\label{sim1step1n200sigma025}
\end{figure}

\begin{figure}[htbp]
\captionsetup[subfigure]{justification=centering}
\centering
 \subcaptionbox{Confounder: strong \\ outcome, weak exposure}[0.45\linewidth]
 {\includegraphics[width=6cm,height=3.5cm]{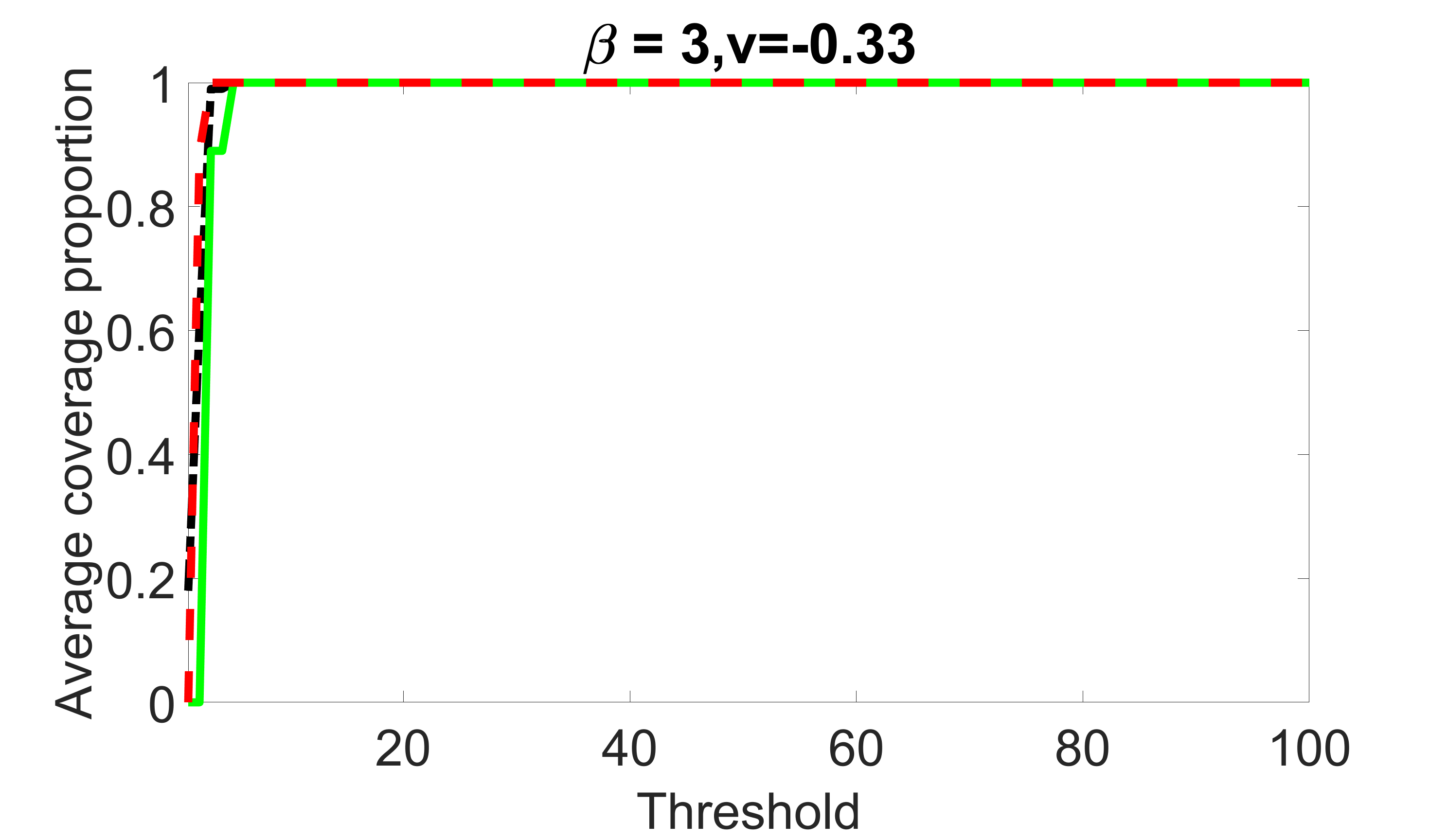}}
 \subcaptionbox{Confounder: medium \\ outcome, medium exposure}[0.45\linewidth]
 {\includegraphics[width=6cm,height=3.5cm]{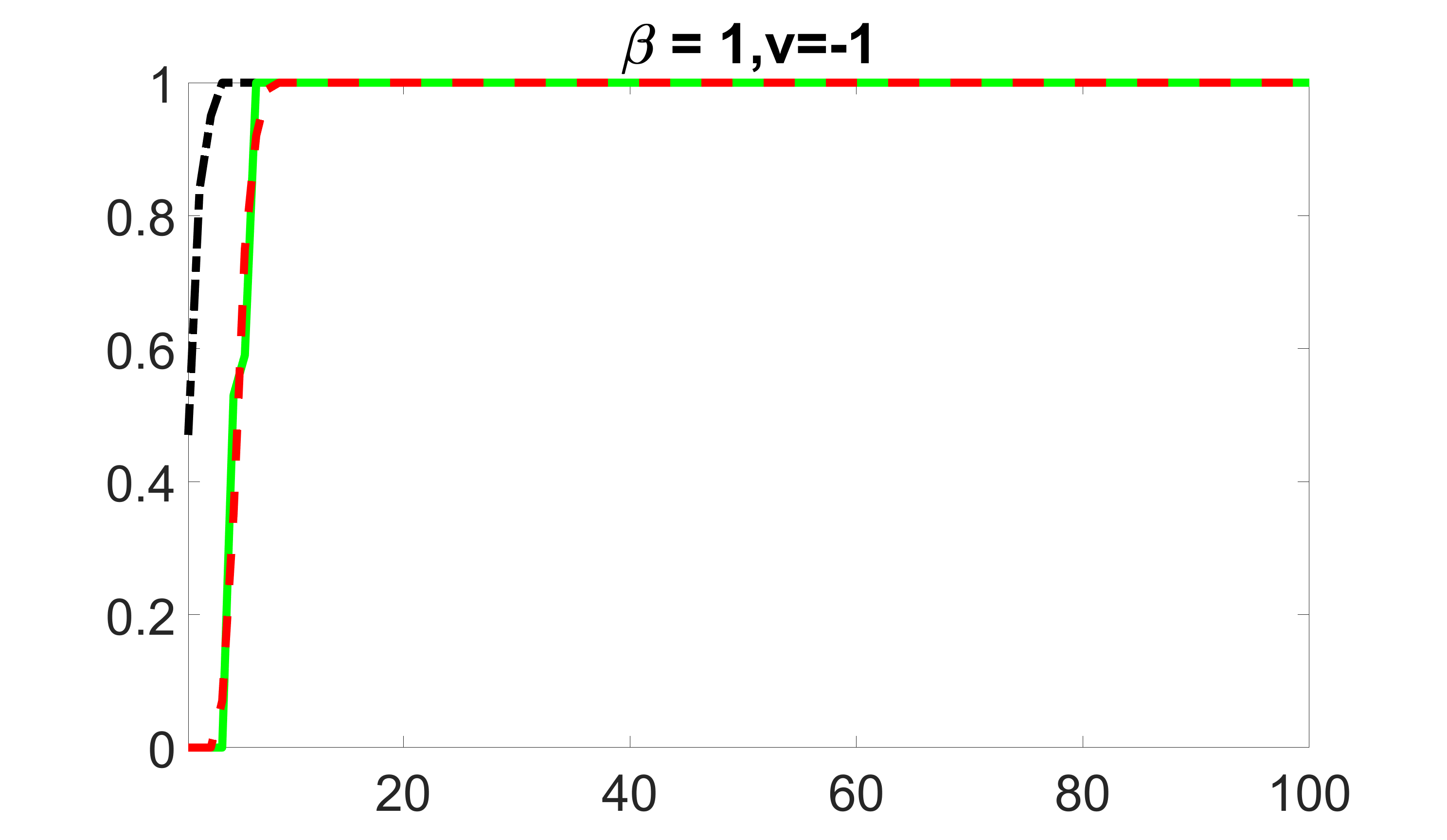}}
  \subcaptionbox{Confounder: weak \\ outcome, strong exposure}[0.45\linewidth]
 {\includegraphics[width=6cm,height=3.5cm]{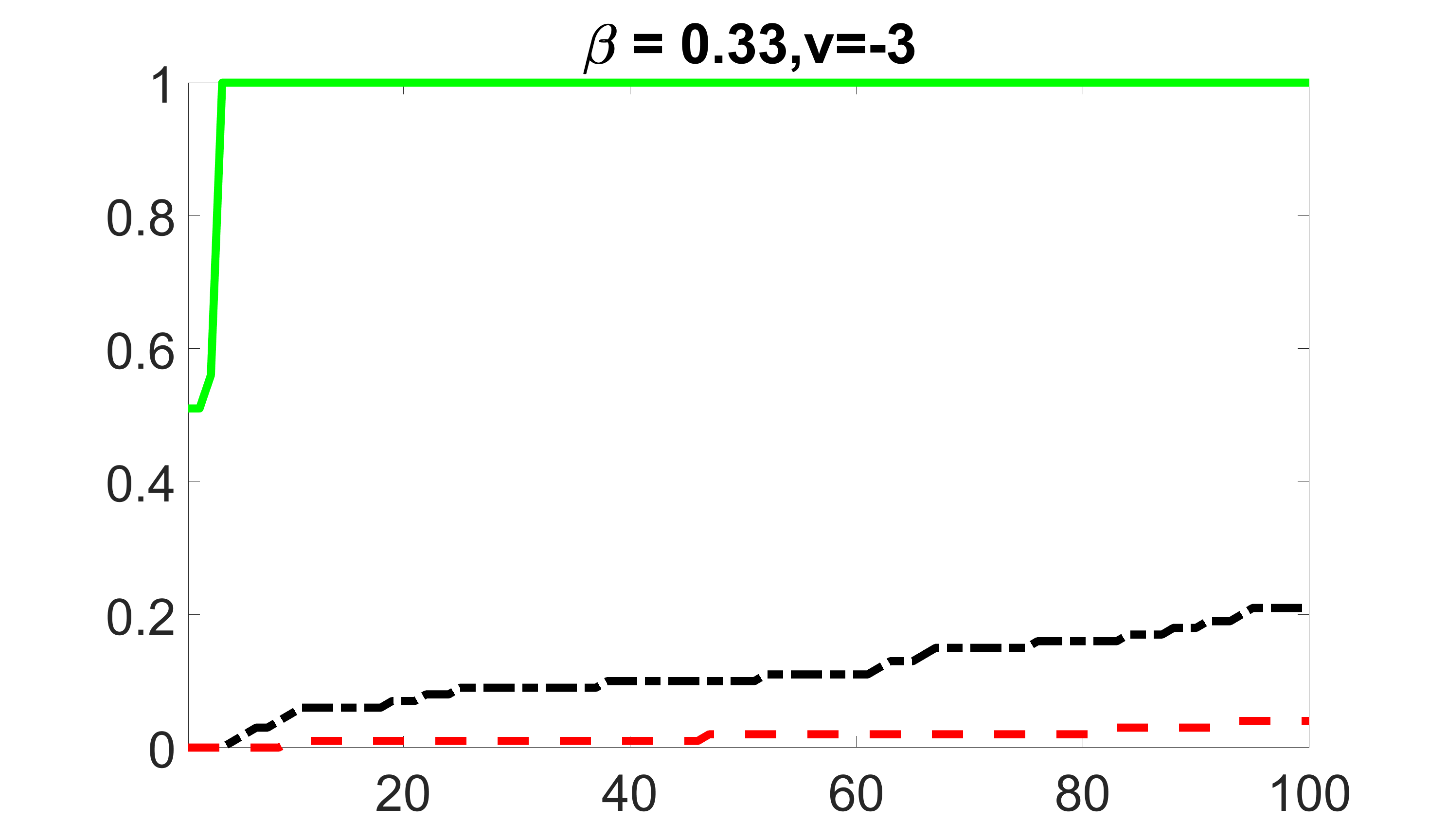}}
  \subcaptionbox{Precision: strong \\ outcome, zero exposure}[0.45\linewidth]
 {\includegraphics[width=6cm,height=3.5cm]{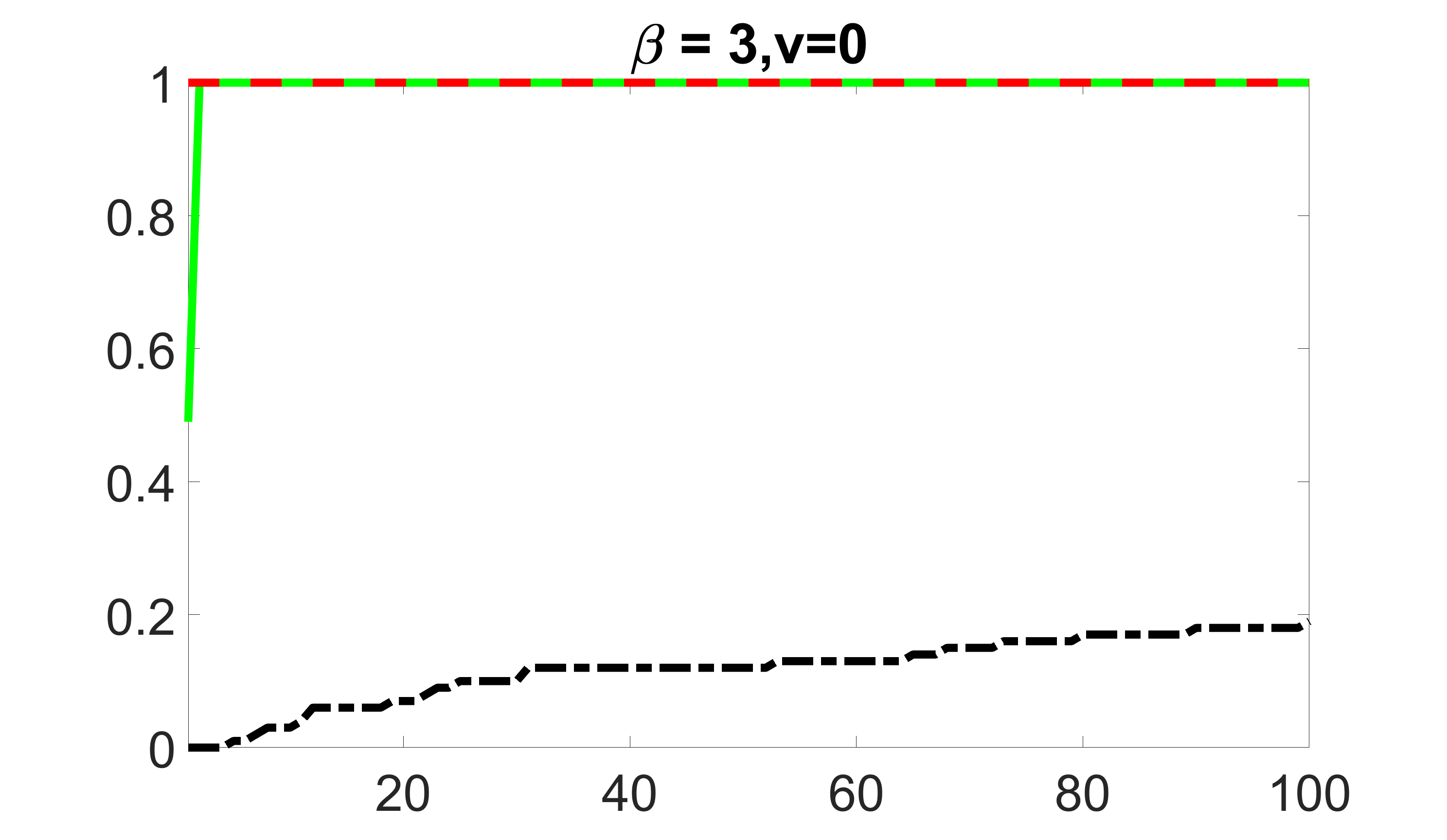}}
  \subcaptionbox{Precision: medium \\ outcome, zero exposure}[0.45\linewidth]
 {\includegraphics[width=6cm,height=3.5cm]{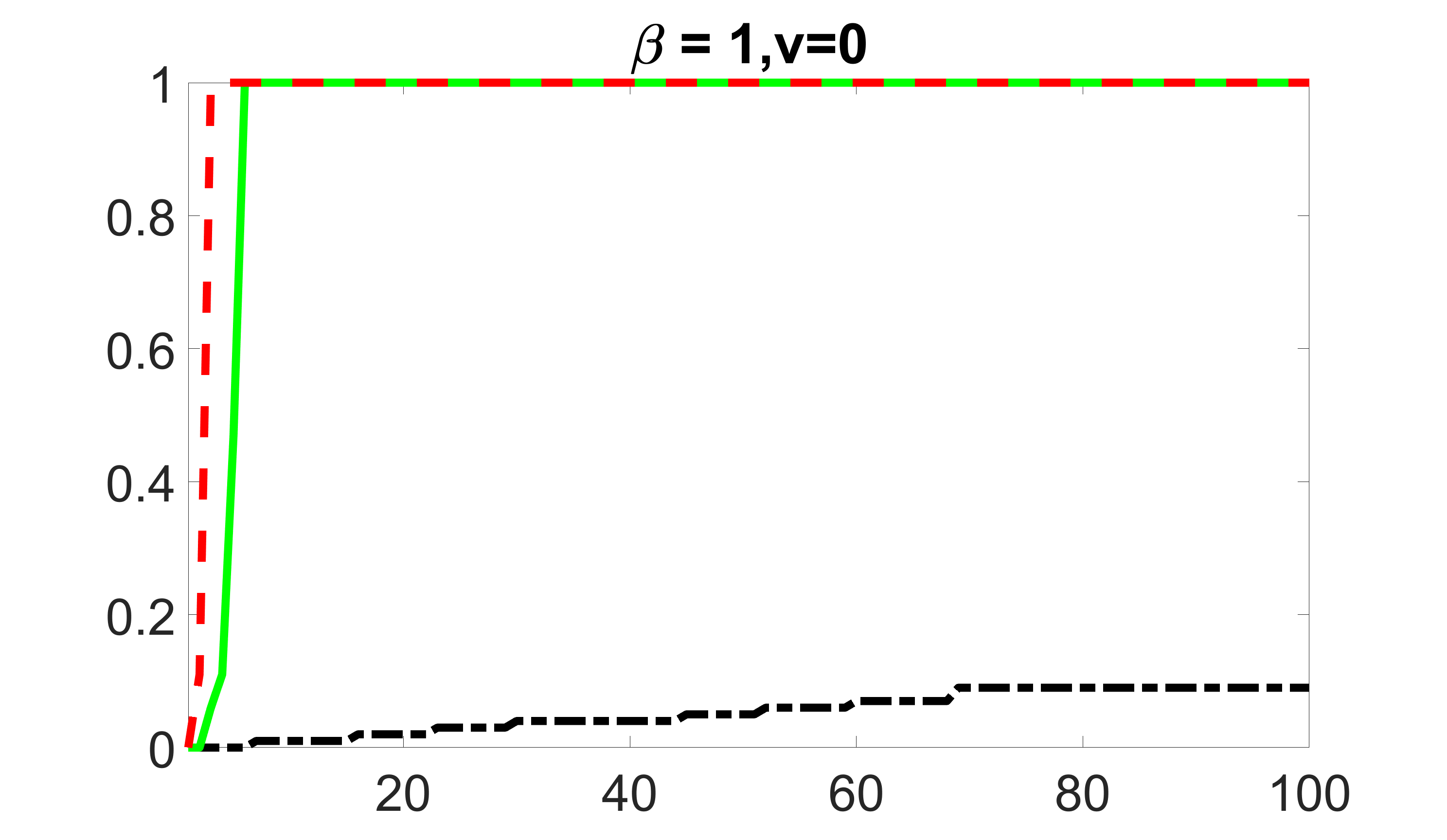}}
  \subcaptionbox{Precision: weak \\ outcome, zero exposure}[0.45\linewidth]
 {\includegraphics[width=6cm,height=3.5cm]{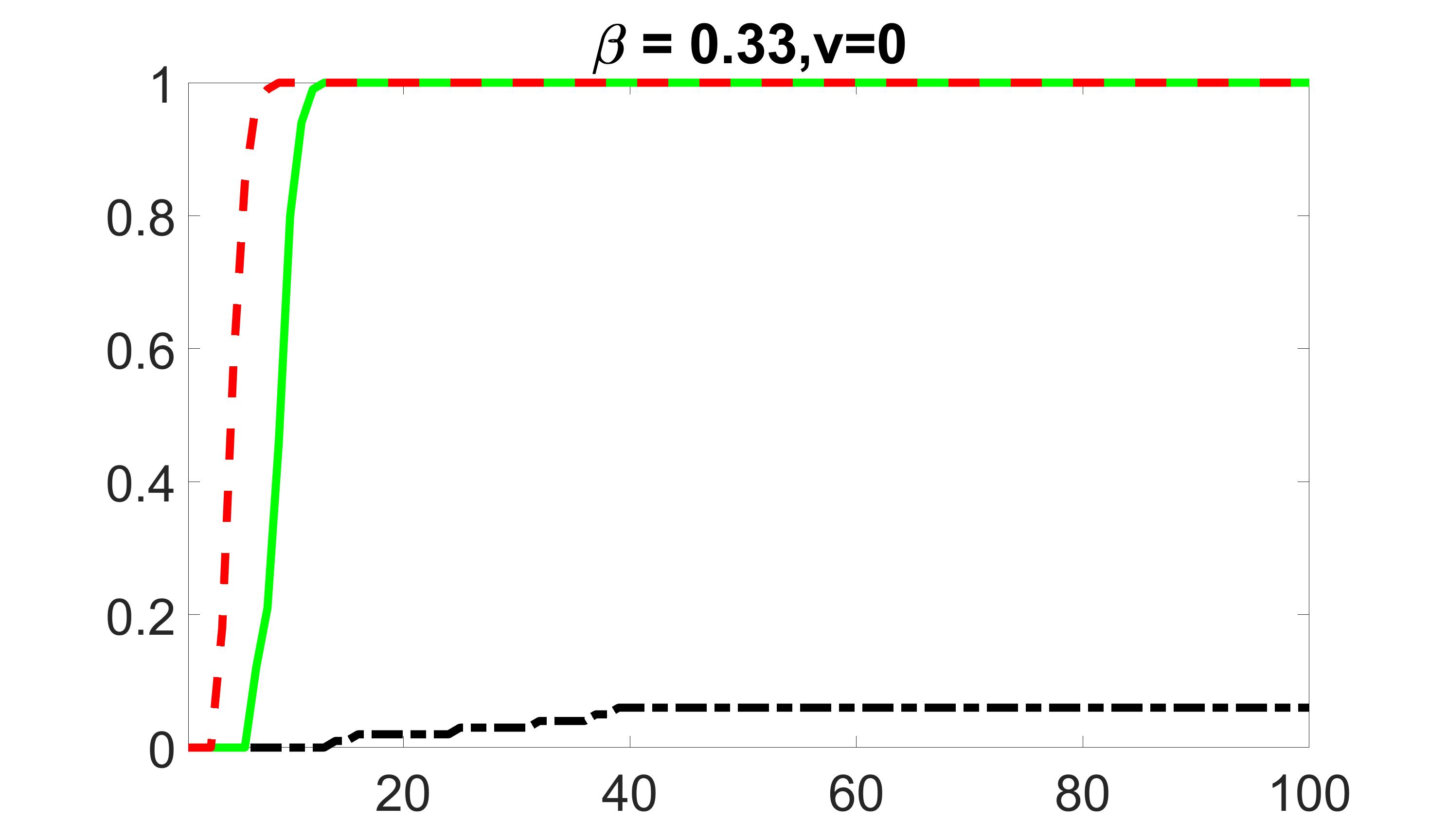}}
  \subcaptionbox{Overall coverage of $\mathcal{M}_1$}[0.45\linewidth]
 {\includegraphics[width=6cm,height=3.5cm]{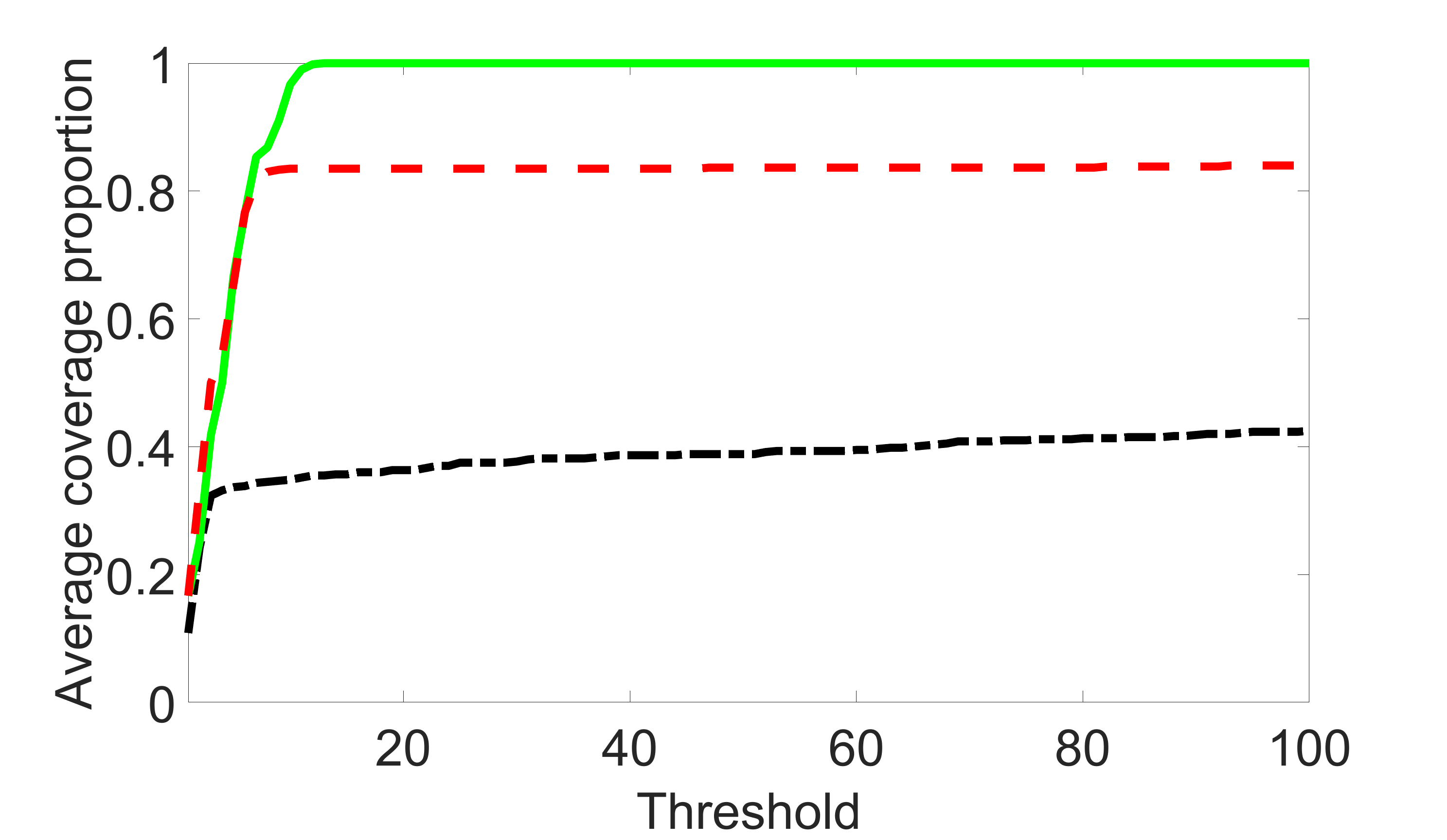}}
\caption{Simulation results for the case $(n,s,\sigma) = (500,5000,1)$: Panels (a) -- (f) plot the average coverage proportion for $X_l$, where $l=1,2,3,104,105$ and $106$. Panels (a) -- (c) correspond to strong outcome and weak exposure predictor, moderate outcome and moderate exposure predictor and weak outcome and strong exposure predictor; Panels (d) -- (f) correspond to strong, moderate and weak predictors of outcome only. Panel (g) plots the average coverage proportion for the index set $\mathcal{M}_1 = \{1,2,3,104,105,106\}$. The x-axis represents the size of $\widehat{\mathcal{M}} $, while
y-axis denotes the average proportion. The green solid, the red dashed and the black dash dotted lines denote our joint screening method, the outcome screening method, and the intersection screening method, respectively.}
\label{sim1step1n500sigma1}
\end{figure}

\begin{figure}[htbp]
\captionsetup[subfigure]{justification=centering}
\centering
 \subcaptionbox{Confounder: strong \\ outcome, weak exposure}[0.45\linewidth]
 {\includegraphics[width=6cm,height=3.5cm]{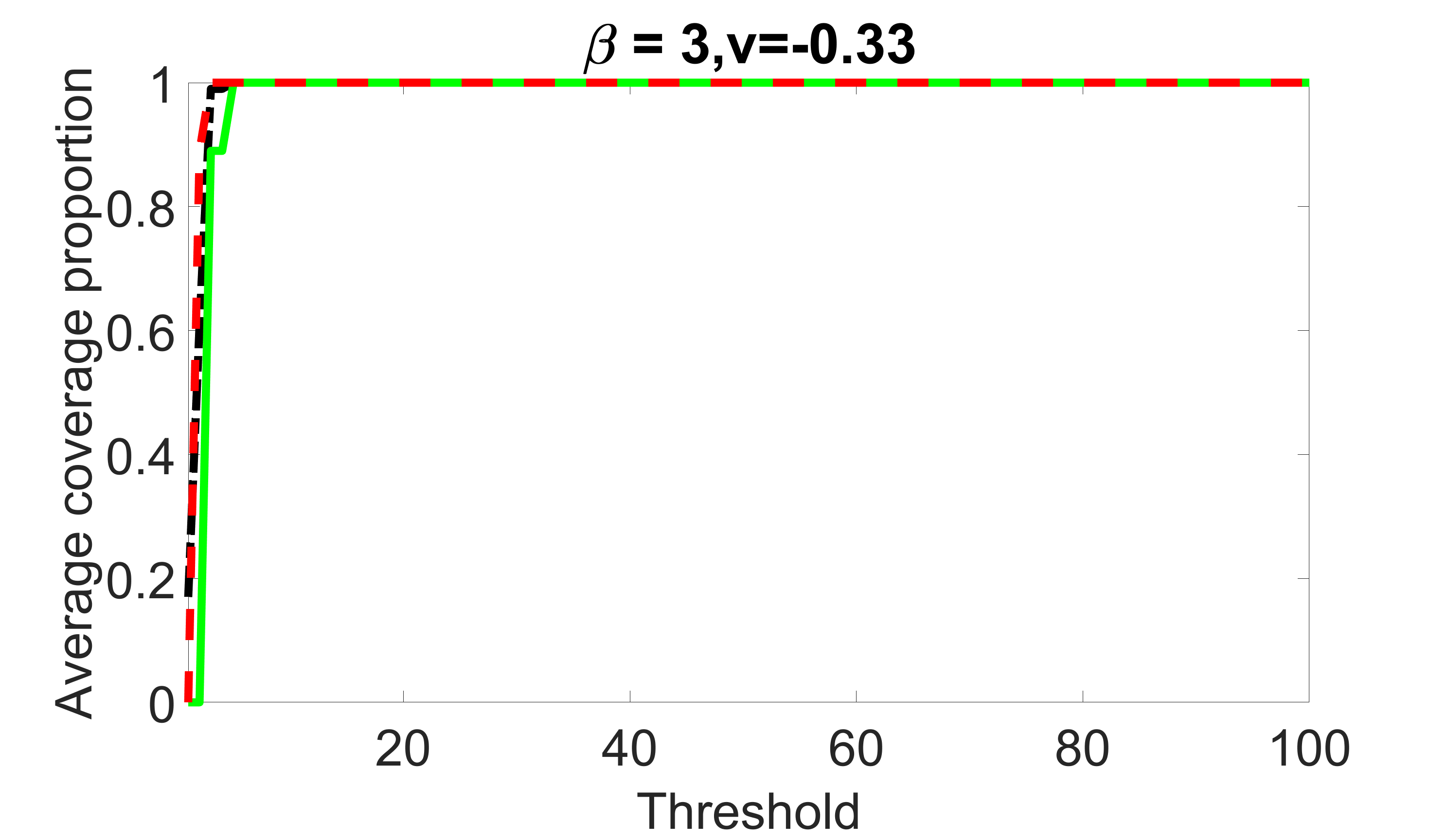}}
 \subcaptionbox{Confounder: medium \\ outcome, medium exposure}[0.45\linewidth]
 {\includegraphics[width=6cm,height=3.5cm]{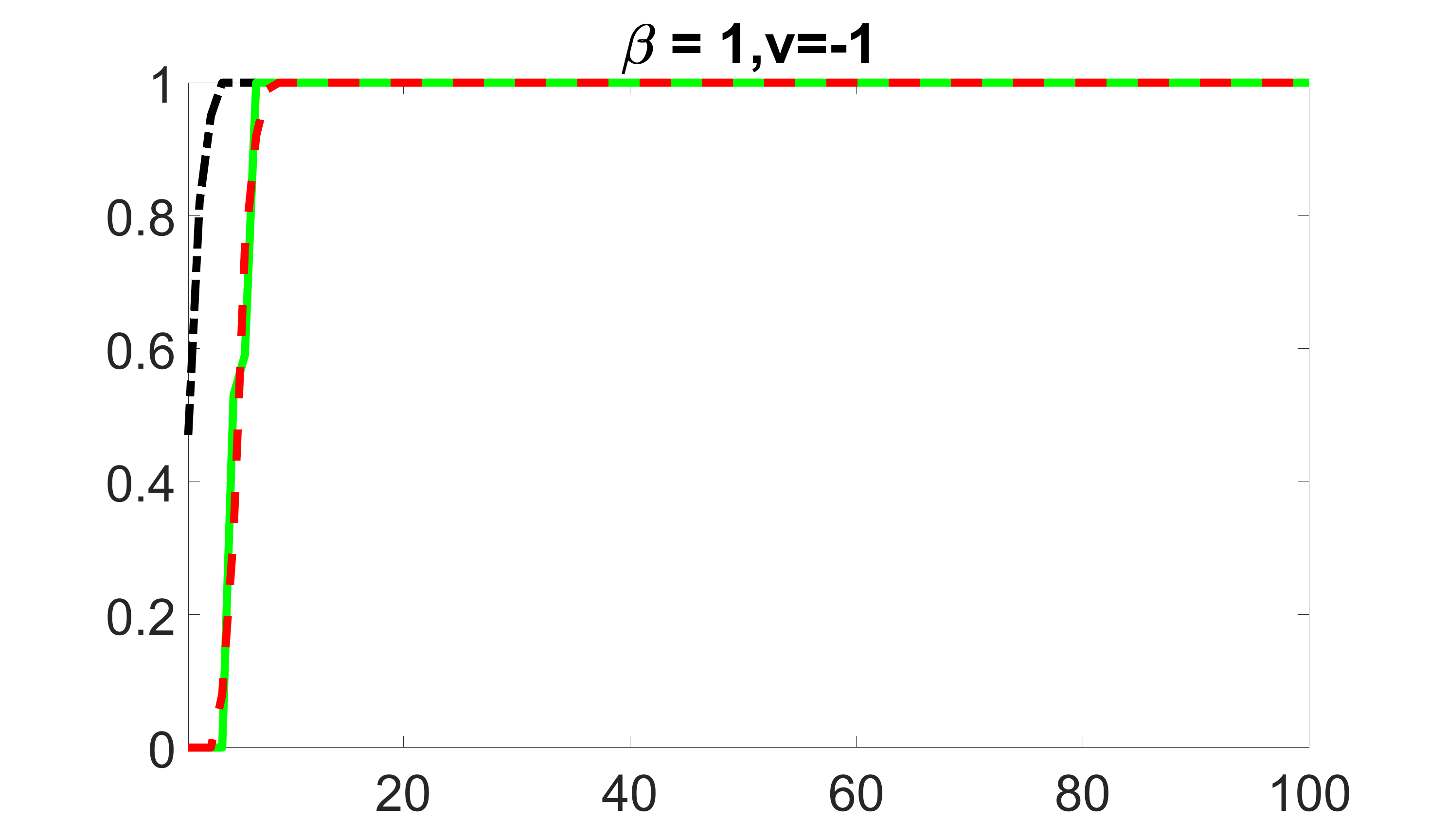}}
  \subcaptionbox{Confounder: weak \\ outcome, strong exposure}[0.45\linewidth]
 {\includegraphics[width=6cm,height=3.5cm]{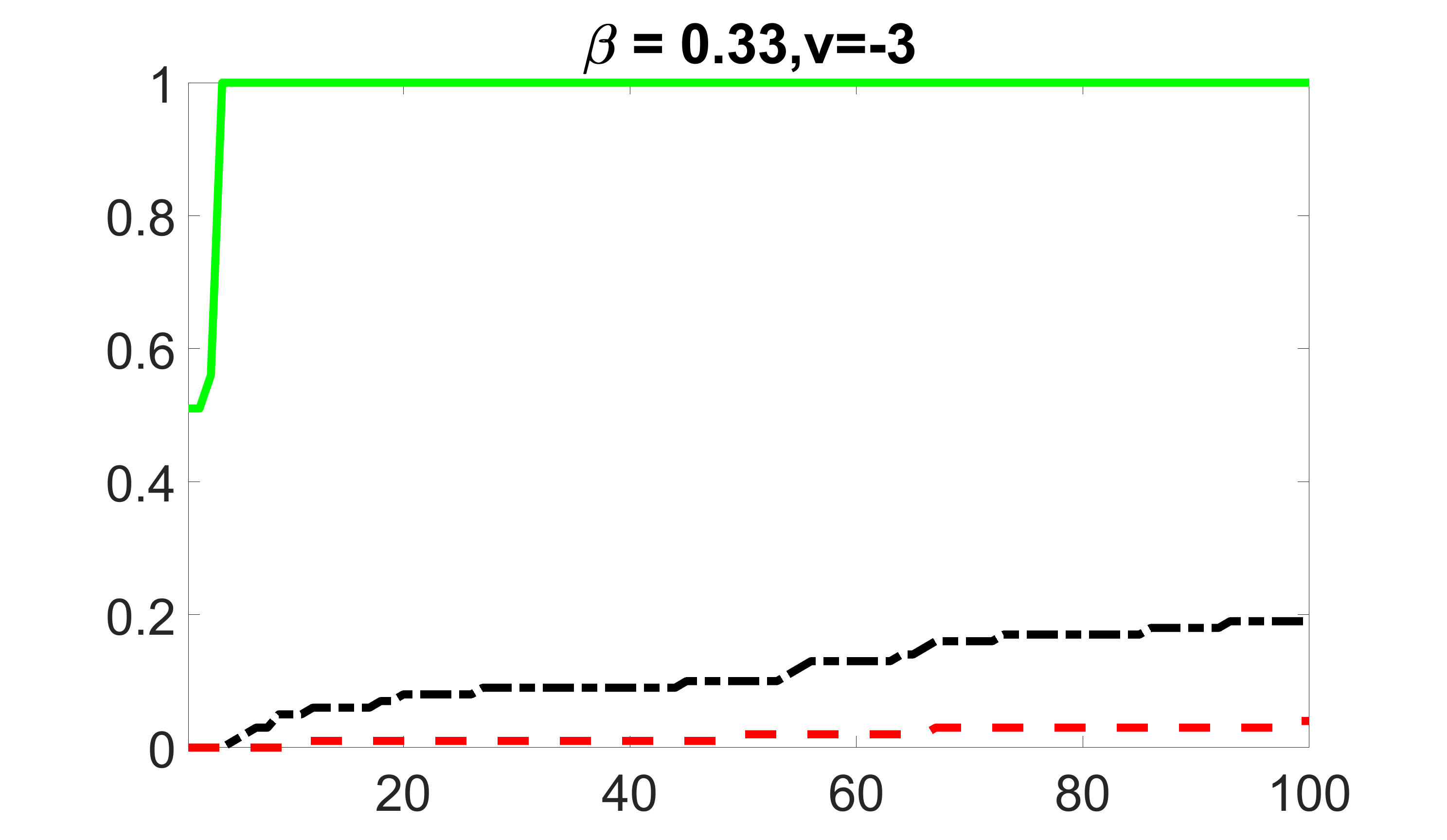}}
  \subcaptionbox{Precision: strong \\ outcome, zero exposure}[0.45\linewidth]
 {\includegraphics[width=6cm,height=3.5cm]{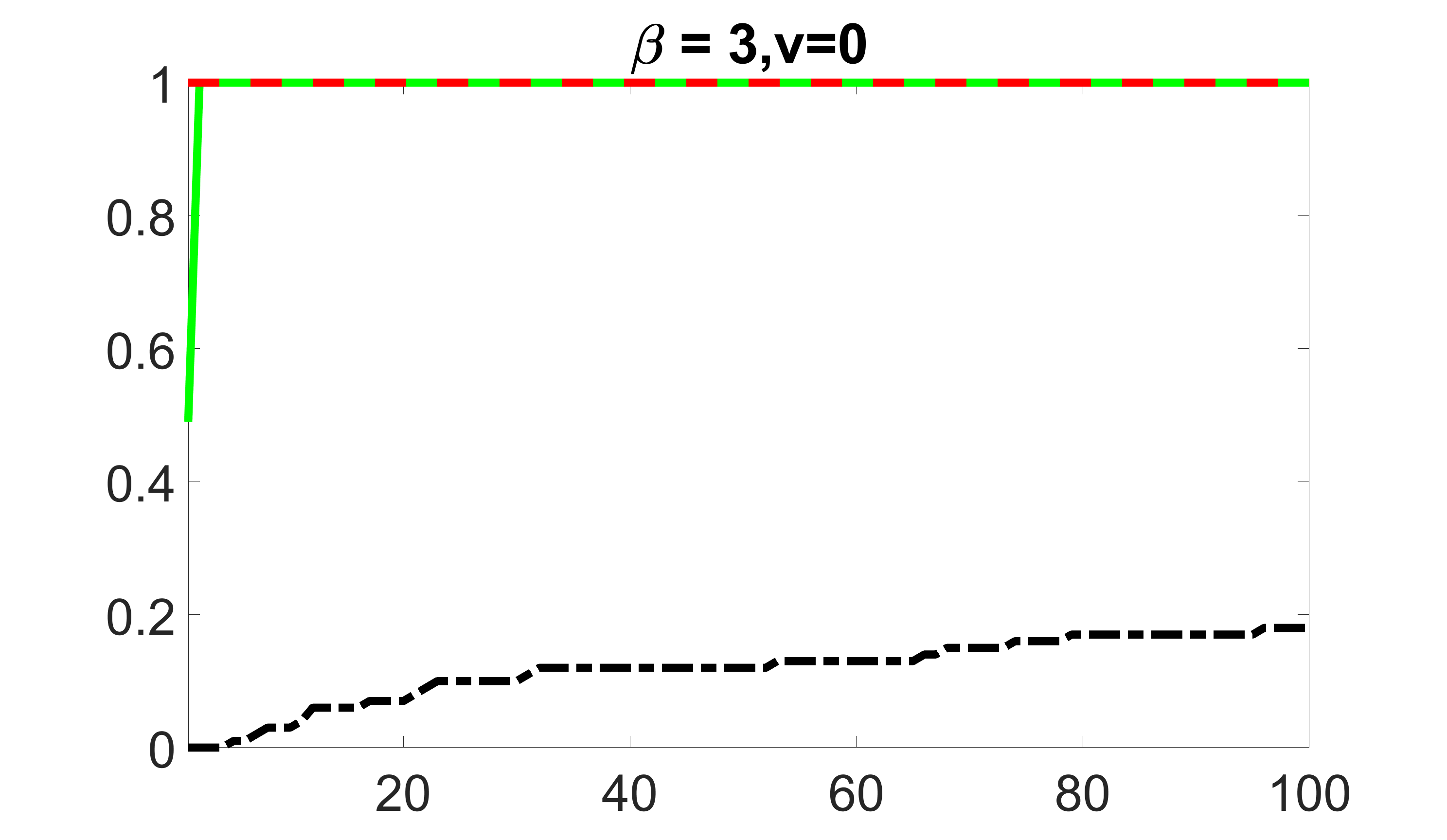}}
  \subcaptionbox{Precision: medium \\ outcome, zero exposure}[0.45\linewidth]
 {\includegraphics[width=6cm,height=3.5cm]{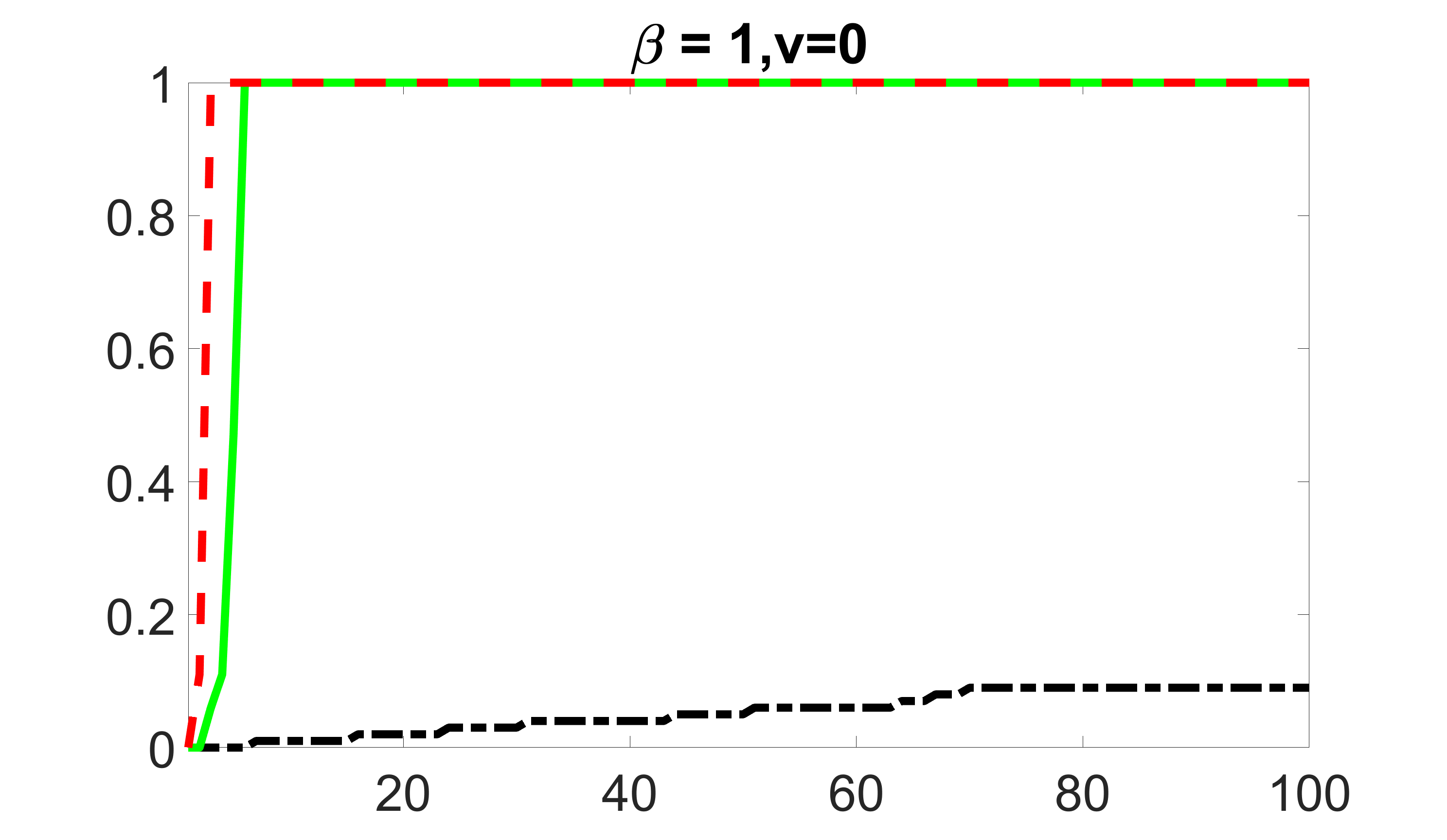}}
  \subcaptionbox{Precision: weak \\ outcome, zero exposure}[0.45\linewidth]
 {\includegraphics[width=6cm,height=3.5cm]{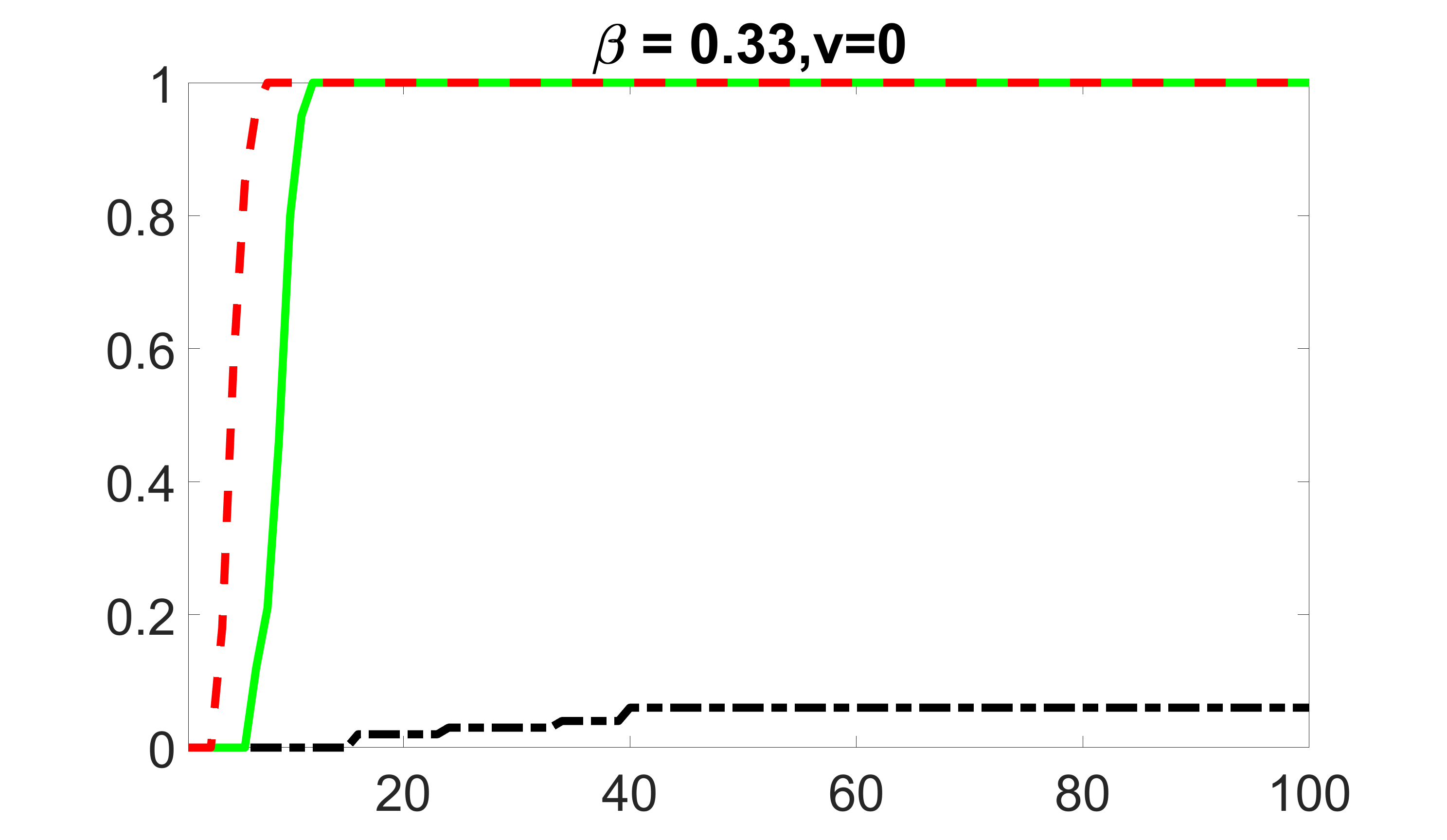}}
  \subcaptionbox{Overall coverage of $\mathcal{M}_1$}[0.45\linewidth]
 {\includegraphics[width=6cm,height=3.5cm]{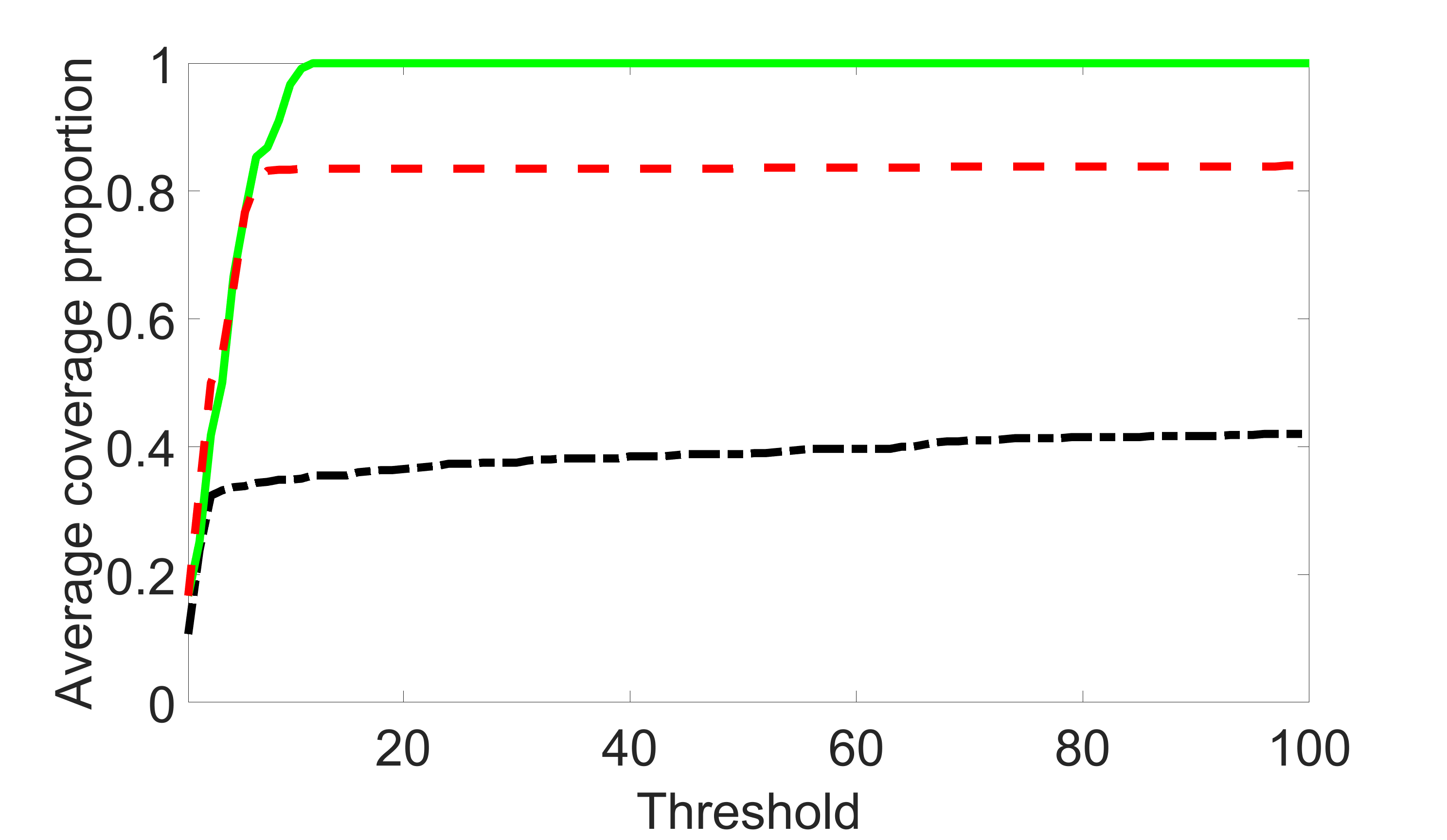}}
\caption{Simulation results for the case $(n,s,\sigma) = (500,5000,0.5)$: Panels (a) -- (f) plot the average coverage proportion for $X_l$, where $l=1,2,3,104,105$ and $106$. Panels (a) -- (c) correspond to strong outcome and weak exposure predictor, moderate outcome and moderate exposure predictor and weak outcome and strong exposure predictor; Panels (d) -- (f) correspond to strong, moderate and weak predictors of outcome only. Panel (g) plots the average coverage proportion for the index set $\mathcal{M}_1 = \{1,2,3,104,105,106\}$. The x-axis represents the size of $\widehat{\mathcal{M}} $, while
y-axis denotes the average proportion. The green solid, the red dashed and the black dash dotted lines denote our joint screening method, the outcome screening method, and the intersection screening method, respectively.}
\label{sim1step1n500sigma025}
\end{figure}

\begin{figure}[htbp]
\captionsetup[subfigure]{justification=centering}
\centering
 \subcaptionbox{Confounder: strong \\ outcome, weak exposure}[0.45\linewidth]
 {\includegraphics[width=6cm,height=3.5cm]{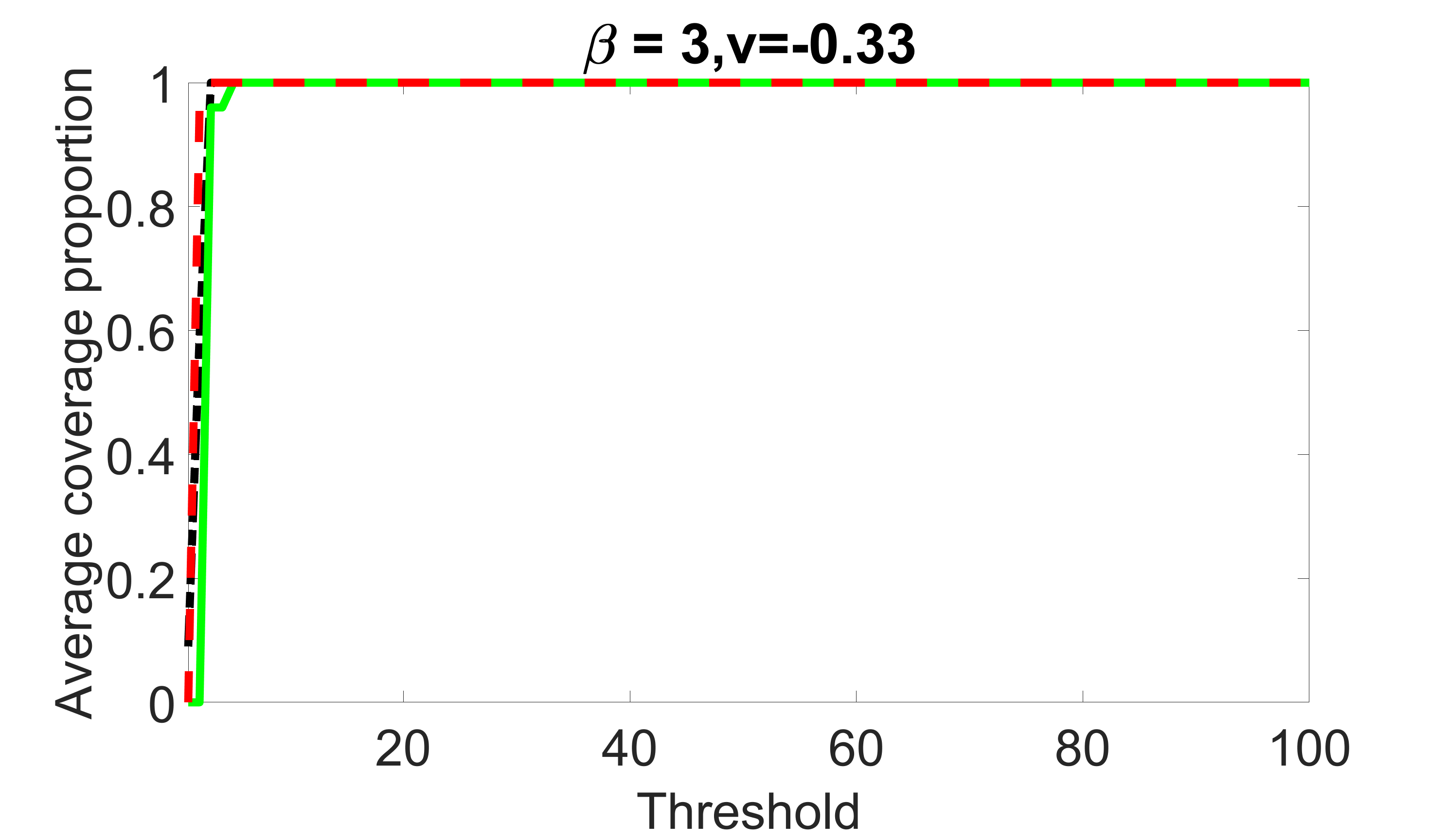}}
 \subcaptionbox{Confounder: medium \\ outcome, medium exposure}[0.45\linewidth]
 {\includegraphics[width=6cm,height=3.5cm]{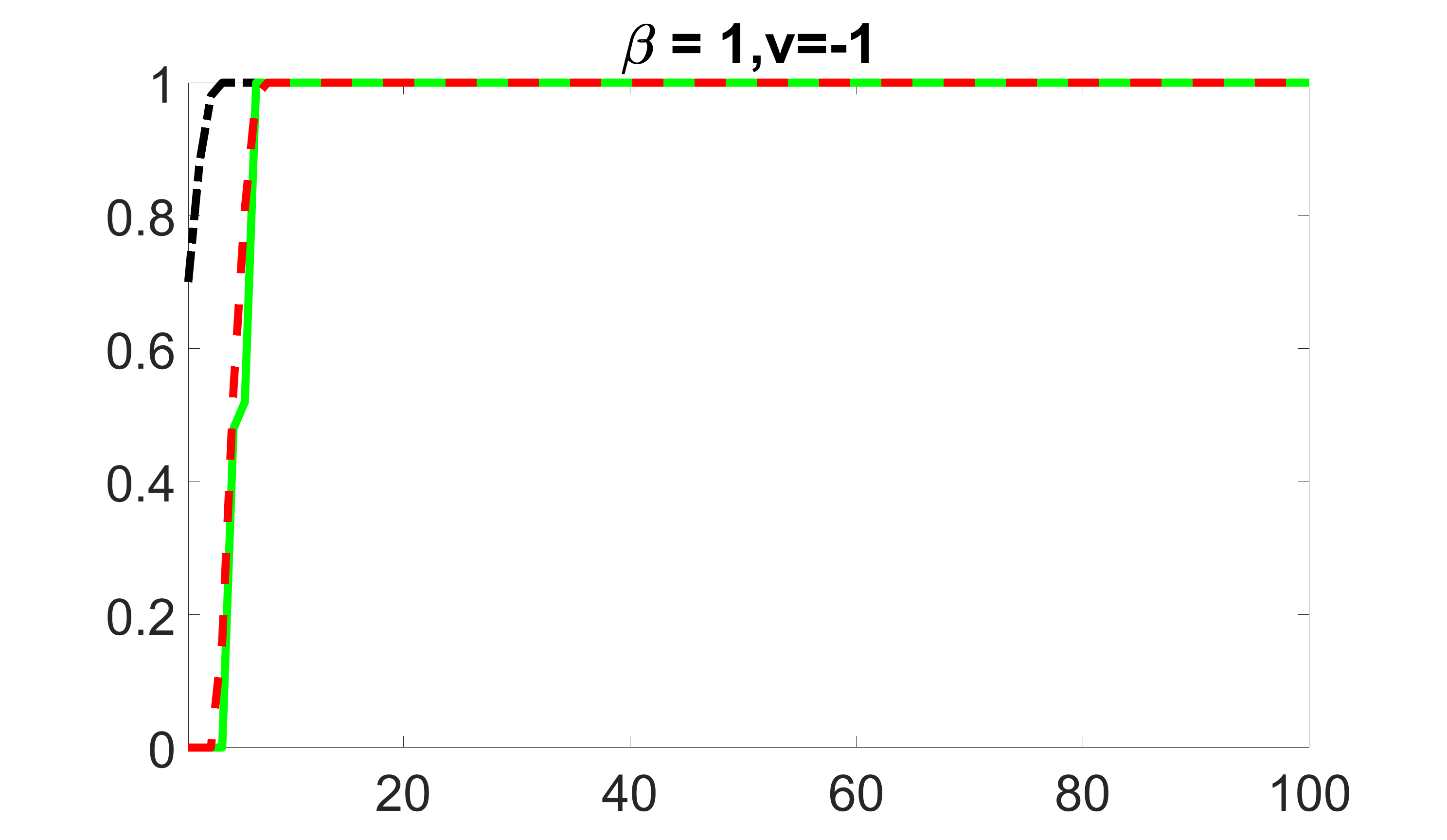}}
  \subcaptionbox{Confounder: weak \\ outcome, strong exposure}[0.45\linewidth]
 {\includegraphics[width=6cm,height=3.5cm]{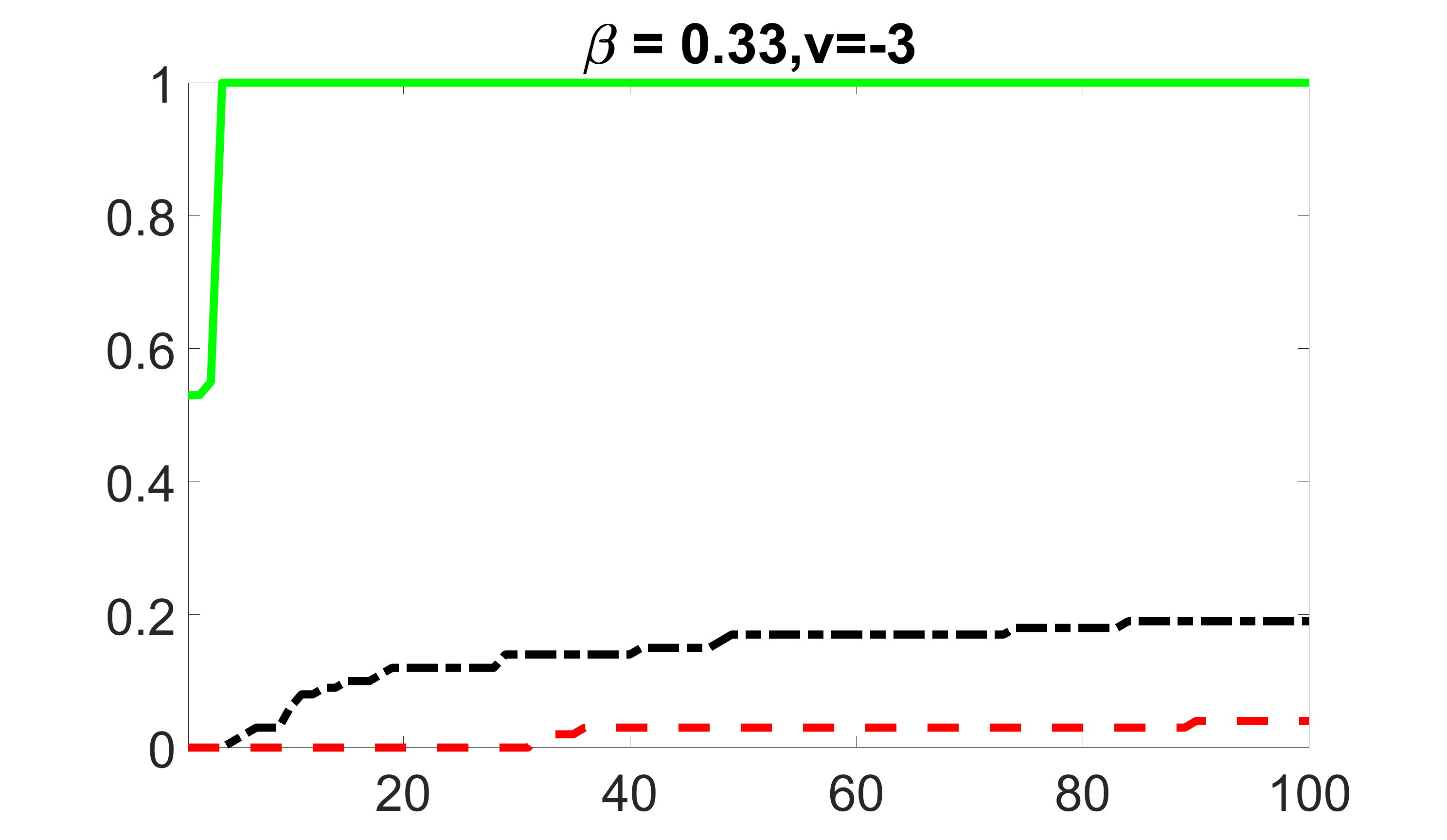}}
  \subcaptionbox{Precision: strong \\ outcome, zero exposure}[0.45\linewidth]
 {\includegraphics[width=6cm,height=3.5cm]{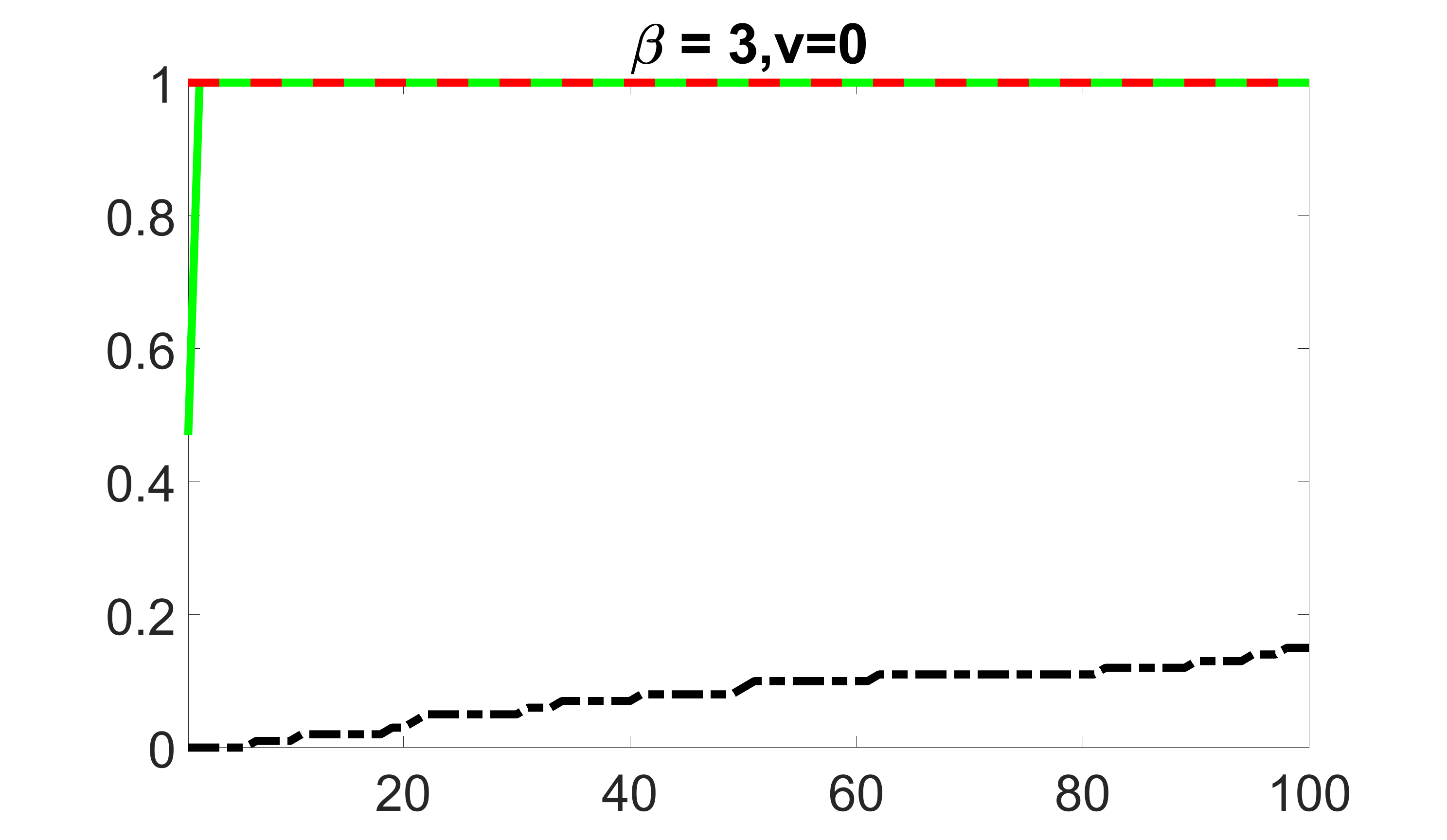}}
  \subcaptionbox{Precision: medium \\ outcome, zero exposure}[0.45\linewidth]
 {\includegraphics[width=6cm,height=3.5cm]{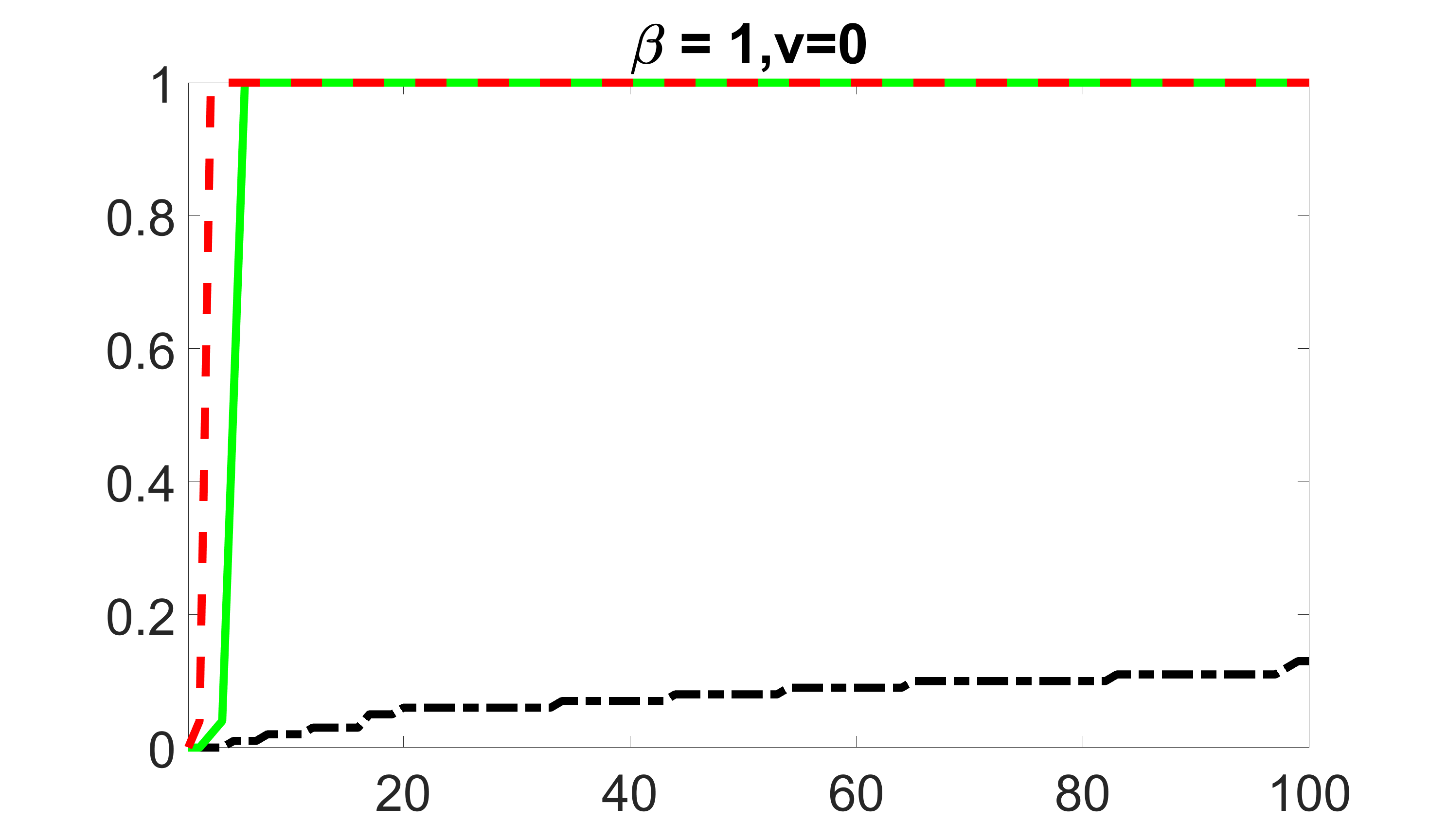}}
  \subcaptionbox{Precision: weak \\ outcome, zero exposure}[0.45\linewidth]
 {\includegraphics[width=6cm,height=3.5cm]{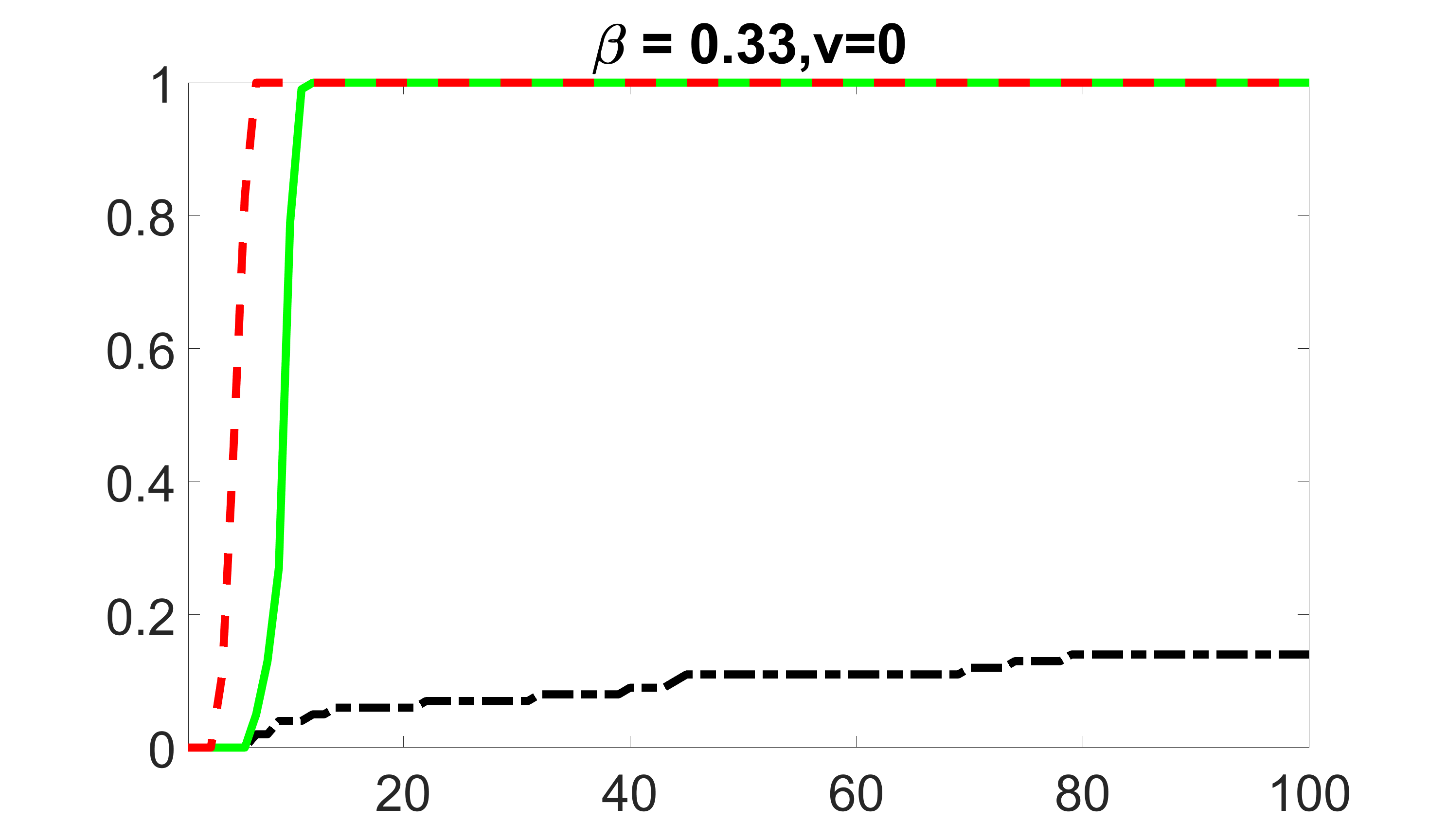}}
  \subcaptionbox{Overall coverage of $\mathcal{M}_1$}[0.45\linewidth]
 {\includegraphics[width=6cm,height=3.5cm]{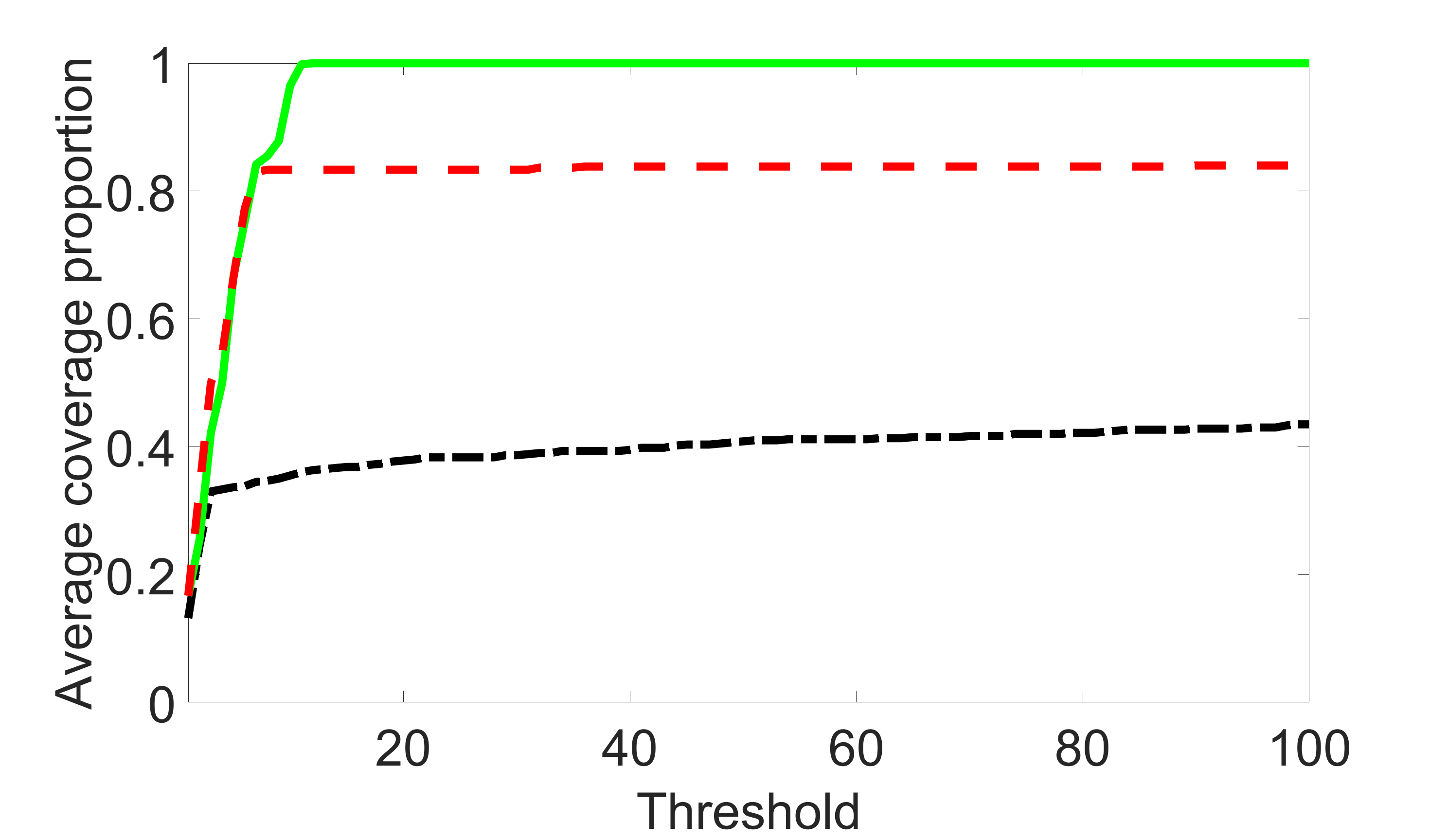}}
\caption{Simulation results for the case $(n,s,\sigma) = (1000,5000,1)$: Panels (a) -- (f) plot the average coverage proportion for $X_l$, where $l=1,2,3,104,105$ and $106$. Panels (a) -- (c) correspond to strong outcome and weak exposure predictor, moderate outcome and moderate exposure predictor and weak outcome and strong exposure predictor; Panels (d) -- (f) correspond to strong, moderate and weak predictors of outcome only. Panel (g) plots the average coverage proportion for the index set $\mathcal{M}_1 = \{1,2,3,104,105,106\}$. The x-axis represents the size of $\widehat{\mathcal{M}} $, while
y-axis denotes the average proportion. The green solid, the red dashed and the black dash dotted lines denote our joint screening method, the outcome screening method, and the intersection screening method, respectively. }
\label{sim1step1n1000sigma1}
\end{figure}

\begin{figure}[htbp]
\captionsetup[subfigure]{justification=centering}
\centering
 \subcaptionbox{Confounder: strong \\ outcome, weak exposure}[0.45\linewidth]
 {\includegraphics[width=6cm,height=3.5cm]{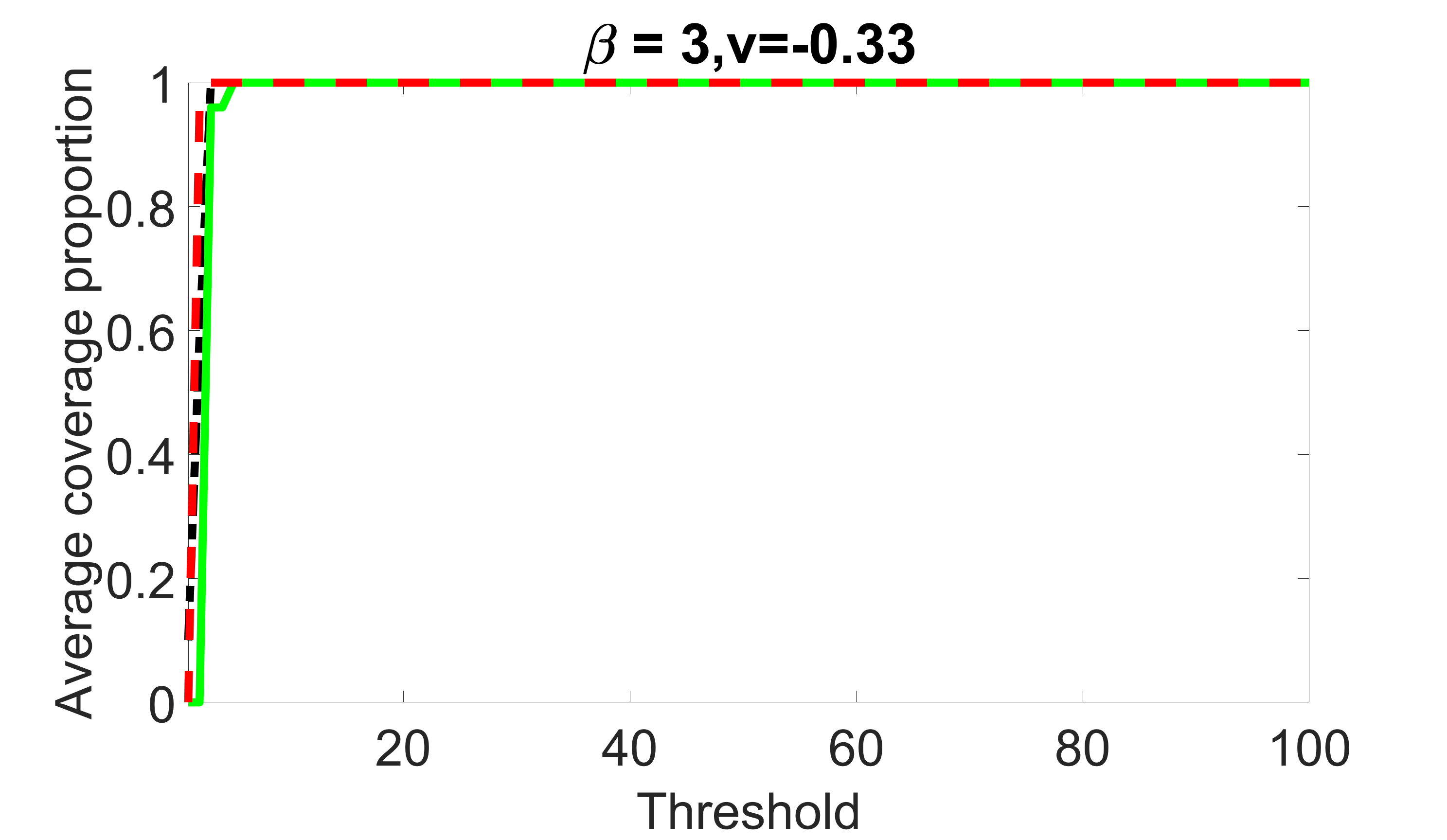}}
 \subcaptionbox{Confounder: medium \\ outcome, medium exposure}[0.45\linewidth]
 {\includegraphics[width=6cm,height=3.5cm]{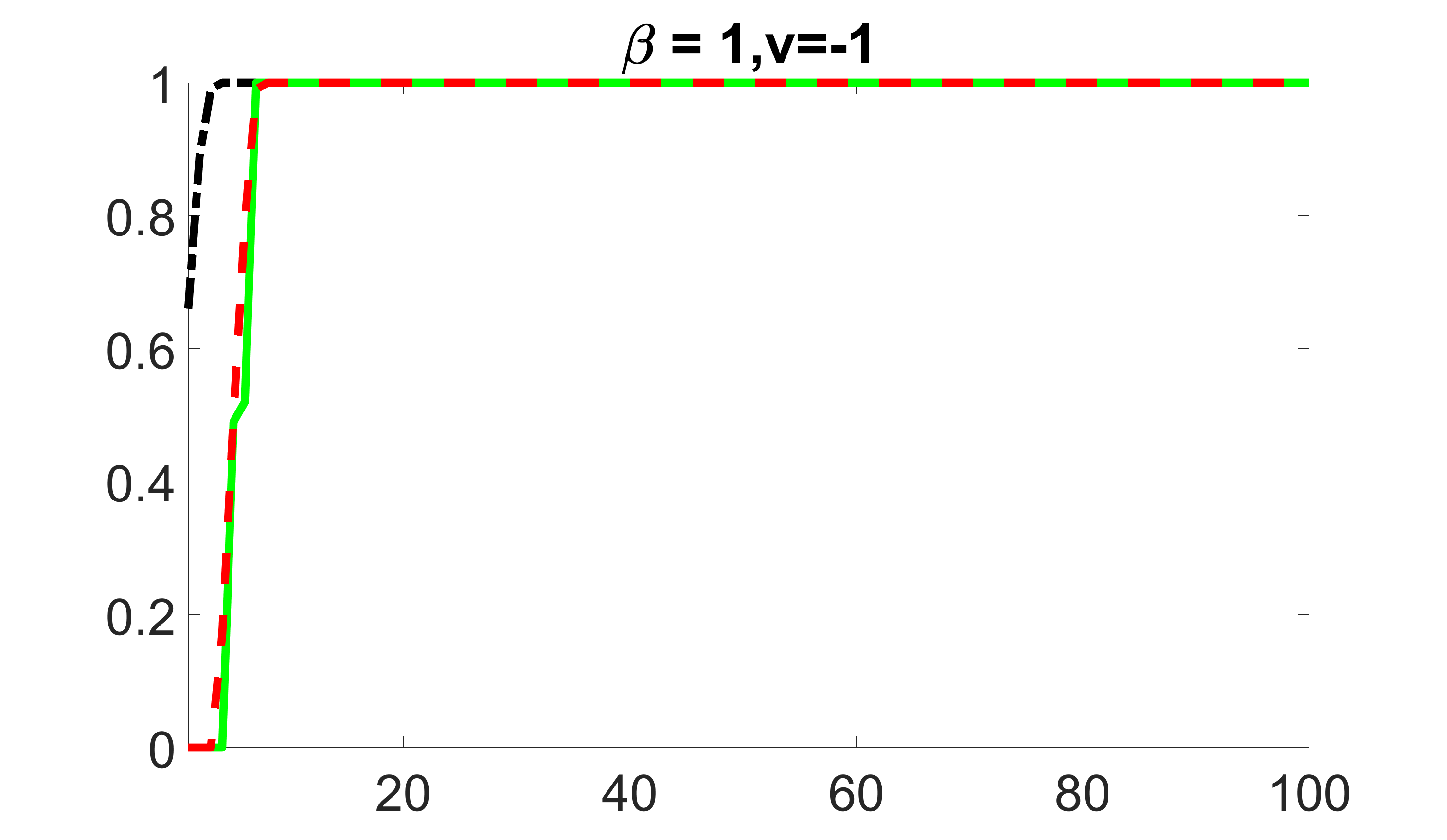}}
  \subcaptionbox{Confounder: weak \\ outcome, strong exposure}[0.45\linewidth]
 {\includegraphics[width=6cm,height=3.5cm]{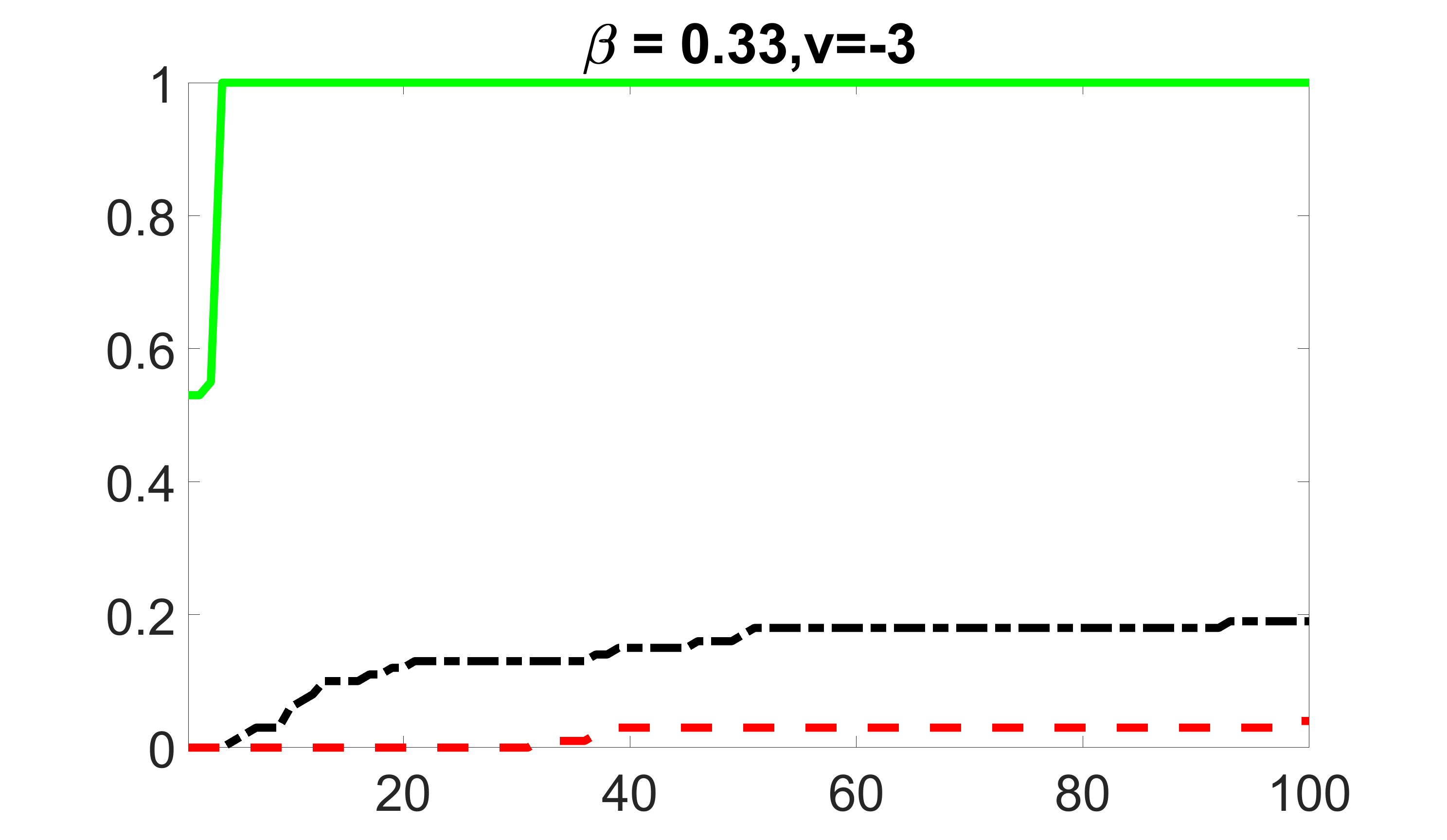}}
  \subcaptionbox{Precision: strong \\ outcome, zero exposure}[0.45\linewidth]
 {\includegraphics[width=6cm,height=3.5cm]{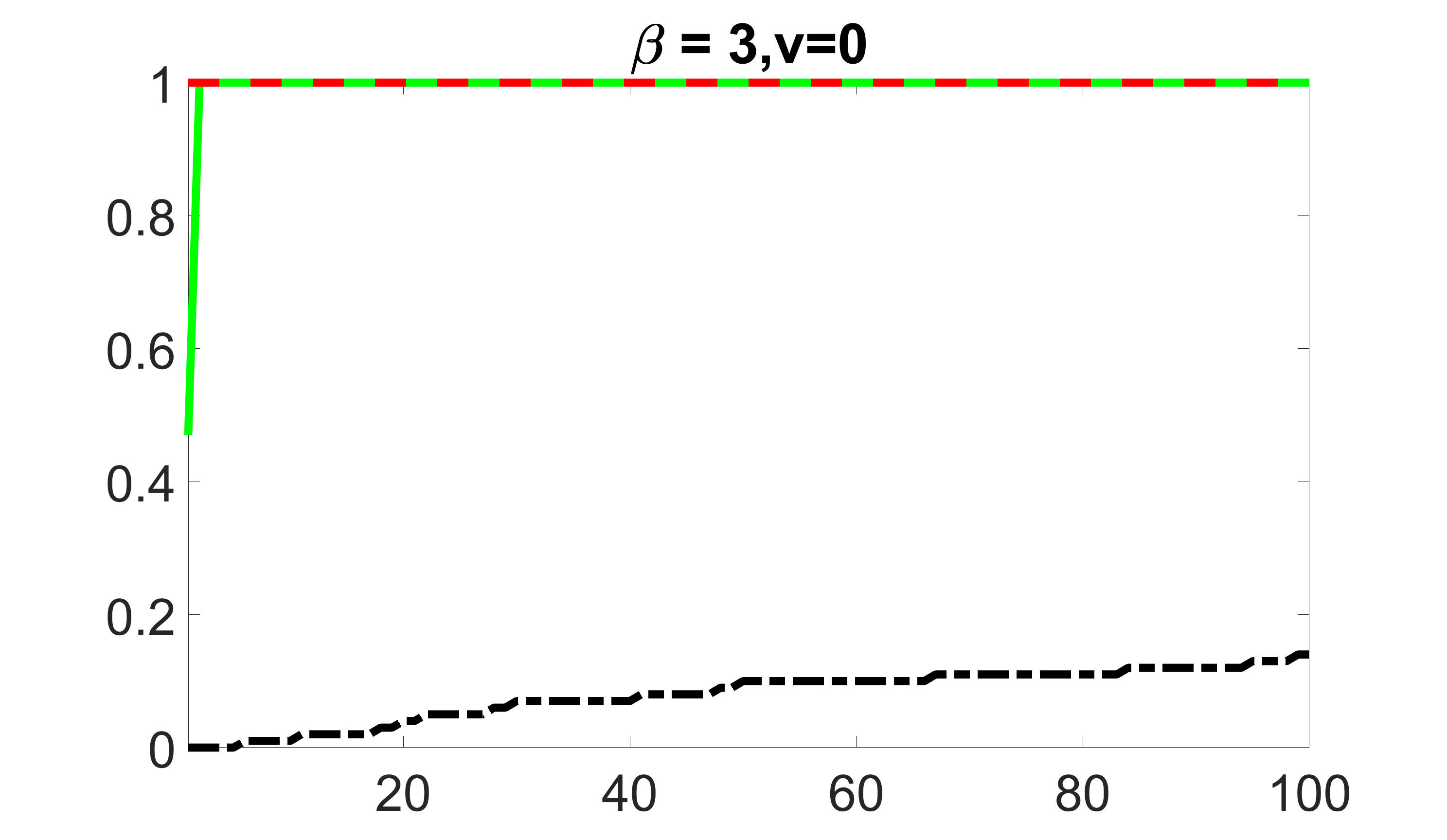}}
  \subcaptionbox{Precision: medium \\ outcome, zero exposure}[0.45\linewidth]
 {\includegraphics[width=6cm,height=3.5cm]{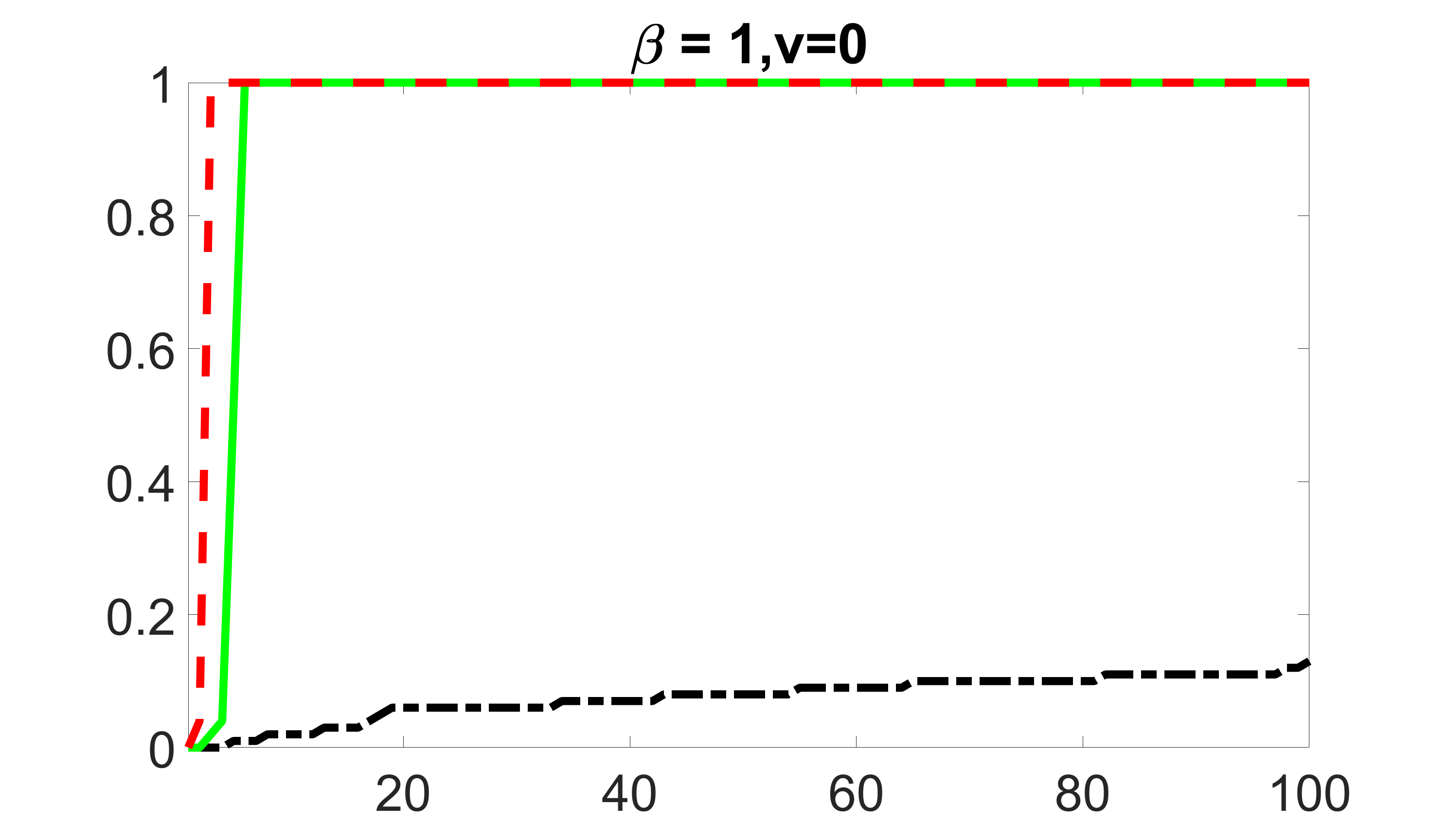}}
  \subcaptionbox{Precision: weak \\ outcome, zero exposure}[0.45\linewidth]
 {\includegraphics[width=6cm,height=3.5cm]{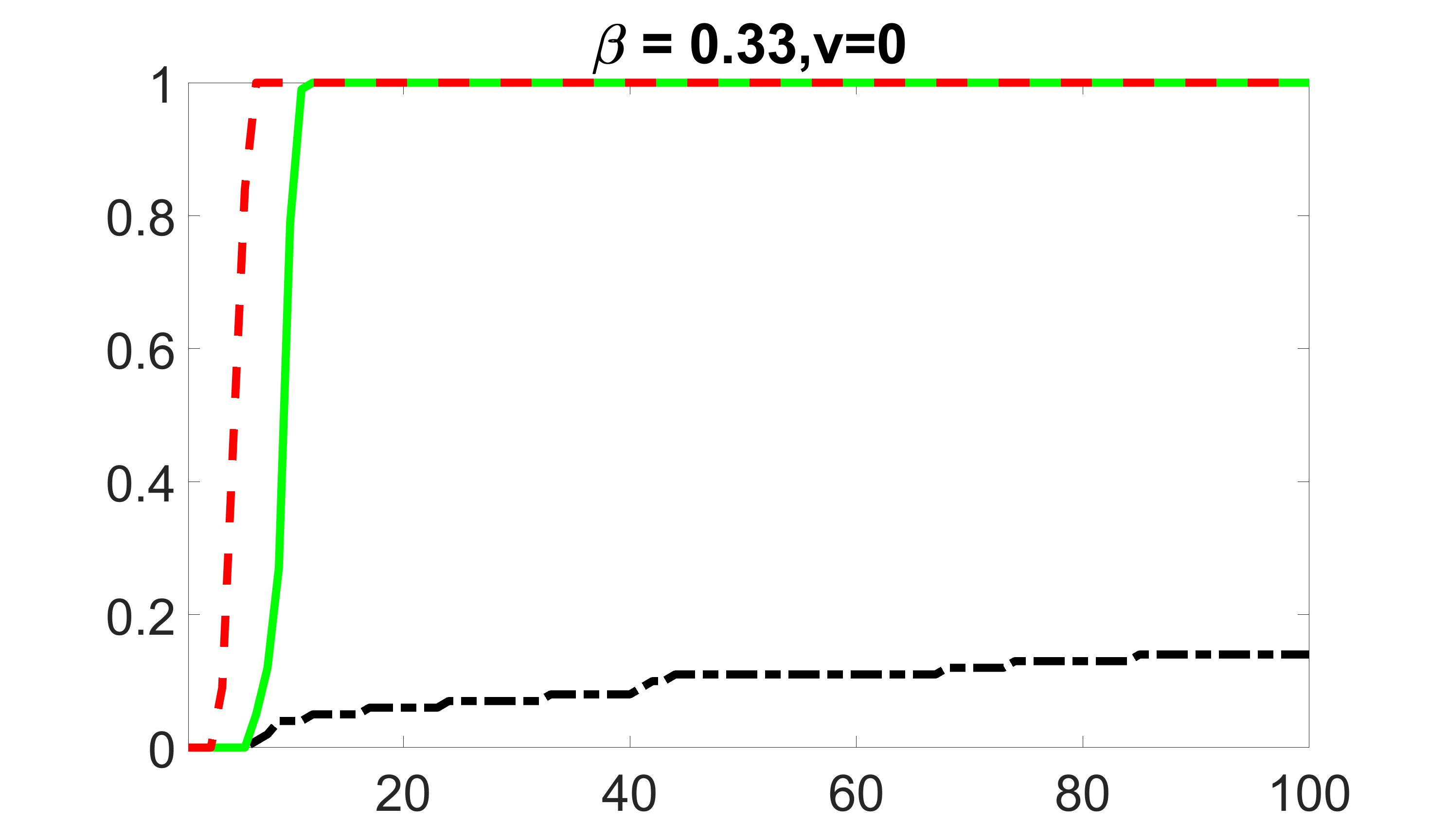}}
  \subcaptionbox{Overall coverage of $\mathcal{M}_1$}[0.45\linewidth]
 {\includegraphics[width=6cm,height=3.5cm]{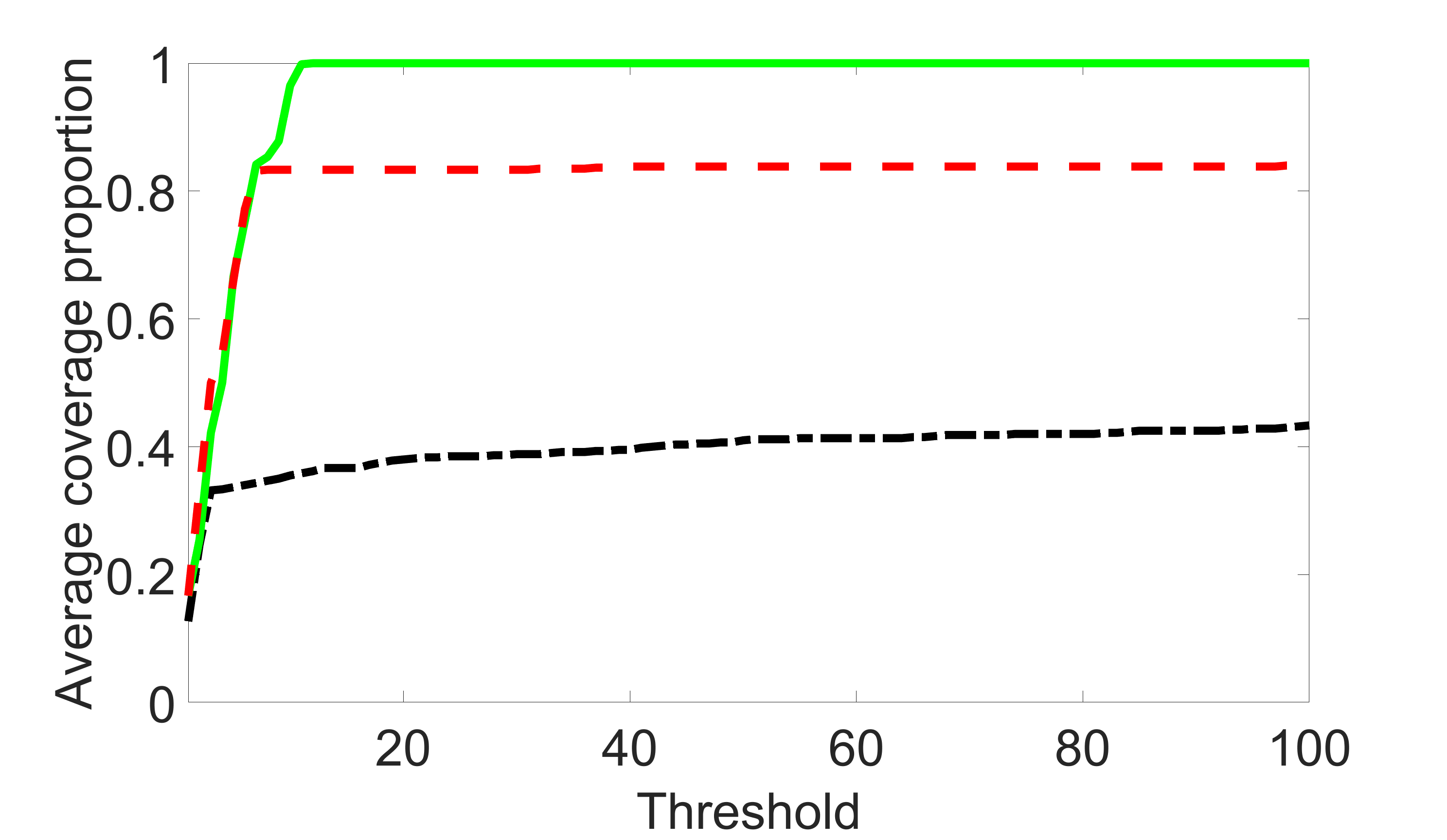}}
\caption{Simulation results for the case $(n,s,\sigma) = (1000,5000,0.5)$: Panels (a) -- (f) plot the average coverage proportion for $X_l$, where $l=1,2,3,104,105$ and $106$. Panels (a) -- (c) correspond to strong outcome and weak exposure predictor, moderate outcome and moderate exposure predictor and weak outcome and strong exposure predictor; Panels (d) -- (f) correspond to strong, moderate and weak predictors of outcome only. Panel (g) plots the average coverage proportion for the index set $\mathcal{M}_1 = \{1,2,3,104,105,106\}$. The x-axis represents the size of $\widehat{\mathcal{M}} $, while
y-axis denotes the average proportion. The green solid, the red dashed and the black dash dotted lines denote our joint screening method, the outcome screening method, and the intersection screening method, respectively.}
\label{sim1step1n1000sigma025}
\end{figure}

\subsection{Sensitivity and specificity analyses of simulation}
\label{Sensitivity and specificity analysis of simulation}
In this subsection, we report the sensitivity and specificity of the estimates. The sensitivity (true positive rate) is defined as $\frac{ |\{ j: \widehat{\beta}_j \neq 0 \} \cap \mathcal{M}_1 |}{ | \mathcal{M}_1|}$, i.e. the proportion of variables in the oracle adjustment set $\mathcal{M}_1$ that are selected by our estimation procedure. The specificity (true negative rate) is defined as $\frac{ |\{ j: \widehat{\beta}_j = 0 \} \cap (\mathcal{I} \cup \mathcal{S})| }{ | \mathcal{I} \cup  \mathcal{S}|}$, i.e. the proportion of variables not in the oracle adjustment set $\mathcal{M}_1$ that are not selected by our estimation procedure. Furthermore, we define the instrumental specificity as $\frac{ |\{ j: \widehat{\beta}_j = 0 \} \cap \mathcal{I}| }{ | \mathcal{I}|}$, i.e. the proportion of variables in the instrumental set $\mathcal{I}$ that are not selected by our estimation procedure.

\begin{table}[htbp]
\centering
\caption{Simulation results for $ \sigma=1 $ and $ \sigma = 0.5 $, when $n=500$. The average $\bm\beta$  sensitivity, $\bm\beta$  instrumental specificity and $\bm\beta$ specificity, MSE for $\bm\beta$, and MSE for $ {\bm B}$, with their associated standard errors in the parentheses are reported. The results are based on 100 Monte Carlo repetitions. ``No Lasso'' estimate is calculated by including all the selected variables from the screening step then estimating $\bm{B}$ using the optimization (\ref{min1}) without the $l_{1}$-regularization. ``Oracle'' estimate is calculated by pretending to know the correct set of confounders and precision variables as $X$ and then estimate $\bm{B}$ using the optimization (\ref{min1}) without the $l_{1}$-regularization. 
}
\begin{tabular}{ cccccc }
n = 500 &  Sensitivity & Instrumental specificity 
&  Specificity & MSE $\bm\beta$ & MSE ${\bm{B}}$\\
\hline
\multicolumn{6}{c}{$\sigma$ = 1.0}\\
Oracle &1.000(0.000)&1.000(0.000)&1.000(0.000) &0.036(0.002)&0.553(0.005)\\
Proposed &0.833(0.000)&0.293(0.020)&0.998(0.000)&0.303(0.008)&0.574(0.006)\\
No Lasso &1.000(0.000)&0.000(0.000)&0.985(0.000)&1.740(0.078)&0.693(0.013)\\
\multicolumn{6}{c}{$\sigma$ = 0.5}\\
Oracle &1.000(0.000)&1.000(0.000)&1.000(0.000)&0.006(0.000)&0.340(0.004)\\
Proposed &0.897(0.008)&0.217(0.017)&0.999(0.000)&0.191(0.005)&0.345(0.004)\\
No Lasso &1.000(0.000)&0.000(0.000)&0.985(0.000)&0.372(0.017)&0.371(0.004)\\
\end{tabular}
\label{sim1t1sssmmn500}
\end{table}

In the simulation studies, we report the results for the case $n=500$, which is close to the sample size $566$ in the real data. From Table \ref{sim1t1sssmmn500}, one can see that the second step can only regularize out some of the instrumental variables. We guess there may be a better tuning method, which can regularize out more instrumental variables, while still keep the confounders and prevision variables in the model. We leave this as future research. Nevertheless, one can also see from Table \ref{sim1t1sssmmn500} that although the proposed method may not remove all of the instrumental variables, eliminating even just some of the instruments greatly reduce the MSEs of both $\bm\beta$ and ${\bm{B}}$, compared to the method where we do not impose $l_{1}$-regularization on $\bm\beta$ in the second-step estimation (denoted by the ``No Lasso'' method). In addition, the estimation of ${\bm B} $ is reasonably good compared to the oracle estimates as shown in Table \ref{sim1t1oracle} of the main article.

\subsection{Screening under different sparsity levels}
\label{Screening under different sparsity levels}
We also consider different sparsity levels in the simulation. It is particularly of interest for our study since when more instrumental variables than confounders and precision variables exist, which could be the case in an imaging-genetic study, the robustness of the proposed method may be undermined. As discussed before, to reduce bias and increase the statistical efficiency of the estimated $\bm{B}$, the ideal adjustment set should include all confounders and precision variables while excluding instrumental variables and irrelevant variables. In particular, we consider three scenarios where the sizes of instrumental variables $\mathcal{I}$ are the same, twice, and eight times of the size of confounders and precision variables $\mathcal{M}_1 $. 

We set $s=5000$ and the settings for $ {\bm B} $ and $ {\bm C} $ remain the same as before: $ {\bm B} $ is as in Figure \ref{figureTCross}(a), and $ {\bm C} $ is as in Figure \ref{figureTCross}(b). Further 
we set $ {\bm C}_l=v_l*\bm{C}$, where $ v_1=-1/3$, $v_2=-1$, $v_3=-3$, $ v_{207}=v_{210}=\ldots=v_{204+6L}=-3$, $v_{208}=v_{211}=\ldots=v_{205+6L}=-1$, $v_{209}=v_{212}=\ldots=v_{206+6L}=-1/3$, and $ v_l=0 $ for $ 4\leq l\leq 206$ and $ 207+6L \leq l \leq s $. Here $L$ is a positive integer.
We set $ \beta_1=3$, $\beta_2=1$, $\beta_3=1/3$, $ \beta_{104}=3$, $\beta_{105}=1$, $\beta_{106}=1/3$, and $ \beta_l=0$ for $4 \leq l \leq 103$ and $ 107\leq l\leq s$. In this setting, we have $\mathcal{C} = \{1, 2, 3\}$, $\mathcal{P} = \{ 104, 105, 106\}$, $\mathcal{I} =\{ 207,208,209,\ldots,206+6L\}$ and $\mathcal{S} = \{1 \ldots, 5000\} \backslash \{1,2,3,104,105,106,207,208,209,\ldots,206+6L\}$. Note that $\frac{|\mathcal{I}|}{|\mathcal{C} \cup \mathcal{P}|}=\frac{|\mathcal{I}|}{|\mathcal{M}_1|} = L$. For $n=200$, we let $\sigma=1$ and $0.5$, and consider three different sparsity levels that $L=1,2,8$.
The complete screening results can be found in Figures \ref{sim1step1n200sigma1sparsity_instru2} -- \ref{sim1step1n200sigma025sparsity_instru16}.

Specifically, as summarized in Figure \ref{sim1step1n200sparsity_instru_summary}, when the number of instrumental variables is much larger than that of confounders and precision variables, the size of $\mathcal{M}_1 \cup \mathcal{M}_2$  is larger than the number of covariates kept in the first screening step. And in this case, our results show that the  screening step  may include many instrumental variables, while missing  some confounders and precision variables. This may deteriorate the accuracy and efficiency of the second step estimation.

\begin{figure}[htbp]
\captionsetup[subfigure]{justification=centering}
\centering
 \subcaptionbox{$\frac{|\mathcal{I}|}{|\mathcal{M}_1|} = 1$, $\sigma=1$}[0.45\linewidth]
 {\includegraphics[width=7cm,height=10cm,keepaspectratio]{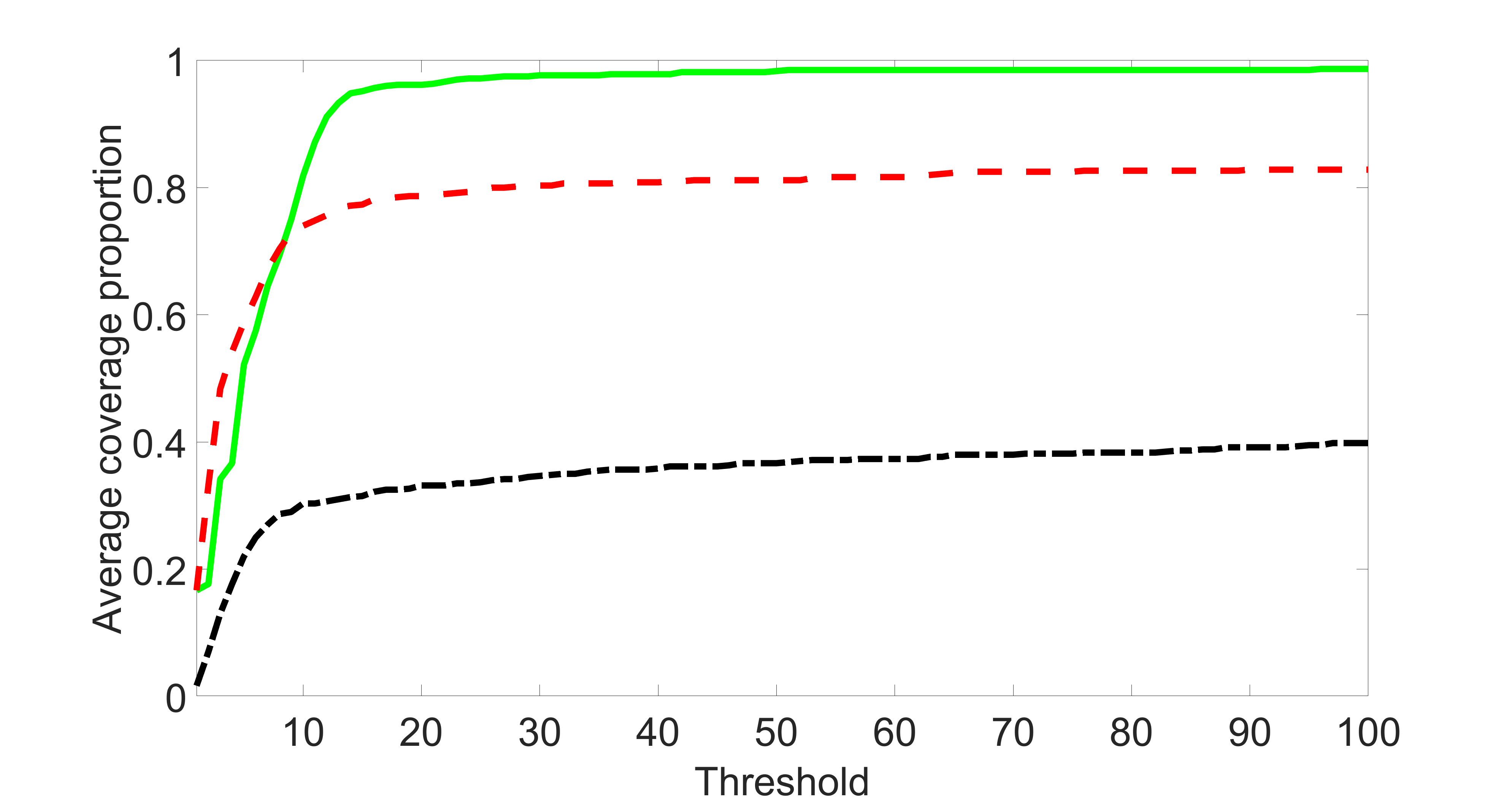}}
 \subcaptionbox{$\frac{|\mathcal{I}|}{|\mathcal{M}_1|} = 1$, $\sigma = 0.5$}[0.45\linewidth]
 {\includegraphics[width=7cm,height=10cm,keepaspectratio]{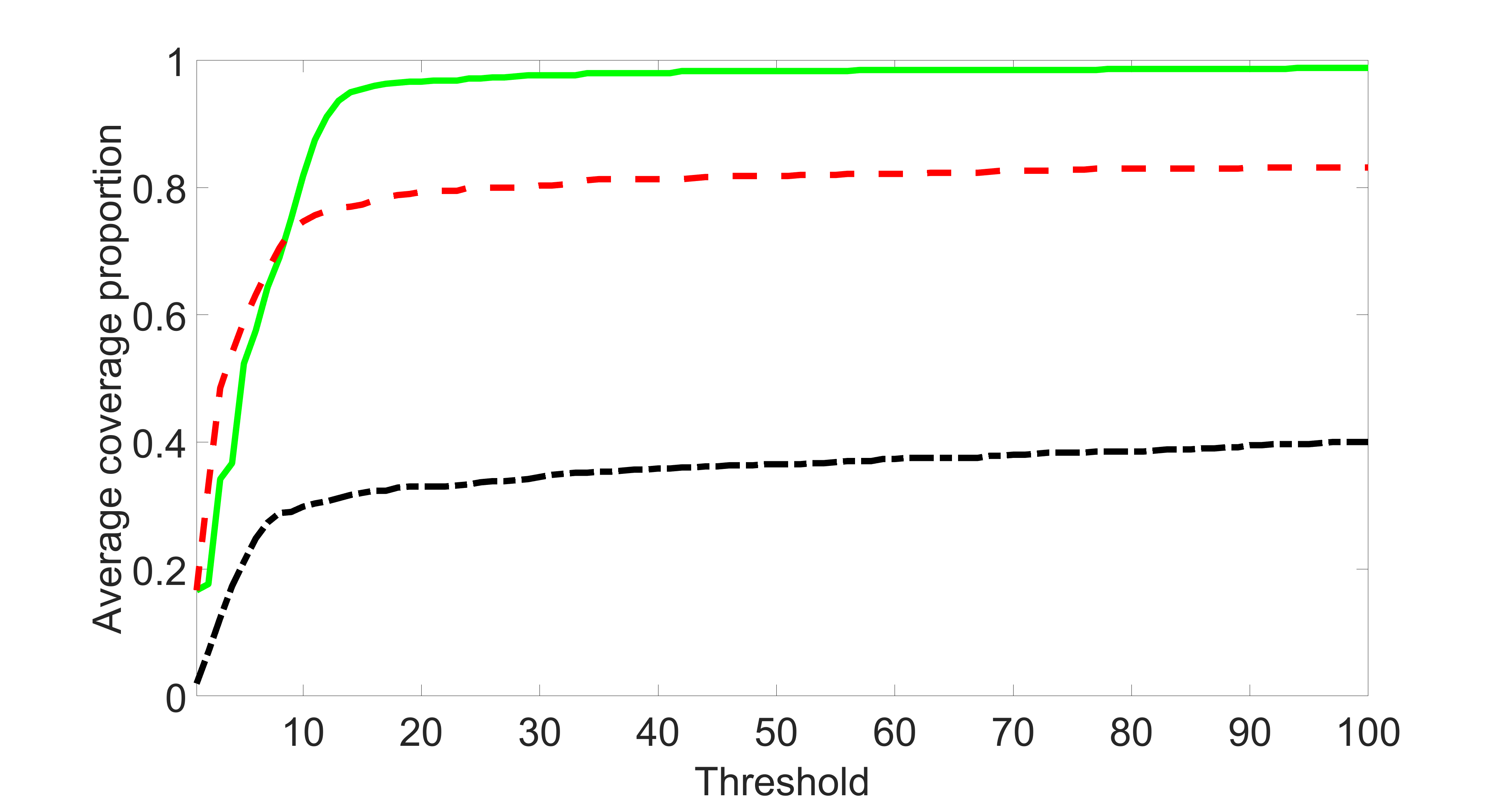}}
  \subcaptionbox{$\frac{|\mathcal{I}|}{|\mathcal{M}_1|} = 2$, $\sigma=1$}[0.45\linewidth]
 {\includegraphics[width=7cm,height=3.75cm]{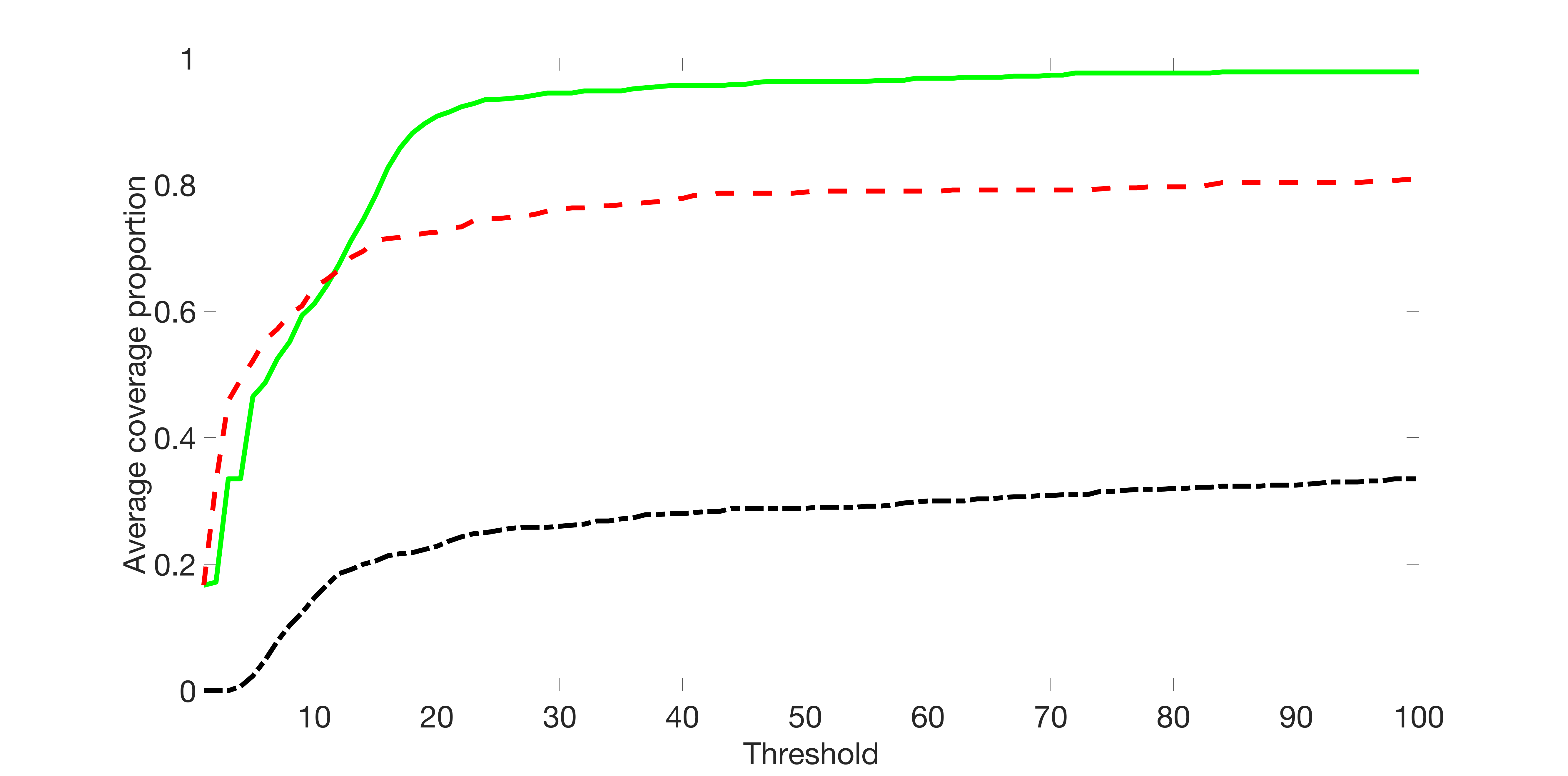}}
  \subcaptionbox{$\frac{|\mathcal{I}|}{|\mathcal{M}_1|} = 2$, $\sigma = 0.5$}[0.45\linewidth]
 {\includegraphics[width=7cm,height=10cm,keepaspectratio]{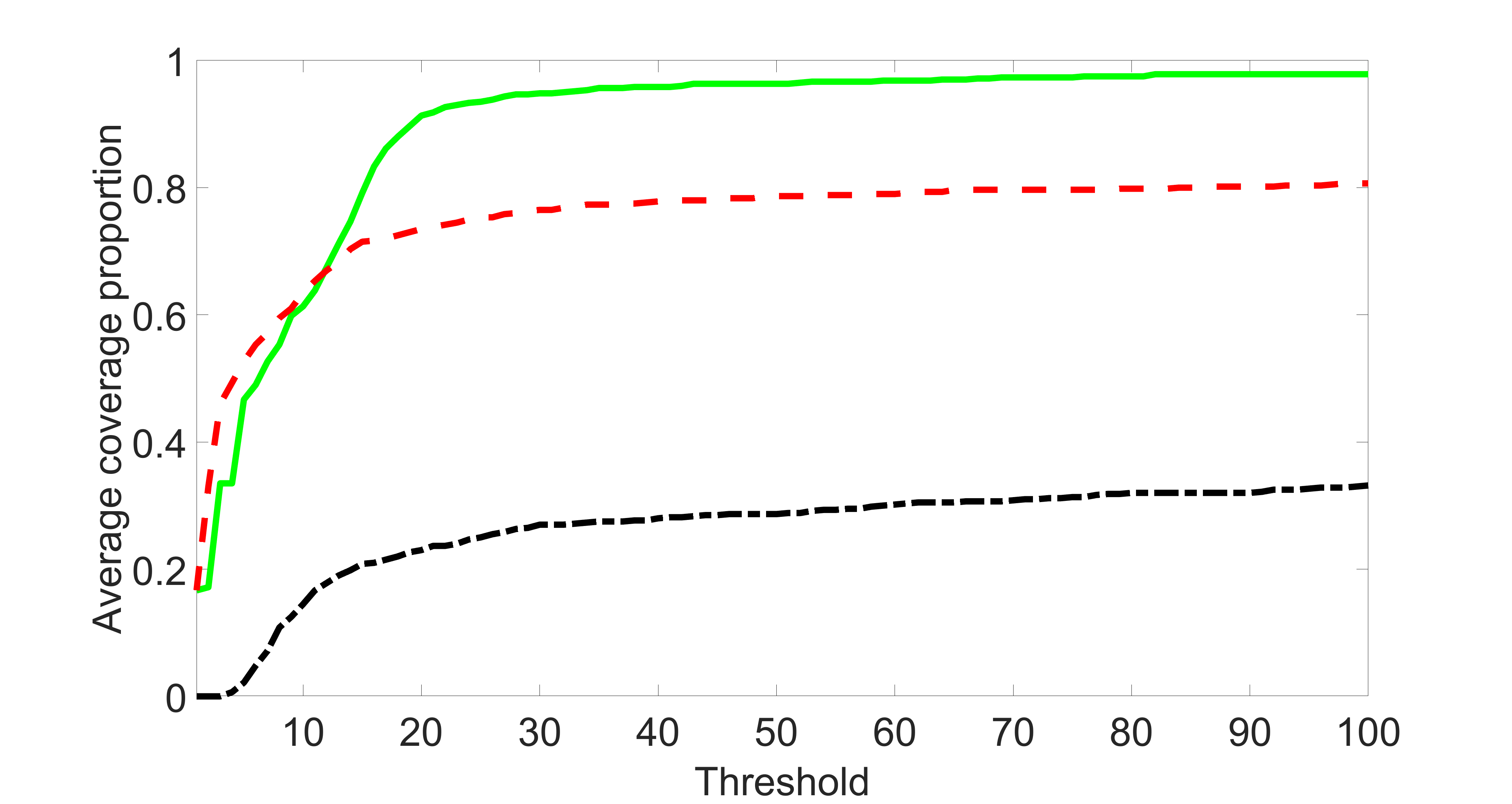}}
  \subcaptionbox{$\frac{|\mathcal{I}|}{|\mathcal{M}_1|} = 8$, $\sigma=1$}[0.45\linewidth]
 {\includegraphics[width=7cm,height=10cm,keepaspectratio]{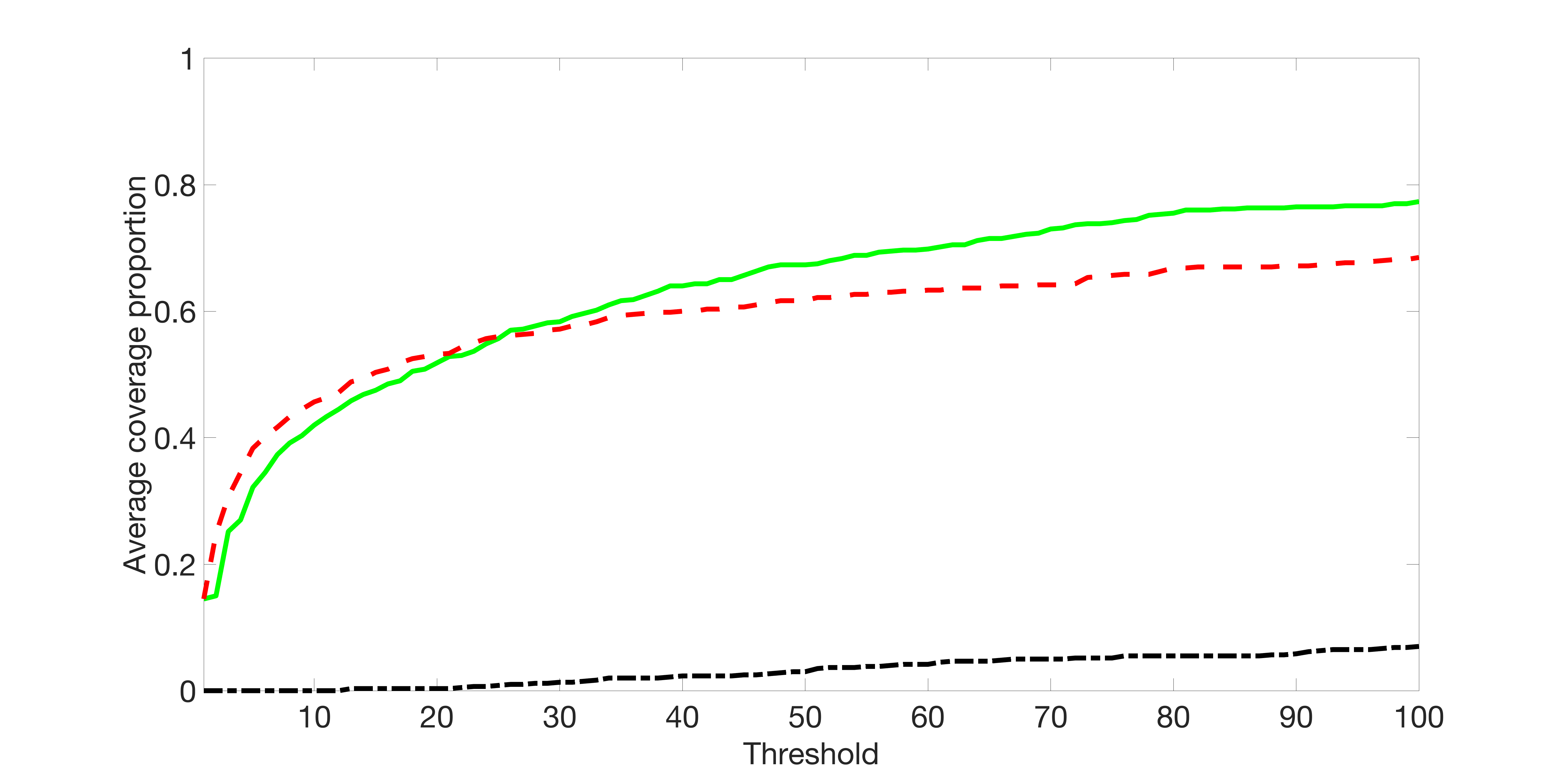}}
  \subcaptionbox{$\frac{|\mathcal{I}|}{|\mathcal{M}_1|} = 8$, $\sigma = 0.5$}[0.45\linewidth]
 {\includegraphics[width=7cm,height=10cm,keepaspectratio]{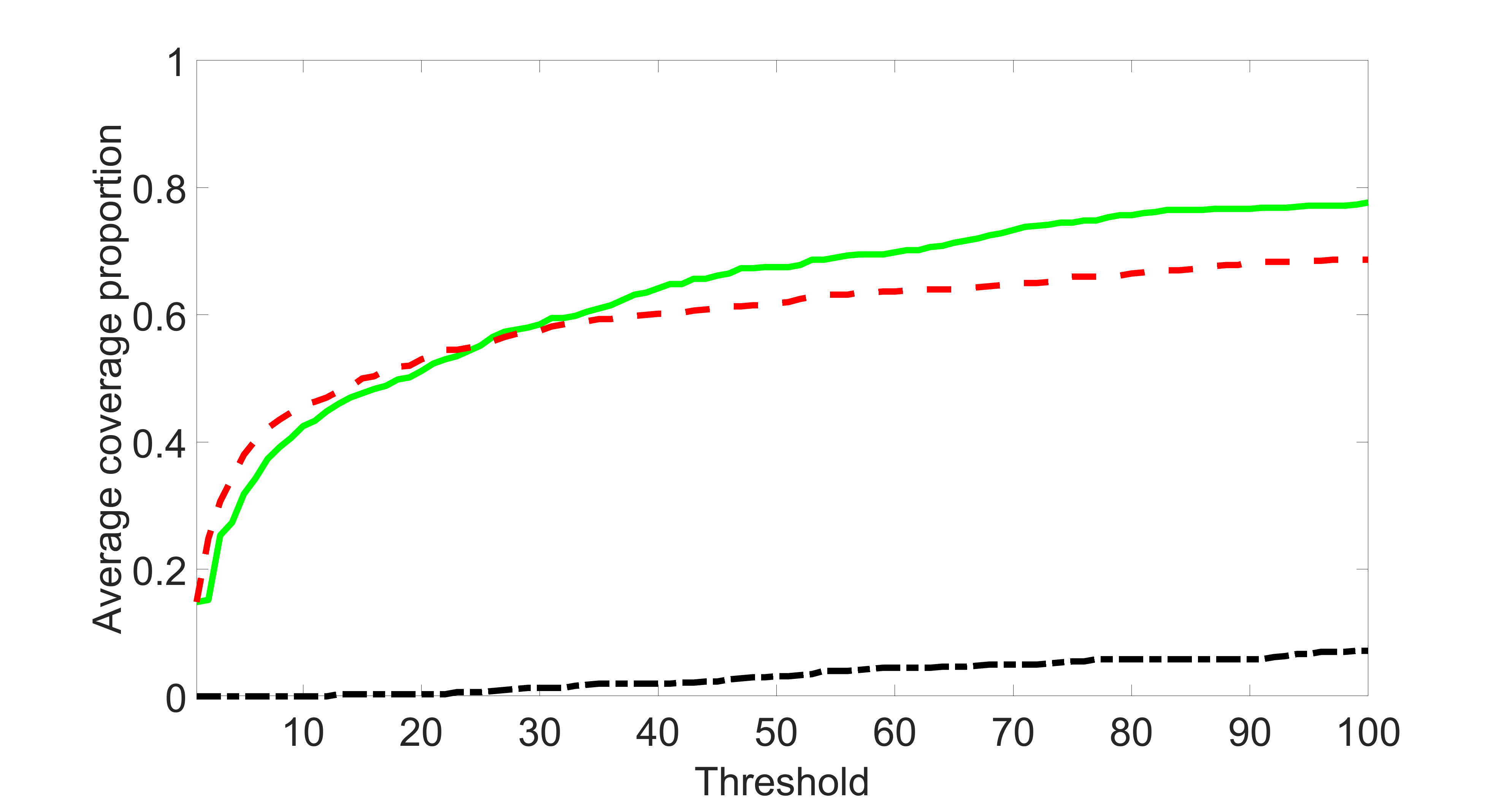}}
\caption{Simulation results where the size of instrumental variables $\mathcal{I}$ are the same, twice and eight times of $\mathcal{M}_1 $ for the case $(n,s) = (200,5000)$: Panels (a) (c) (e) plot the average coverage proportion for the index set $\mathcal{M}_1 = \{1,2,3,104,105,106\}$ when $\sigma=1$. Panels (b) (d) (f) plot the average coverage proportion for the index set $\mathcal{M}_1 = \{1,2,3,104,105,106\}$ when $\sigma = 0.5$. The x-axis represents the size of $\widehat{\mathcal{M}} $, while
y-axis denotes the average proportion. The green solid, the red dashed and the black dash dotted lines denote our joint screening method, the outcome screening method, and the intersection screening method, respectively. }
\label{sim1step1n200sparsity_instru_summary}
\end{figure}

\begin{figure}[htbp]
\captionsetup[subfigure]{justification=centering}
\centering
 \subcaptionbox{Confounder: strong \\ outcome, weak exposure}[0.45\linewidth]
 {\includegraphics[width=6cm,height=3.5cm]{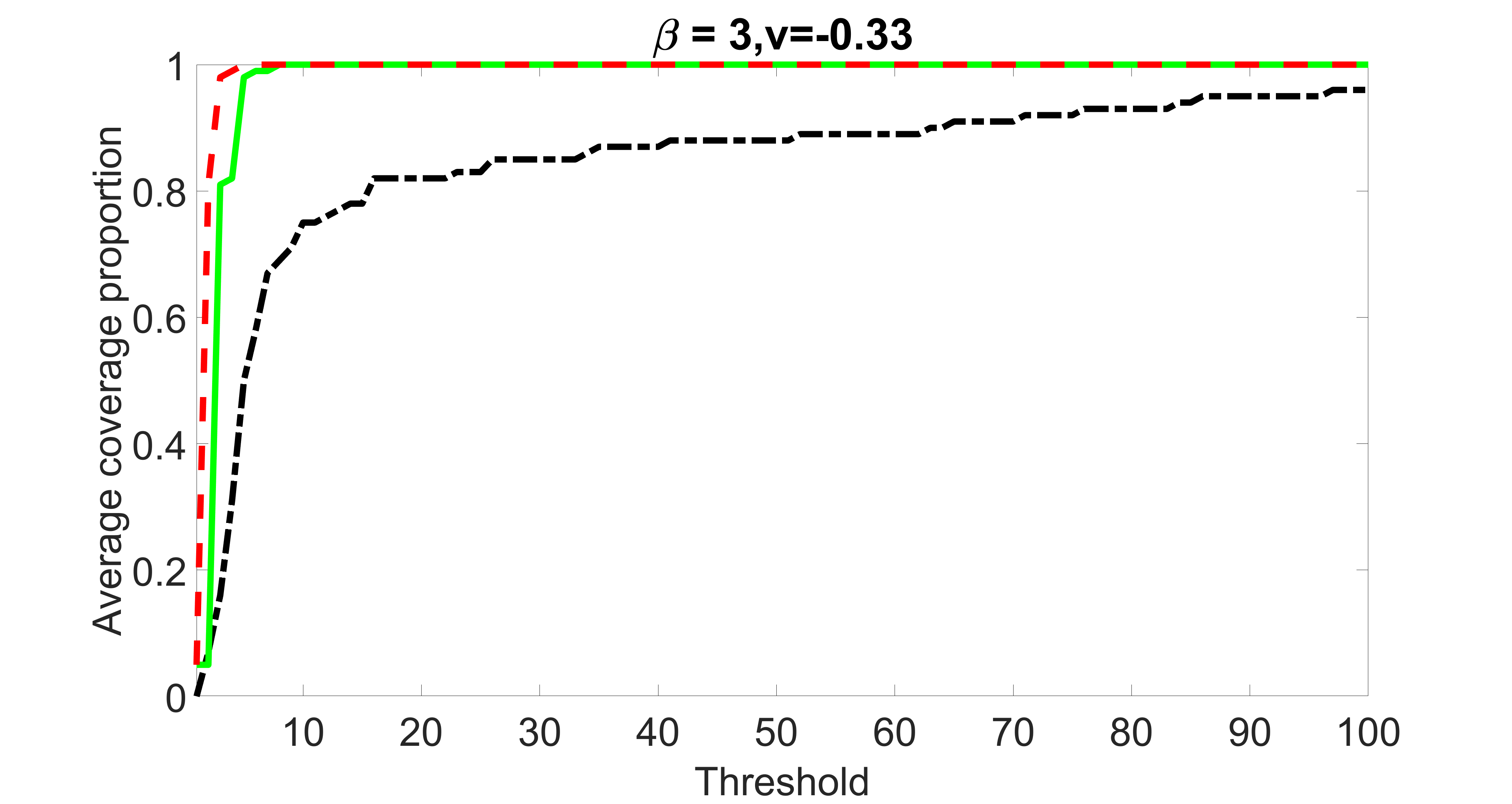}}
 \subcaptionbox{Confounder: medium \\ outcome, medium exposure}[0.45\linewidth]
 {\includegraphics[width=6cm,height=3.5cm]{./plotsMainArkSupp/sim1oal_p64_v2_optNorm2_n200s5000snry9sparsity_instru2count1.png}}
  \subcaptionbox{Confounder: weak \\ outcome, strong exposure}[0.45\linewidth]
 {\includegraphics[width=6cm,height=3.5cm]{./plotsMainArkSupp/sim1oal_p64_v2_optNorm2_n200s5000snry9sparsity_instru2count1.png}}
  \subcaptionbox{Precision: strong \\ outcome, zero exposure}[0.45\linewidth]
 {\includegraphics[width=6cm,height=3.5cm]{./plotsMainArkSupp/sim1oal_p64_v2_optNorm2_n200s5000snry9sparsity_instru2count1.png}}
  \subcaptionbox{Precision: medium \\ outcome, zero exposure}[0.45\linewidth]
 {\includegraphics[width=6cm,height=3.5cm]{./plotsMainArkSupp/sim1oal_p64_v2_optNorm2_n200s5000snry9sparsity_instru2count1.png}}
  \subcaptionbox{Precision: weak \\ outcome, zero exposure}[0.45\linewidth]
 {\includegraphics[width=6cm,height=3.5cm]{./plotsMainArkSupp/sim1oal_p64_v2_optNorm2_n200s5000snry9sparsity_instru2count1.png}}
  \subcaptionbox{Overall coverage of $\mathcal{M}_1$}[0.45\linewidth]
 {\includegraphics[width=6cm,height=3.5cm]{./plotsMainArkSupp/sim1oal_p64_v2_optNorm2_n200s5000sparsity_instru2coverage_snry9.png}}
\caption{Simulation results where the number of instrumental variables are the same of $\mathcal{M}_1 $ for the case $(n,s,\sigma) = (200,5000,1)$: Panels (a) -- (f) plot the average coverage proportion for $X_l$, where $l=1,2,3,104,105$ and $106$. Panels (a) -- (c) correspond to strong outcome and weak exposure predictor, moderate outcome and moderate exposure predictor and weak outcome and strong exposure predictor; Panels (d) -- (f) correspond to strong, moderate and weak predictors of outcome only. Panel (g) plots the average coverage proportion for the index set $\mathcal{M}_1 = \{1,2,3,104,105,106\}$. The x-axis represents the size of $\widehat{\mathcal{M}} $, while
y-axis denotes the average proportion. The green solid, the red dashed and the black dash dotted lines denote our joint screening method, the outcome screening method, and the intersection screening method, respectively. }
\label{sim1step1n200sigma1sparsity_instru2}
\end{figure}

\begin{figure}[htbp]
\captionsetup[subfigure]{justification=centering}
\centering
 \subcaptionbox{Confounder: strong \\ outcome, weak exposure}[0.45\linewidth]
 {\includegraphics[width=6cm,height=3.5cm]{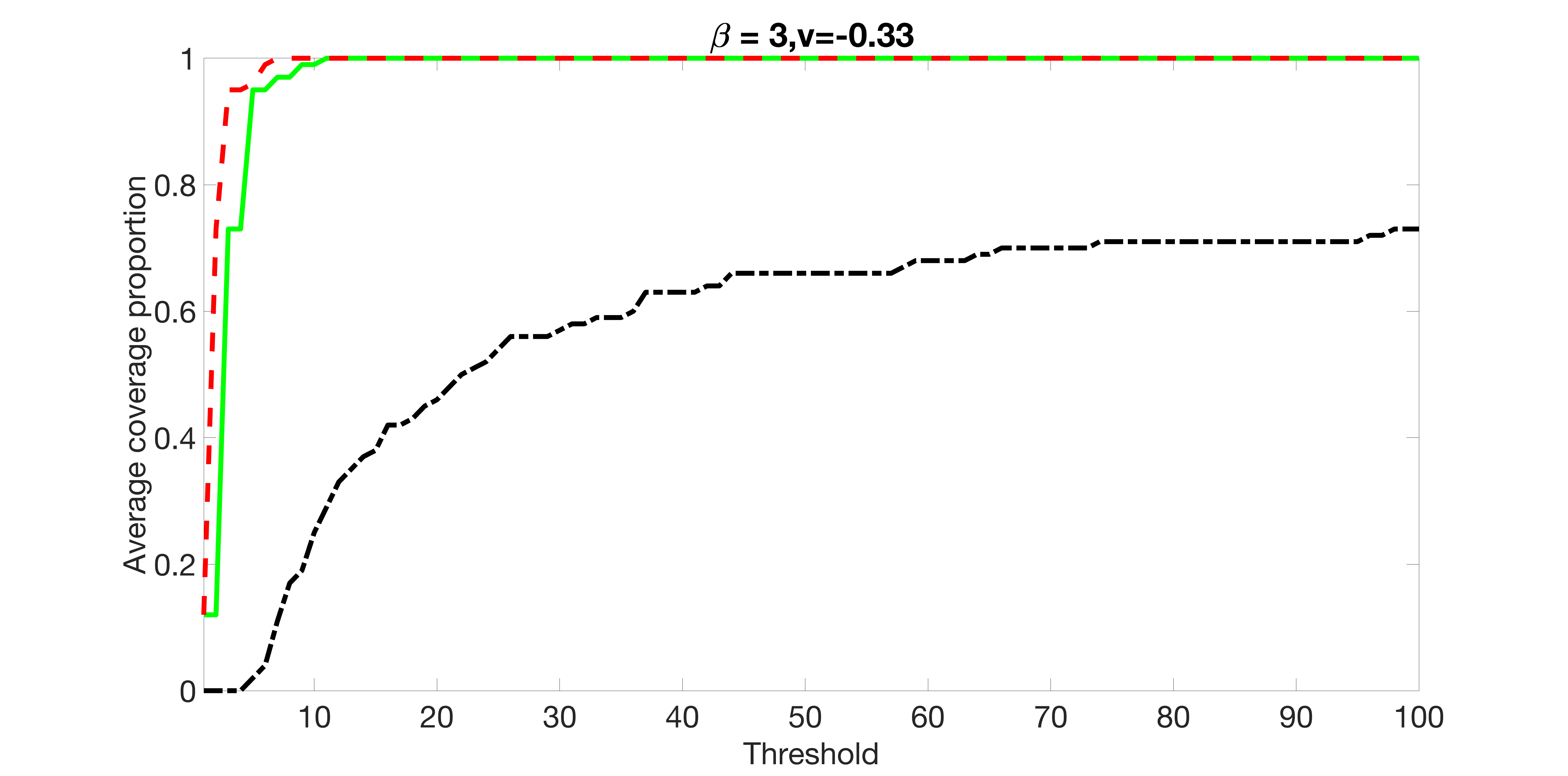}}
 \subcaptionbox{Confounder: medium \\ outcome, medium exposure}[0.45\linewidth]
 {\includegraphics[width=6cm,height=3.5cm]{./plotsMainArkSupp/sim1oal_p64_v2_optNorm2_n200s5000snry9sparsity_instru4count1.png}}
  \subcaptionbox{Confounder: weak \\ outcome, strong exposure}[0.45\linewidth]
 {\includegraphics[width=6cm,height=3.5cm]{./plotsMainArkSupp/sim1oal_p64_v2_optNorm2_n200s5000snry9sparsity_instru4count1.png}}
  \subcaptionbox{Precision: strong \\ outcome, zero exposure}[0.45\linewidth]
 {\includegraphics[width=6cm,height=3.5cm]{./plotsMainArkSupp/sim1oal_p64_v2_optNorm2_n200s5000snry9sparsity_instru4count1.png}}
  \subcaptionbox{Precision: medium \\ outcome, zero exposure}[0.45\linewidth]
 {\includegraphics[width=6cm,height=3.5cm]{./plotsMainArkSupp/sim1oal_p64_v2_optNorm2_n200s5000snry9sparsity_instru4count1.png}}
  \subcaptionbox{Precision: weak \\ outcome, zero exposure}[0.45\linewidth]
 {\includegraphics[width=6cm,height=3.5cm]{./plotsMainArkSupp/sim1oal_p64_v2_optNorm2_n200s5000snry9sparsity_instru4count1.png}}
  \subcaptionbox{Overall coverage of $\mathcal{M}_1$}[0.45\linewidth]
 {\includegraphics[width=6cm,height=3.5cm]{./plotsMainArkSupp/sim1oal_p64_v2_optNorm2_n200s5000sparsity_instru4coverage_snry9.png}}
\caption{Simulation results where the number of instrumental variables are twice of $\mathcal{M}_1 $ for the case $(n,s,\sigma) = (200,5000,1)$: Panels (a) -- (f) plot the average coverage proportion for $X_l$, where $l=1,2,3,104,105$ and $106$. Panels (a) -- (c) correspond to strong outcome and weak exposure predictor, moderate outcome and moderate exposure predictor and weak outcome and strong exposure predictor; Panels (d) -- (f) correspond to strong, moderate and weak predictors of outcome only. Panel (g) plots the average coverage proportion for the index set $\mathcal{M}_1 = \{1,2,3,104,105,106\}$. The x-axis represents the size of $\widehat{\mathcal{M}} $, while
y-axis denotes the average proportion. The green solid, the red dashed and the black dash dotted lines denote our joint screening method, the outcome screening method, and the intersection screening method, respectively. }
\label{sim1step1n200sigma1sparsity_instru4}
\end{figure}

\begin{figure}[htbp]
\captionsetup[subfigure]{justification=centering}
\centering
 \subcaptionbox{Confounder: strong \\ outcome, weak exposure}[0.45\linewidth]
 {\includegraphics[width=6cm,height=3.5cm]{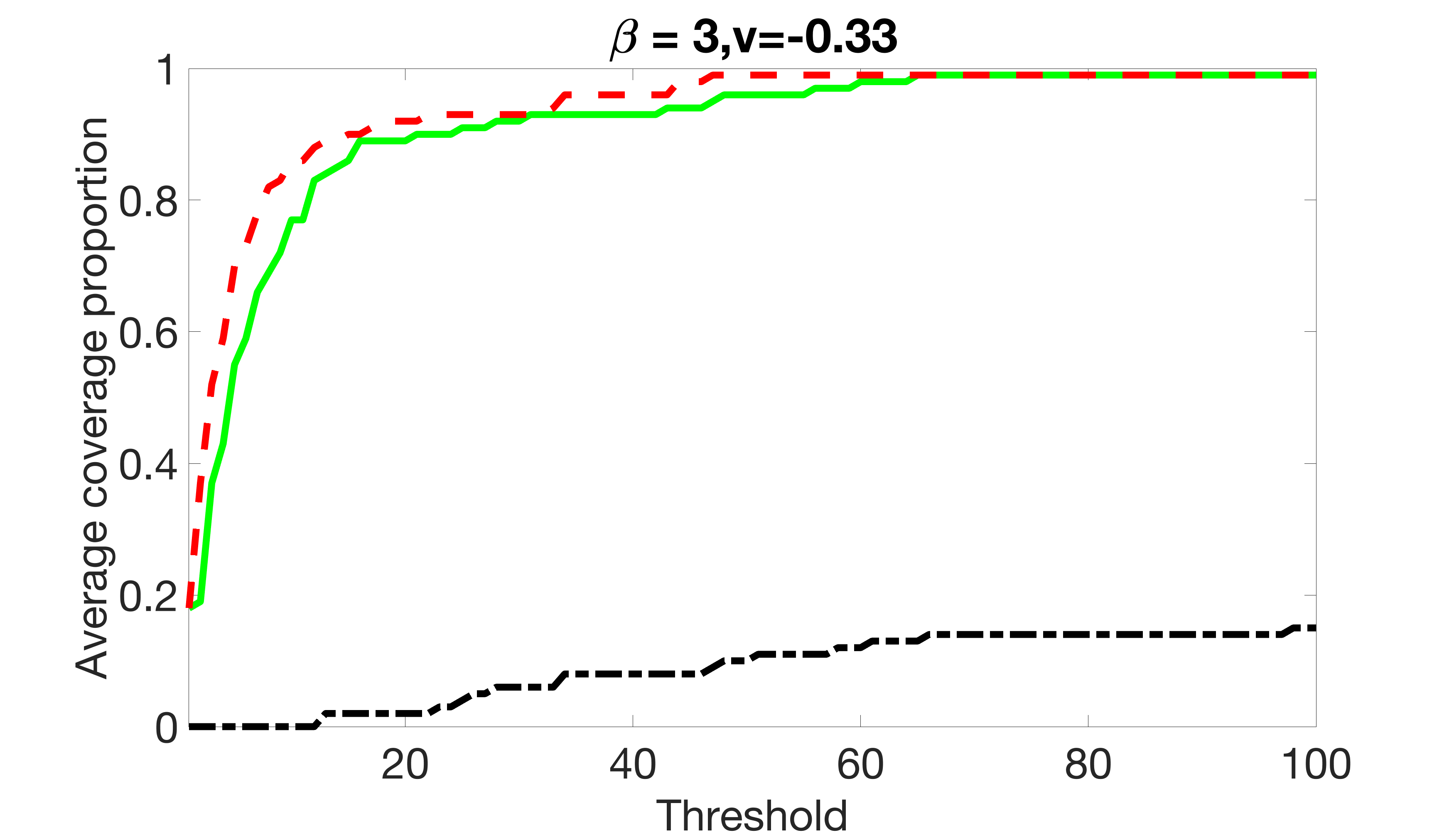}}
 \subcaptionbox{Confounder: medium \\ outcome, medium exposure}[0.45\linewidth]
 {\includegraphics[width=6cm,height=3.5cm]{./plotsMainArkSupp/sim1oal_p64_v2_optNorm2_n200s5000snry9sparsity_instru16count1.png}}
  \subcaptionbox{Confounder: weak \\ outcome, strong exposure}[0.45\linewidth]
 {\includegraphics[width=6cm,height=3.5cm]{./plotsMainArkSupp/sim1oal_p64_v2_optNorm2_n200s5000snry9sparsity_instru16count1.png}}
  \subcaptionbox{Precision: strong \\ outcome, zero exposure}[0.45\linewidth]
 {\includegraphics[width=6cm,height=3.5cm]{./plotsMainArkSupp/sim1oal_p64_v2_optNorm2_n200s5000snry9sparsity_instru16count1.png}}
  \subcaptionbox{Precision: medium \\ outcome, zero exposure}[0.45\linewidth]
 {\includegraphics[width=6cm,height=3.5cm]{./plotsMainArkSupp/sim1oal_p64_v2_optNorm2_n200s5000snry9sparsity_instru16count1.png}}
  \subcaptionbox{Precision: weak \\ outcome, zero exposure}[0.45\linewidth]
 {\includegraphics[width=6cm,height=3.5cm]{./plotsMainArkSupp/sim1oal_p64_v2_optNorm2_n200s5000snry9sparsity_instru16count1.png}}
  \subcaptionbox{Overall coverage of $\mathcal{M}_1$}[0.45\linewidth]
 {\includegraphics[width=6cm,height=3.5cm]{./plotsMainArkSupp/sim1oal_p64_v2_optNorm2_n200s5000sparsity_instru16coverage_snry9.png}}
\caption{Simulation results where the number of instrumental variables are eight times of $\mathcal{M}_1 $ for the case $(n,s,\sigma) = (200,5000,1)$: Panels (a) -- (f) plot the average coverage proportion for $X_l$, where $l=1,2,3,104,105$ and $106$. Panels (a) -- (c) correspond to strong outcome and weak exposure predictor, moderate outcome and moderate exposure predictor and weak outcome and strong exposure predictor; Panels (d) -- (f) correspond to strong, moderate and weak predictors of outcome only. Panel (g) plots the average coverage proportion for the index set $\mathcal{M}_1 = \{1,2,3,104,105,106\}$. The x-axis represents the size of $\widehat{\mathcal{M}} $, while
y-axis denotes the average proportion. The green solid, the red dashed and the black dash dotted lines denote our joint screening method, the outcome screening method, and the intersection screening method, respectively. }
\label{sim1step1n200sigma1sparsity_instru16}
\end{figure}

\begin{figure}[htbp]
\captionsetup[subfigure]{justification=centering}
\centering
 \subcaptionbox{Confounder: strong \\ outcome, weak exposure}[0.45\linewidth]
 {\includegraphics[width=6cm,height=3.5cm]{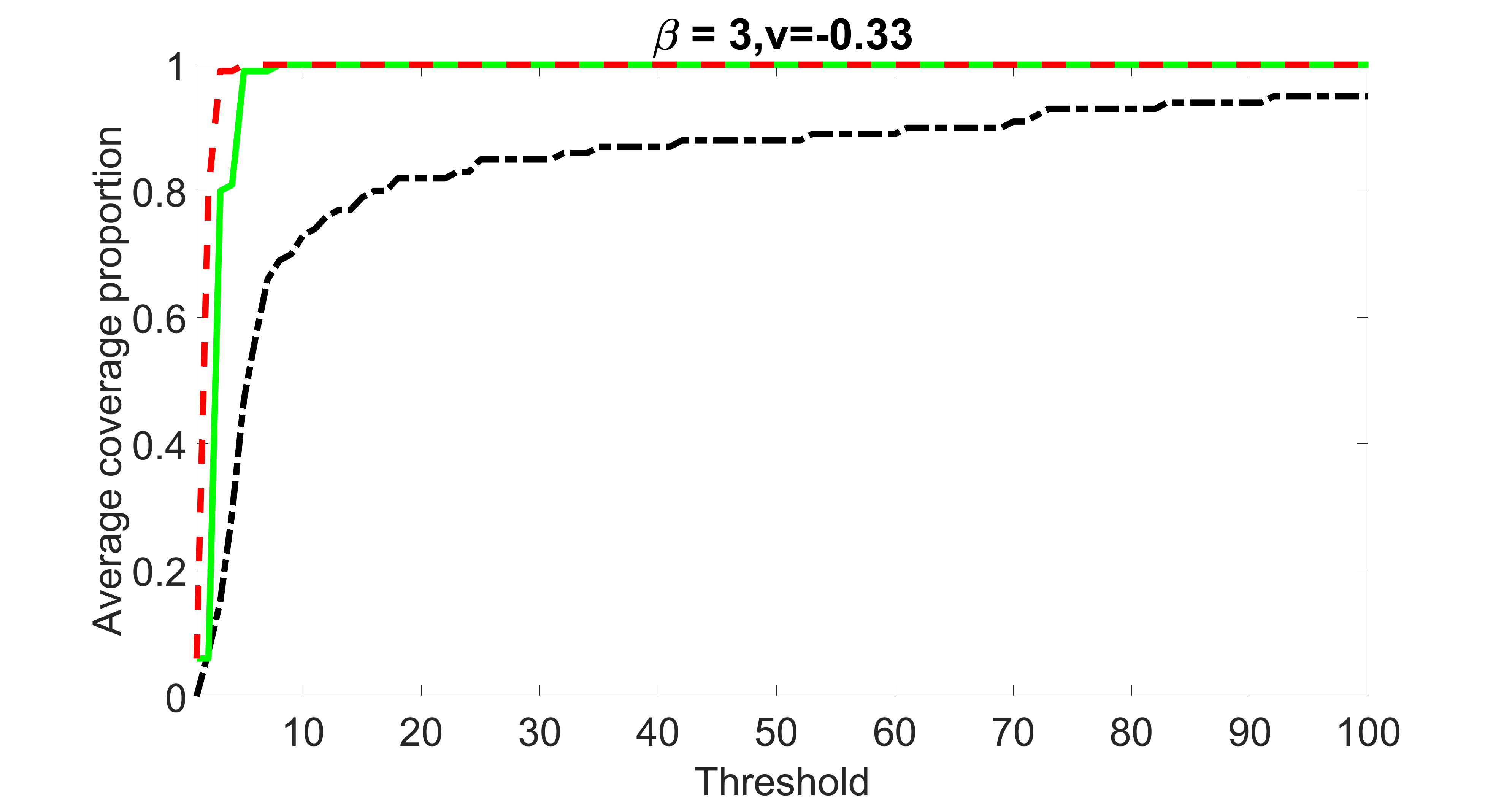}}
 \subcaptionbox{Confounder: medium \\ outcome, medium exposure}[0.45\linewidth]
 {\includegraphics[width=6cm,height=3.5cm]{./plotsMainArkSupp/sim1oal_p64_v2_optNorm2_n200s5000snry8sparsity_instru2count1.png}}
  \subcaptionbox{Confounder: weak \\ outcome, strong exposure}[0.45\linewidth]
 {\includegraphics[width=6cm,height=3.5cm]{./plotsMainArkSupp/sim1oal_p64_v2_optNorm2_n200s5000snry8sparsity_instru2count1.png}}
  \subcaptionbox{Precision: strong \\ outcome, zero exposure}[0.45\linewidth]
 {\includegraphics[width=6cm,height=3.5cm]{./plotsMainArkSupp/sim1oal_p64_v2_optNorm2_n200s5000snry8sparsity_instru2count1.png}}
  \subcaptionbox{Precision: medium \\ outcome, zero exposure}[0.45\linewidth]
 {\includegraphics[width=6cm,height=3.5cm]{./plotsMainArkSupp/sim1oal_p64_v2_optNorm2_n200s5000snry8sparsity_instru2count1.png}}
  \subcaptionbox{Precision: weak \\ outcome, zero exposure}[0.45\linewidth]
 {\includegraphics[width=6cm,height=3.5cm]{./plotsMainArkSupp/sim1oal_p64_v2_optNorm2_n200s5000snry8sparsity_instru2count1.png}}
  \subcaptionbox{Overall coverage of $\mathcal{M}_1$}[0.45\linewidth]
 {\includegraphics[width=6cm,height=3.5cm]{./plotsMainArkSupp/sim1oal_p64_v2_optNorm2_n200s5000sparsity_instru2coverage_snry8.png}}
\caption{Simulation results where the number of instrumental variables are the same of $\mathcal{M}_1 $ for the case $(n,s,\sigma) = (200,5000,0.5)$: Panels (a) -- (f) plot the average coverage proportion for $X_l$, where $l=1,2,3,104,105$ and $106$. Panels (a) -- (c) correspond to strong outcome and weak exposure predictor, moderate outcome and moderate exposure predictor and weak outcome and strong exposure predictor; Panels (d) -- (f) correspond to strong, moderate and weak predictors of outcome only. Panel (g) plots the average coverage proportion for the index set $\mathcal{M}_1 = \{1,2,3,104,105,106\}$. The x-axis represents the size of $\widehat{\mathcal{M}} $, while
y-axis denotes the average proportion. The green solid, the red dashed and the black dash dotted lines denote our joint screening method, the outcome screening method, and the intersection screening method, respectively. }
\label{sim1step1n200sigma025sparsity_instru2}
\end{figure}

\begin{figure}[htbp]
\captionsetup[subfigure]{justification=centering}
\centering
 \subcaptionbox{Confounder: strong \\ outcome, weak exposure}[0.45\linewidth]
 {\includegraphics[width=6cm,height=3.5cm]{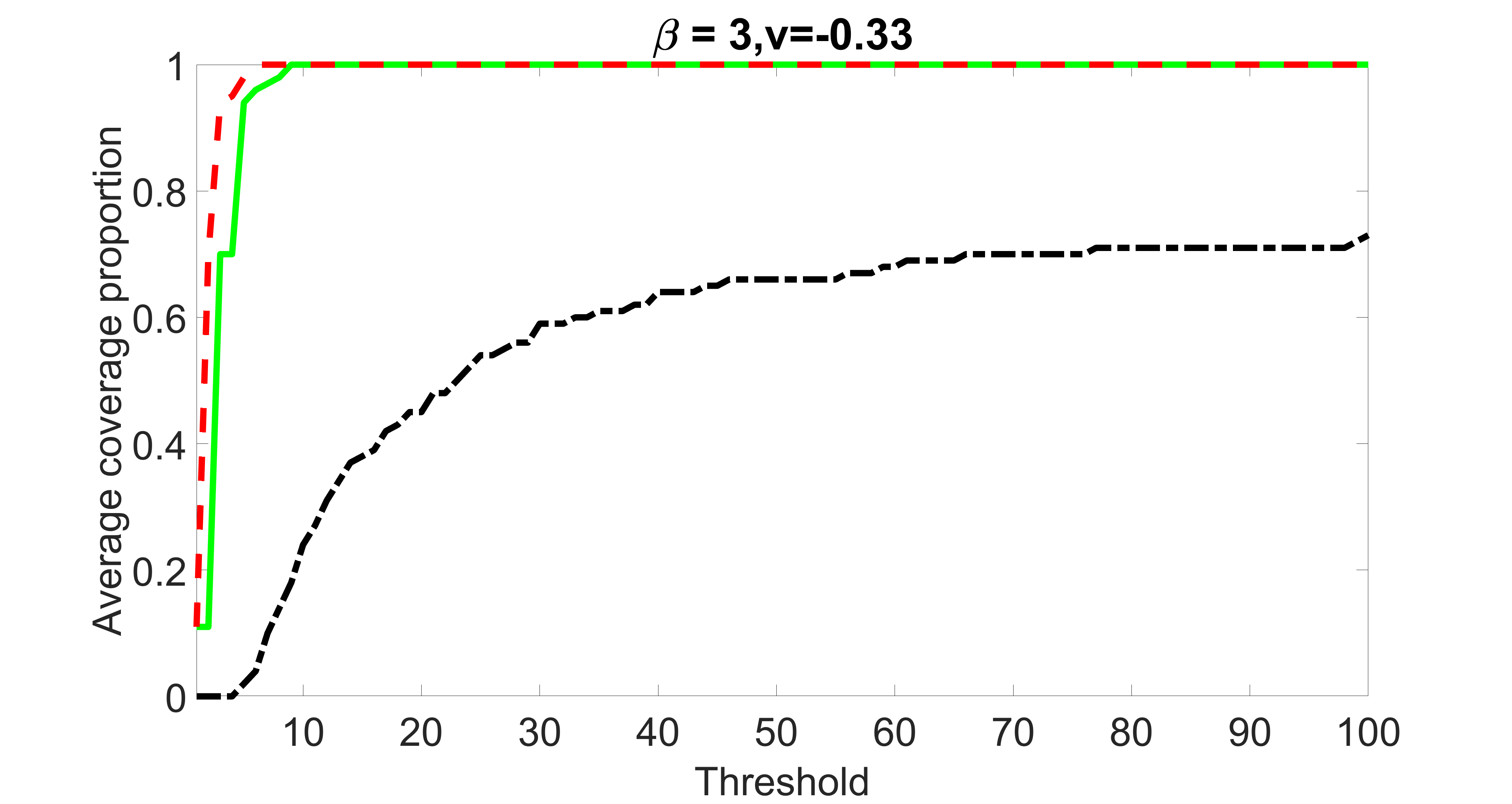}}
 \subcaptionbox{Confounder: medium \\ outcome, medium exposure}[0.45\linewidth]
 {\includegraphics[width=6cm,height=3.5cm]{./plotsMainArkSupp/sim1oal_p64_v2_optNorm2_n200s5000snry8sparsity_instru4count1.png}}
  \subcaptionbox{Confounder: weak \\ outcome, strong exposure}[0.45\linewidth]
 {\includegraphics[width=6cm,height=3.5cm]{./plotsMainArkSupp/sim1oal_p64_v2_optNorm2_n200s5000snry8sparsity_instru4count1.png}}
  \subcaptionbox{Precision: strong \\ outcome, zero exposure}[0.45\linewidth]
 {\includegraphics[width=6cm,height=3.5cm]{./plotsMainArkSupp/sim1oal_p64_v2_optNorm2_n200s5000snry8sparsity_instru4count1.png}}
  \subcaptionbox{Precision: medium \\ outcome, zero exposure}[0.45\linewidth]
 {\includegraphics[width=6cm,height=3.5cm]{./plotsMainArkSupp/sim1oal_p64_v2_optNorm2_n200s5000snry8sparsity_instru4count1.png}}
  \subcaptionbox{Precision: weak \\ outcome, zero exposure}[0.45\linewidth]
 {\includegraphics[width=6cm,height=3.5cm]{./plotsMainArkSupp/sim1oal_p64_v2_optNorm2_n200s5000snry8sparsity_instru4count1.png}}
  \subcaptionbox{Overall coverage of $\mathcal{M}_1$}[0.45\linewidth]
 {\includegraphics[width=6cm,height=3.5cm]{./plotsMainArkSupp/sim1oal_p64_v2_optNorm2_n200s5000sparsity_instru4coverage_snry8.png}}
\caption{Simulation results where the number of instrumental variables are twice of $\mathcal{M}_1 $ for the case $(n,s,\sigma) = (200,5000,0.5)$: Panels (a) -- (f) plot the average coverage proportion for $X_l$, where $l=1,2,3,104,105$ and $106$. Panels (a) -- (c) correspond to strong outcome and weak exposure predictor, moderate outcome and moderate exposure predictor and weak outcome and strong exposure predictor; Panels (d) -- (f) correspond to strong, moderate and weak predictors of outcome only. Panel (g) plots the average coverage proportion for the index set $\mathcal{M}_1 = \{1,2,3,104,105,106\}$. The x-axis represents the size of $\widehat{\mathcal{M}} $, while
y-axis denotes the average proportion. The green solid, the red dashed and the black dash dotted lines denote our joint screening method, the outcome screening method, and the intersection screening method, respectively. }
\label{sim1step1n200sigma025sparsity_instru4}
\end{figure}

\begin{figure}[htbp]
\captionsetup[subfigure]{justification=centering}
\centering
 \subcaptionbox{Confounder: strong \\ outcome, weak exposure}[0.45\linewidth]
 {\includegraphics[width=6cm,height=3.5cm]{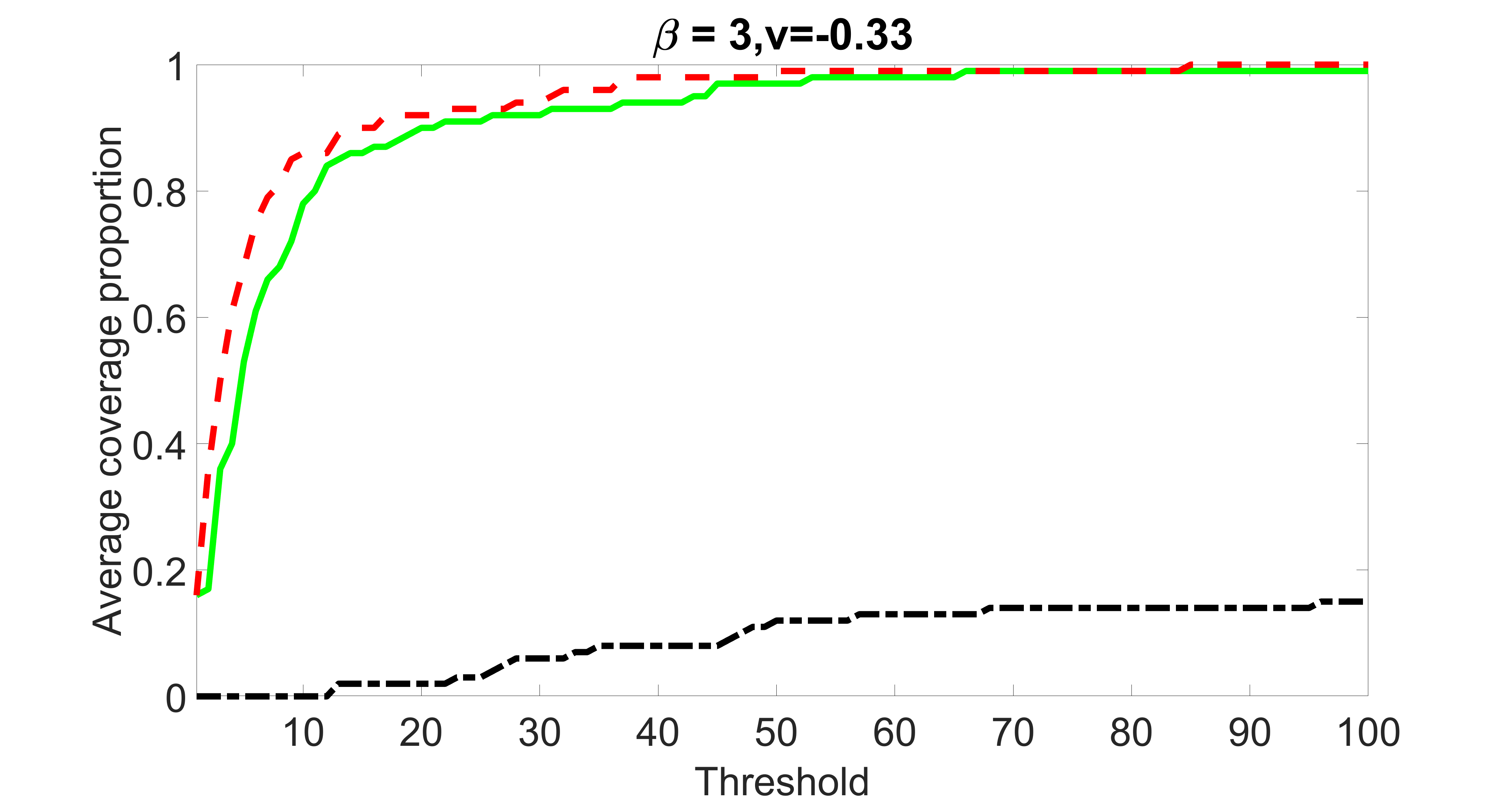}}
 \subcaptionbox{Confounder: medium \\ outcome, medium exposure}[0.45\linewidth]
 {\includegraphics[width=6cm,height=3.5cm]{./plotsMainArkSupp/sim1oal_p64_v2_optNorm2_n200s5000snry8sparsity_instru16count1.png}}
  \subcaptionbox{Confounder: weak \\ outcome, strong exposure}[0.45\linewidth]
 {\includegraphics[width=6cm,height=3.5cm]{./plotsMainArkSupp/sim1oal_p64_v2_optNorm2_n200s5000snry8sparsity_instru16count1.png}}
  \subcaptionbox{Precision: strong \\ outcome, zero exposure}[0.45\linewidth]
 {\includegraphics[width=6cm,height=3.5cm]{./plotsMainArkSupp/sim1oal_p64_v2_optNorm2_n200s5000snry8sparsity_instru16count1.png}}
  \subcaptionbox{Precision: medium \\ outcome, zero exposure}[0.45\linewidth]
 {\includegraphics[width=6cm,height=3.5cm]{./plotsMainArkSupp/sim1oal_p64_v2_optNorm2_n200s5000snry8sparsity_instru16count1.png}}
  \subcaptionbox{Precision: weak \\ outcome, zero exposure}[0.45\linewidth]
 {\includegraphics[width=6cm,height=3.5cm]{./plotsMainArkSupp/sim1oal_p64_v2_optNorm2_n200s5000snry8sparsity_instru16count1.png}}
  \subcaptionbox{Overall coverage of $\mathcal{M}_1$}[0.45\linewidth]
 {\includegraphics[width=6cm,height=3.5cm]{./plotsMainArkSupp/sim1oal_p64_v2_optNorm2_n200s5000sparsity_instru16coverage_snry8.png}}
\caption{Simulation results where the number of instrumental variables are eight times of $\mathcal{M}_1 $ for the case $(n,s,\sigma) = (200,5000,0.5)$: Panels (a) -- (f) plot the average coverage proportion for $X_l$, where $l=1,2,3,104,105$ and $106$. Panels (a) -- (c) correspond to strong outcome and weak exposure predictor, moderate outcome and moderate exposure predictor and weak outcome and strong exposure predictor; Panels (d) -- (f) correspond to strong, moderate and weak predictors of outcome only. Panel (g) plots the average coverage proportion for the index set $\mathcal{M}_1 = \{1,2,3,104,105,106\}$. The x-axis represents the size of $\widehat{\mathcal{M}} $, while
y-axis denotes the average proportion. The green solid, the red dashed and the black dash dotted lines denote our joint screening method, the outcome screening method, and the intersection screening method, respectively. }
\label{sim1step1n200sigma025sparsity_instru16}
\end{figure}

\subsection{Screening and estimation under different covariances of exposure errors}
\label{Screening under different covariance structures of exposure errors}
We also consider various exposure errors in the simulation. We use the same setting as Section \ref{Simulation for screening} of the main paper but taking three different covariance structures of exposure errors $\bm{E}_{i}$. In particular, the random error $\mathrm{vec}(\bm{E}_i)$ is independently generated from $N(\bm{0},\bm\Sigma_e)$, where we set the standard deviations of all elements in  $\bm{E}_i$ to be $\sigma_e = 0.2$ and the correlation between $\bm{E}_{i,jk}$ and $\bm{E}_{i,j^\prime k^\prime}$ to be $\rho_2^{|j - j^\prime| + |k-k^\prime|}$ for $1 \leq j,k,j^\prime, k^\prime \leq 64$ with $\rho_2$. We consider three scenarios that $\rho_2 = 0.2$, $0.5$, and $0.8$, and report the selected covariates from the screening step. We consider $ \sigma=1$ or $0.5$ and fix the sample size $n=200$. The complete screening results can be found in Figure \ref{sim1step1n200sigma1} as well as Figures \ref{sim1step1n200sigma025}, \ref{sim1step1n200sigma1rho202} -- \ref{sim1step1n200sigma025rho208} here ($\rho_2$ is set to be $0.5$ in Section \ref{Simulation for screening} of the main paper).

Specifically, as summarized in Figure \ref{sim1step1n200rho2_summary}, when $\rho_2$ increases, the average coverage proportion for the index set $\mathcal{M}_1 = \{1,2,3,104,105,106\}$ does not change too much. 

\begin{figure}[htbp]
\captionsetup[subfigure]{justification=centering}
\centering
 \subcaptionbox{$\rho_2 = 0.2$, $\sigma=1$}[0.45\linewidth]
 {\includegraphics[width=7cm,height=10cm,keepaspectratio]{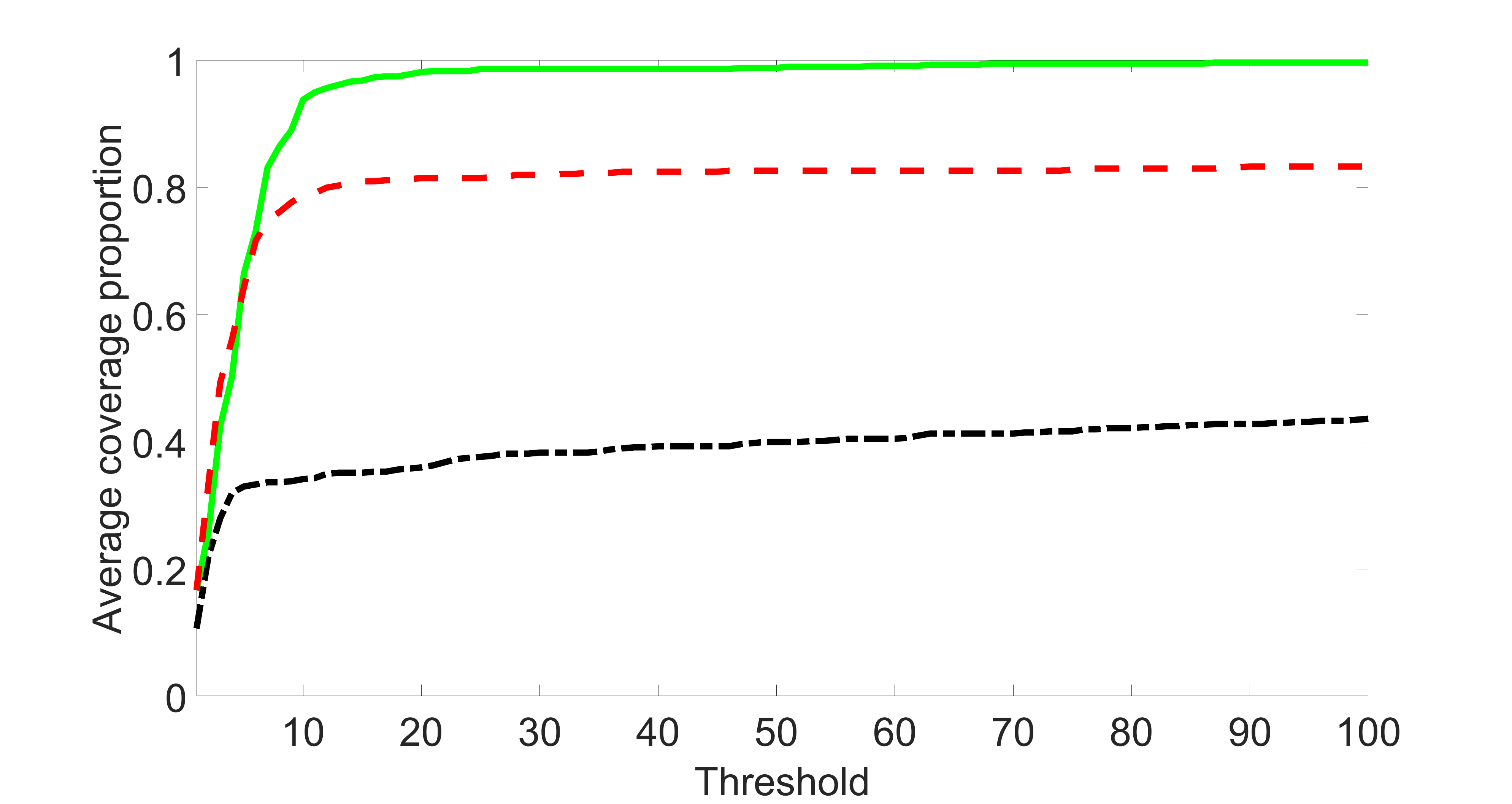}}
 \subcaptionbox{$\rho_2 = 0.2$, $\sigma = 0.5$}[0.45\linewidth]
 {\includegraphics[width=7cm,height=10cm,keepaspectratio]{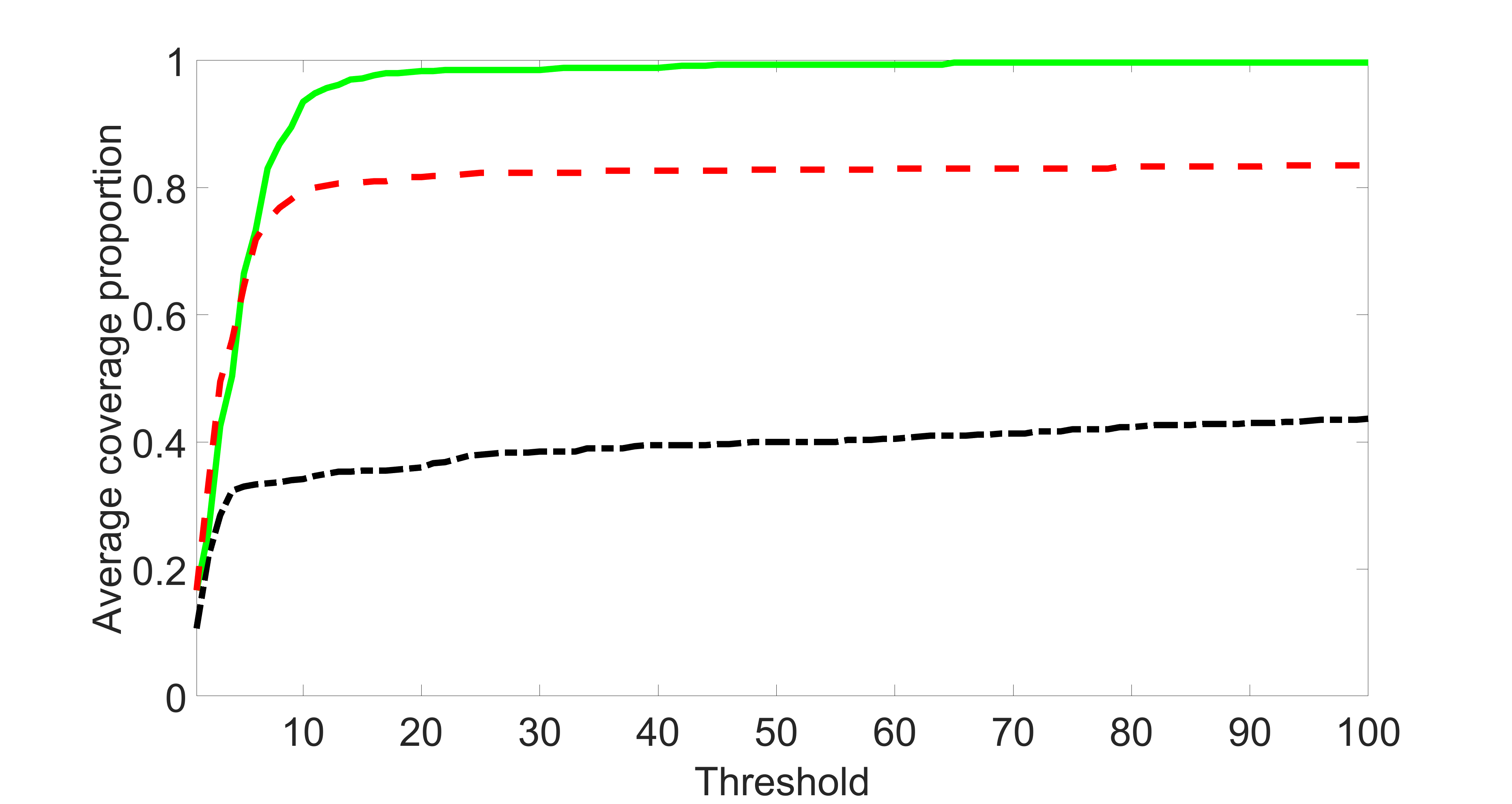}}
  \subcaptionbox{$\rho_2 = 0.5$, $\sigma=1$}[0.45\linewidth]
 {\includegraphics[width=7cm,height=10cm,keepaspectratio]{./plotsMainArkSupp/sim1oal_p64_v2_optNorm2_n200s5000coverageoptB2C1snry9.png}}
  \subcaptionbox{$\rho_2 = 0.5$, $\sigma = 0.5$}[0.45\linewidth]
 {\includegraphics[width=7cm,height=10cm,keepaspectratio]{./plotsMainArkSupp/sim1oal_p64_v2_optNorm2_n200s5000coverageoptB2C1snry8.png}}
  \subcaptionbox{$\rho_2 = 0.8$, $\sigma=1$}[0.45\linewidth]
 {\includegraphics[width=7cm,height=3.75cm]{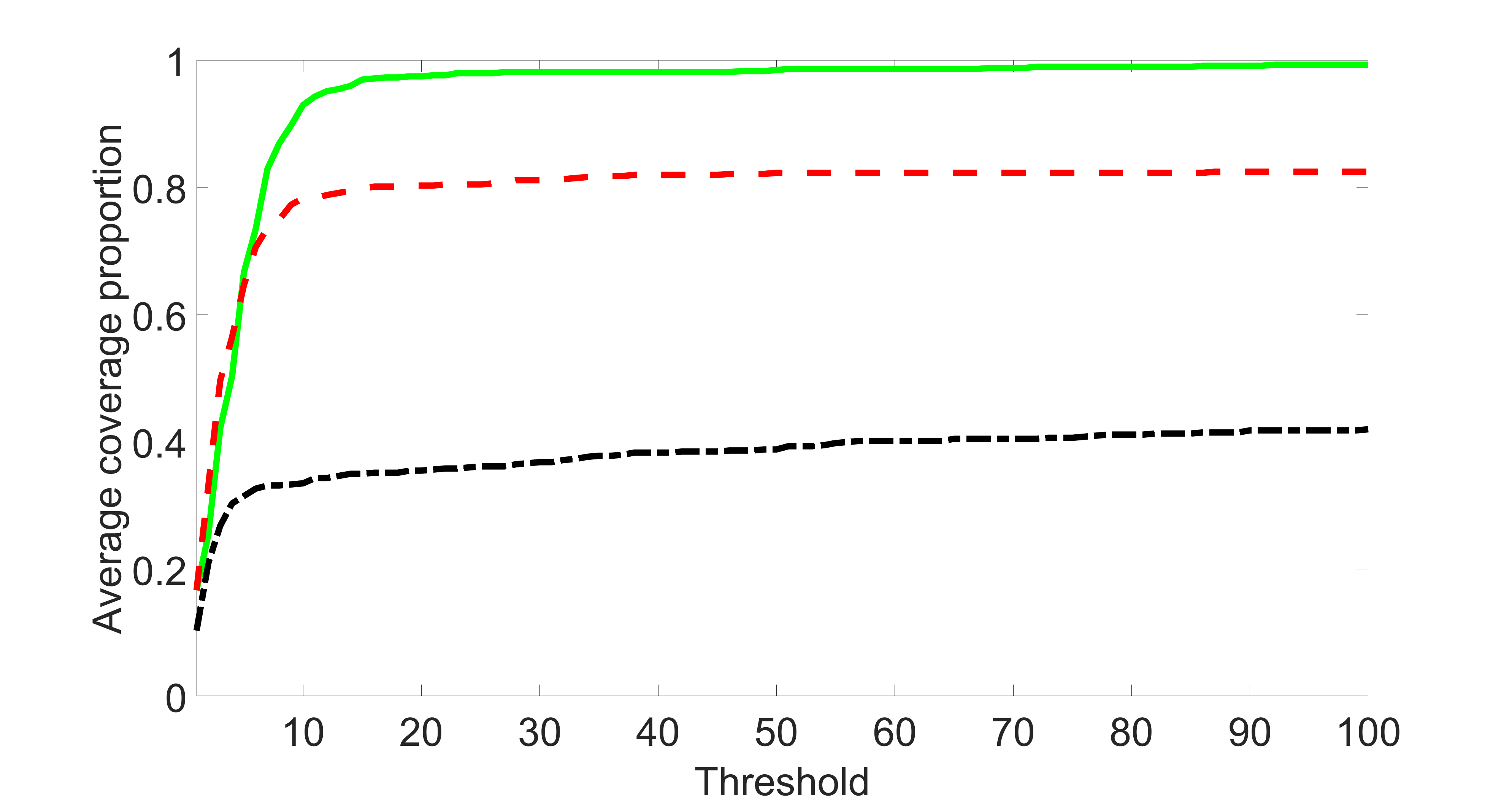}}
  \subcaptionbox{$\rho_2 = 0.8$, $\sigma = 0.5$}[0.45\linewidth]
 {\includegraphics[width=7cm,height=10cm,keepaspectratio]{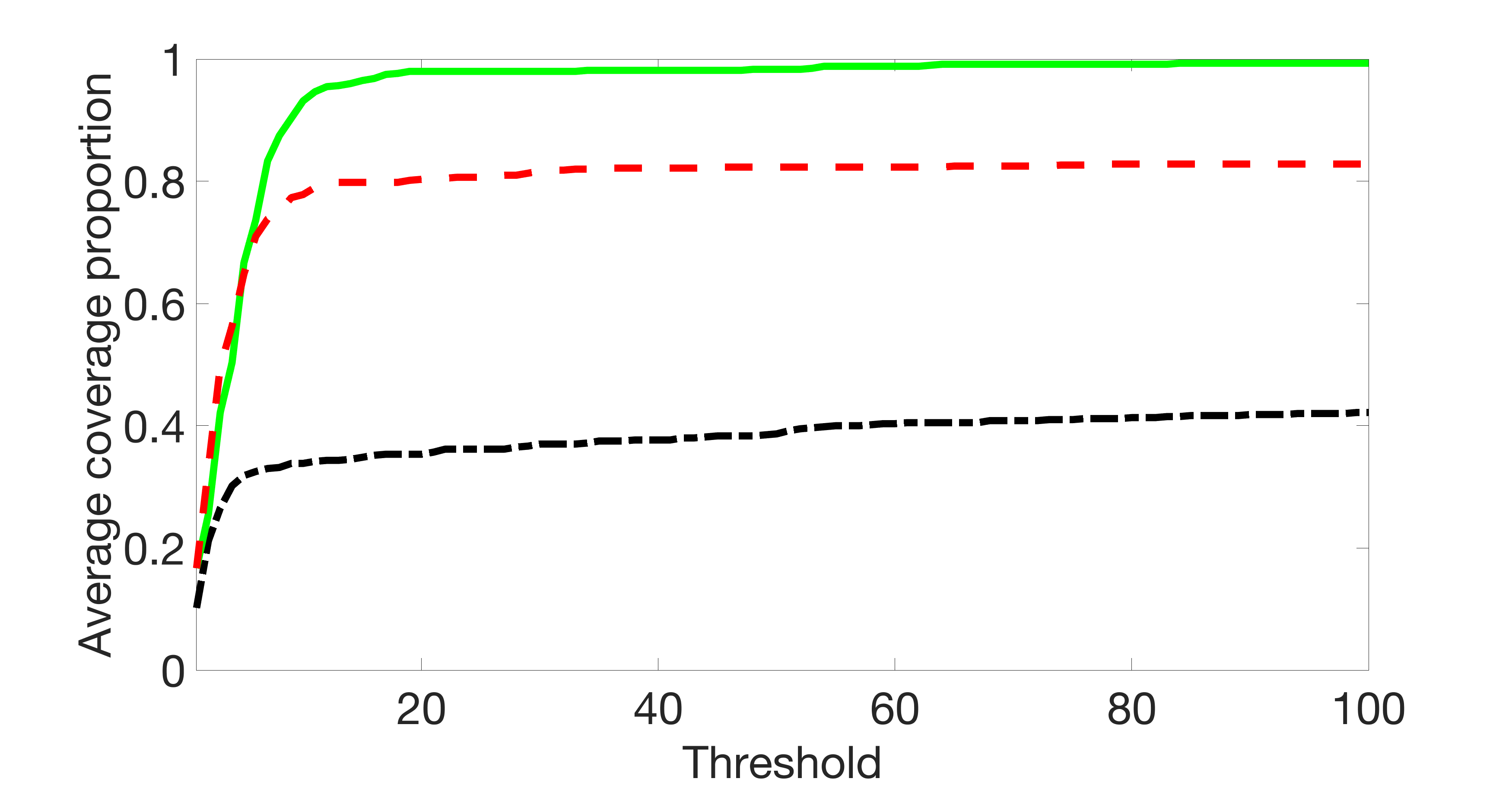}}
\caption{Simulation results where the size of instrumental variables $\mathcal{I}$ are the same, twice and eight times of $\mathcal{M}_1 $ for the case $(n,s) = (200,5000)$: Panels (a) (c) (e) plot the average coverage proportion for the index set $\mathcal{M}_1 = \{1,2,3,104,105,106\}$ when $\sigma=1$. Panels (b) (d) (f) plot the average coverage proportion for the index set $\mathcal{M}_1 = \{1,2,3,104,105,106\}$ when $\sigma = 0.5$. The x-axis represents the size of $\widehat{\mathcal{M}} $, while
y-axis denotes the average proportion. The green solid, the red dashed and the black dash dotted lines denote our joint screening method, the outcome screening method, and the intersection screening method, respectively. }
\label{sim1step1n200rho2_summary}
\end{figure}

In addition, we evaluate the performance of our estimation procedure after the first-step screening. For the size of $ \widehat{\mathcal{M}} $ in the screening step, we set $|\widehat{\mathcal{M}}| = \lfloor n / \log(n) \rfloor$, so that $|\widehat{\mathcal{M}}|=37$ for sample size $n = 200$. We report the mean squared errors (MSEs) for $\bm{\beta}$ and $\bm{B}$ defined as
$||{\bm\beta}_{}-\widehat{\bm\beta}||_2^2$ and $ \|{\bm{B}}-\widehat{\bm{B}}\|_F^2$, respectively.

Table \ref{sim1rho2t1} summarizes the average MSEs for $\bm\beta$ and $\bm{B}$ among 100 Monte Carlo runs. We can see that the MSE decreases with $\rho_2$ increases. As nuclear norm penalization estimation procedure can be regarded as one way of spatial smoothing, the large correlations among $\bm{E}_{i}$ actually help with the spatial smoothing, and thus we have better estimation accuracy when $\rho_2$ increases.

\begin{table}[htbp]
\centering
\caption{Simulation results for $ \sigma=1 $ and $ \sigma = 0.5 $: the average MSEs for $\bm\beta$ and $ {\bm B}$, and their associated standard errors in the parentheses are reported. The results are based on 100 Monte Carlo repetitions.}
\begin{tabular}{ lcc | lcccc }
$\sigma = 1.0$ & MSE $\bm\beta$ & MSE ${\bm{B}}$ &$\sigma = 0.5$  &MSE $\bm\beta$ & MSE ${\bm{B}}$  \\
\hline
$\rho_2=0.2$   &0.986(0.099)&0.802(0.011)&  $\rho_2=0.2$   &0.397(0.042)&0.701(0.006)\\
$\rho_2=0.5$  &0.496(0.021)&0.667(0.005)&  $\rho_2=0.5$  &0.276(0.009)&0.528(0.005)\\
$\rho_2=0.8 $  &0.252(0.010)&0.439(0.007)&  $\rho_2=0.8$  &0.097(0.006)&0.305(0.004)\\
\end{tabular}
\label{sim1rho2t1}
\end{table}

\begin{figure}[htbp]
\captionsetup[subfigure]{justification=centering}
\centering
 \subcaptionbox{Confounder: strong \\ outcome, weak exposure}[0.45\linewidth]
 {\includegraphics[width=6cm,height=3.5cm]{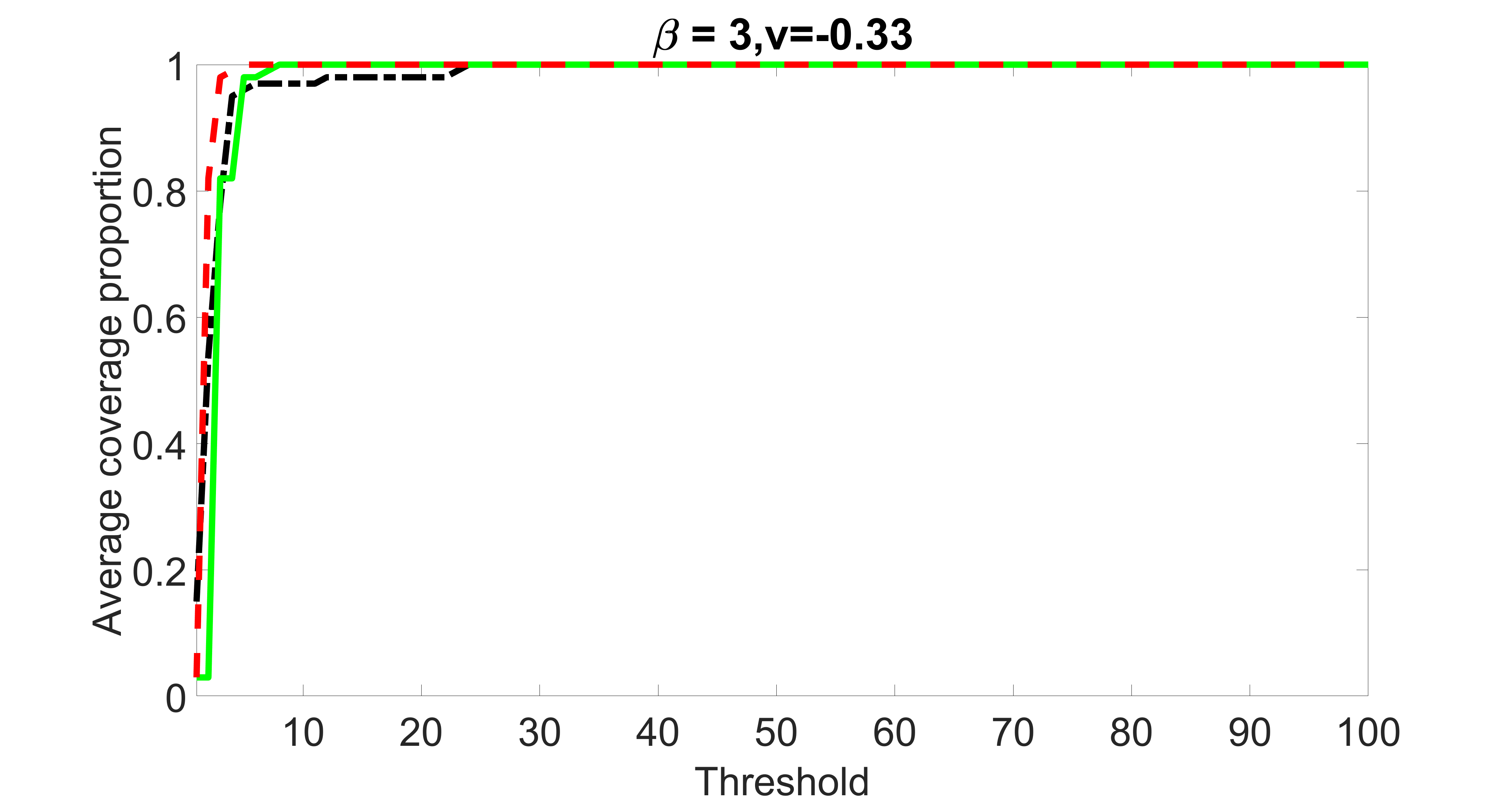}}
 \subcaptionbox{Confounder: medium \\ outcome, medium exposure}[0.45\linewidth]
 {\includegraphics[width=6cm,height=3.5cm]{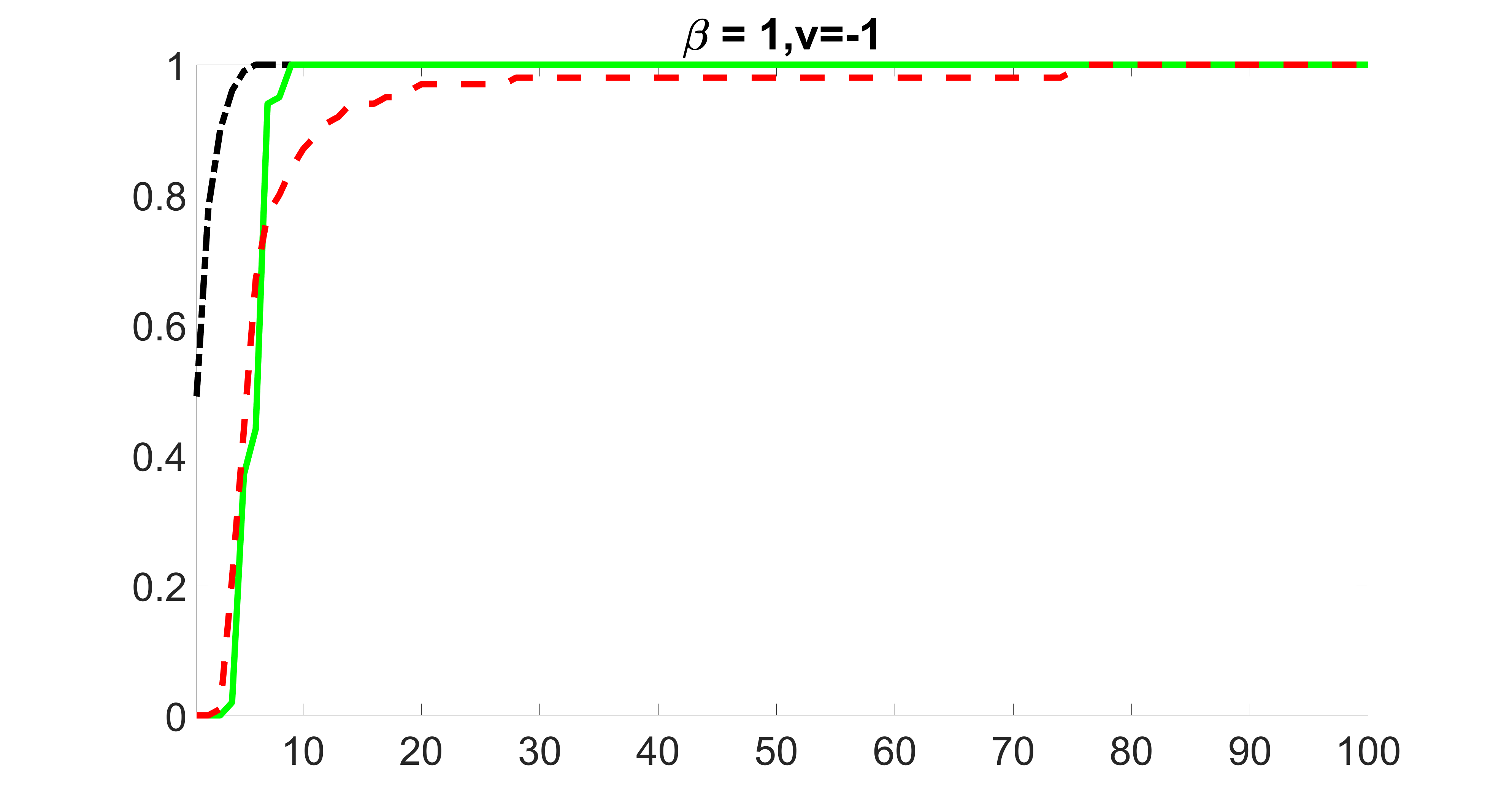}}
  \subcaptionbox{Confounder: weak \\ outcome, strong exposure}[0.45\linewidth]
 {\includegraphics[width=6cm,height=3.5cm]{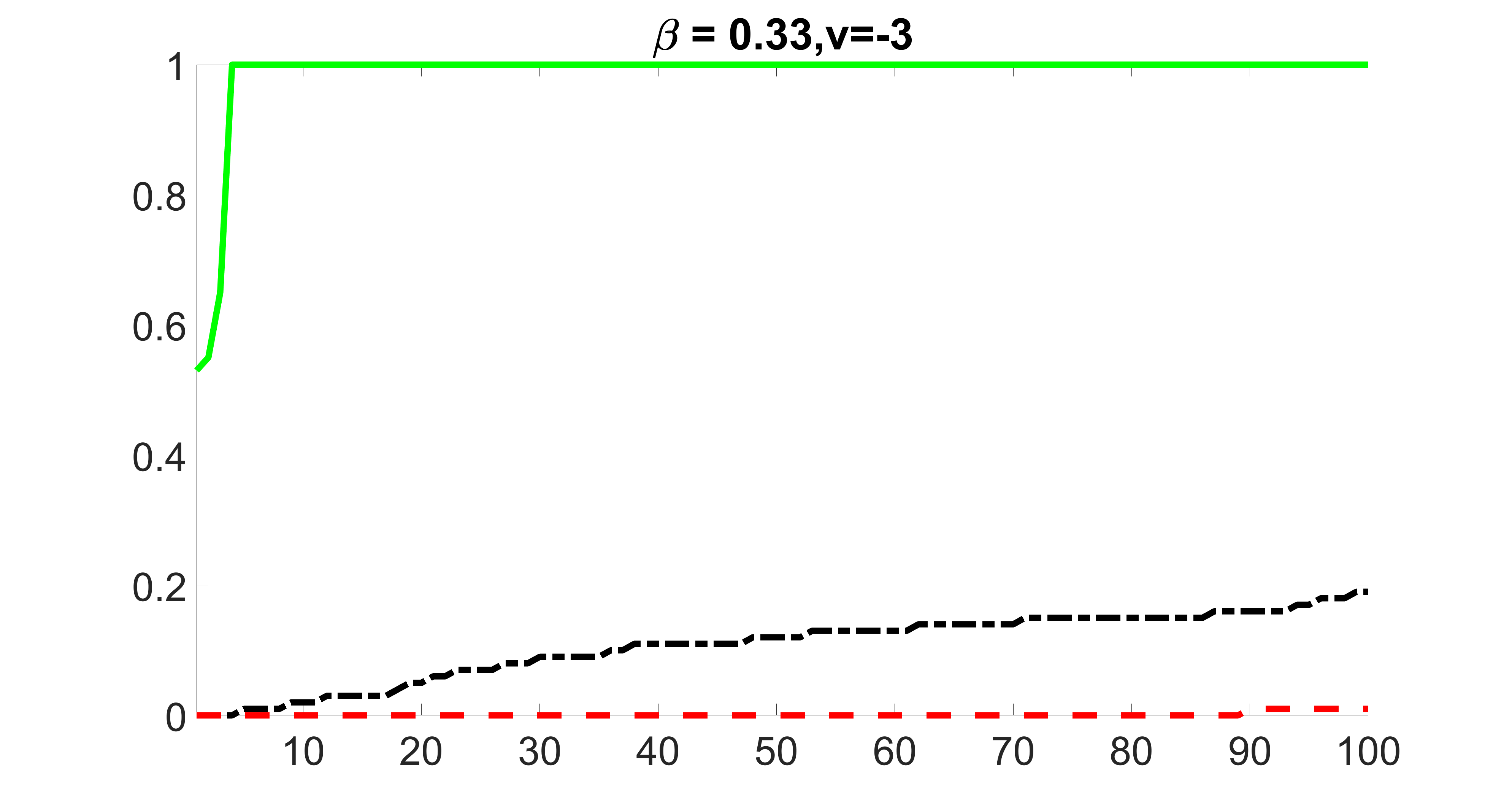}}
  \subcaptionbox{Precision: strong \\ outcome, zero exposure}[0.45\linewidth]
 {\includegraphics[width=6cm,height=3.5cm]{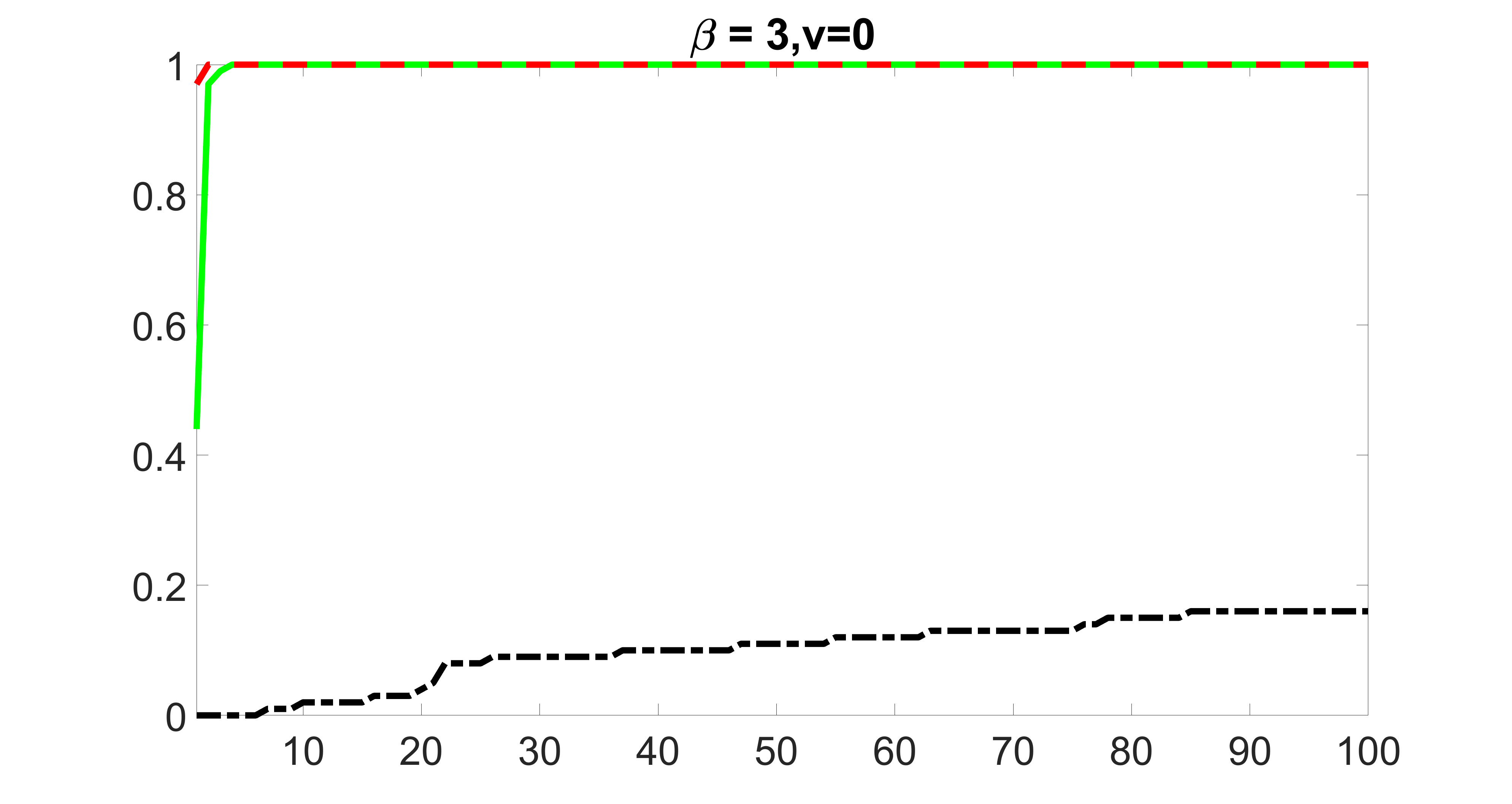}}
  \subcaptionbox{Precision: medium \\ outcome, zero exposure}[0.45\linewidth]
 {\includegraphics[width=6cm,height=3.5cm]{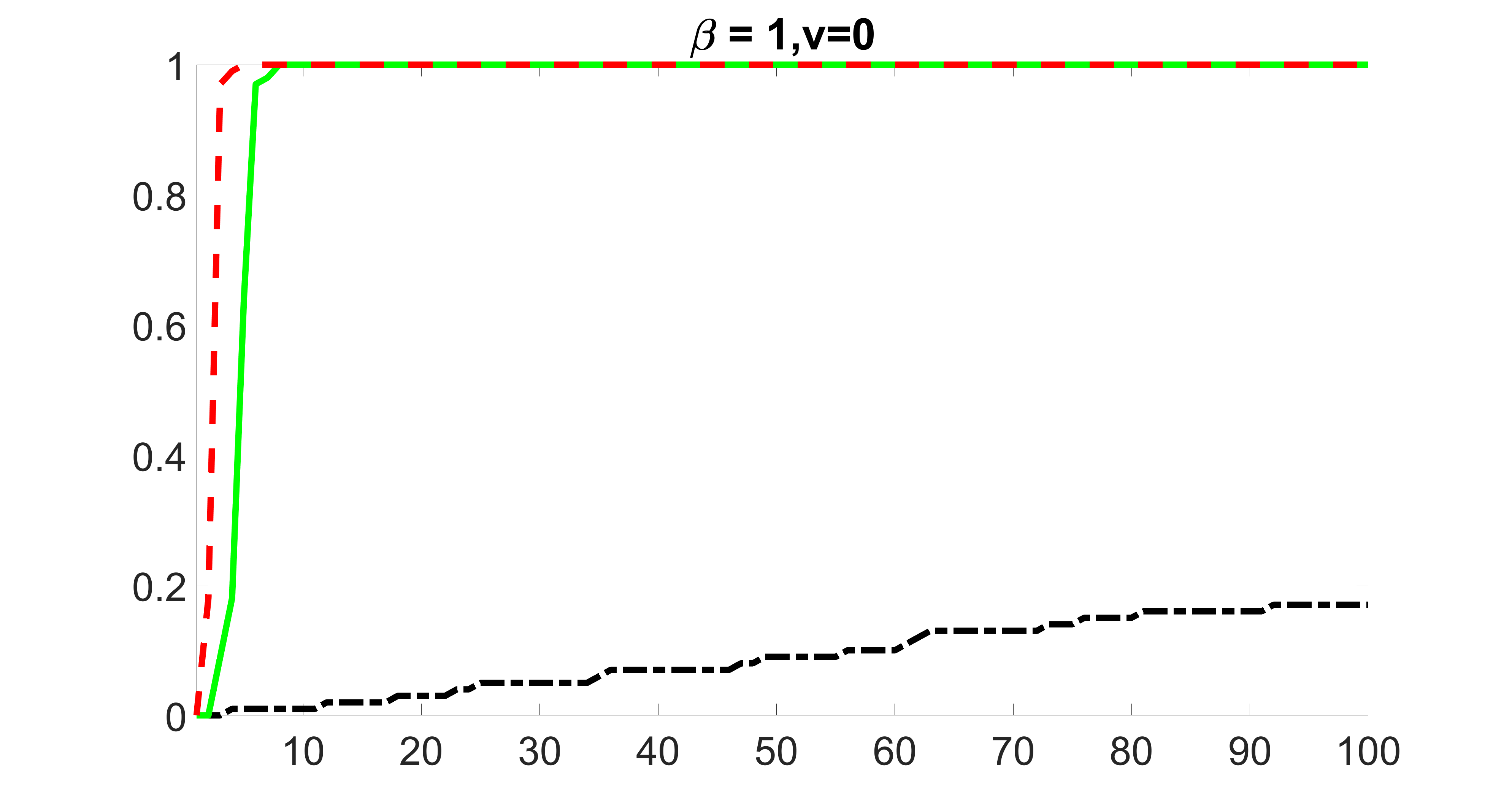}}
  \subcaptionbox{Precision: weak \\ outcome, zero exposure}[0.45\linewidth]
 {\includegraphics[width=6cm,height=3.5cm]{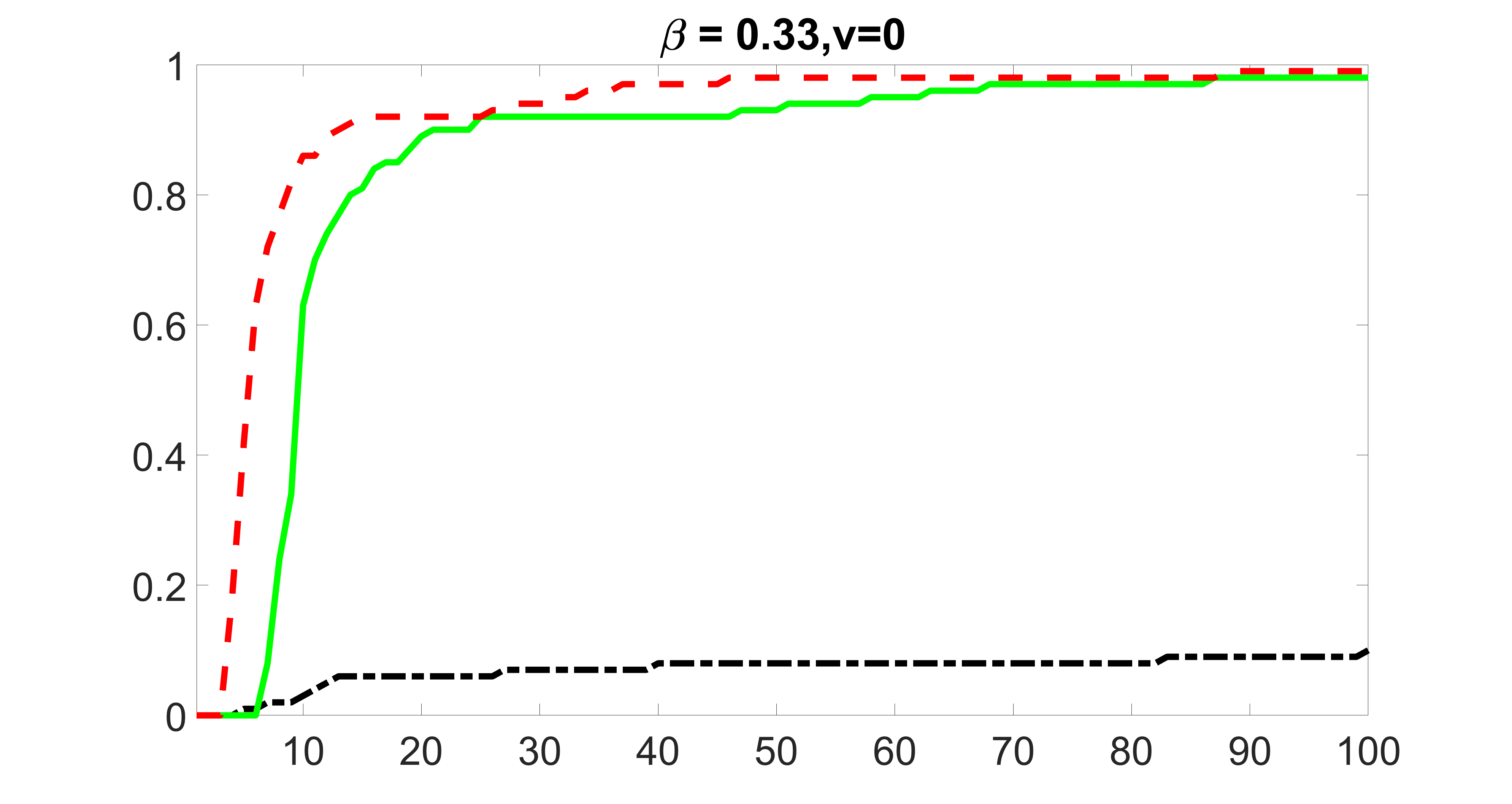}}
  \subcaptionbox{Overall coverage of $\mathcal{M}_1$}[0.45\linewidth]
 {\includegraphics[width=6cm,height=3.5cm]{./plotsMainArkSupp/sim1oal_p64_v2_optNorm2_n200s5000coverage_snry9rho202.png}}
\caption{Simulation results for the case $(n,s,\sigma,\rho_2) = (200,5000,1,0.2)$: Panels (a) -- (f) plot the average coverage proportion for $X_l$, where $l=1,2,3,104,105$ and $106$. Panels (a) -- (c) correspond to strong outcome and weak exposure predictor, moderate outcome and moderate exposure predictor and weak outcome and strong exposure predictor; Panels (d) -- (f) correspond to strong, moderate and weak predictors of outcome only. Panel (g) plots the average coverage proportion for the index set $\mathcal{M}_1 = \{1,2,3,104,105,106\}$. The x-axis represents the size of $\widehat{\mathcal{M}} $, while
y-axis denotes the average proportion. The green solid, the red dashed and the black dash dotted lines denote our joint screening method, the outcome screening method, and the intersection screening method, respectively. }
\label{sim1step1n200sigma1rho202}
\end{figure}

\begin{figure}[htbp]
\captionsetup[subfigure]{justification=centering}
\centering
 \subcaptionbox{Confounder: strong \\ outcome, weak exposure}[0.45\linewidth]
 {\includegraphics[width=6cm,height=3.5cm]{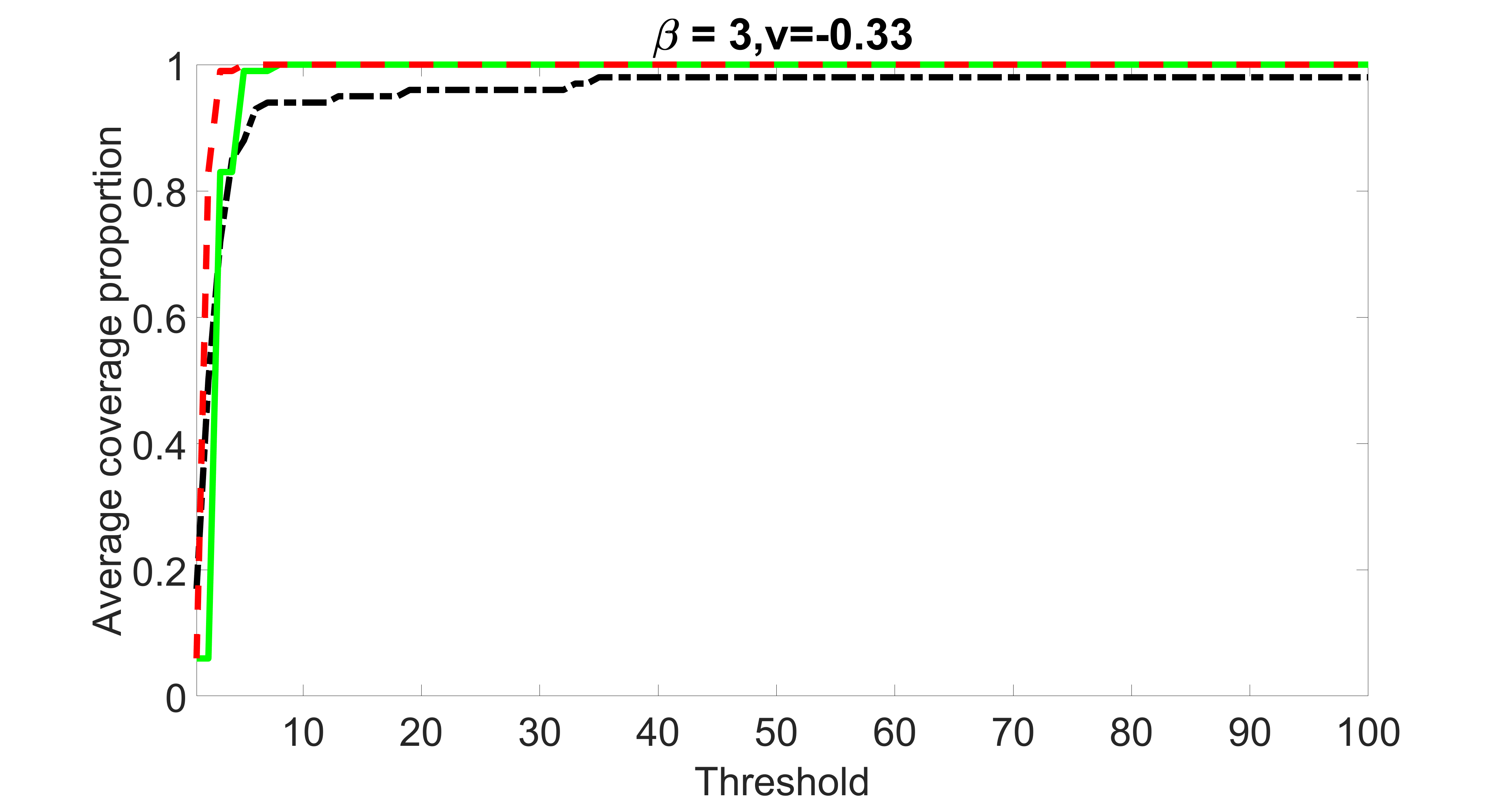}}
 \subcaptionbox{Confounder: medium \\ outcome, medium exposure}[0.45\linewidth]
 {\includegraphics[width=6cm,height=3.5cm]{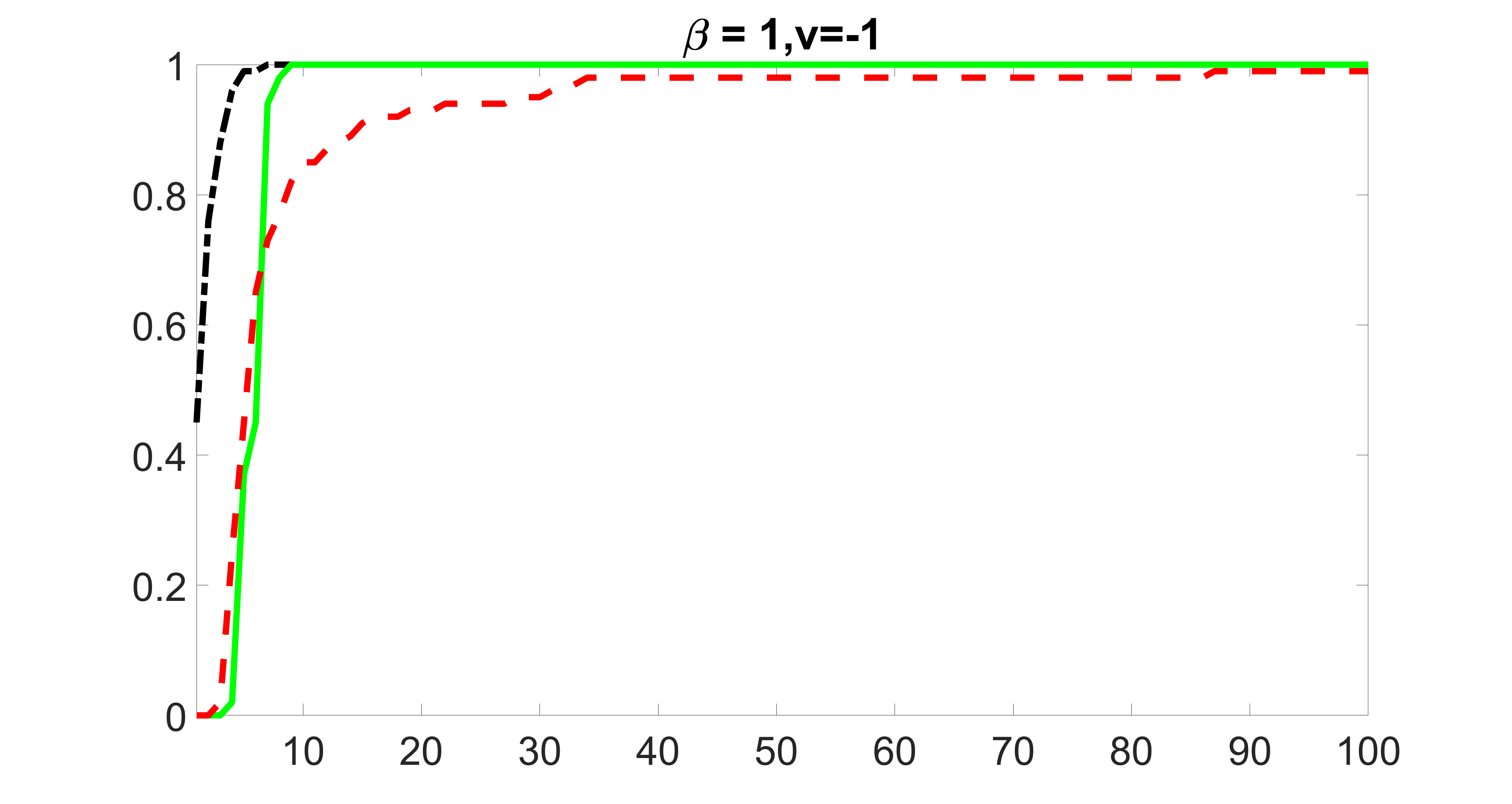}}
  \subcaptionbox{Confounder: weak \\ outcome, strong exposure}[0.45\linewidth]
 {\includegraphics[width=6cm,height=3.5cm]{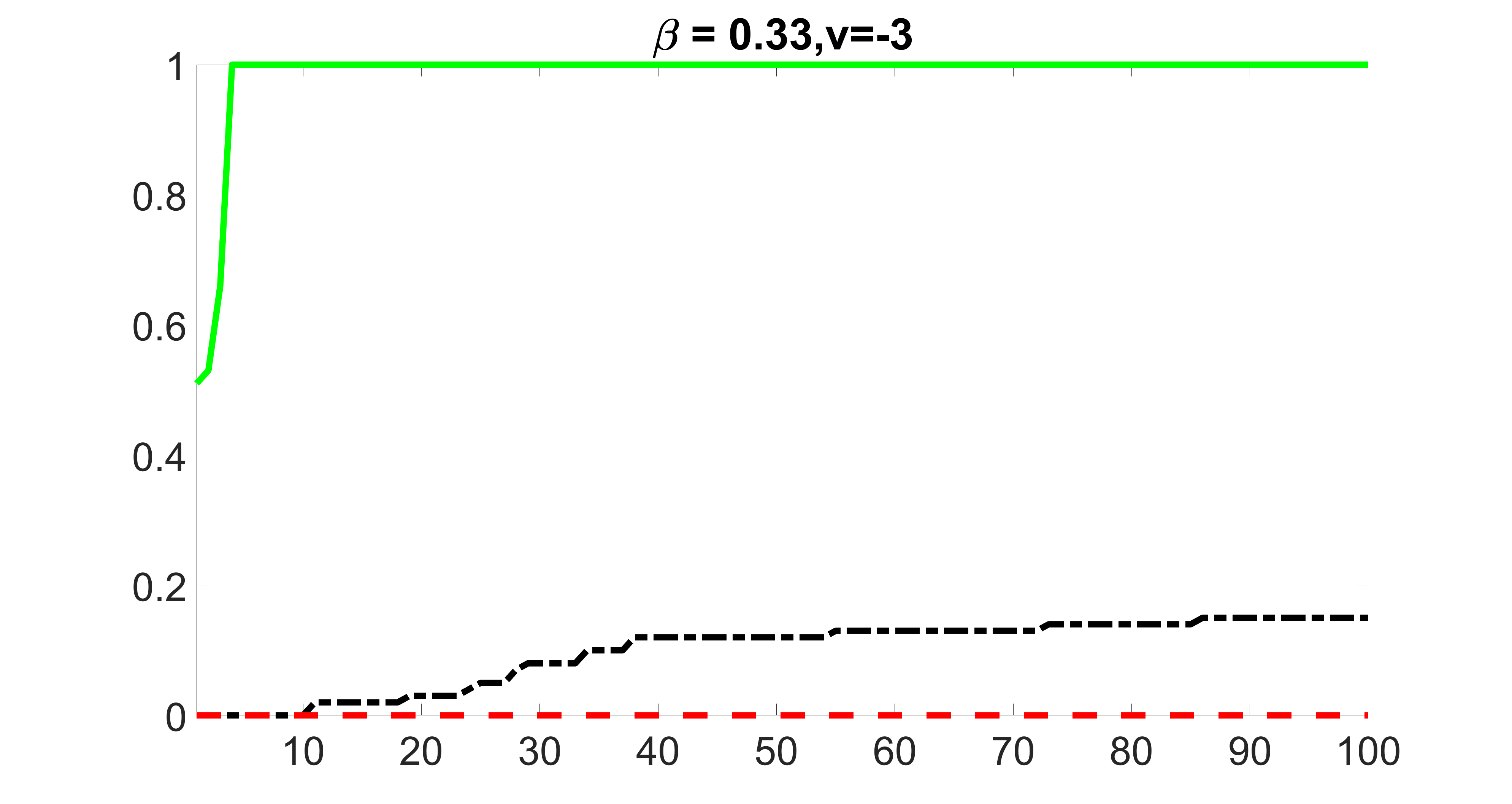}}
  \subcaptionbox{Precision: strong \\ outcome, zero exposure}[0.45\linewidth]
 {\includegraphics[width=6cm,height=3.5cm]{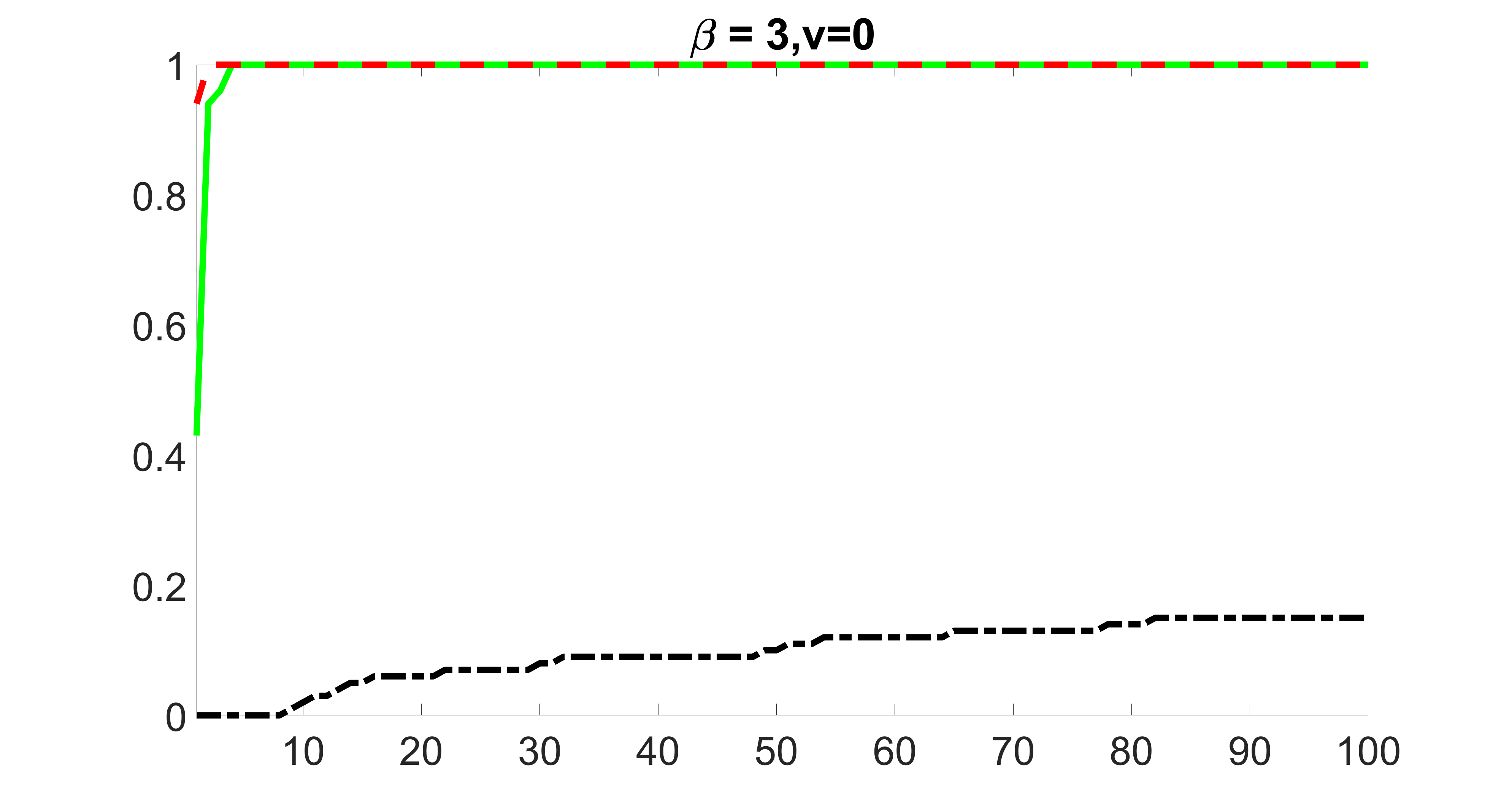}}
  \subcaptionbox{Precision: medium \\ outcome, zero exposure}[0.45\linewidth]
 {\includegraphics[width=6cm,height=3.5cm]{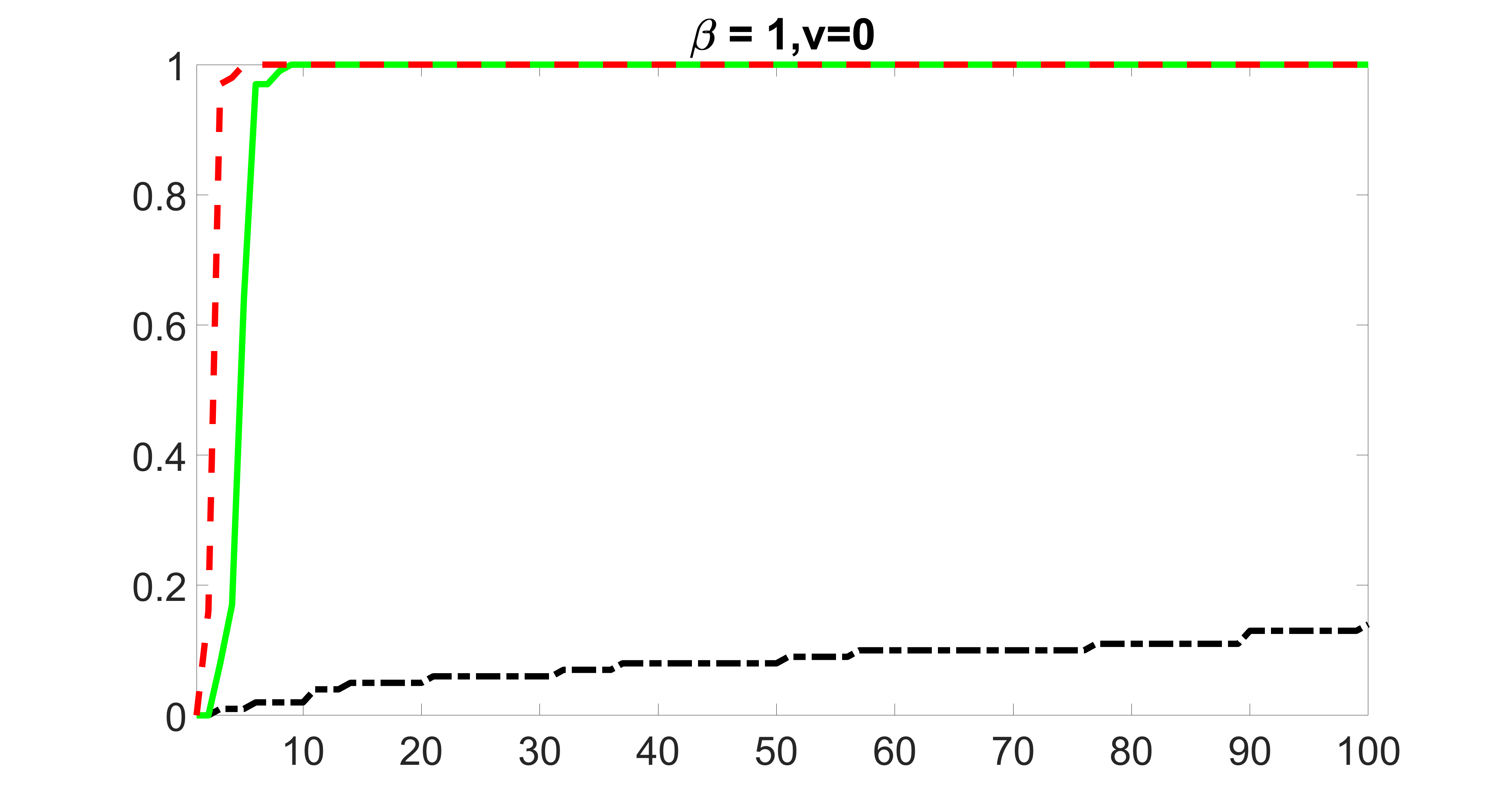}}
  \subcaptionbox{Precision: weak \\ outcome, zero exposure}[0.45\linewidth]
 {\includegraphics[width=6cm,height=3.5cm]{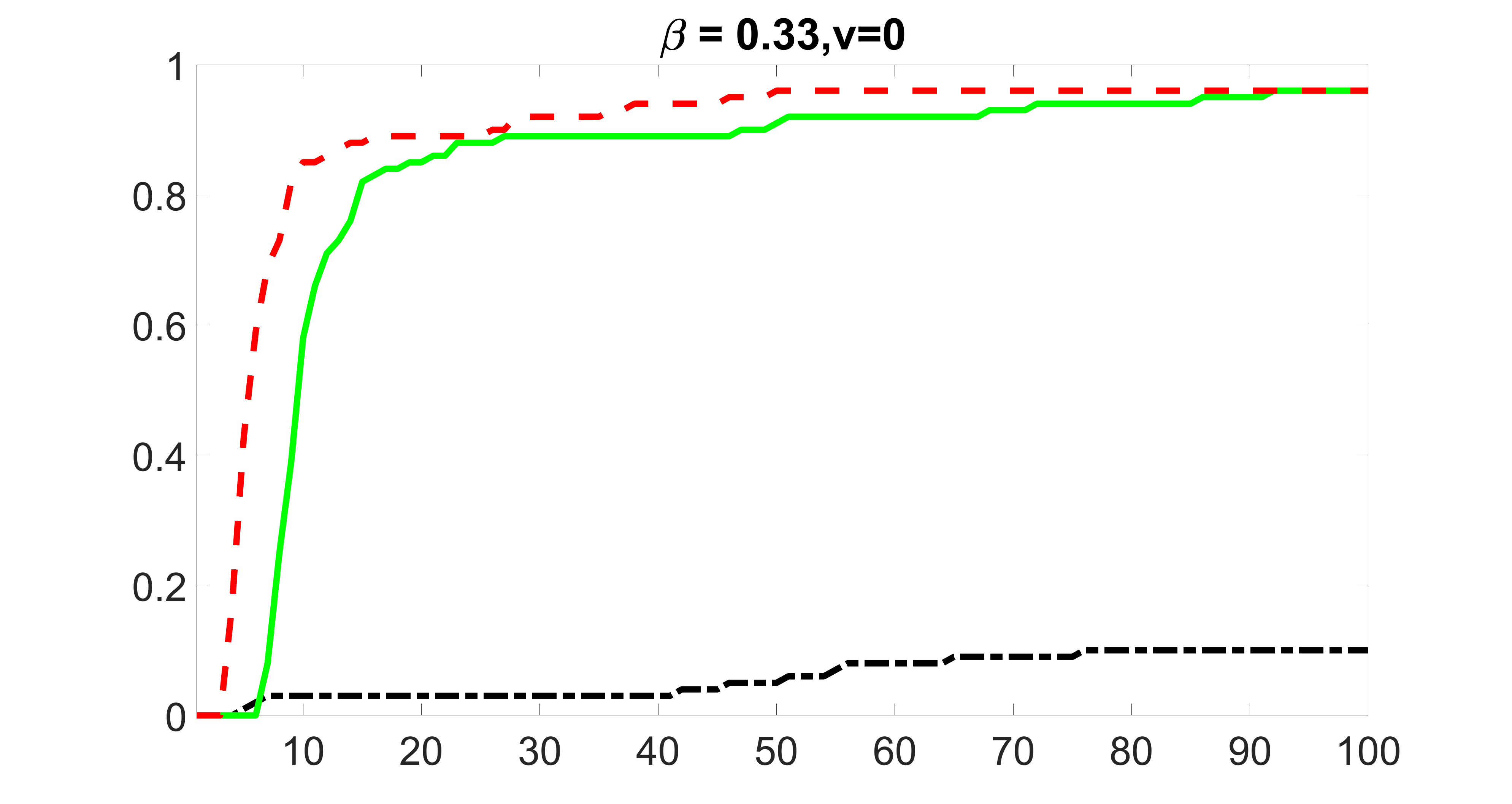}}
  \subcaptionbox{Overall coverage of $\mathcal{M}_1$}[0.45\linewidth]
 {\includegraphics[width=6cm,height=3.5cm]{./plotsMainArkSupp/sim1oal_p64_v2_optNorm2_n200s5000coverage_snry9rho208.png}}
\caption{Simulation results for the case $(n,s,\sigma,\rho_2) = (200,5000,1,0.8)$: Panels (a) -- (f) plot the average coverage proportion for $X_l$, where $l=1,2,3,104,105$ and $106$. Panels (a) -- (c) correspond to strong outcome and weak exposure predictor, moderate outcome and moderate exposure predictor and weak outcome and strong exposure predictor; Panels (d) -- (f) correspond to strong, moderate and weak predictors of outcome only. Panel (g) plots the average coverage proportion for the index set $\mathcal{M}_1 = \{1,2,3,104,105,106\}$. The x-axis represents the size of $\widehat{\mathcal{M}} $, while
y-axis denotes the average proportion. The green solid, the red dashed and the black dash dotted lines denote our joint screening method, the outcome screening method, and the intersection screening method, respectively. }
\label{sim1step1n200sigma1rho208}
\end{figure}

\begin{figure}[htbp]
\captionsetup[subfigure]{justification=centering}
\centering
 \subcaptionbox{Confounder: strong \\ outcome, weak exposure}[0.45\linewidth]
 {\includegraphics[width=6cm,height=3.5cm]{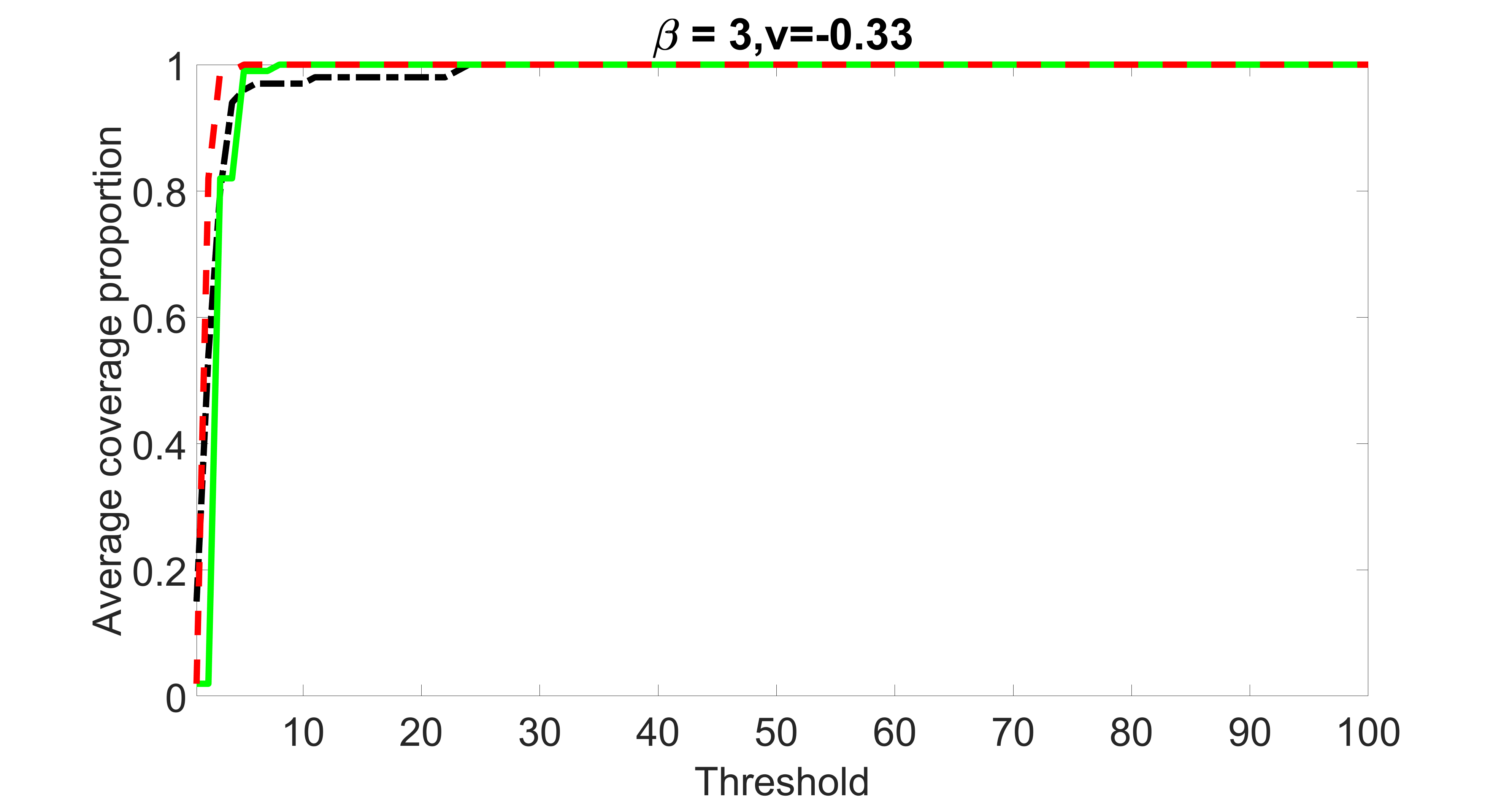}}
 \subcaptionbox{Confounder: medium \\ outcome, medium exposure}[0.45\linewidth]
 {\includegraphics[width=6cm,height=3.5cm]{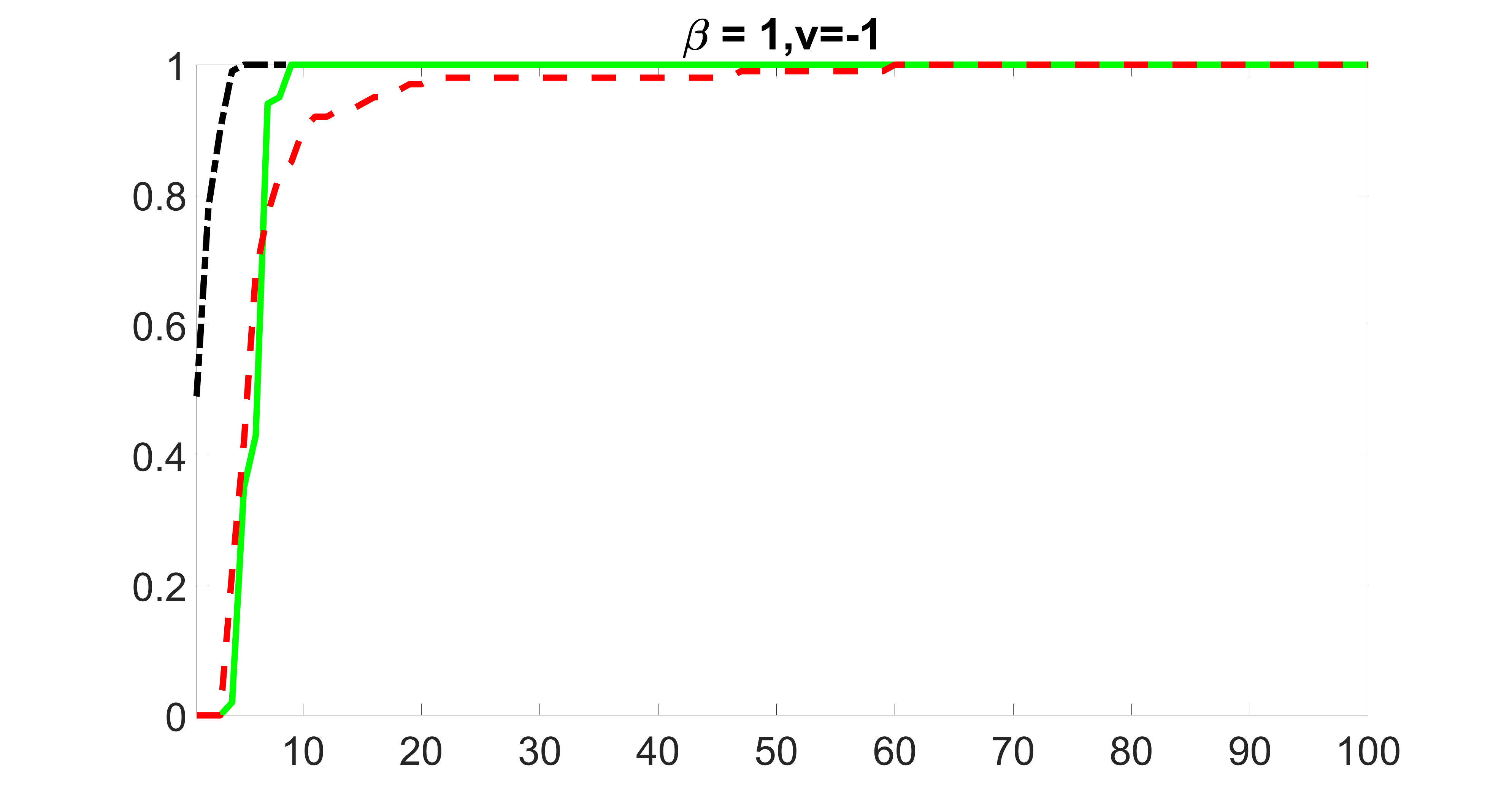}}
  \subcaptionbox{Confounder: weak \\ outcome, strong exposure}[0.45\linewidth]
 {\includegraphics[width=6cm,height=3.5cm]{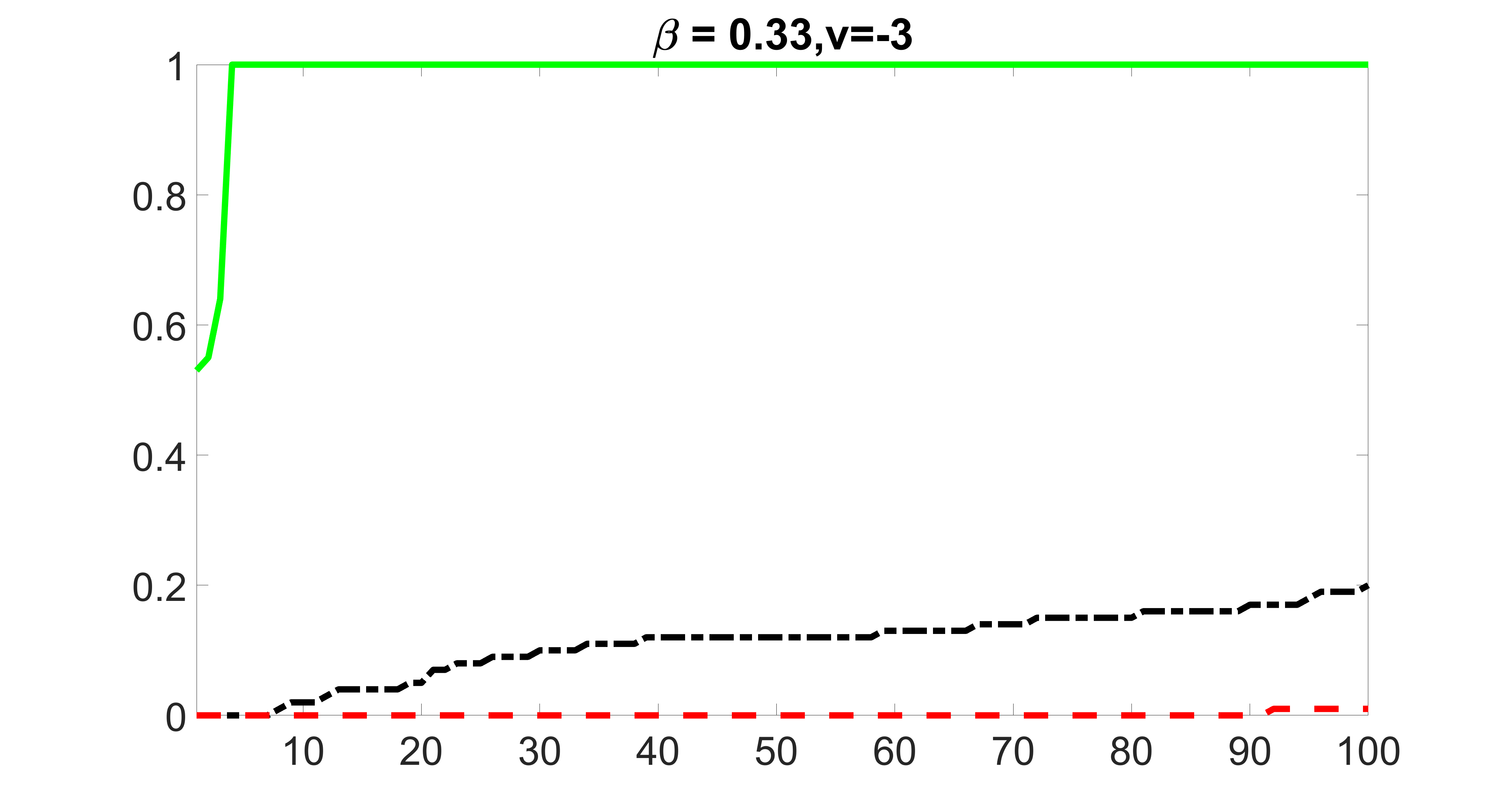}}
  \subcaptionbox{Precision: strong \\ outcome, zero exposure}[0.45\linewidth]
 {\includegraphics[width=6cm,height=3.5cm]{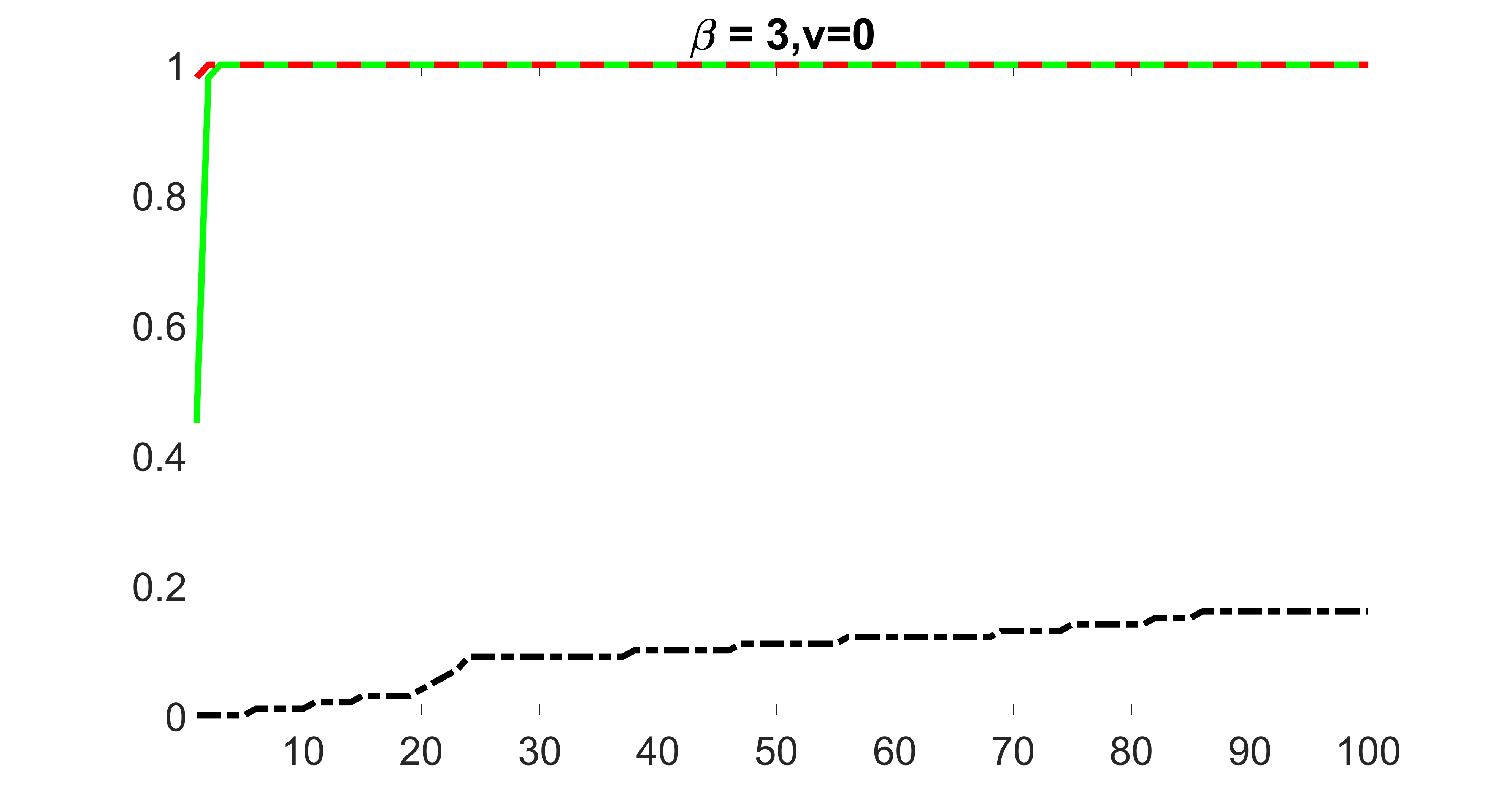}}
  \subcaptionbox{Precision: medium \\ outcome, zero exposure}[0.45\linewidth]
 {\includegraphics[width=6cm,height=3.5cm]{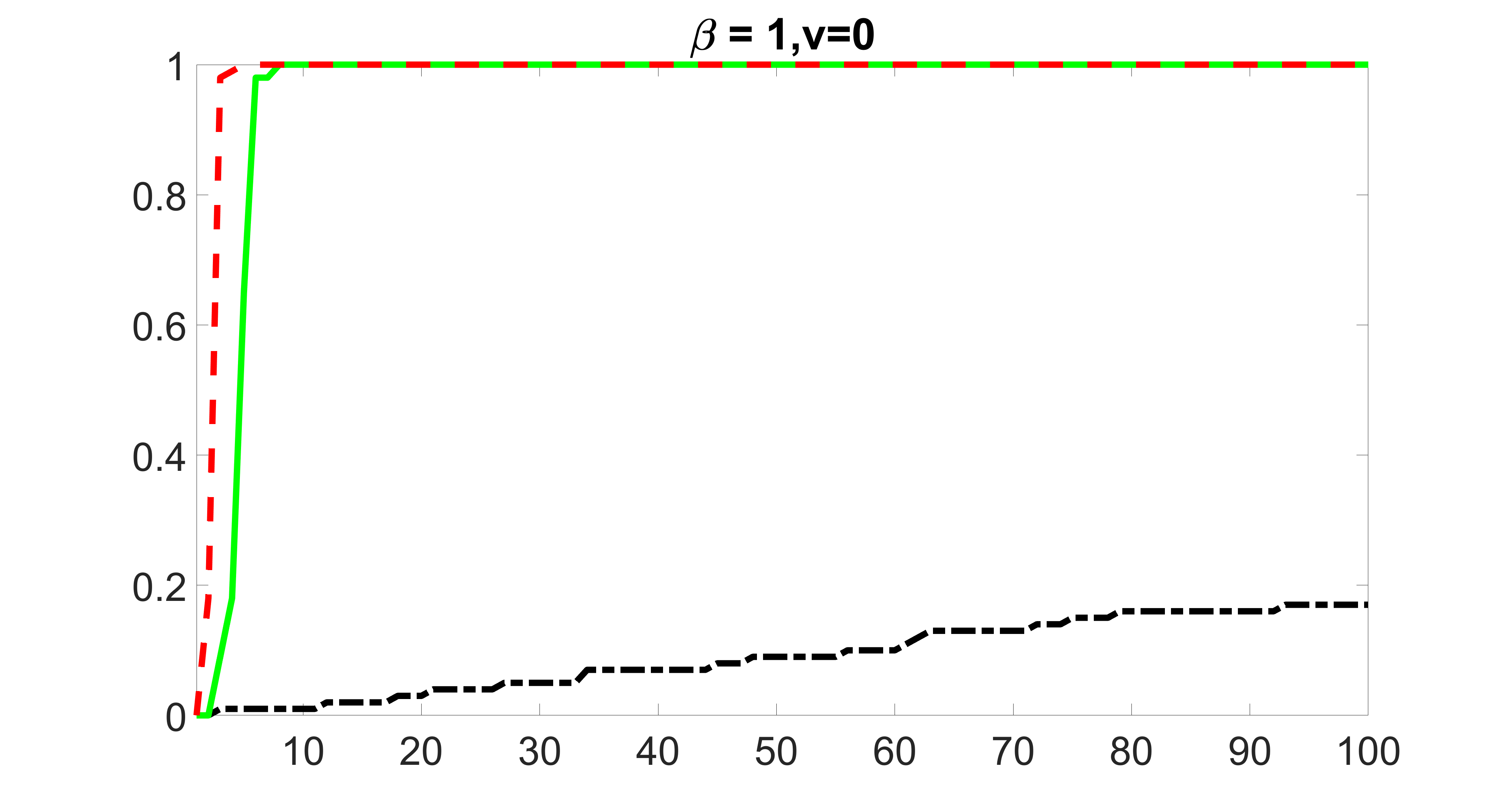}}
  \subcaptionbox{Precision: weak \\ outcome, zero exposure}[0.45\linewidth]
 {\includegraphics[width=6cm,height=3.5cm]{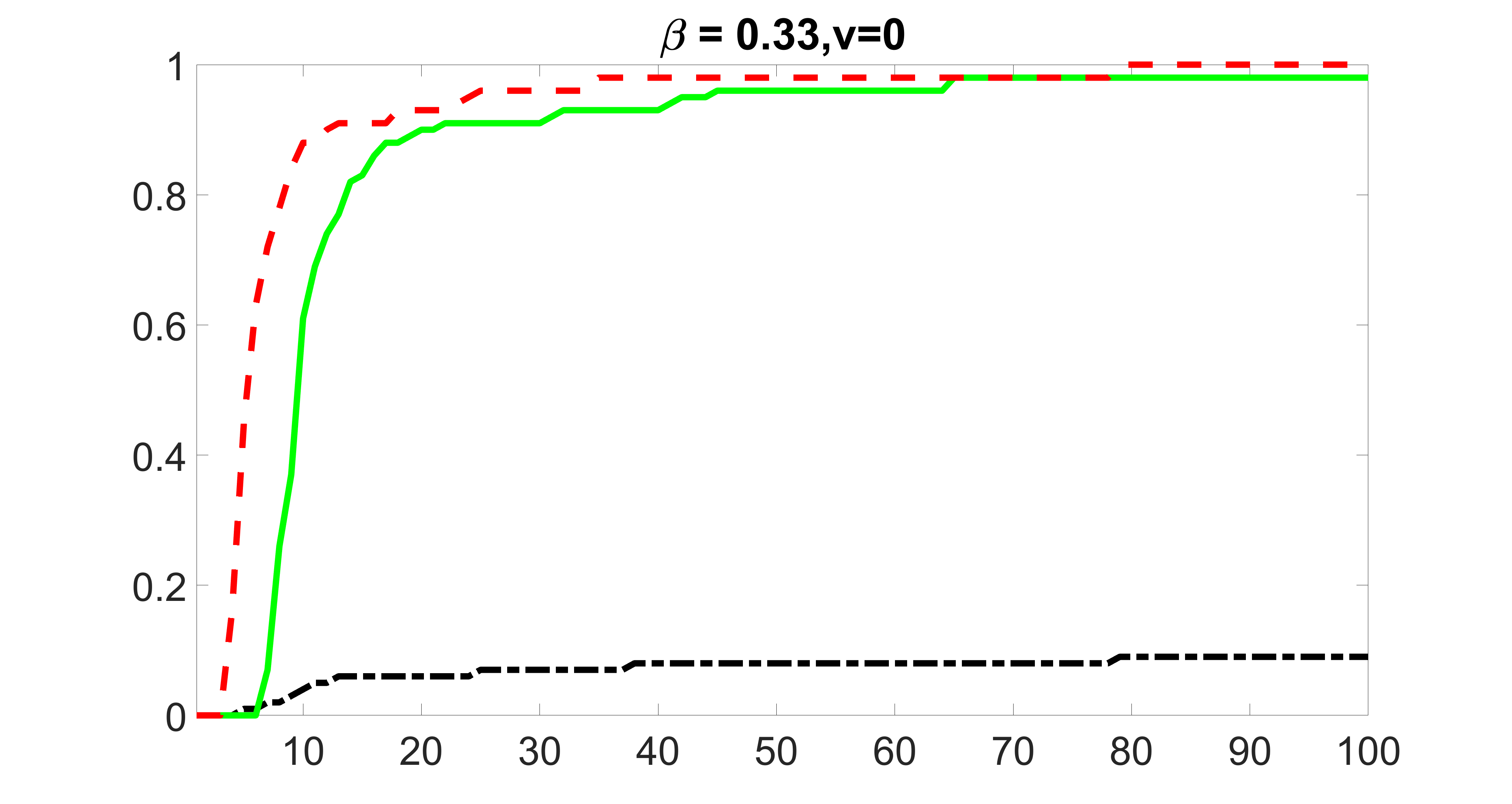}}
  \subcaptionbox{Overall coverage of $\mathcal{M}_1$}[0.45\linewidth]
 {\includegraphics[width=6cm,height=3.5cm]{./plotsMainArkSupp/sim1oal_p64_v2_optNorm2_n200s5000coverage_snry8rho202.png}}
\caption{Simulation results for the case $(n,s,\sigma,\rho_2) = (200,5000,0.5,0.2)$: Panels (a) -- (f) plot the average coverage proportion for $X_l$, where $l=1,2,3,104,105$ and $106$. Panels (a) -- (c) correspond to strong outcome and weak exposure predictor, moderate outcome and moderate exposure predictor and weak outcome and strong exposure predictor; Panels (d) -- (f) correspond to strong, moderate and weak predictors of outcome only. Panel (g) plots the average coverage proportion for the index set $\mathcal{M}_1 = \{1,2,3,104,105,106\}$. The x-axis represents the size of $\widehat{\mathcal{M}} $, while
y-axis denotes the average proportion. The green solid, the red dashed and the black dash dotted lines denote our joint screening method, the outcome screening method, and the intersection screening method, respectively. }
\label{sim1step1n200sigma025rho202}
\end{figure}

\begin{figure}[htbp]
\captionsetup[subfigure]{justification=centering}
\centering
 \subcaptionbox{Confounder: strong \\ outcome, weak exposure}[0.45\linewidth]
 {\includegraphics[width=6cm,height=3.5cm]{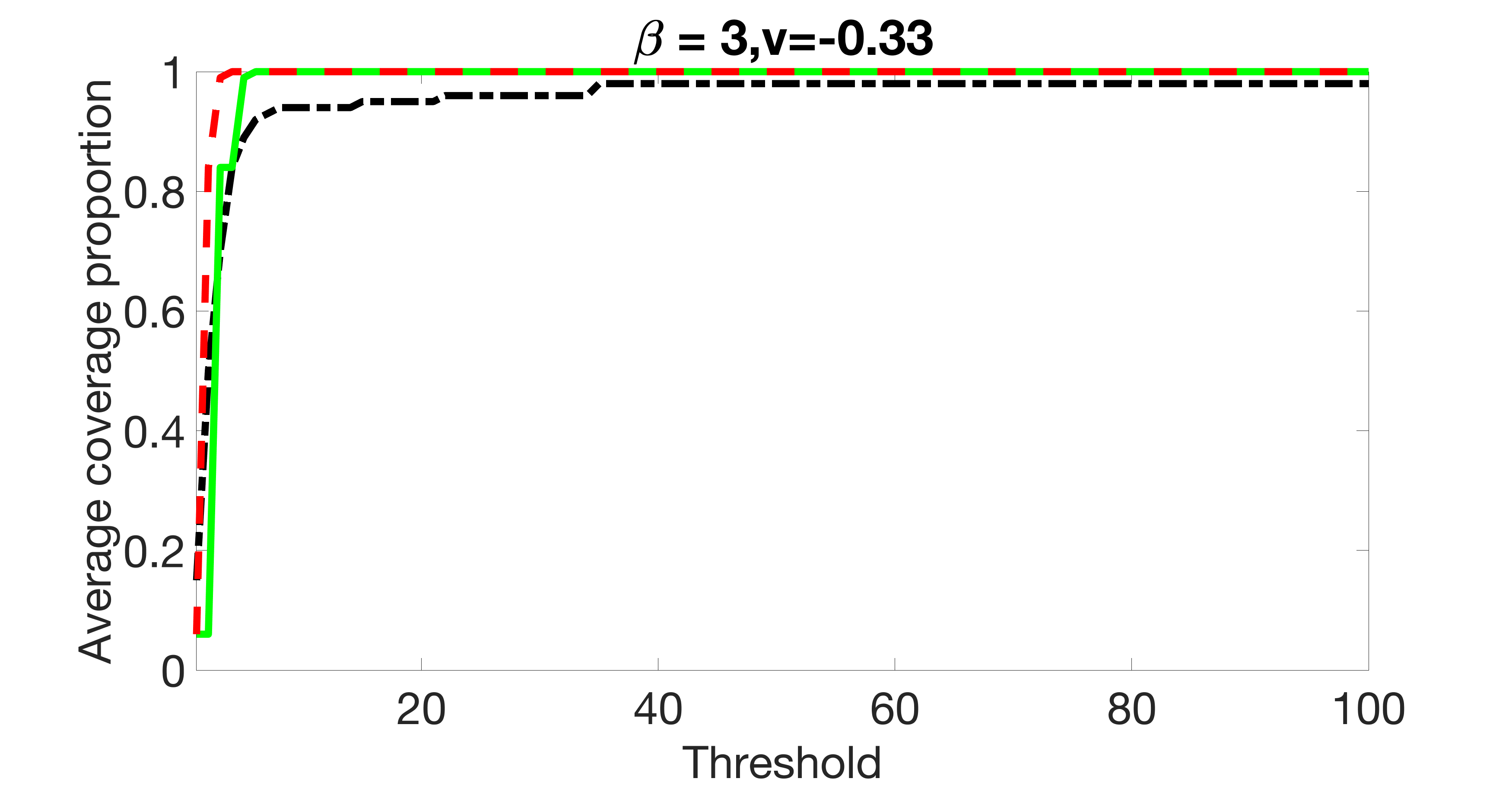}}
 \subcaptionbox{Confounder: medium \\ outcome, medium exposure}[0.45\linewidth]
 {\includegraphics[width=6cm,height=3.5cm]{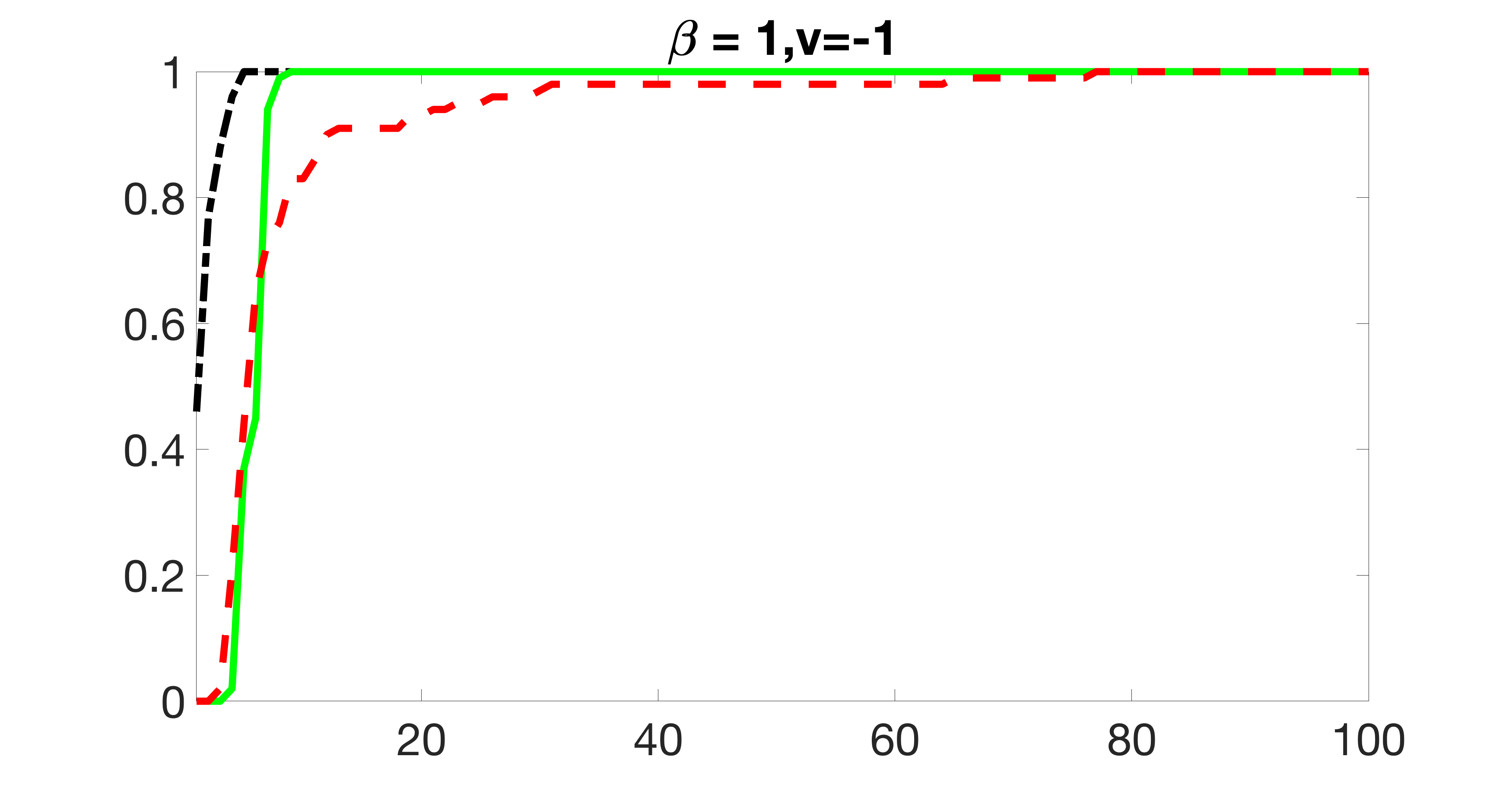}}
  \subcaptionbox{Confounder: weak \\ outcome, strong exposure}[0.45\linewidth]
 {\includegraphics[width=6cm,height=3.5cm]{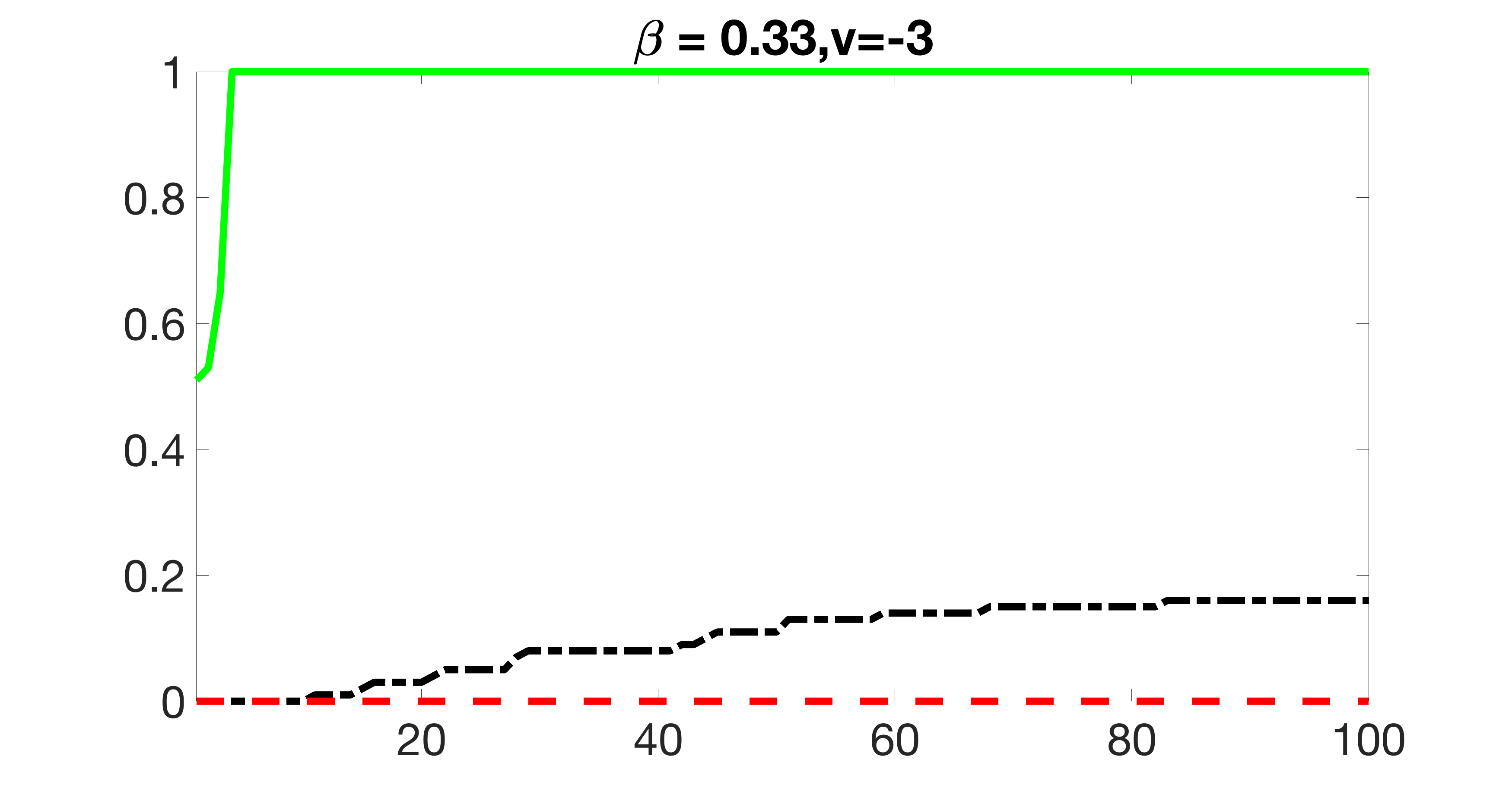}}
  \subcaptionbox{Precision: strong \\ outcome, zero exposure}[0.45\linewidth]
 {\includegraphics[width=6cm,height=3.5cm]{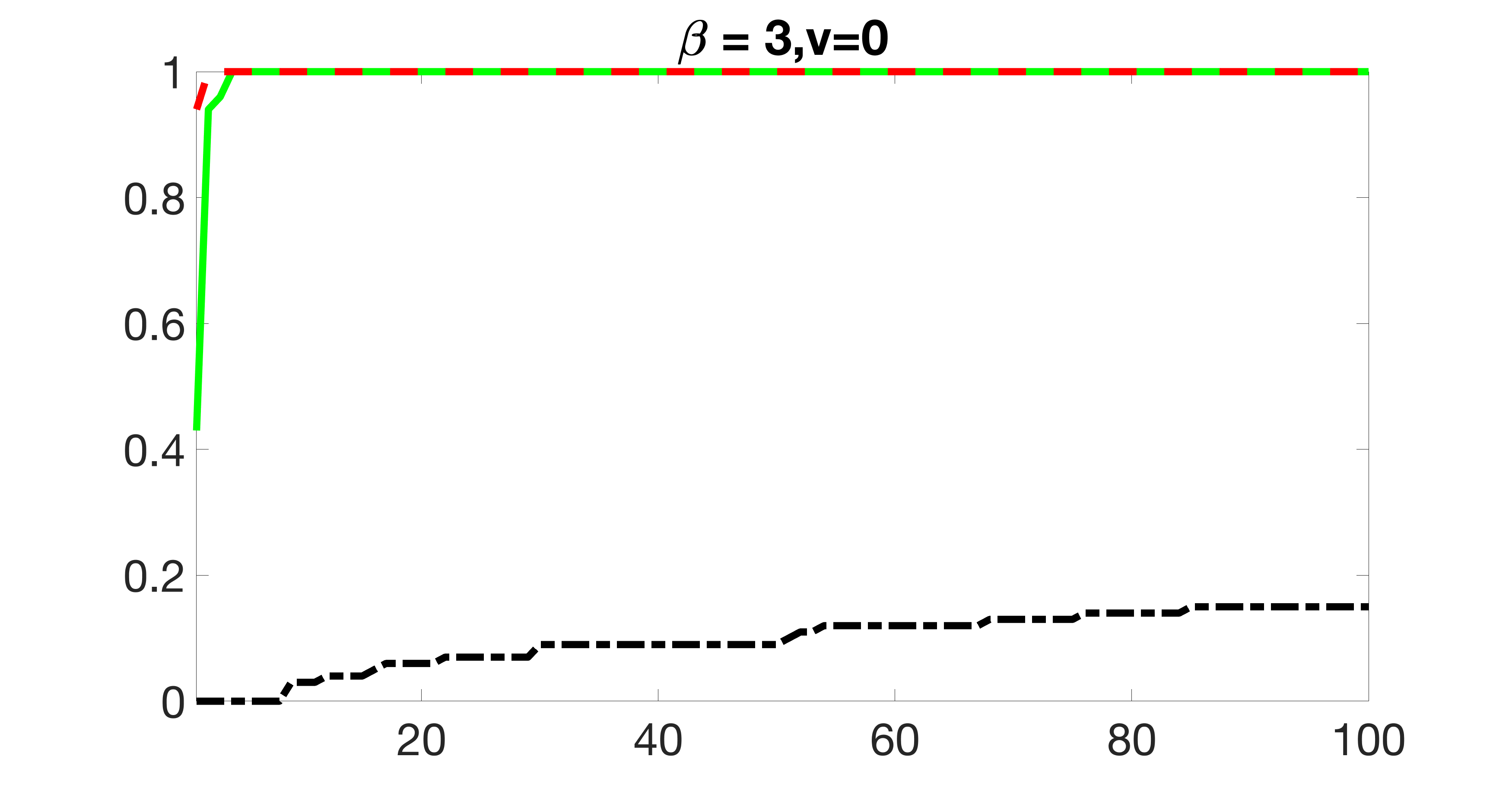}}
  \subcaptionbox{Precision: medium \\ outcome, zero exposure}[0.45\linewidth]
 {\includegraphics[width=6cm,height=3.5cm]{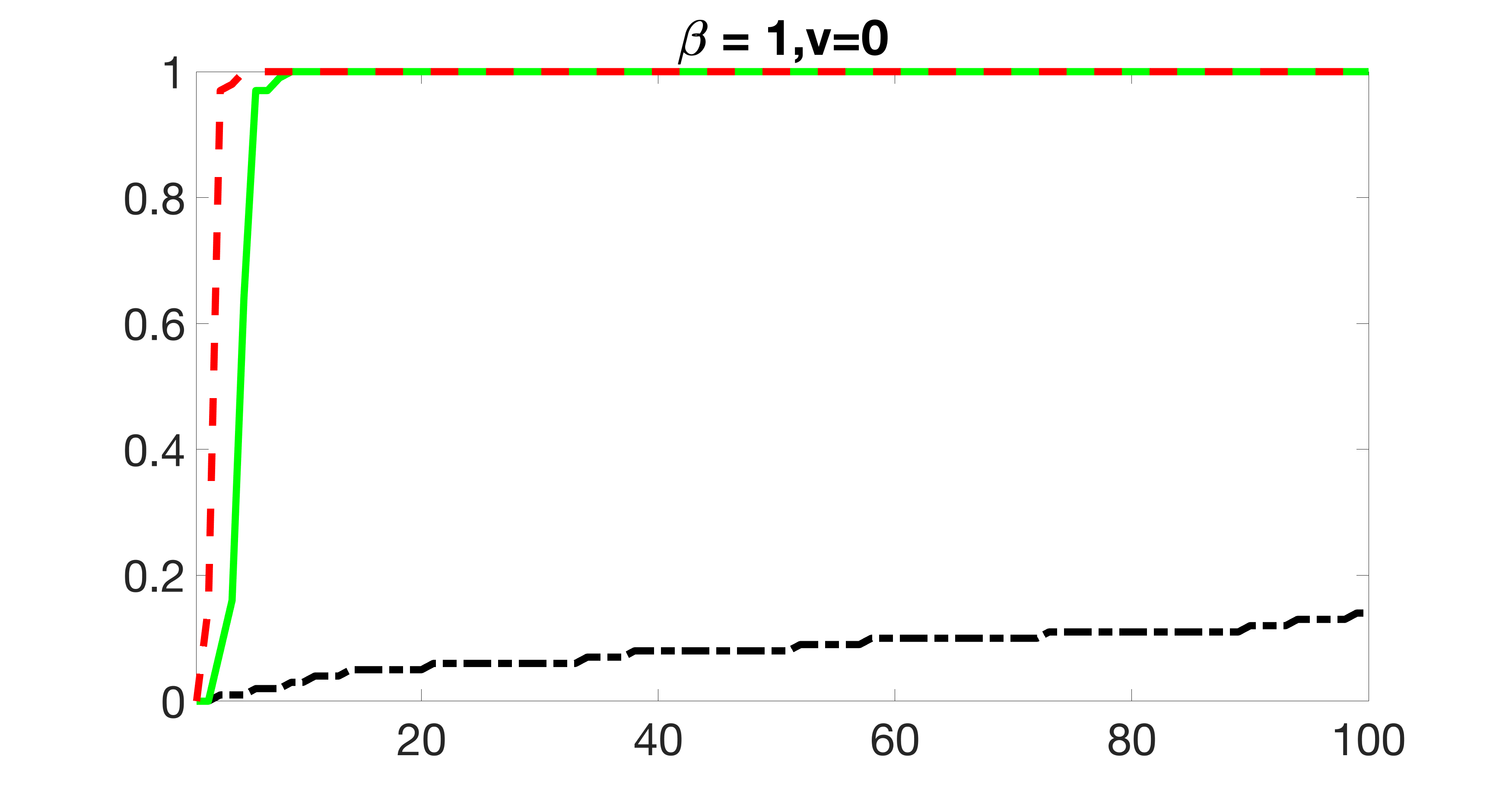}}
  \subcaptionbox{Precision: weak \\ outcome, zero exposure}[0.45\linewidth]
 {\includegraphics[width=6cm,height=3.5cm]{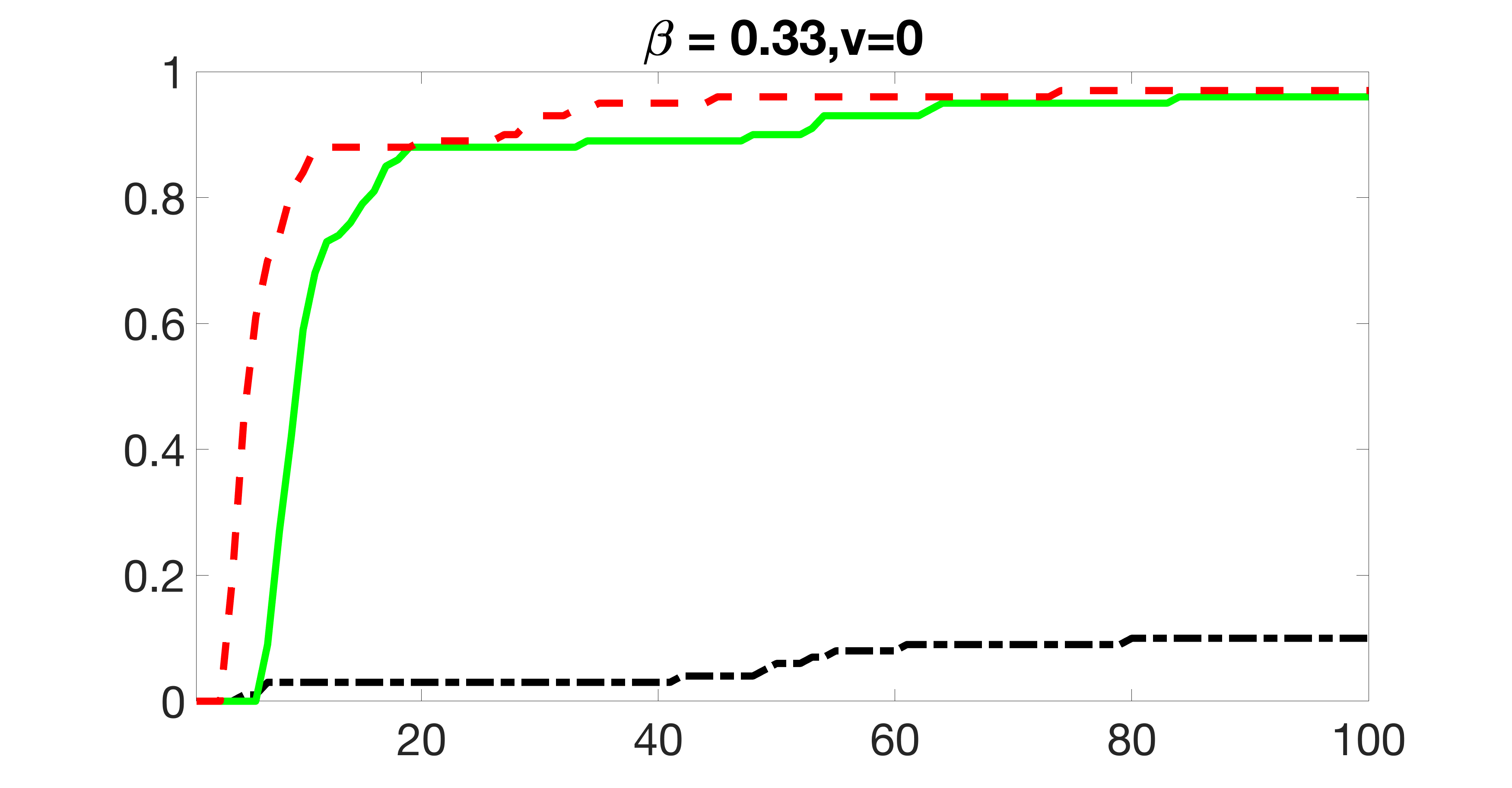}}
  \subcaptionbox{Overall coverage of $\mathcal{M}_1$}[0.45\linewidth]
 {\includegraphics[width=6cm,height=3.5cm]{./plotsMainArkSupp/sim1oal_p64_v2_optNorm2_n200s5000coverage_snry8rho208.png}}
\caption{Simulation results for the case $(n,s,\sigma,\rho_2) = (200,5000,0.5,0.8)$: Panels (a) -- (f) plot the average coverage proportion for $X_l$, where $l=1,2,3,104,105$ and $106$. Panels (a) -- (c) correspond to strong outcome and weak exposure predictor, moderate outcome and moderate exposure predictor and weak outcome and strong exposure predictor; Panels (d) -- (f) correspond to strong, moderate and weak predictors of outcome only. Panel (g) plots the average coverage proportion for the index set $\mathcal{M}_1 = \{1,2,3,104,105,106\}$. The x-axis represents the size of $\widehat{\mathcal{M}} $, while
y-axis denotes the average proportion. The green solid, the red dashed and the black dash dotted lines denote our joint screening method, the outcome screening method, and the intersection screening method, respectively. }
\label{sim1step1n200sigma025rho208}
\end{figure}

\subsection{Screening and estimation under different sizes of $\widehat{\mathcal{M}}_{1}^{*}$ and $\widehat{\mathcal{M}}_{2}$}
\label{Screening and estimation under different sizes}

In addition, we conduct a similar study following the same setting as described in Section \ref{Simulation for screening} of the main paper, where $|\widehat{\mathcal{M}}_2|/|\widehat{\mathcal{M}}_1^*|=2$ and $1/2$. We use $n = 500$ here since it is close to the number of observations $n = 566$ in the real data analysis. Specifically, as summarized in Figures \ref{sim1step1n500sigma1_propSb05} and \ref{sim1step1n500sigma1_propSb2}, when the ratio of $|\widehat{\mathcal{M}}_2|/|\widehat{\mathcal{M}}_1^*|$ is taken as $1/2$ or $2$, the performances of the proposed joint screening method are quite similar to each other. By comparing them with Figure \ref{sim1step1n500sigma1}, where $|\widehat{\mathcal{M}}_2|/|\widehat{\mathcal{M}}_1^*|=1$, the performances of the screening step results are quite similar. 

\begin{figure}[htbp]
\captionsetup[subfigure]{justification=centering}
\centering
 \subcaptionbox{Confounder: strong \\ outcome, weak exposure}[0.45\linewidth]
 {\includegraphics[width=6cm,height=3.5cm]{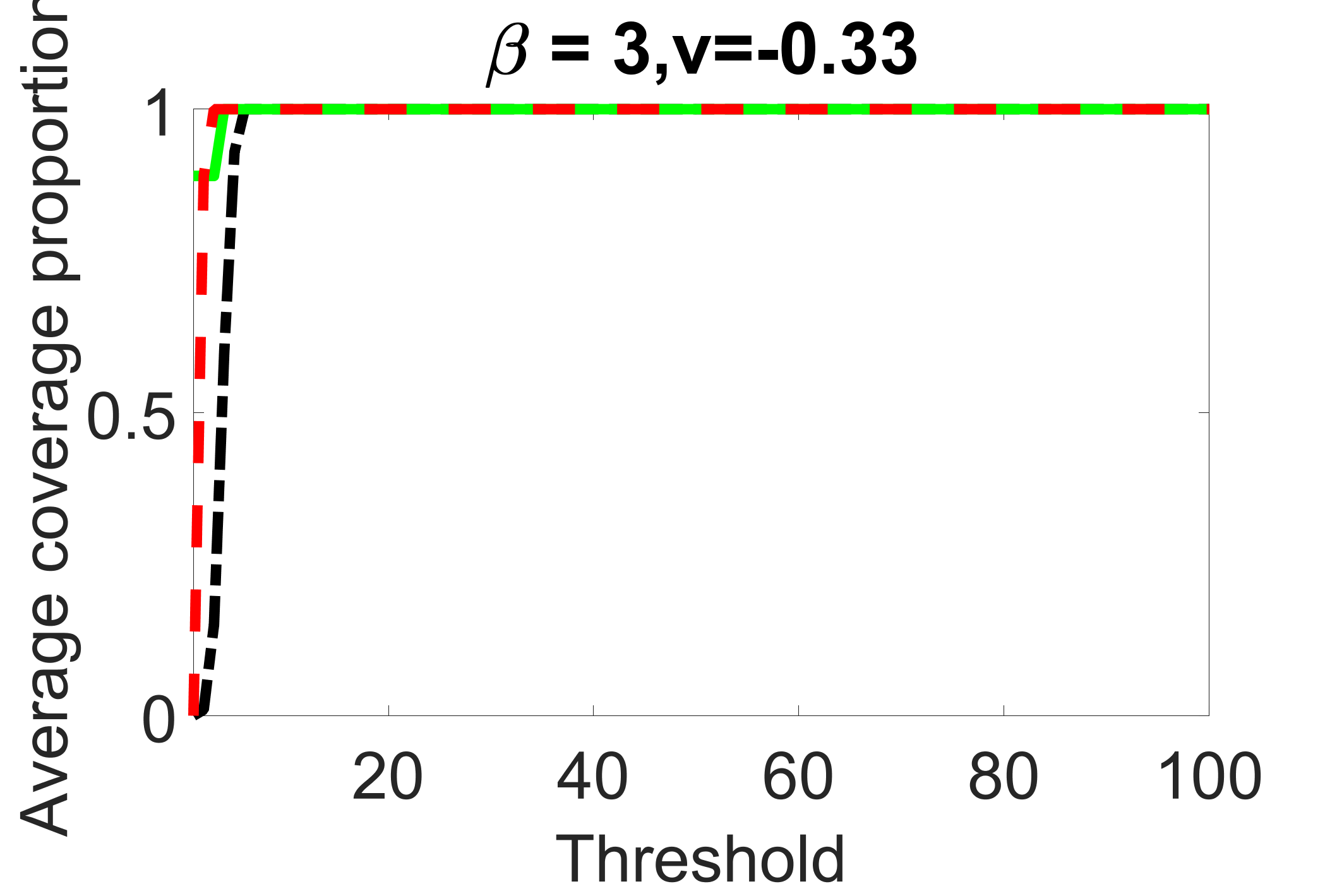}}
 \subcaptionbox{Confounder: medium \\ outcome, medium exposure}[0.45\linewidth]
 {\includegraphics[width=6cm,height=3.5cm]{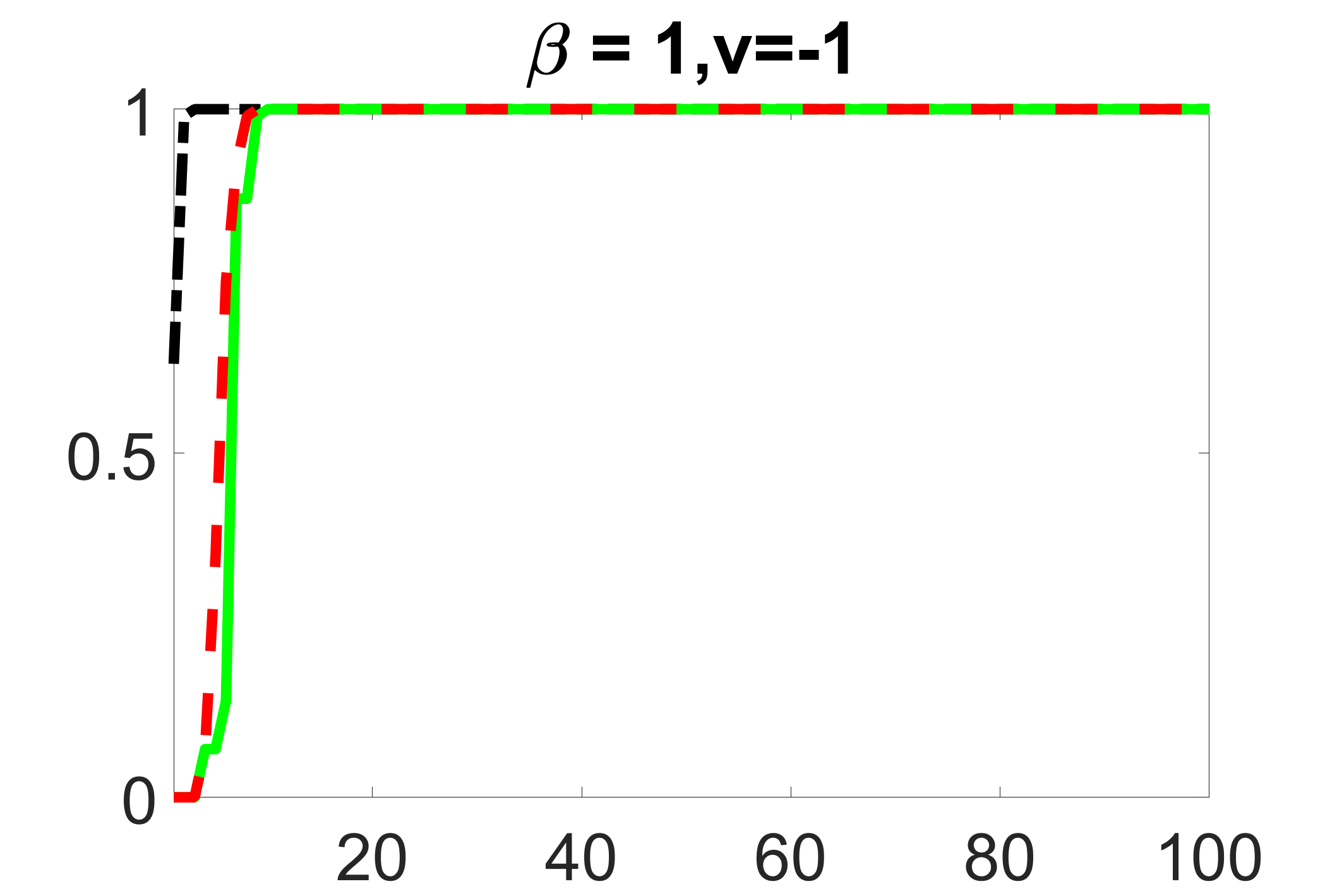}}
  \subcaptionbox{Confounder: weak \\ outcome, strong exposure}[0.45\linewidth]
 {\includegraphics[width=6cm,height=3.5cm]{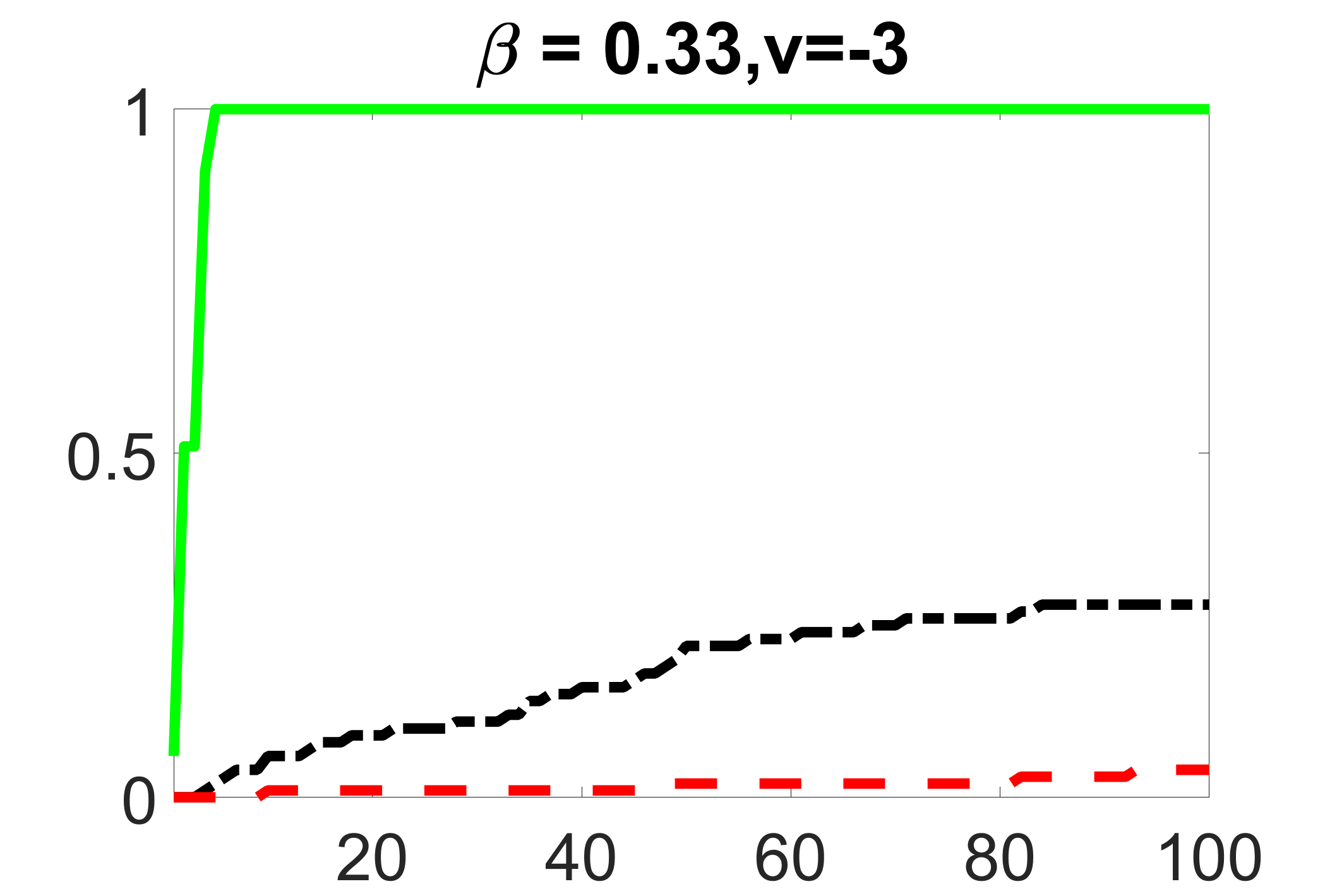}}
  \subcaptionbox{Precision: strong \\ outcome, zero exposure}[0.45\linewidth]
 {\includegraphics[width=6cm,height=3.5cm]{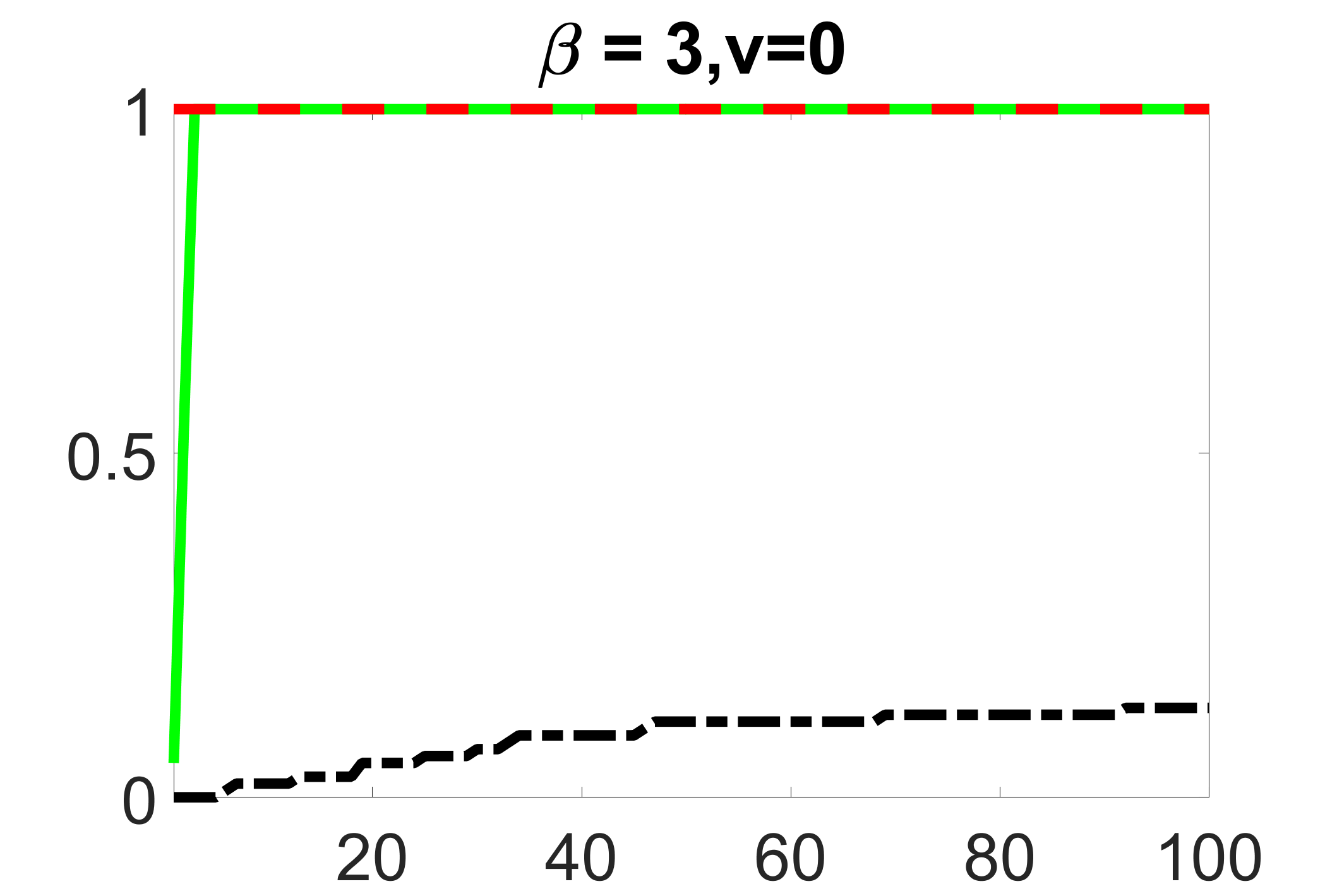}}
  \subcaptionbox{Precision: medium \\ outcome, zero exposure}[0.45\linewidth]
 {\includegraphics[width=6cm,height=3.5cm]{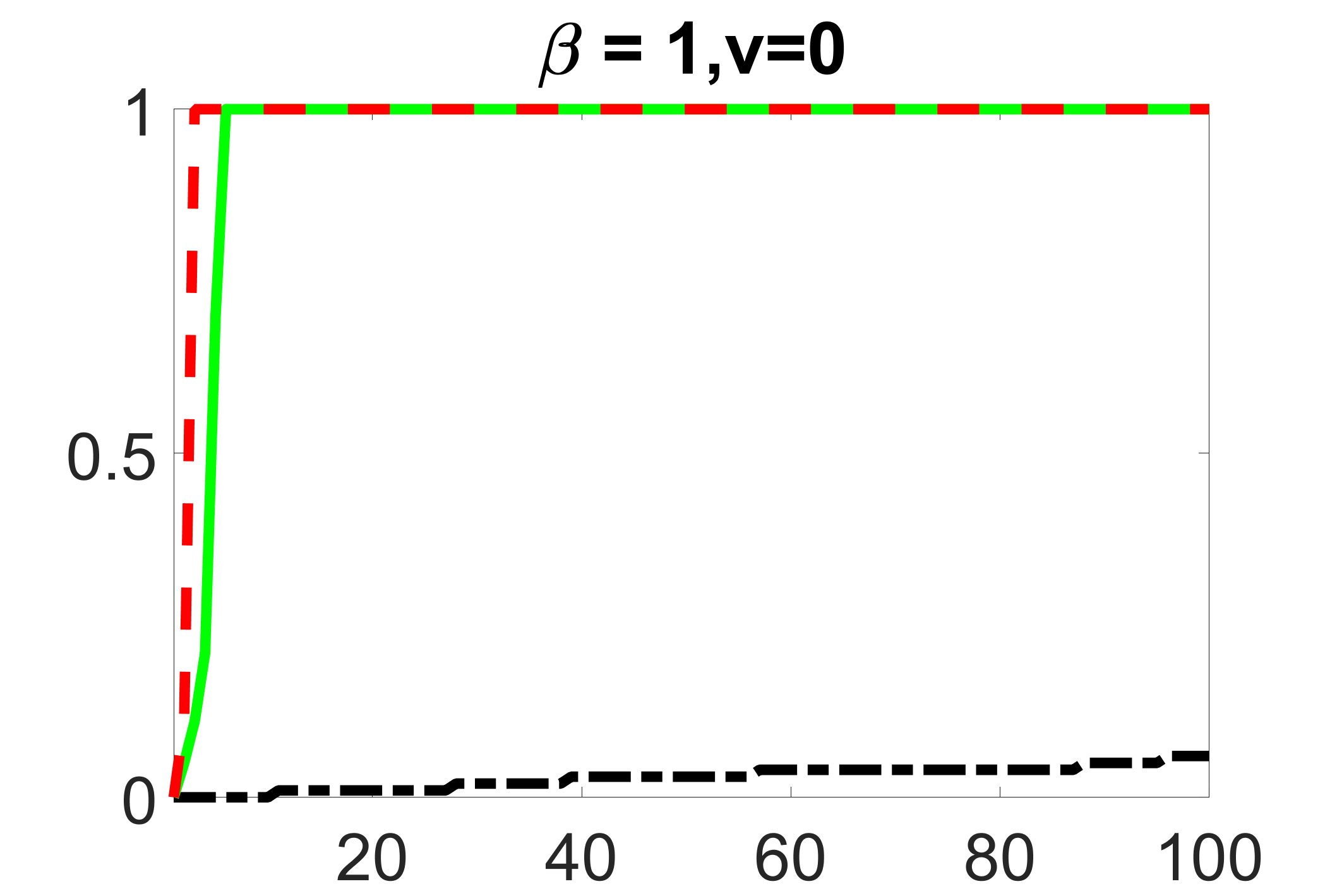}}
  \subcaptionbox{Precision: weak \\ outcome, zero exposure}[0.45\linewidth]
 {\includegraphics[width=6cm,height=3.5cm]{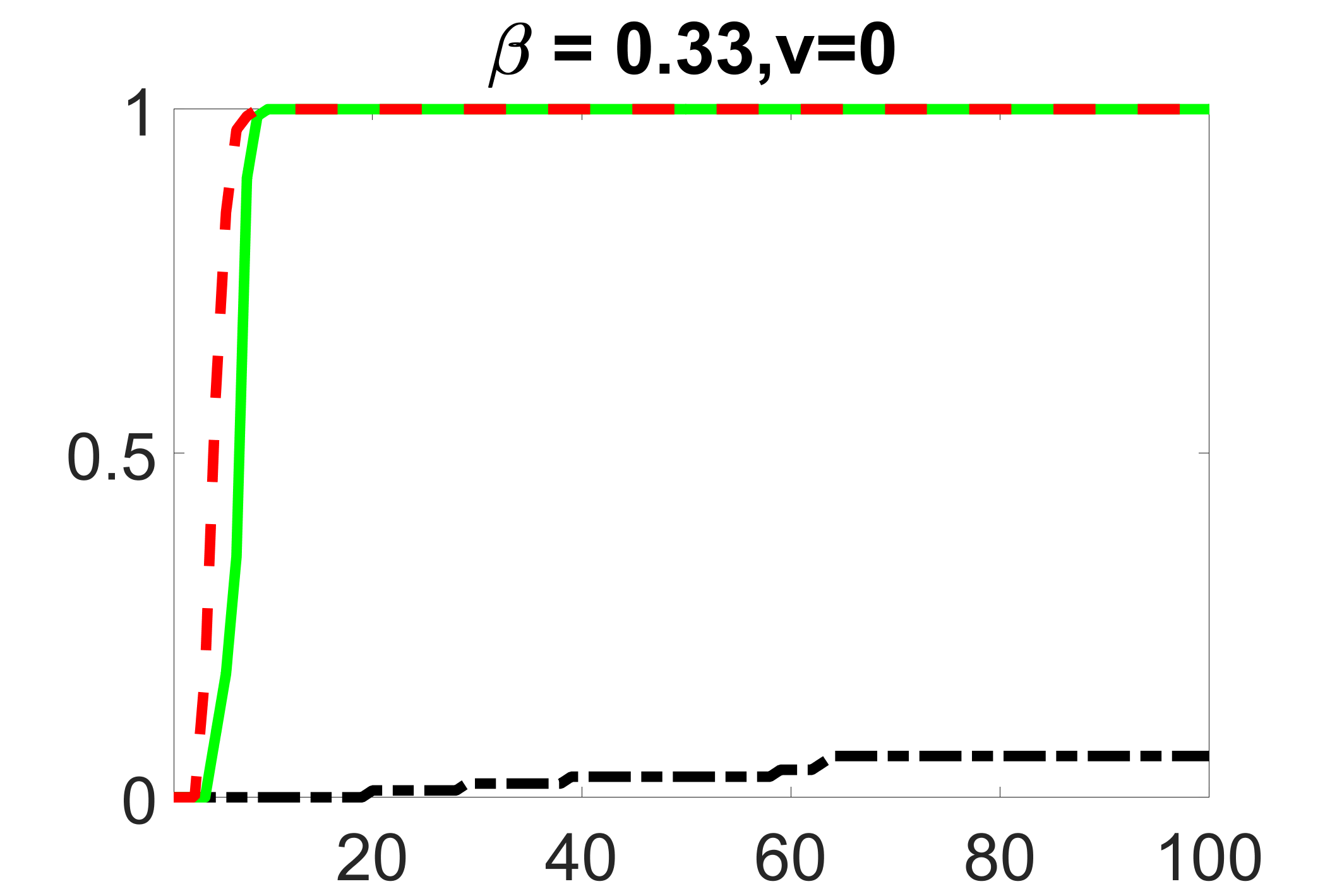}}
  \subcaptionbox{Overall coverage of $\mathcal{M}_1$}[0.45\linewidth]
 {\includegraphics[width=6cm,height=3.5cm]{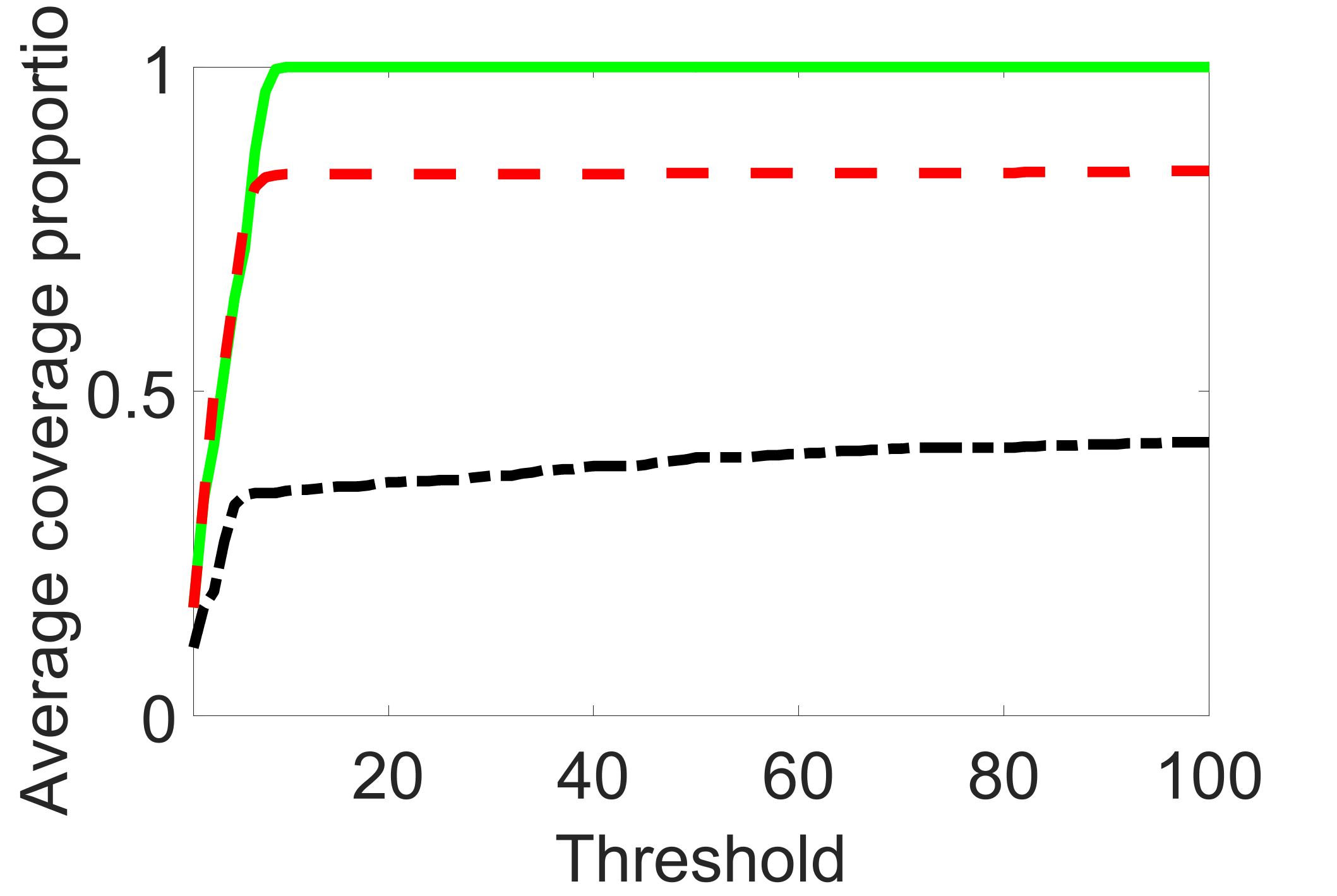}}
\caption{Simulation results for the case $(n,s,\sigma) = (500,5000,1)$  when $|\widehat{\mathcal{M}}_2|/|\widehat{\mathcal{M}}_1^*|=1/2$: Panels (a) -- (f) plot the average coverage proportion for $X_l$, where $l=1,2,3,104,105$ and $106$. Panels (a) -- (c) correspond to strong outcome and weak exposure predictor, moderate outcome and moderate exposure predictor and weak outcome and strong exposure predictor; Panels (d) -- (f) correspond to strong, moderate and weak predictors of outcome only. Panel (g) plots the average coverage proportion for the index set $\mathcal{M}_1 = \{1,2,3,104,105,106\}$. The x-axis represents the size of $\widehat{\mathcal{M}} $, while
y-axis denotes the average proportion. The green solid, the red dashed and the black dash dotted lines denote our joint screening method, the outcome screening method, and the intersection screening method, respectively.}
\label{sim1step1n500sigma1_propSb05}
\end{figure}

\begin{figure}[htbp]
\captionsetup[subfigure]{justification=centering}
\centering
 \subcaptionbox{Confounder: strong \\ outcome, weak exposure}[0.45\linewidth]
 {\includegraphics[width=6cm,height=3.5cm]{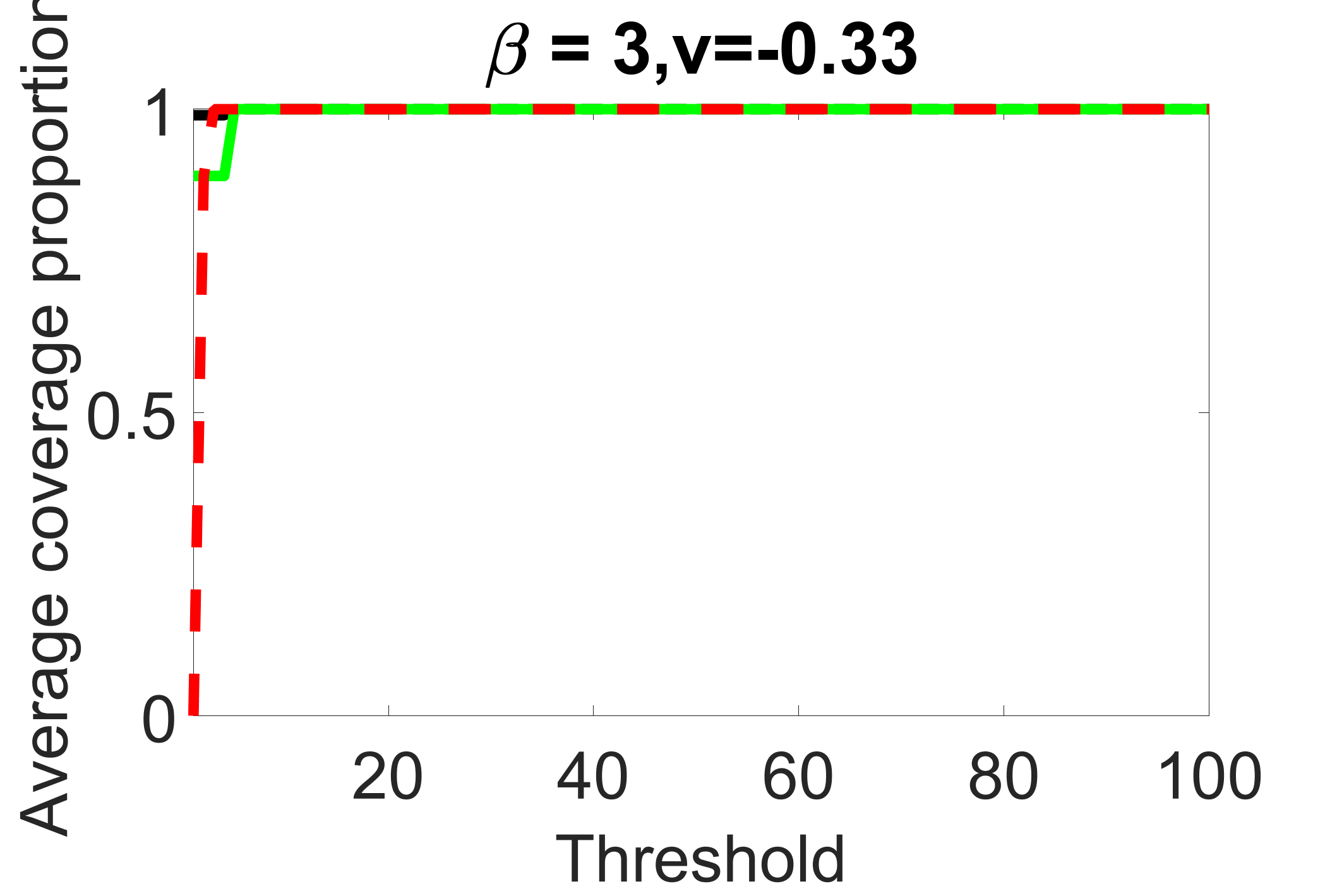}}
 \subcaptionbox{Confounder: medium \\ outcome, medium exposure}[0.45\linewidth]
 {\includegraphics[width=6cm,height=3.5cm]{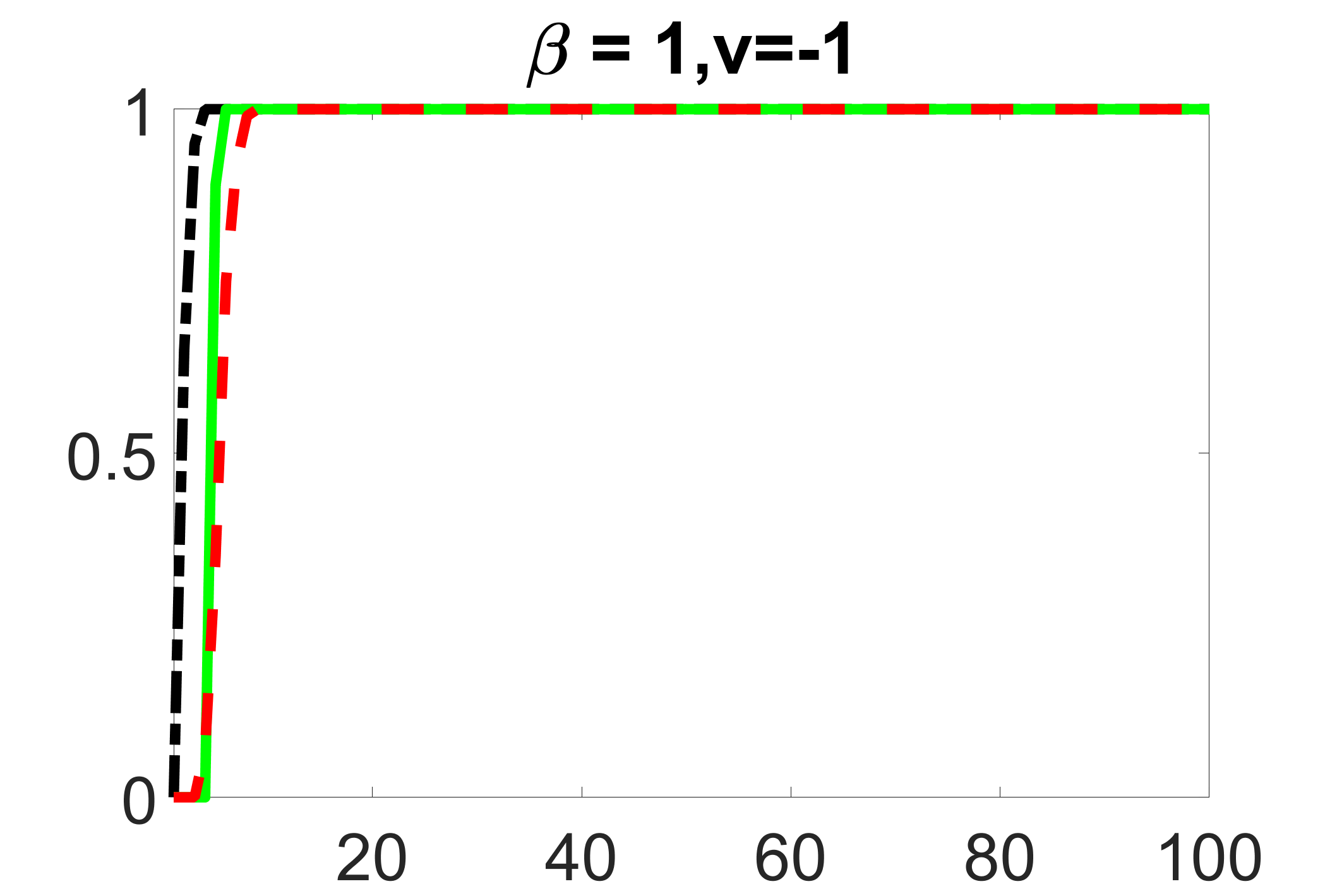}}
  \subcaptionbox{Confounder: weak \\ outcome, strong exposure}[0.45\linewidth]
 {\includegraphics[width=6cm,height=3.5cm]{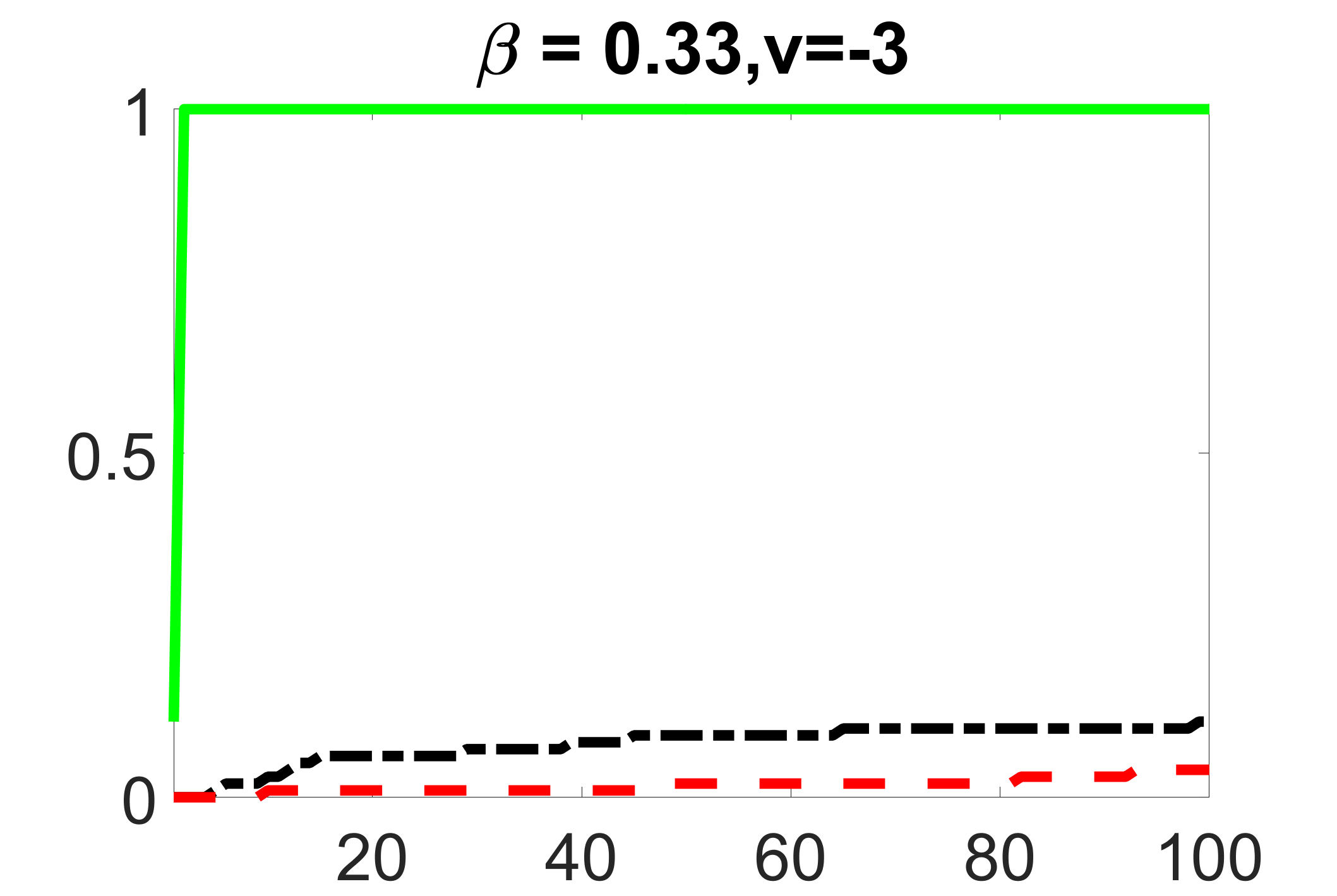}}
  \subcaptionbox{Precision: strong \\ outcome, zero exposure}[0.45\linewidth]
 {\includegraphics[width=6cm,height=3.5cm]{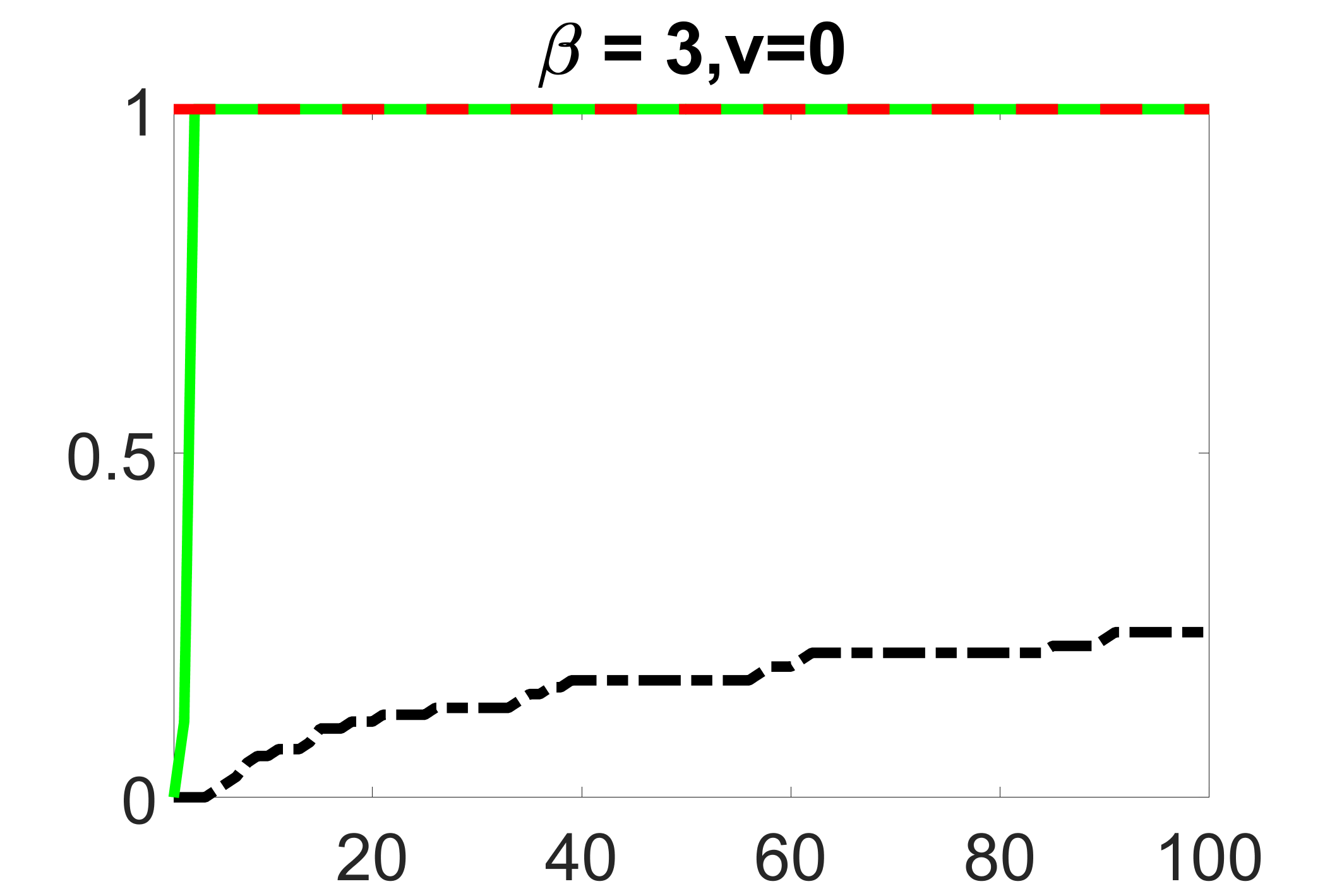}}
  \subcaptionbox{Precision: medium \\ outcome, zero exposure}[0.45\linewidth]
 {\includegraphics[width=6cm,height=3.5cm]{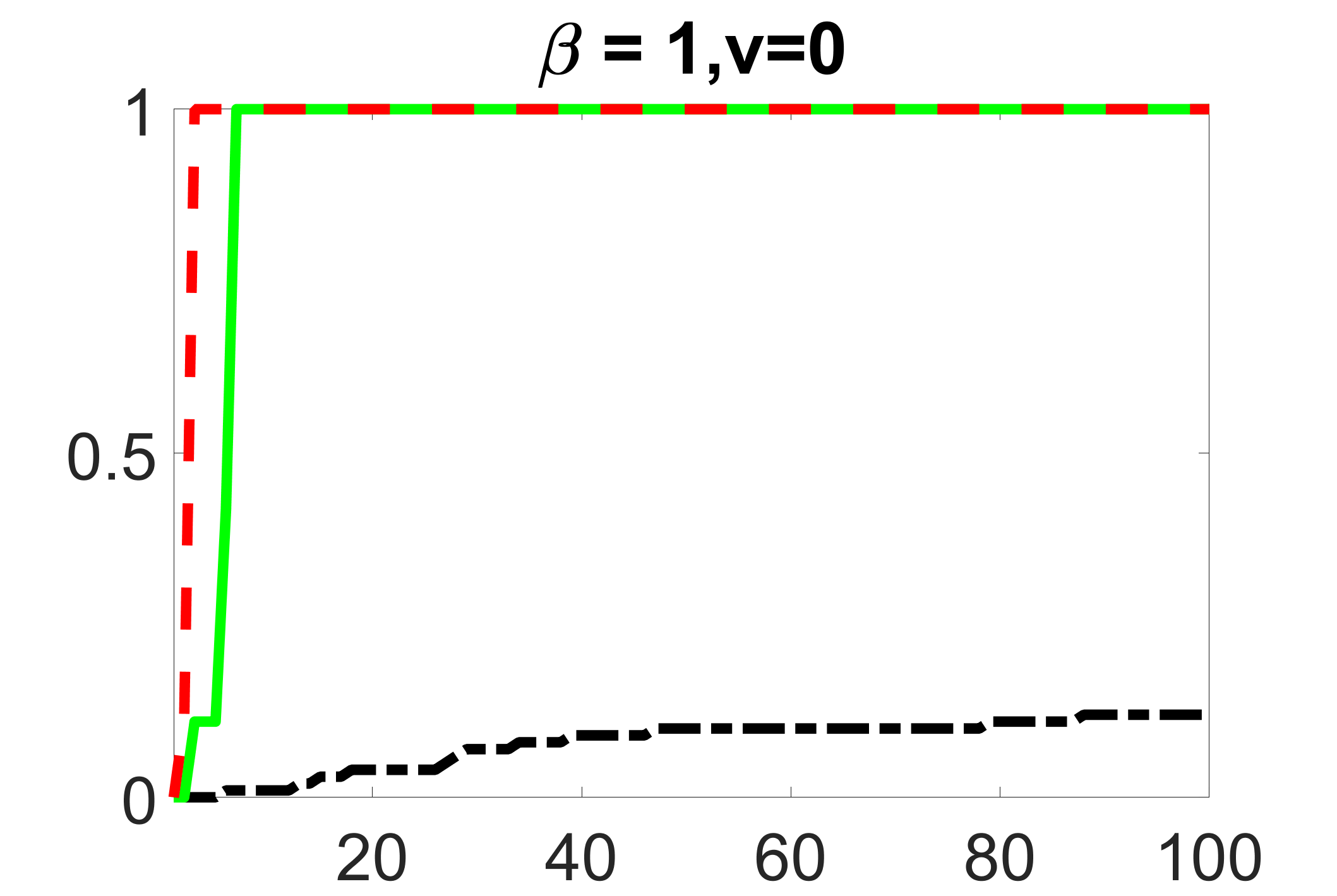}}
  \subcaptionbox{Precision: weak \\ outcome, zero exposure}[0.45\linewidth]
 {\includegraphics[width=6cm,height=3.5cm]{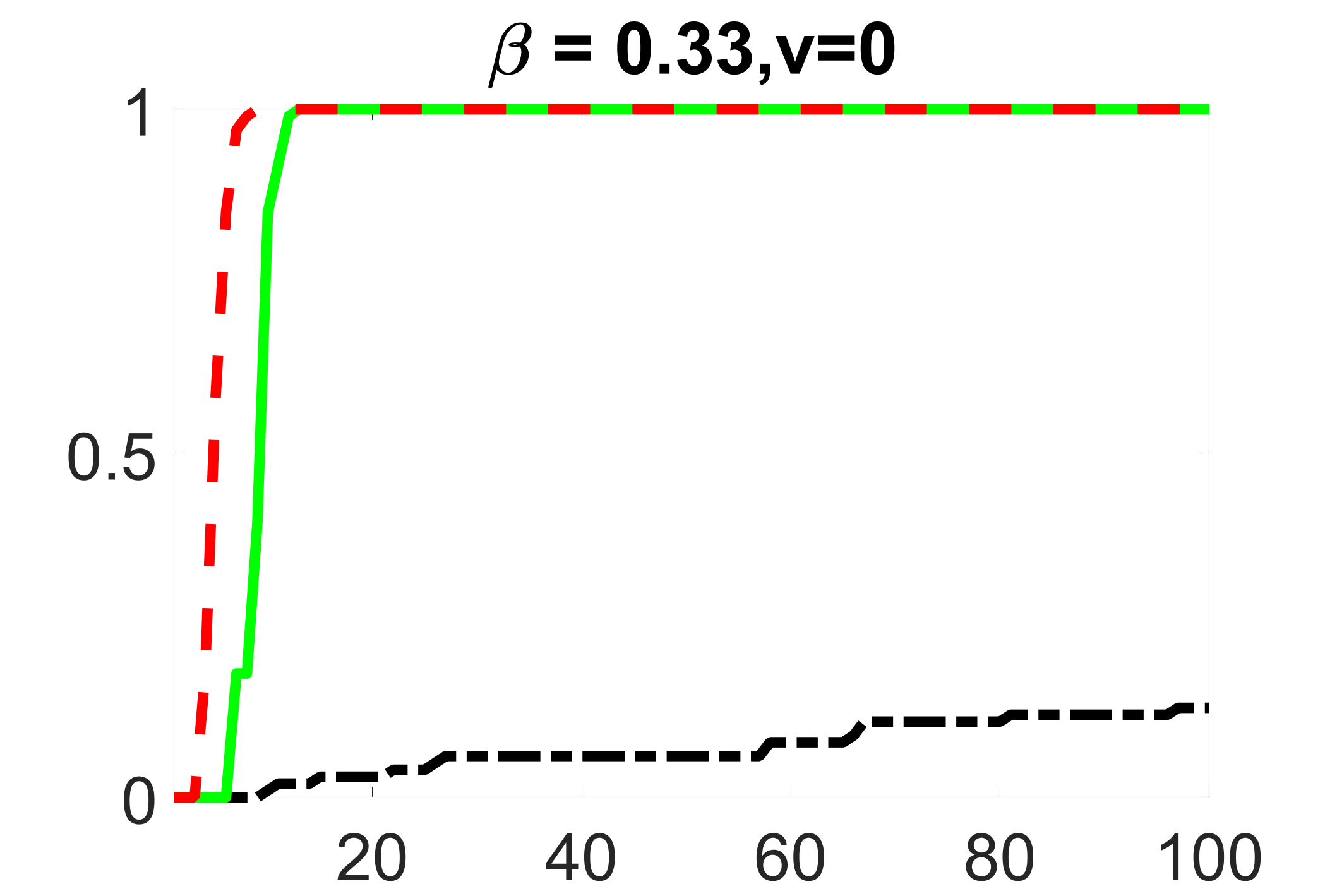}}
  \subcaptionbox{Overall coverage of $\mathcal{M}_1$}[0.45\linewidth]
 {\includegraphics[width=7cm,height=3.75cm]{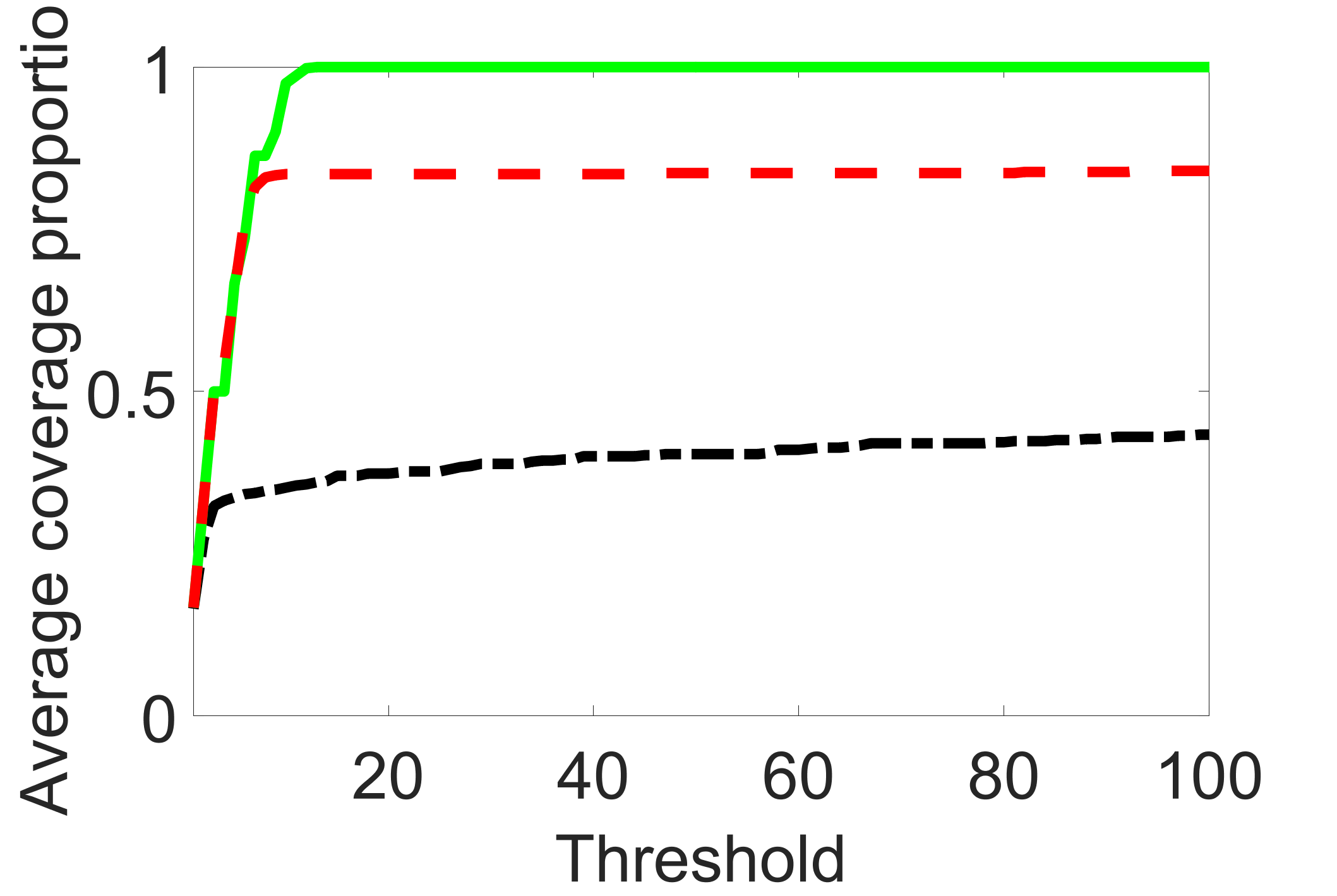}}
\caption{Simulation results for the case $(n,s,\sigma) = (500,5000,1)$  when $|\widehat{\mathcal{M}}_2|/|\widehat{\mathcal{M}}_1^*|=2$: Panels (a) -- (f) plot the average coverage proportion for $X_l$, where $l=1,2,3,104,105$ and $106$. Panels (a) -- (c) correspond to strong outcome and weak exposure predictor, moderate outcome and moderate exposure predictor and weak outcome and strong exposure predictor; Panels (d) -- (f) correspond to strong, moderate and weak predictors of outcome only. Panel (g) plots the average coverage proportion for the index set $\mathcal{M}_1 = \{1,2,3,104,105,106\}$. The x-axis represents the size of $\widehat{\mathcal{M}} $, while
y-axis denotes the average proportion. The green solid, the red dashed and the black dash dotted lines denote our joint screening method, the outcome screening method, and the intersection screening method, respectively.}
\label{sim1step1n500sigma1_propSb2}
\end{figure}

Furthermore, we evaluate the performance of our estimation procedure after the first-step screening. For the size of $ \widehat{\mathcal{M}} $ in the screening step, we set $|\widehat{\mathcal{M}}| = \lfloor n / \log(n) \rfloor$, so that $|\widehat{\mathcal{M}}|=89$ for sample size $n = 500$. We report the mean squared errors (MSEs) for $\bm{\beta}$ and $\bm{B}$ defined as
$||{\bm\beta}_{}-\widehat{\bm\beta}||_2^2$ and $ \|{\bm{B}}-\widehat{\bm{B}}\|_F^2$, respectively. As summarized in Table \ref{sim1t1mmn500sigma1propSb}, the average MSEs for $\bm\beta$ and $\bm{B}$ among 100 Monte Carlo runs are all similar to each other for the different choices of $|\widehat{\mathcal{M}}_1^*|$ and $|\widehat{\mathcal{M}}_2|$. Therefore, depending on the prior knowledge about the sizes and strengths of signals of confounding, precision and instrumental variables, one may choose $\widehat{\mathcal{M}}_1^*$ and $\widehat{\mathcal{M}}_2$ differently, though the estimations of  $\bm{B}$ appear to be similar among the different choices.

\begin{table}[htbp]
\caption{Simulation results of the proposed estimates for $(n,s,\sigma) = (500,5000,1)$ , when $|\widehat{\mathcal{M}}_2|/|\widehat{\mathcal{M}}_1^*|$ is taken as 0.5, 1.0 and 2.0: the average MSEs for $\bm\beta$ and $ {\bm B}$, and their associated standard errors in the parentheses are reported. The results are based on 100 Monte Carlo repetitions.}
\centering
\begin{tabular}{ ccc}
$|\widehat{\mathcal{M}}_2|/|\widehat{\mathcal{M}}_1^*|$ & MSE $\bm\beta$ & MSE ${\bm{B}}$\\
\hline
0.5 &0.301(0.008)&0.567(0.005)\\
1.0&0.303(0.008)&0.574(0.006)\\
2.0&0.302(0.008)&0.574(0.006)\\
\end{tabular}
\label{sim1t1mmn500sigma1propSb}
\end{table}

\subsection{Screening and estimation using blockwise joint screening}
\label{Group Selection Using Linkage Disequilibrium Information}
In this section, we list additional screening and estimation results for Section \ref{Group Selection Using Linkage Disequilibrium Information main} of the main article. In particular, the results for the screening step, in which $s=5000$, $\sigma=1$, $n=200,500,1000$, and $K =2,4,6,12,24,52$, can be found in Figures \ref{sim3step1n200sizesig2sigma1} -- \ref{sim3step1n1000sizesig52sigma1}. The complete results for the second step estimation, in which $s=5000$, $\sigma=1$, $n=200, 500, 1000$, and $K =2,4,6,12,24,52$, can be found in Table \ref{sim1t2}.

\begin{figure}[htbp]
\captionsetup[subfigure]{justification=centering}
\centering
 \subcaptionbox{\footnotesize Confounder: strong \\ outcome, weak exposure}[0.45\linewidth]
 {\includegraphics[width=6cm,height=3.5cm]{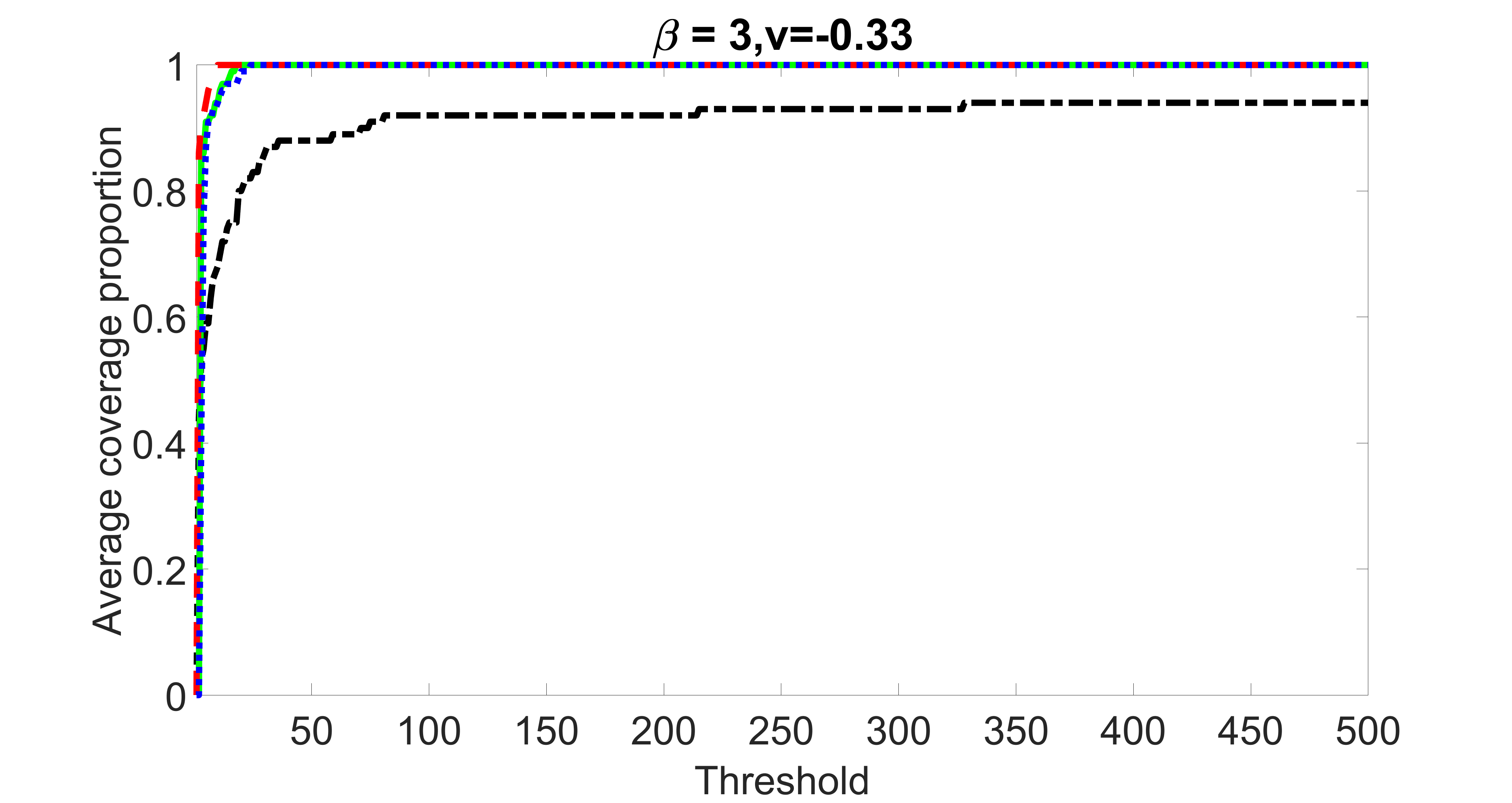}}
 \subcaptionbox{\footnotesize Confounder: medium \\ outcome, medium exposure}[0.45\linewidth]
 {\includegraphics[width=6cm,height=3.5cm]{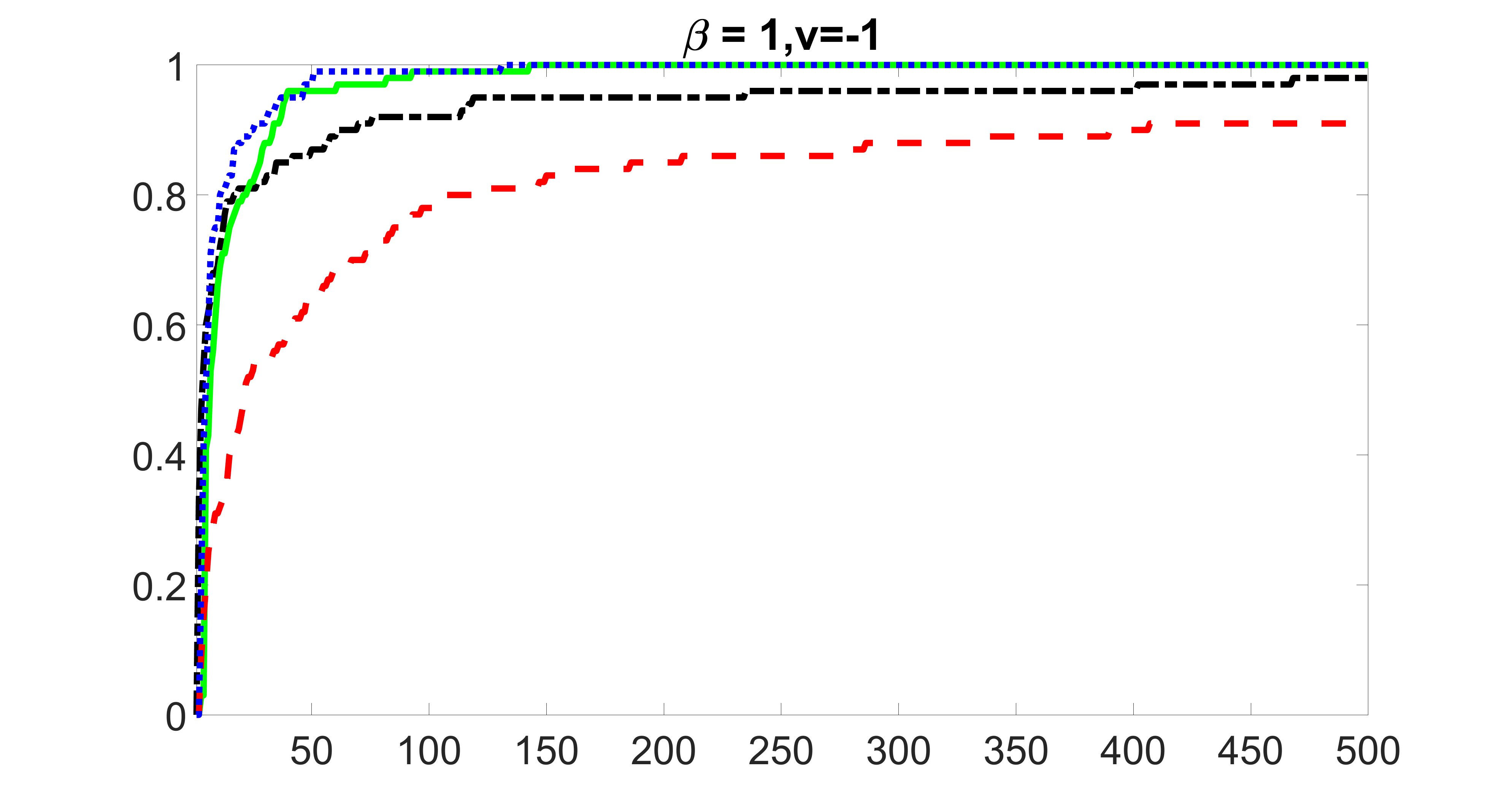}}
  \subcaptionbox{\footnotesize Confounder: weak \\ outcome, strong exposure}[0.45\linewidth]
 {\includegraphics[width=6cm,height=3.5cm]{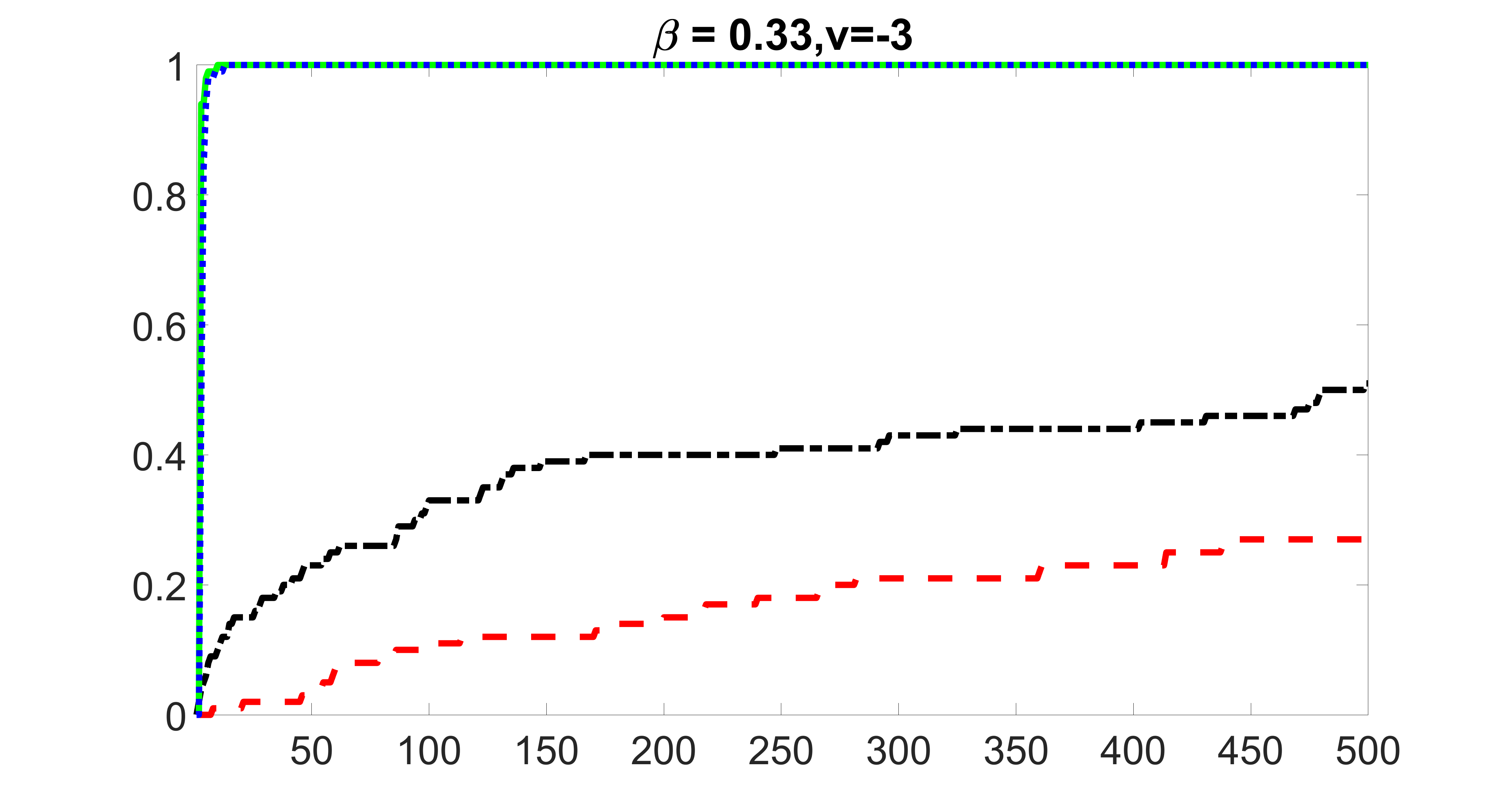}}
  \subcaptionbox{\footnotesize Precision: strong \\ outcome, zero exposure}[0.45\linewidth]
 {\includegraphics[width=6cm,height=3.5cm]{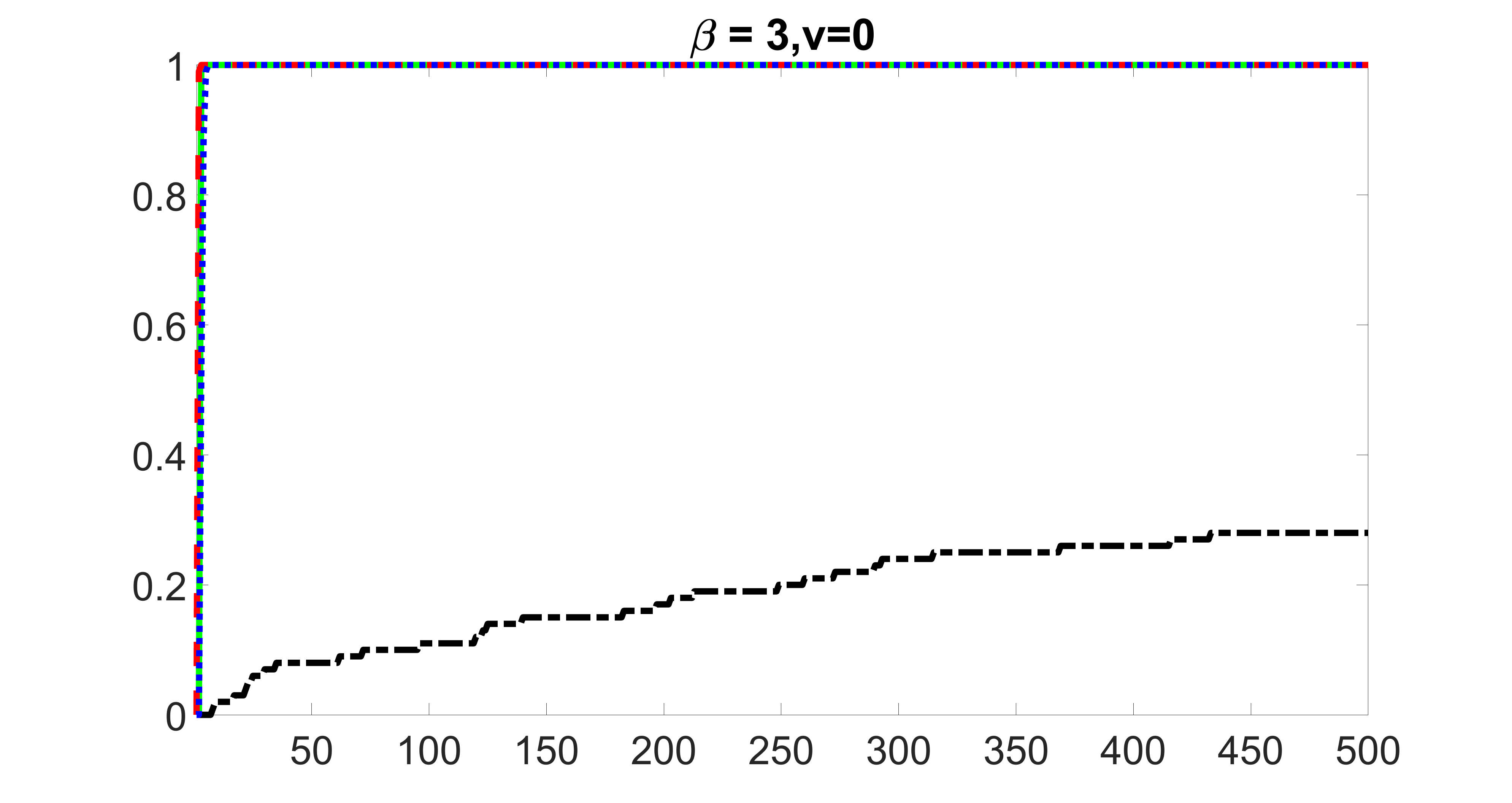}}
  \subcaptionbox{\footnotesize Precision: medium \\ outcome, zero exposure}[0.45\linewidth]
 {\includegraphics[width=6cm,height=3.5cm]{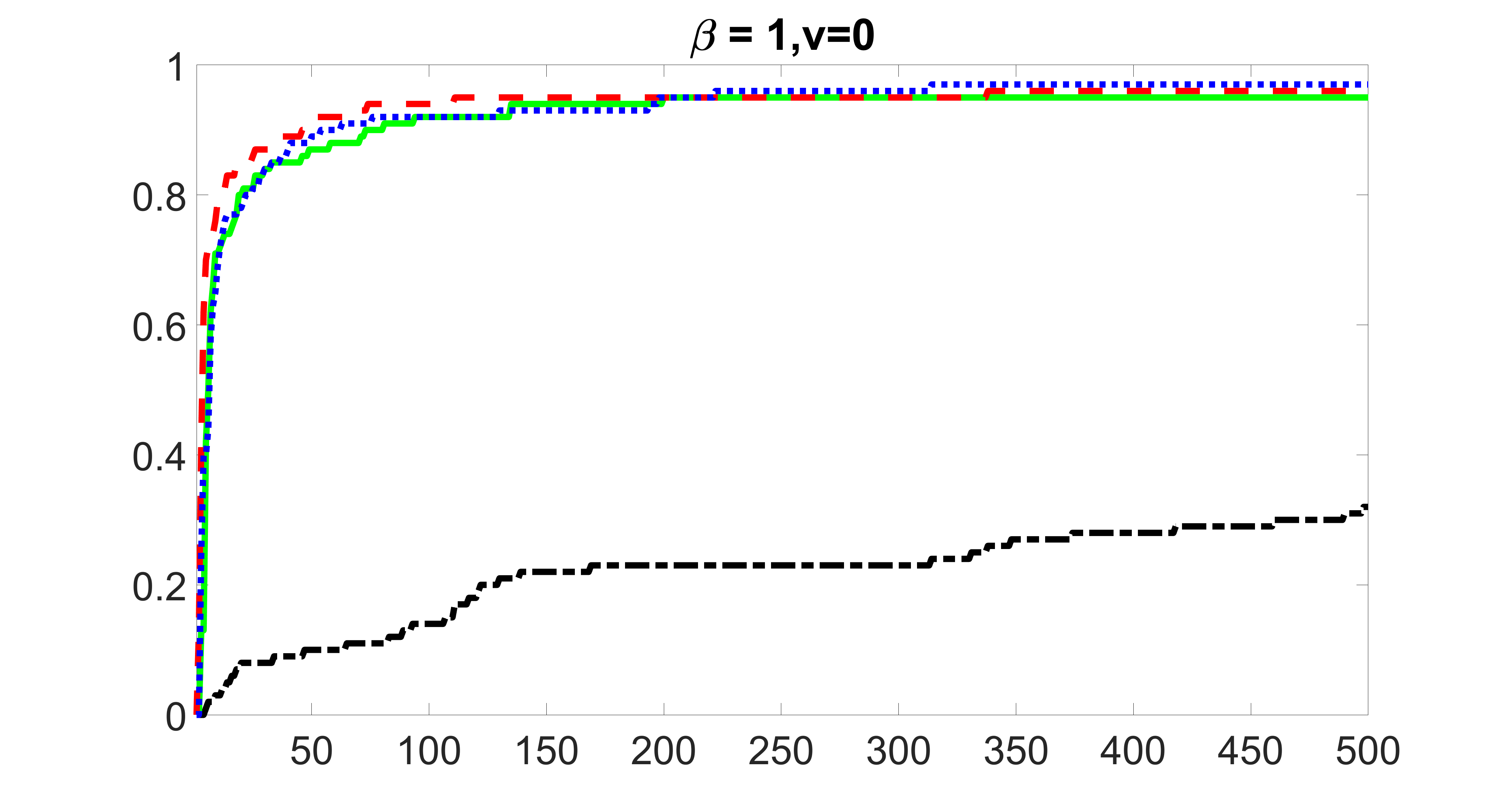}}
  \subcaptionbox{\footnotesize Precision: weak \\ outcome, zero exposure}[0.45\linewidth]
 {\includegraphics[width=6cm,height=3.5cm]{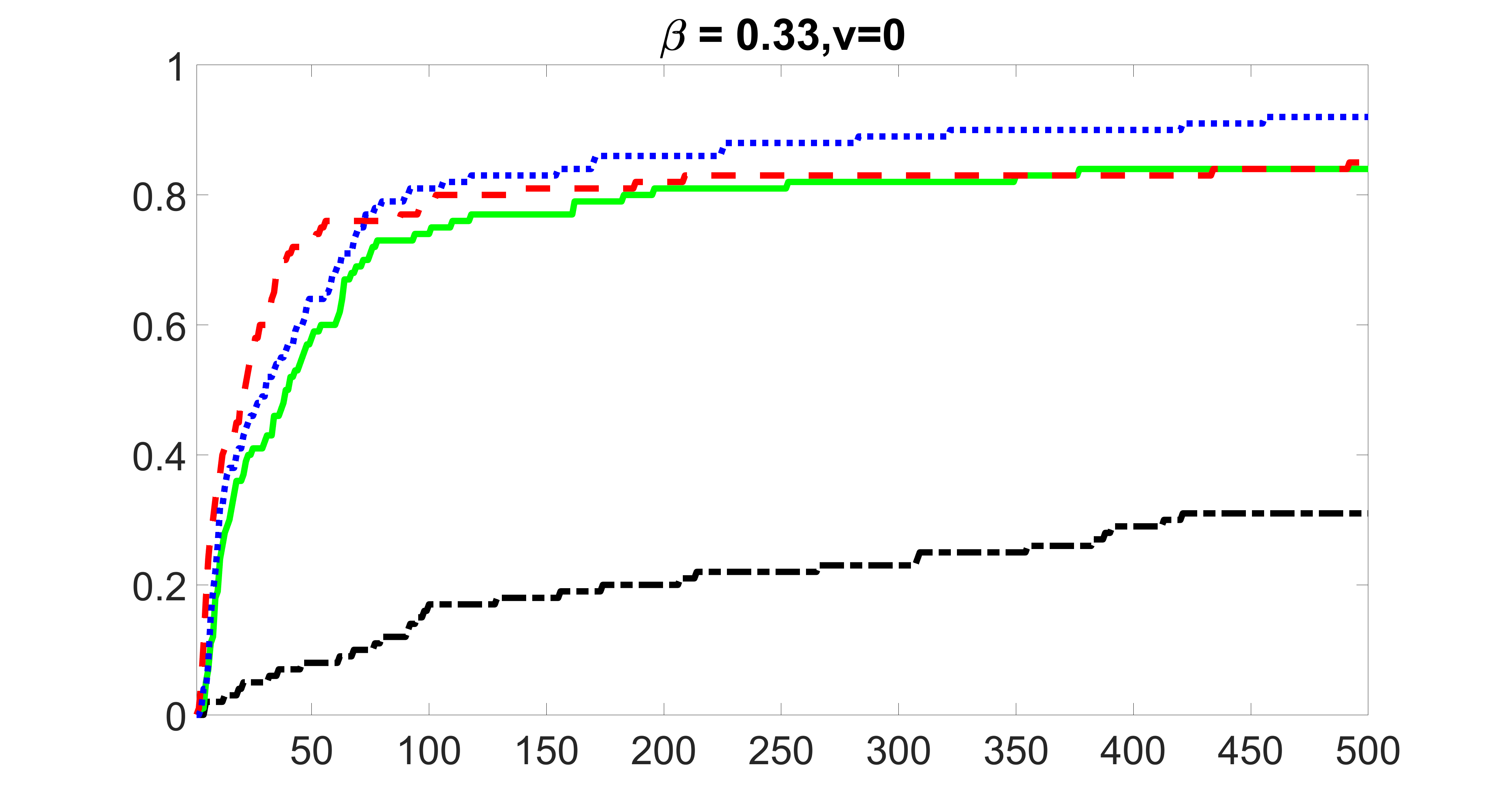}}
 \subcaptionbox{\footnotesize Precision: weaker \\ outcome, zero exposure}[0.45\linewidth]
 {\includegraphics[width=6cm,height=3.5cm]{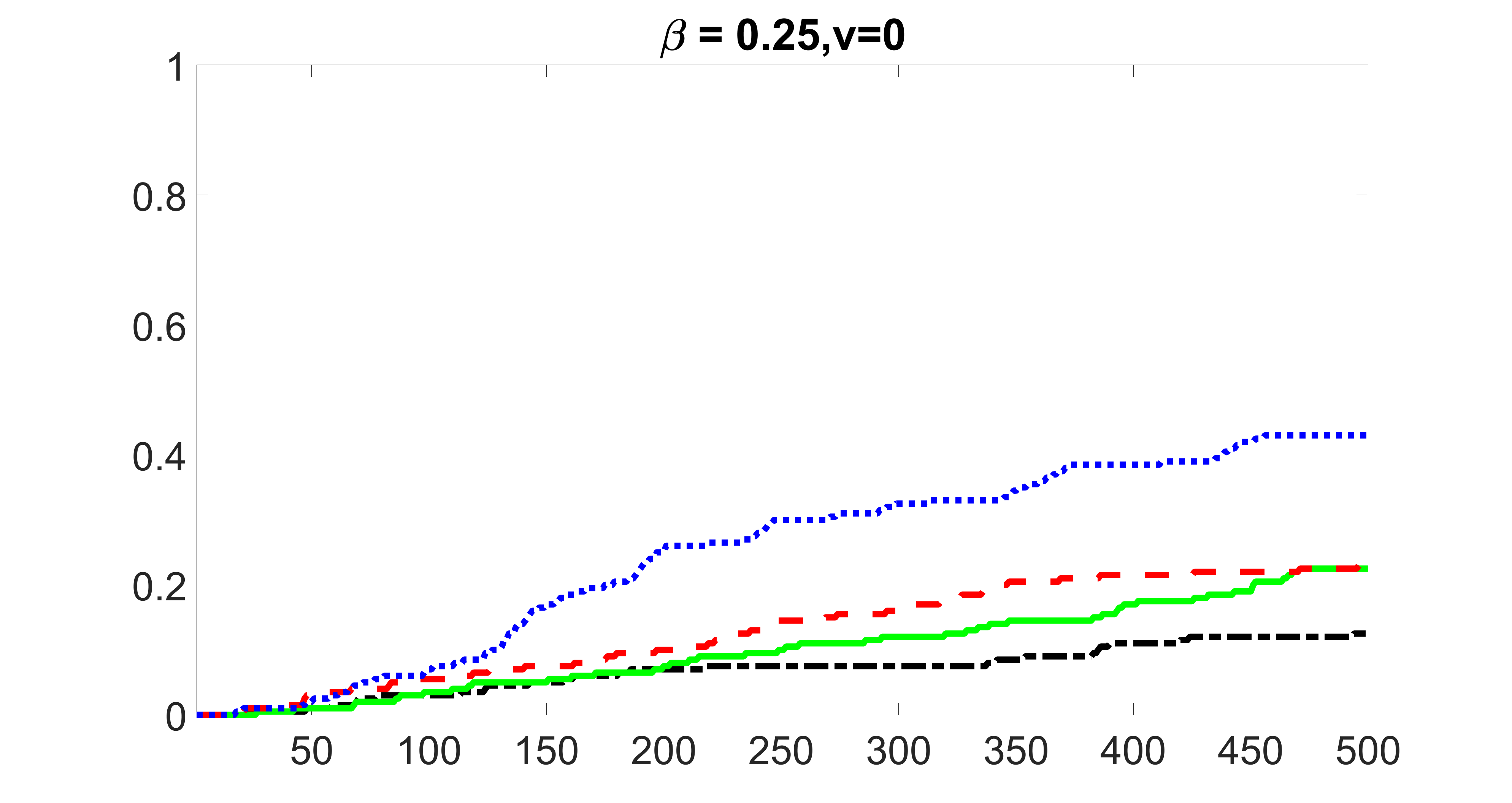}}
  \subcaptionbox{Overall coverage of $\mathcal{M}_1$}[0.45\linewidth]
 {\includegraphics[width=6cm,height=3.5cm]{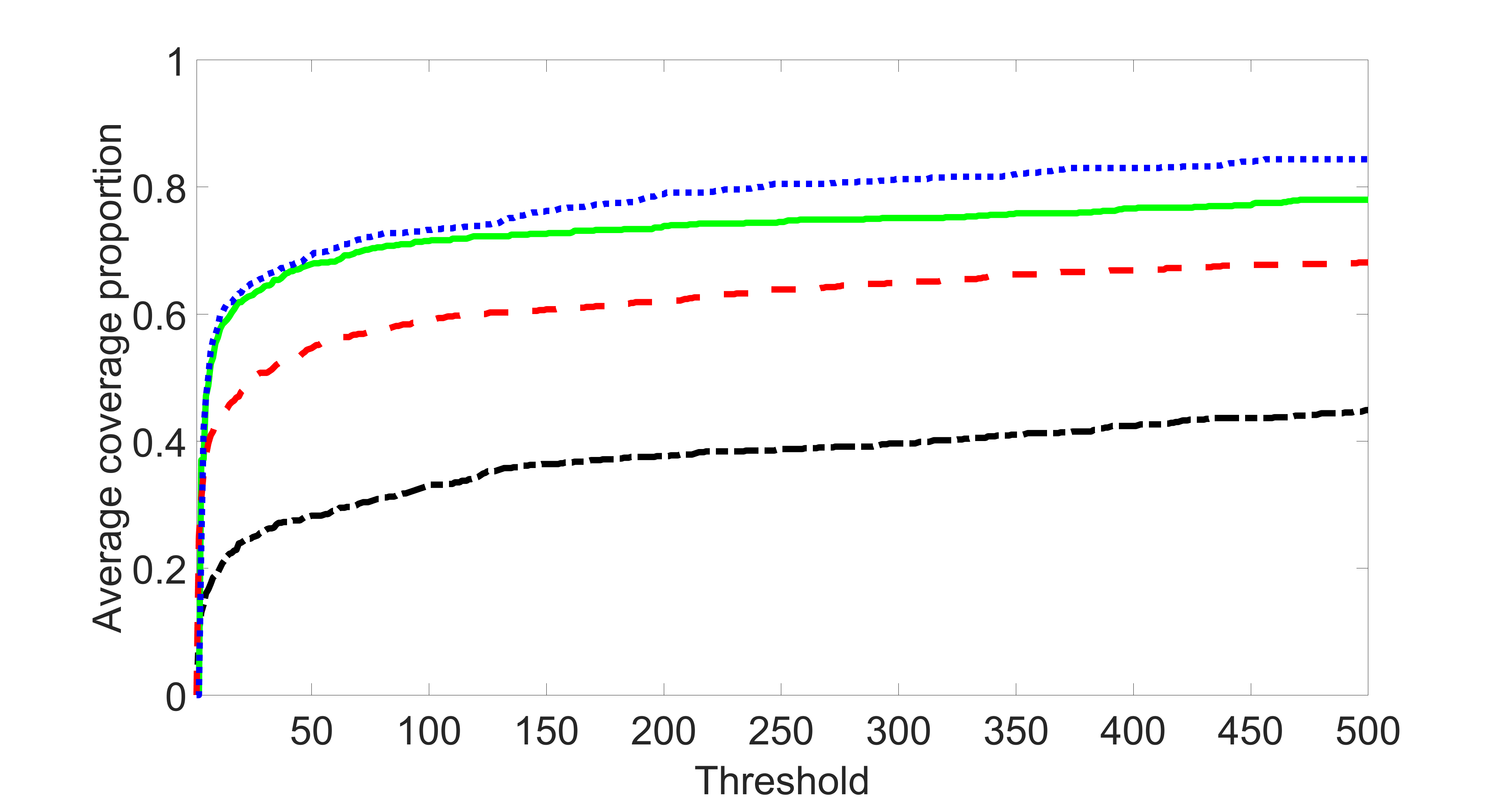}}
\caption{ Simulation results for the case $(n,s,K,\sigma) = (200,5000,2,1)$: Panels (a) -- (g) plot the average coverage proportion for $X_l$, where $l \in \mathcal{M}_1 =  \{1,2,3,104,105, 106\} \cup \mathcal{P}_{LD}$. Panels (a) -- (c) correspond to strong outcome and weak exposure predictor, moderate outcome and moderate exposure predictor and weak outcome and strong exposure predictor; Panels (d) -- (g) correspond to strong, moderate, and weak predictors of outcome only. Panel (g) plots the average coverage proportion for the index set $\mathcal{P}_{LD}$. Panel (h) plots the average coverage proportion for the index set $\mathcal{M}_1$. The x-axis represents the size of $\widehat{\mathcal{M}} $, while
y-axis denotes the average proportion. The blue dot, green solid, red dashed and black dash dotted lines denote the blockwise joint screening, joint screening, outcome screening, and intersection screening methods, respectively.}
\label{sim3step1n200sizesig2sigma1}
\end{figure}

\begin{figure}[htbp]
\captionsetup[subfigure]{justification=centering}
\centering
 \subcaptionbox{\footnotesize Confounder: strong \\ outcome, weak exposure}[0.45\linewidth]
 {\includegraphics[width=6cm,height=3.5cm]{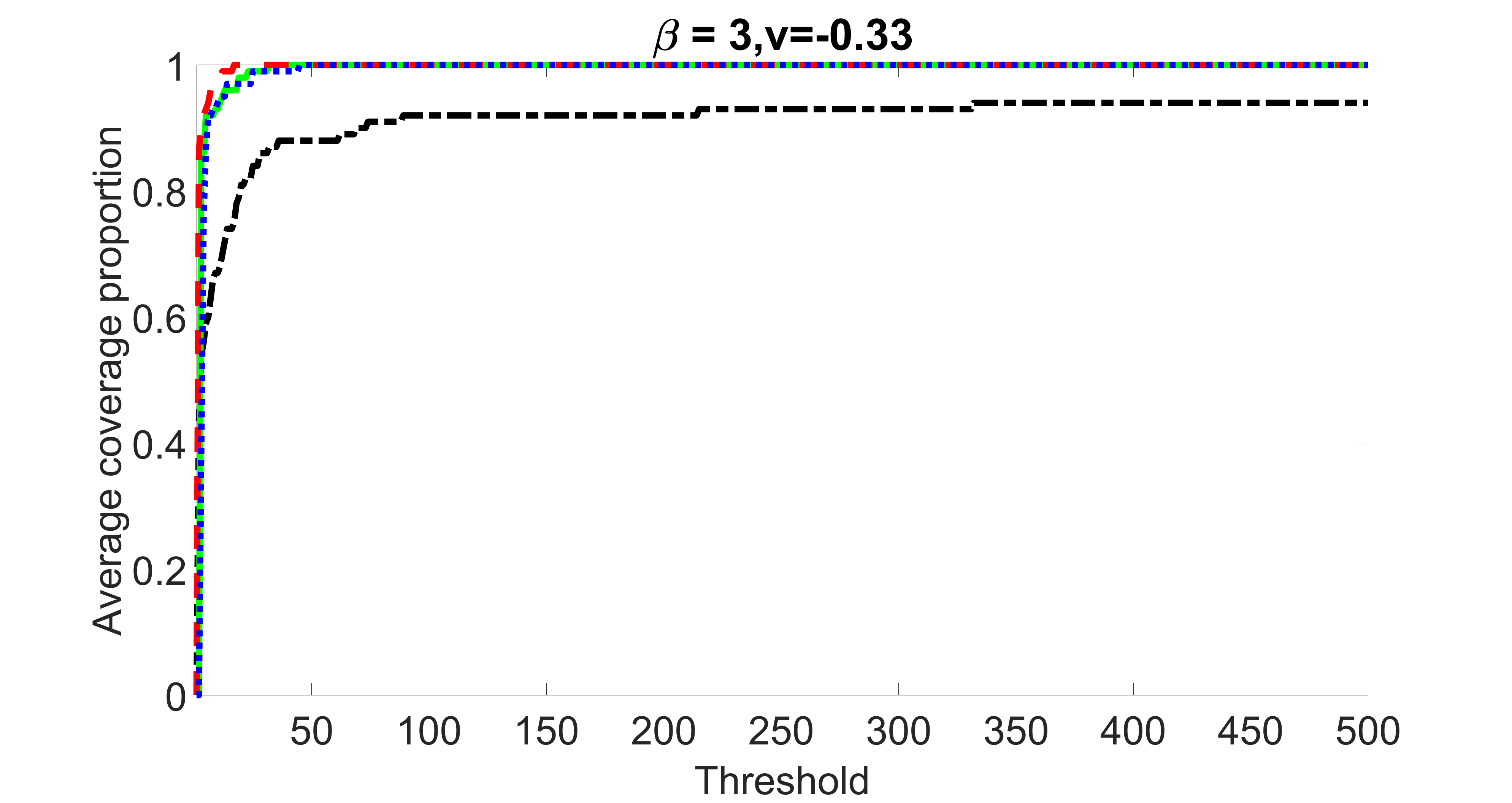}}
 \subcaptionbox{\footnotesize Confounder: medium \\ outcome, medium exposure}[0.45\linewidth]
 {\includegraphics[width=6cm,height=3.5cm]{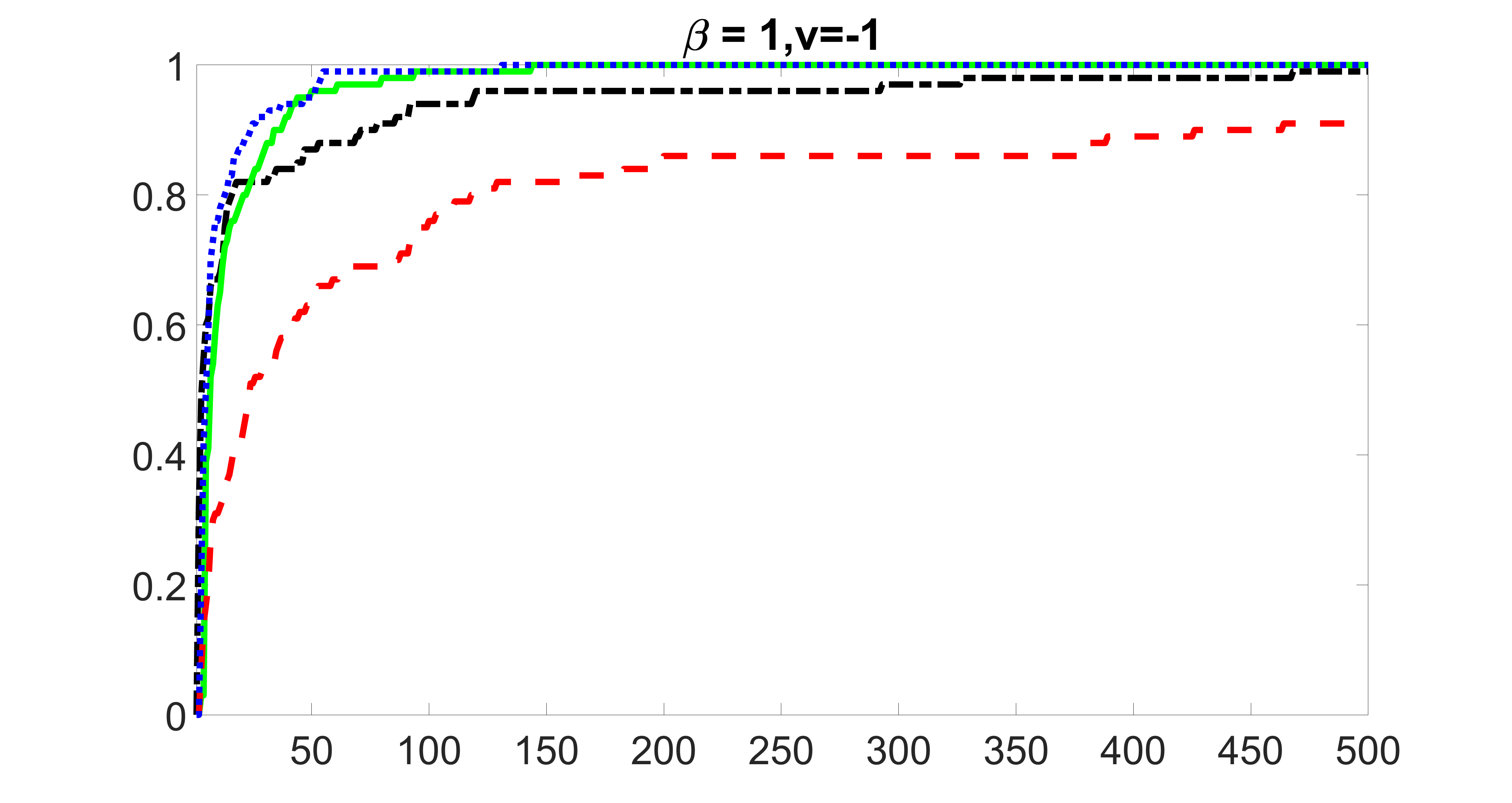}}
  \subcaptionbox{\footnotesize Confounder: weak \\ outcome, strong exposure}[0.45\linewidth]
 {\includegraphics[width=6cm,height=3.5cm]{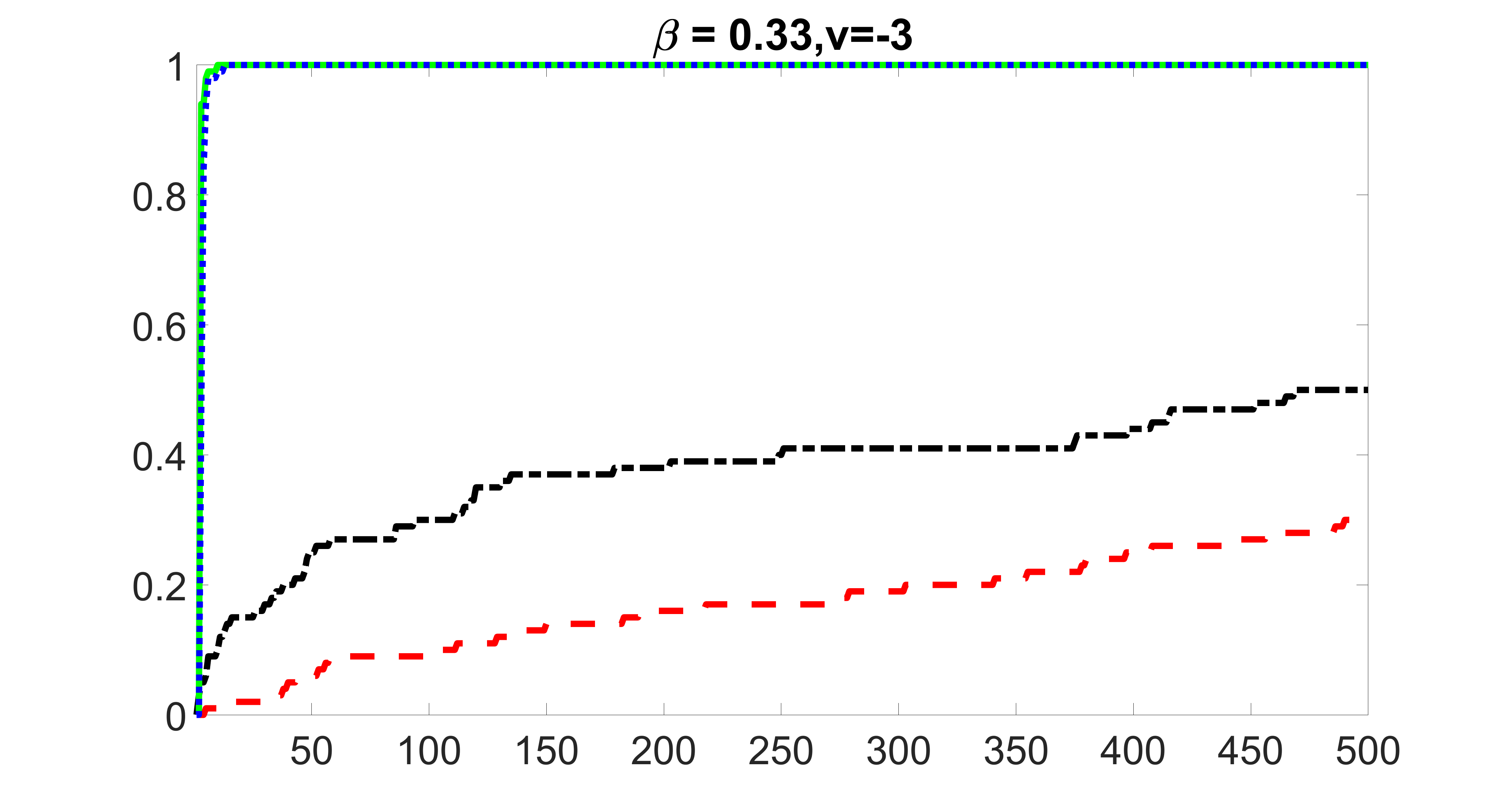}}
  \subcaptionbox{\footnotesize Precision: strong \\ outcome, zero exposure}[0.45\linewidth]
 {\includegraphics[width=6cm,height=3.5cm]{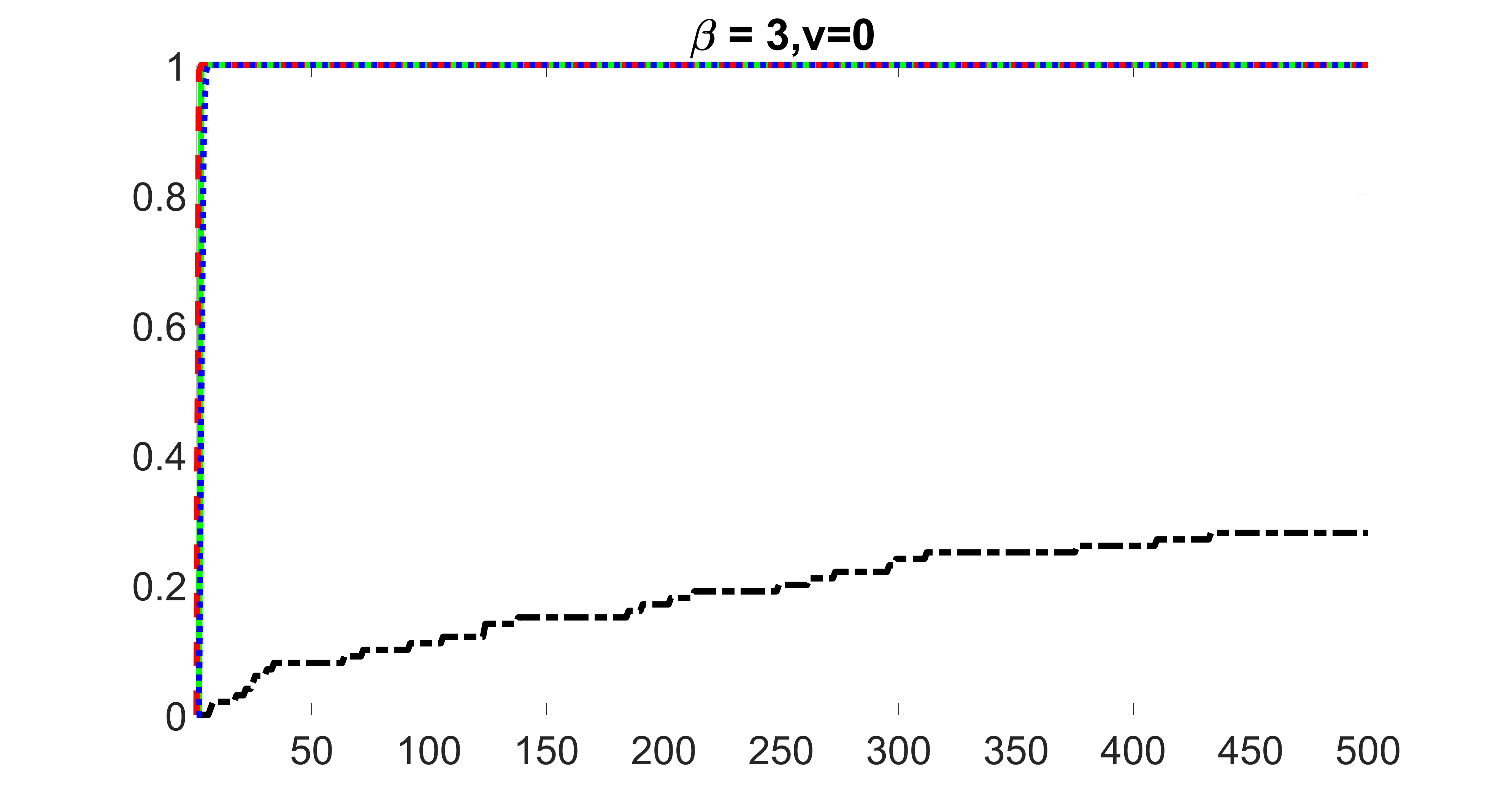}}
  \subcaptionbox{\footnotesize Precision: medium \\ outcome, zero exposure}[0.45\linewidth]
 {\includegraphics[width=6cm,height=3.5cm]{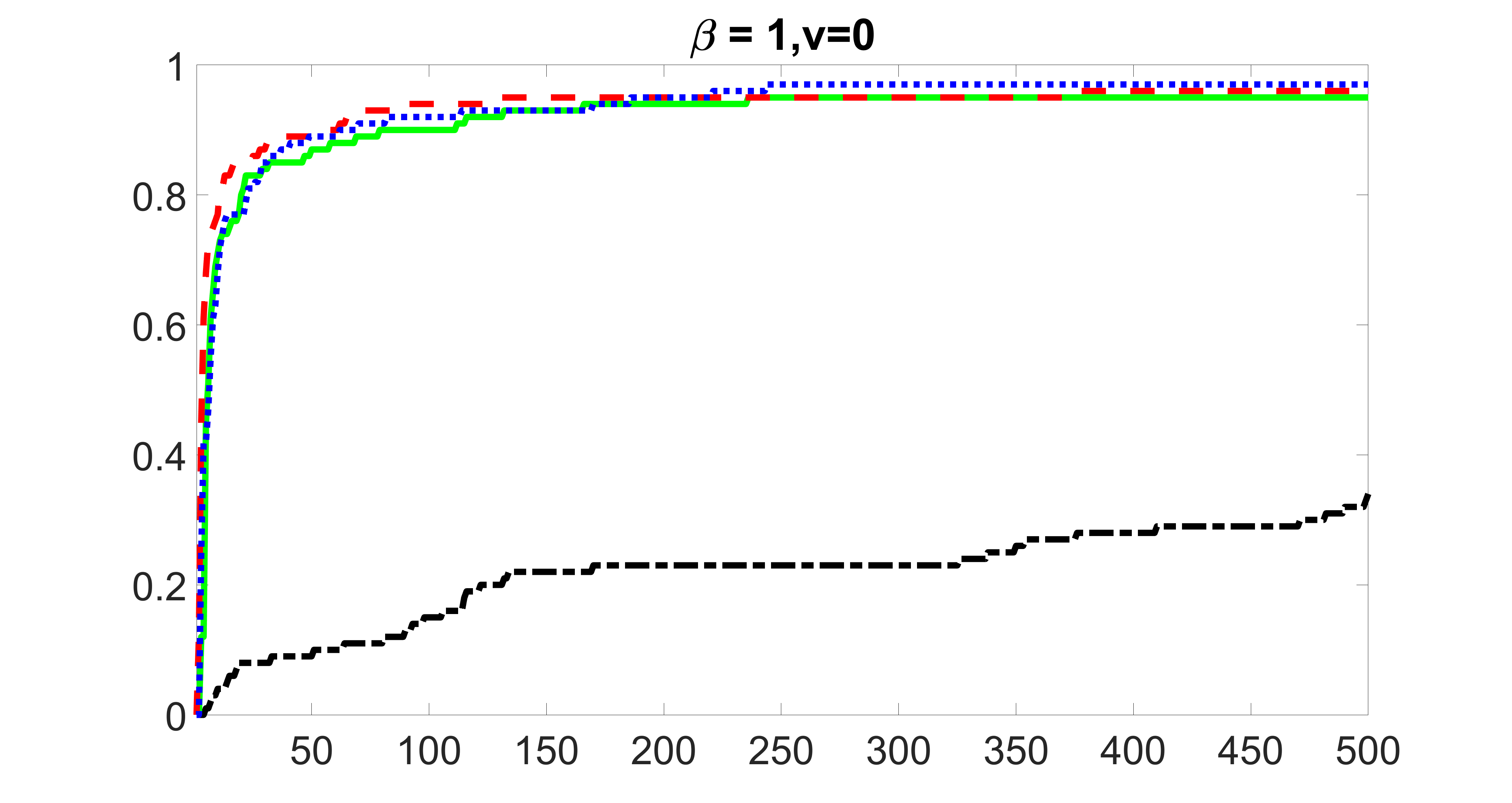}}
  \subcaptionbox{\footnotesize Precision: weak \\ outcome, zero exposure}[0.45\linewidth]
 {\includegraphics[width=6cm,height=3.5cm]{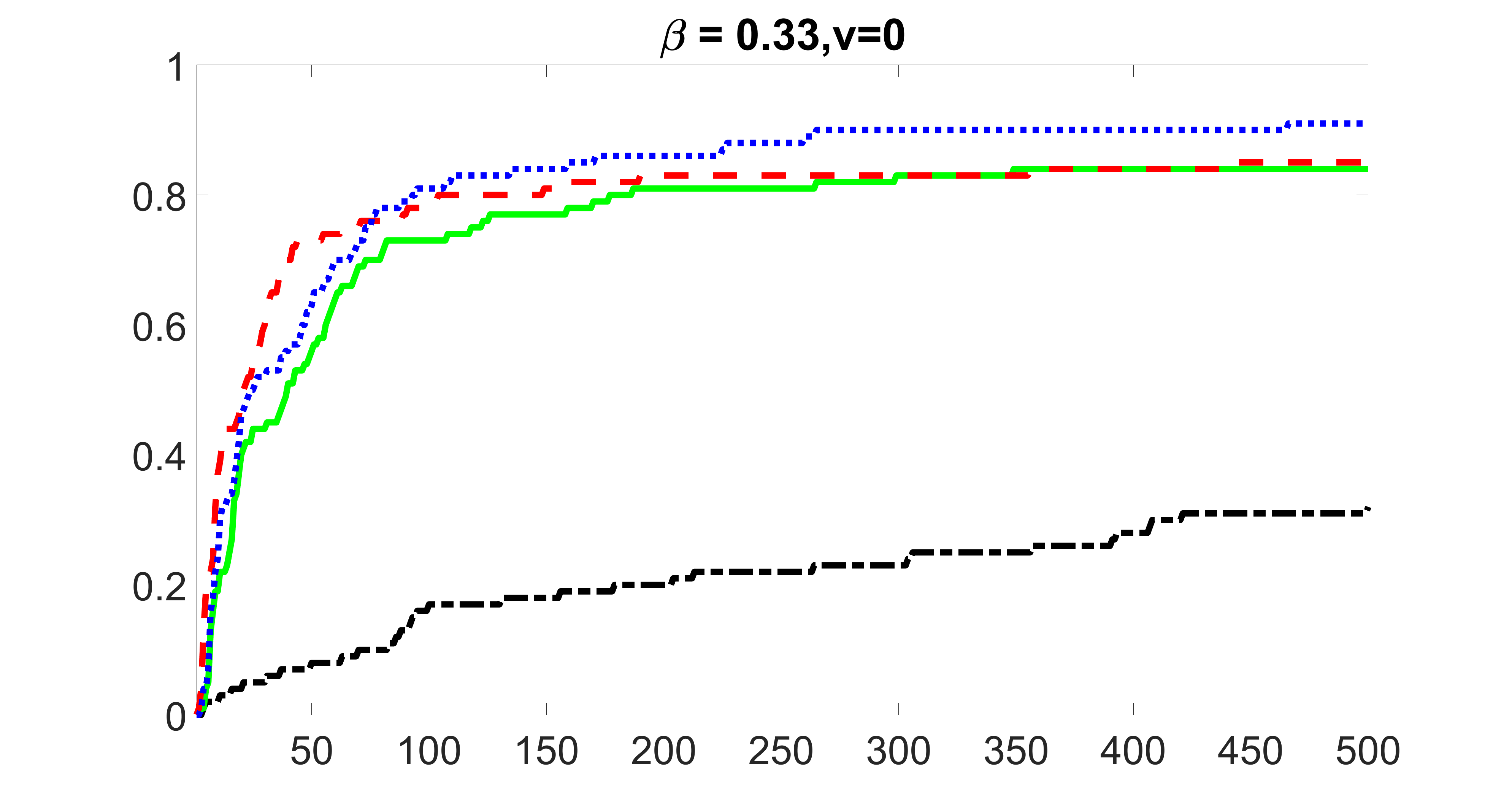}}
 \subcaptionbox{\footnotesize Precision: weaker \\ outcome, zero exposure}[0.45\linewidth]
 {\includegraphics[width=6cm,height=3.5cm]{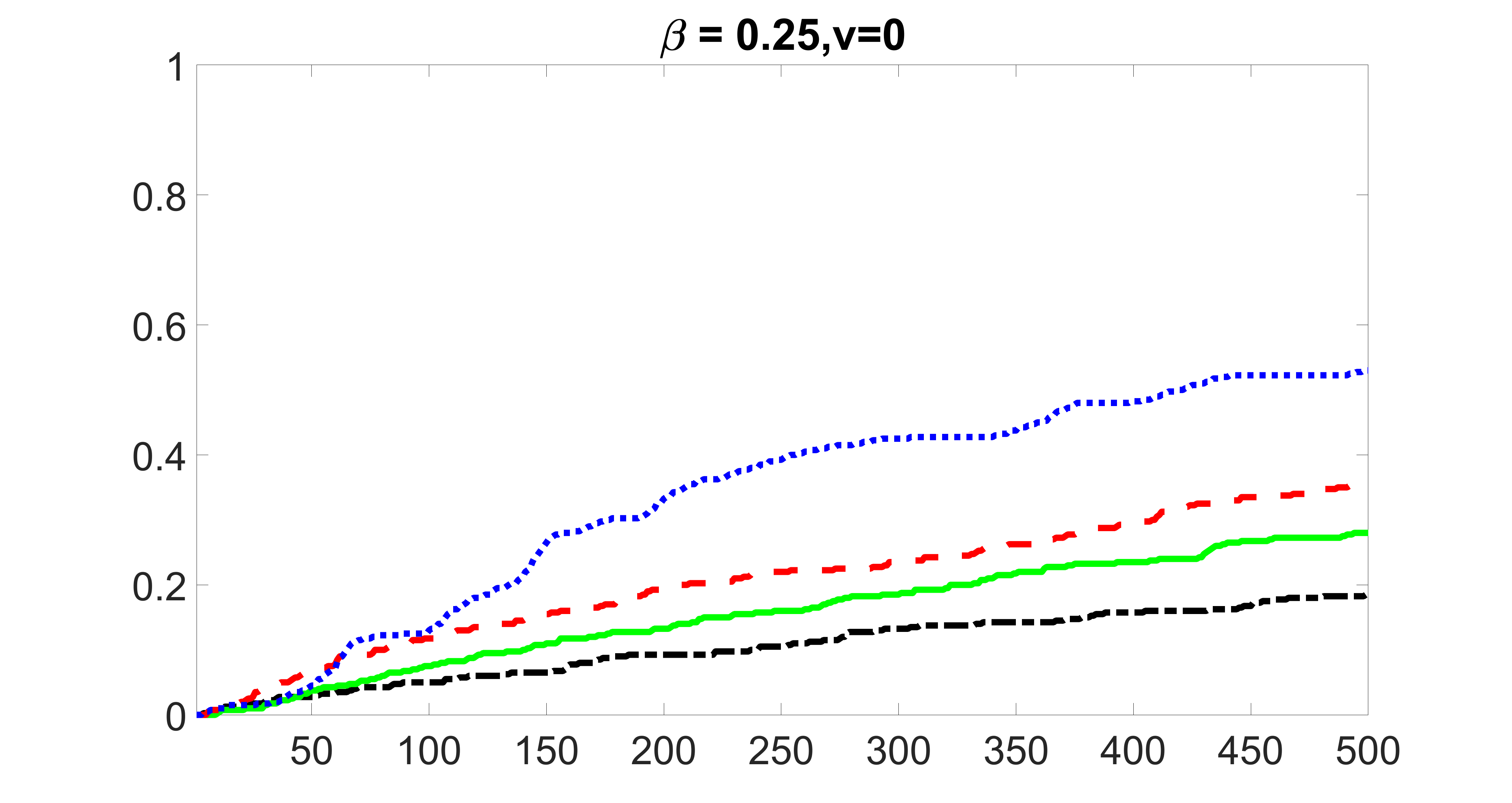}}
  \subcaptionbox{Overall coverage of $\mathcal{M}_1$}[0.45\linewidth]
 {\includegraphics[width=6cm,height=3.5cm]{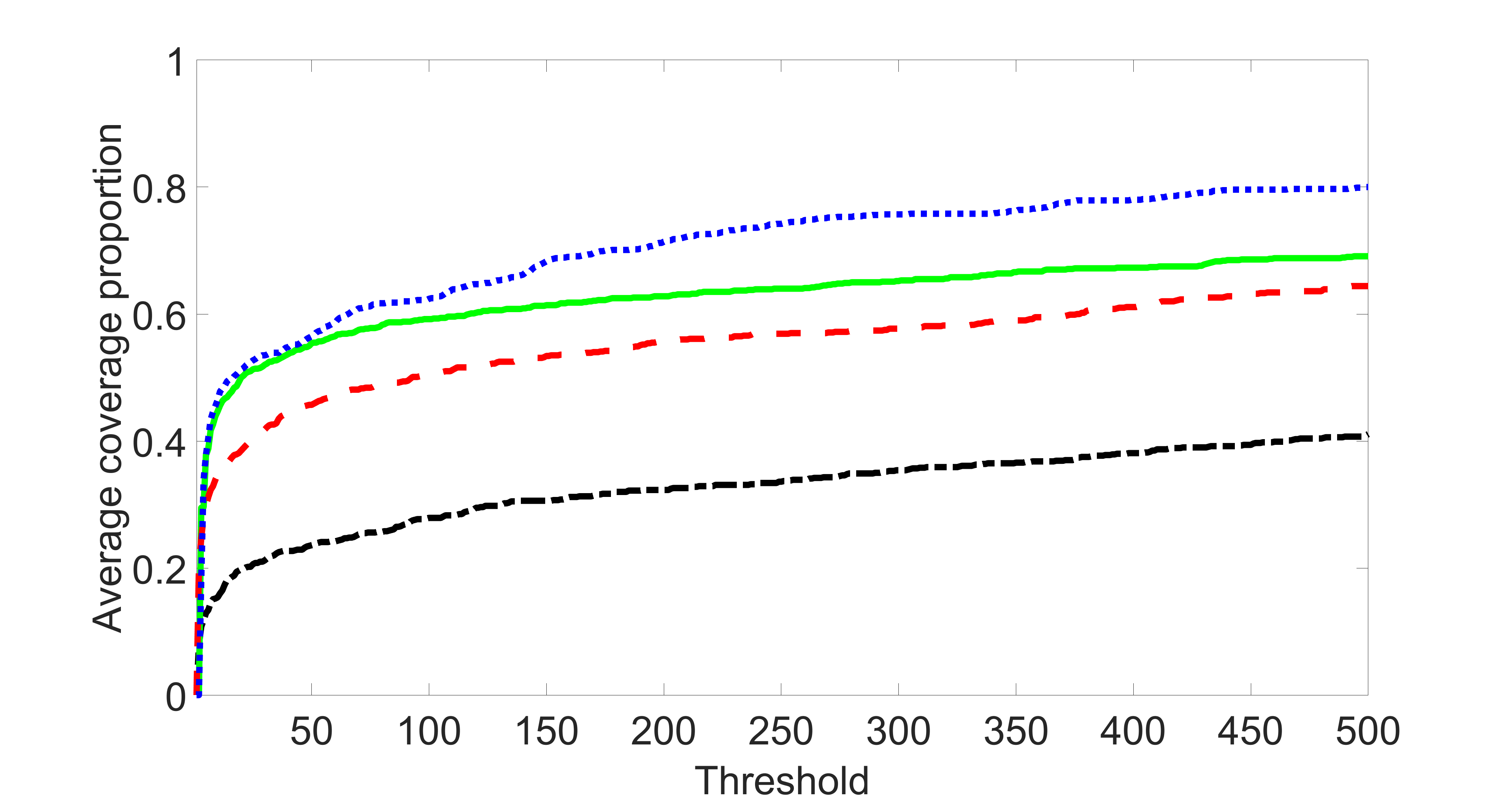}}
\caption{ Simulation results for the case $(n,s,K,\sigma) = (200,5000,4,1)$: Panels (a) -- (g) plot the average coverage proportion for $X_l$, where $l \in \mathcal{M}_1 =  \{1,2,3,104,105, 106\} \cup \mathcal{P}_{LD}$. Panels (a) -- (c) correspond to strong outcome and weak exposure predictor, moderate outcome and moderate exposure predictor and weak outcome and strong exposure predictor; Panels (d) -- (g) correspond to strong, moderate, and weak predictors of outcome only. Panel (g) plots the average coverage proportion for the index set $\mathcal{P}_{LD}$. Panel (h) plots the average coverage proportion for the index set $\mathcal{M}_1$. The x-axis represents the size of $\widehat{\mathcal{M}} $, while
y-axis denotes the average proportion. The blue dot, green solid, red dashed and black dash dotted lines denote the blockwise joint screening, joint screening, outcome screening, and intersection screening methods, respectively.}
\label{sim3step1n200sizesig4sigma1}
\end{figure}

\begin{figure}[htbp]
\captionsetup[subfigure]{justification=centering}
\centering
 \subcaptionbox{\footnotesize Confounder: strong \\ outcome, weak exposure}[0.45\linewidth]
 {\includegraphics[width=6cm,height=3.5cm]{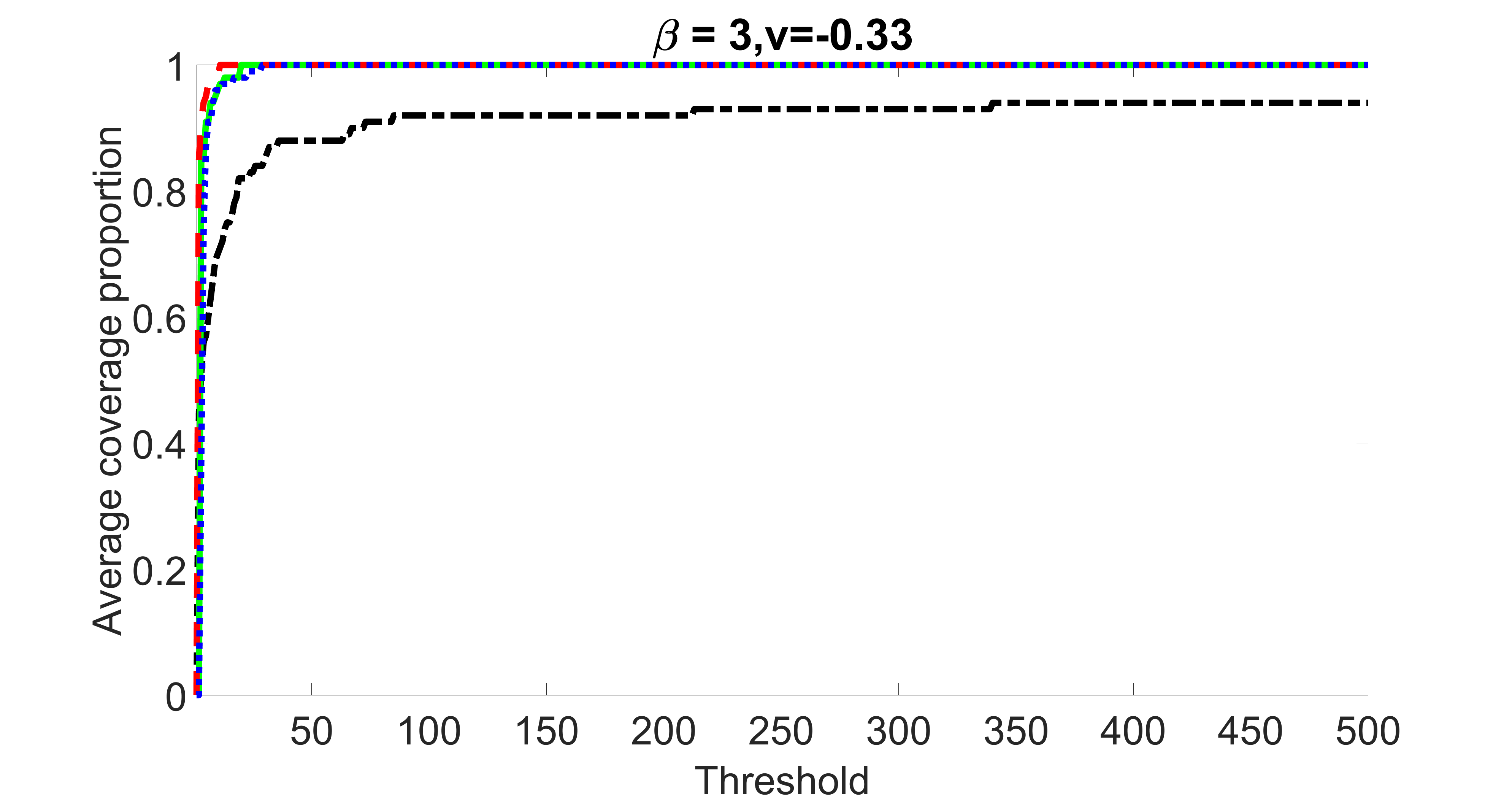}}
 \subcaptionbox{\footnotesize Confounder: medium \\ outcome, medium exposure}[0.45\linewidth]
 {\includegraphics[width=6cm,height=3.5cm]{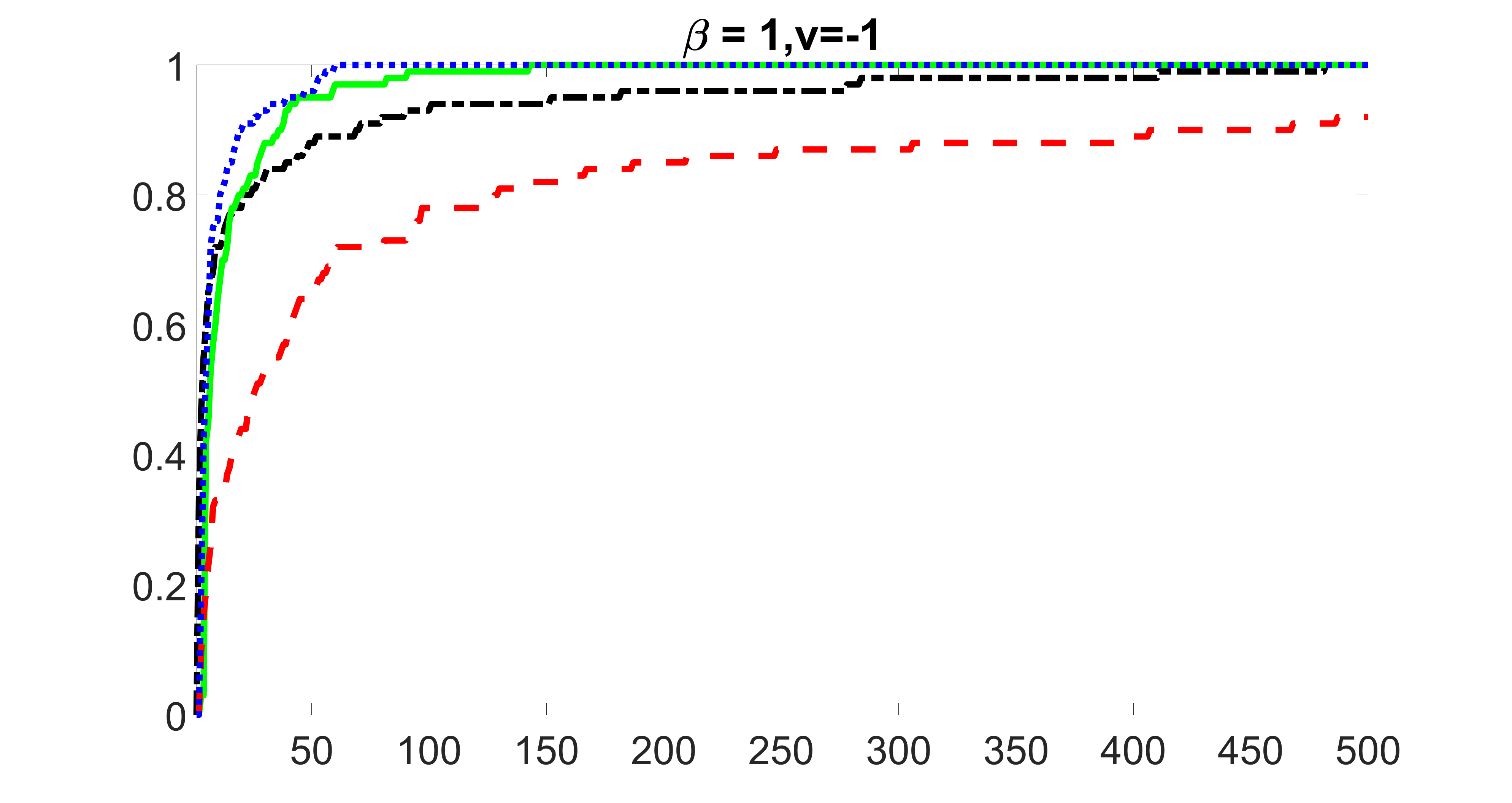}}
  \subcaptionbox{\footnotesize Confounder: weak \\ outcome, strong exposure}[0.45\linewidth]
 {\includegraphics[width=6cm,height=3.5cm]{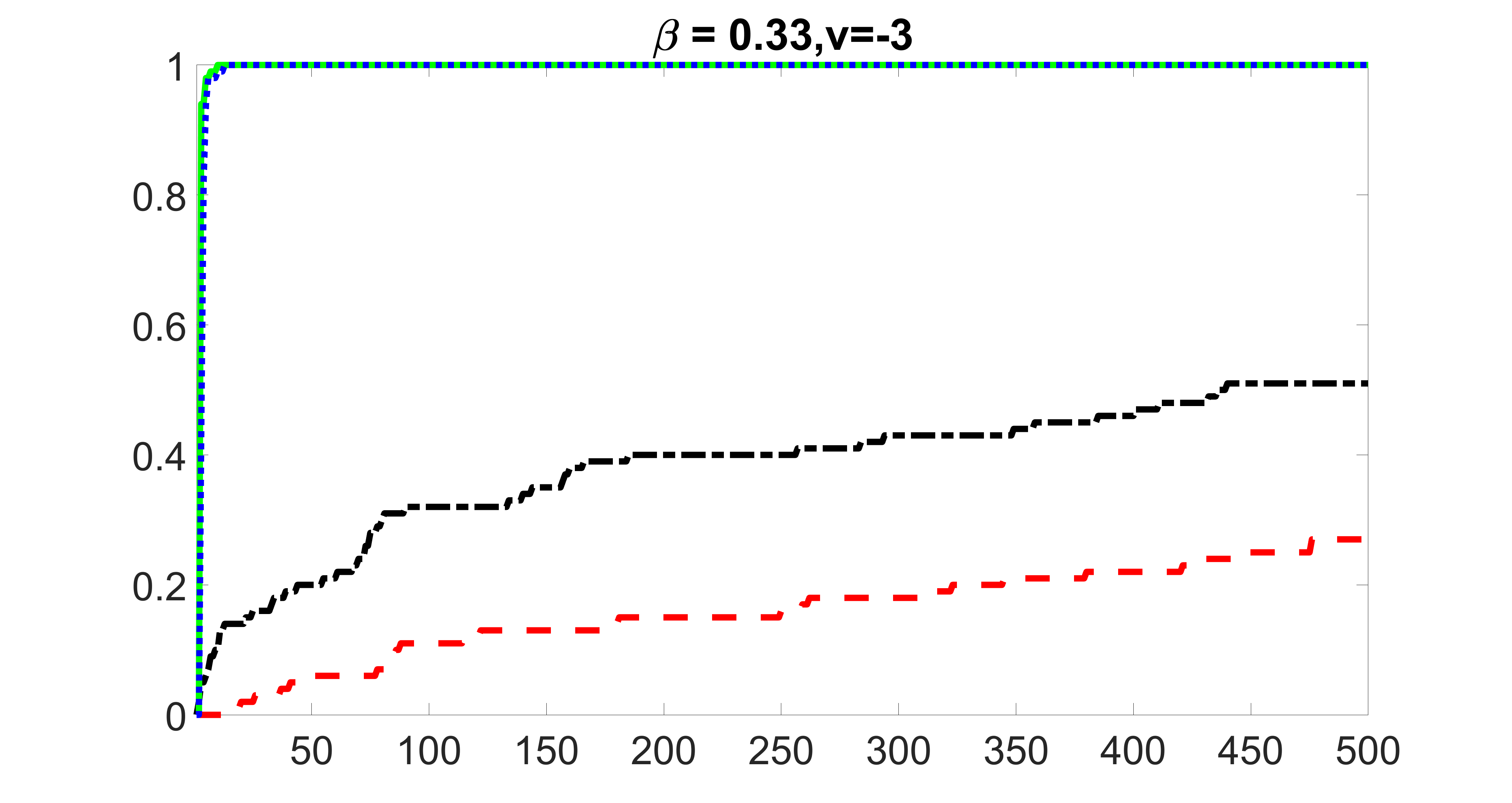}}
  \subcaptionbox{\footnotesize Precision: strong \\ outcome, zero exposure}[0.45\linewidth]
 {\includegraphics[width=6cm,height=3.5cm]{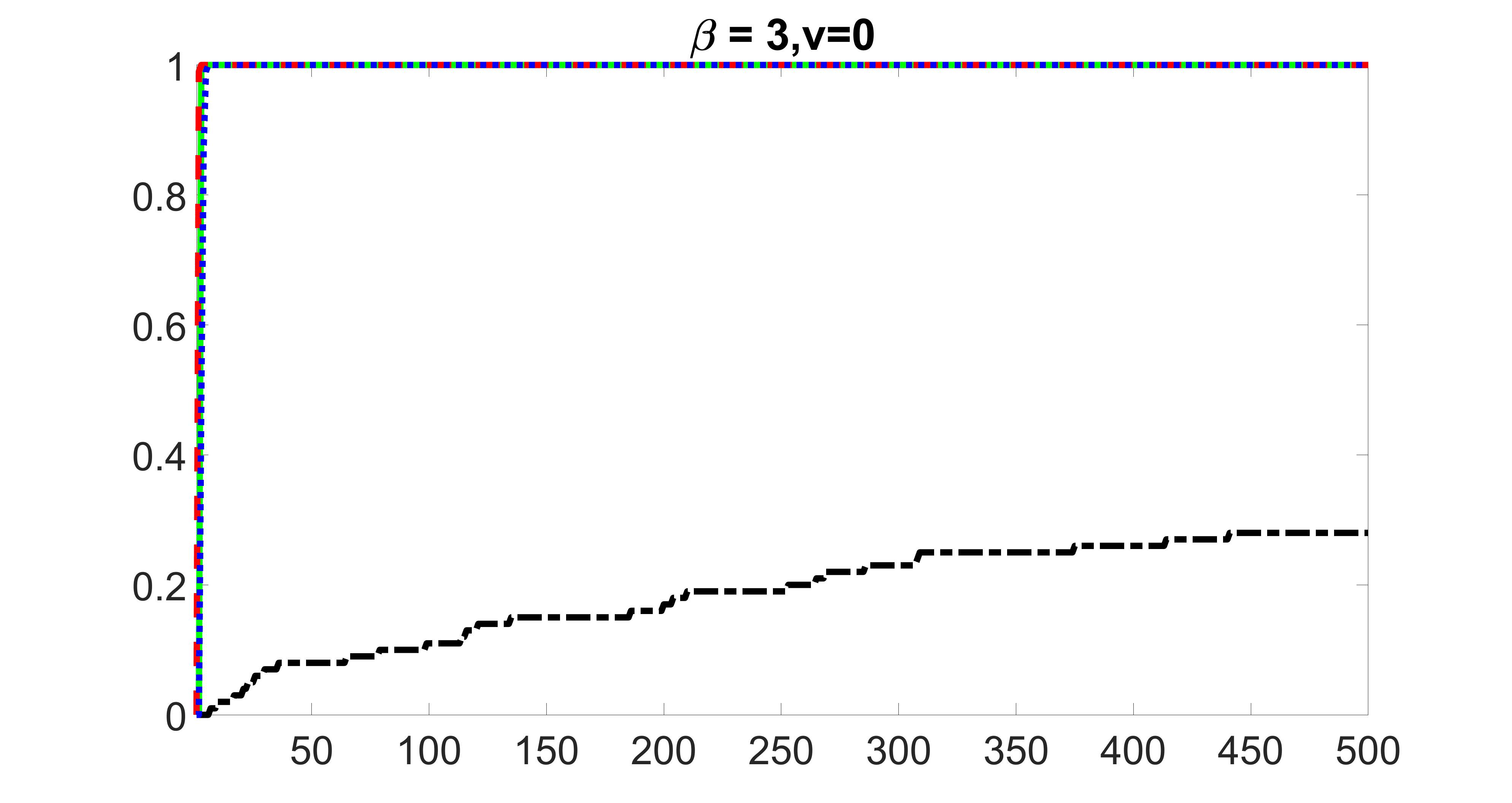}}
  \subcaptionbox{\footnotesize Precision: medium \\ outcome, zero exposure}[0.45\linewidth]
 {\includegraphics[width=6cm,height=3.5cm]{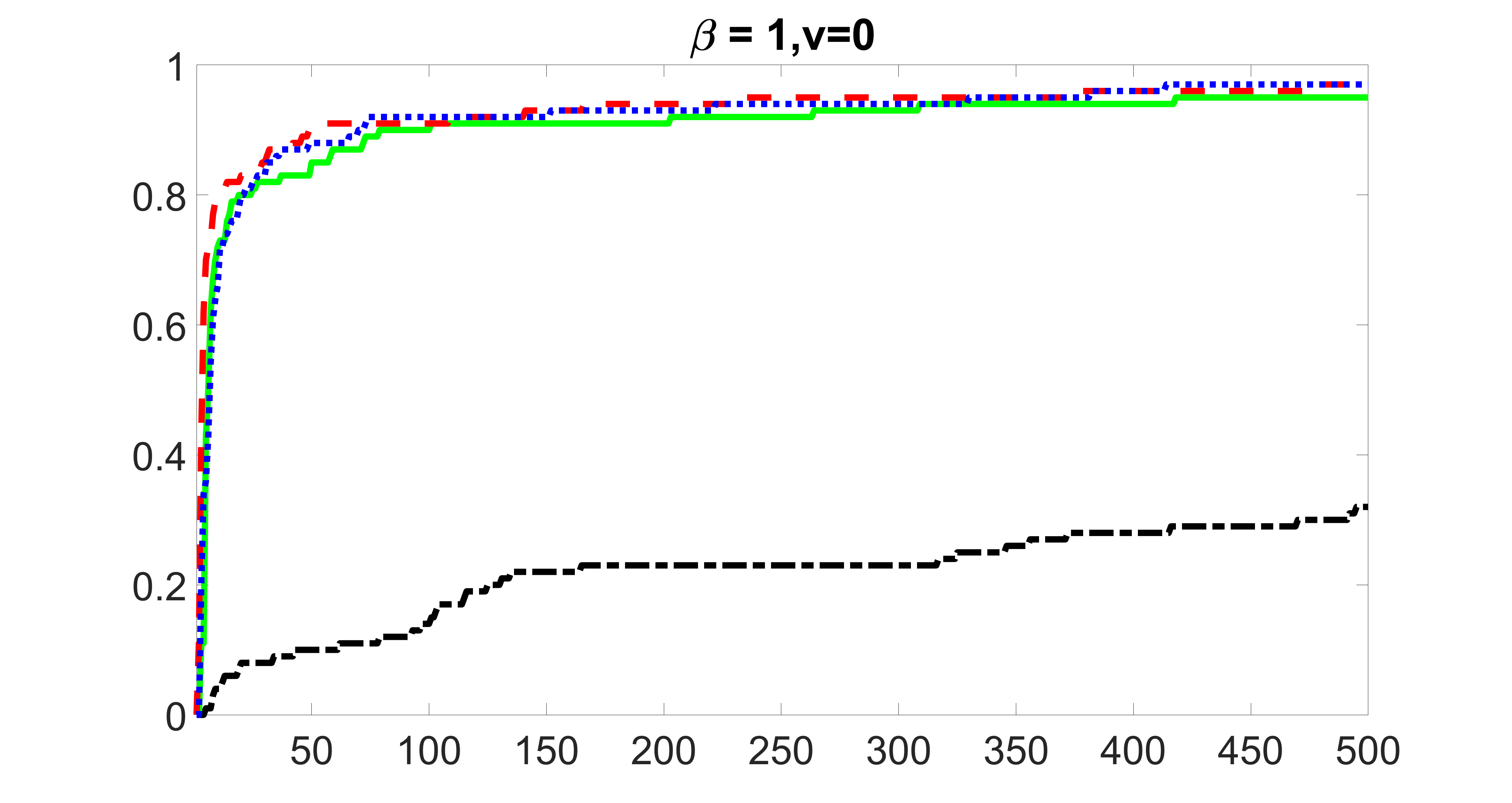}}
  \subcaptionbox{\footnotesize Precision: weak \\ outcome, zero exposure}[0.45\linewidth]
 {\includegraphics[width=6cm,height=3.5cm]{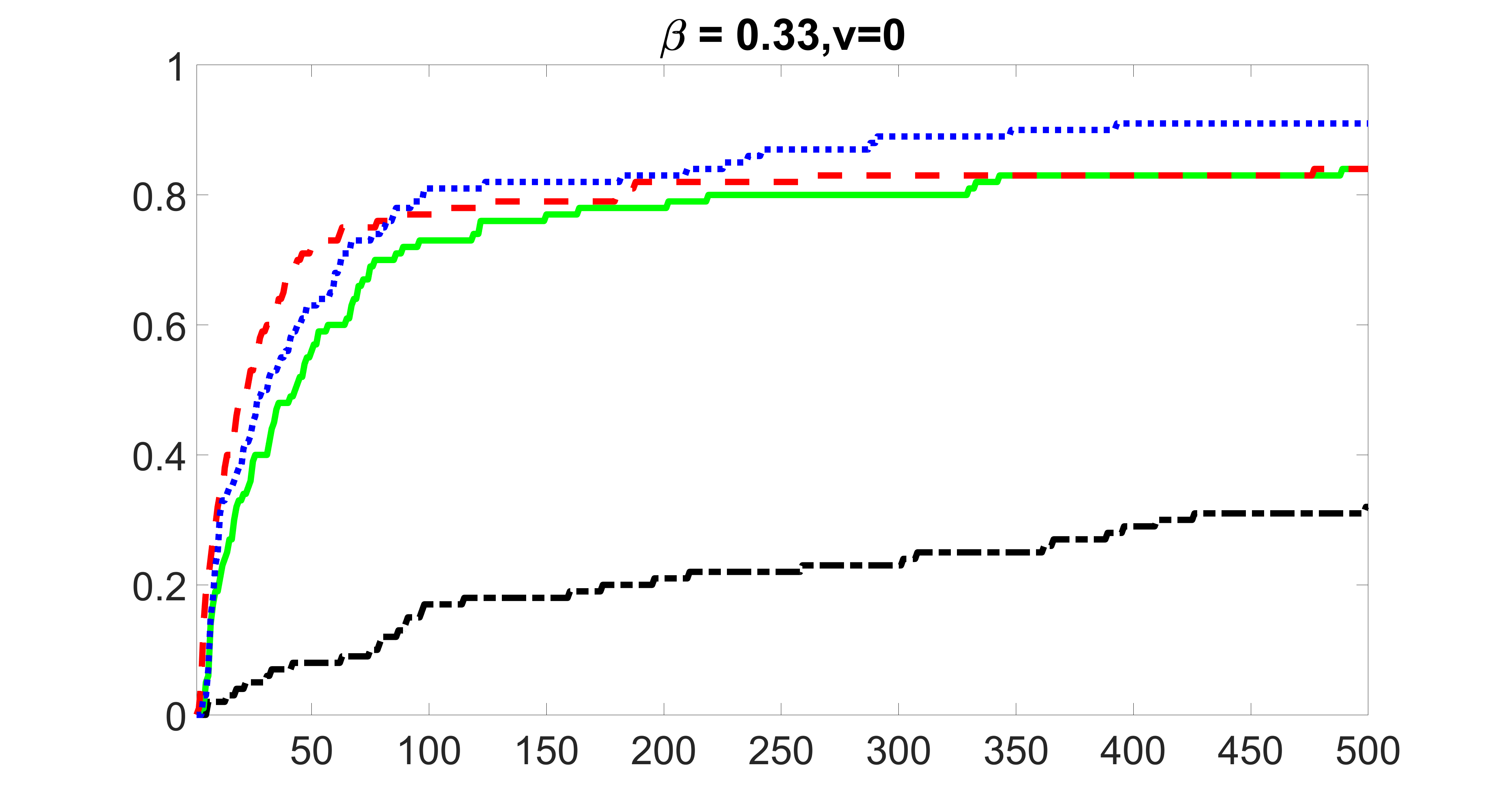}}
 \subcaptionbox{\footnotesize Precision: weaker \\ outcome, zero exposure}[0.45\linewidth]
 {\includegraphics[width=6cm,height=3.5cm]{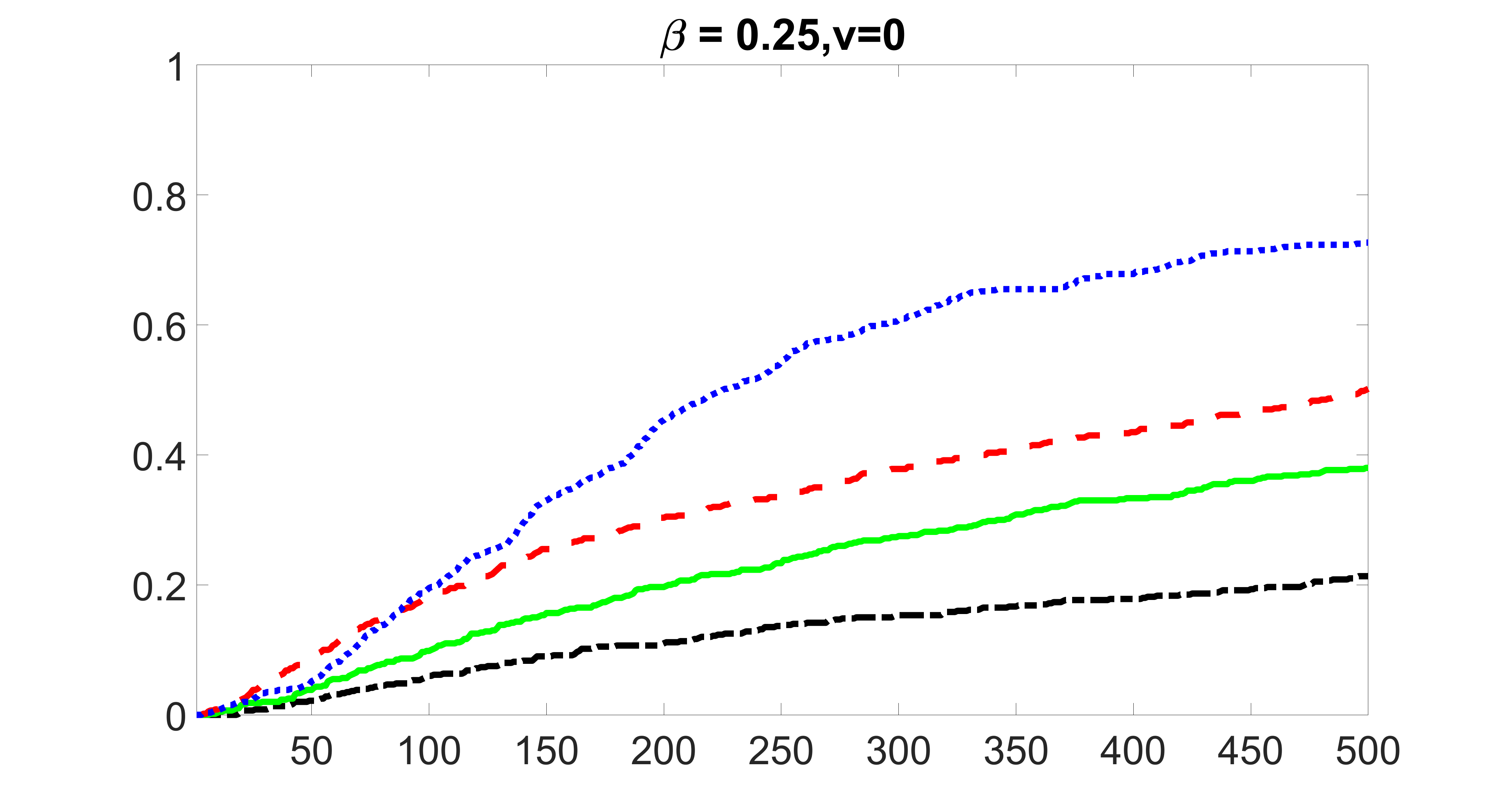}}
  \subcaptionbox{Overall coverage of $\mathcal{M}_1$}[0.45\linewidth]
 {\includegraphics[width=6cm,height=3.5cm]{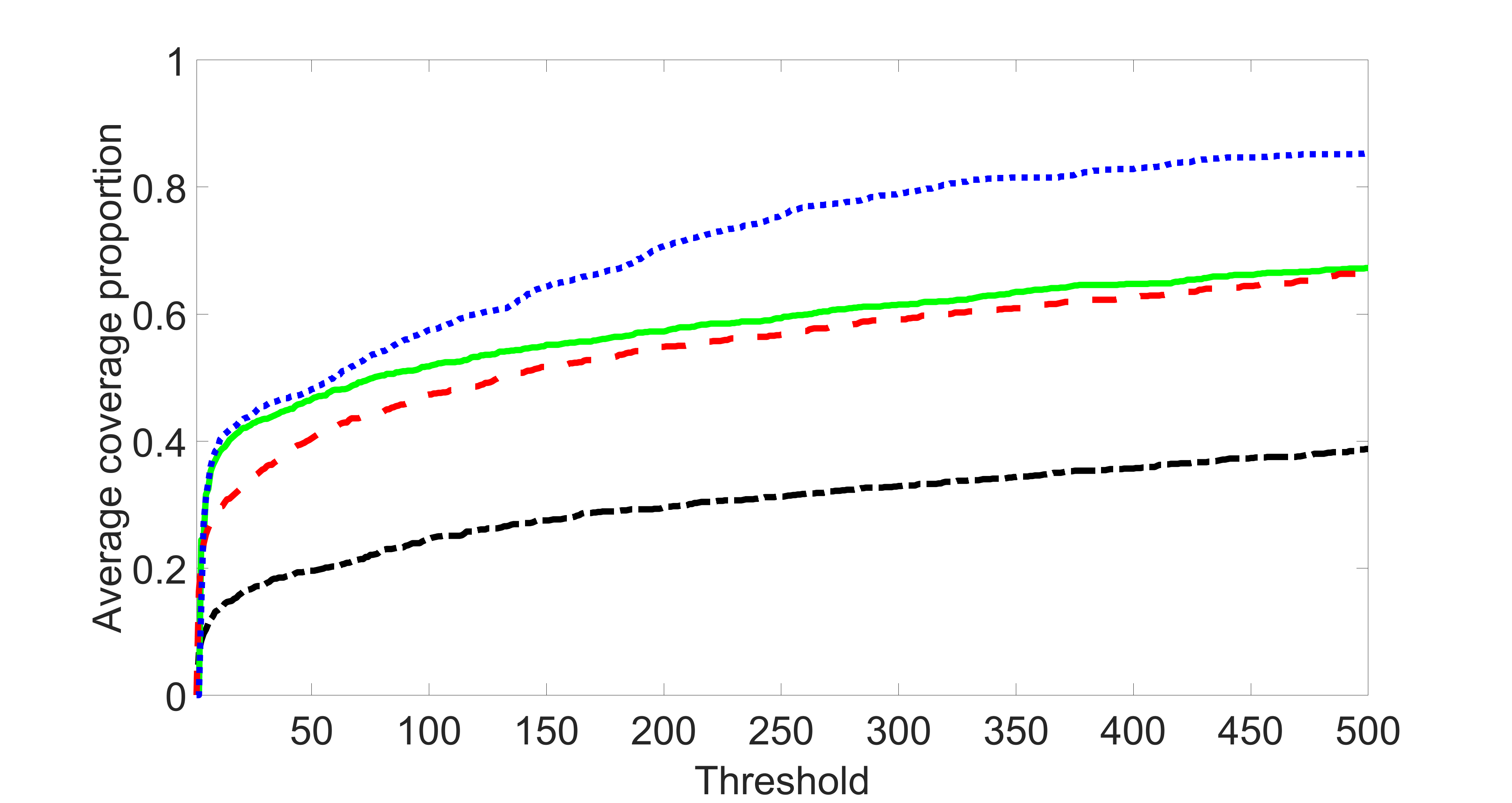}}
\caption{ Simulation results for the case $(n,s,K,\sigma) = (200,5000,6,1)$: Panels (a) -- (g) plot the average coverage proportion for $X_l$, where $l \in \mathcal{M}_1 =  \{1,2,3,104,105, 106\} \cup \mathcal{P}_{LD}$. Panels (a) -- (c) correspond to strong outcome and weak exposure predictor, moderate outcome and moderate exposure predictor and weak outcome and strong exposure predictor; Panels (d) -- (g) correspond to strong, moderate, and weak predictors of outcome only. Panel (g) plots the average coverage proportion for the index set $\mathcal{P}_{LD}$. Panel (h) plots the average coverage proportion for the index set $\mathcal{M}_1$. The x-axis represents the size of $\widehat{\mathcal{M}} $, while
y-axis denotes the average proportion. The blue dot, green solid, red dashed and black dash dotted lines denote the blockwise joint screening, joint screening, outcome screening, and intersection screening methods, respectively.}
\label{sim3step1n200sizesig6sigma1}
\end{figure}

\begin{figure}[htbp]
\captionsetup[subfigure]{justification=centering}
\centering
 \subcaptionbox{\footnotesize Confounder: strong \\ outcome, weak exposure}[0.45\linewidth]
 {\includegraphics[width=6cm,height=3.5cm]{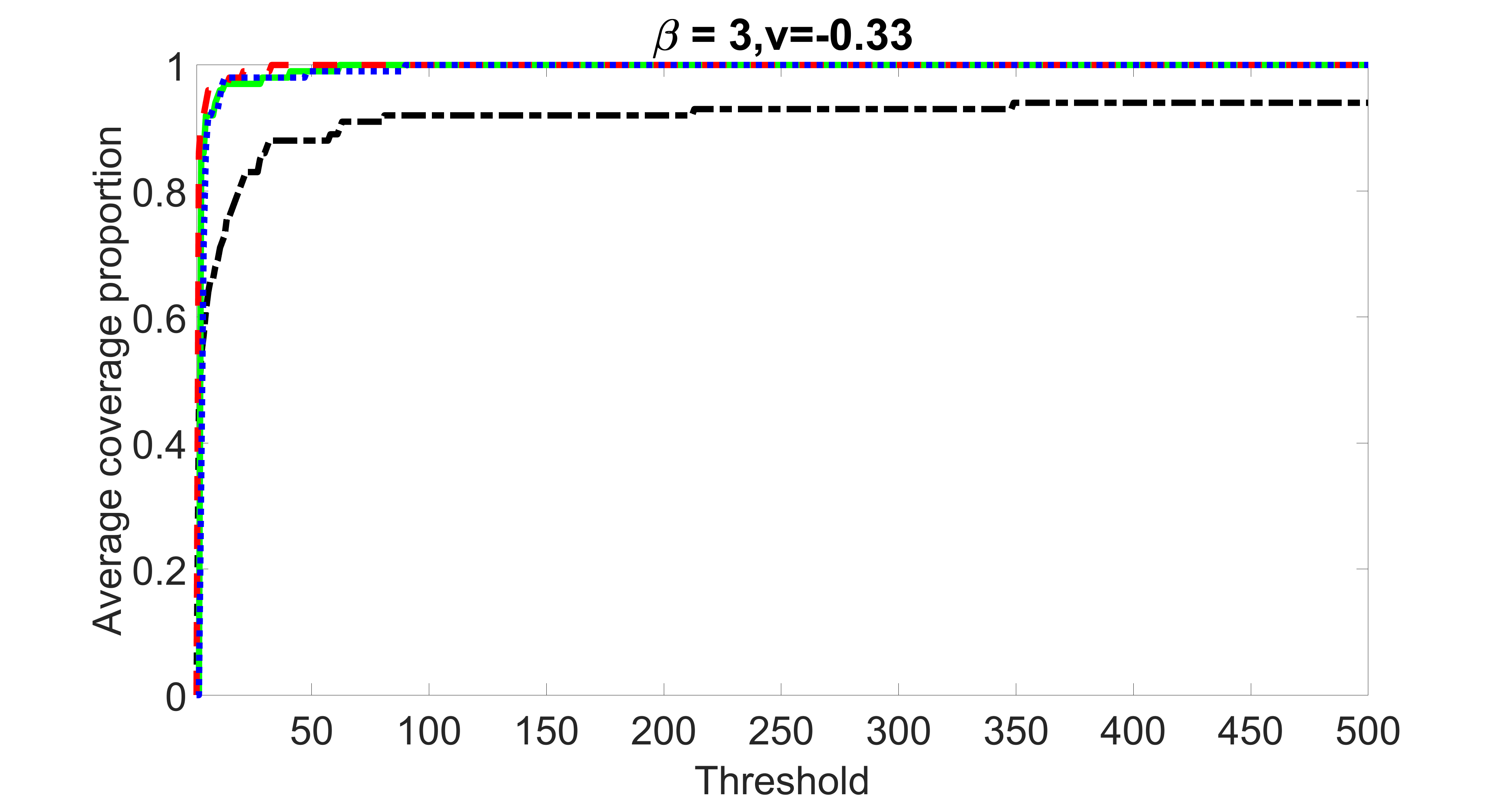}}
 \subcaptionbox{\footnotesize Confounder: medium \\ outcome, medium exposure}[0.45\linewidth]
 {\includegraphics[width=6cm,height=3.5cm]{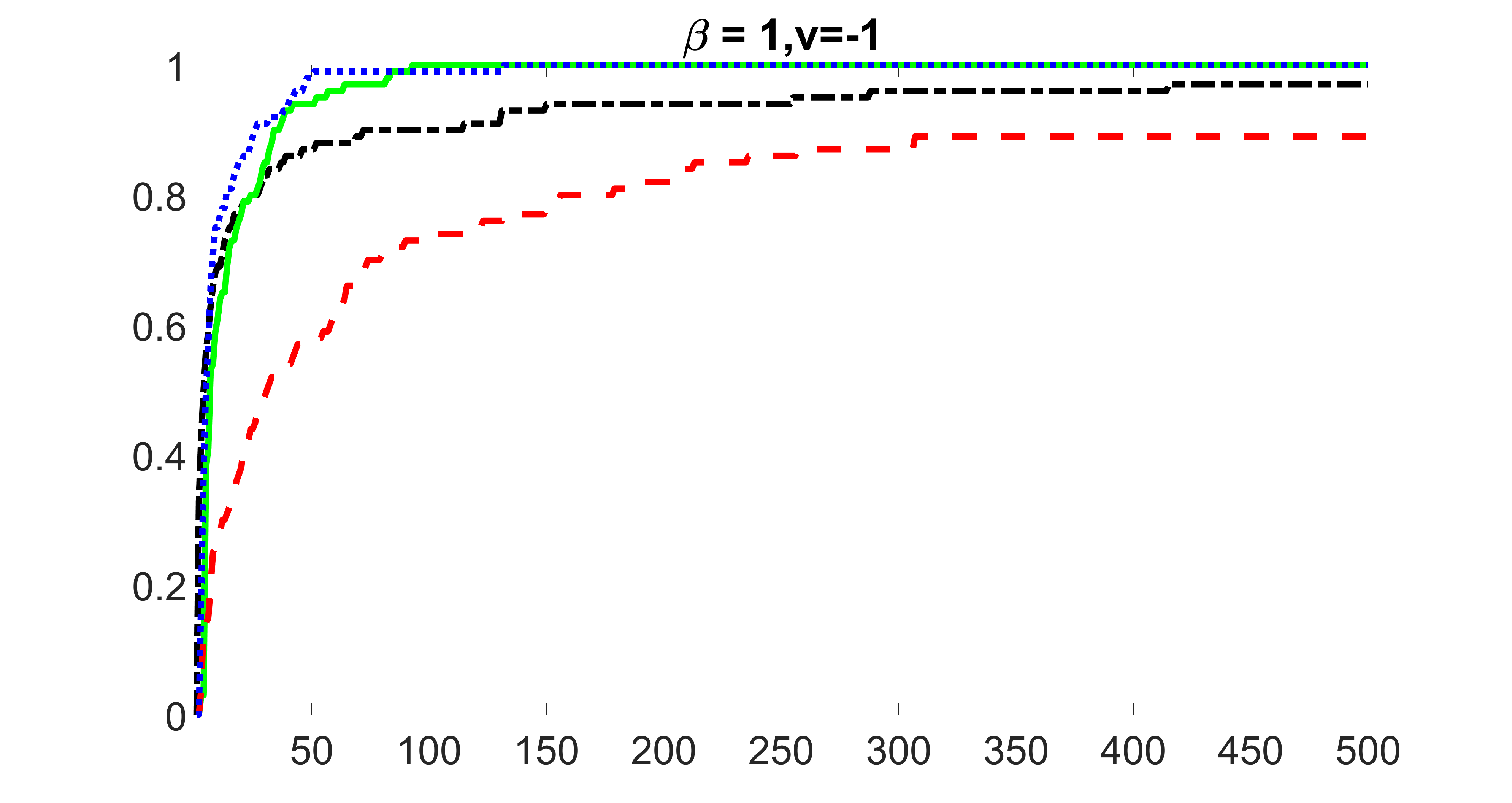}}
  \subcaptionbox{\footnotesize Confounder: weak \\ outcome, strong exposure}[0.45\linewidth]
 {\includegraphics[width=6cm,height=3.5cm]{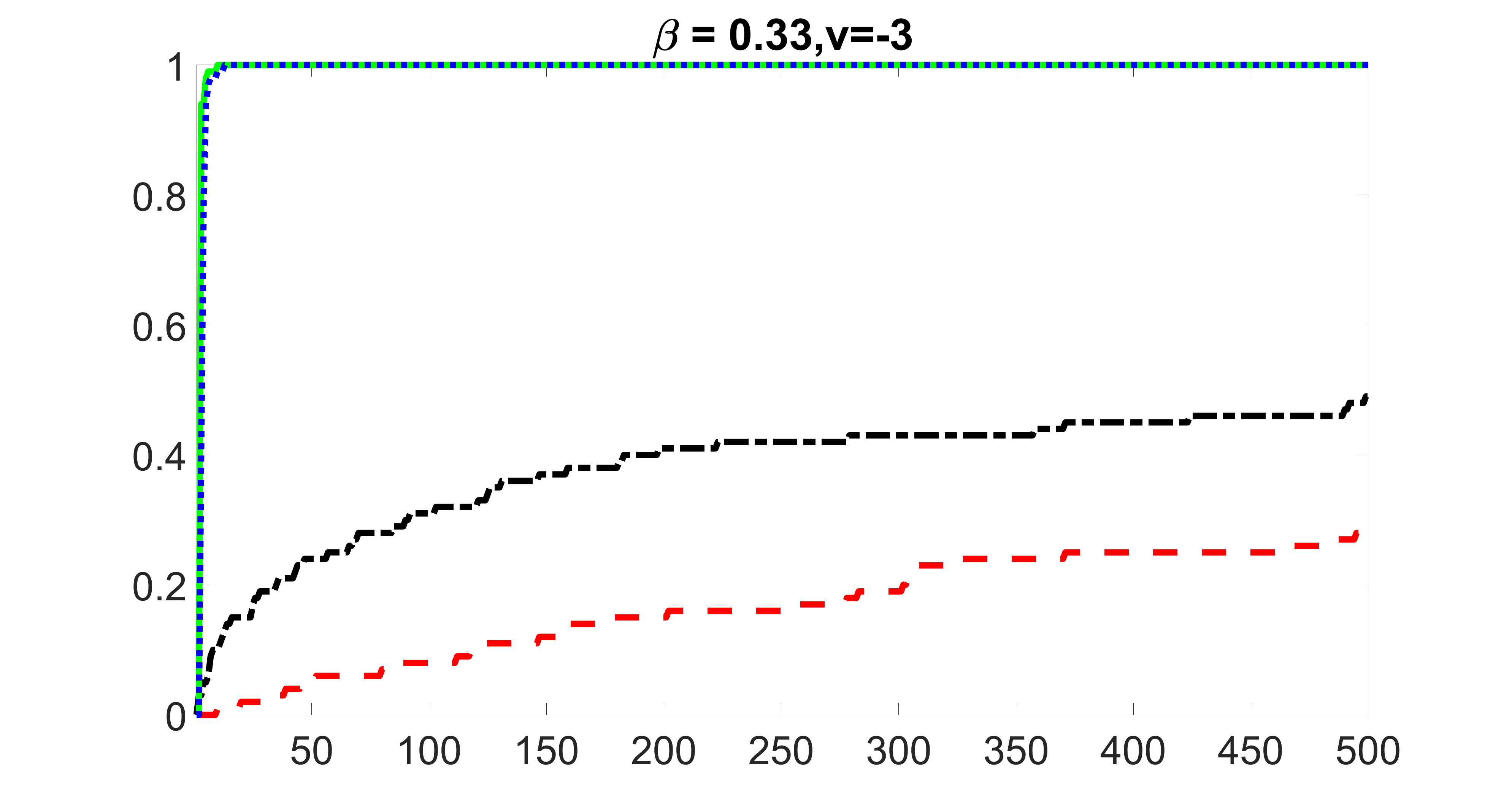}}
  \subcaptionbox{\footnotesize Precision: strong \\ outcome, zero exposure}[0.45\linewidth]
 {\includegraphics[width=6cm,height=3.5cm]{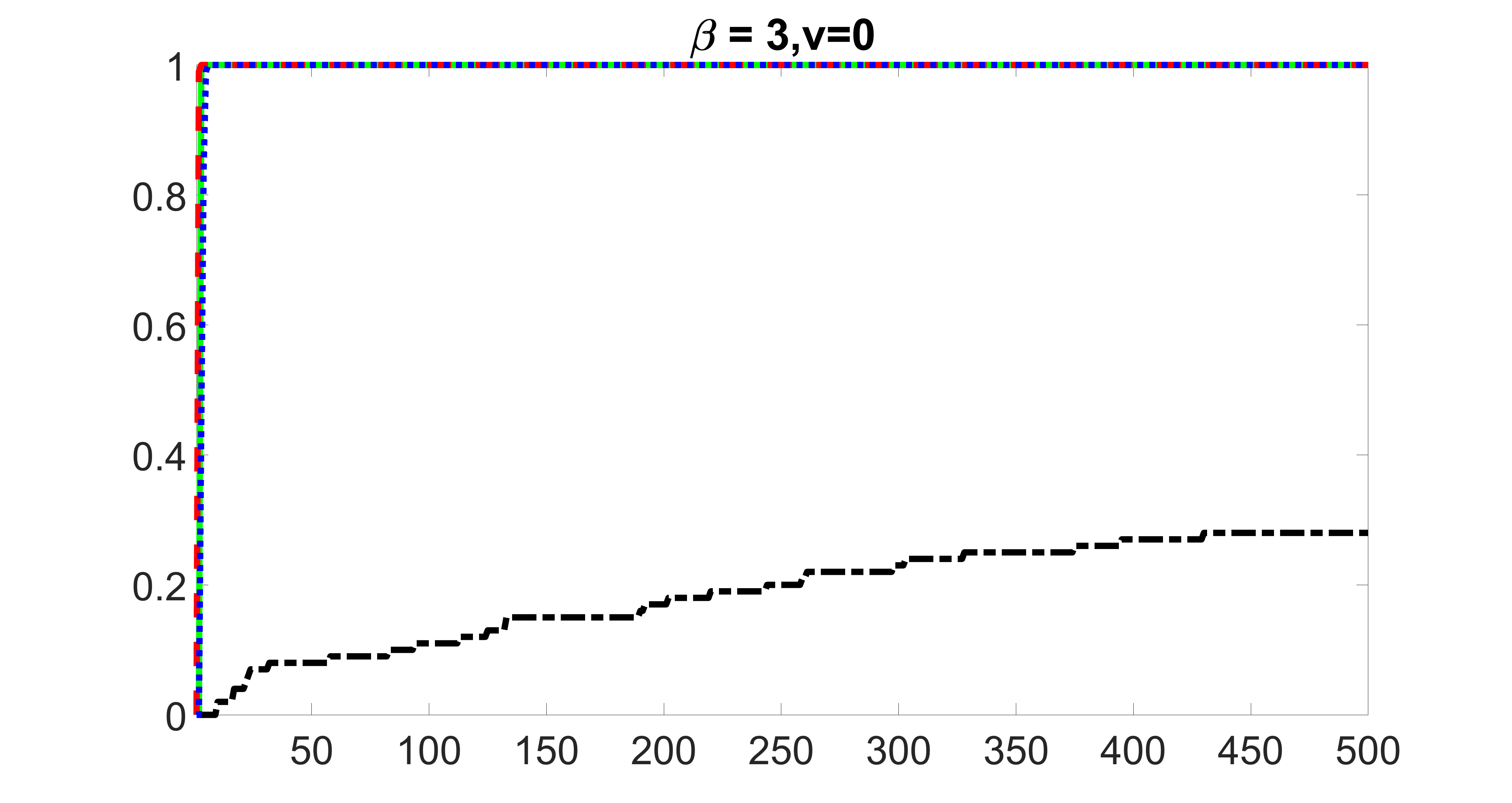}}
  \subcaptionbox{\footnotesize Precision: medium \\ outcome, zero exposure}[0.45\linewidth]
 {\includegraphics[width=6cm,height=3.5cm]{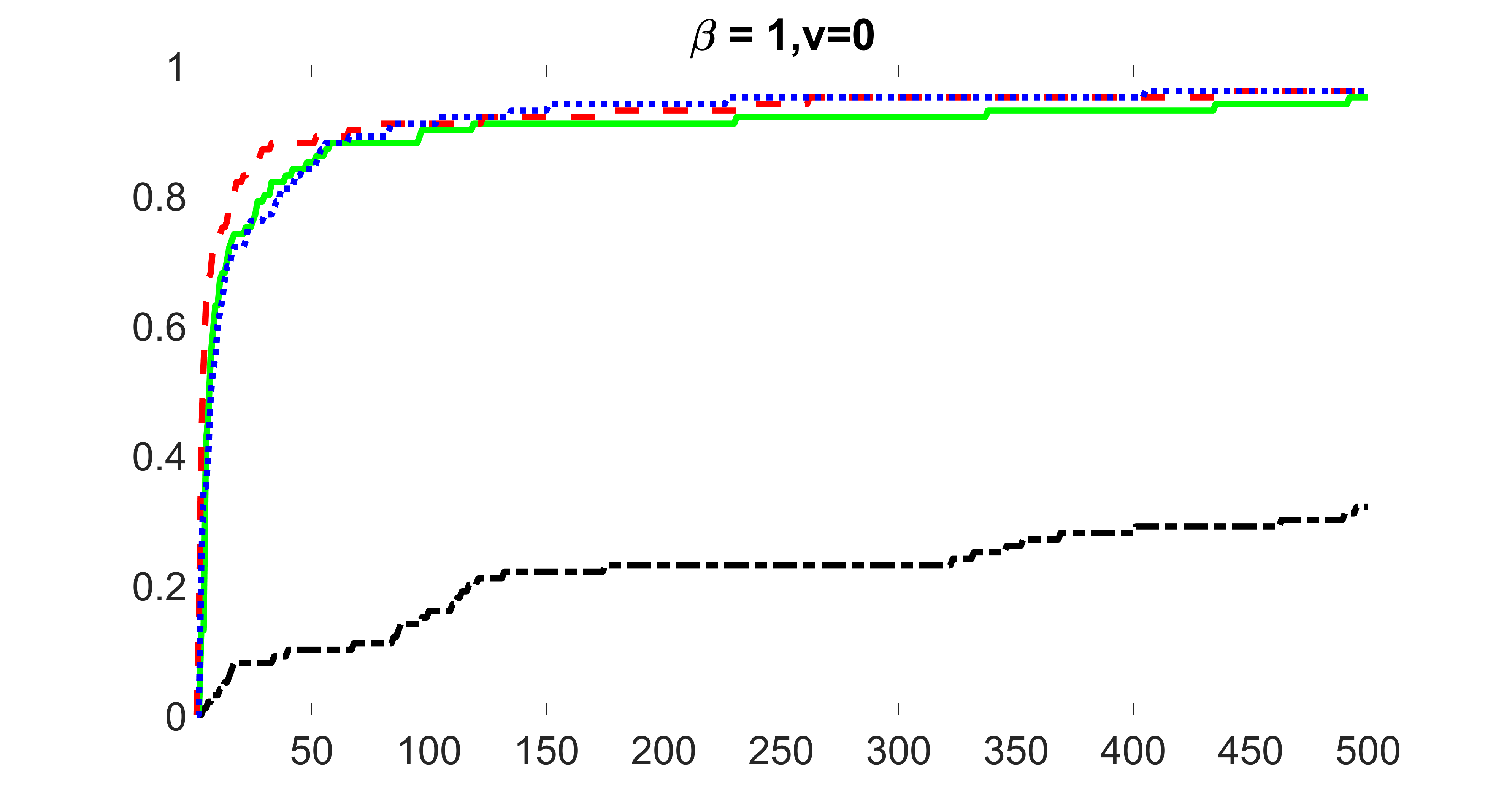}}
  \subcaptionbox{\footnotesize Precision: weak \\ outcome, zero exposure}[0.45\linewidth]
 {\includegraphics[width=6cm,height=3.5cm]{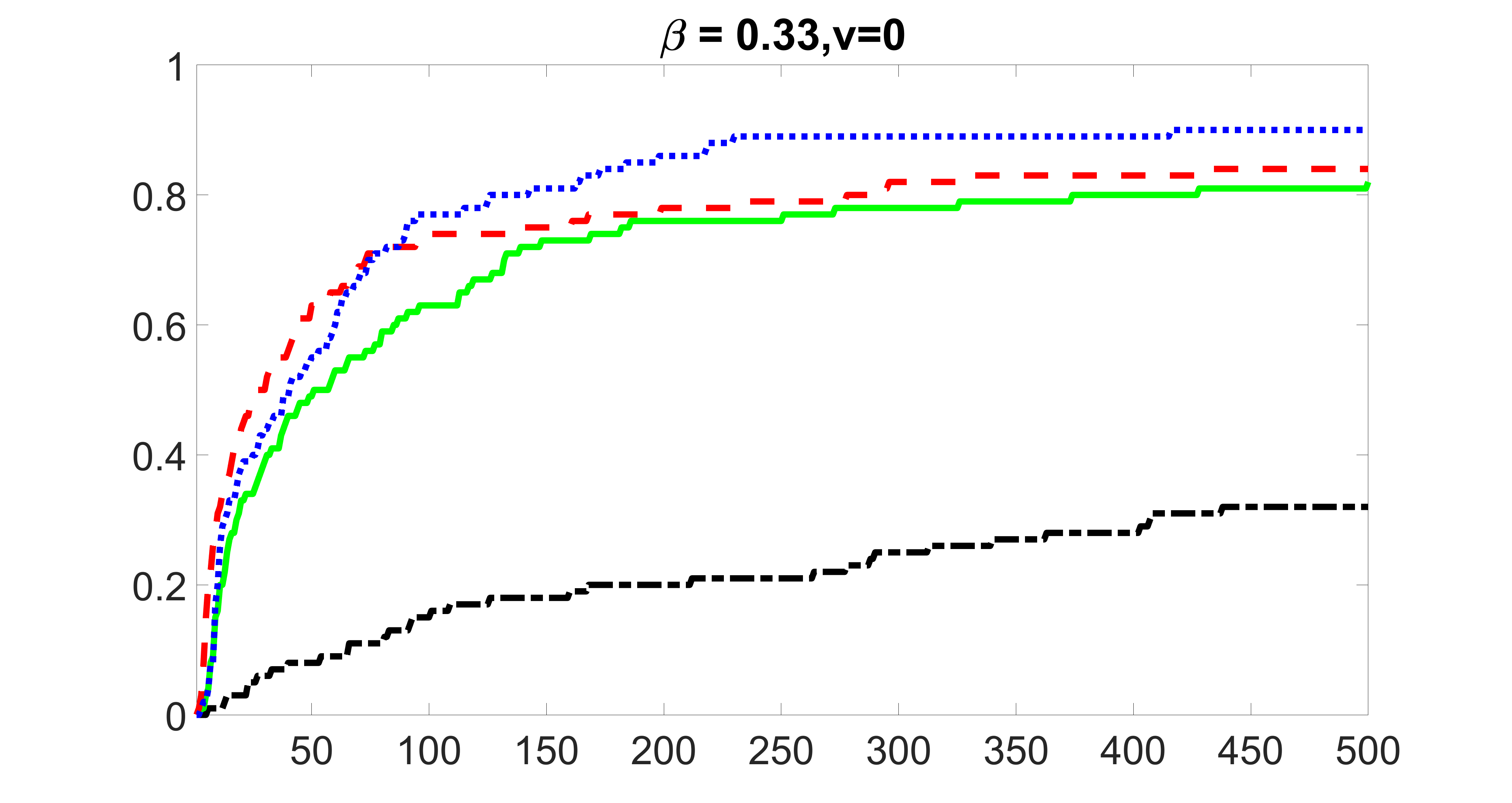}}
 \subcaptionbox{\footnotesize Precision: weaker \\ outcome, zero exposure}[0.45\linewidth]
 {\includegraphics[width=6cm,height=3.5cm]{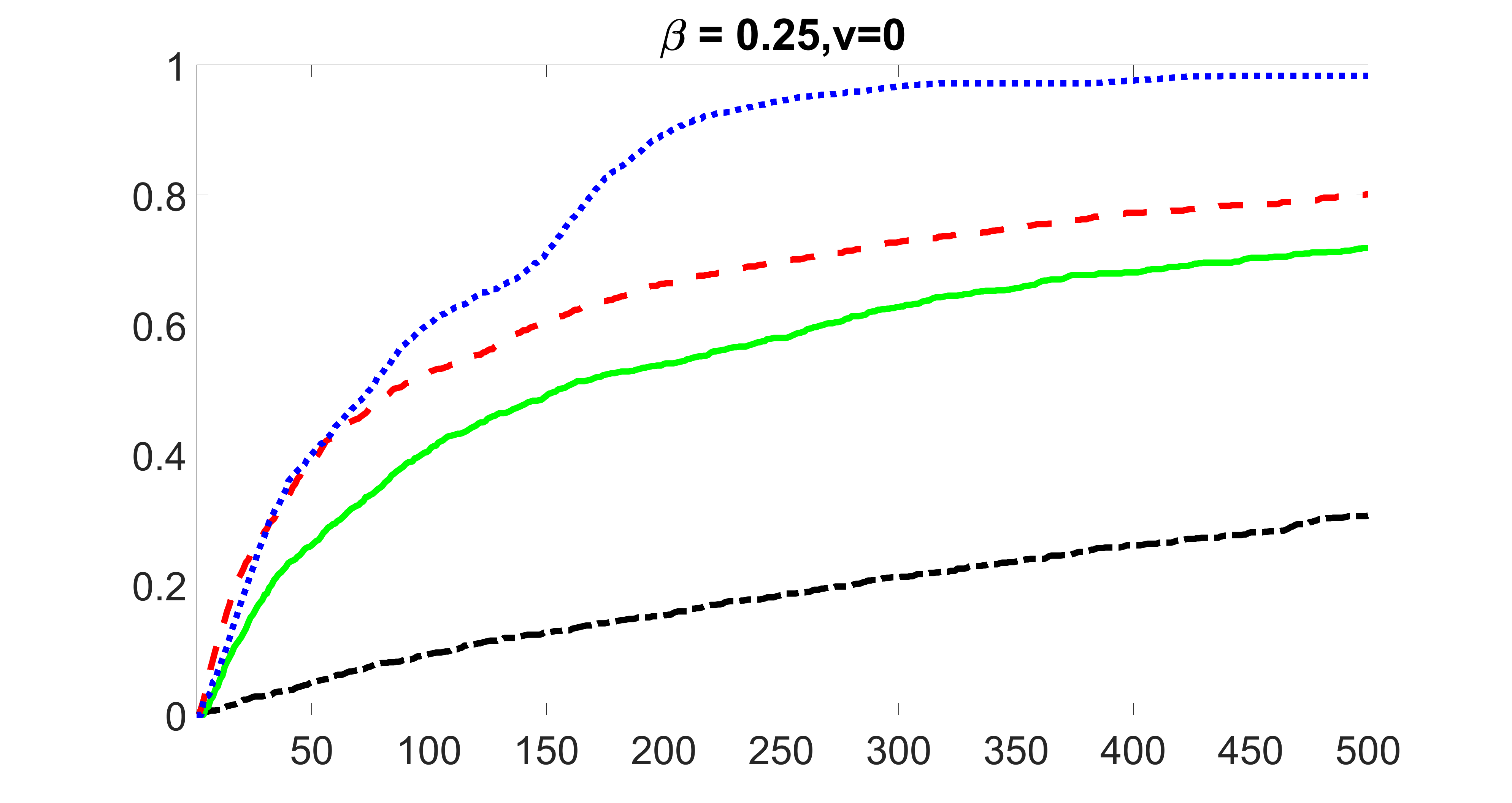}}
  \subcaptionbox{Overall coverage of $\mathcal{M}_1$}[0.45\linewidth]
 {\includegraphics[width=6cm,height=3.5cm]{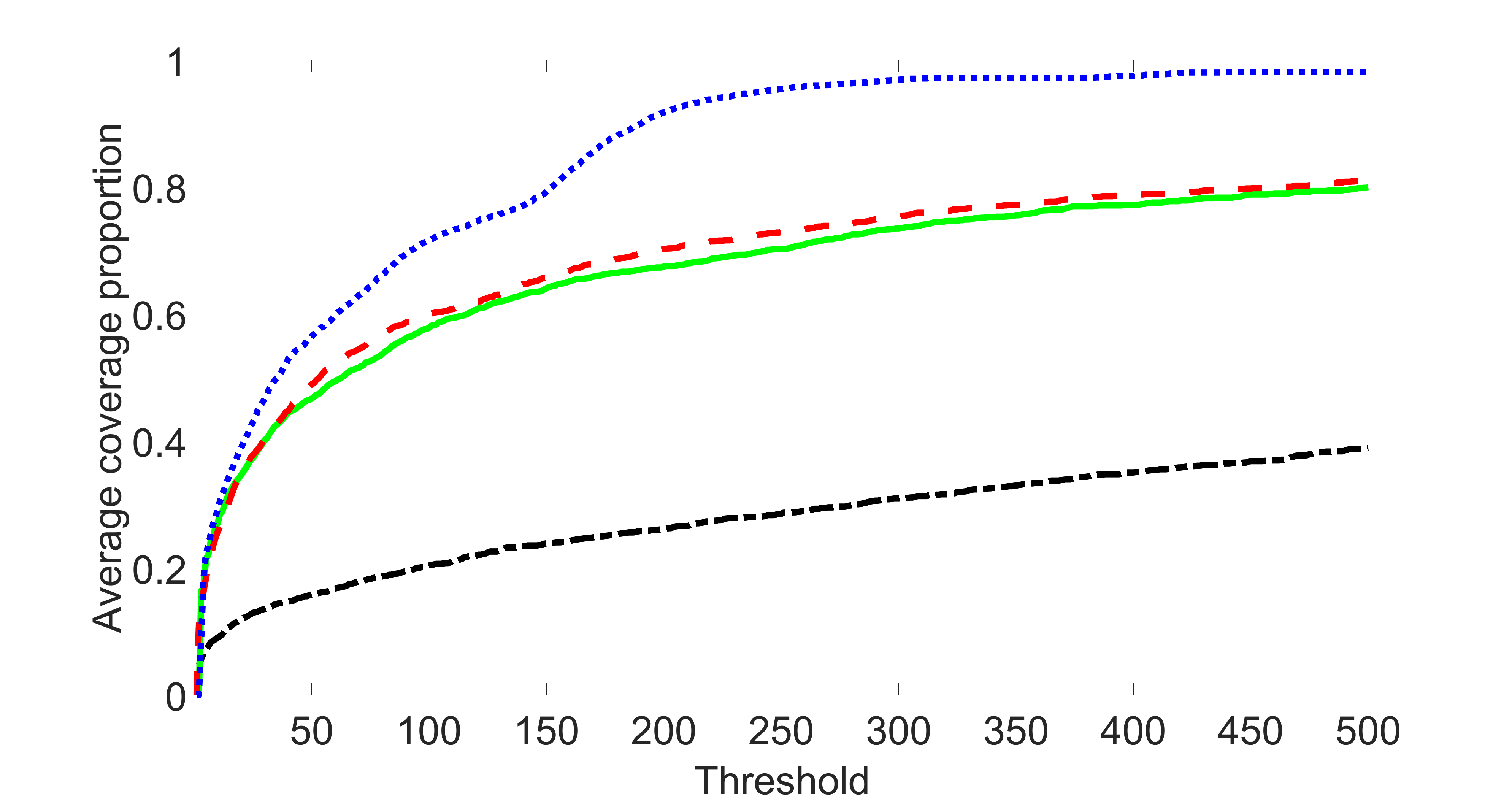}}
\caption{ Simulation results for the case $(n,s,K,\sigma) = (200,5000,12,1)$: Panels (a) -- (g) plot the average coverage proportion for $X_l$, where $l \in \mathcal{M}_1 =  \{1,2,3,104,105, 106\} \cup \mathcal{P}_{LD}$. Panels (a) -- (c) correspond to strong outcome and weak exposure predictor, moderate outcome and moderate exposure predictor and weak outcome and strong exposure predictor; Panels (d) -- (g) correspond to strong, moderate, and weak predictors of outcome only. Panel (g) plots the average coverage proportion for the index set $\mathcal{P}_{LD}$. Panel (h) plots the average coverage proportion for the index set $\mathcal{M}_1$. The x-axis represents the size of $\widehat{\mathcal{M}} $, while
y-axis denotes the average proportion. The blue dot, green solid, red dashed and black dash dotted lines denote the blockwise joint screening, joint screening, outcome screening, and intersection screening methods, respectively.}
\label{sim3step1n200sizesig12sigma1}
\end{figure}

\begin{figure}[htbp]
\captionsetup[subfigure]{justification=centering}
\centering
 \subcaptionbox{\footnotesize Confounder: strong \\ outcome, weak exposure}[0.45\linewidth]
 {\includegraphics[width=6cm,height=3.5cm]{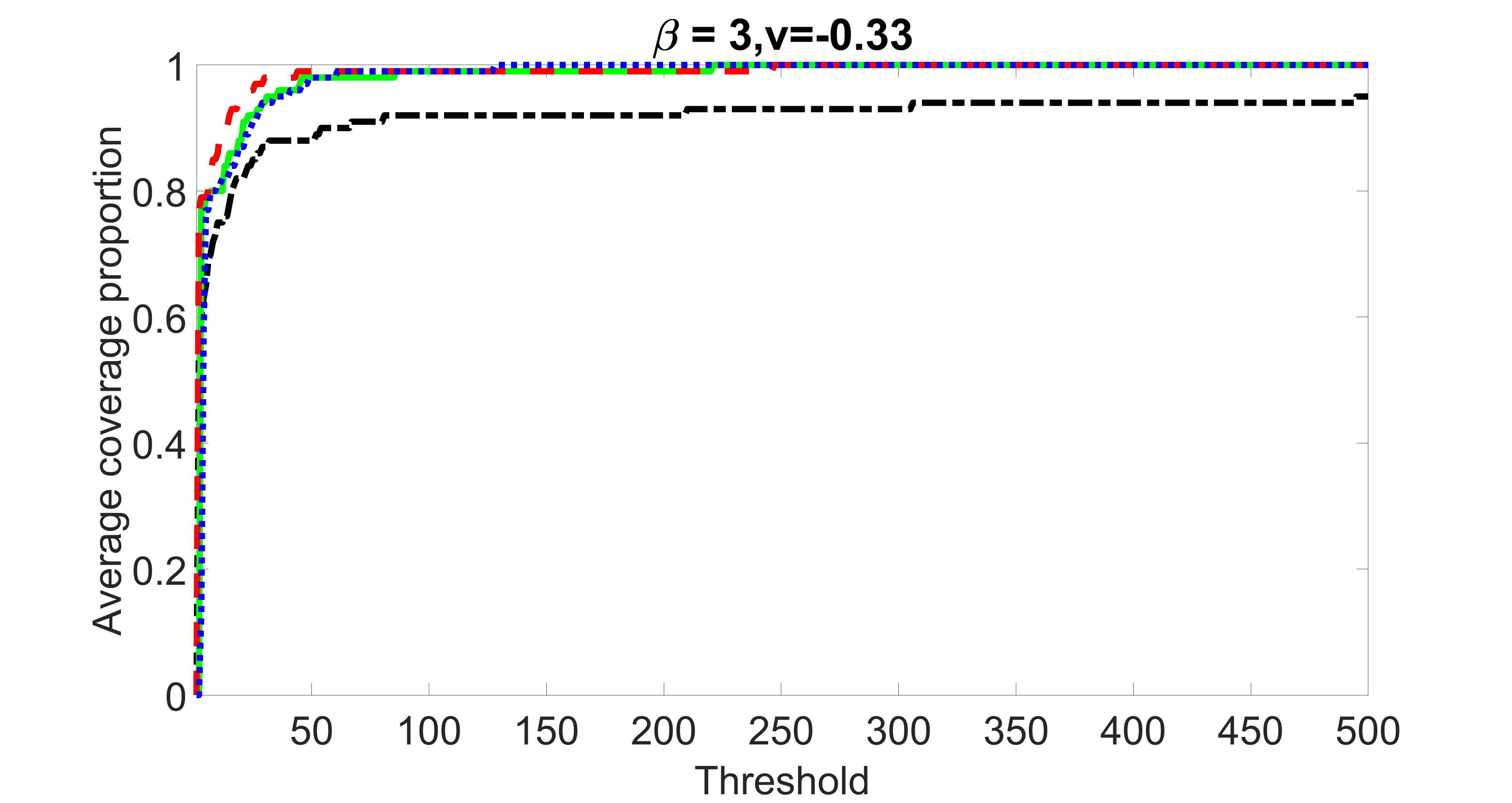}}
 \subcaptionbox{\footnotesize Confounder: medium \\ outcome, medium exposure}[0.45\linewidth]
 {\includegraphics[width=6cm,height=3.5cm]{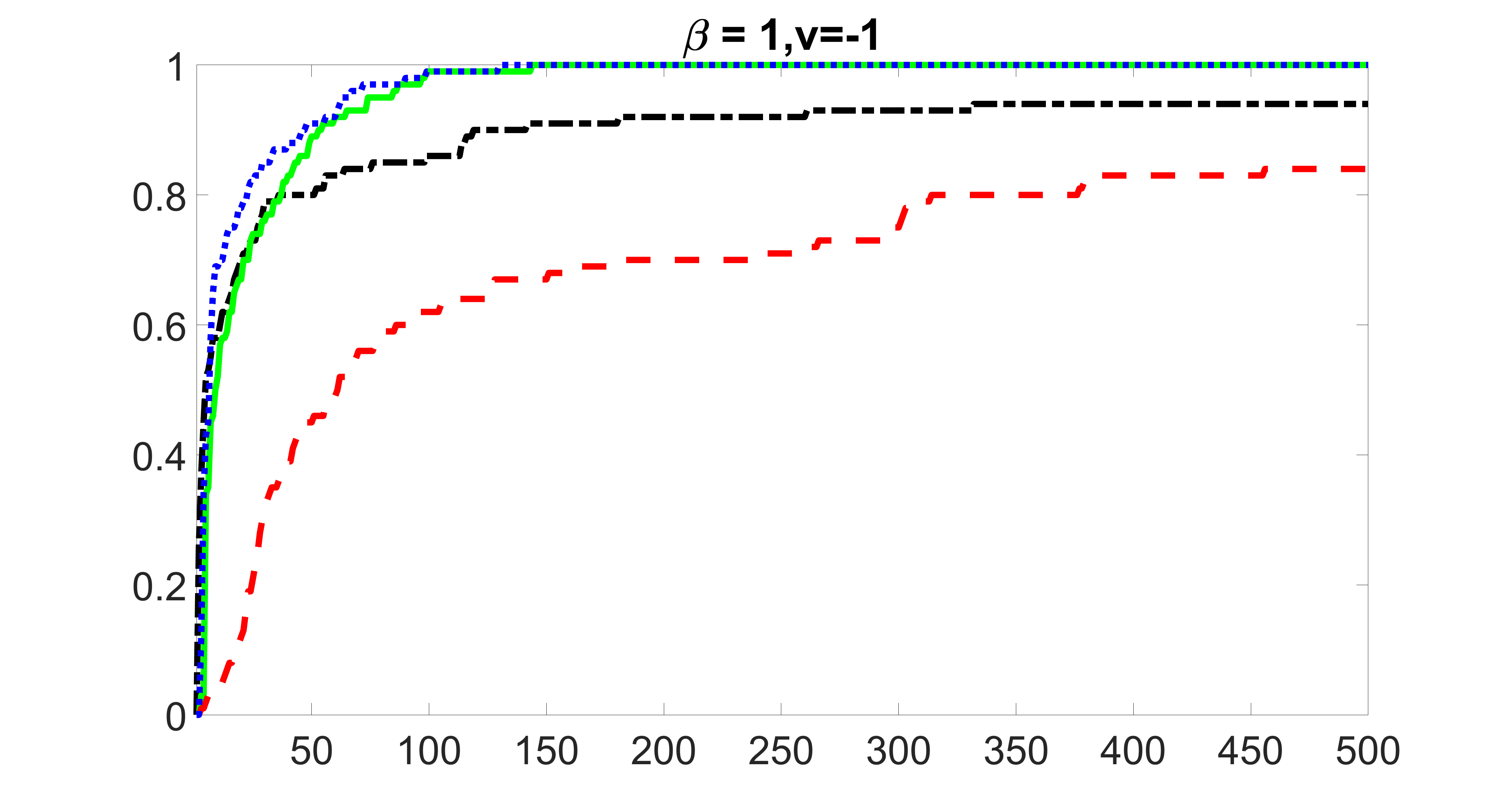}}
  \subcaptionbox{\footnotesize Confounder: weak \\ outcome, strong exposure}[0.45\linewidth]
 {\includegraphics[width=6cm,height=3.5cm]{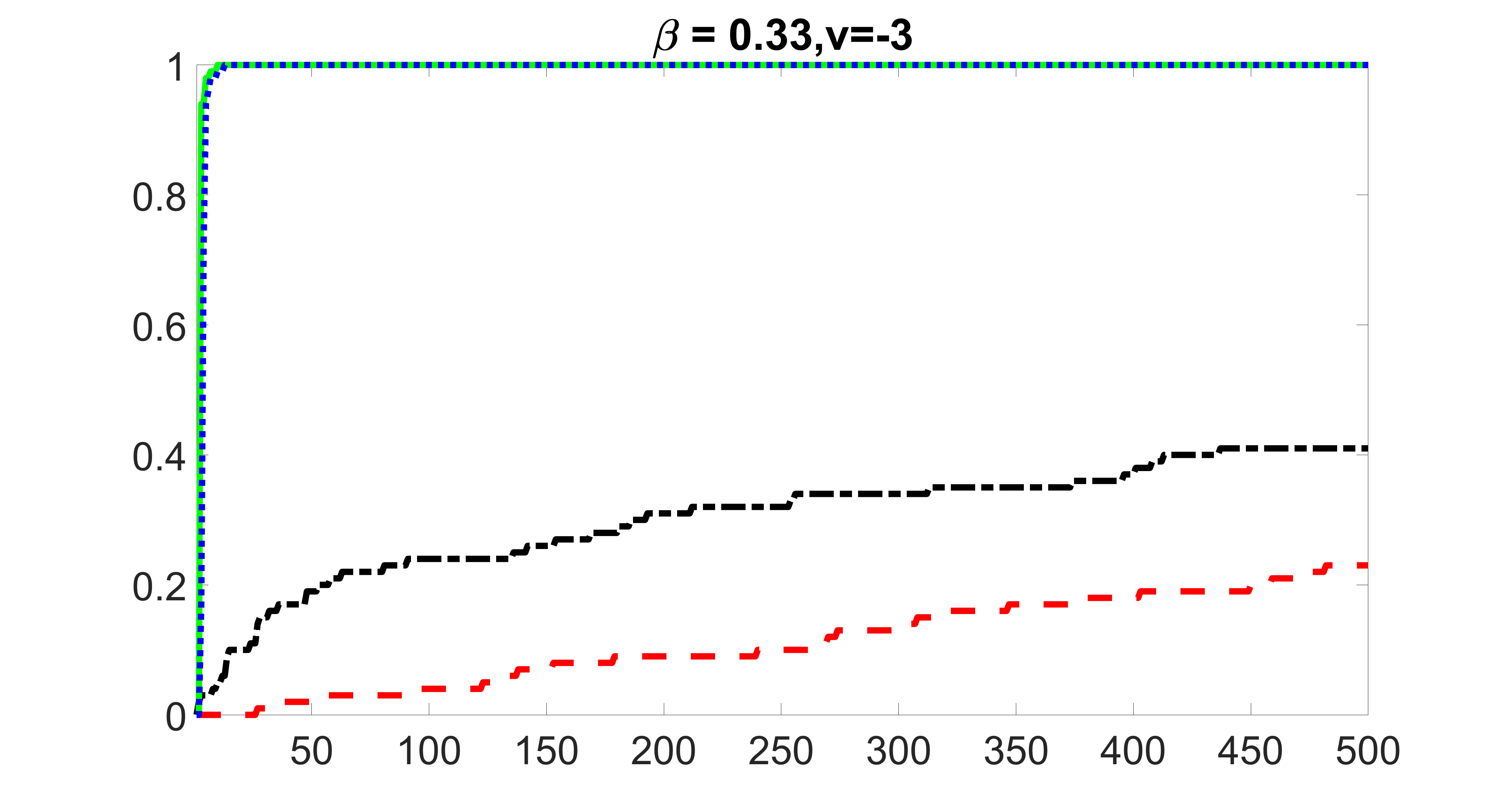}}
  \subcaptionbox{\footnotesize Precision: strong \\ outcome, zero exposure}[0.45\linewidth]
 {\includegraphics[width=6cm,height=3.5cm]{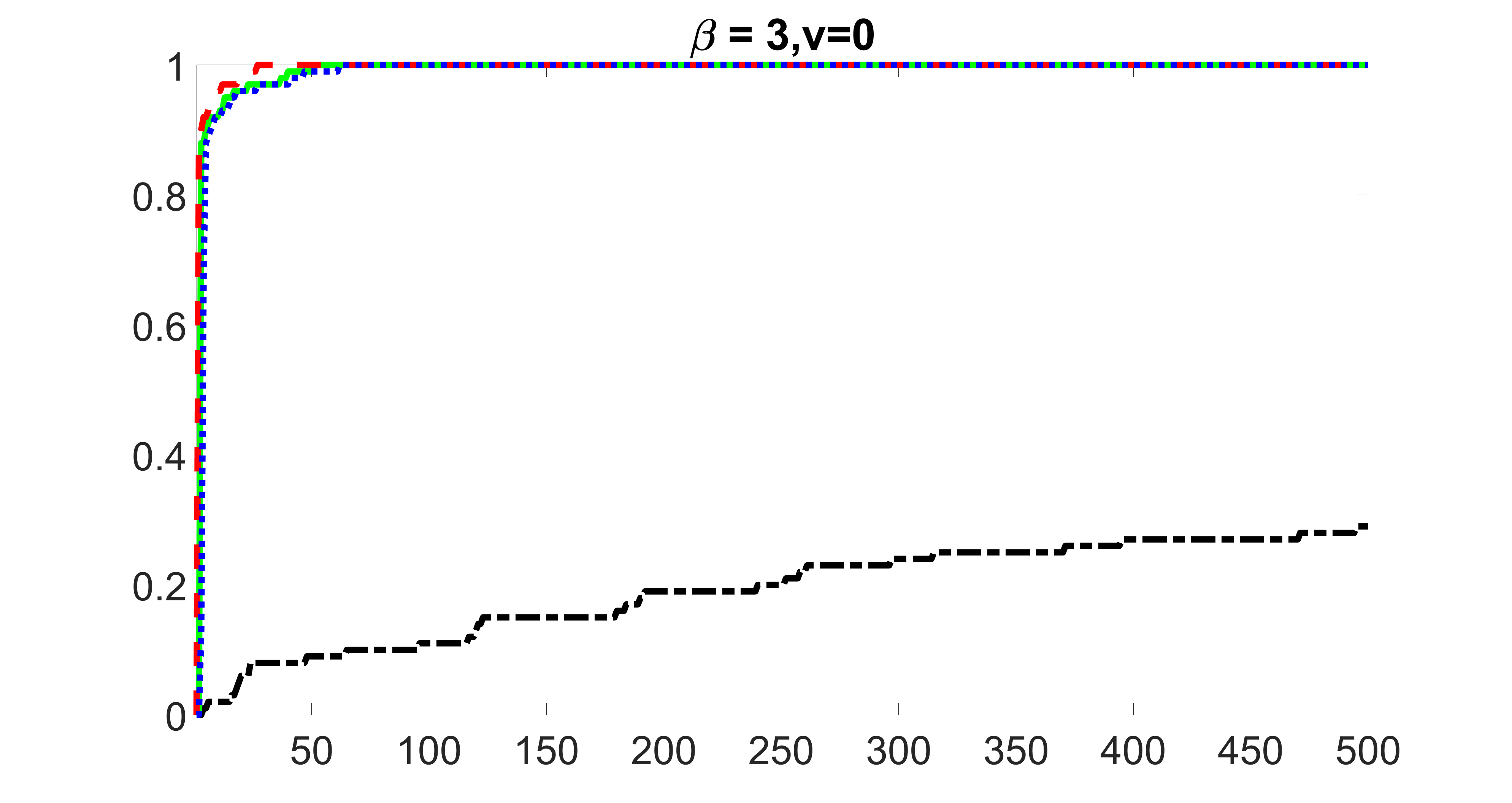}}
  \subcaptionbox{\footnotesize Precision: medium \\ outcome, zero exposure}[0.45\linewidth]
 {\includegraphics[width=6cm,height=3.5cm]{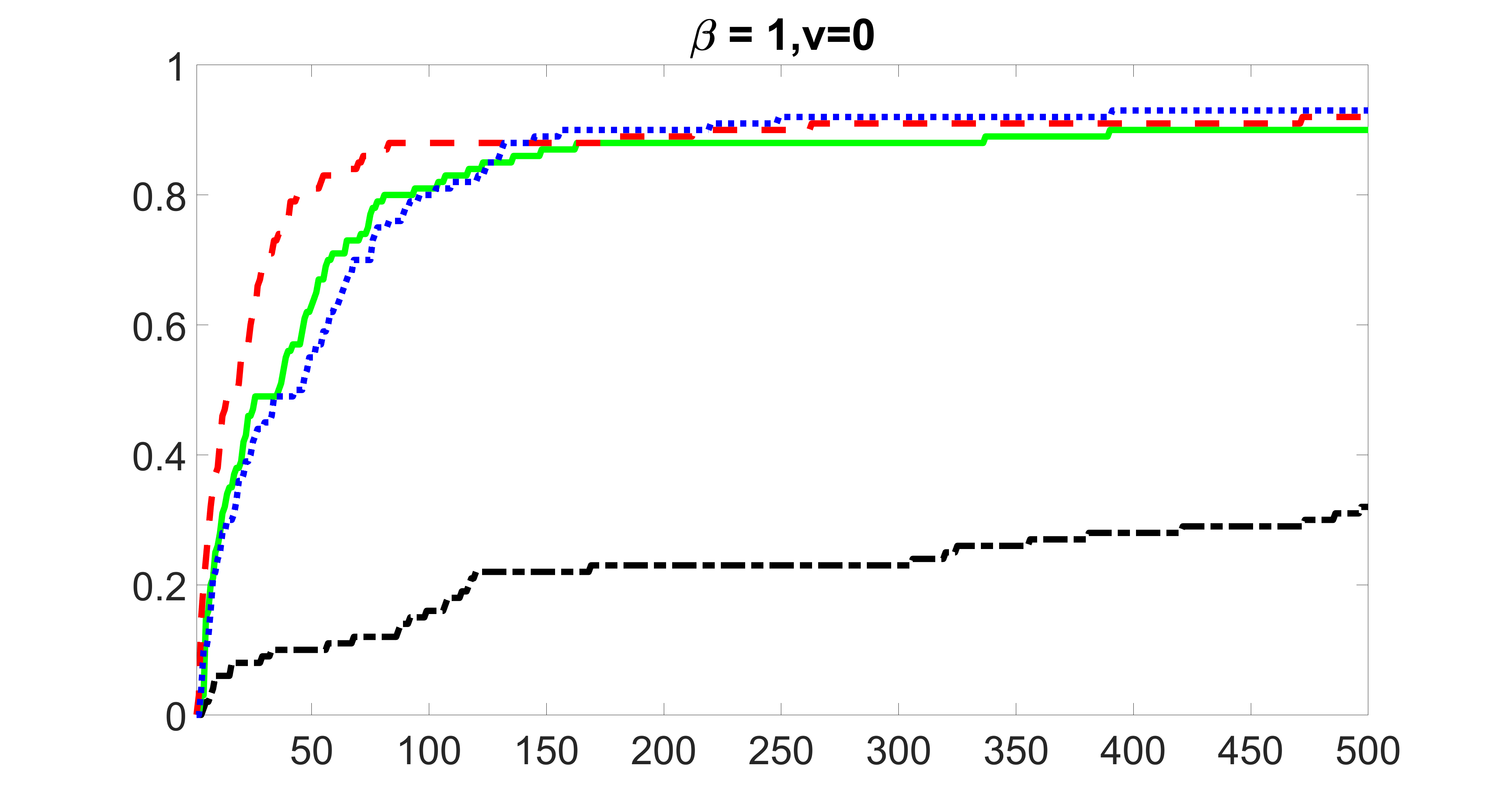}}
  \subcaptionbox{\footnotesize Precision: weak \\ outcome, zero exposure}[0.45\linewidth]
 {\includegraphics[width=6cm,height=3.5cm]{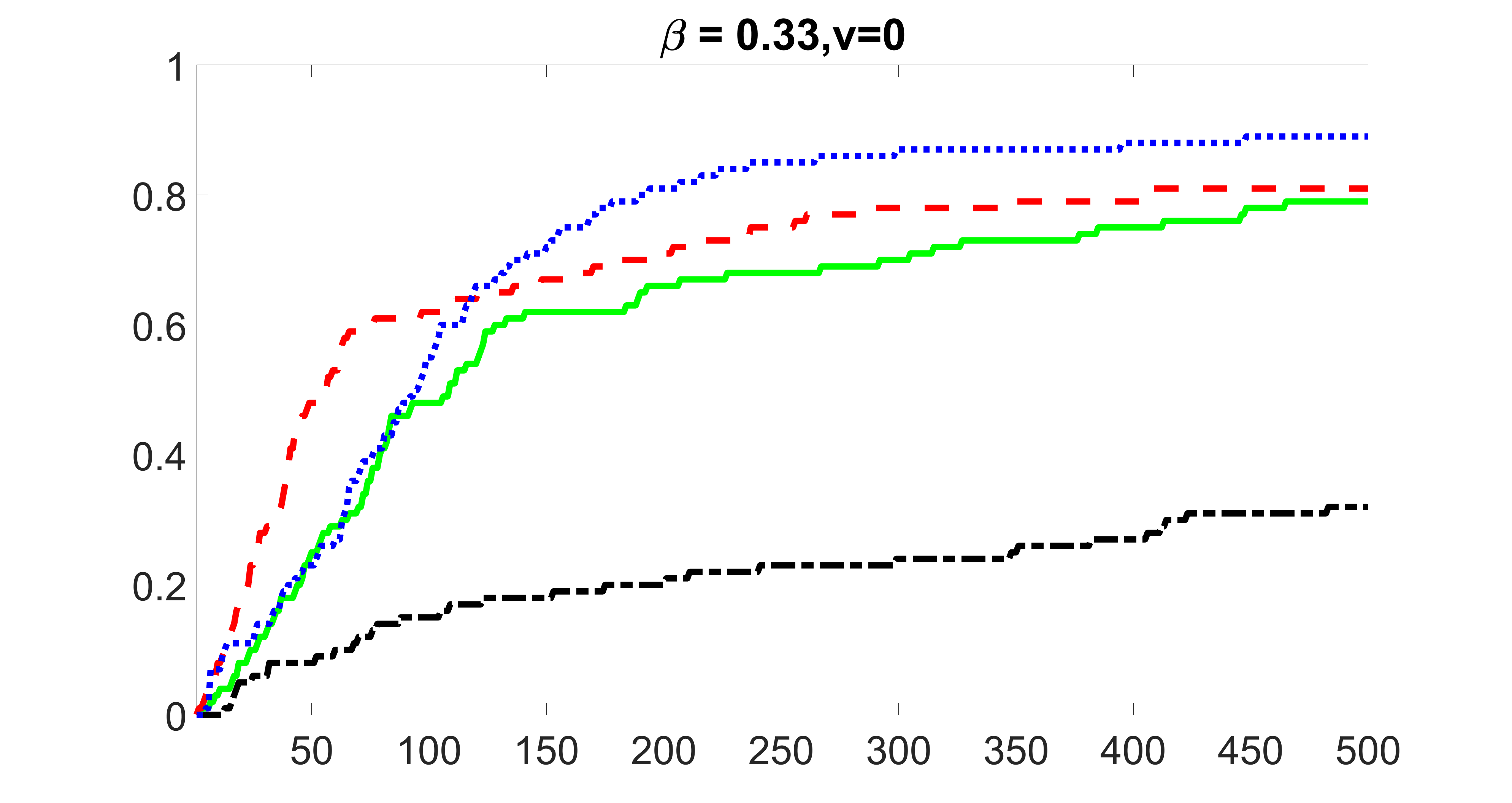}}
 \subcaptionbox{\footnotesize Precision: weaker \\ outcome, zero exposure}[0.45\linewidth]
 {\includegraphics[width=6cm,height=3.5cm]{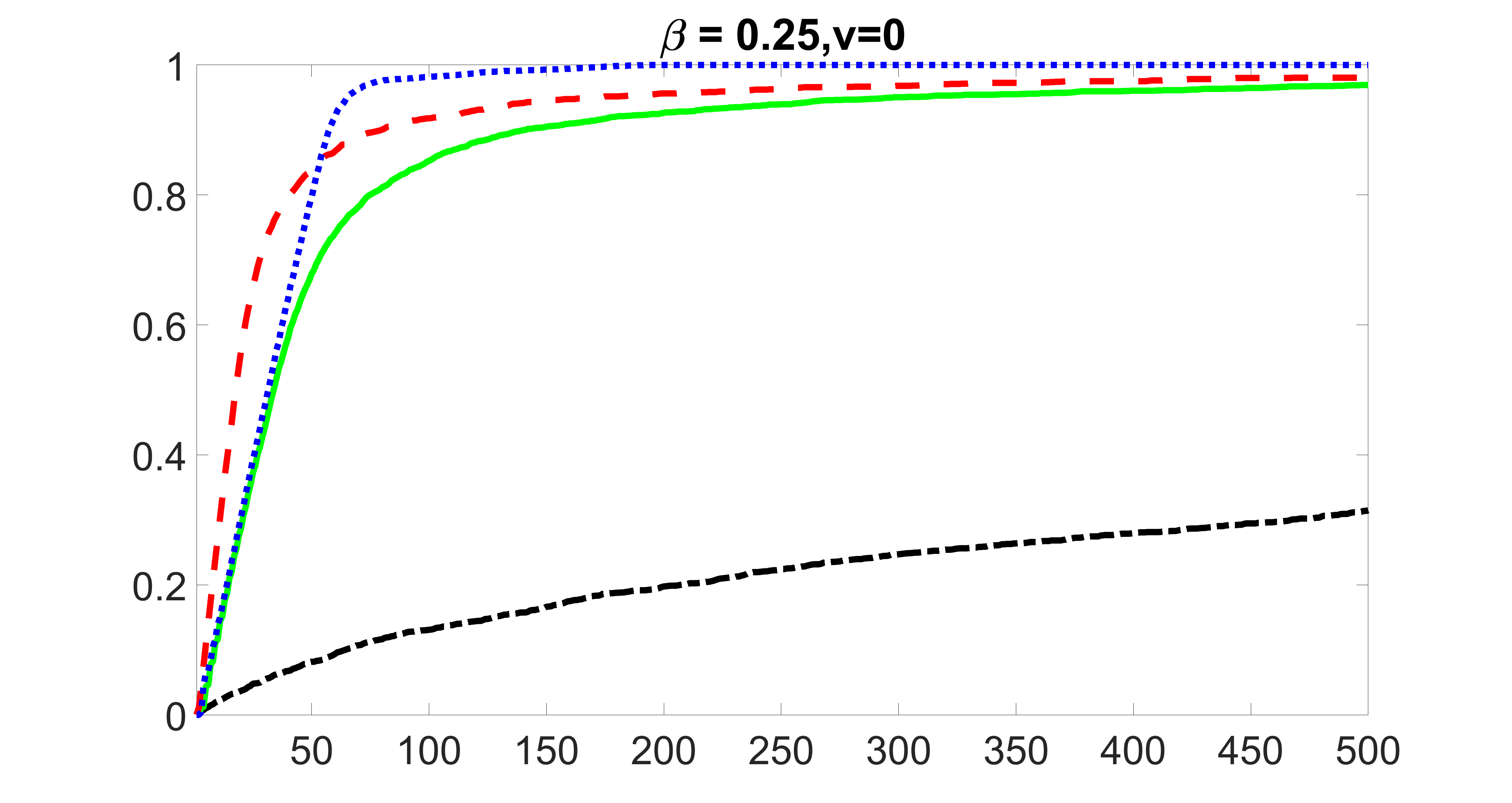}}
  \subcaptionbox{Overall coverage of $\mathcal{M}_1$}[0.45\linewidth]
 {\includegraphics[width=6cm,height=3.5cm]{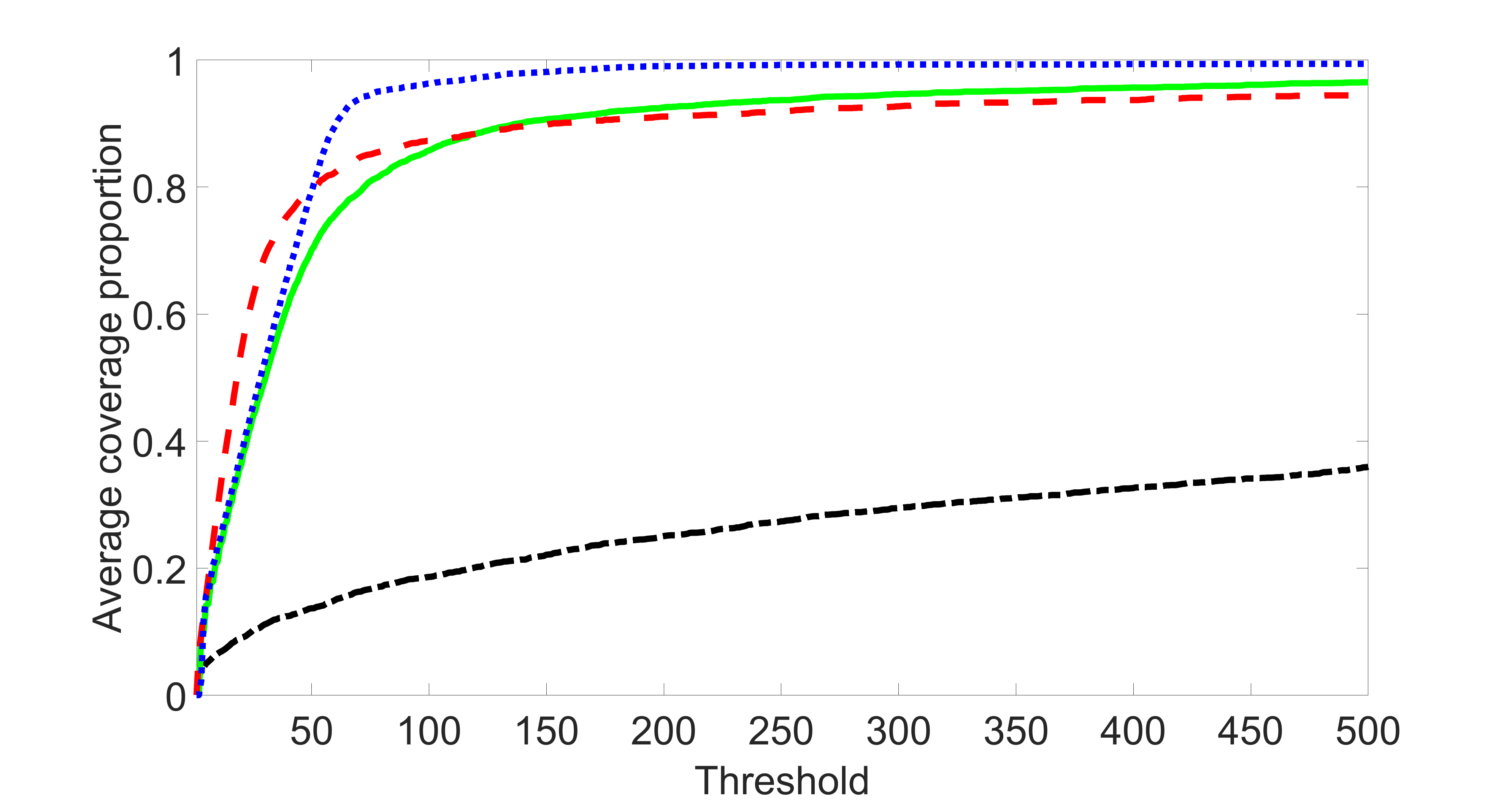}}
\caption{ Simulation results for the case $(n,s,K,\sigma) = (200,5000,24,1)$: Panels (a) -- (g) plot the average coverage proportion for $X_l$, where $l \in \mathcal{M}_1 =  \{1,2,3,104,105, 106\} \cup \mathcal{P}_{LD}$. Panels (a) -- (c) correspond to strong outcome and weak exposure predictor, moderate outcome and moderate exposure predictor and weak outcome and strong exposure predictor; Panels (d) -- (g) correspond to strong, moderate, and weak predictors of outcome only. Panel (g) plots the average coverage proportion for the index set $\mathcal{P}_{LD}$. Panel (h) plots the average coverage proportion for the index set $\mathcal{M}_1$. The x-axis represents the size of $\widehat{\mathcal{M}} $, while
y-axis denotes the average proportion. The blue dot, green solid, red dashed and black dash dotted lines denote the blockwise joint screening, joint screening, outcome screening, and intersection screening methods, respectively.}
\label{sim3step1n200sizesig24sigma1}
\end{figure}

\begin{figure}[htbp]
\captionsetup[subfigure]{justification=centering}
\centering
 \subcaptionbox{\footnotesize Confounder: strong \\ outcome, weak exposure}[0.45\linewidth]
 {\includegraphics[width=6cm,height=3.5cm]{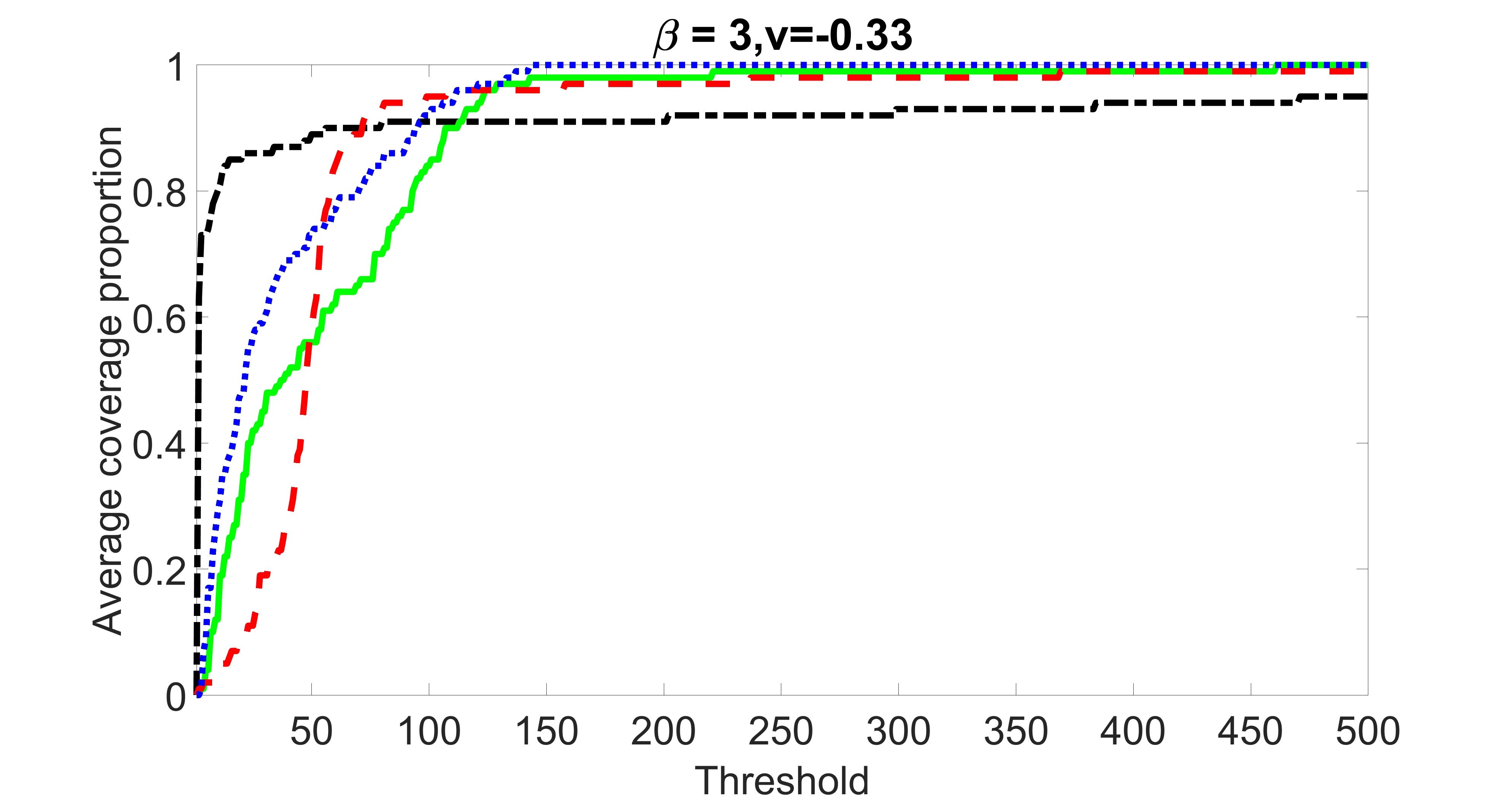}}
 \subcaptionbox{\footnotesize Confounder: medium \\ outcome, medium exposure}[0.45\linewidth]
 {\includegraphics[width=6cm,height=3.5cm]{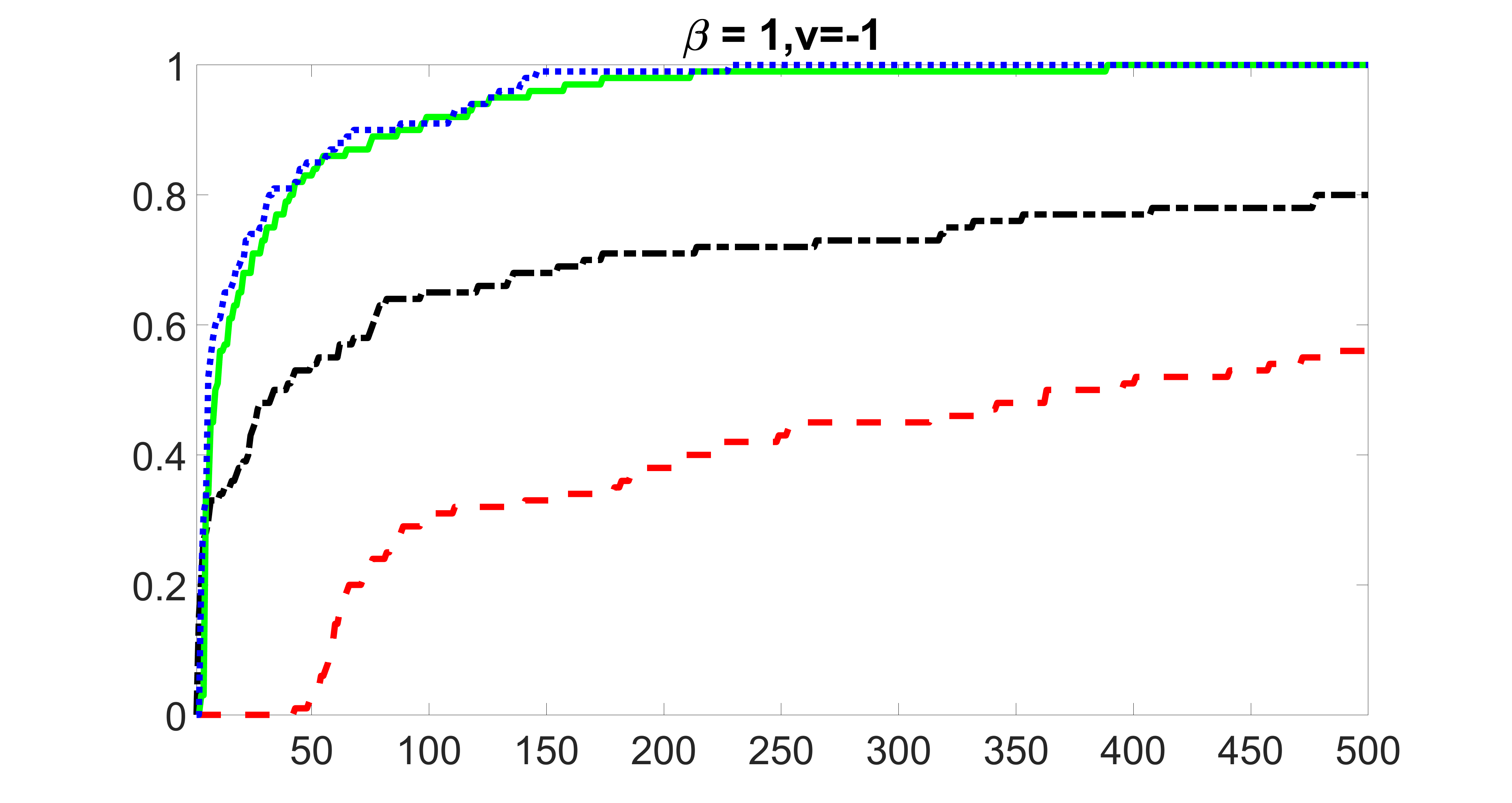}}
  \subcaptionbox{\footnotesize Confounder: weak \\ outcome, strong exposure}[0.45\linewidth]
 {\includegraphics[width=6cm,height=3.5cm]{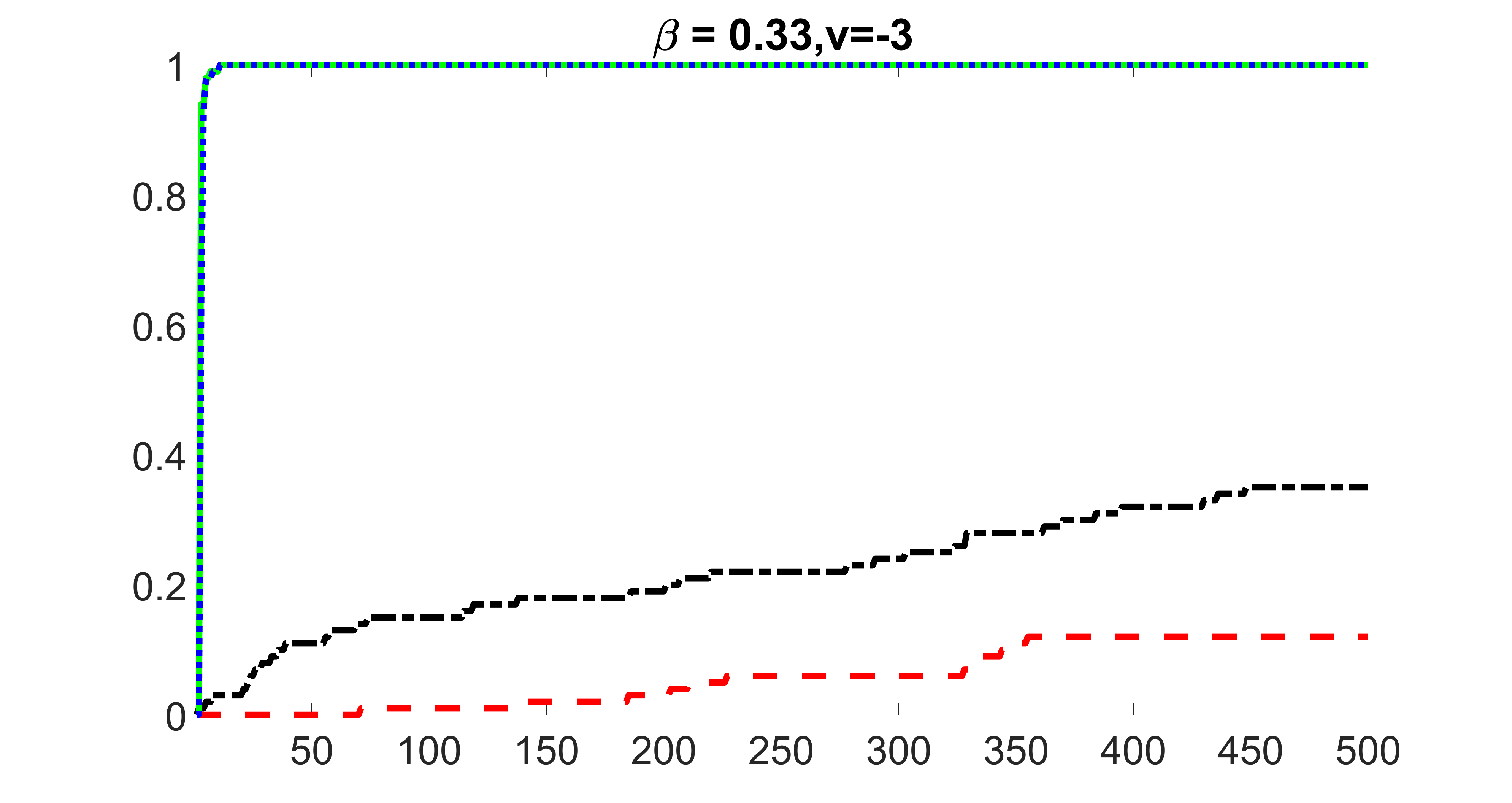}}
  \subcaptionbox{\footnotesize Precision: strong \\ outcome, zero exposure}[0.45\linewidth]
 {\includegraphics[width=6cm,height=3.5cm]{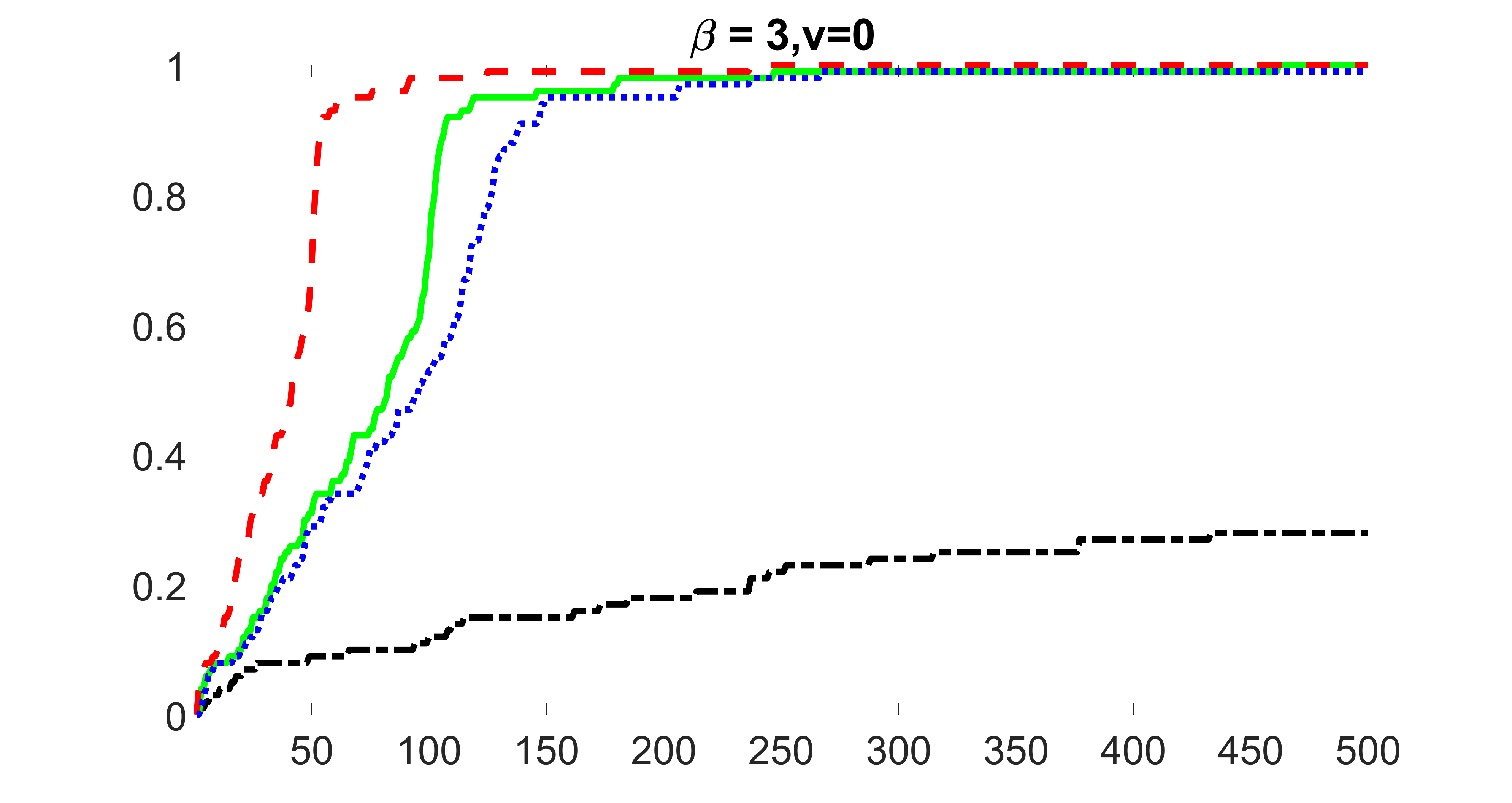}}
  \subcaptionbox{\footnotesize Precision: medium \\ outcome, zero exposure}[0.45\linewidth]
 {\includegraphics[width=6cm,height=3.5cm]{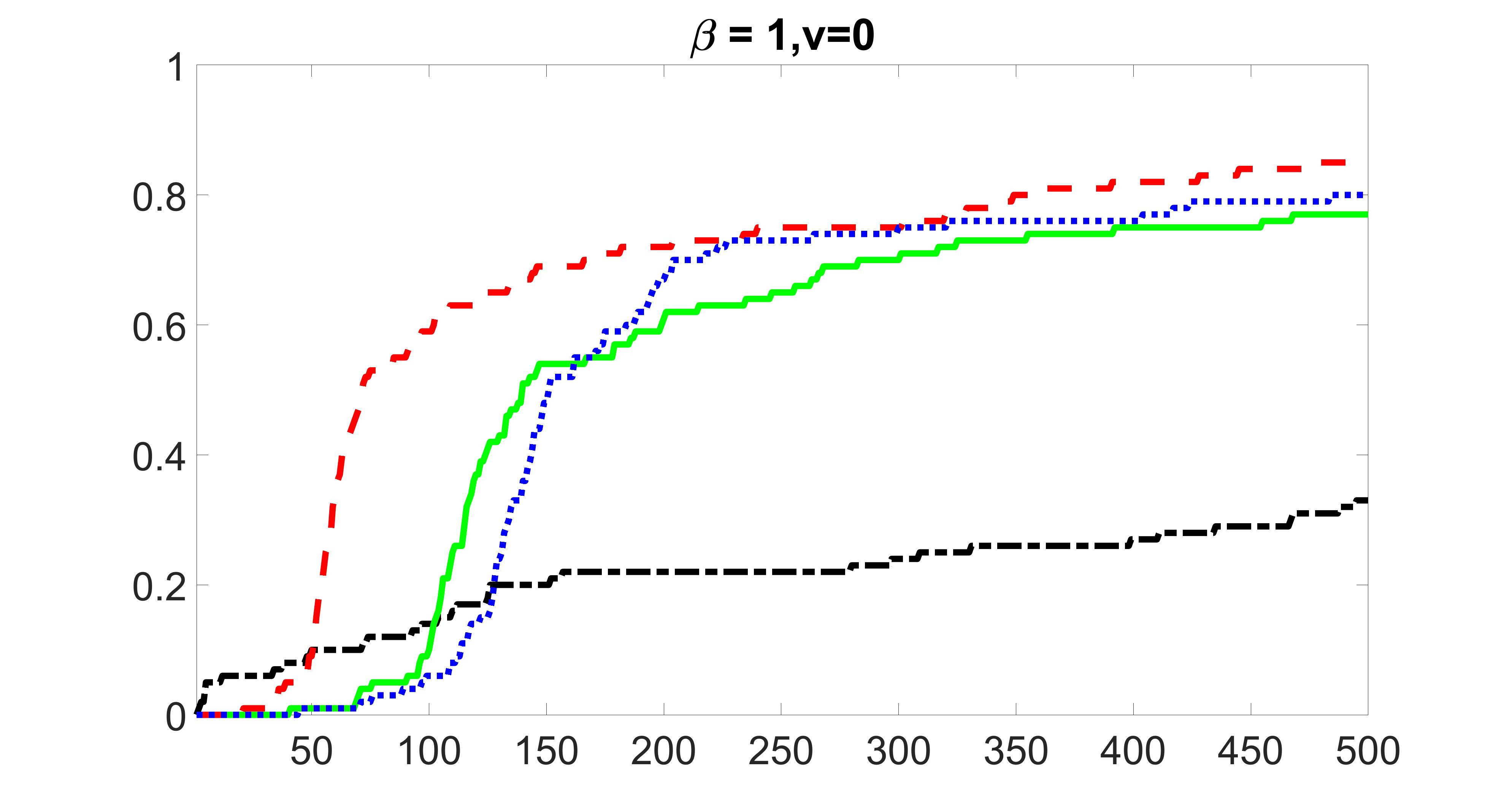}}
  \subcaptionbox{\footnotesize Precision: weak \\ outcome, zero exposure}[0.45\linewidth]
 {\includegraphics[width=6cm,height=3.5cm]{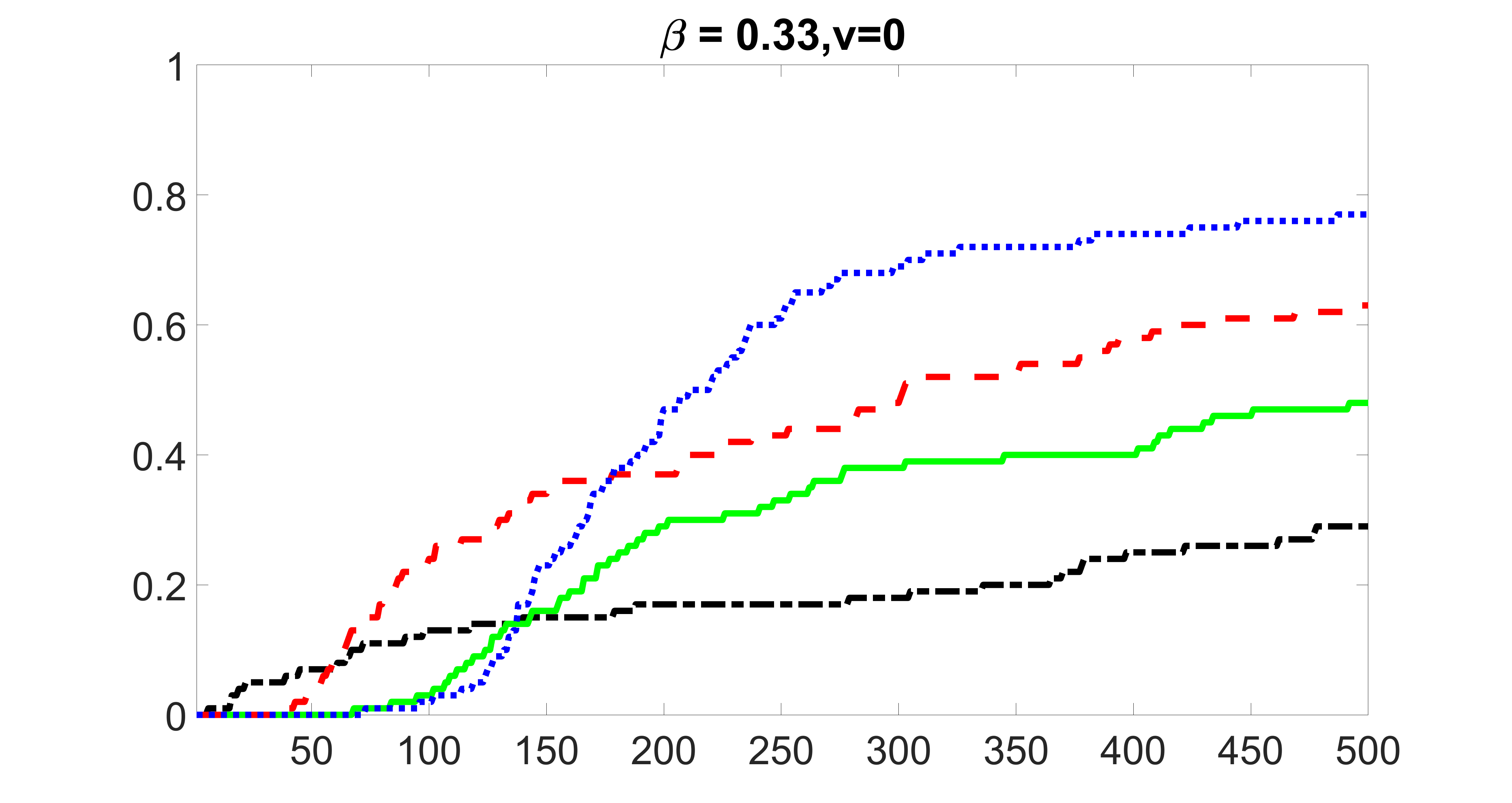}}
 \subcaptionbox{\footnotesize Precision: weaker \\ outcome, zero exposure}[0.45\linewidth]
 {\includegraphics[width=6cm,height=3.5cm]{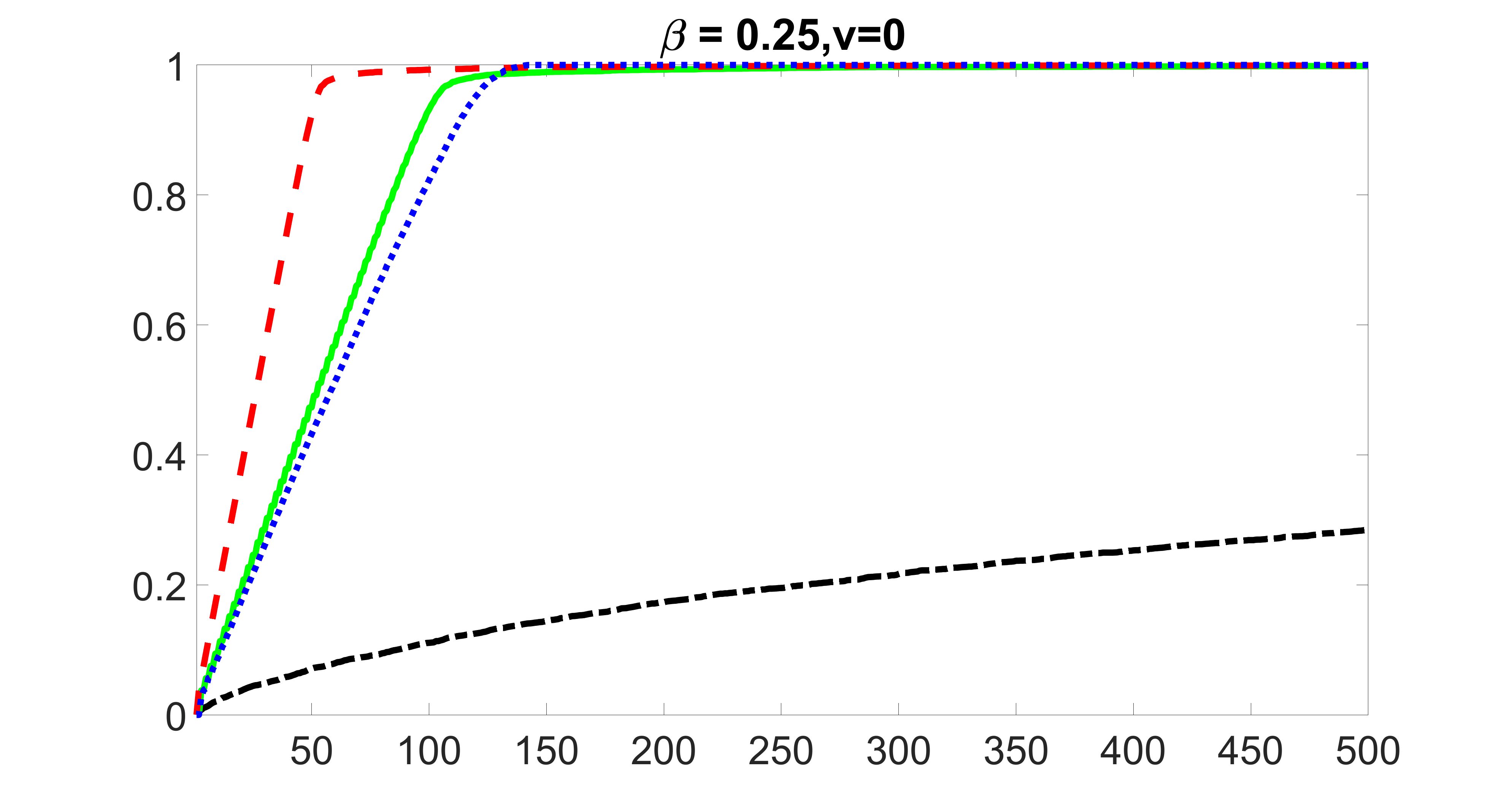}}
  \subcaptionbox{Overall coverage of $\mathcal{M}_1$}[0.45\linewidth]
 {\includegraphics[width=6cm,height=3.5cm]{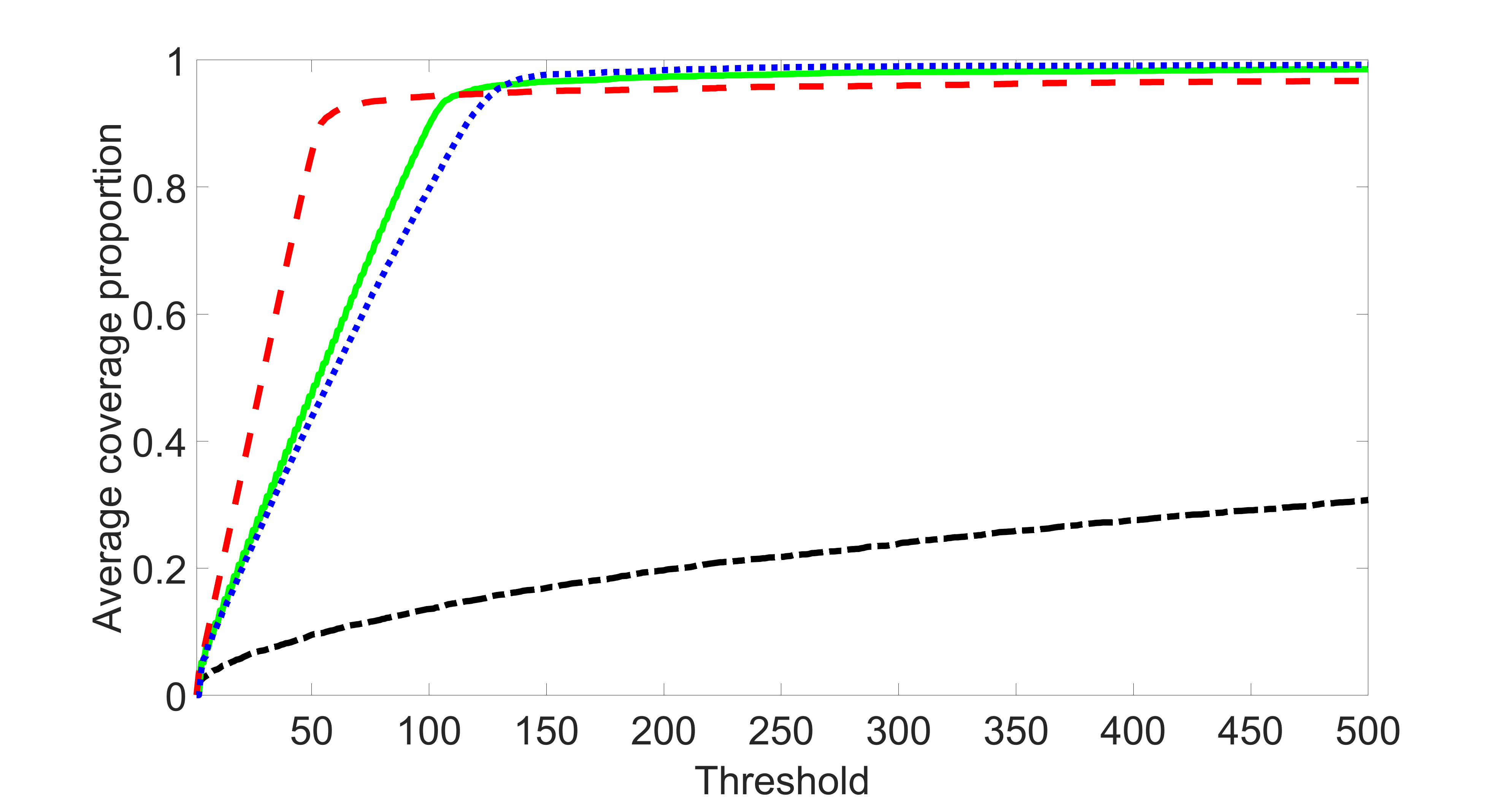}}
\caption{ Simulation results for the case $(n,s,K,\sigma) = (200,5000,52,1)$: Panels (a) -- (g) plot the average coverage proportion for $X_l$, where $l \in \mathcal{M}_1 =  \{1,2,3,104,105, 106\} \cup \mathcal{P}_{LD}$. Panels (a) -- (c) correspond to strong outcome and weak exposure predictor, moderate outcome and moderate exposure predictor and weak outcome and strong exposure predictor; Panels (d) -- (g) correspond to strong, moderate, and weak predictors of outcome only. Panel (g) plots the average coverage proportion for the index set $\mathcal{P}_{LD}$. Panel (h) plots the average coverage proportion for the index set $\mathcal{M}_1$. The x-axis represents the size of $\widehat{\mathcal{M}} $, while
y-axis denotes the average proportion. The blue dot, green solid, red dashed and black dash dotted lines denote the blockwise joint screening, joint screening, outcome screening, and intersection screening methods, respectively.}
\label{sim3step1n200sizesig52sigma1}
\end{figure}

\begin{figure}[htbp]
\captionsetup[subfigure]{justification=centering}
\centering
 \subcaptionbox{\footnotesize Confounder: strong \\ outcome, weak exposure}[0.45\linewidth]
 {\includegraphics[width=6cm,height=3.5cm]{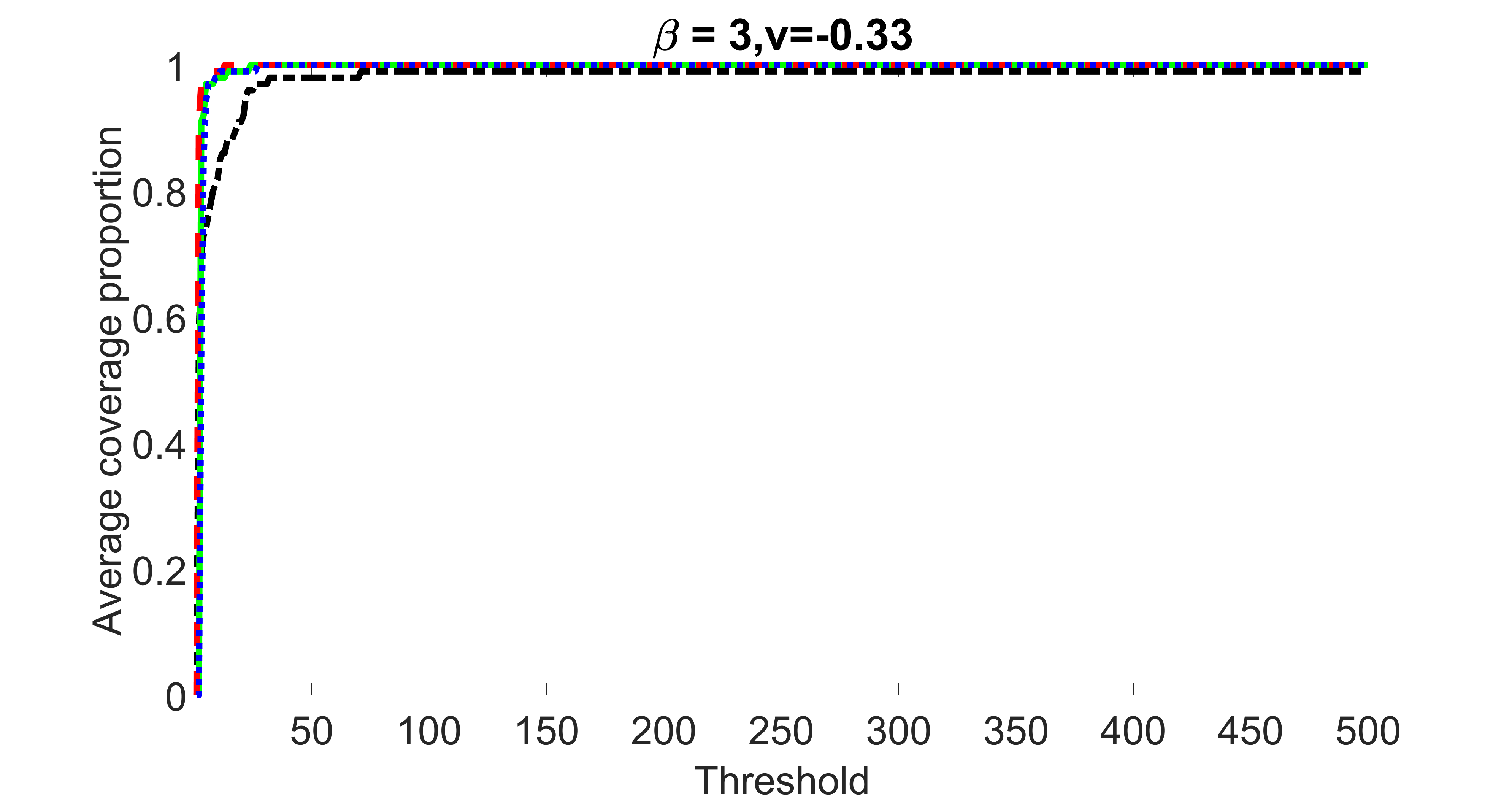}}
 \subcaptionbox{\footnotesize Confounder: medium \\ outcome, medium exposure}[0.45\linewidth]
 {\includegraphics[width=6cm,height=3.5cm]{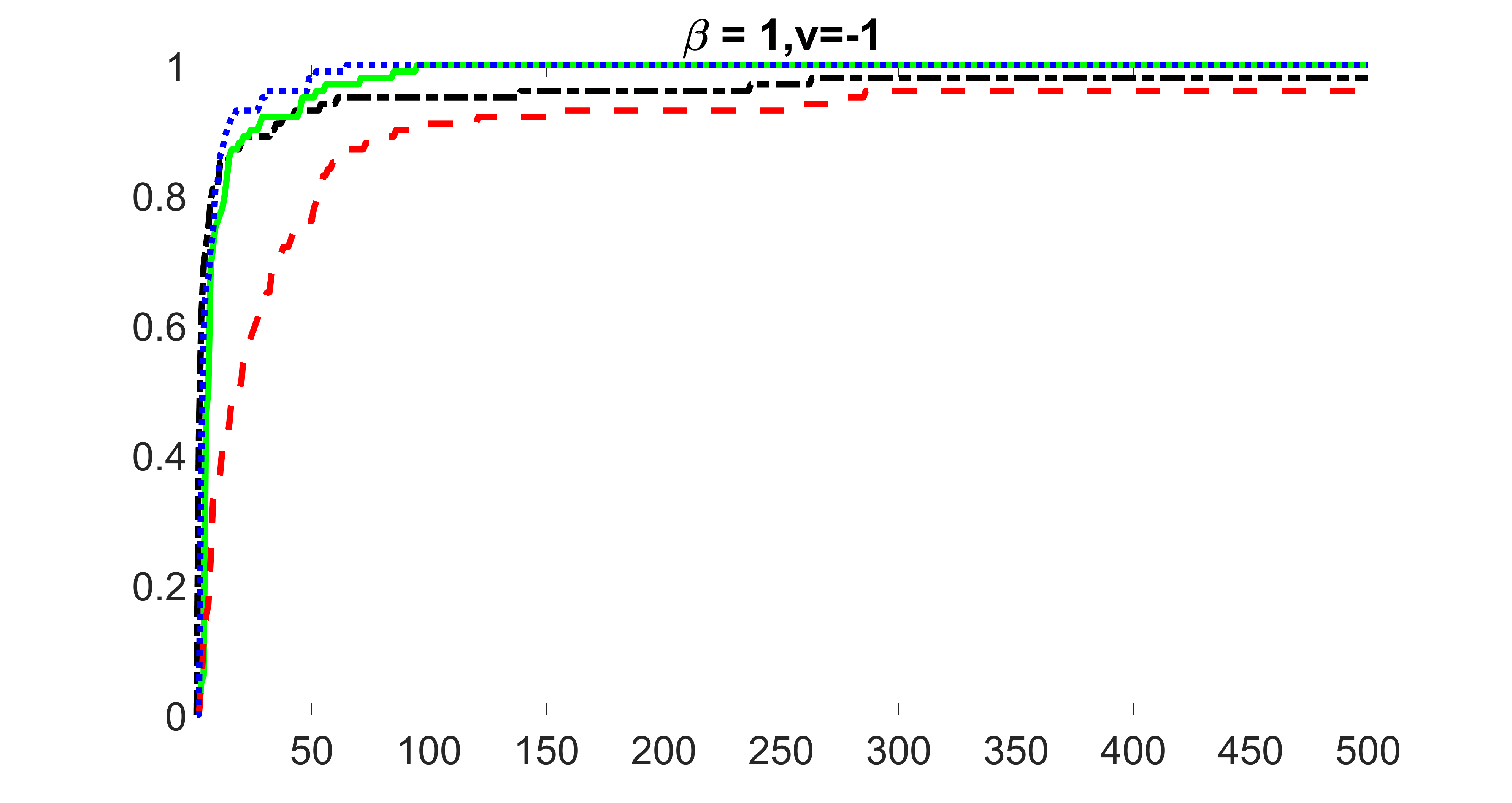}}
  \subcaptionbox{\footnotesize Confounder: weak \\ outcome, strong exposure}[0.45\linewidth]
 {\includegraphics[width=6cm,height=3.5cm]{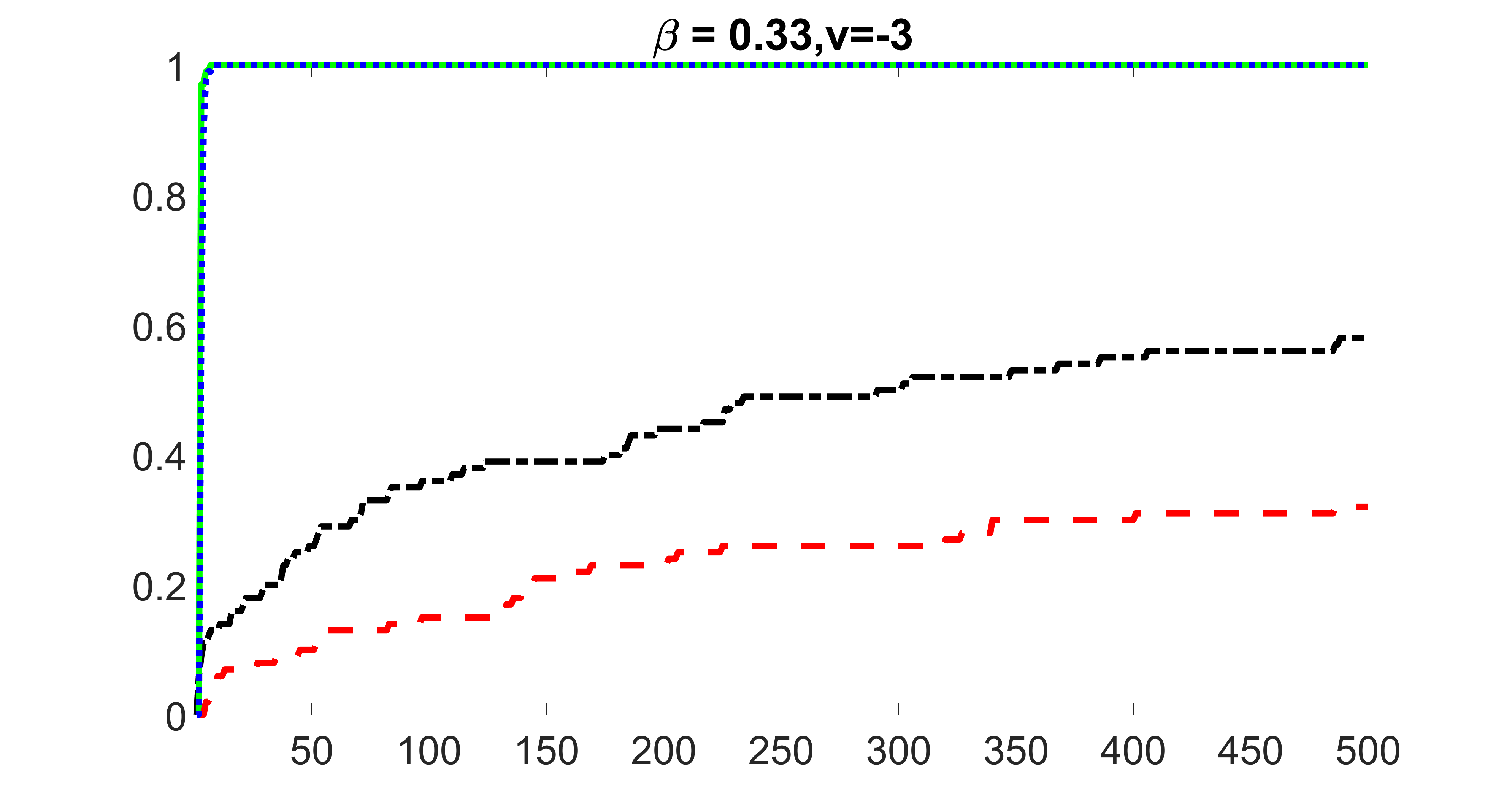}}
  \subcaptionbox{\footnotesize Precision: strong \\ outcome, zero exposure}[0.45\linewidth]
 {\includegraphics[width=6cm,height=3.5cm]{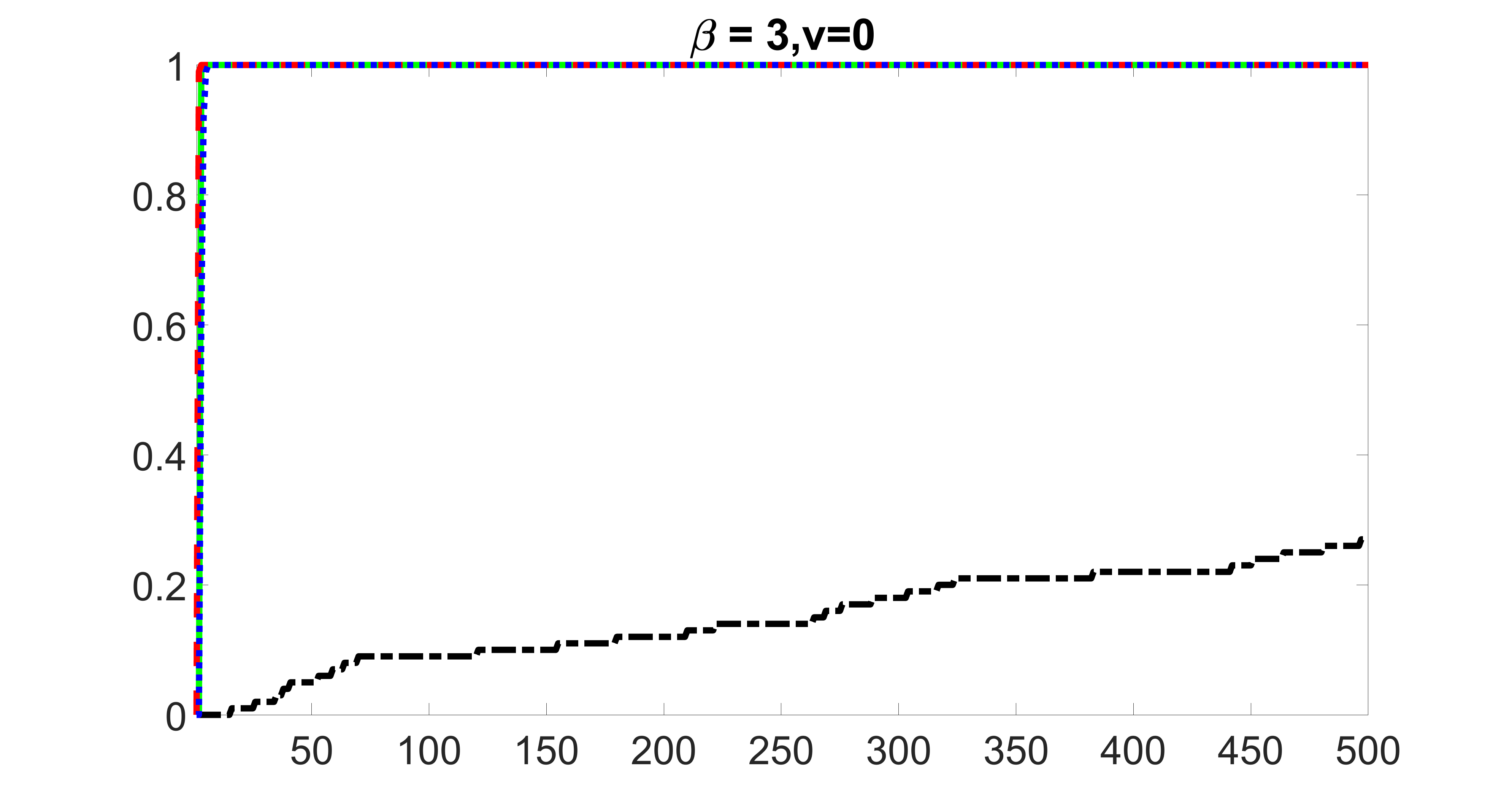}}
  \subcaptionbox{\footnotesize Precision: medium \\ outcome, zero exposure}[0.45\linewidth]
 {\includegraphics[width=6cm,height=3.5cm]{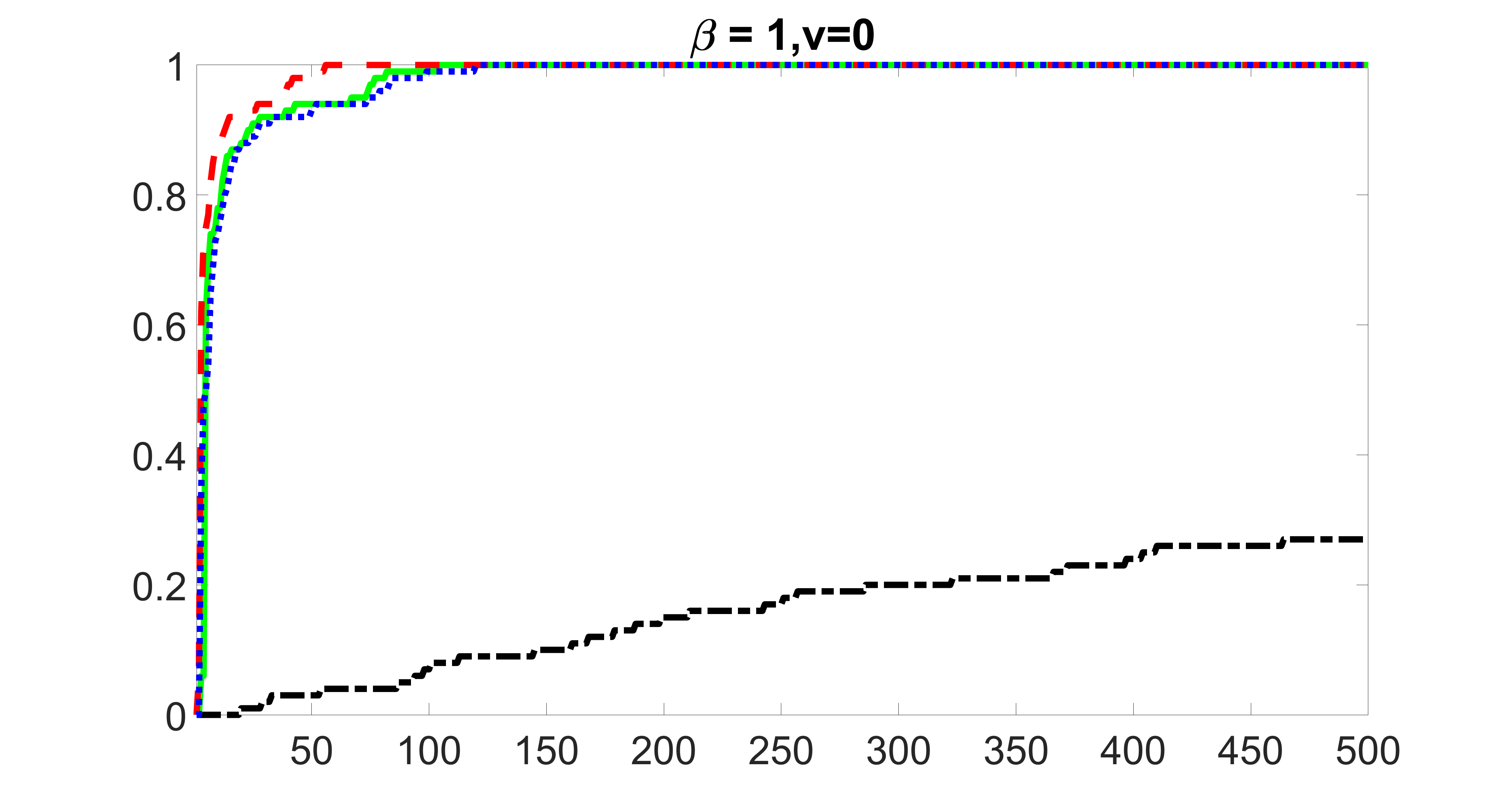}}
  \subcaptionbox{\footnotesize Precision: weak \\ outcome, zero exposure}[0.45\linewidth]
 {\includegraphics[width=6cm,height=3.5cm]{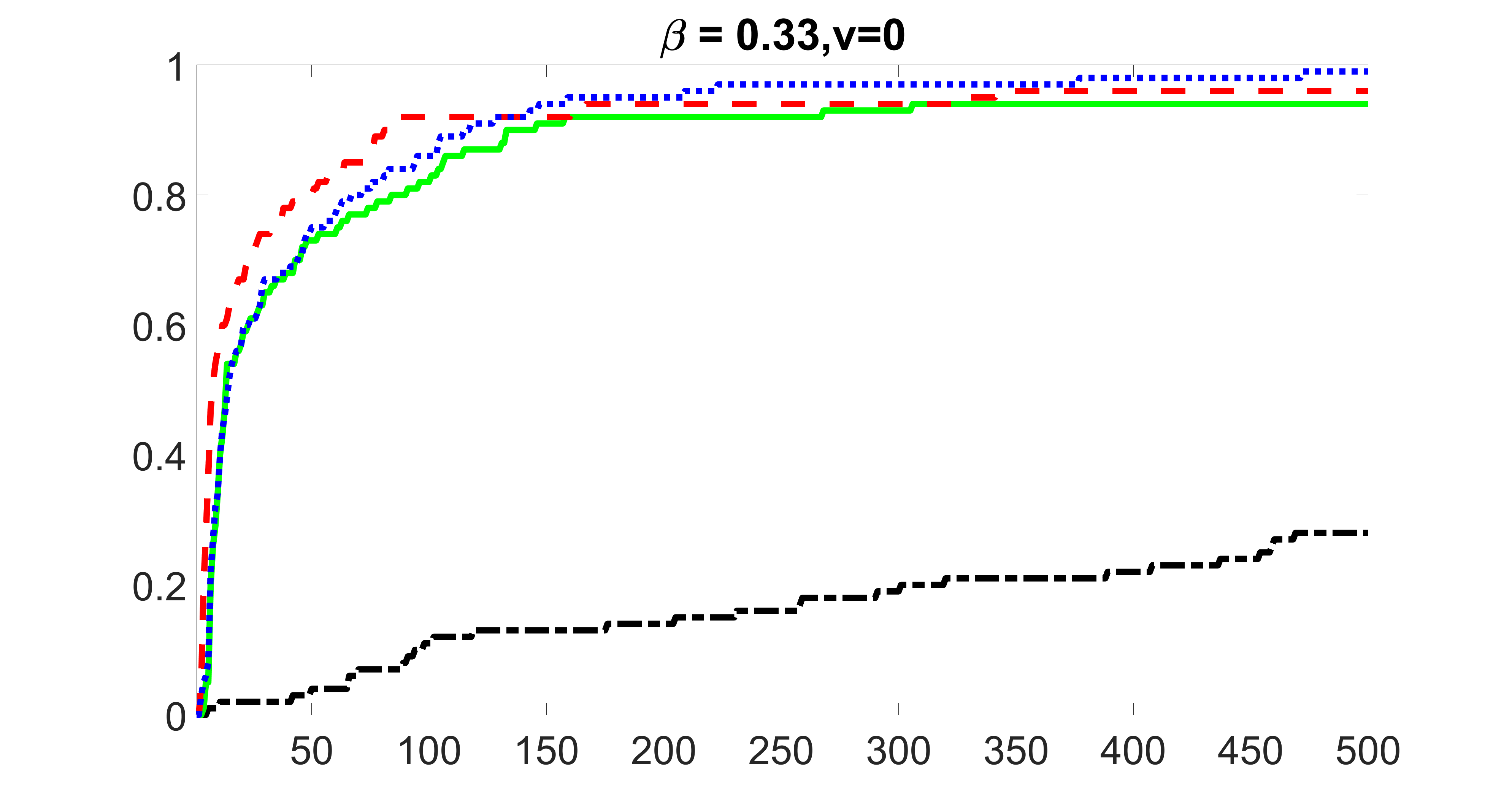}}
 \subcaptionbox{\footnotesize Precision: weaker \\ outcome, zero exposure}[0.45\linewidth]
 {\includegraphics[width=6cm,height=3.5cm]{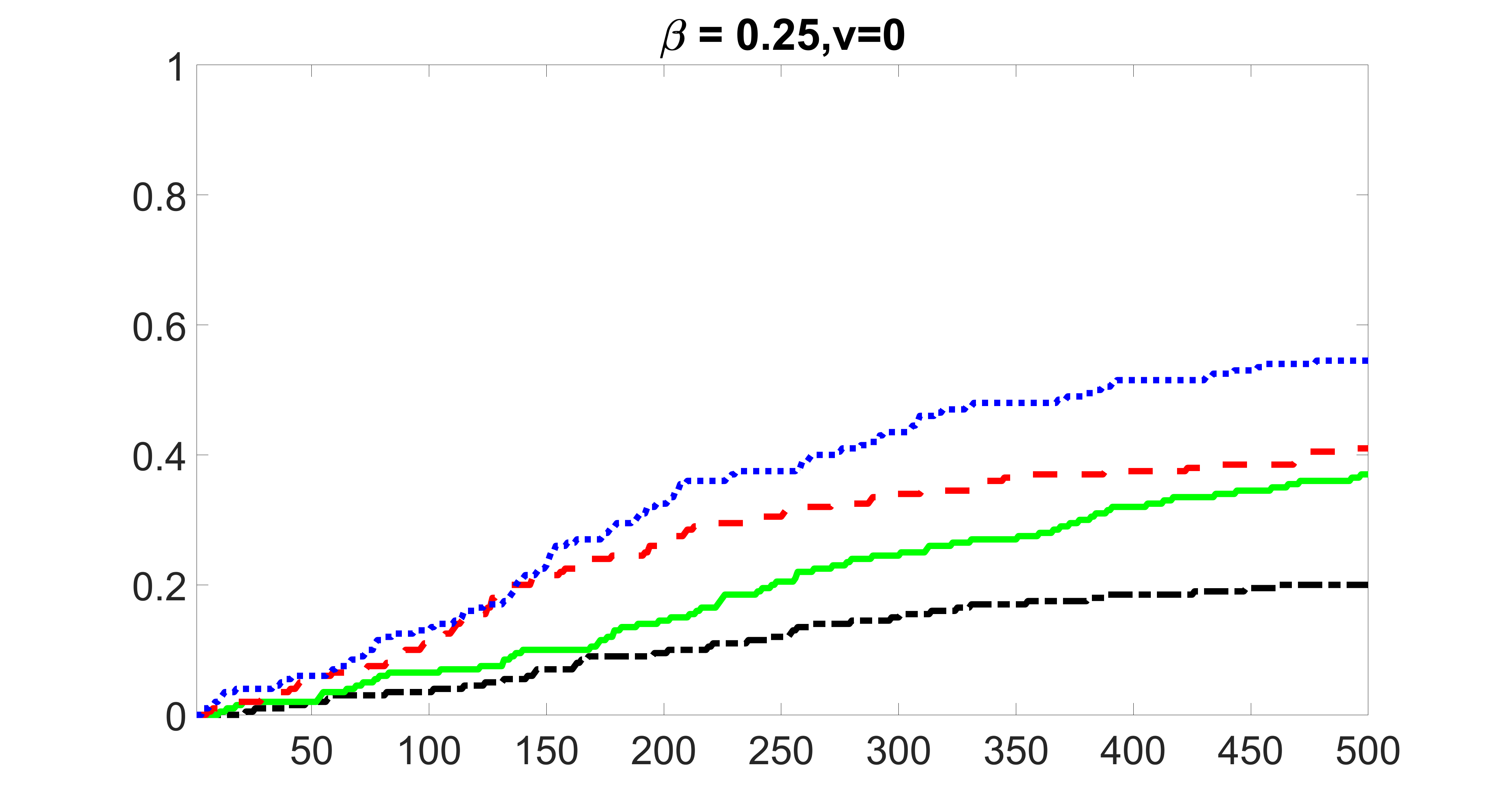}}
  \subcaptionbox{Overall coverage of $\mathcal{M}_1$}[0.45\linewidth]
 {\includegraphics[width=6cm,height=3.5cm]{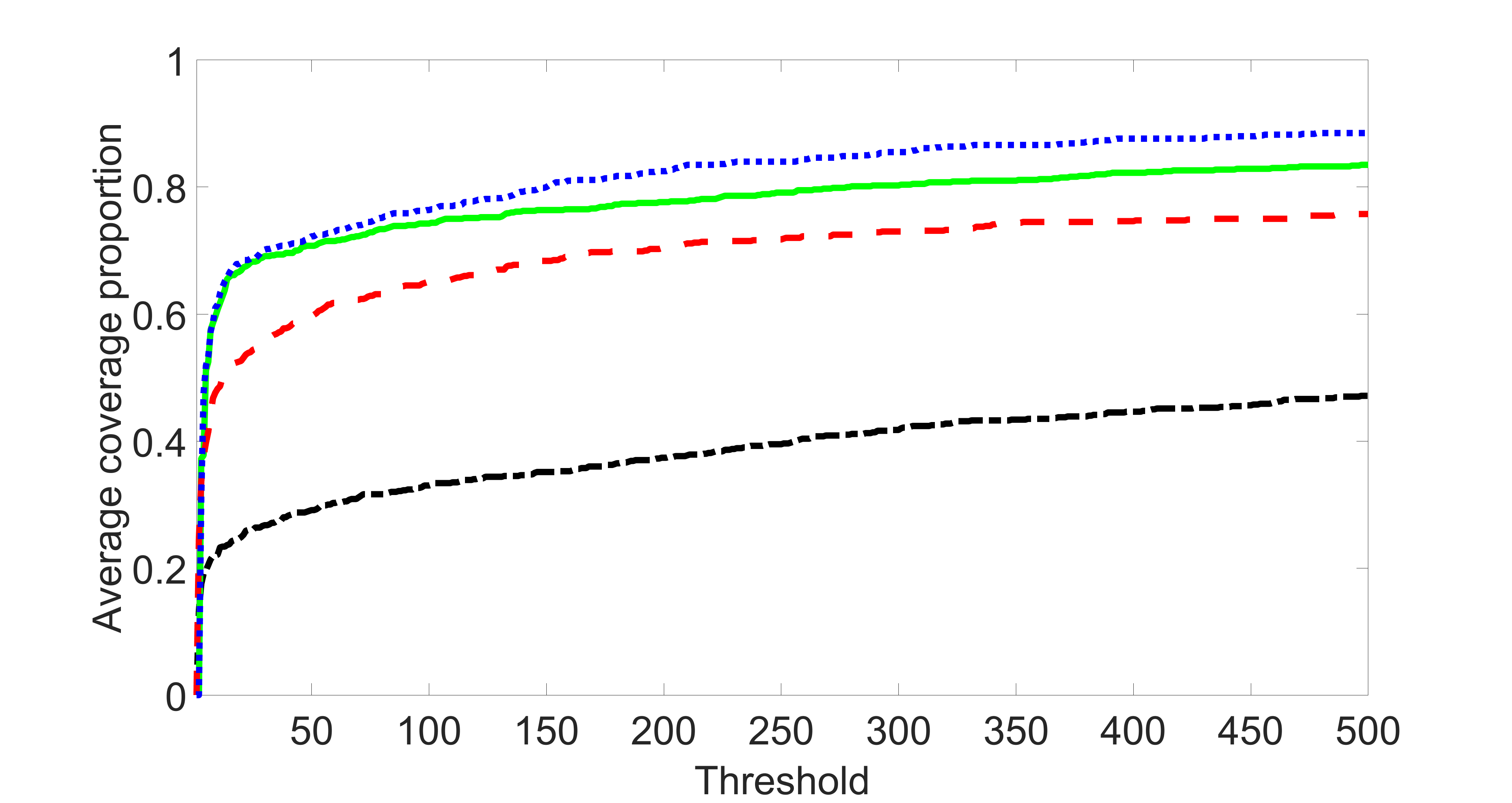}}
\caption{Simulation results for the case $(n,s,K,\sigma) = (500,5000,2,1)$: Panels (a) -- (g) plot the average coverage proportion for $X_l$, where $l \in \mathcal{M}_1 =  \{1,2,3,104,105, 106\} \cup \mathcal{P}_{LD}$. Panels (a) -- (c) correspond to strong outcome and weak exposure predictor, moderate outcome and moderate exposure predictor and weak outcome and strong exposure predictor; Panels (d) -- (g) correspond to strong, moderate, and weak predictors of outcome only. Panel (g) plots the average coverage proportion for the index set $\mathcal{P}_{LD}$. Panel (h) plots the average coverage proportion for the index set $\mathcal{M}_1$. The x-axis represents the size of $\widehat{\mathcal{M}} $, while
y-axis denotes the average proportion. The blue dot, green solid, red dashed and black dash dotted lines denote the blockwise joint screening, joint screening, outcome screening, and intersection screening methods, respectively.}
\label{sim3step1n500sizesig2sigma1}
\end{figure}

\begin{figure}[htbp]
\captionsetup[subfigure]{justification=centering}
\centering
 \subcaptionbox{\footnotesize Confounder: strong \\ outcome, weak exposure}[0.45\linewidth]
 {\includegraphics[width=6cm,height=3.5cm]{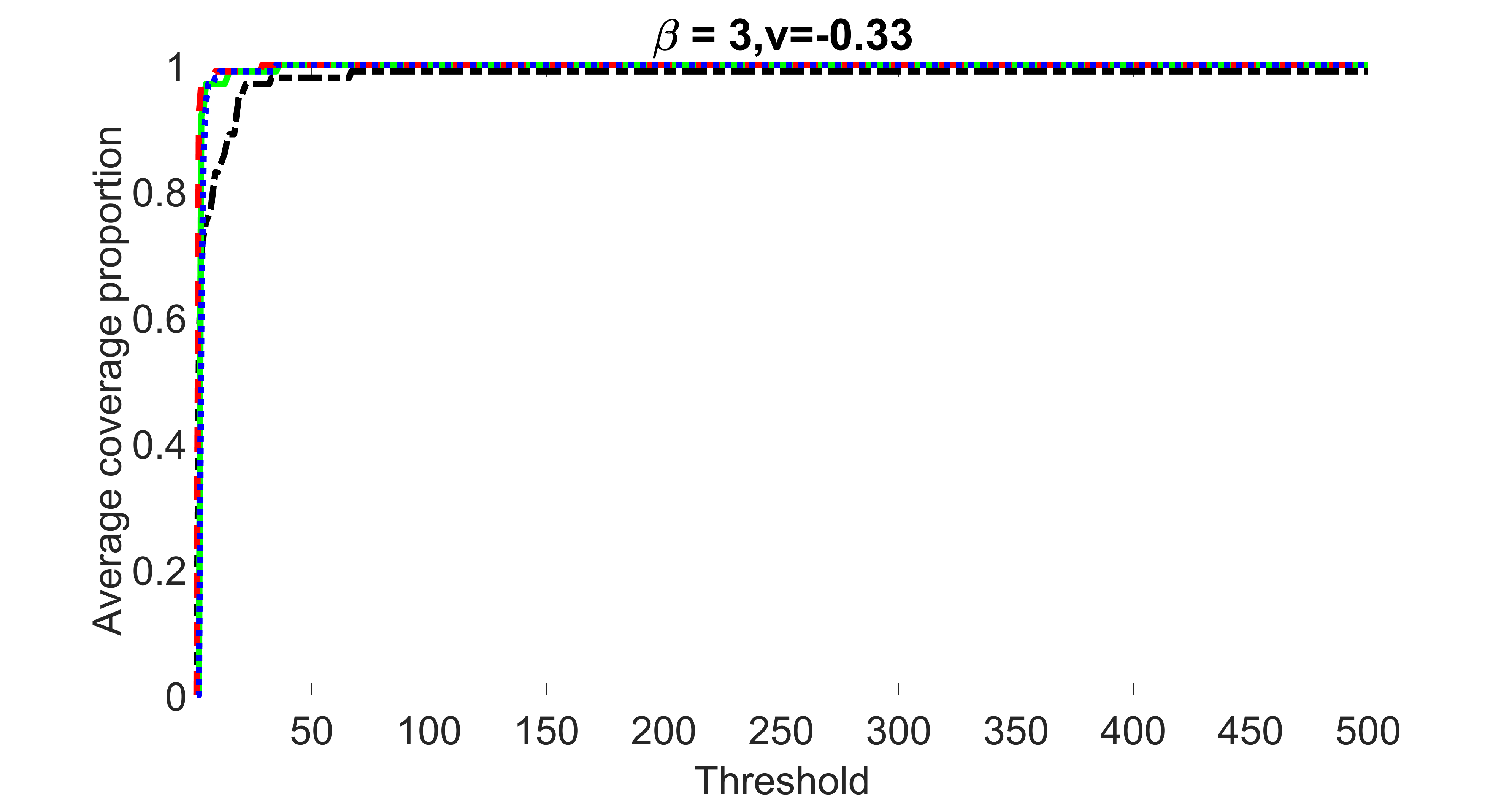}}
 \subcaptionbox{\footnotesize Confounder: medium \\ outcome, medium exposure}[0.45\linewidth]
 {\includegraphics[width=6cm,height=3.5cm]{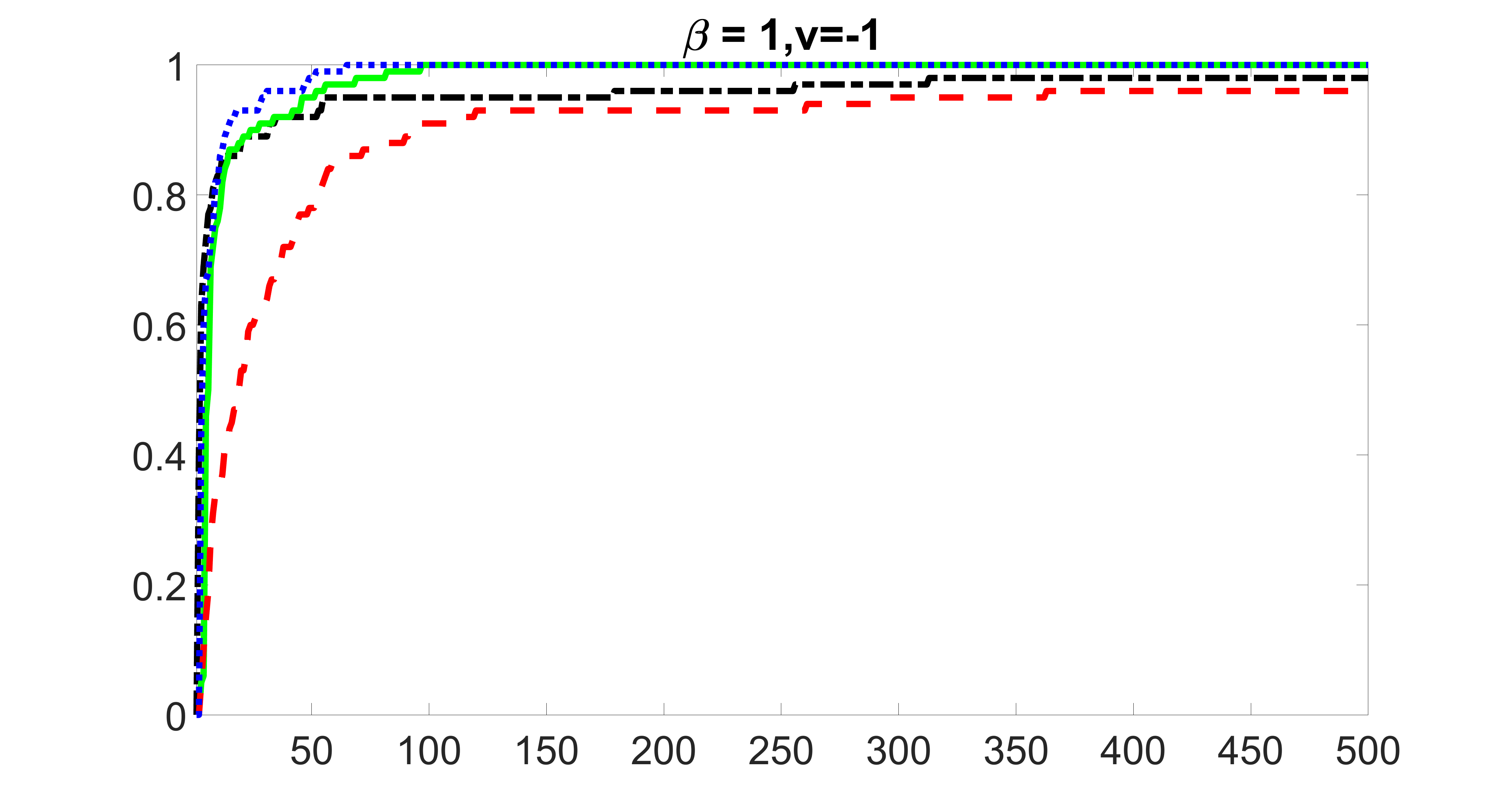}}
  \subcaptionbox{\footnotesize Confounder: weak \\ outcome, strong exposure}[0.45\linewidth]
 {\includegraphics[width=6cm,height=3.5cm]{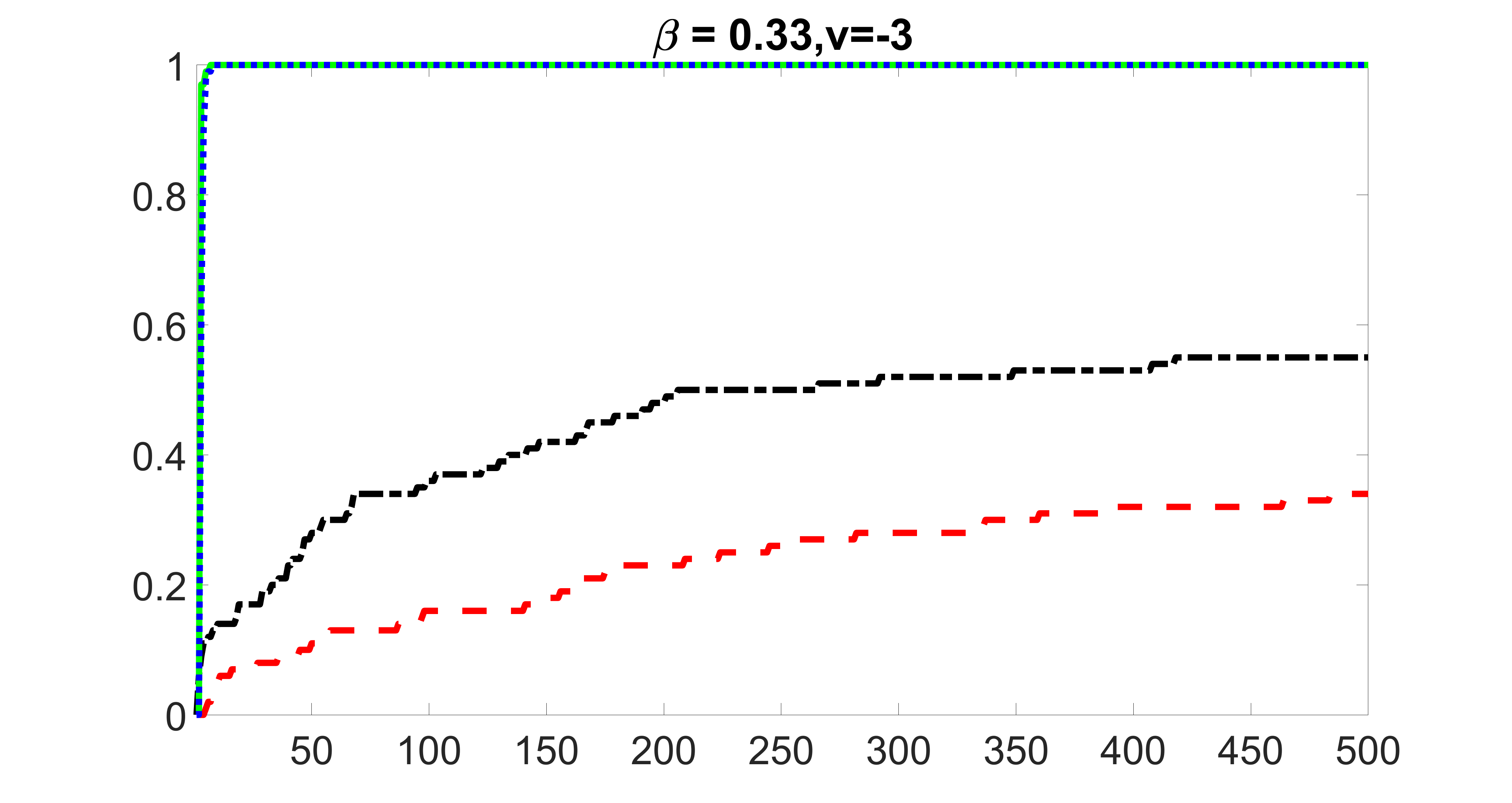}}
  \subcaptionbox{\footnotesize Precision: strong \\ outcome, zero exposure}[0.45\linewidth]
 {\includegraphics[width=6cm,height=3.5cm]{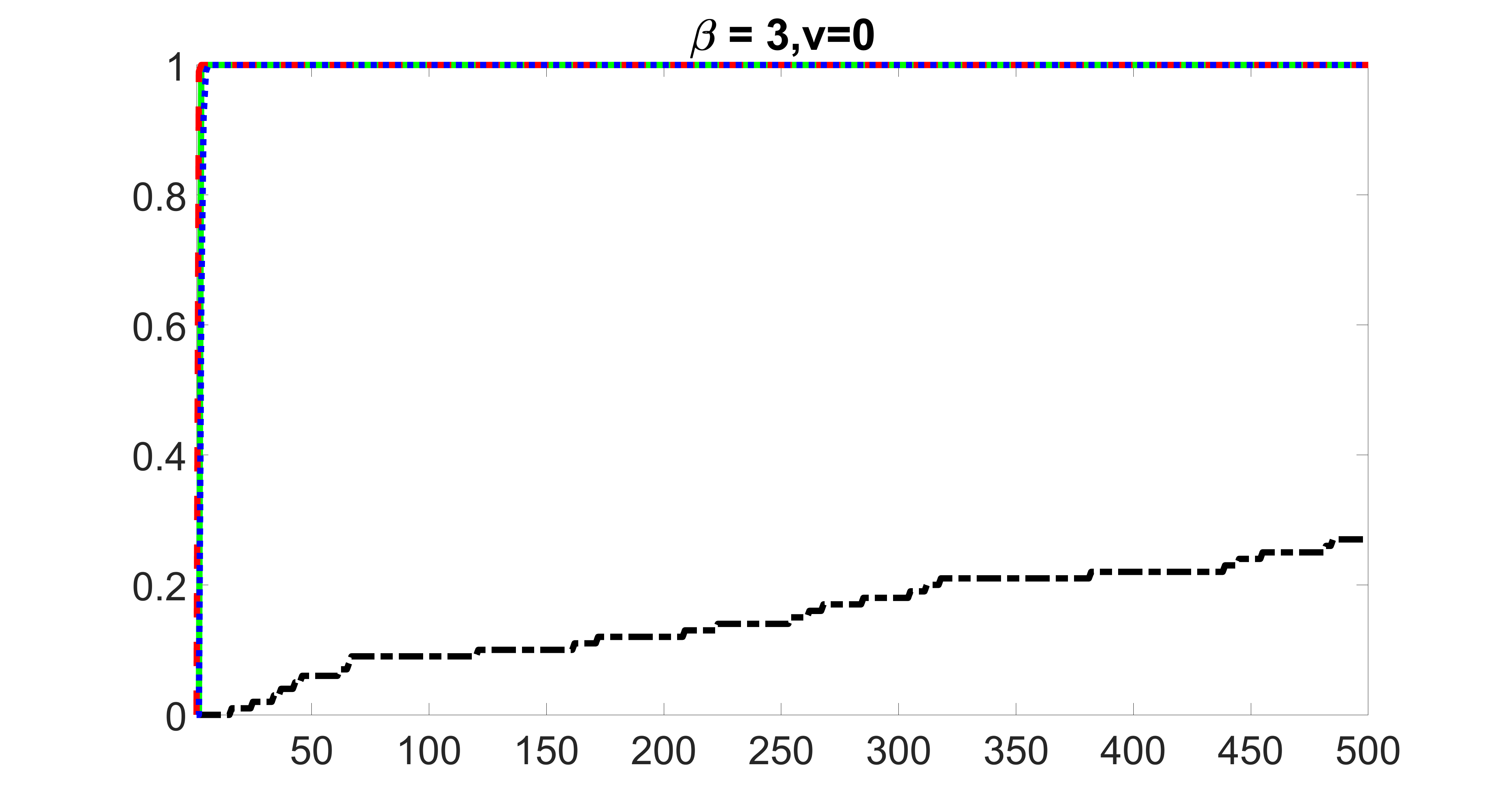}}
  \subcaptionbox{\footnotesize Precision: medium \\ outcome, zero exposure}[0.45\linewidth]
 {\includegraphics[width=6cm,height=3.5cm]{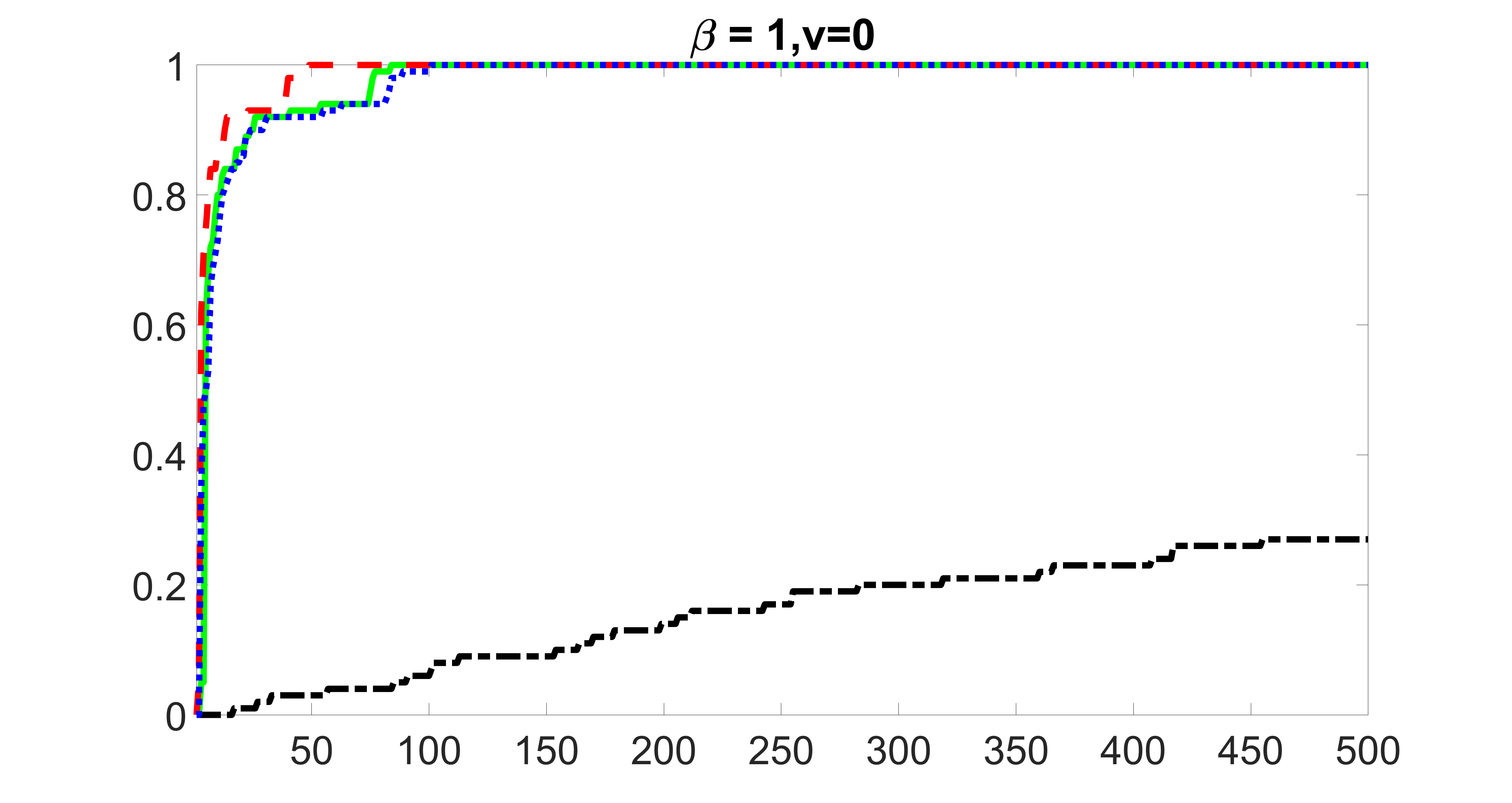}}
  \subcaptionbox{\footnotesize Precision: weak \\ outcome, zero exposure}[0.45\linewidth]
 {\includegraphics[width=6cm,height=3.5cm]{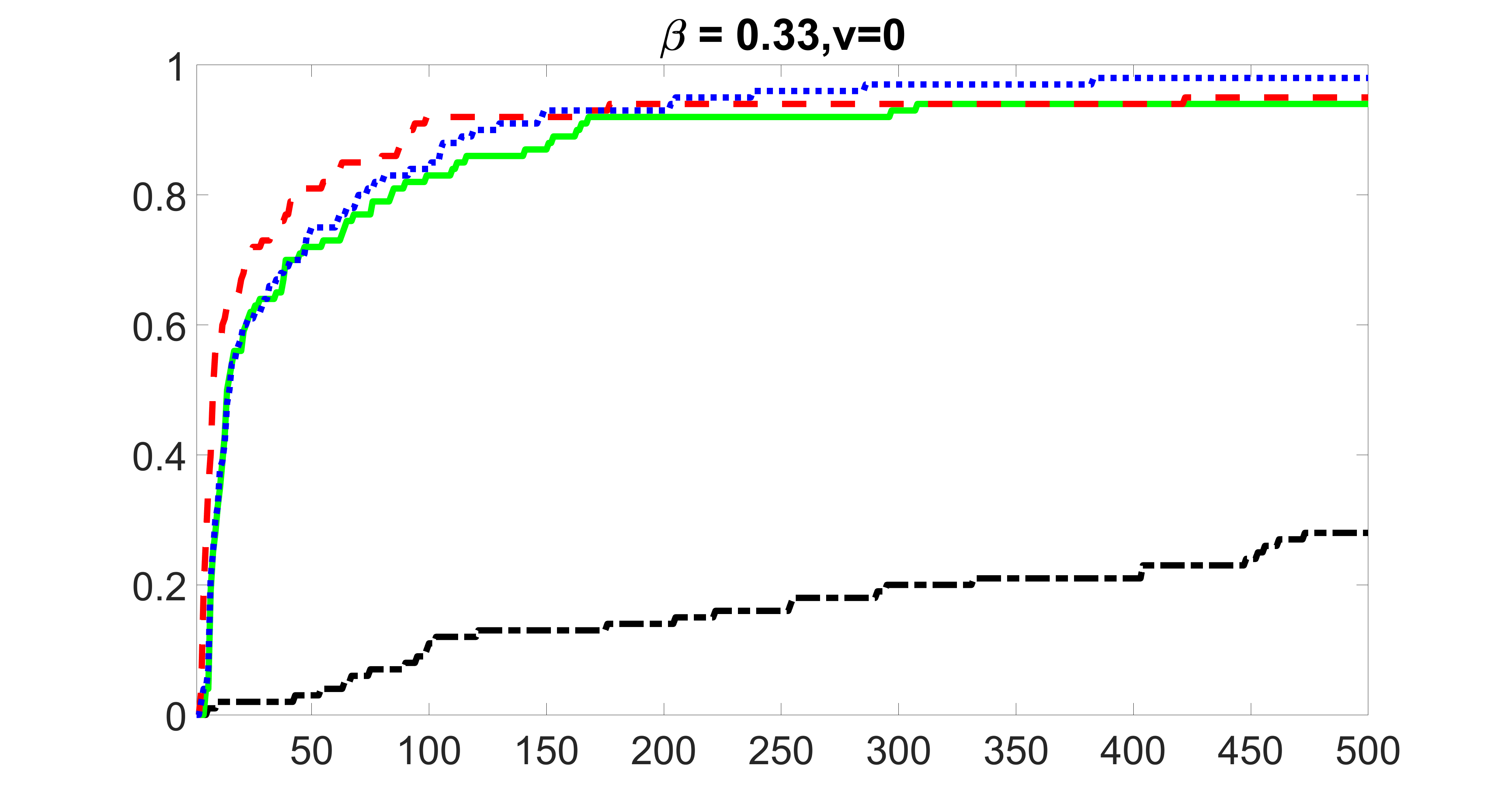}}
 \subcaptionbox{\footnotesize Precision: weaker \\ outcome, zero exposure}[0.45\linewidth]
 {\includegraphics[width=6cm,height=3.5cm]{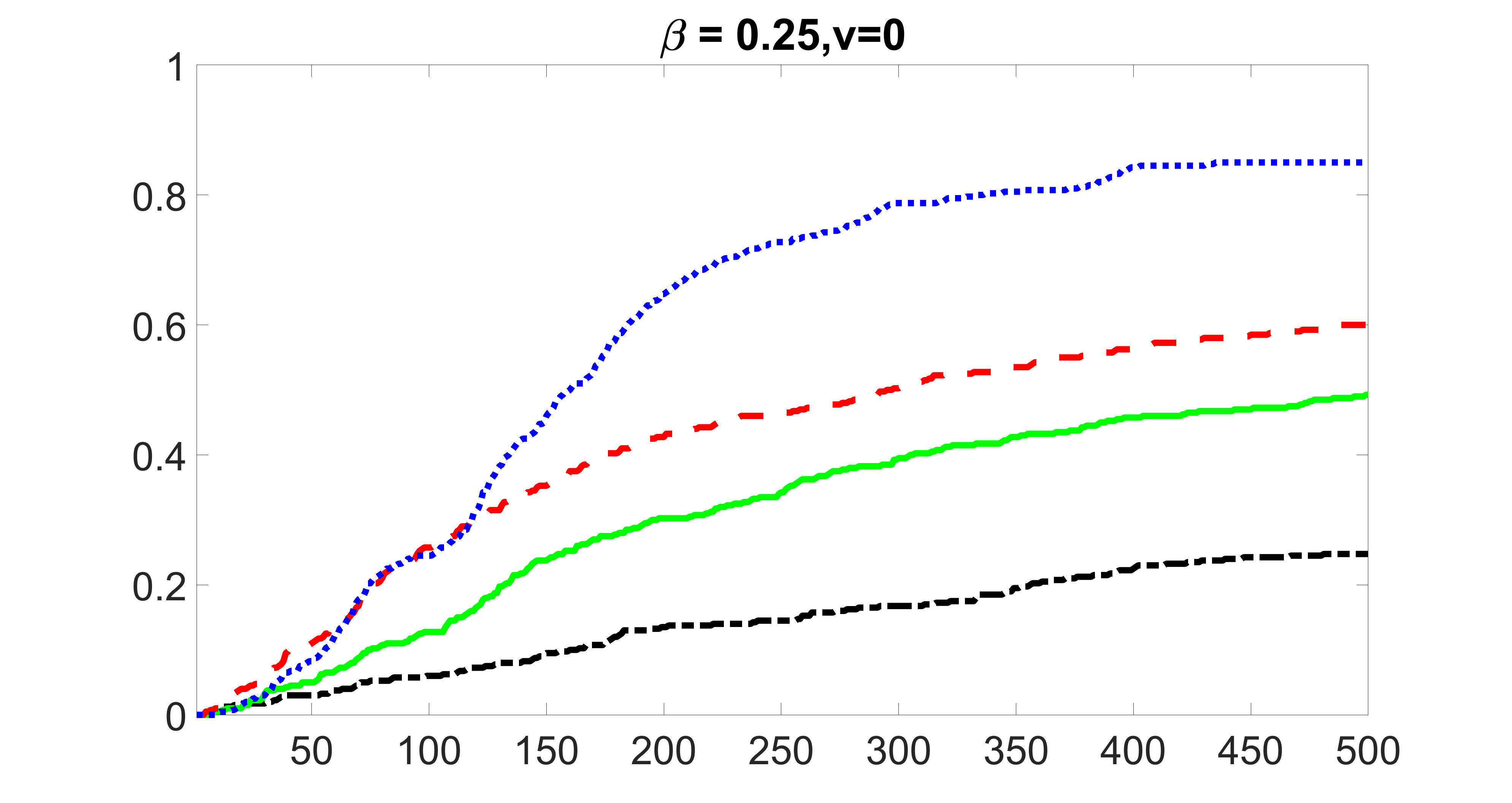}}
  \subcaptionbox{Overall coverage of $\mathcal{M}_1$}[0.45\linewidth]
 {\includegraphics[width=6cm,height=3.5cm]{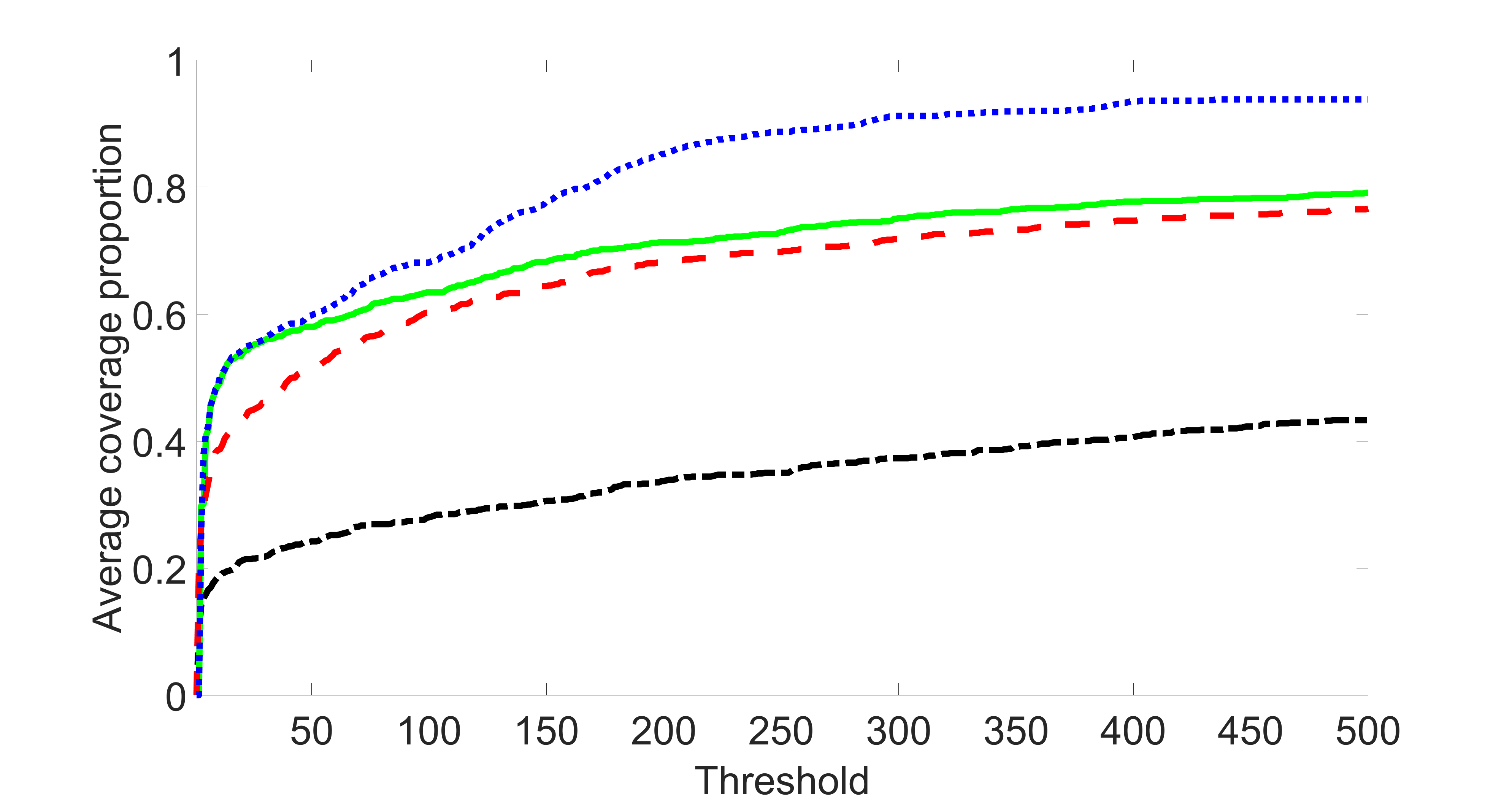}}
\caption{Simulation results for the case $(n,s,K,\sigma) = (500,5000,4,1)$: Panels (a) -- (g) plot the average coverage proportion for $X_l$, where $l \in \mathcal{M}_1 =  \{1,2,3,104,105, 106\} \cup \mathcal{P}_{LD}$. Panels (a) -- (c) correspond to strong outcome and weak exposure predictor, moderate outcome and moderate exposure predictor and weak outcome and strong exposure predictor; Panels (d) -- (g) correspond to strong, moderate, and weak predictors of outcome only. Panel (g) plots the average coverage proportion for the index set $\mathcal{P}_{LD}$. Panel (h) plots the average coverage proportion for the index set $\mathcal{M}_1$. The x-axis represents the size of $\widehat{\mathcal{M}} $, while
y-axis denotes the average proportion. The blue dot, green solid, red dashed and black dash dotted lines denote the blockwise joint screening, joint screening, outcome screening, and intersection screening methods, respectively.}
\label{sim3step1n500sizesig4sigma1}
\end{figure}

\begin{figure}[htbp]
\captionsetup[subfigure]{justification=centering}
\centering
 \subcaptionbox{\footnotesize Confounder: strong \\ outcome, weak exposure}[0.45\linewidth]
 {\includegraphics[width=6cm,height=3.5cm]{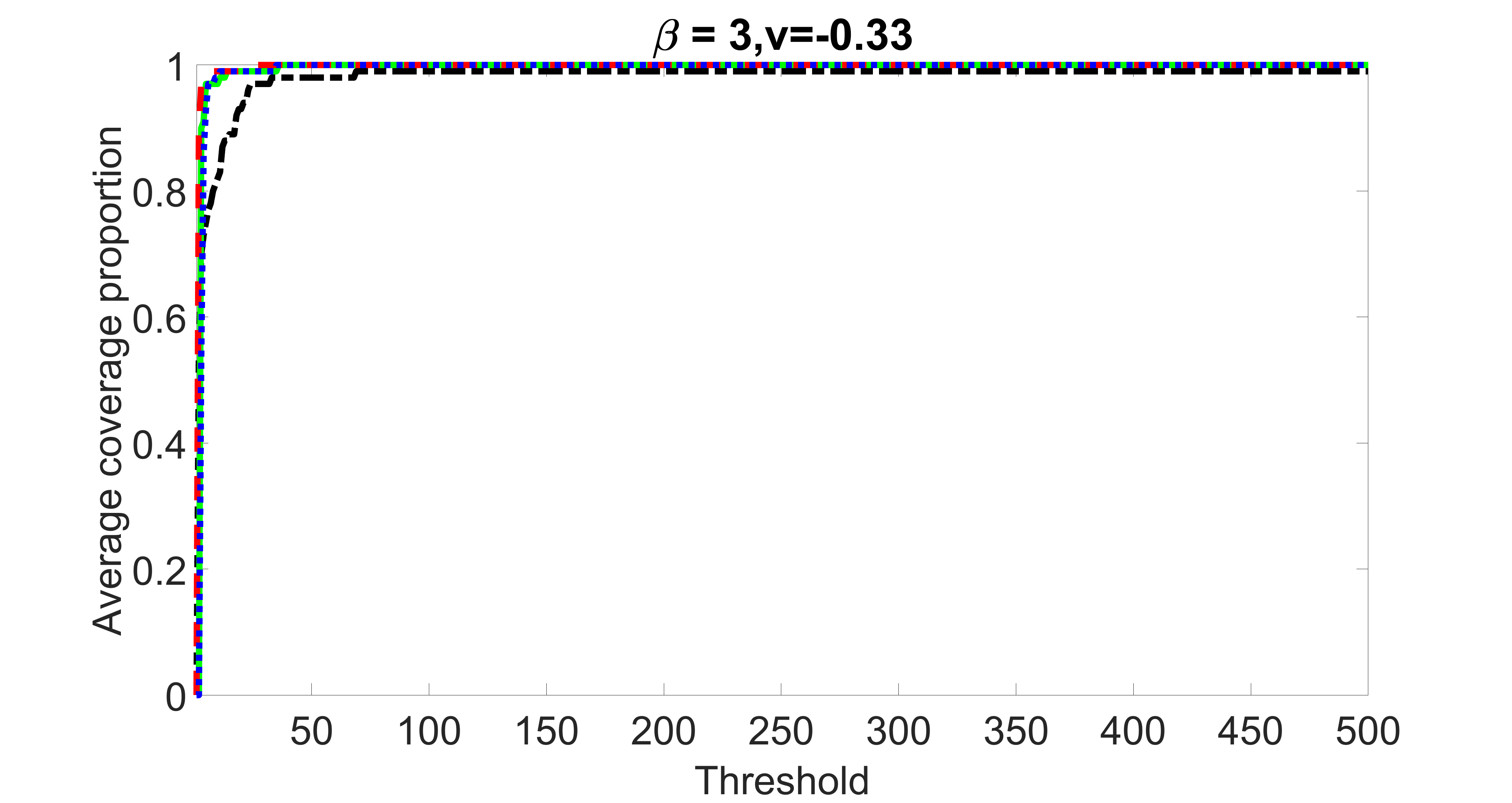}}
 \subcaptionbox{\footnotesize Confounder: medium \\ outcome, medium exposure}[0.45\linewidth]
 {\includegraphics[width=6cm,height=3.5cm]{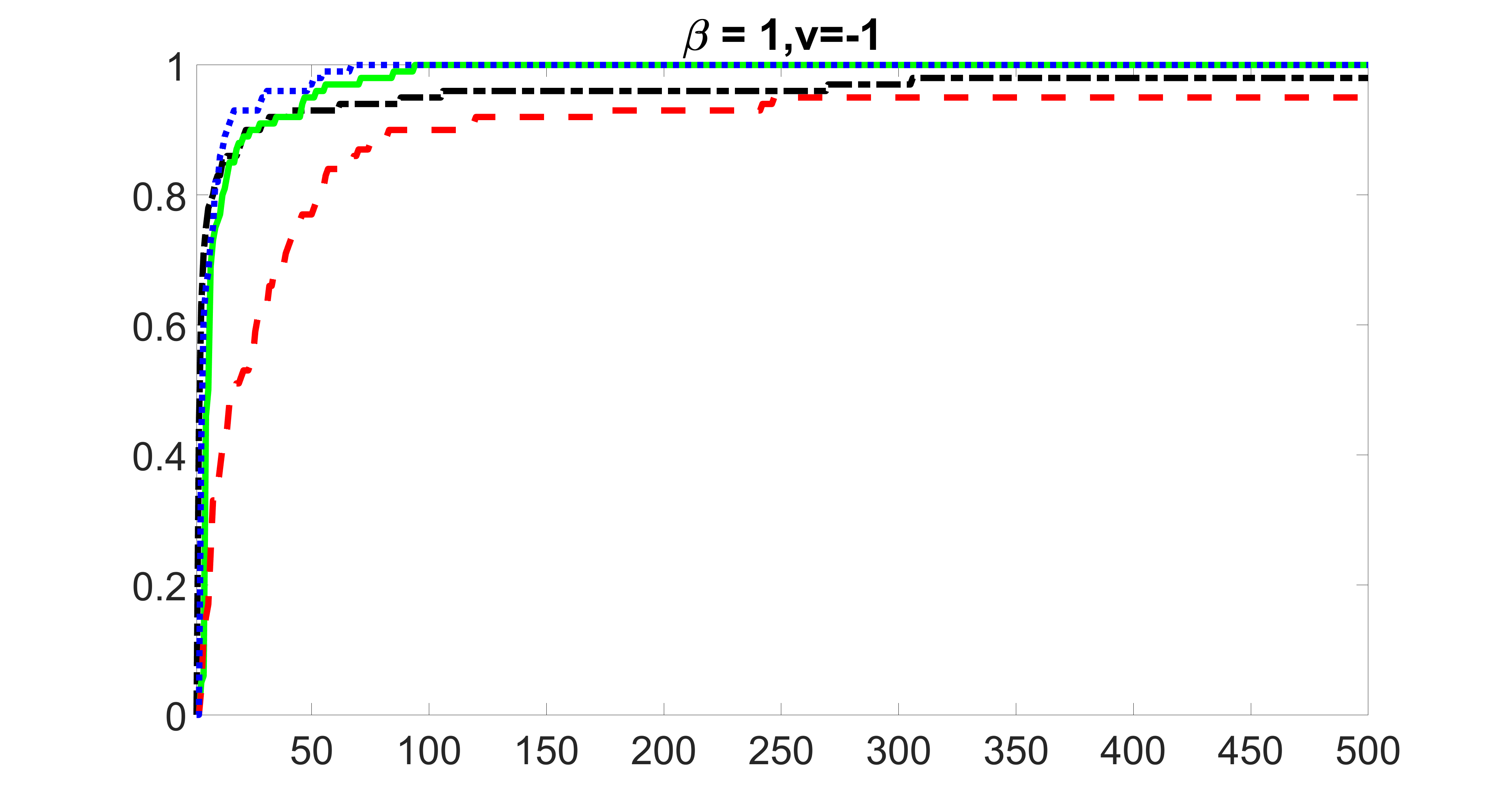}}
  \subcaptionbox{\footnotesize Confounder: weak \\ outcome, strong exposure}[0.45\linewidth]
 {\includegraphics[width=6cm,height=3.5cm]{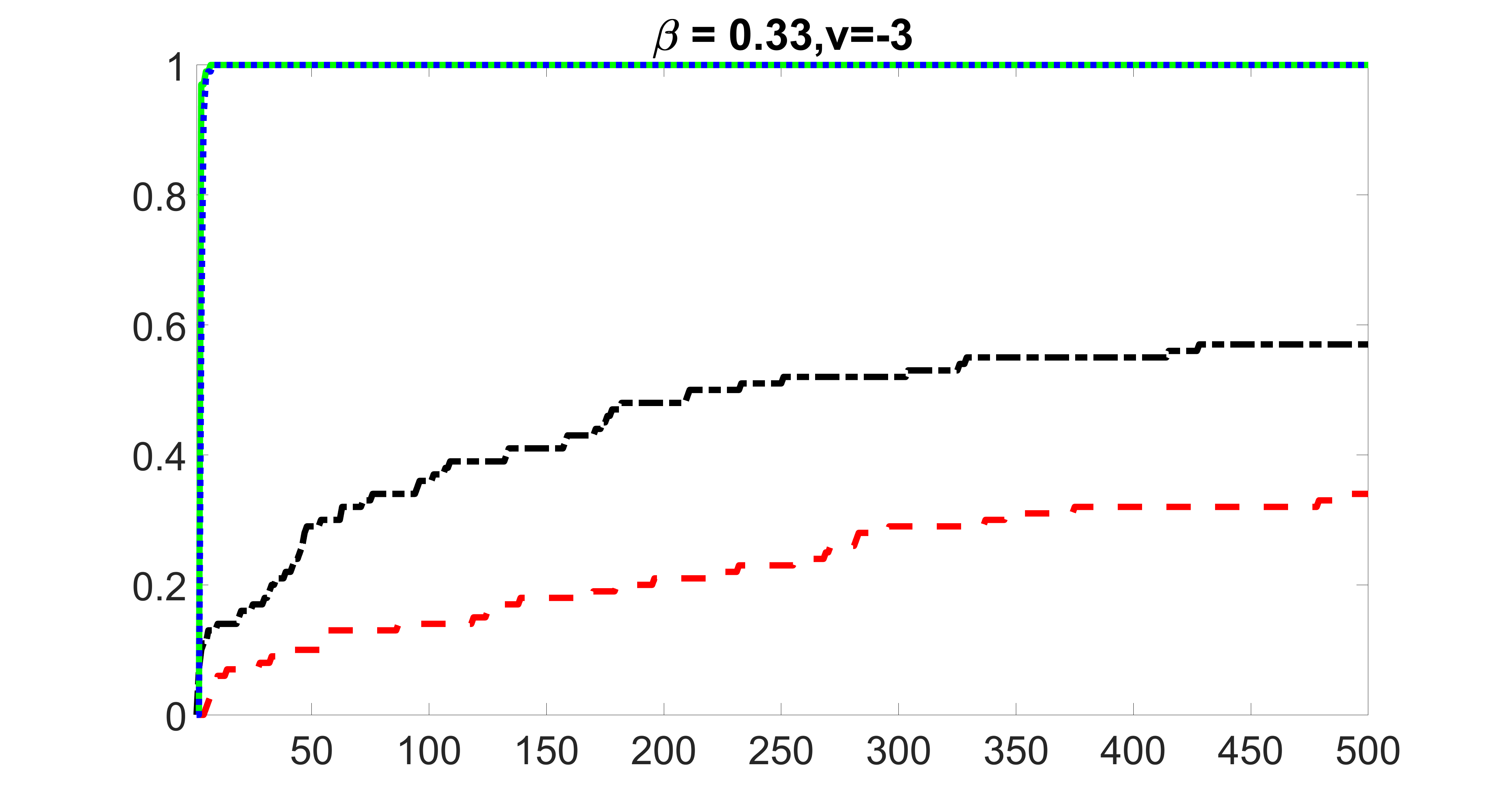}}
  \subcaptionbox{\footnotesize Precision: strong \\ outcome, zero exposure}[0.45\linewidth]
 {\includegraphics[width=6cm,height=3.5cm]{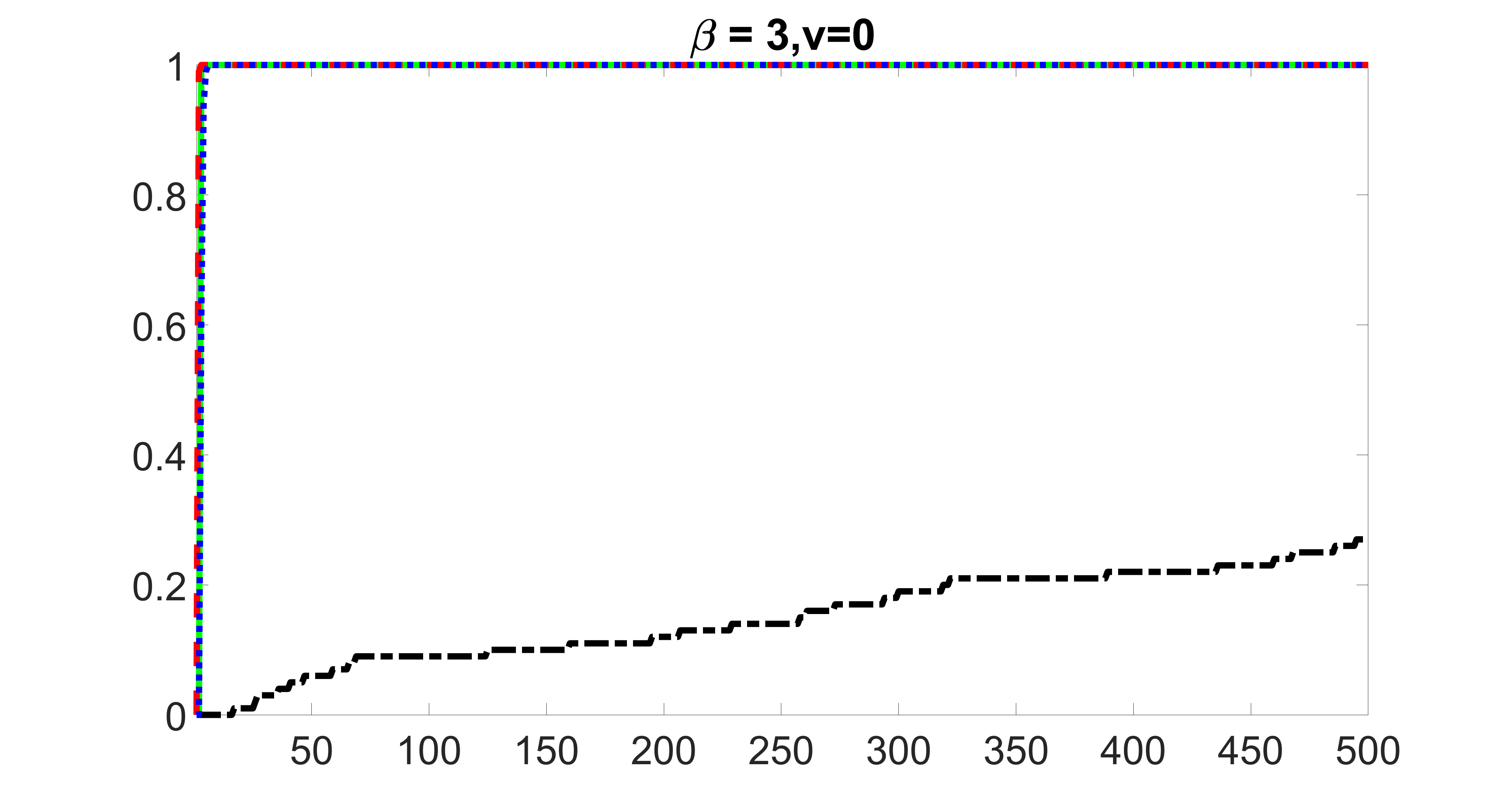}}
  \subcaptionbox{\footnotesize Precision: medium \\ outcome, zero exposure}[0.45\linewidth]
 {\includegraphics[width=6cm,height=3.5cm]{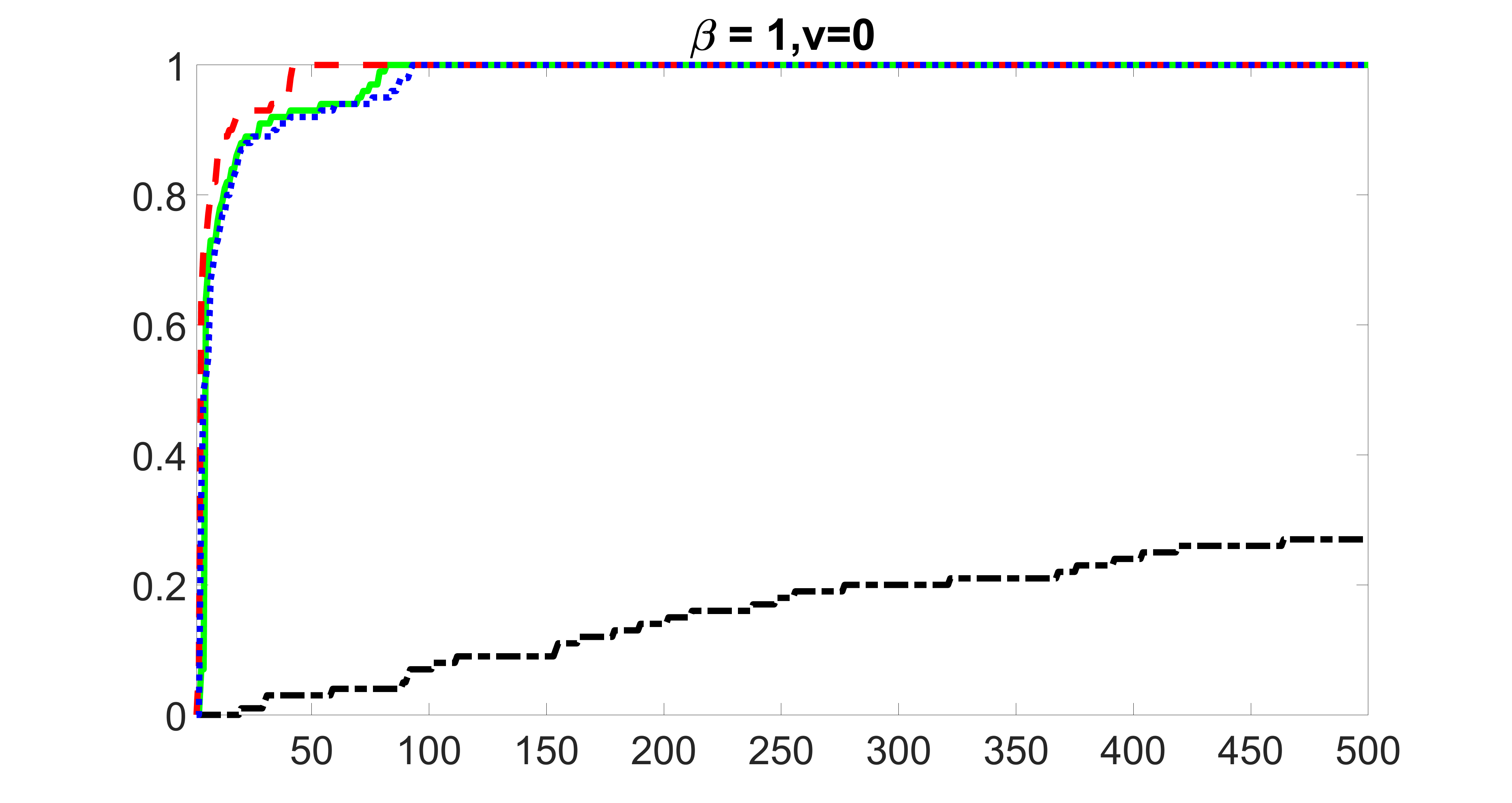}}
  \subcaptionbox{\footnotesize Precision: weak \\ outcome, zero exposure}[0.45\linewidth]
 {\includegraphics[width=6cm,height=3.5cm]{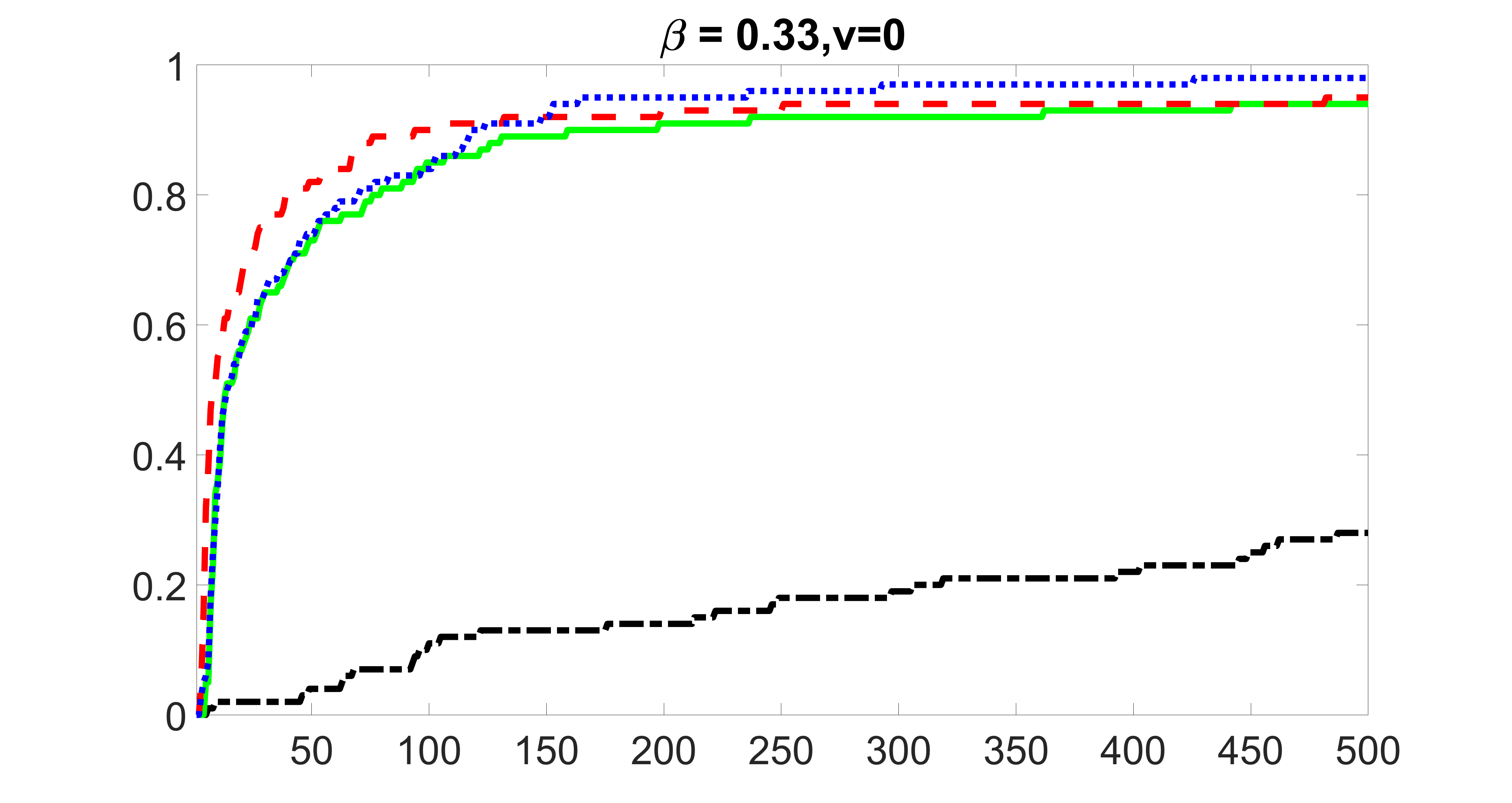}}
 \subcaptionbox{\footnotesize Precision: weaker \\ outcome, zero exposure}[0.45\linewidth]
 {\includegraphics[width=6cm,height=3.5cm]{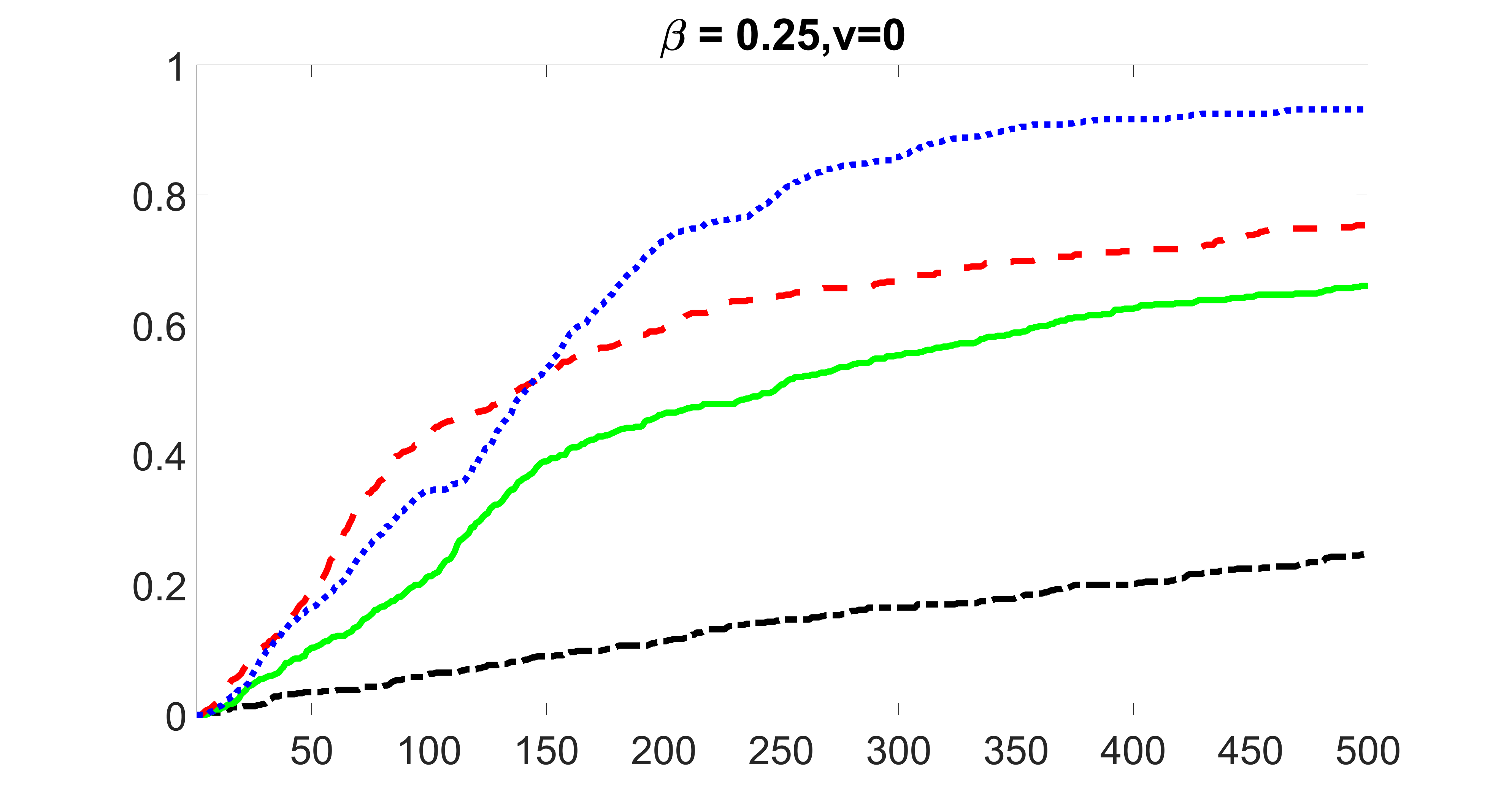}}
  \subcaptionbox{Overall coverage of $\mathcal{M}_1$}[0.45\linewidth]
 {\includegraphics[width=6cm,height=3.5cm]{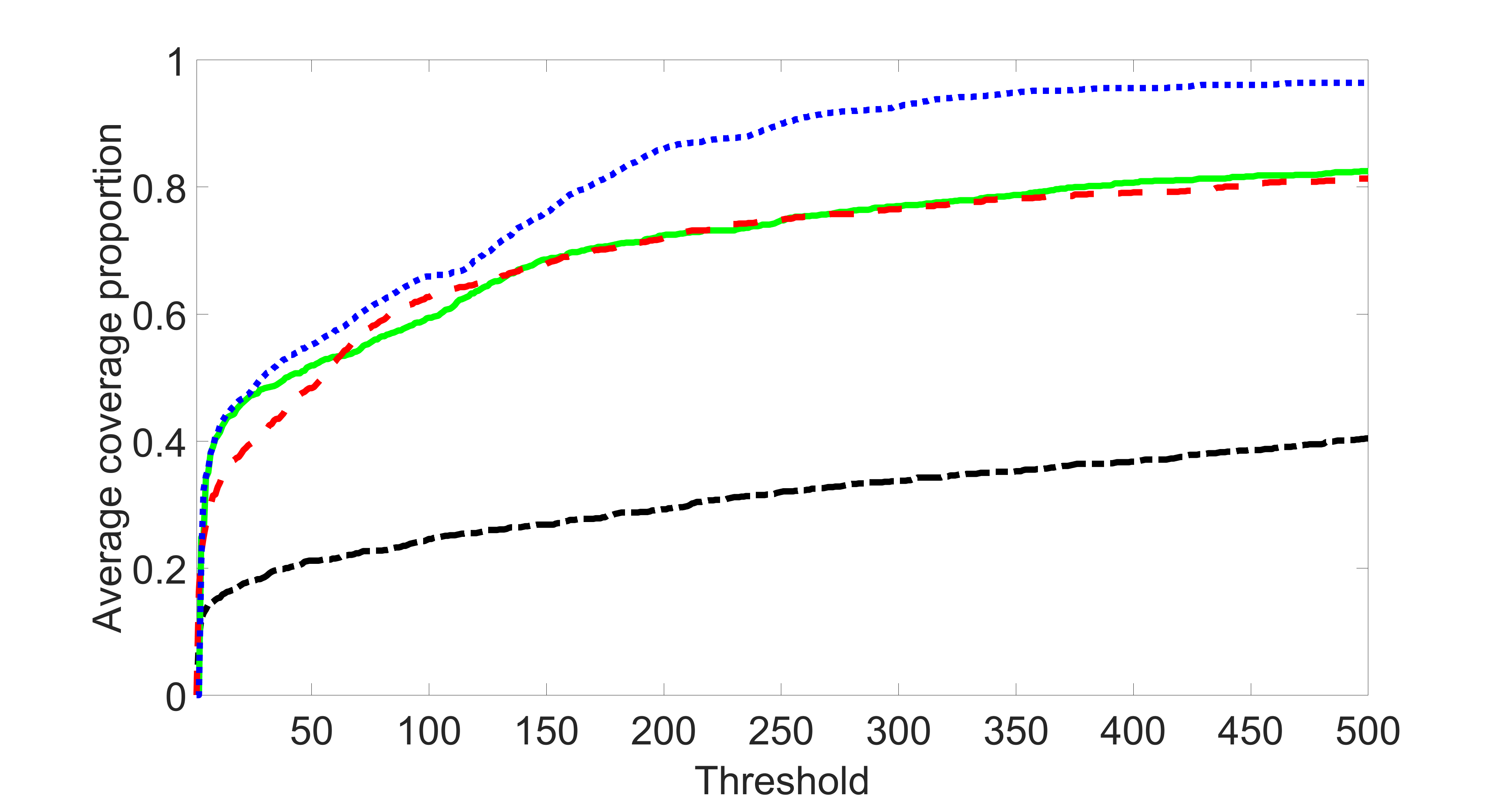}}
\caption{ Simulation results for the case $(n,s,K,\sigma) = (500,5000,6,1)$: Panels (a) -- (g) plot the average coverage proportion for $X_l$, where $l \in \mathcal{M}_1 =  \{1,2,3,104,105, 106\} \cup \mathcal{P}_{LD}$. Panels (a) -- (c) correspond to strong outcome and weak exposure predictor, moderate outcome and moderate exposure predictor and weak outcome and strong exposure predictor; Panels (d) -- (g) correspond to strong, moderate, and weak predictors of outcome only. Panel (g) plots the average coverage proportion for the index set $\mathcal{P}_{LD}$. Panel (h) plots the average coverage proportion for the index set $\mathcal{M}_1$. The x-axis represents the size of $\widehat{\mathcal{M}} $, while
y-axis denotes the average proportion. The blue dot, green solid, red dashed and black dash dotted lines denote the blockwise joint screening, joint screening, outcome screening, and intersection screening methods, respectively.}
\label{sim3step1n500sizesig6sigma1}
\end{figure}

\begin{figure}[htbp]
\captionsetup[subfigure]{justification=centering}
\centering
 \subcaptionbox{\footnotesize Confounder: strong \\ outcome, weak exposure}[0.45\linewidth]
 {\includegraphics[width=6cm,height=3.5cm]{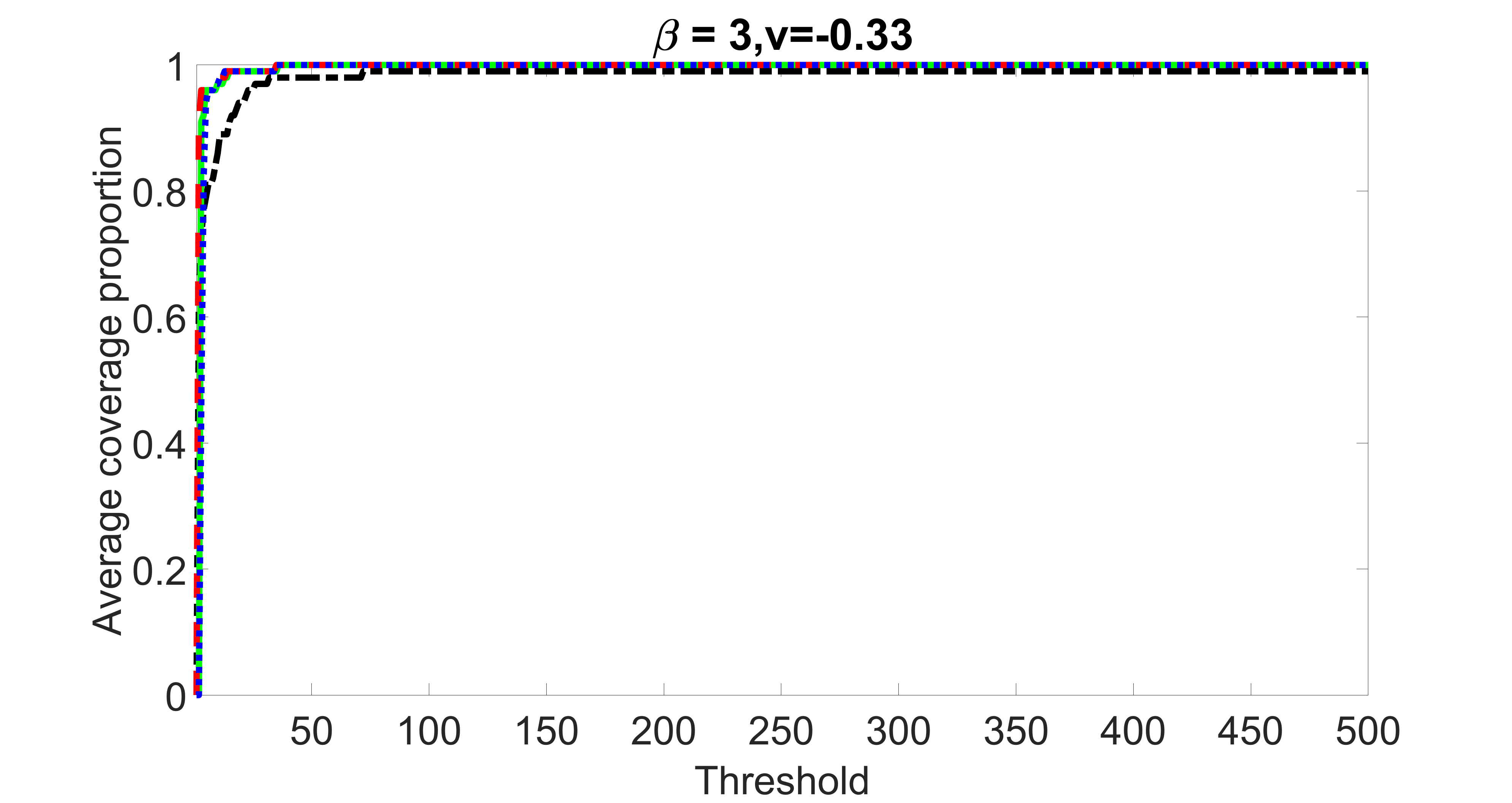}}
 \subcaptionbox{\footnotesize Confounder: medium \\ outcome, medium exposure}[0.45\linewidth]
 {\includegraphics[width=6cm,height=3.5cm]{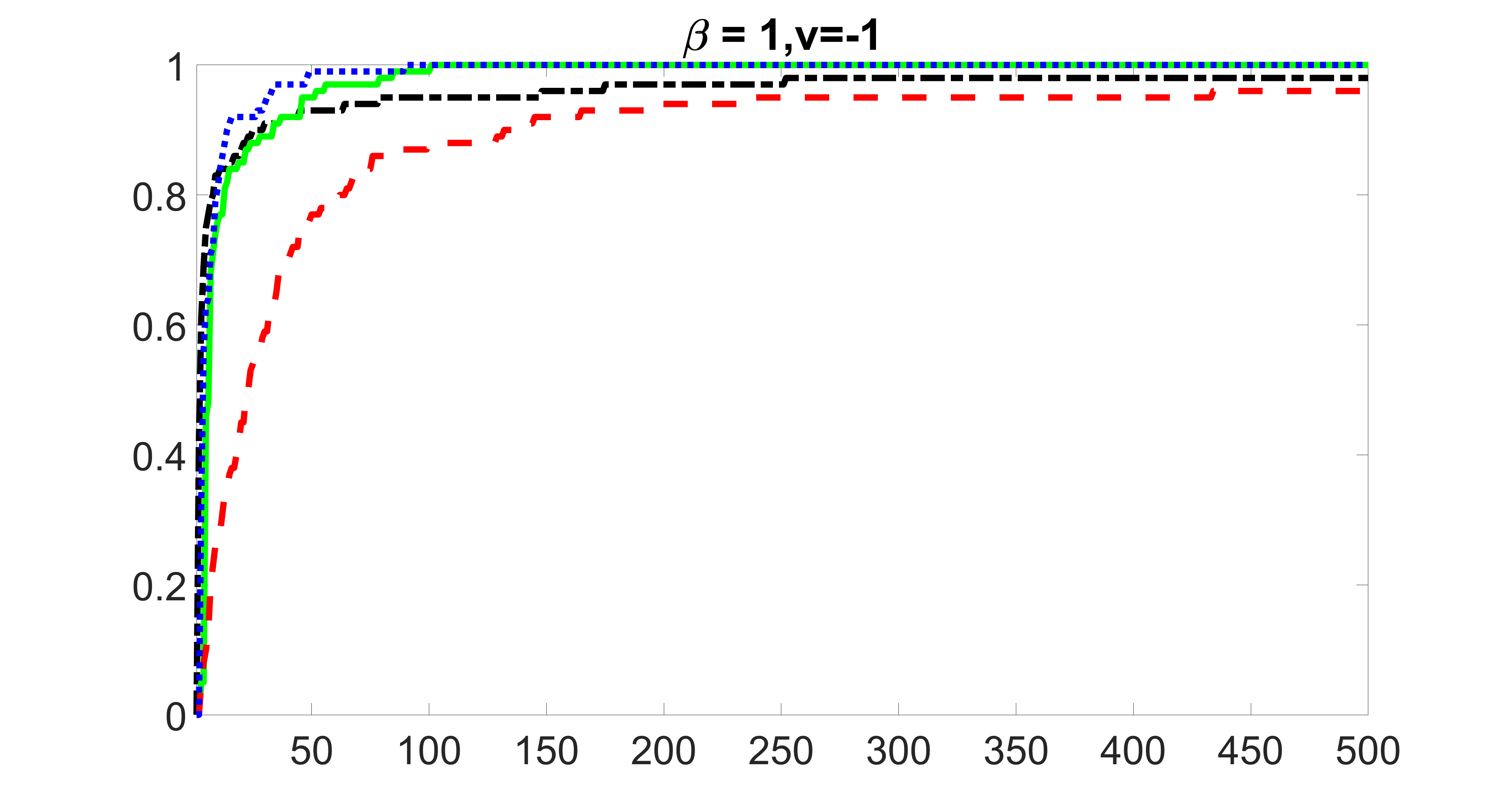}}
  \subcaptionbox{\footnotesize Confounder: weak \\ outcome, strong exposure}[0.45\linewidth]
 {\includegraphics[width=6cm,height=3.5cm]{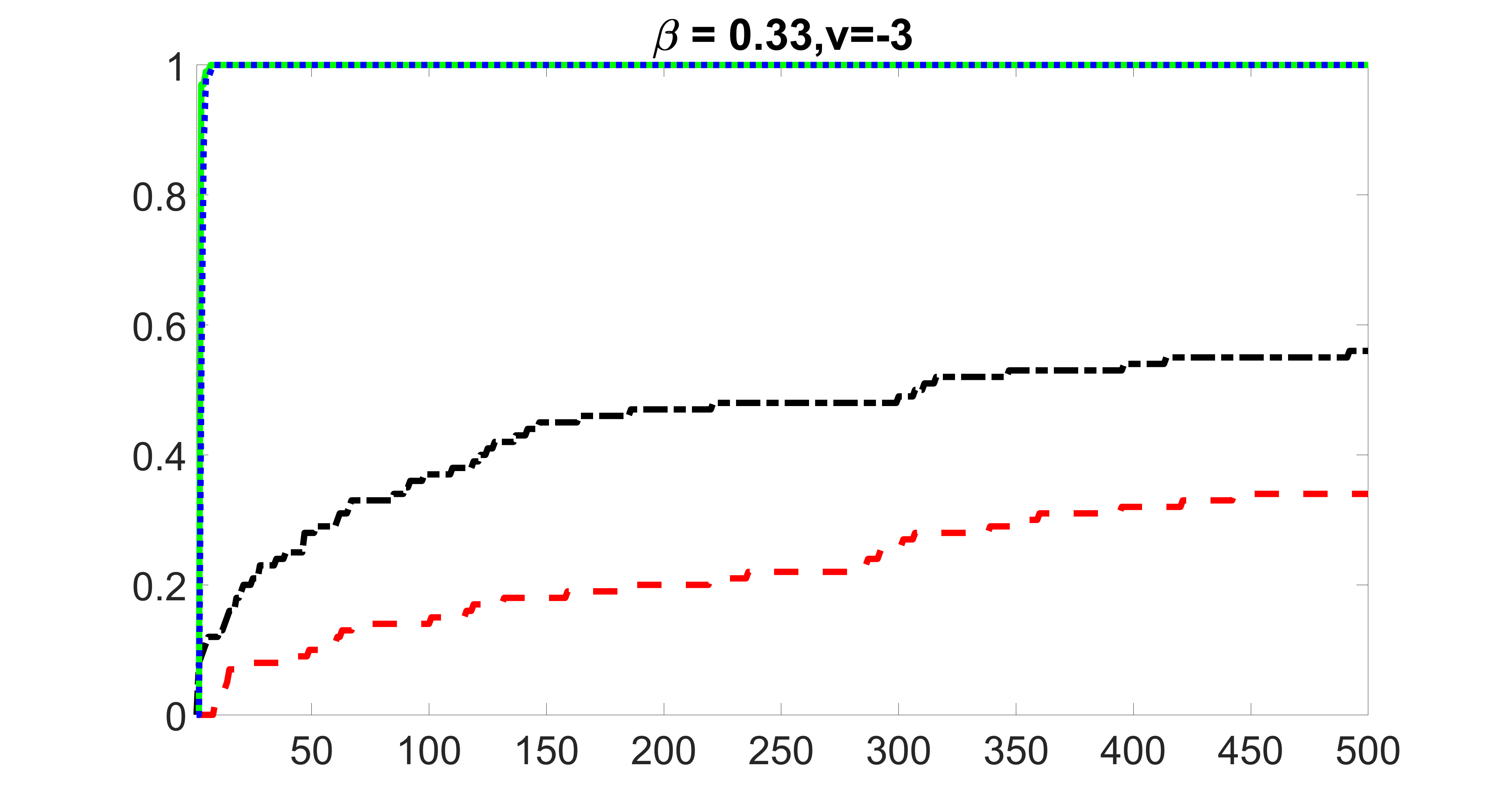}}
  \subcaptionbox{\footnotesize Precision: strong \\ outcome, zero exposure}[0.45\linewidth]
 {\includegraphics[width=6cm,height=3.5cm]{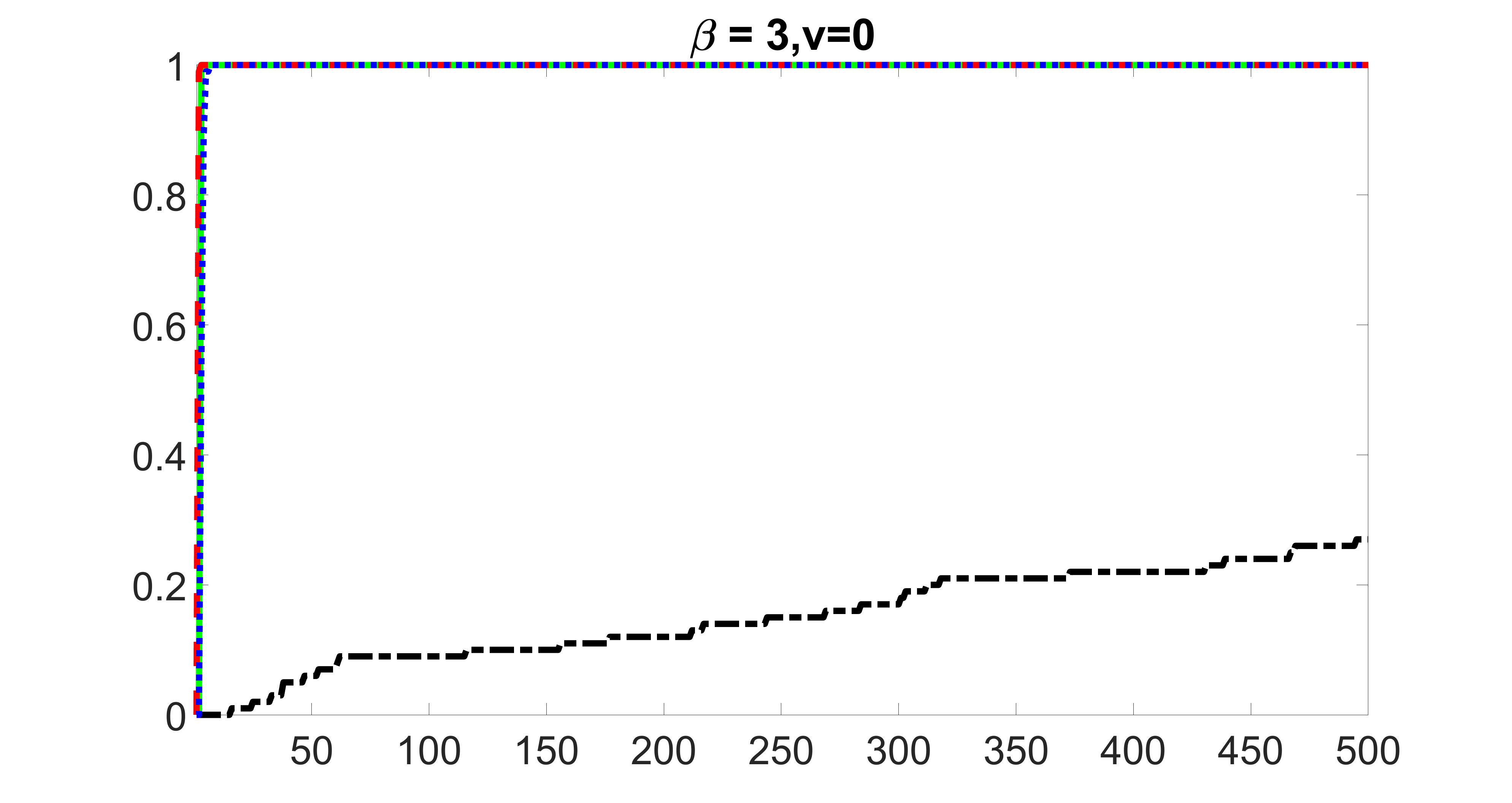}}
  \subcaptionbox{\footnotesize Precision: medium \\ outcome, zero exposure}[0.45\linewidth]
 {\includegraphics[width=6cm,height=3.5cm]{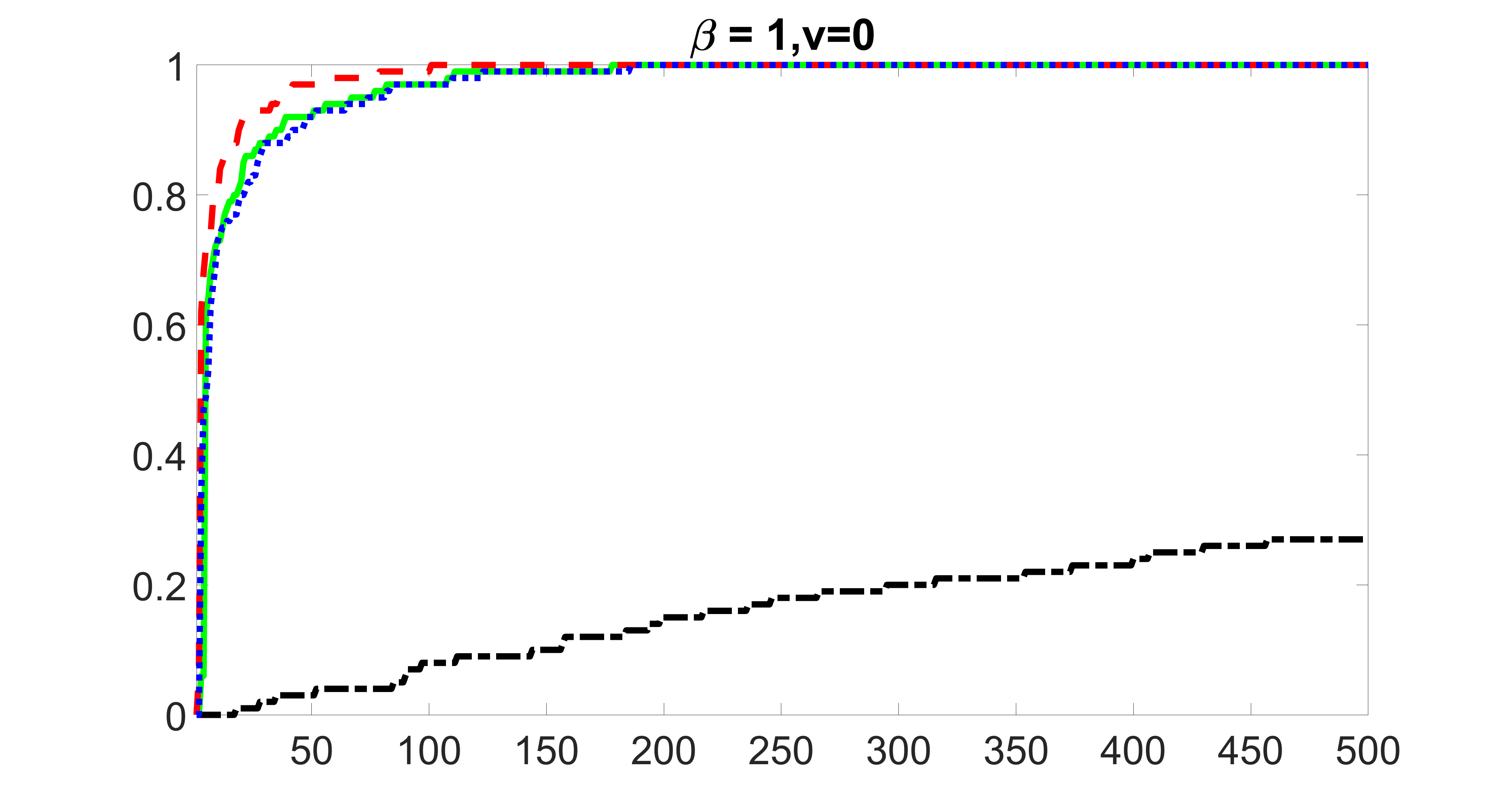}}
  \subcaptionbox{\footnotesize Precision: weak \\ outcome, zero exposure}[0.45\linewidth]
 {\includegraphics[width=6cm,height=3.5cm]{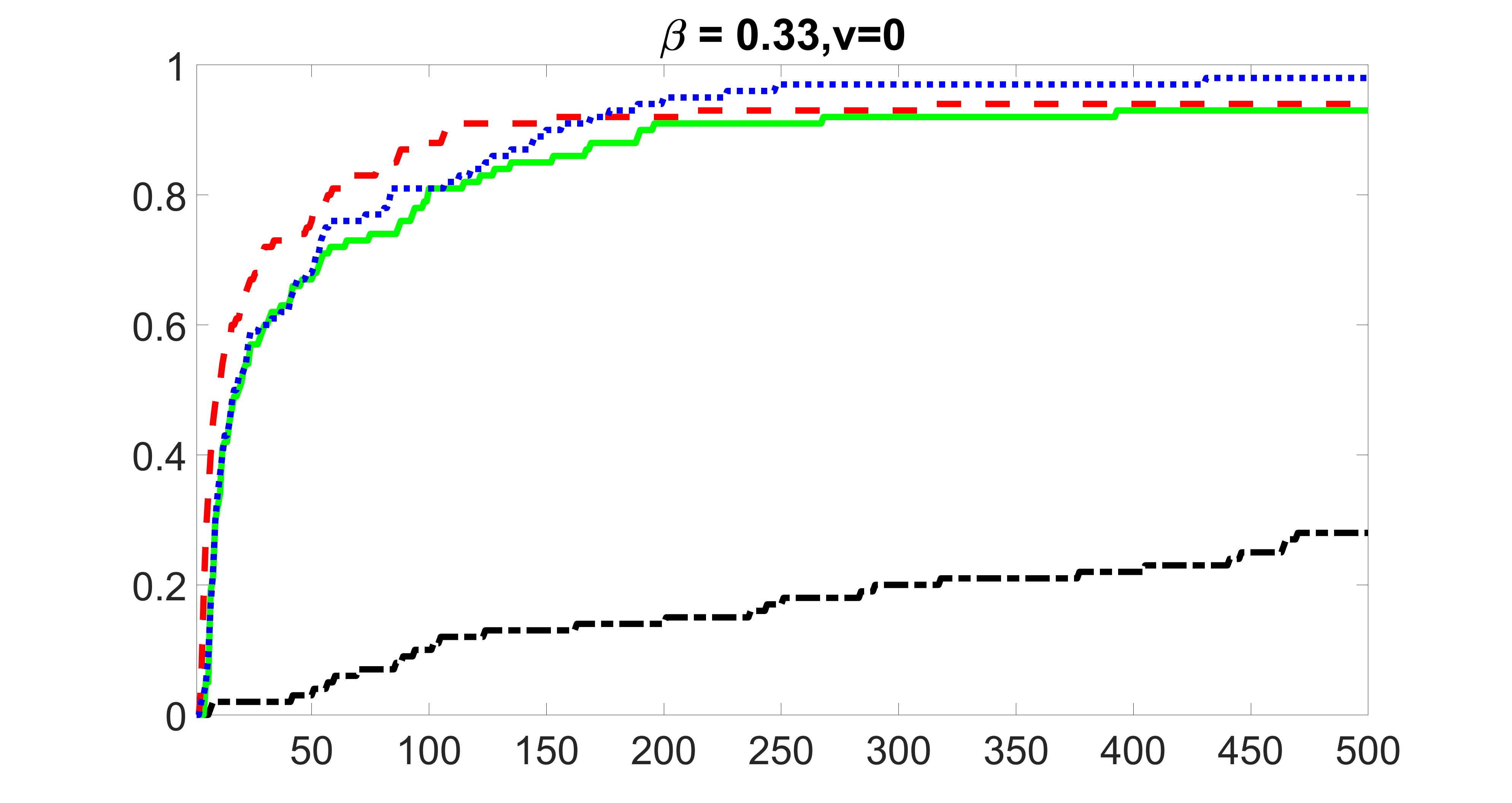}}
 \subcaptionbox{\footnotesize Precision: weaker \\ outcome, zero exposure}[0.45\linewidth]
 {\includegraphics[width=6cm,height=3.5cm]{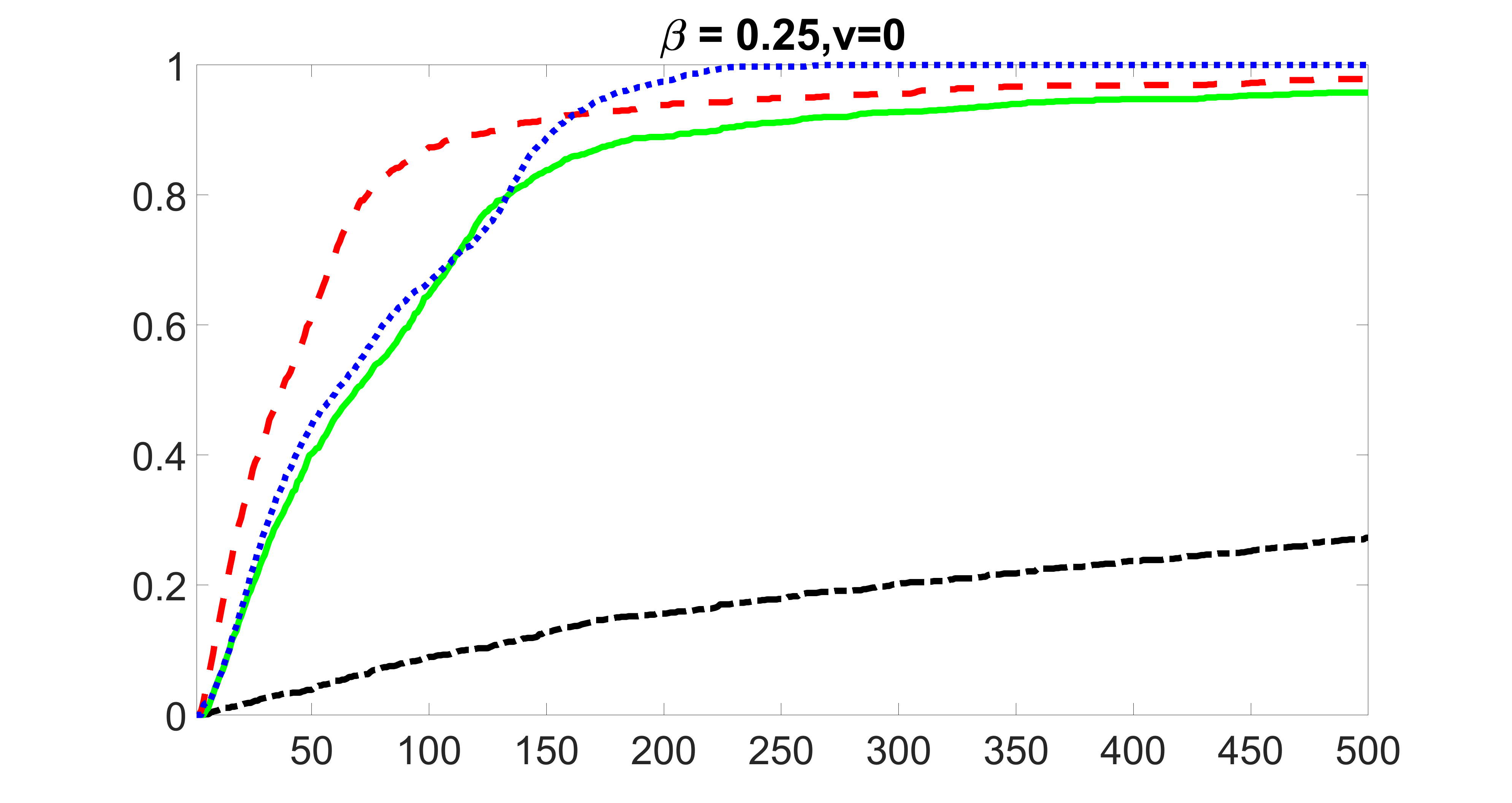}}
  \subcaptionbox{Overall coverage of $\mathcal{M}_1$}[0.45\linewidth]
 {\includegraphics[width=6cm,height=3.5cm]{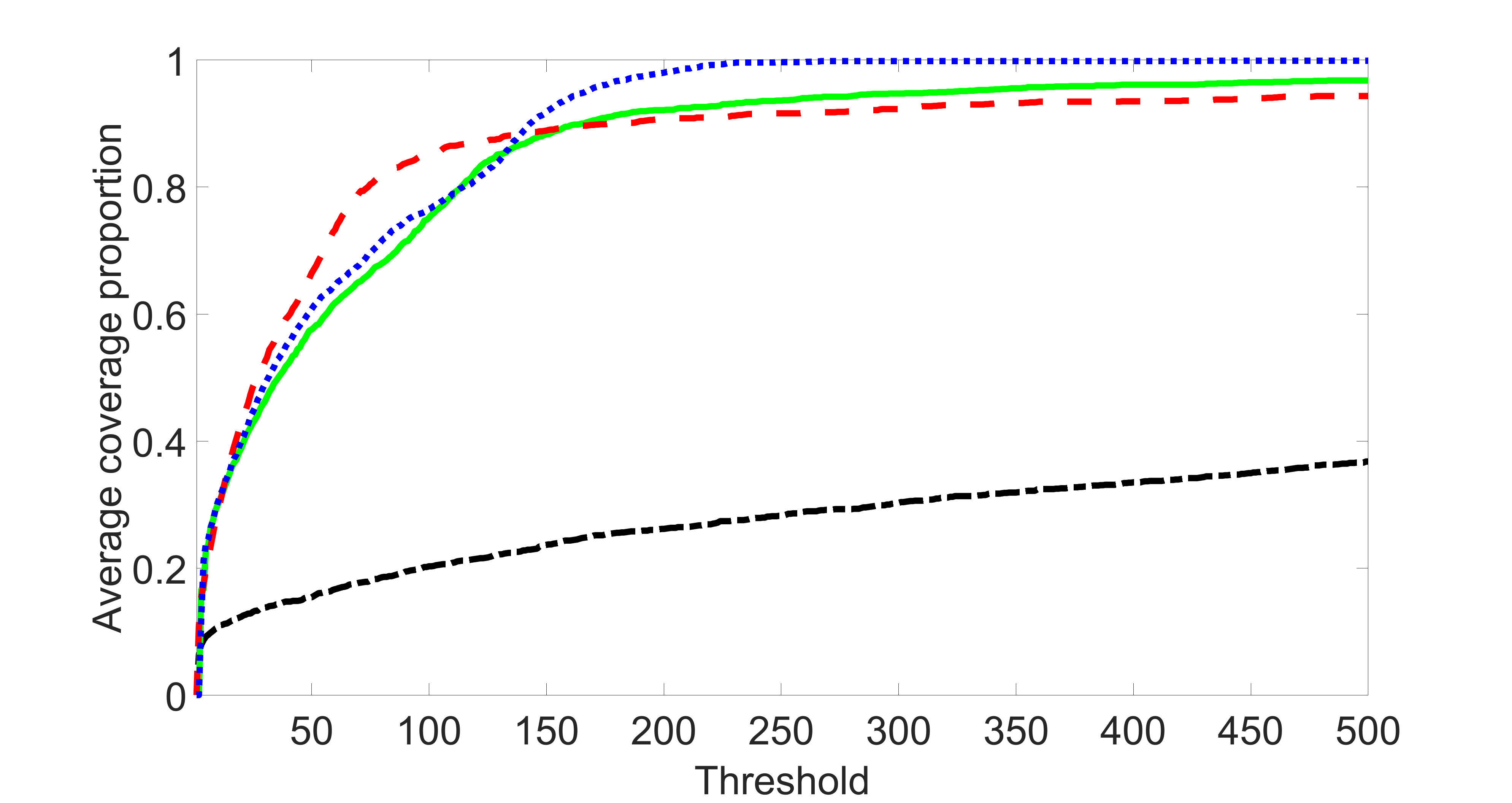}}
\caption{ Simulation results for the case $(n,s,K,\sigma) = (500,5000,12,1)$: Panels (a) -- (g) plot the average coverage proportion for $X_l$, where $l \in \mathcal{M}_1 =  \{1,2,3,104,105, 106\} \cup \mathcal{P}_{LD}$. Panels (a) -- (c) correspond to strong outcome and weak exposure predictor, moderate outcome and moderate exposure predictor and weak outcome and strong exposure predictor; Panels (d) -- (g) correspond to strong, moderate, and weak predictors of outcome only. Panel (g) plots the average coverage proportion for the index set $\mathcal{P}_{LD}$. Panel (h) plots the average coverage proportion for the index set $\mathcal{M}_1$. The x-axis represents the size of $\widehat{\mathcal{M}} $, while
y-axis denotes the average proportion. The blue dot, green solid, red dashed and black dash dotted lines denote the blockwise joint screening, joint screening, outcome screening, and intersection screening methods, respectively.}
\label{sim3step1n500sizesig12sigma1}
\end{figure}

\begin{figure}[htbp]
\captionsetup[subfigure]{justification=centering}
\centering
 \subcaptionbox{\footnotesize Confounder: strong \\ outcome, weak exposure}[0.45\linewidth]
 {\includegraphics[width=6cm,height=3.5cm]{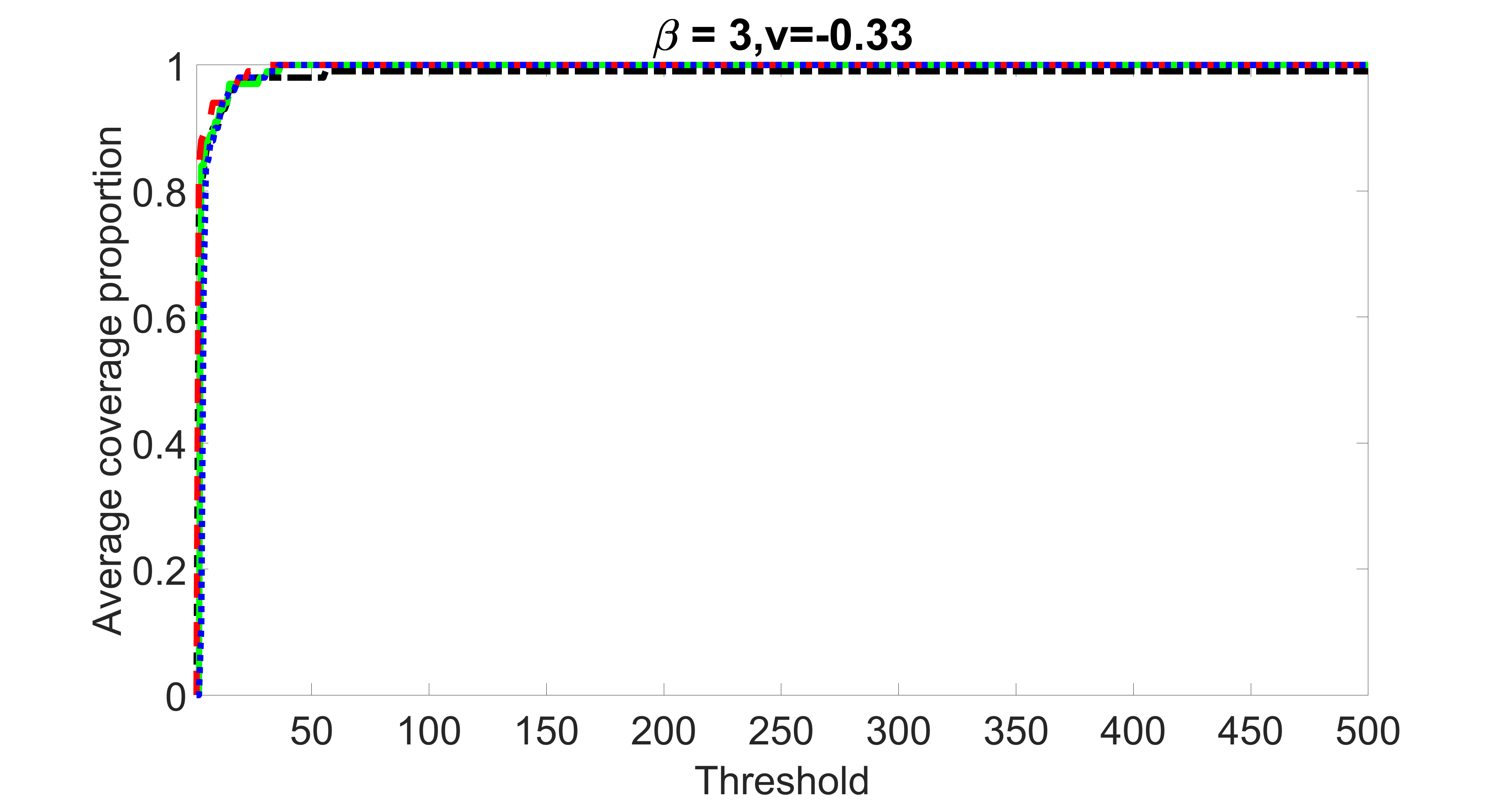}}
 \subcaptionbox{\footnotesize Confounder: medium \\ outcome, medium exposure}[0.45\linewidth]
 {\includegraphics[width=6cm,height=3.5cm]{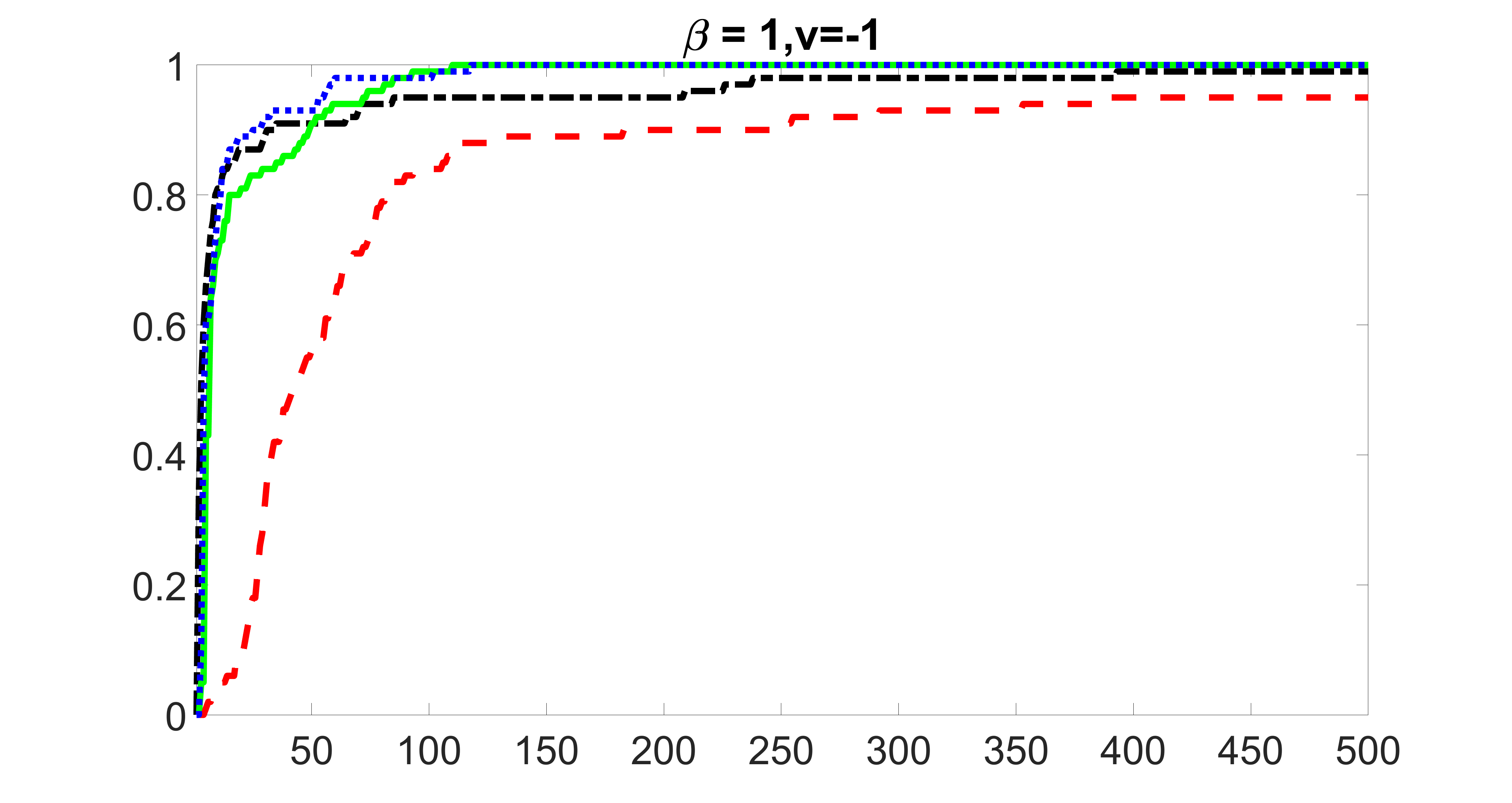}}
  \subcaptionbox{\footnotesize Confounder: weak \\ outcome, strong exposure}[0.45\linewidth]
 {\includegraphics[width=6cm,height=3.5cm]{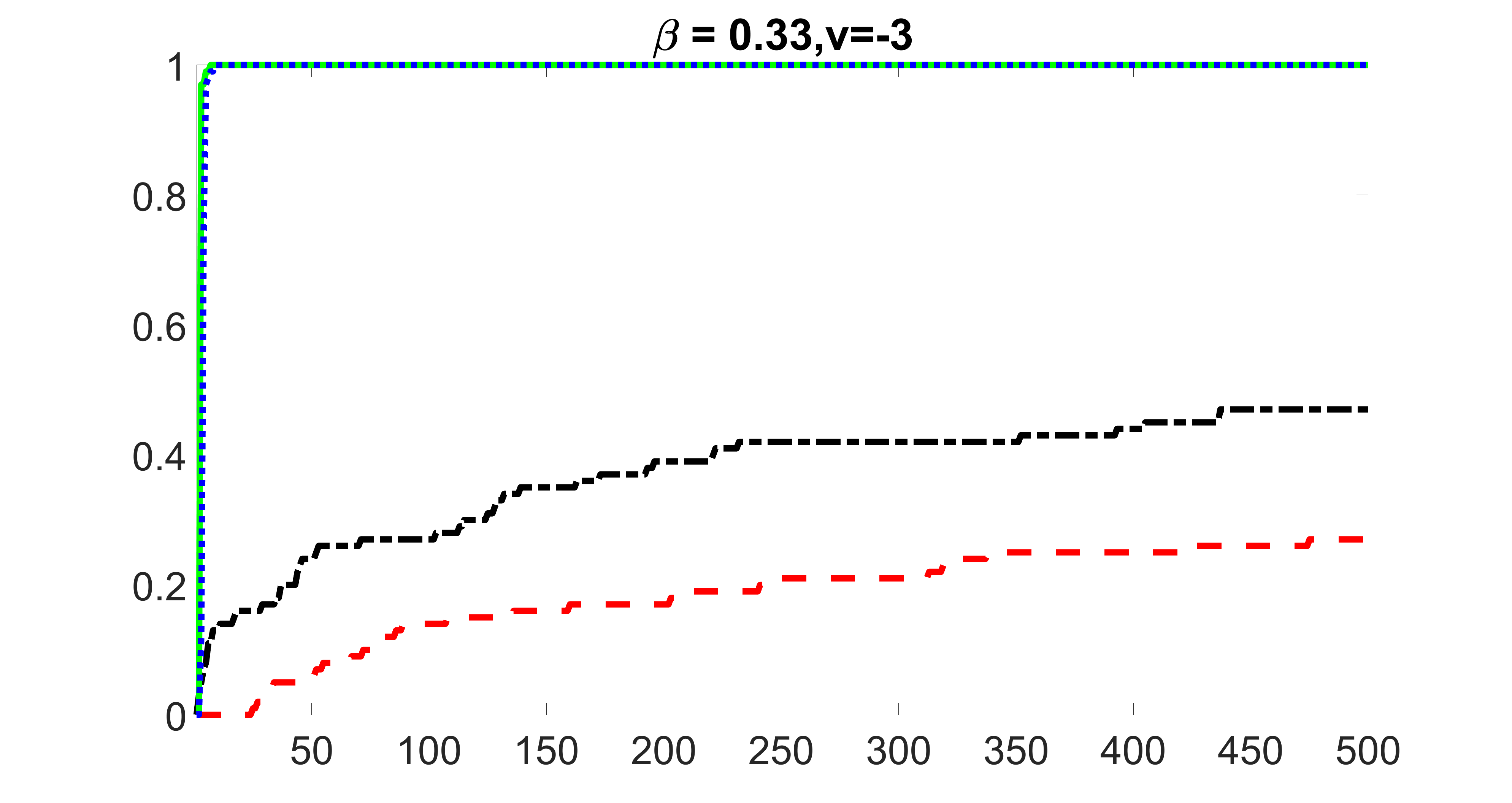}}
  \subcaptionbox{\footnotesize Precision: strong \\ outcome, zero exposure}[0.45\linewidth]
 {\includegraphics[width=6cm,height=3.5cm]{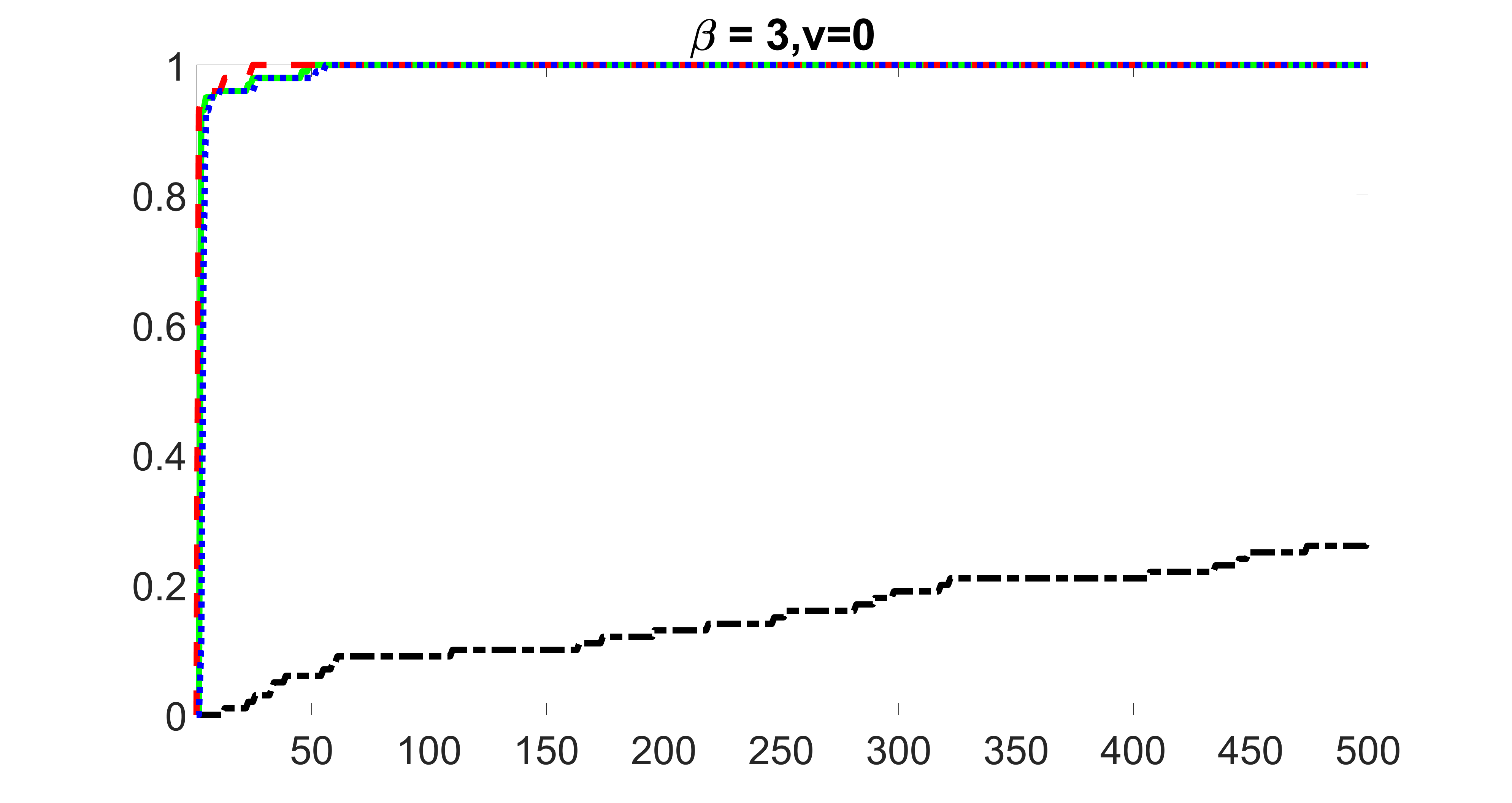}}
  \subcaptionbox{\footnotesize Precision: medium \\ outcome, zero exposure}[0.45\linewidth]
 {\includegraphics[width=6cm,height=3.5cm]{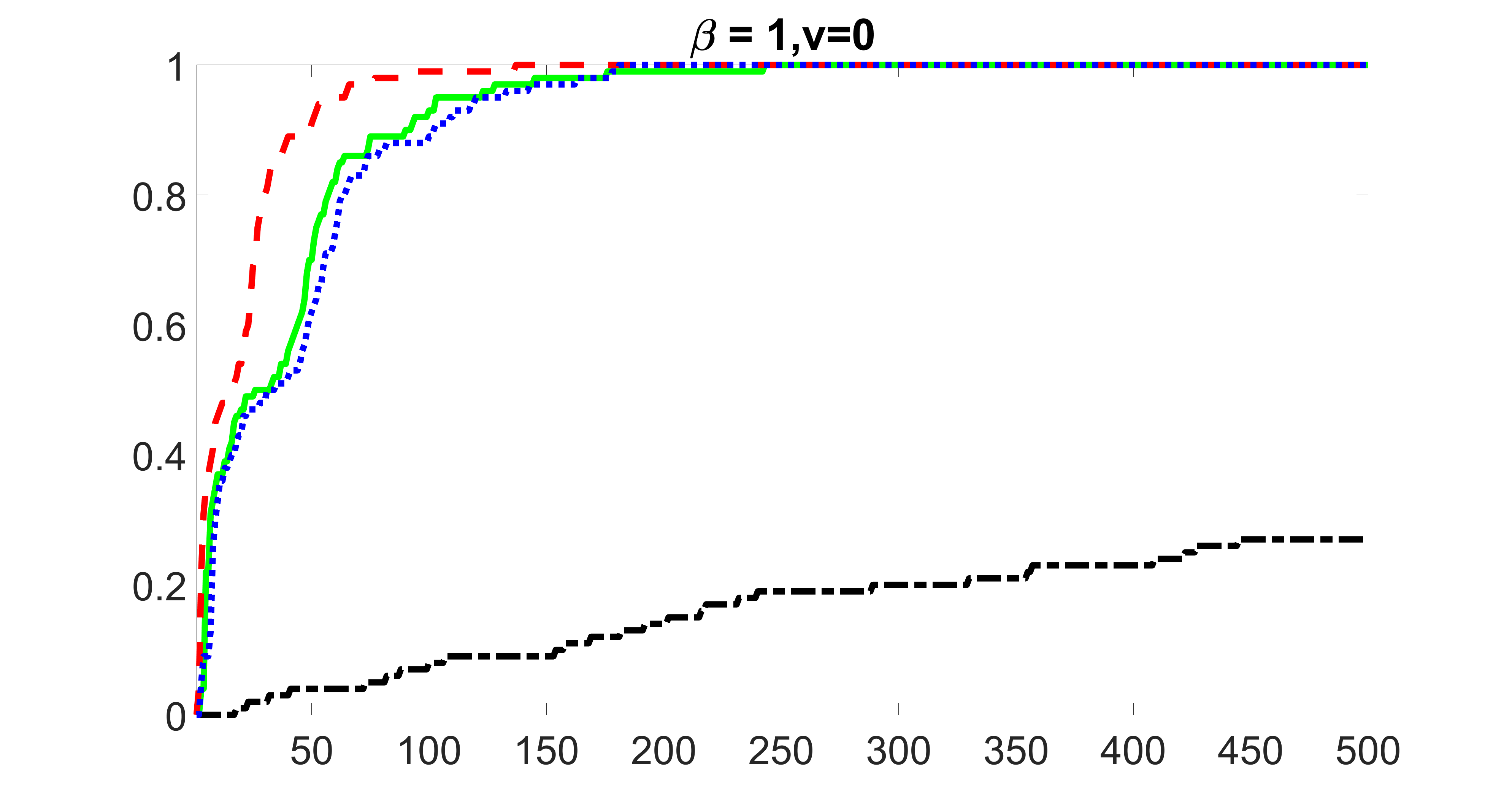}}
  \subcaptionbox{\footnotesize Precision: weak \\ outcome, zero exposure}[0.45\linewidth]
 {\includegraphics[width=6cm,height=3.5cm]{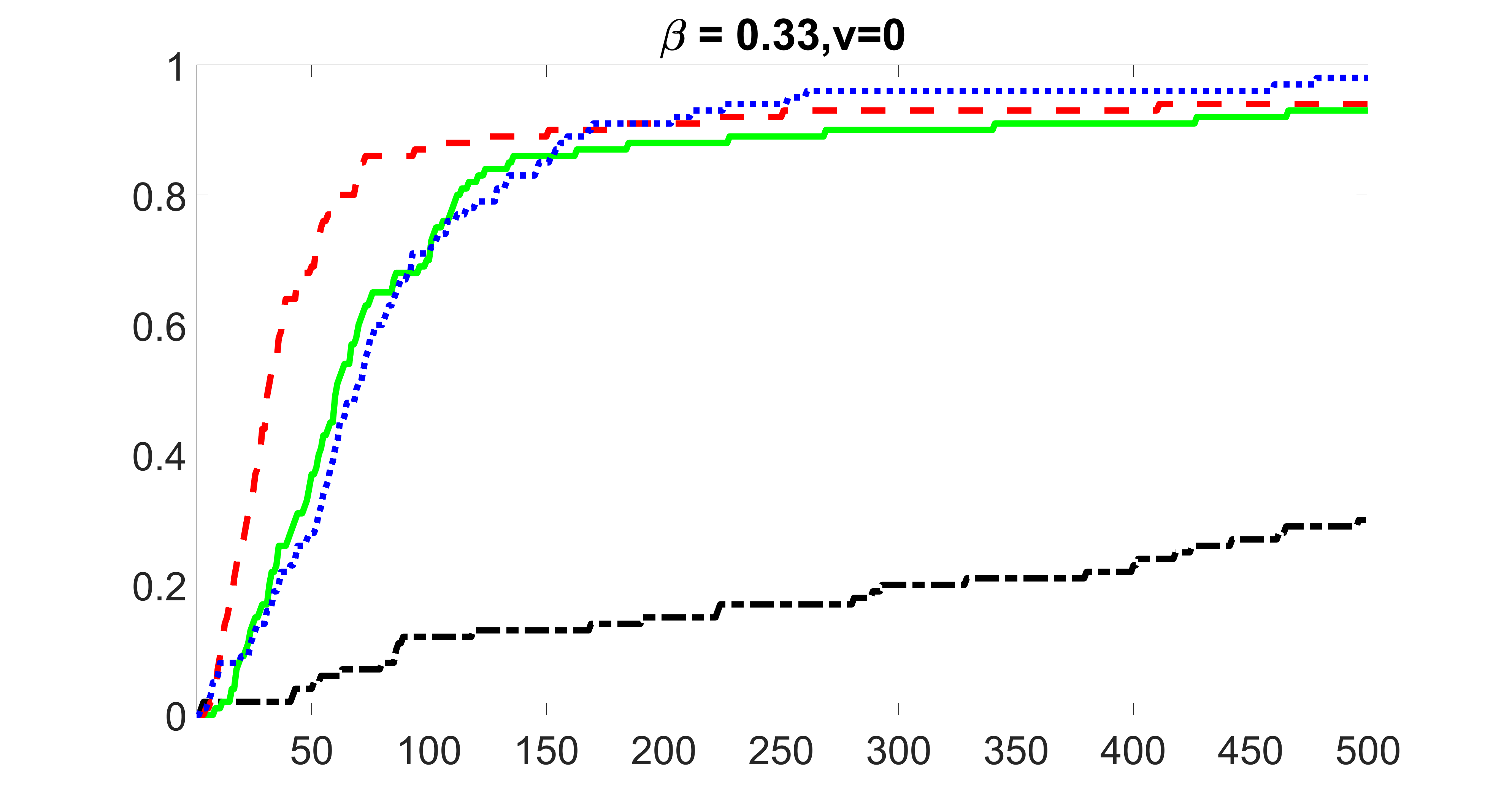}}
 \subcaptionbox{\footnotesize Precision: weaker \\ outcome, zero exposure}[0.45\linewidth]
 {\includegraphics[width=6cm,height=3.5cm]{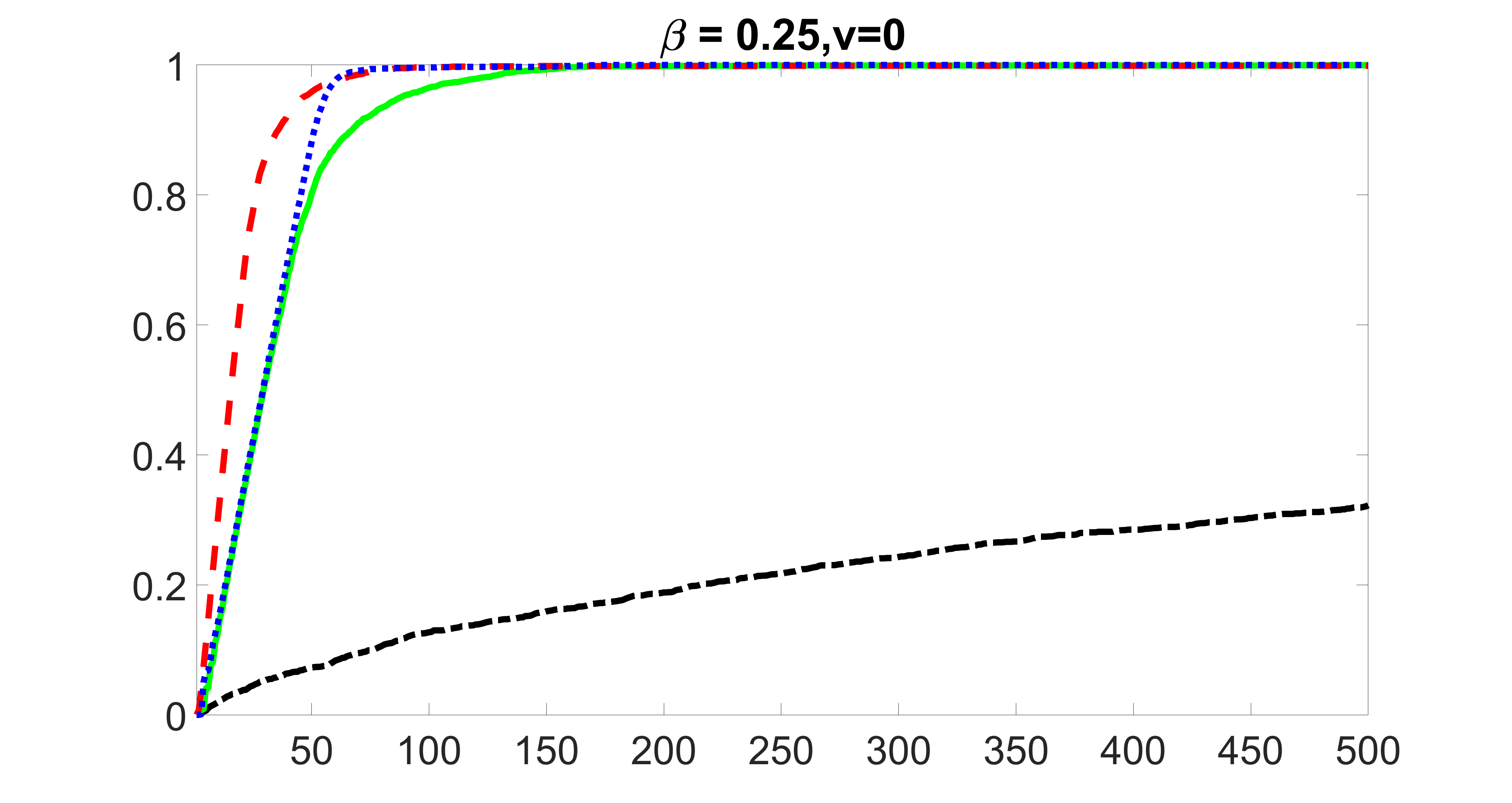}}
  \subcaptionbox{Overall coverage of $\mathcal{M}_1$}[0.45\linewidth]
 {\includegraphics[width=6cm,height=3.5cm]{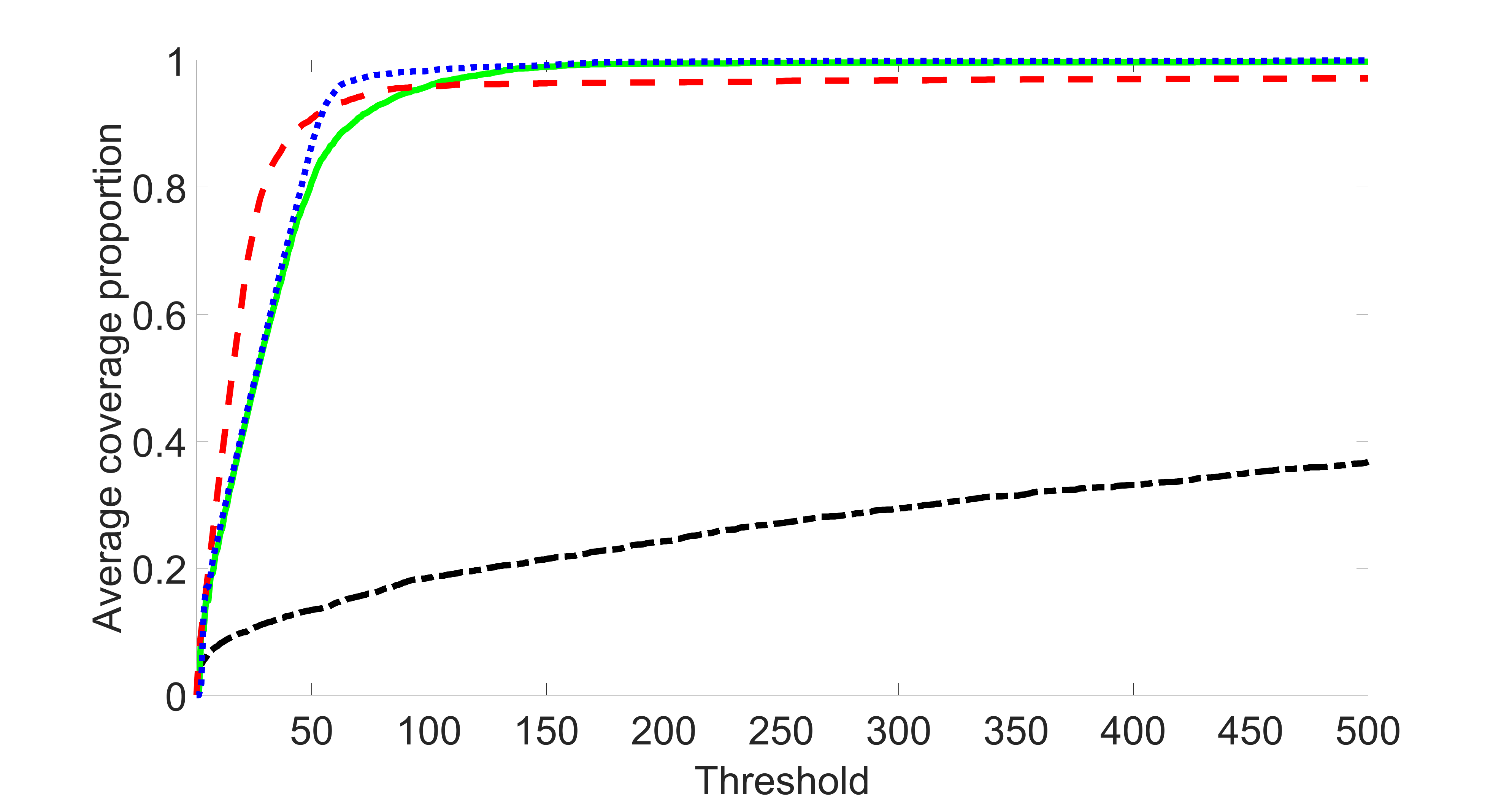}}
\caption{ Simulation results for the case $(n,s,K,\sigma) = (500,5000,24,1)$: Panels (a) -- (g) plot the average coverage proportion for $X_l$, where $l \in \mathcal{M}_1 =  \{1,2,3,104,105, 106\} \cup \mathcal{P}_{LD}$. Panels (a) -- (c) correspond to strong outcome and weak exposure predictor, moderate outcome and moderate exposure predictor and weak outcome and strong exposure predictor; Panels (d) -- (g) correspond to strong, moderate, and weak predictors of outcome only. Panel (g) plots the average coverage proportion for the index set $\mathcal{P}_{LD}$. Panel (h) plots the average coverage proportion for the index set $\mathcal{M}_1$. The x-axis represents the size of $\widehat{\mathcal{M}} $, while
y-axis denotes the average proportion. The blue dot, green solid, red dashed and black dash dotted lines denote the blockwise joint screening, joint screening, outcome screening, and intersection screening methods, respectively.}
\label{sim3step1n500sizesig24sigma1}
\end{figure}

\begin{figure}[htbp]
\captionsetup[subfigure]{justification=centering}
\centering
 \subcaptionbox{\footnotesize Confounder: strong \\ outcome, weak exposure}[0.45\linewidth]
 {\includegraphics[width=6cm,height=3.5cm]{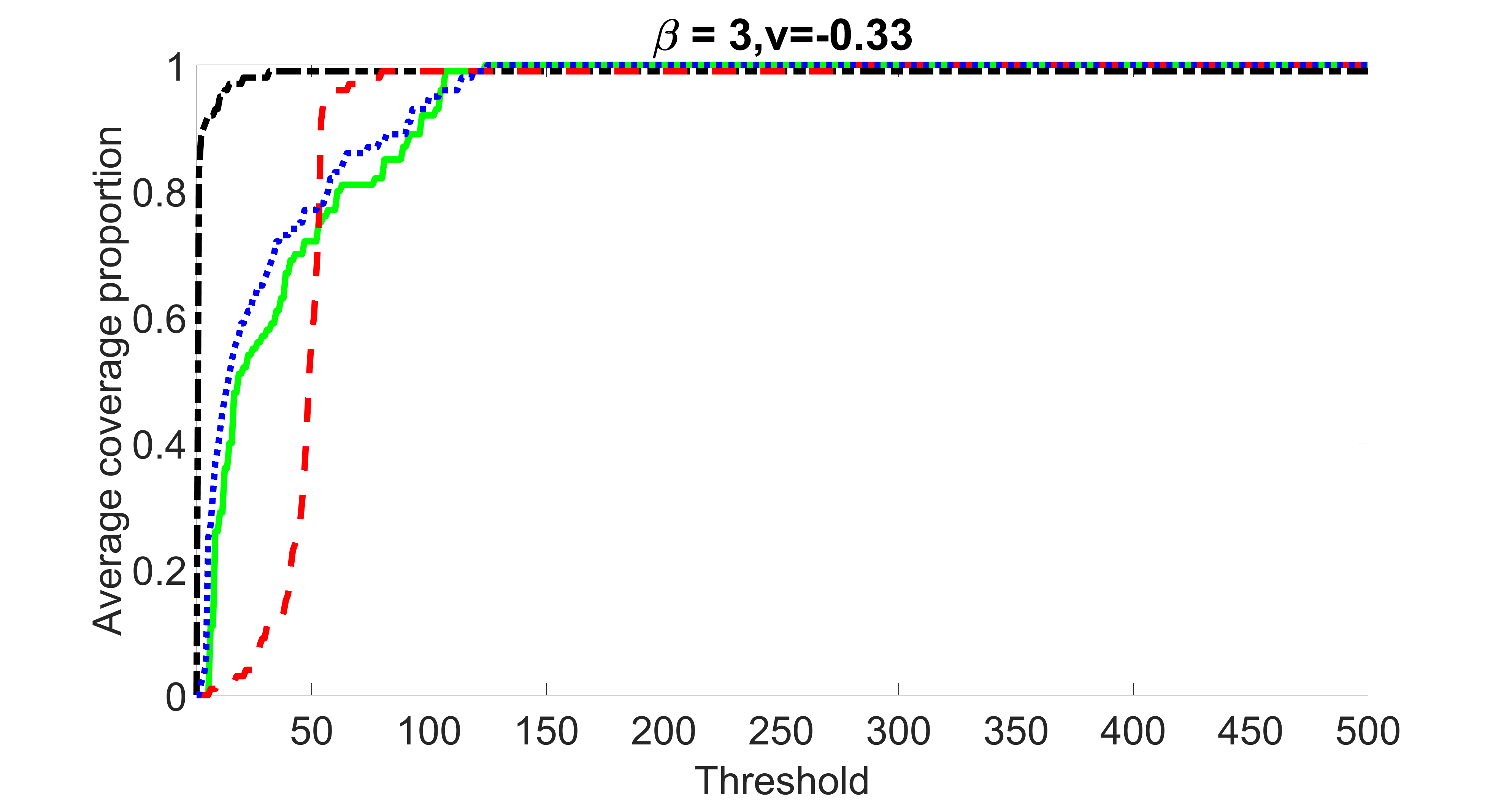}}
 \subcaptionbox{\footnotesize Confounder: medium \\ outcome, medium exposure}[0.45\linewidth]
 {\includegraphics[width=6cm,height=3.5cm]{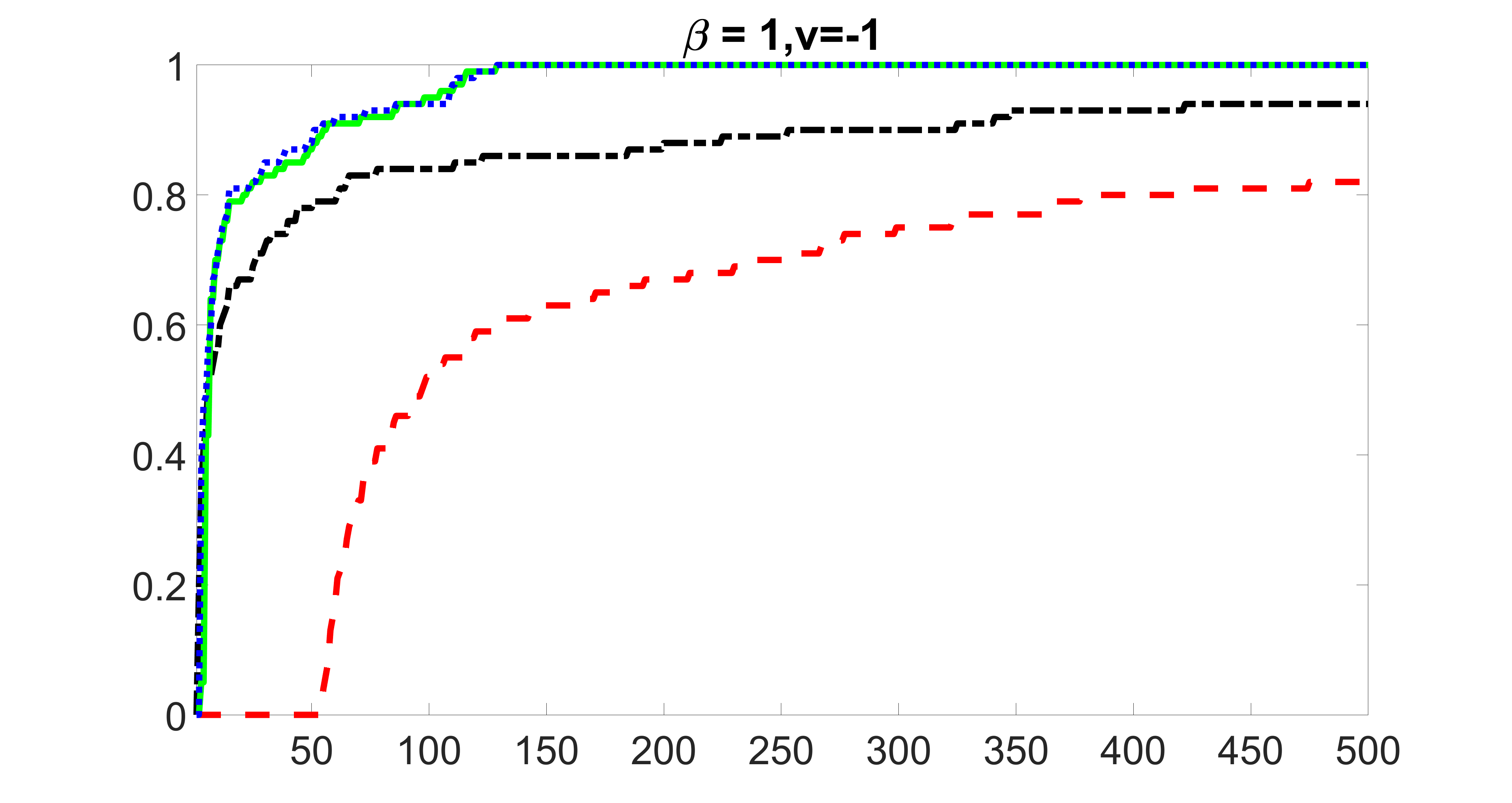}}
  \subcaptionbox{\footnotesize Confounder: weak \\ outcome, strong exposure}[0.45\linewidth]
 {\includegraphics[width=6cm,height=3.5cm]{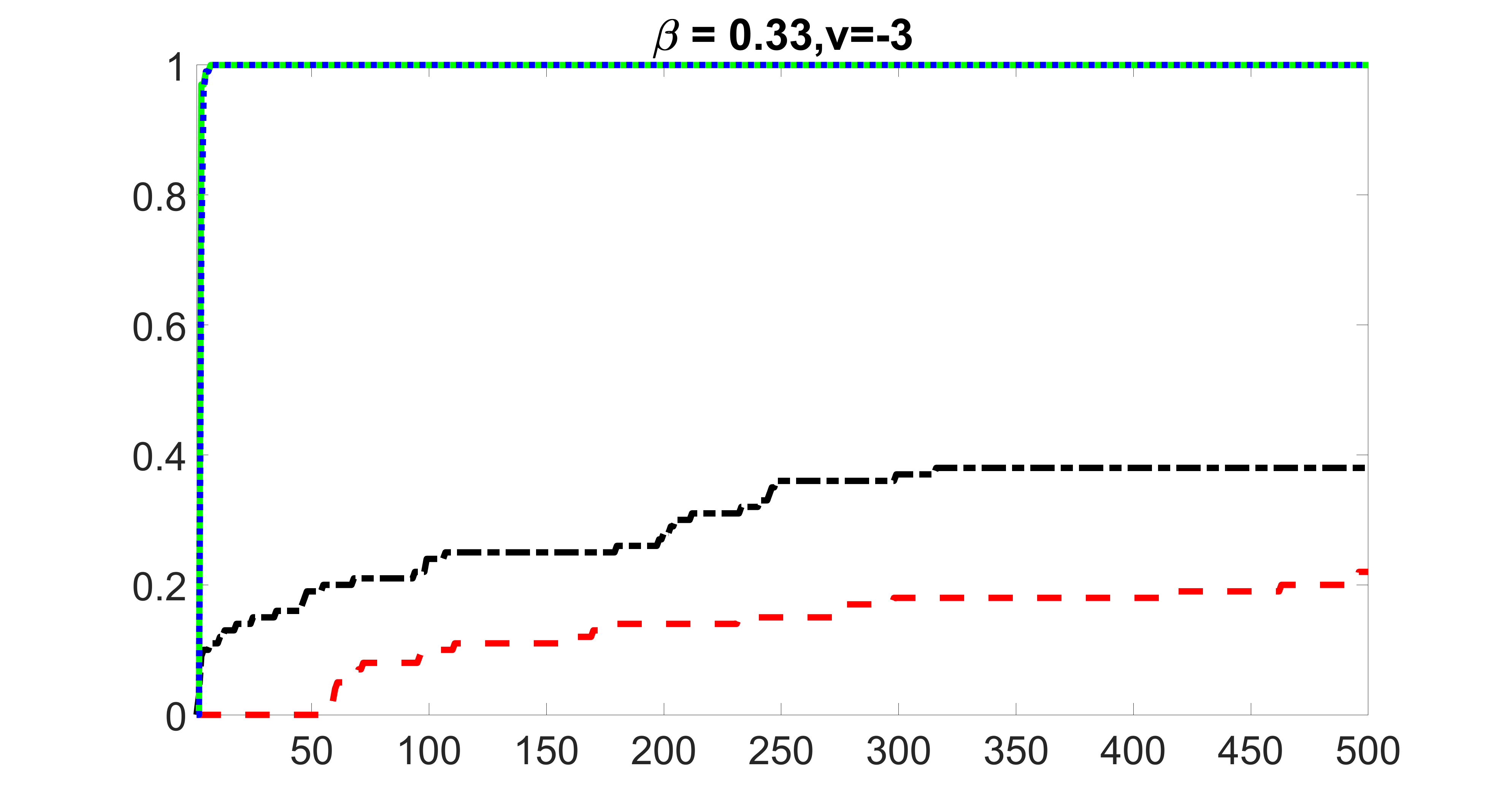}}
  \subcaptionbox{\footnotesize Precision: strong \\ outcome, zero exposure}[0.45\linewidth]
 {\includegraphics[width=6cm,height=3.5cm]{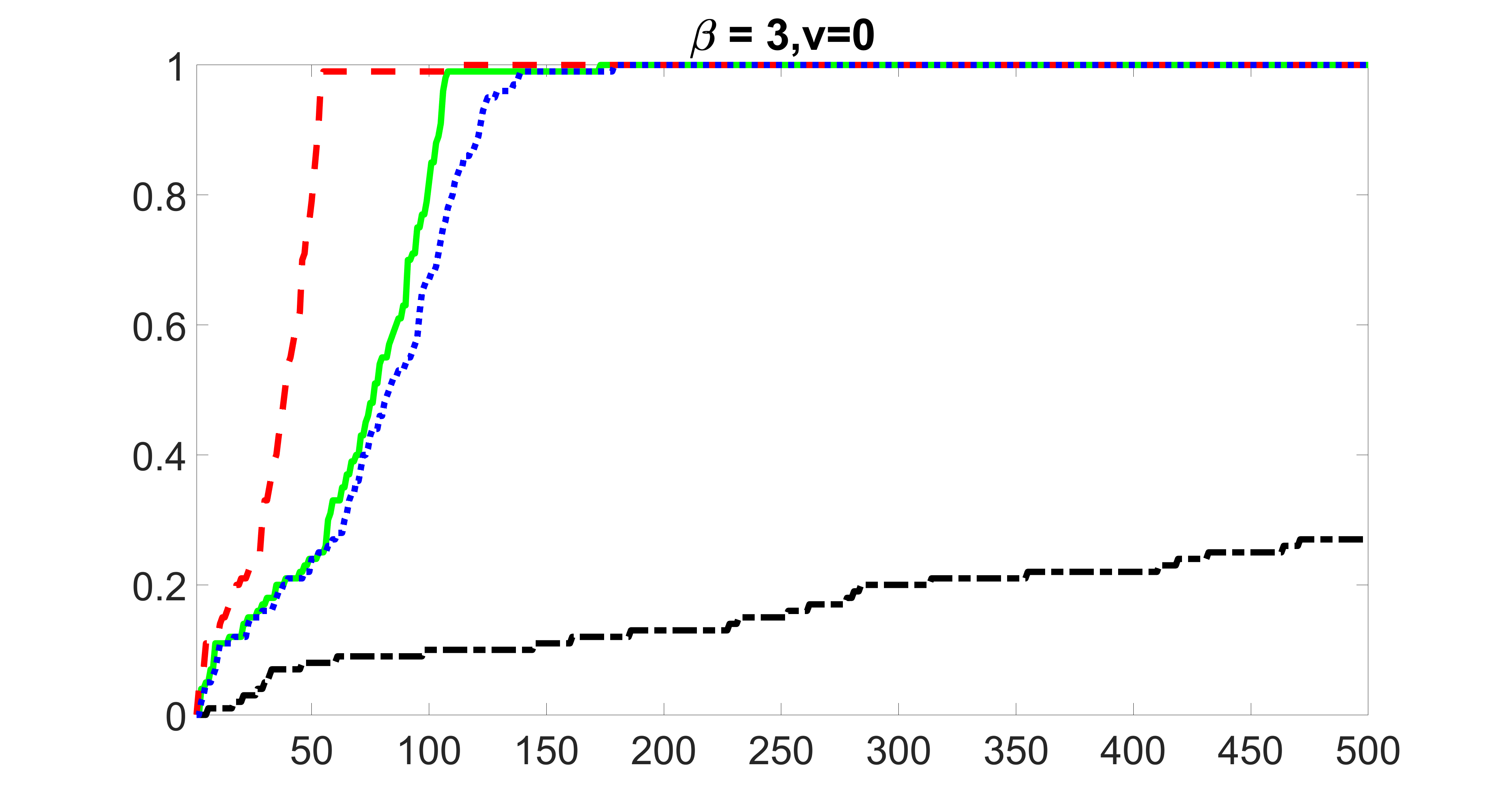}}
  \subcaptionbox{\footnotesize Precision: medium \\ outcome, zero exposure}[0.45\linewidth]
 {\includegraphics[width=6cm,height=3.5cm]{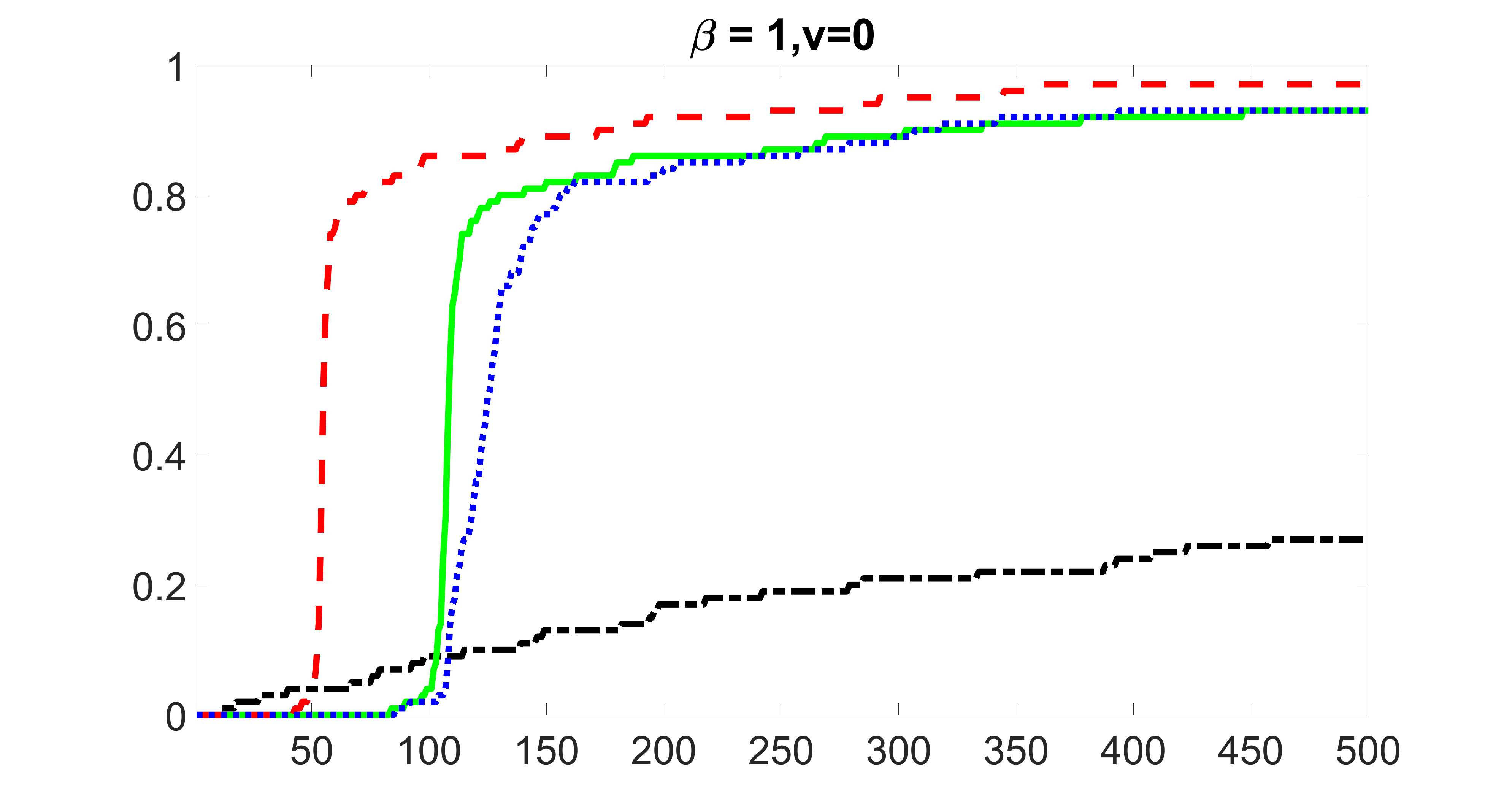}}
  \subcaptionbox{\footnotesize Precision: weak \\ outcome, zero exposure}[0.45\linewidth]
 {\includegraphics[width=6cm,height=3.5cm]{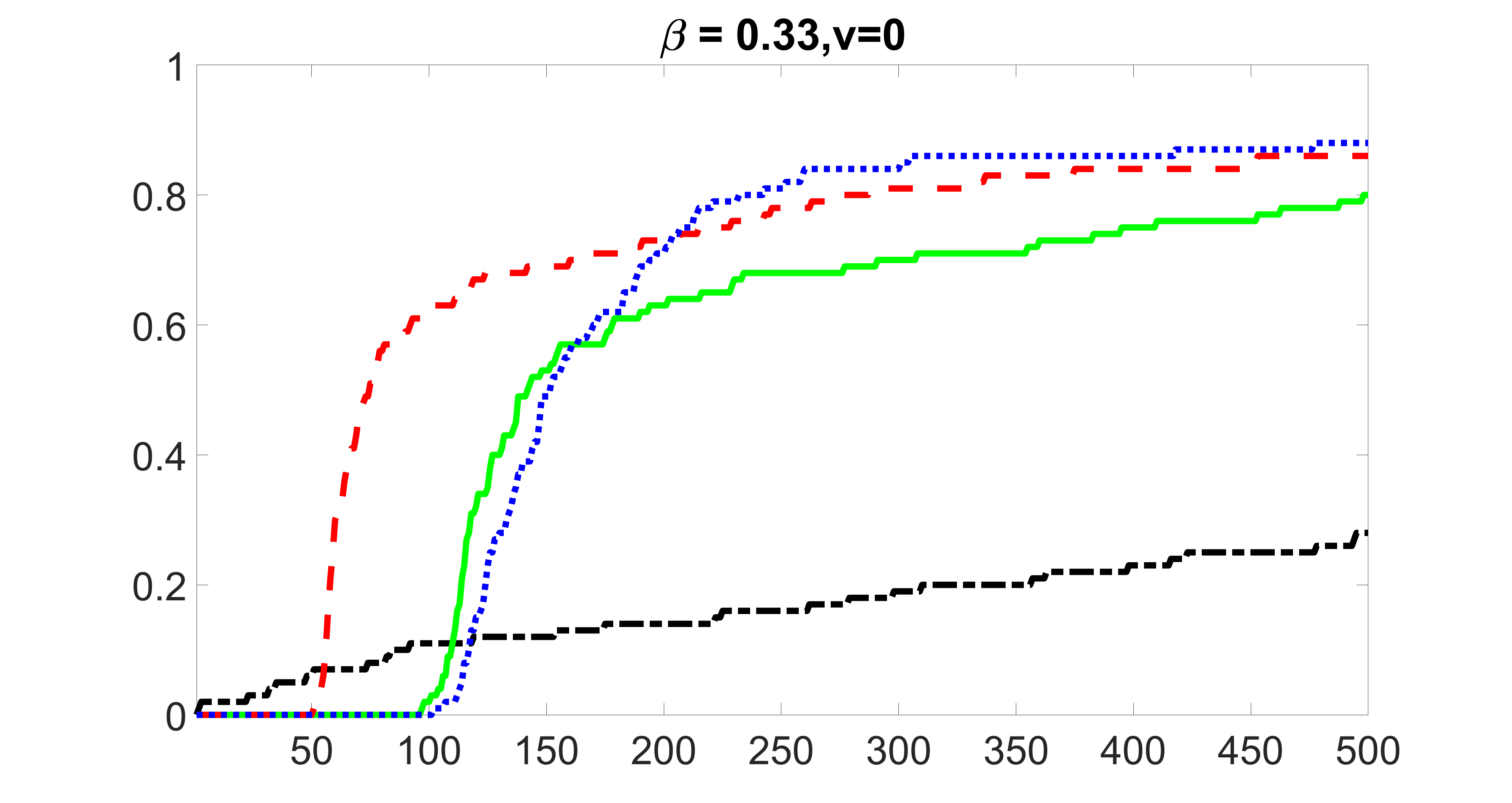}}
 \subcaptionbox{\footnotesize Precision: weaker \\ outcome, zero exposure}[0.45\linewidth]
 {\includegraphics[width=6cm,height=3.5cm]{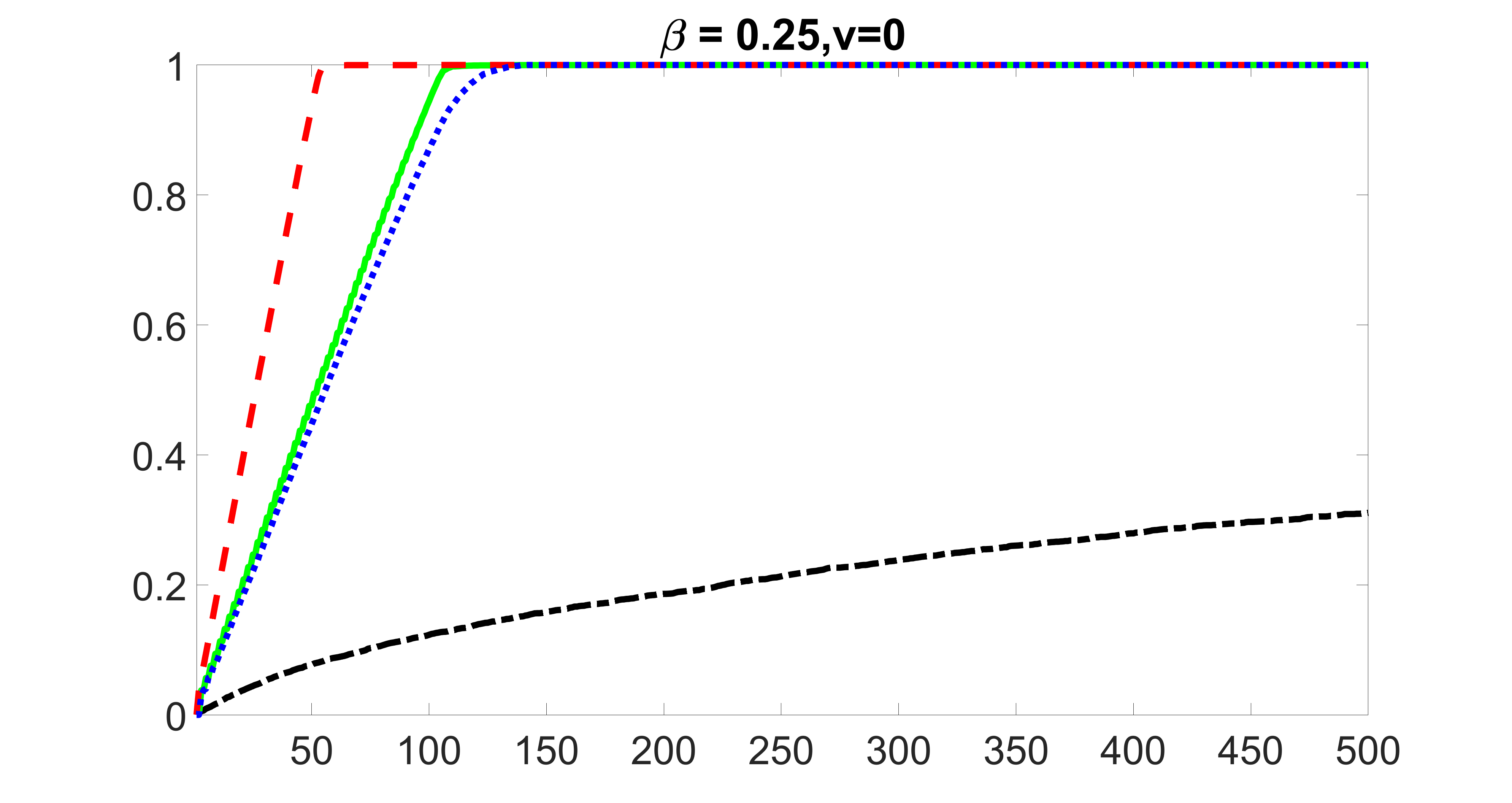}}
  \subcaptionbox{Overall coverage of $\mathcal{M}_1$}[0.45\linewidth]
 {\includegraphics[width=6cm,height=3.5cm]{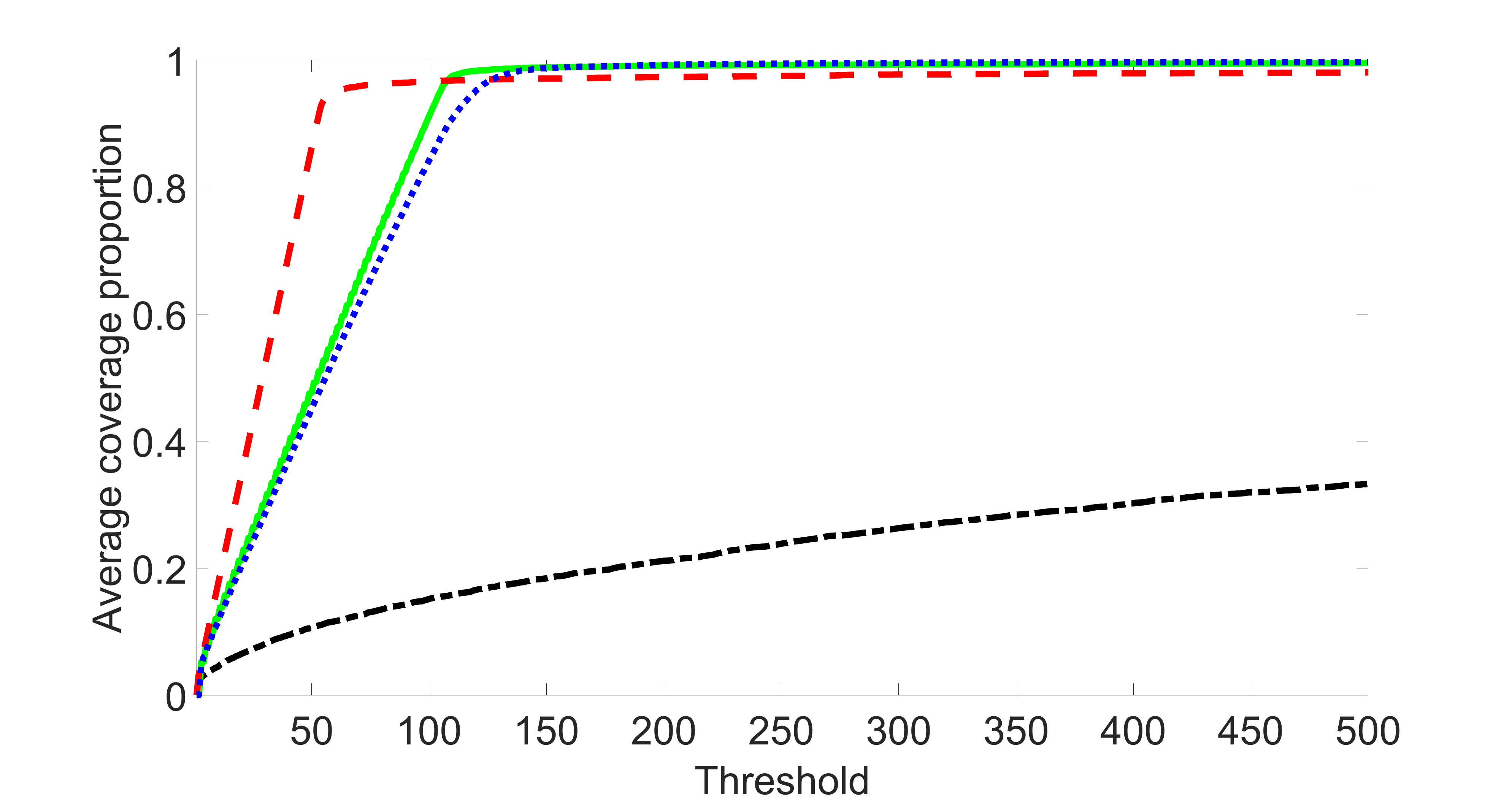}}
\caption{ Simulation results for the case $(n,s,K,\sigma) = (500,5000,52,1)$: Panels (a) -- (g) plot the average coverage proportion for $X_l$, where $l \in \mathcal{M}_1 =  \{1,2,3,104,105, 106\} \cup \mathcal{P}_{LD}$. Panels (a) -- (c) correspond to strong outcome and weak exposure predictor, moderate outcome and moderate exposure predictor and weak outcome and strong exposure predictor; Panels (d) -- (g) correspond to strong, moderate, and weak predictors of outcome only. Panel (g) plots the average coverage proportion for the index set $\mathcal{P}_{LD}$. Panel (h) plots the average coverage proportion for the index set $\mathcal{M}_1$. The x-axis represents the size of $\widehat{\mathcal{M}} $, while
y-axis denotes the average proportion. The blue dot, green solid, red dashed and black dash dotted lines denote the blockwise joint screening, joint screening, outcome screening, and intersection screening methods, respectively.}
\label{sim3step1n500sizesig52sigma1}
\end{figure}

\begin{figure}[htbp]
\captionsetup[subfigure]{justification=centering}
\centering
 \subcaptionbox{\footnotesize Confounder: strong \\ outcome, weak exposure}[0.45\linewidth]
 {\includegraphics[width=6cm,height=3.5cm]{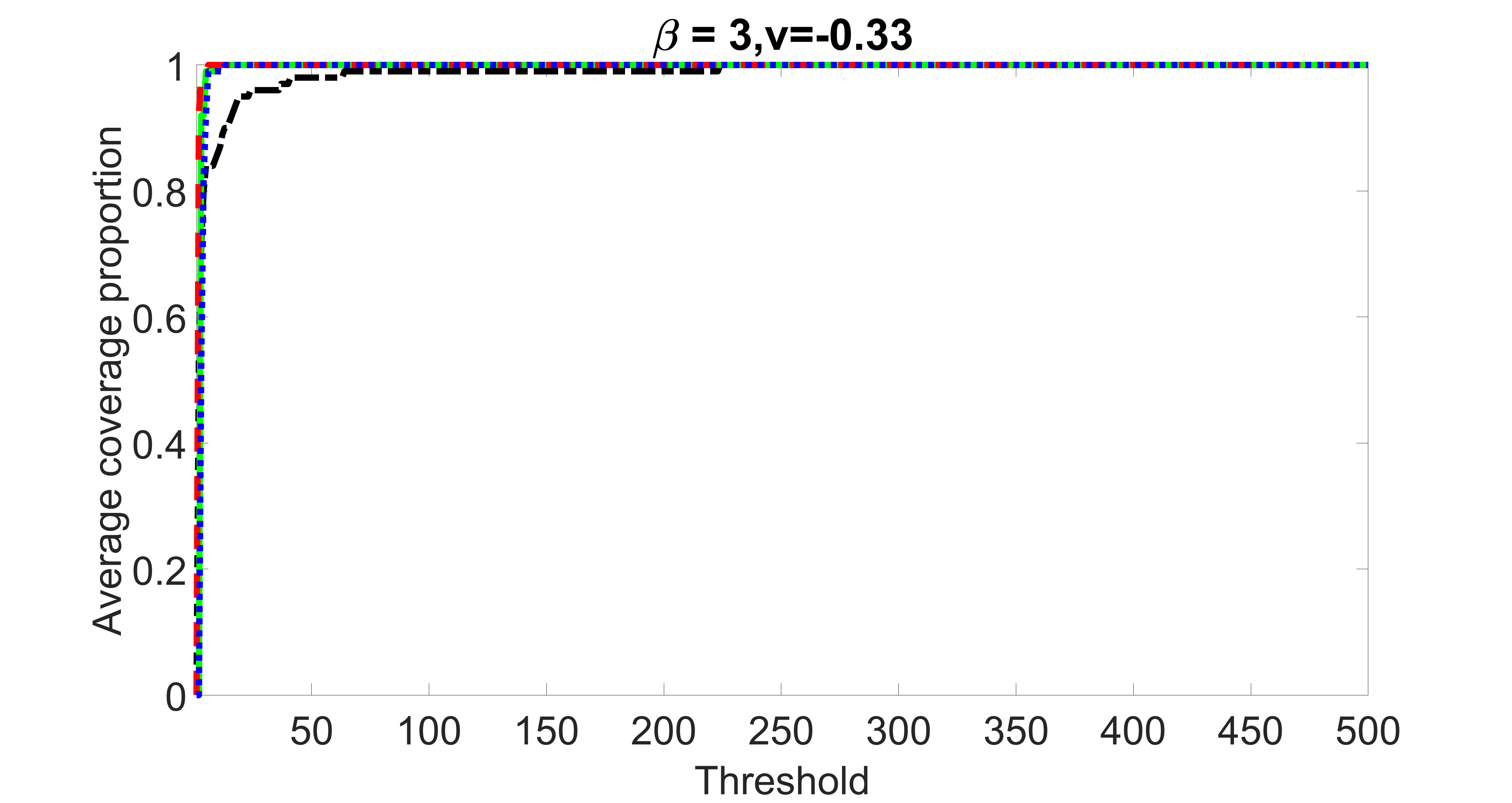}}
 \subcaptionbox{\footnotesize Confounder: medium \\ outcome, medium exposure}[0.45\linewidth]
 {\includegraphics[width=6cm,height=3.5cm]{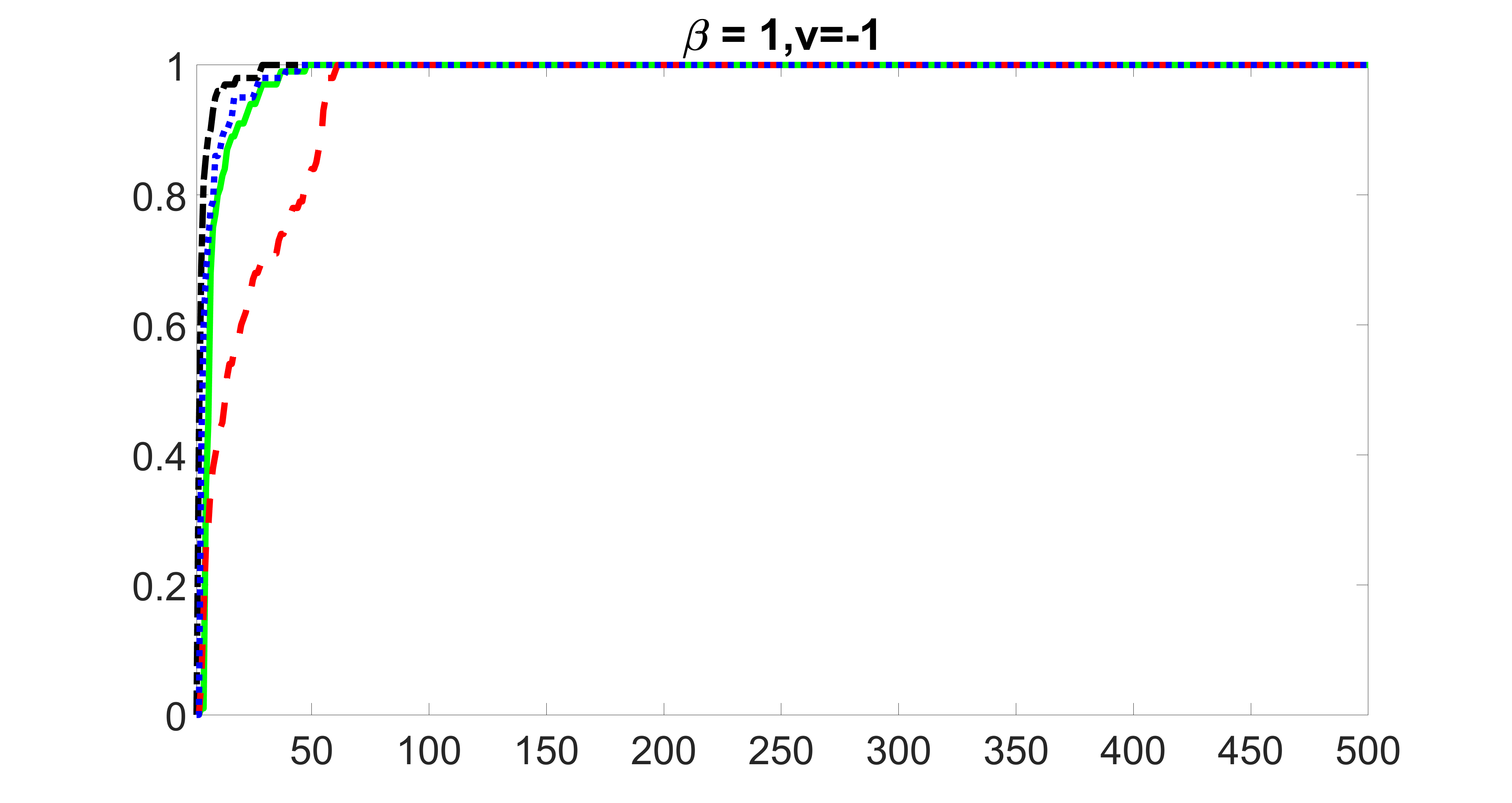}}
  \subcaptionbox{\footnotesize Confounder: weak \\ outcome, strong exposure}[0.45\linewidth]
 {\includegraphics[width=6cm,height=3.5cm]{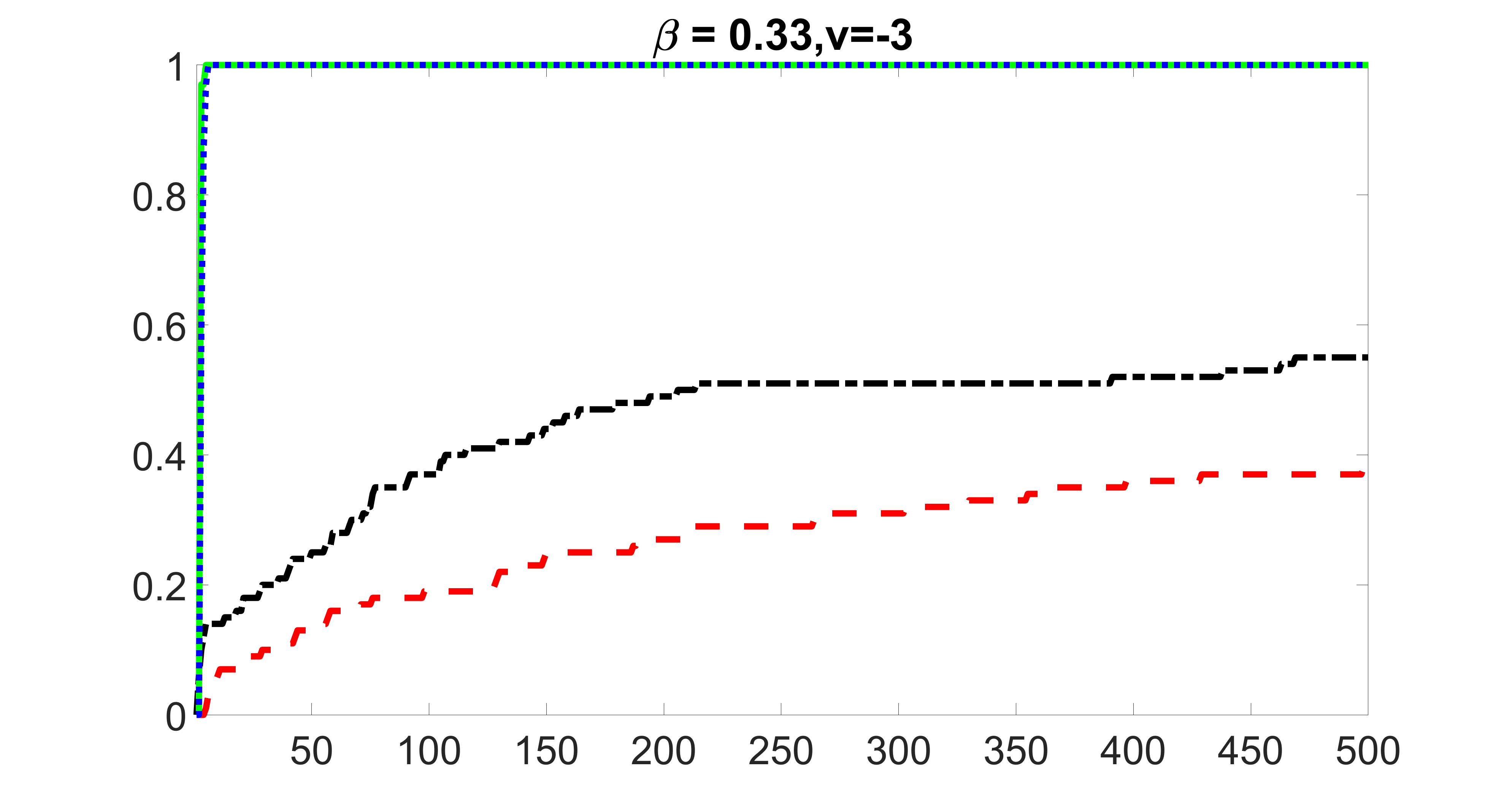}}
  \subcaptionbox{\footnotesize Precision: strong \\ outcome, zero exposure}[0.45\linewidth]
 {\includegraphics[width=6cm,height=3.5cm]{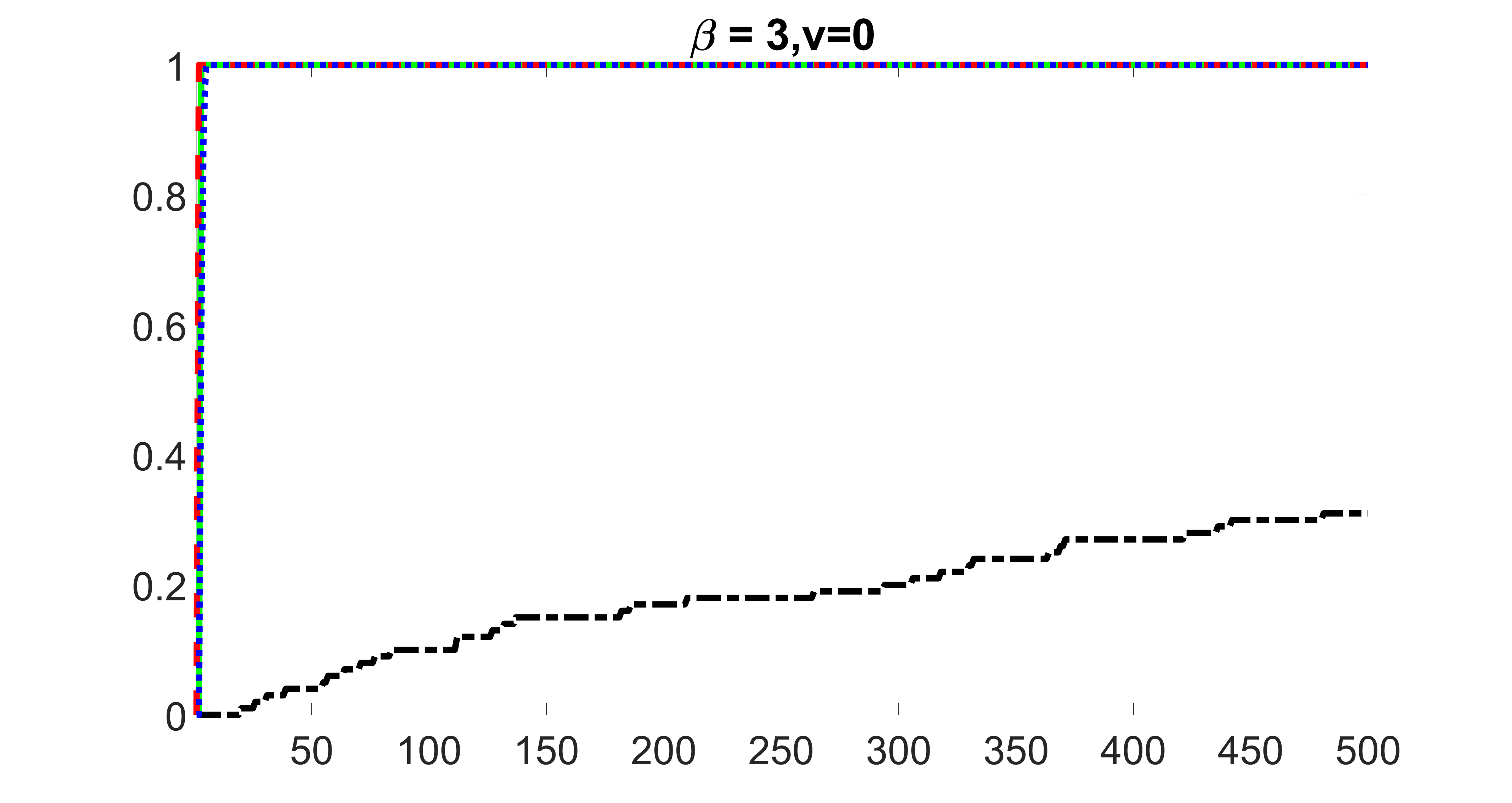}}
  \subcaptionbox{\footnotesize Precision: medium \\ outcome, zero exposure}[0.45\linewidth]
 {\includegraphics[width=6cm,height=3.5cm]{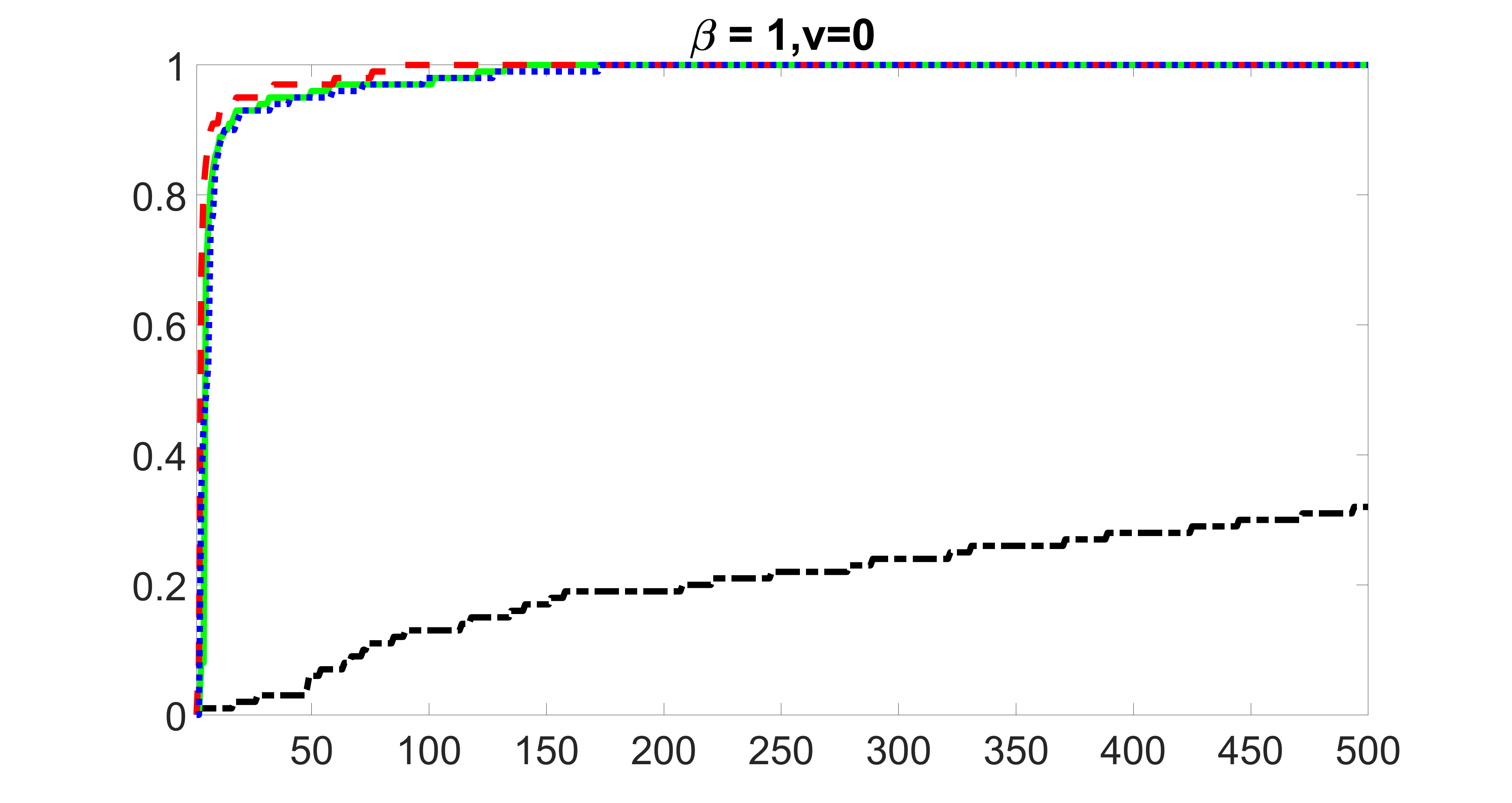}}
  \subcaptionbox{\footnotesize Precision: weak \\ outcome, zero exposure}[0.45\linewidth]
 {\includegraphics[width=6cm,height=3.5cm]{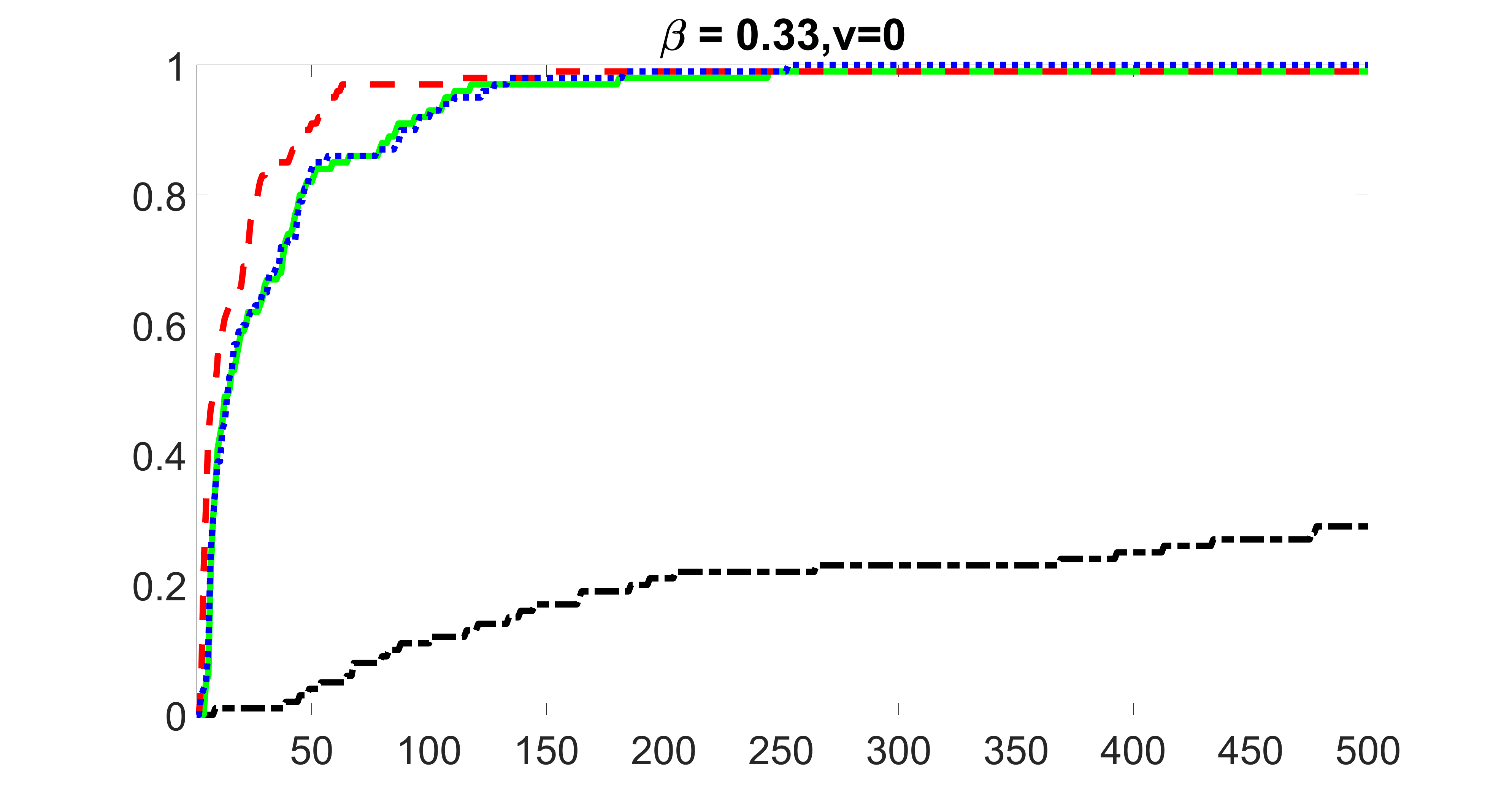}}
 \subcaptionbox{\footnotesize Precision: weaker \\ outcome, zero exposure}[0.45\linewidth]
 {\includegraphics[width=6cm,height=3.5cm]{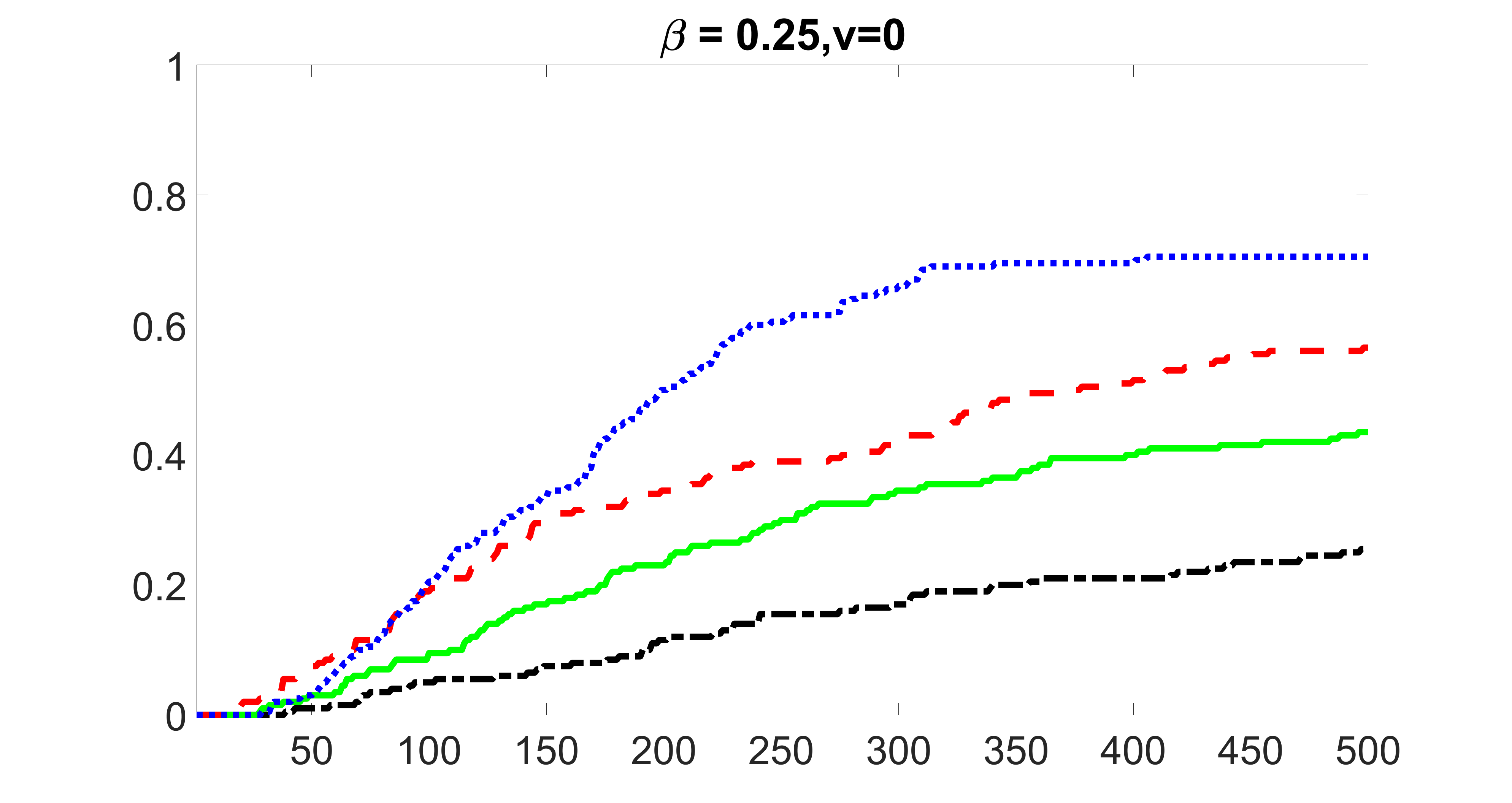}}
  \subcaptionbox{Overall coverage of $\mathcal{M}_1$}[0.45\linewidth]
 {\includegraphics[width=6cm,height=3.5cm]{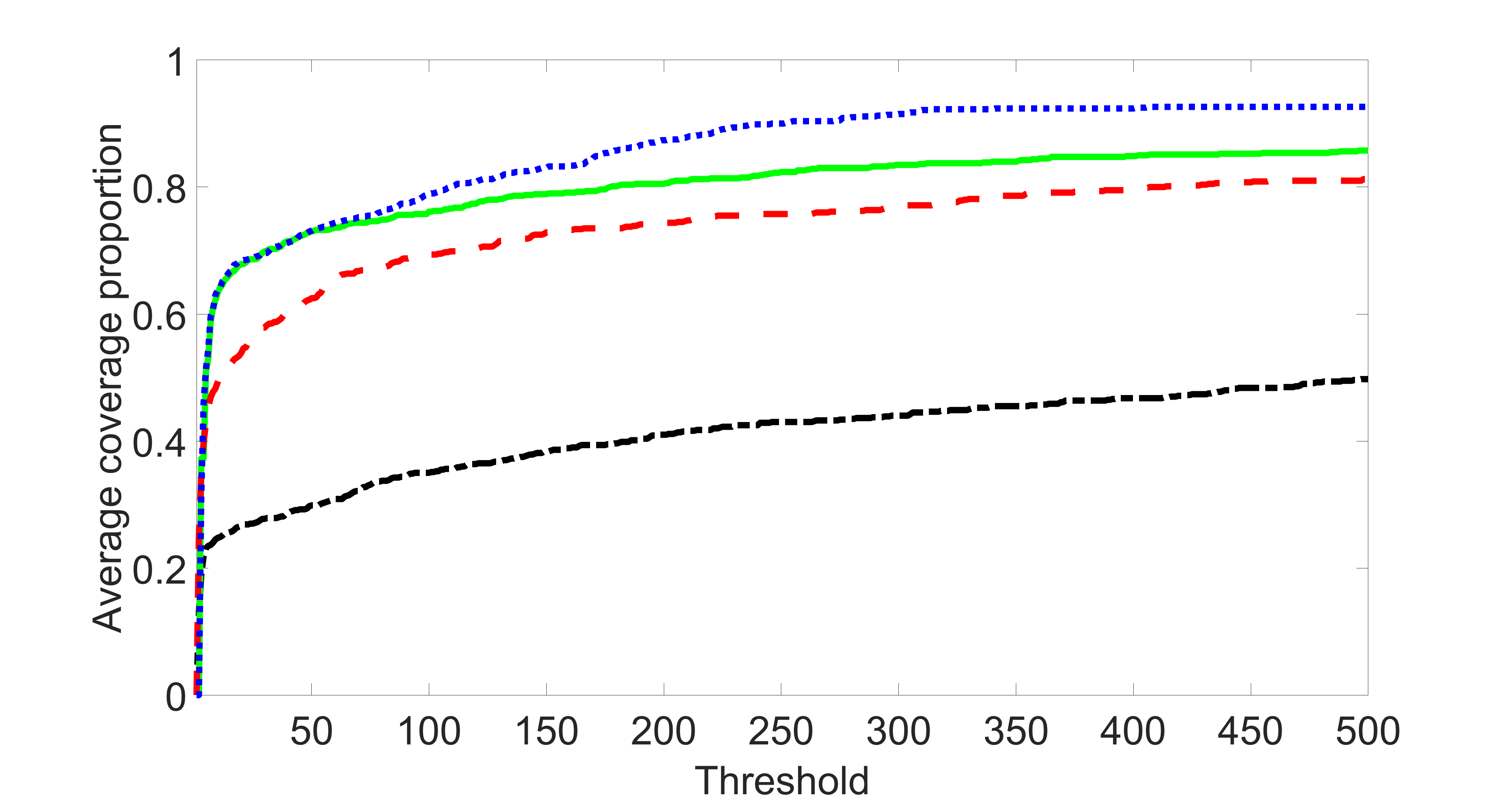}}
\caption{ Simulation results for the case $(n,s,K,\sigma) = (1000,5000,2,1)$: Panels (a) -- (g) plot the average coverage proportion for $X_l$, where $l \in \mathcal{M}_1 =  \{1,2,3,104,105, 106\} \cup \mathcal{P}_{LD}$. Panels (a) -- (c) correspond to strong outcome and weak exposure predictor, moderate outcome and moderate exposure predictor and weak outcome and strong exposure predictor; Panels (d) -- (g) correspond to strong, moderate, and weak predictors of outcome only. Panel (g) plots the average coverage proportion for the index set $\mathcal{P}_{LD}$. Panel (h) plots the average coverage proportion for the index set $\mathcal{M}_1$. The x-axis represents the size of $\widehat{\mathcal{M}} $, while
y-axis denotes the average proportion. The blue dot, green solid, red dashed and black dash dotted lines denote the blockwise joint screening, joint screening, outcome screening, and intersection screening methods, respectively.}
\label{sim3step1n1000sizesig2sigma1}
\end{figure}

\begin{figure}[htbp]
\captionsetup[subfigure]{justification=centering}
\centering
 \subcaptionbox{\footnotesize Confounder: strong \\ outcome, weak exposure}[0.45\linewidth]
 {\includegraphics[width=6cm,height=3.5cm]{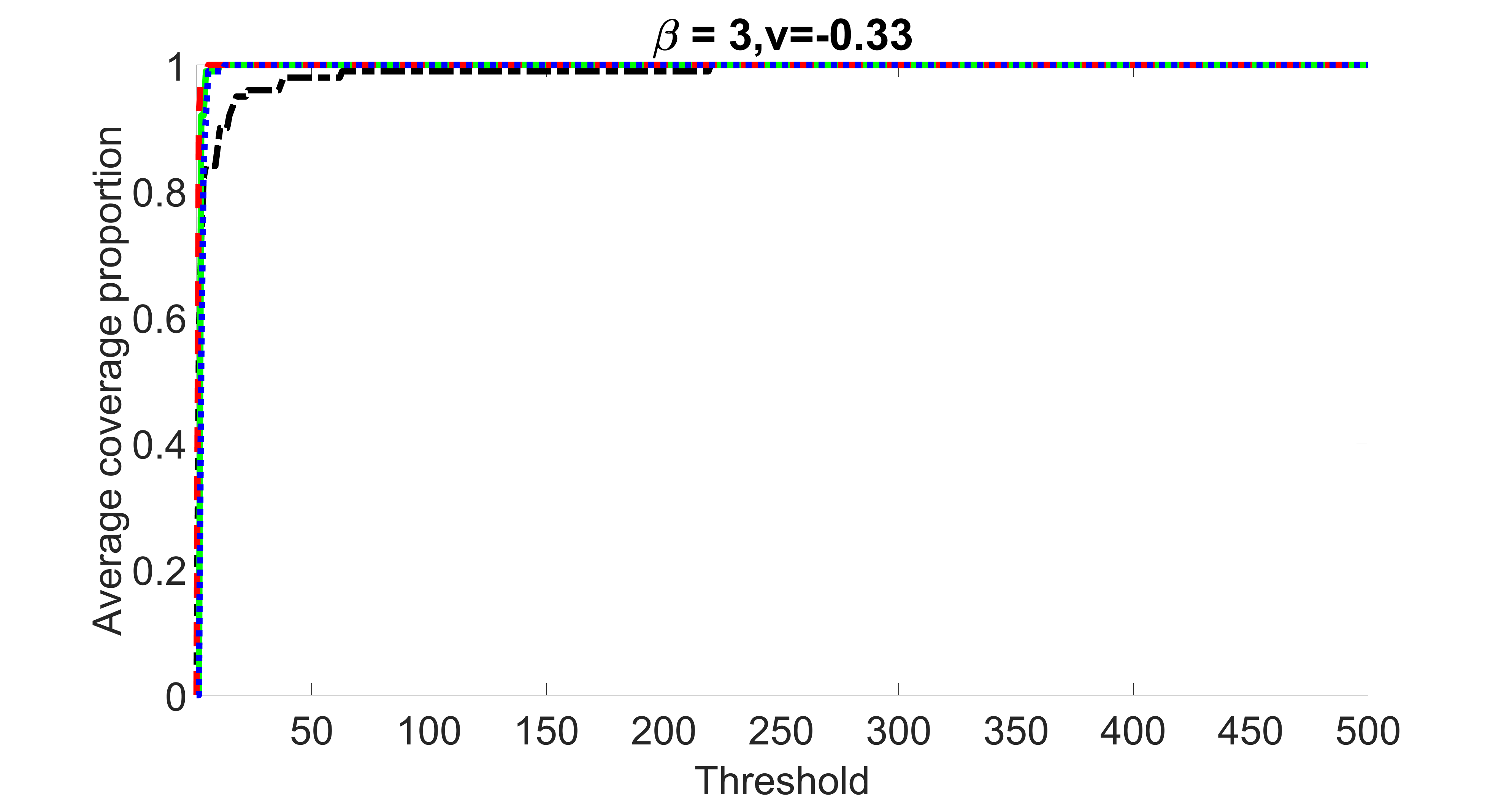}}
 \subcaptionbox{\footnotesize Confounder: medium \\ outcome, medium exposure}[0.45\linewidth]
 {\includegraphics[width=6cm,height=3.5cm]{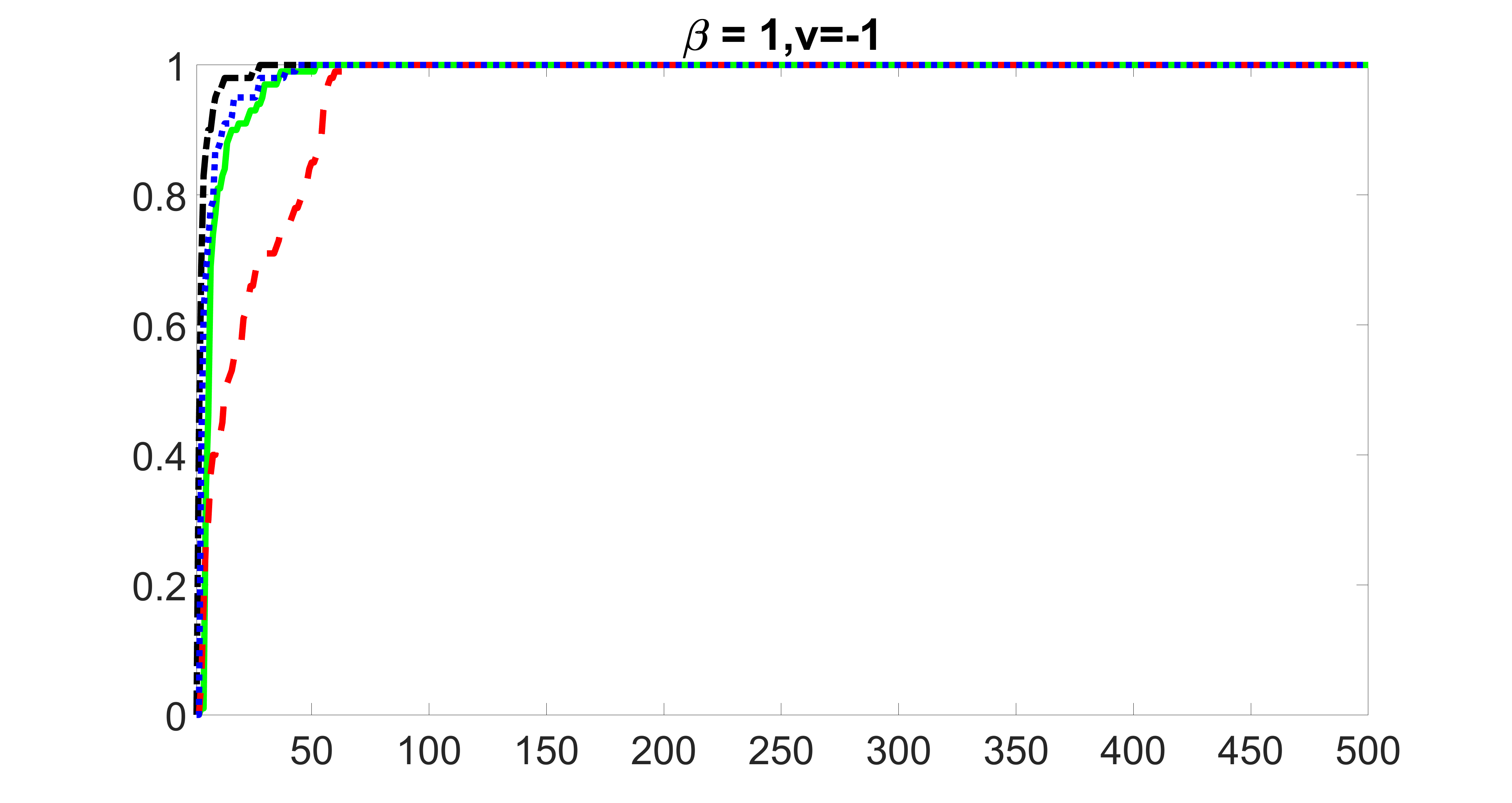}}
  \subcaptionbox{\footnotesize Confounder: weak \\ outcome, strong exposure}[0.45\linewidth]
 {\includegraphics[width=6cm,height=3.5cm]{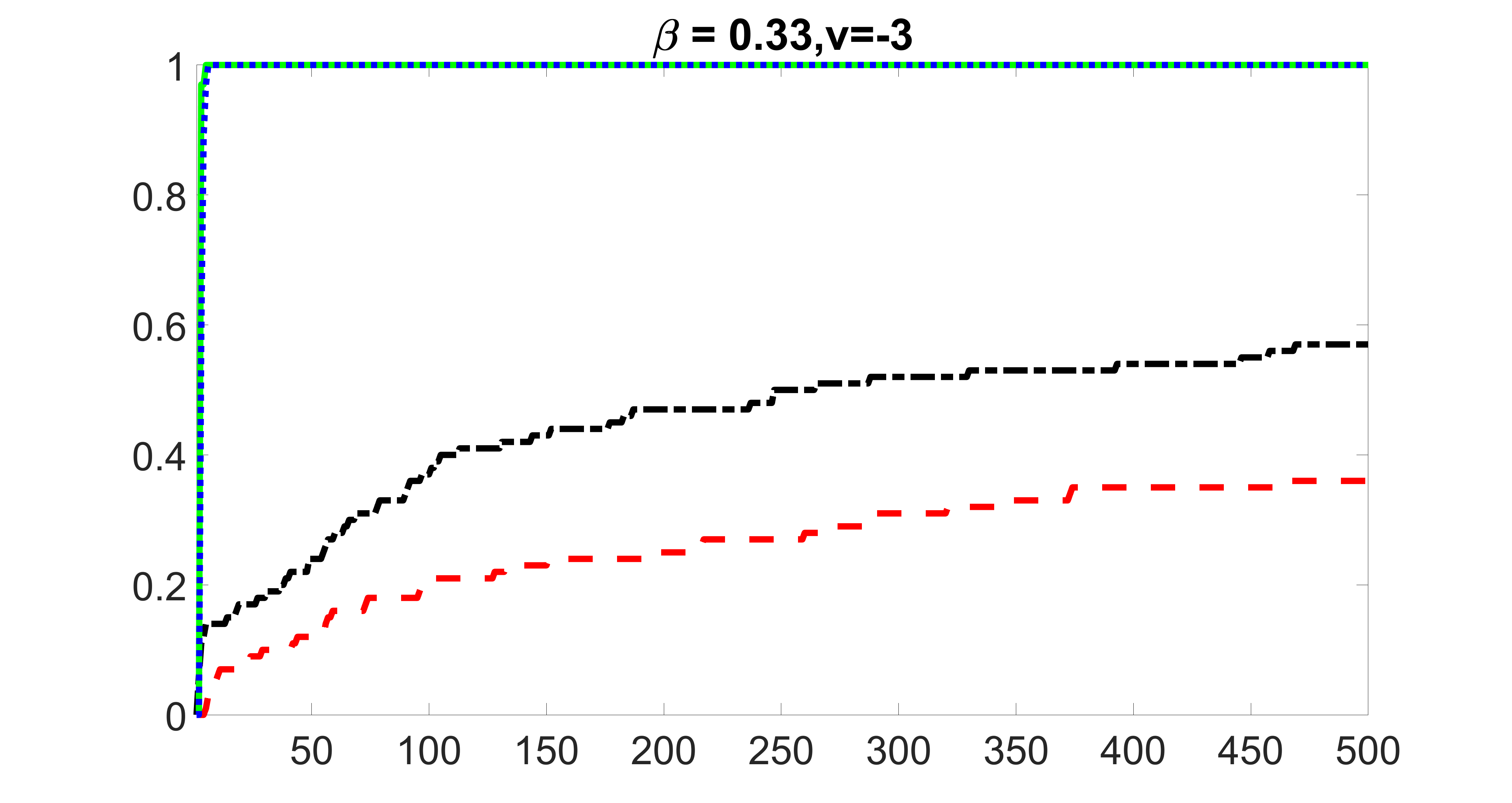}}
  \subcaptionbox{\footnotesize Precision: strong \\ outcome, zero exposure}[0.45\linewidth]
 {\includegraphics[width=6cm,height=3.5cm]{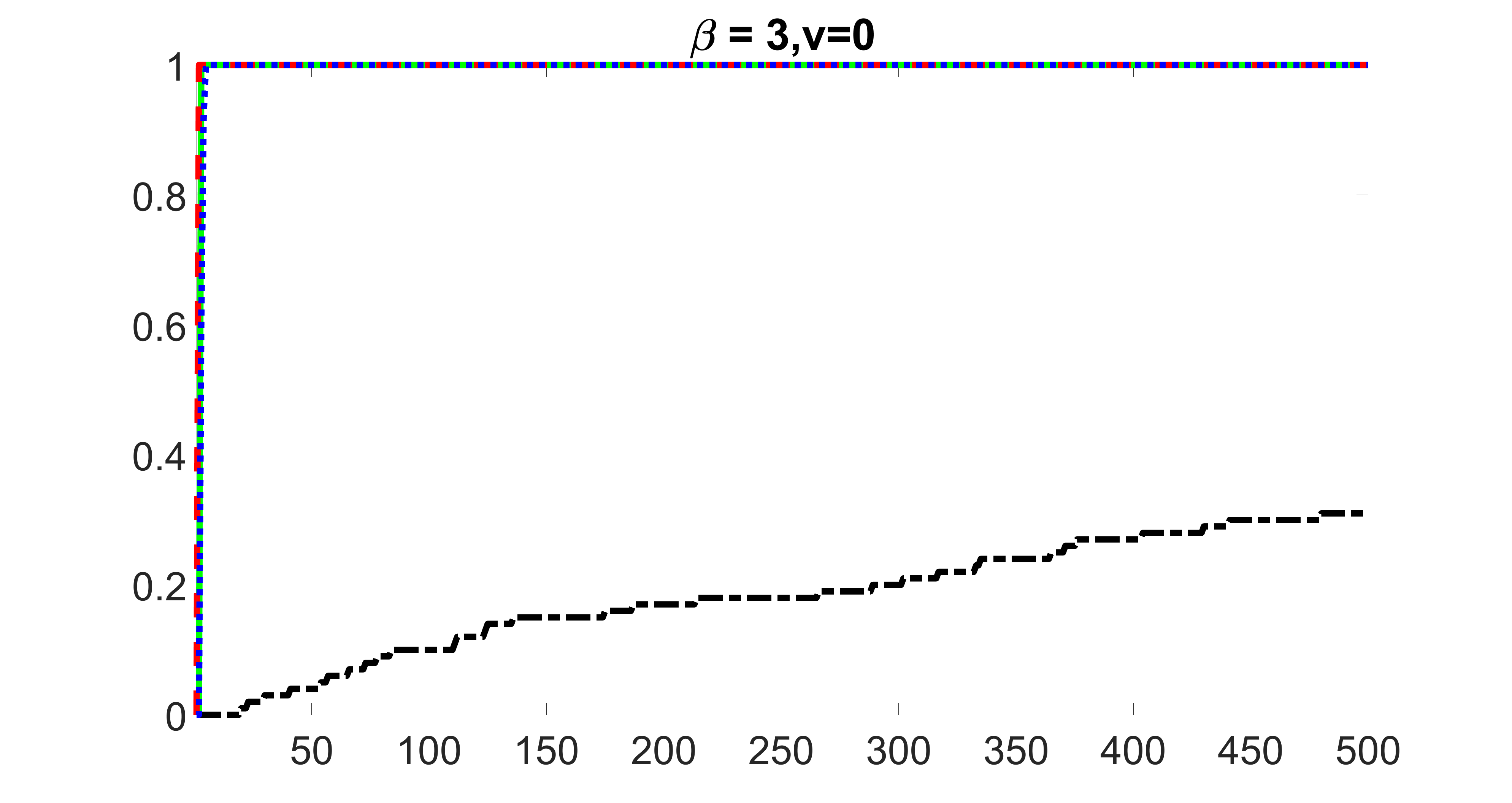}}
  \subcaptionbox{\footnotesize Precision: medium \\ outcome, zero exposure}[0.45\linewidth]
 {\includegraphics[width=6cm,height=3.5cm]{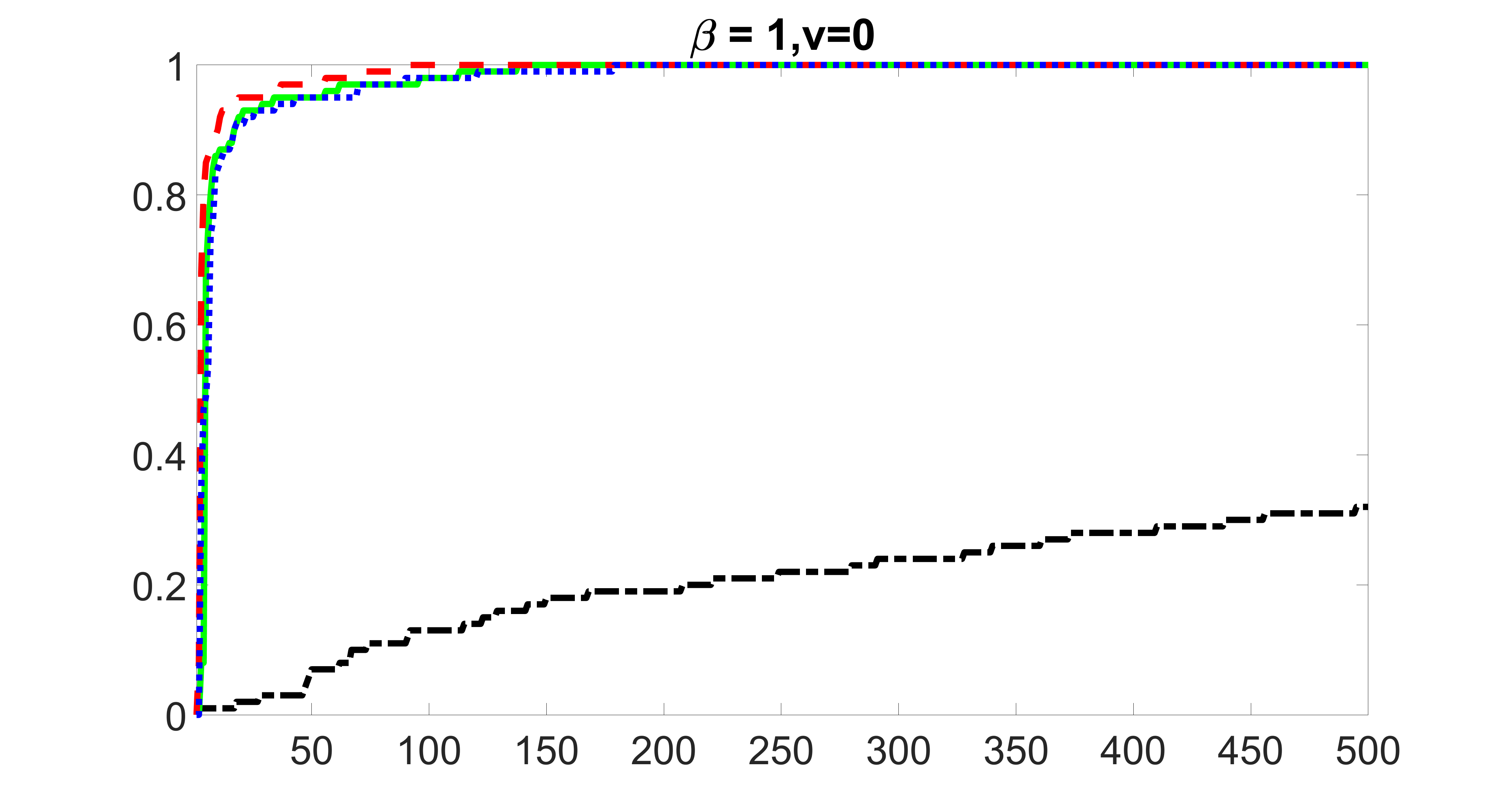}}
  \subcaptionbox{\footnotesize Precision: weak \\ outcome, zero exposure}[0.45\linewidth]
 {\includegraphics[width=6cm,height=3.5cm]{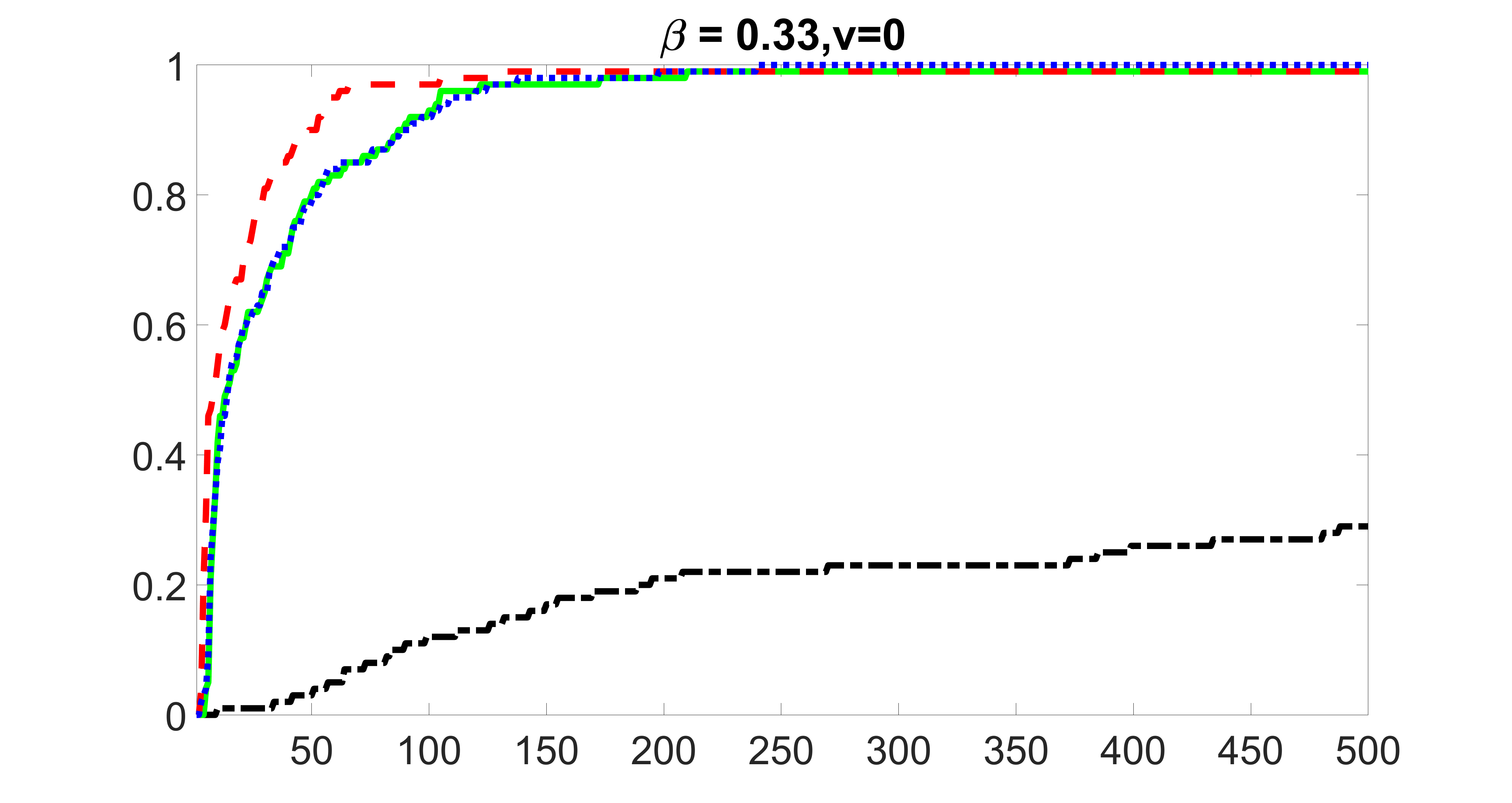}}
 \subcaptionbox{\footnotesize Precision: weaker \\ outcome, zero exposure}[0.45\linewidth]
 {\includegraphics[width=6cm,height=3.5cm]{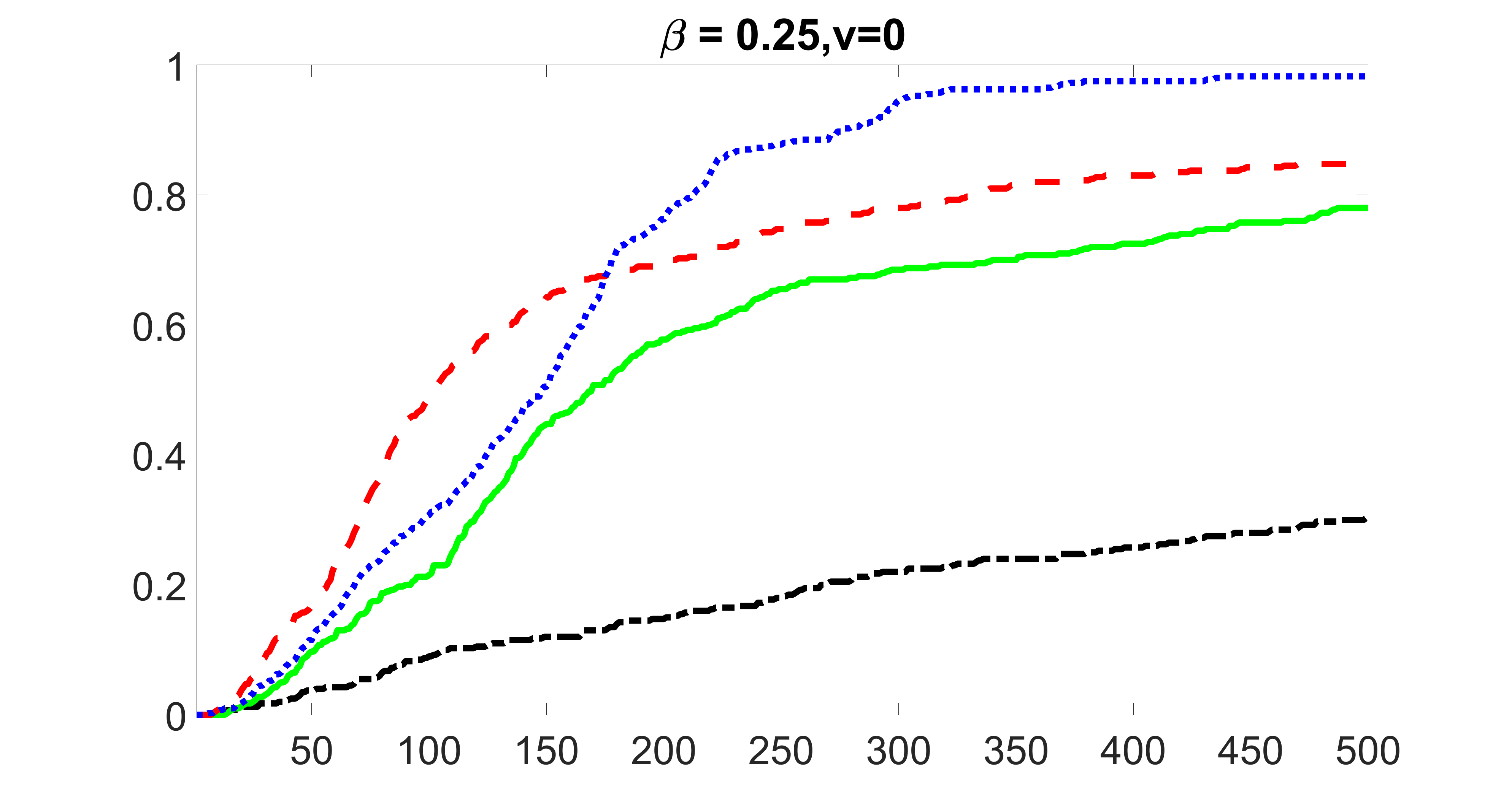}}
  \subcaptionbox{Overall coverage of $\mathcal{M}_1$}[0.45\linewidth]
 {\includegraphics[width=6cm,height=3.5cm]{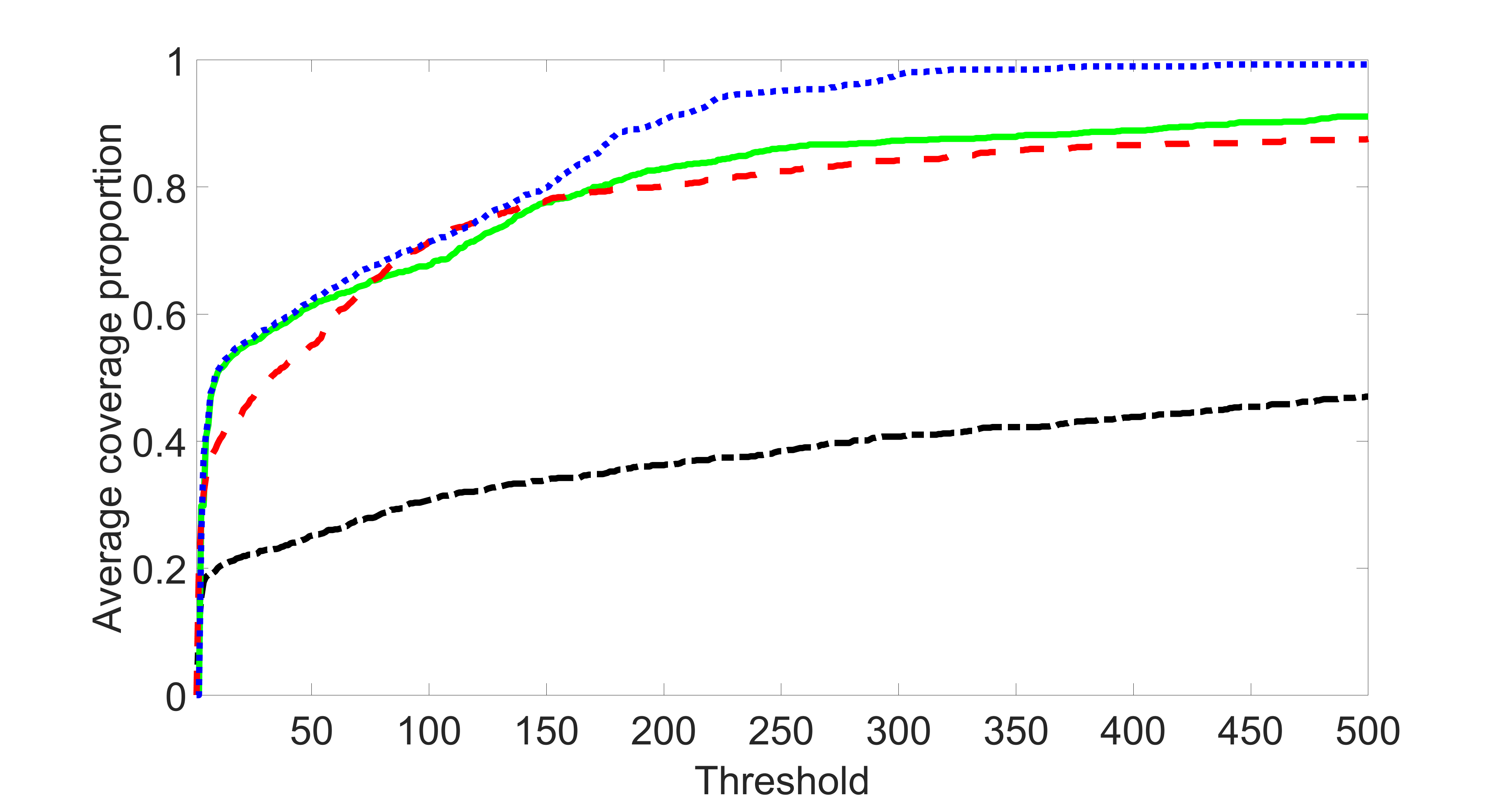}}
\caption{ Simulation results for the case $(n,s,K,\sigma) = (1000,5000,4,1)$: Panels (a) -- (g) plot the average coverage proportion for $X_l$, where $l \in \mathcal{M}_1 =  \{1,2,3,104,105, 106\} \cup \mathcal{P}_{LD}$. Panels (a) -- (c) correspond to strong outcome and weak exposure predictor, moderate outcome and moderate exposure predictor and weak outcome and strong exposure predictor; Panels (d) -- (g) correspond to strong, moderate, and weak predictors of outcome only. Panel (g) plots the average coverage proportion for the index set $\mathcal{P}_{LD}$. Panel (h) plots the average coverage proportion for the index set $\mathcal{M}_1$. The x-axis represents the size of $\widehat{\mathcal{M}} $, while
y-axis denotes the average proportion. The blue dot, green solid, red dashed and black dash dotted lines denote the blockwise joint screening, joint screening, outcome screening, and intersection screening methods, respectively.}
\label{sim3step1n1000sizesig4sigma1}
\end{figure}

\begin{figure}[htbp]
\captionsetup[subfigure]{justification=centering}
\centering
 \subcaptionbox{\footnotesize Confounder: strong \\ outcome, weak exposure}[0.45\linewidth]
 {\includegraphics[width=6cm,height=3.5cm]{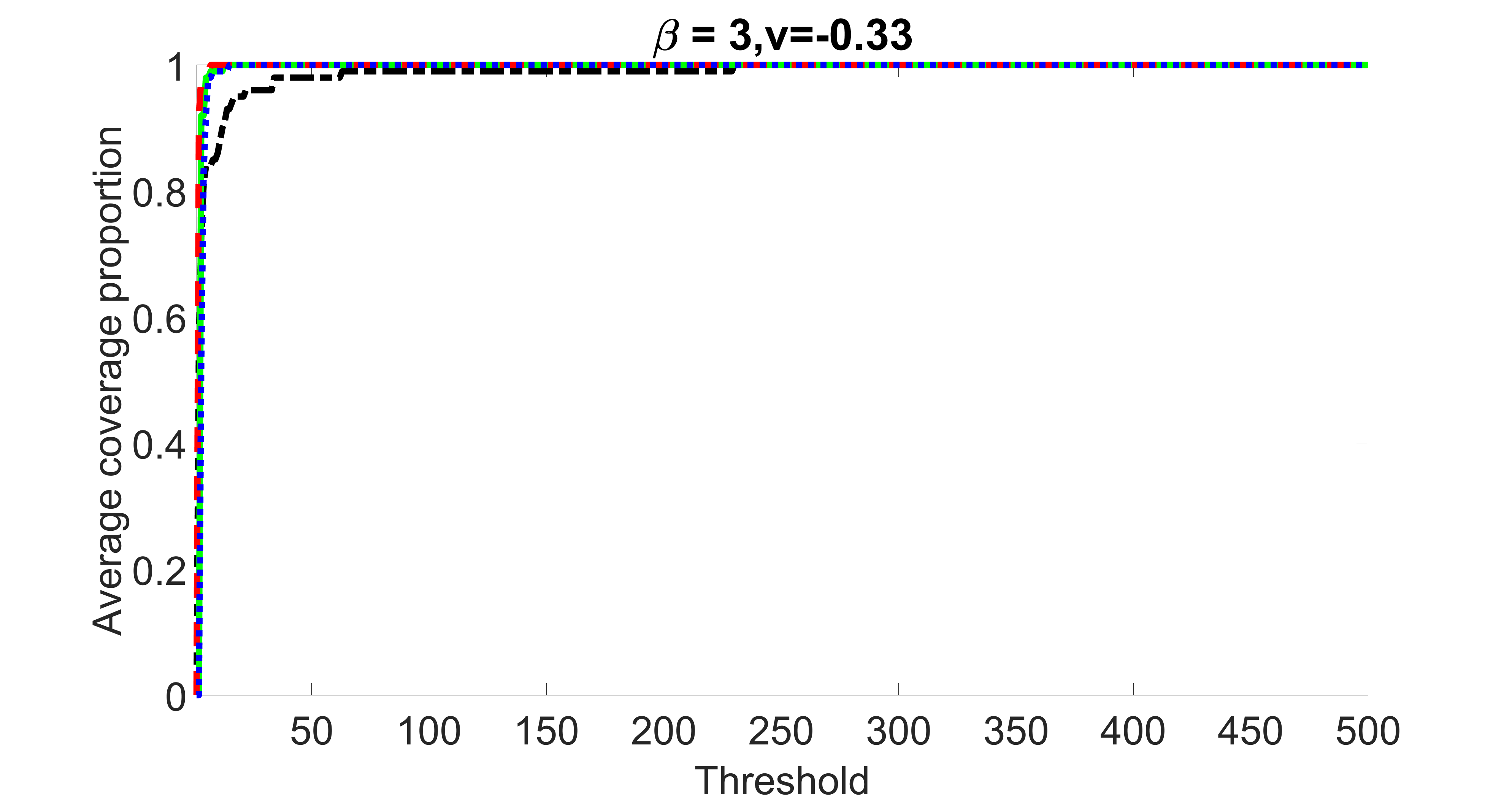}}
 \subcaptionbox{\footnotesize Confounder: medium \\ outcome, medium exposure}[0.45\linewidth]
 {\includegraphics[width=6cm,height=3.5cm]{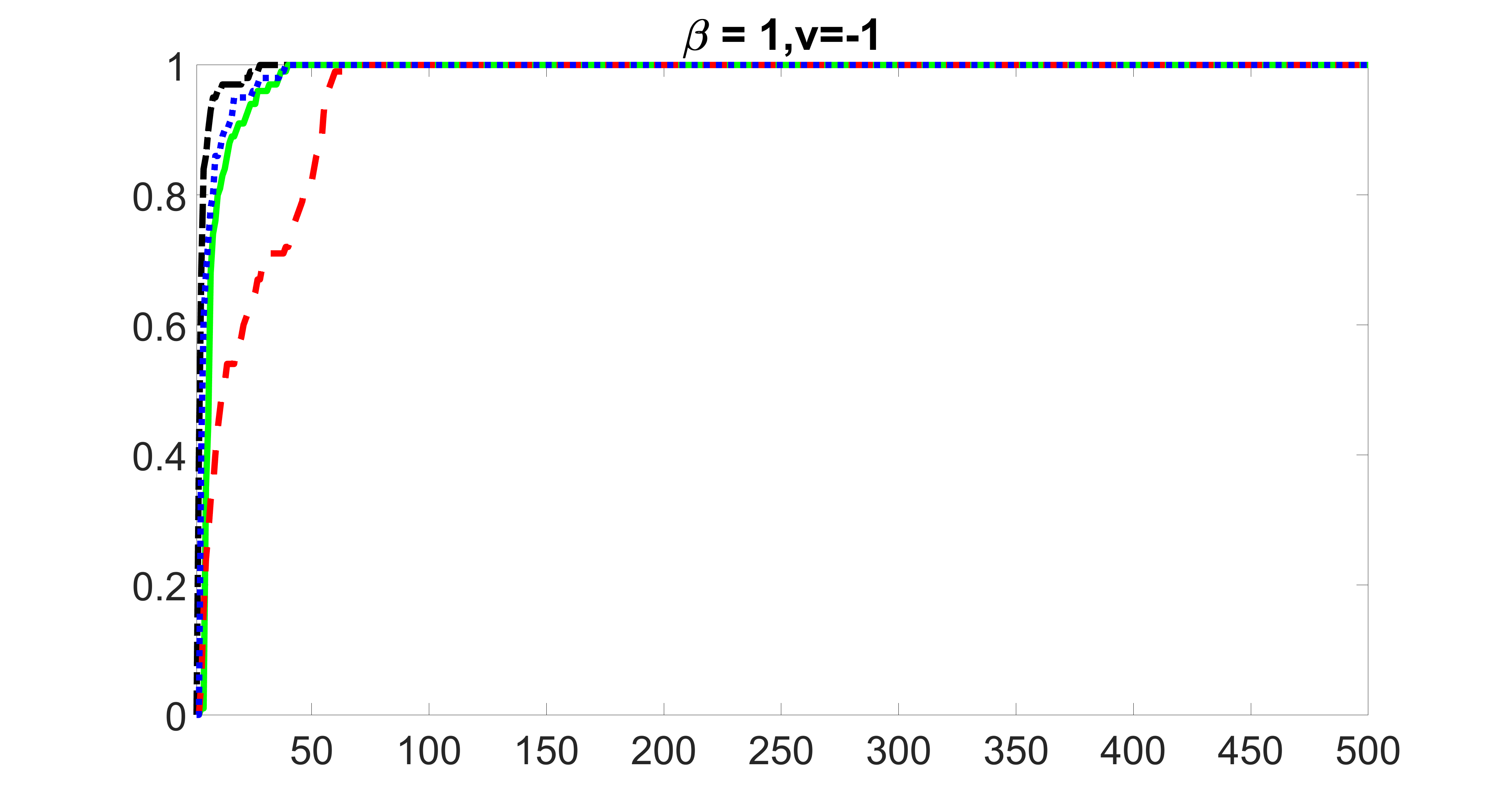}}
  \subcaptionbox{\footnotesize Confounder: weak \\ outcome, strong exposure}[0.45\linewidth]
 {\includegraphics[width=6cm,height=3.5cm]{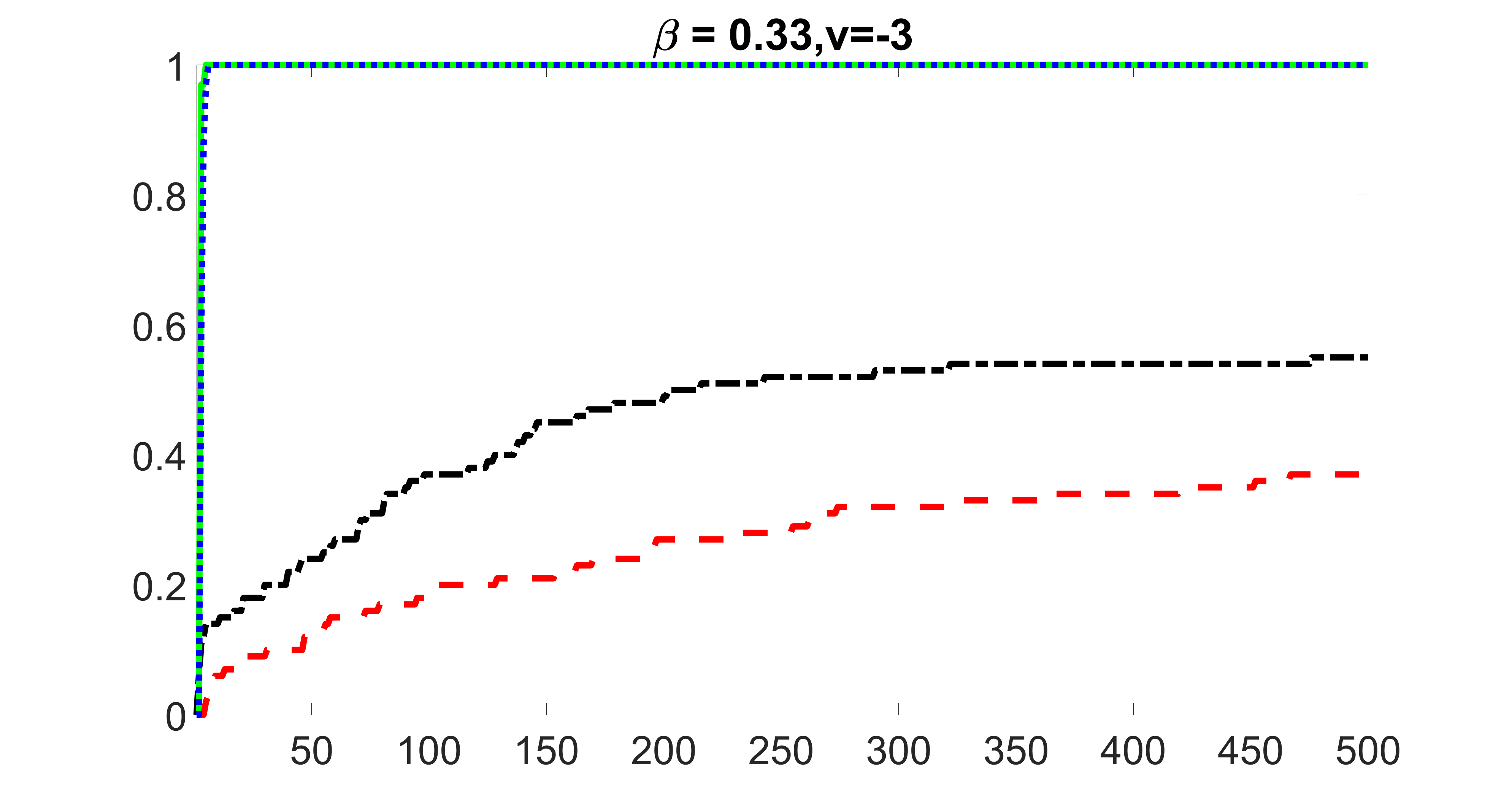}}
  \subcaptionbox{\footnotesize Precision: strong \\ outcome, zero exposure}[0.45\linewidth]
 {\includegraphics[width=6cm,height=3.5cm]{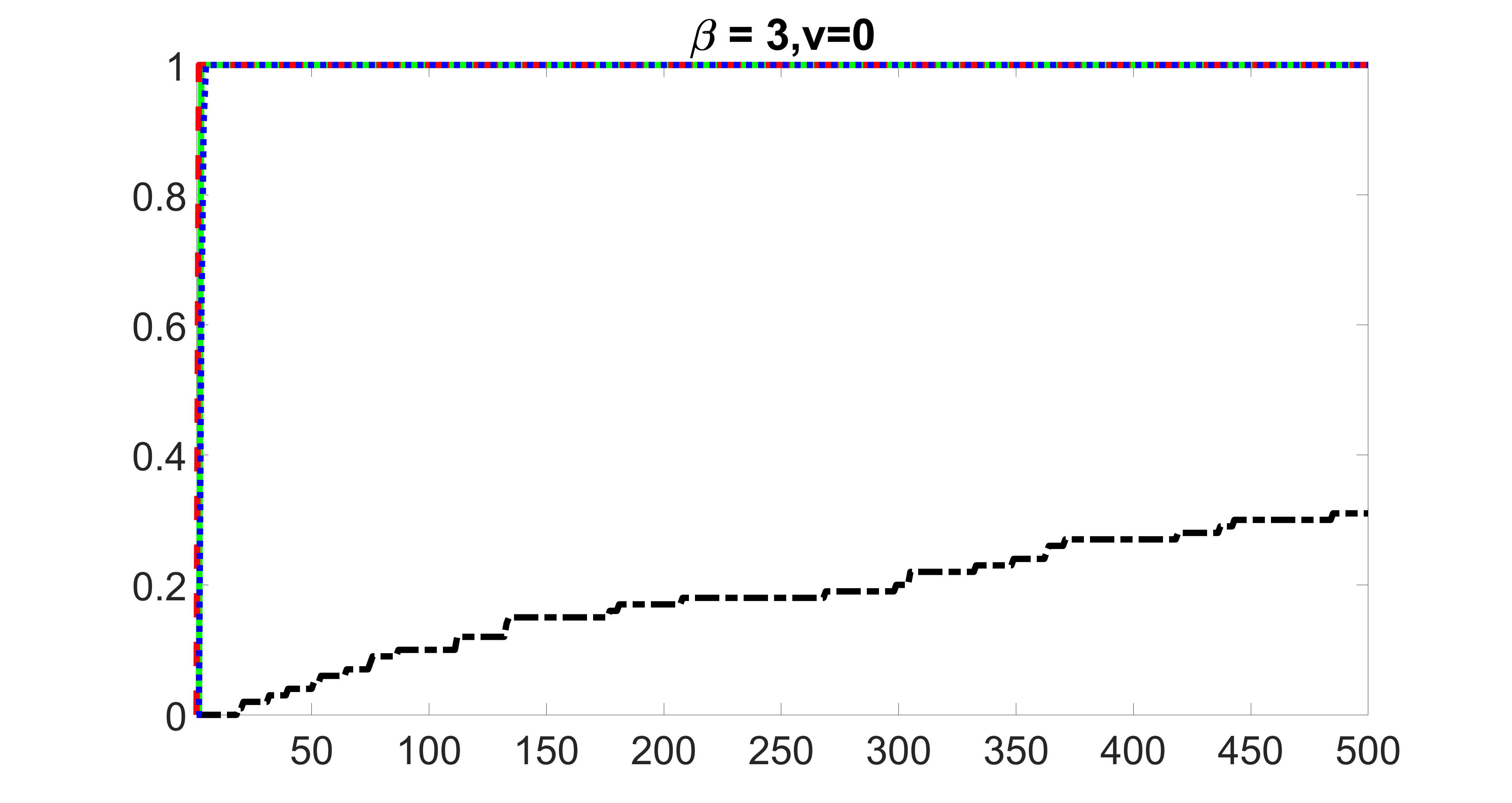}}
  \subcaptionbox{\footnotesize Precision: medium \\ outcome, zero exposure}[0.45\linewidth]
 {\includegraphics[width=6cm,height=3.5cm]{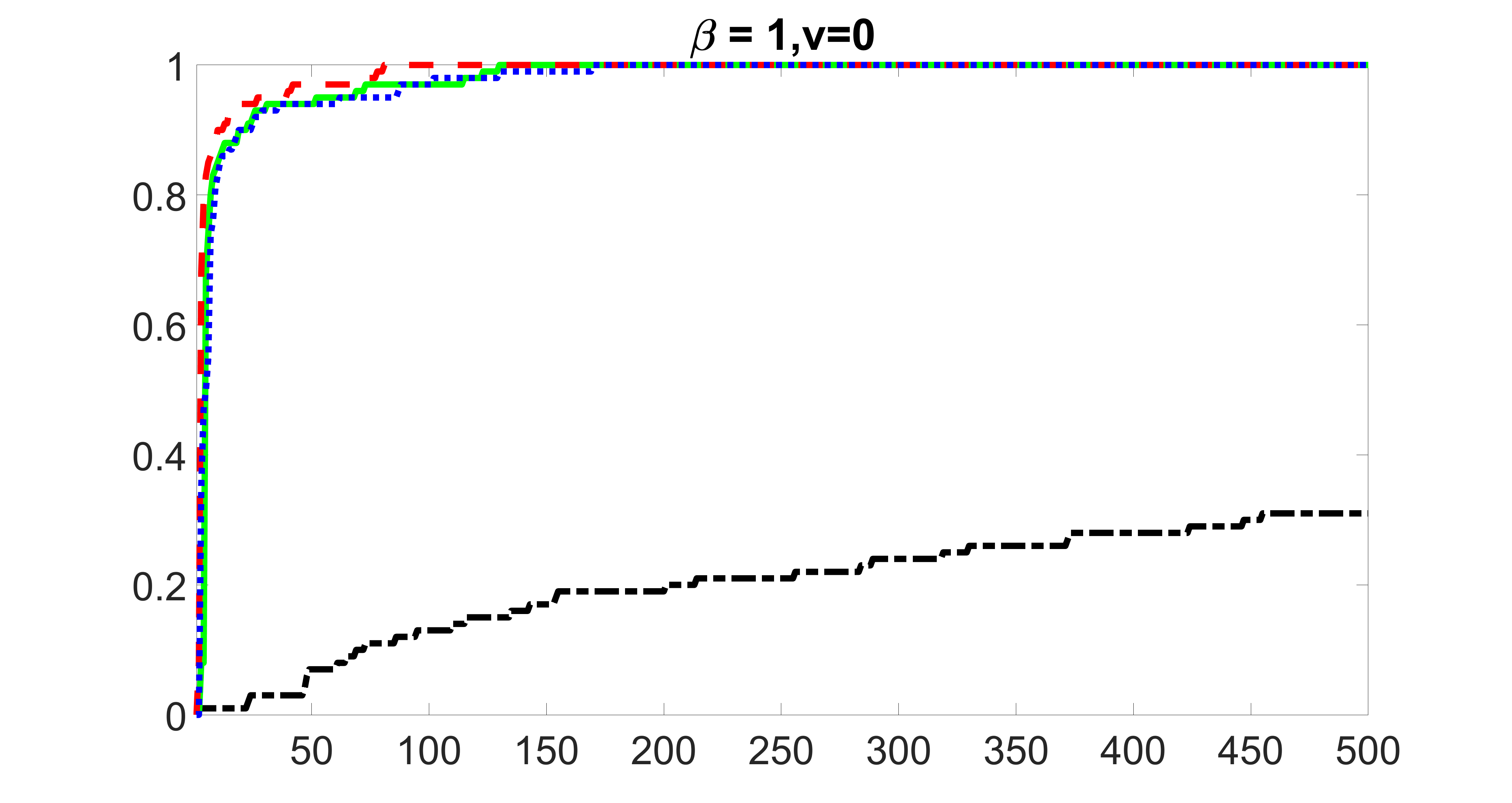}}
  \subcaptionbox{\footnotesize Precision: weak \\ outcome, zero exposure}[0.45\linewidth]
 {\includegraphics[width=6cm,height=3.5cm]{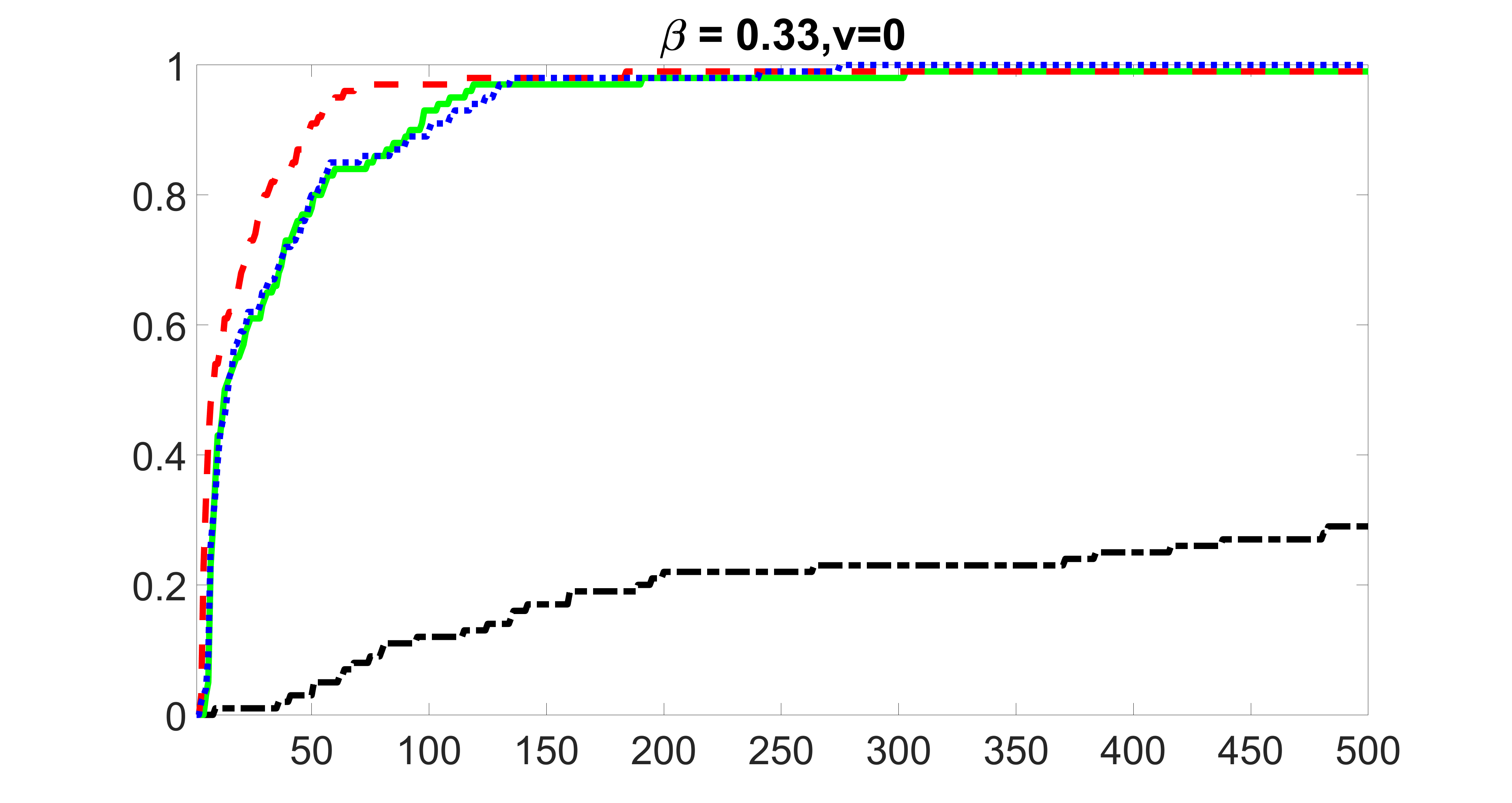}}
 \subcaptionbox{\footnotesize Precision: weaker \\ outcome, zero exposure}[0.45\linewidth]
 {\includegraphics[width=6cm,height=3.5cm]{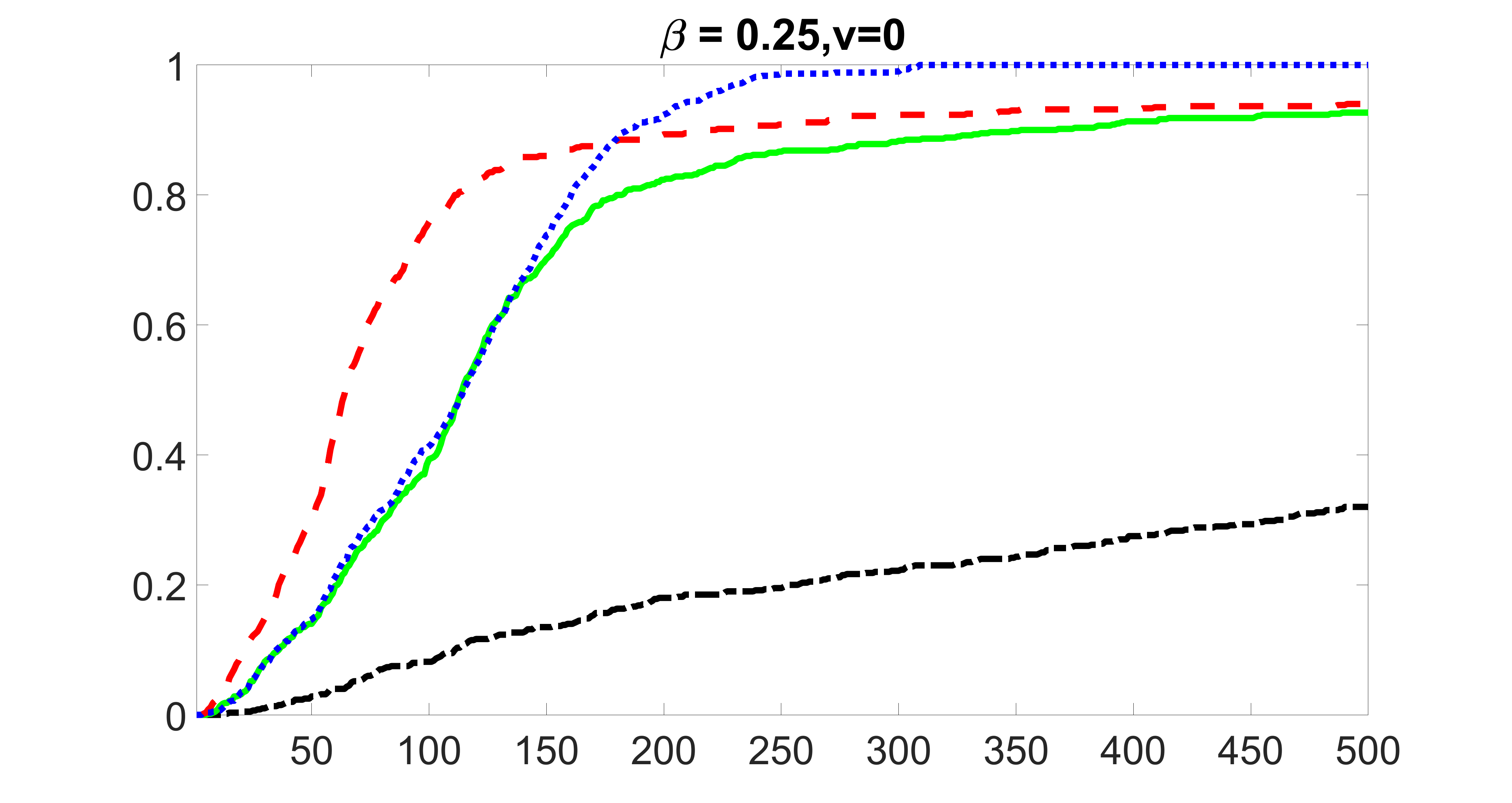}}
  \subcaptionbox{Overall coverage of $\mathcal{M}_1$}[0.45\linewidth]
 {\includegraphics[width=6cm,height=3.5cm]{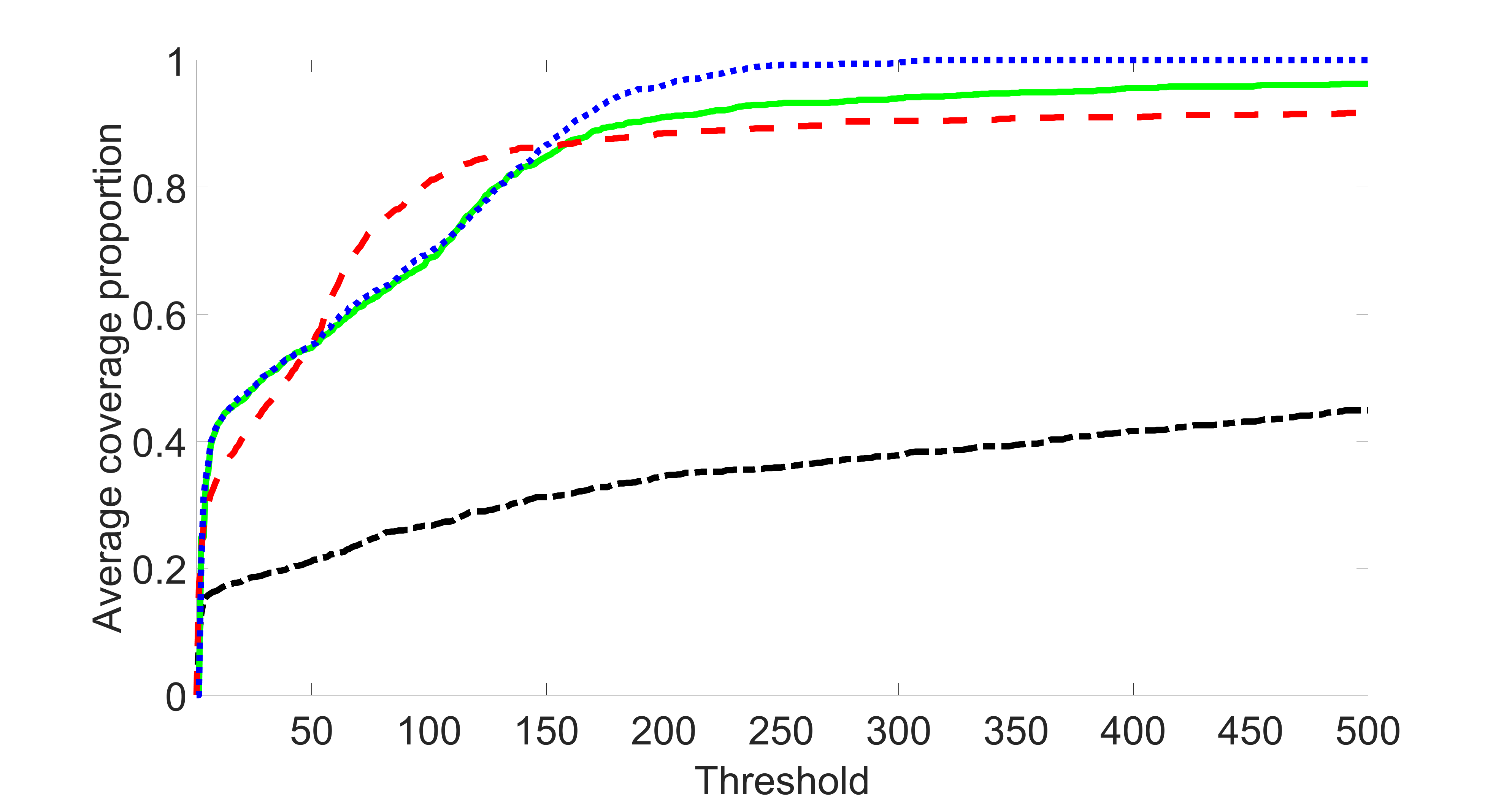}}
\caption{ Simulation results for the case $(n,s,K,\sigma) = (1000,5000,6,1)$: Panels (a) -- (g) plot the average coverage proportion for $X_l$, where $l \in \mathcal{M}_1 =  \{1,2,3,104,105, 106\} \cup \mathcal{P}_{LD}$. Panels (a) -- (c) correspond to strong outcome and weak exposure predictor, moderate outcome and moderate exposure predictor and weak outcome and strong exposure predictor; Panels (d) -- (g) correspond to strong, moderate, and weak predictors of outcome only. Panel (g) plots the average coverage proportion for the index set $\mathcal{P}_{LD}$. Panel (h) plots the average coverage proportion for the index set $\mathcal{M}_1$. The x-axis represents the size of $\widehat{\mathcal{M}} $, while
y-axis denotes the average proportion. The blue dot, green solid, red dashed and black dash dotted lines denote the blockwise joint screening, joint screening, outcome screening, and intersection screening methods, respectively.}
\label{sim3step1n1000sizesig6sigma1}
\end{figure}

\begin{figure}[htbp]
\captionsetup[subfigure]{justification=centering}
\centering
 \subcaptionbox{\footnotesize Confounder: strong \\ outcome, weak exposure}[0.45\linewidth]
 {\includegraphics[width=6cm,height=3.5cm]{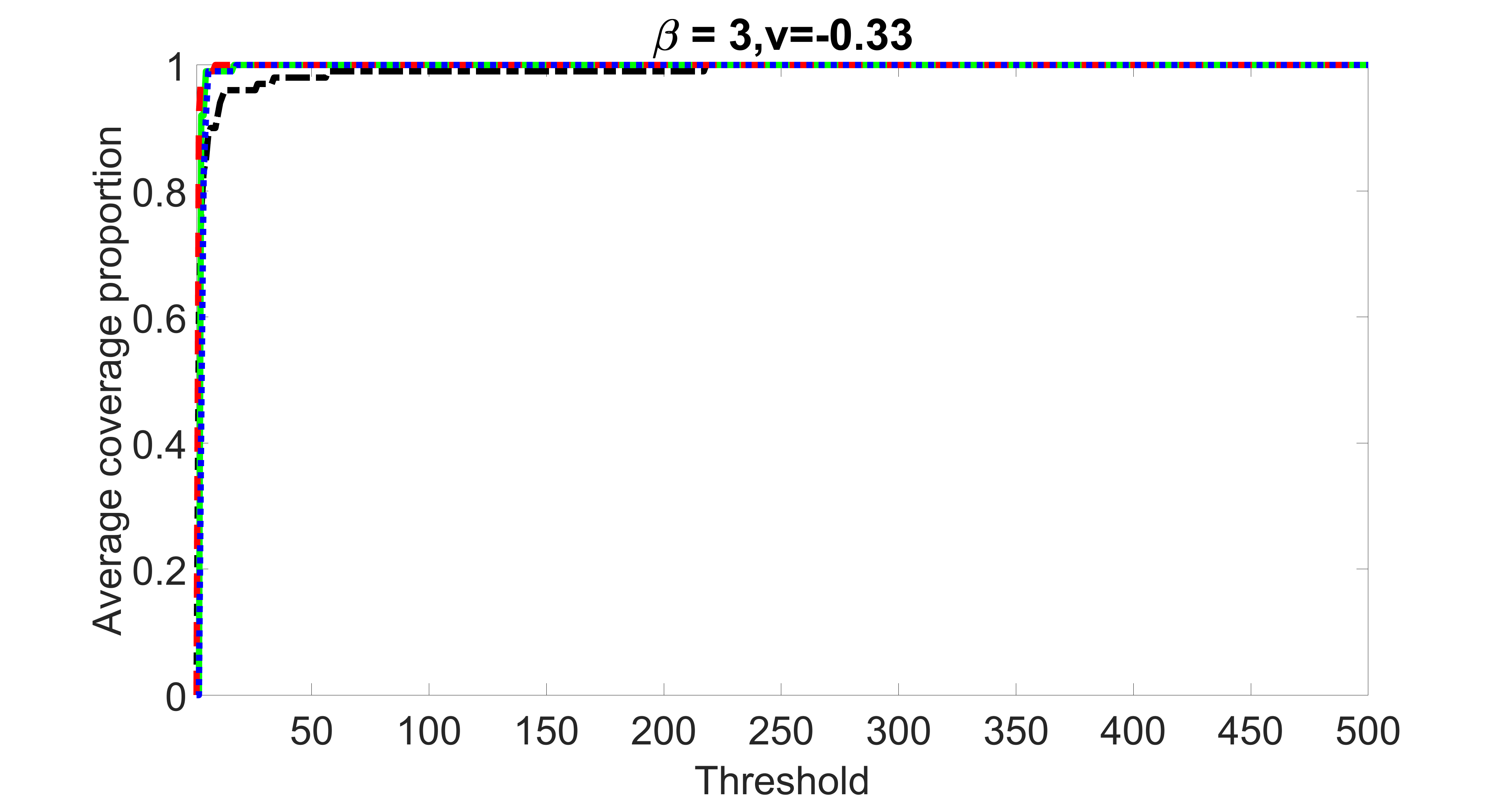}}
 \subcaptionbox{\footnotesize Confounder: medium \\ outcome, medium exposure}[0.45\linewidth]
 {\includegraphics[width=6cm,height=3.5cm]{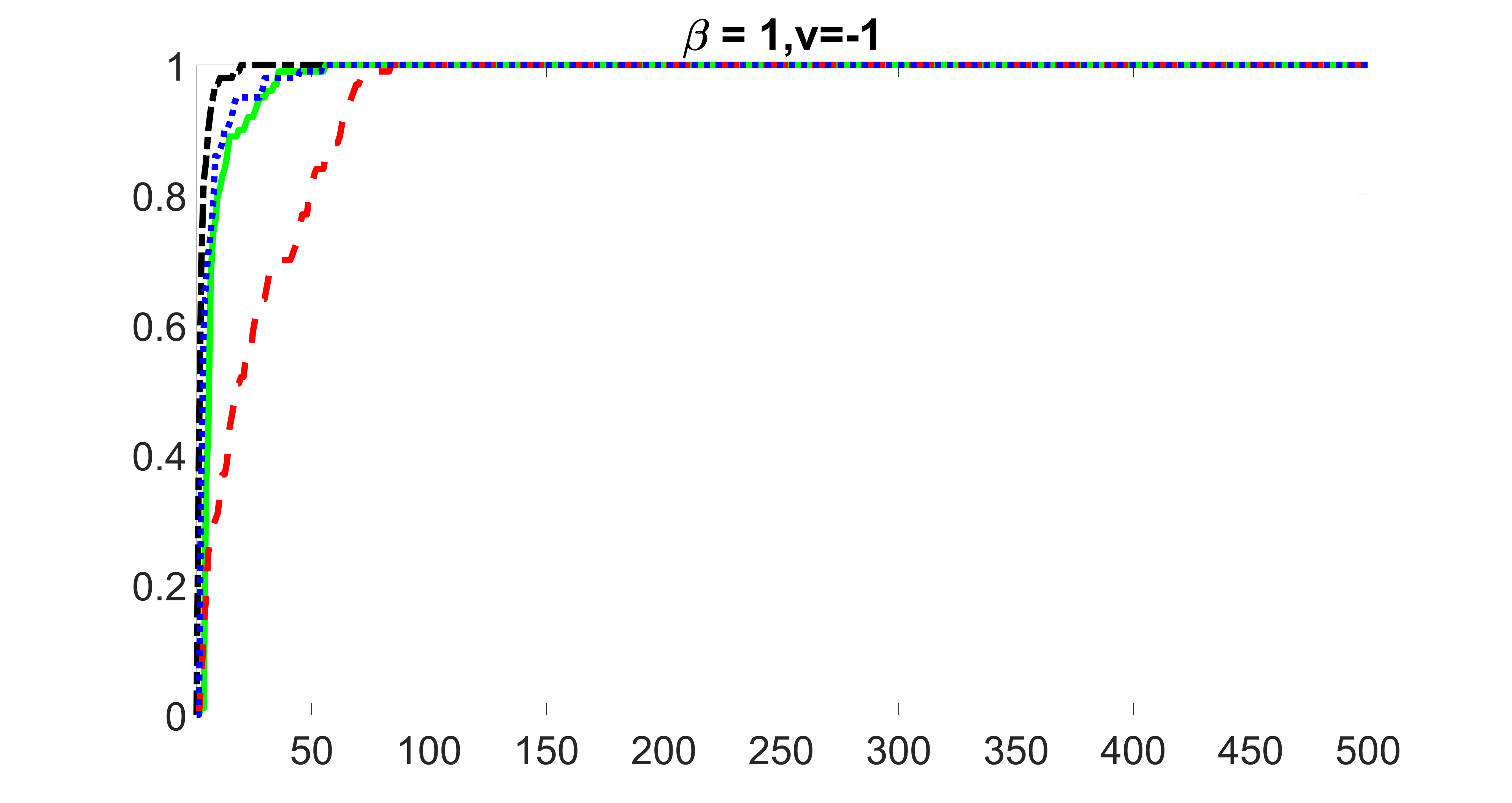}}
  \subcaptionbox{\footnotesize Confounder: weak \\ outcome, strong exposure}[0.45\linewidth]
 {\includegraphics[width=6cm,height=3.5cm]{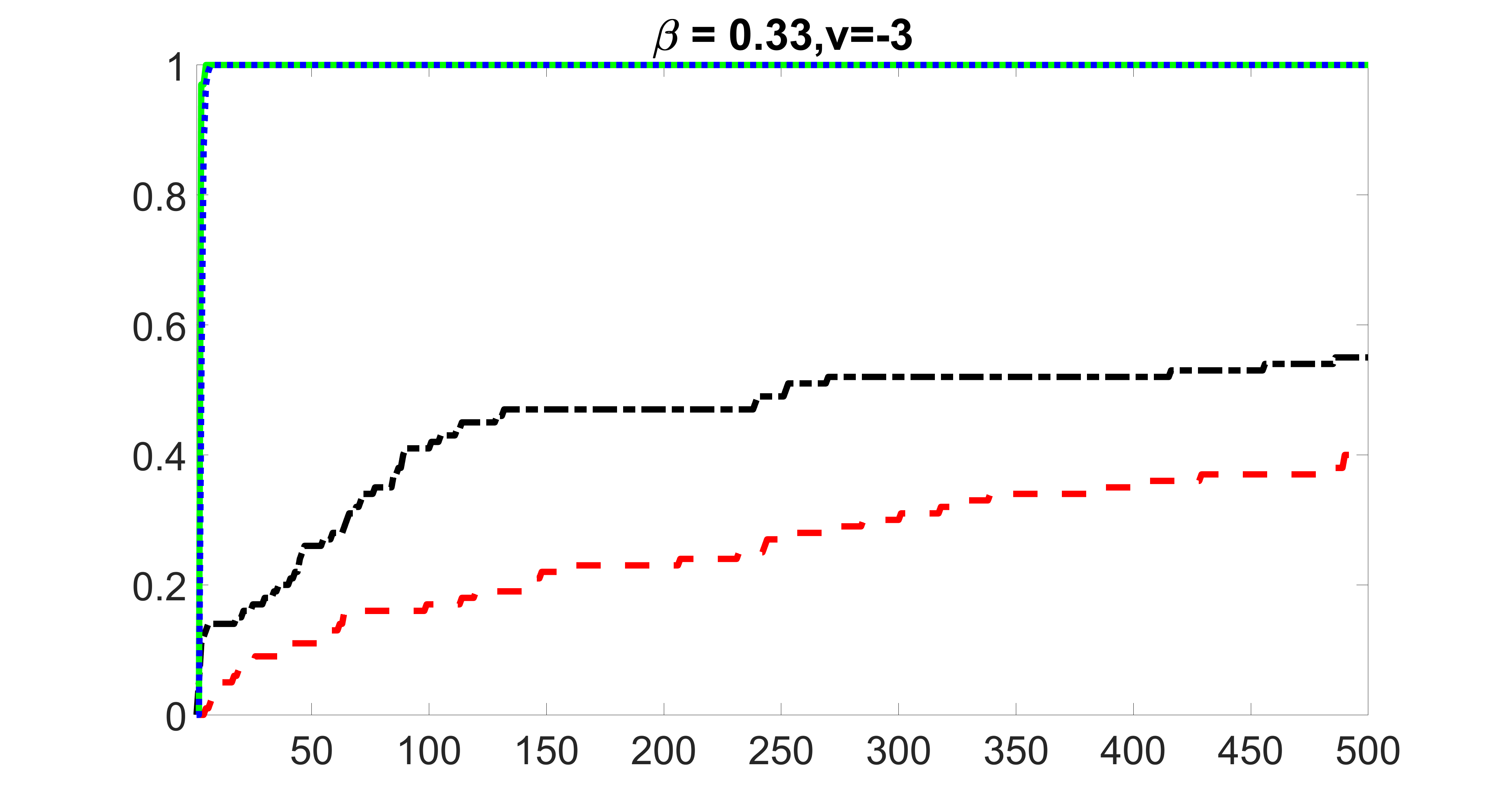}}
  \subcaptionbox{\footnotesize Precision: strong \\ outcome, zero exposure}[0.45\linewidth]
 {\includegraphics[width=6cm,height=3.5cm]{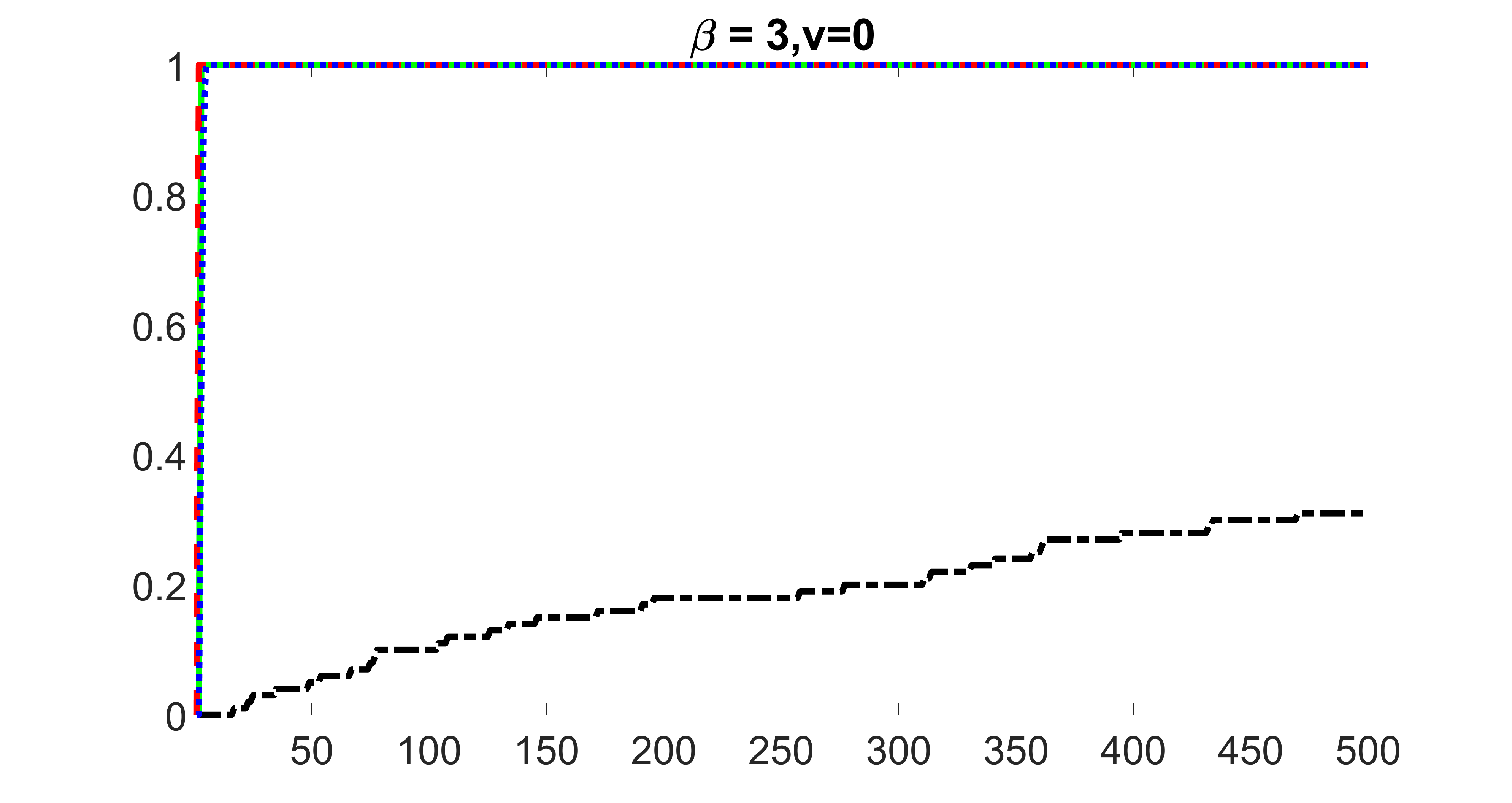}}
  \subcaptionbox{\footnotesize Precision: medium \\ outcome, zero exposure}[0.45\linewidth]
 {\includegraphics[width=6cm,height=3.5cm]{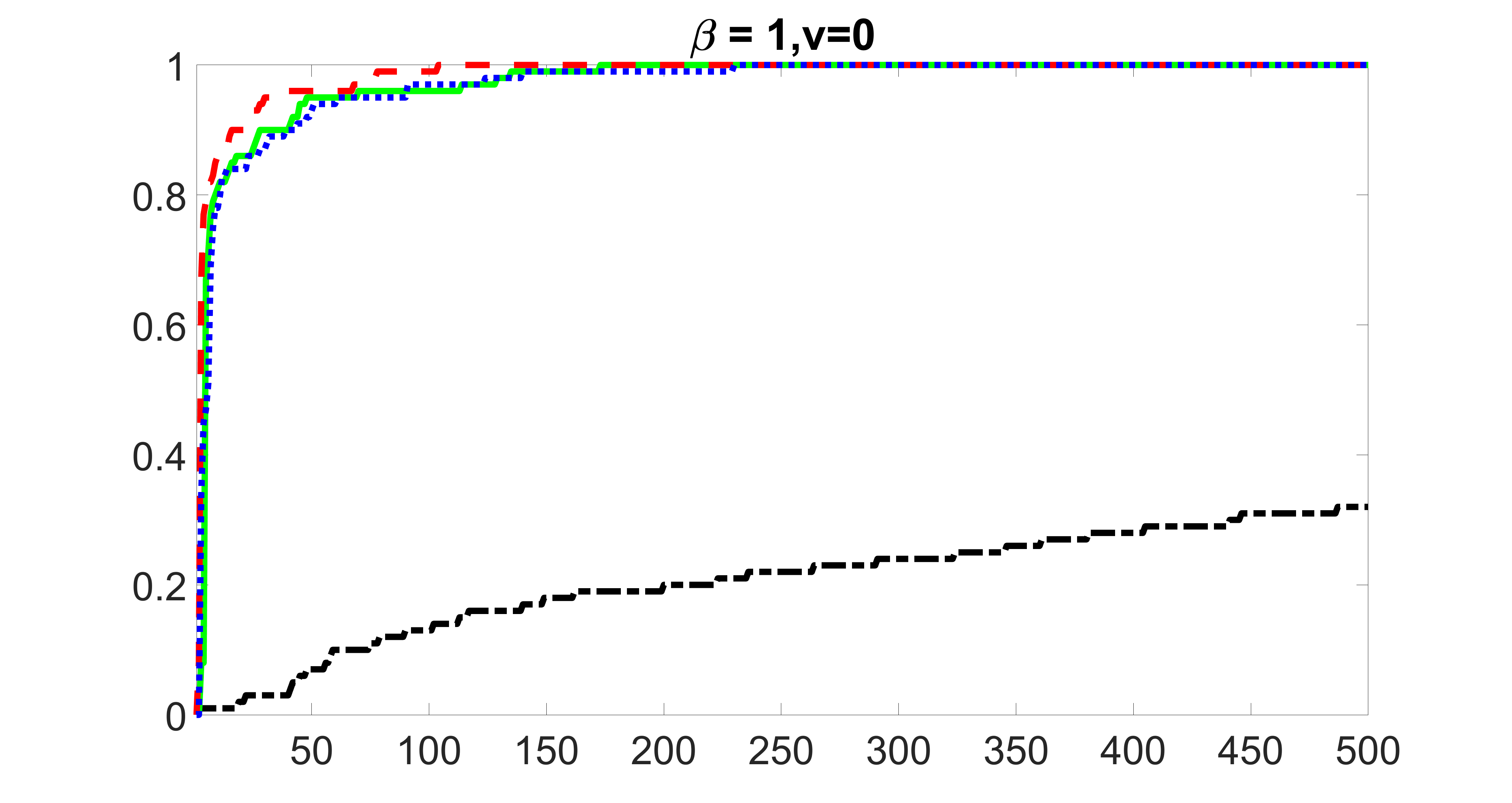}}
  \subcaptionbox{\footnotesize Precision: weak \\ outcome, zero exposure}[0.45\linewidth]
 {\includegraphics[width=6cm,height=3.5cm]{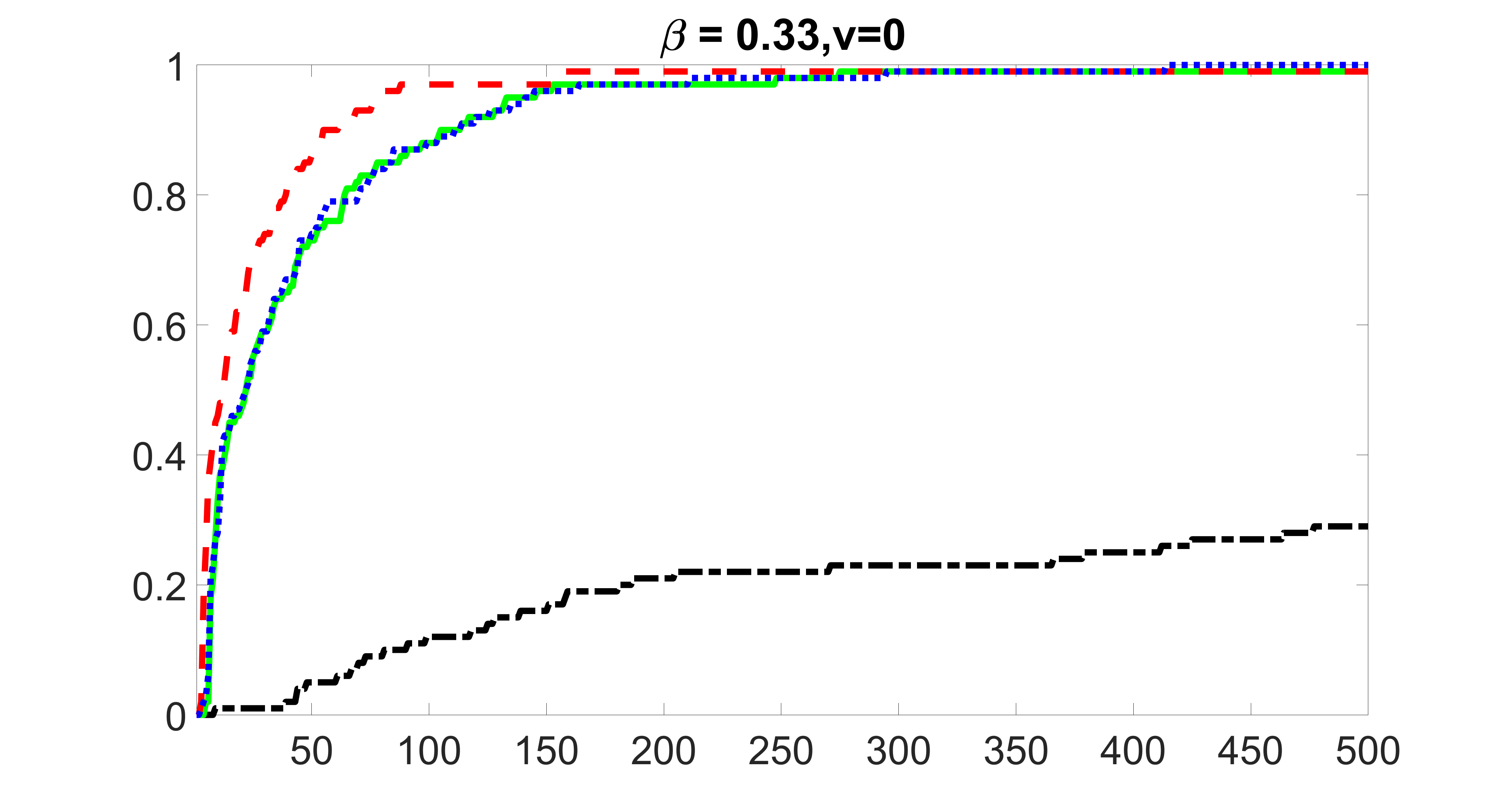}}
 \subcaptionbox{\footnotesize Precision: weaker \\ outcome, zero exposure}[0.45\linewidth]
 {\includegraphics[width=6cm,height=3.5cm]{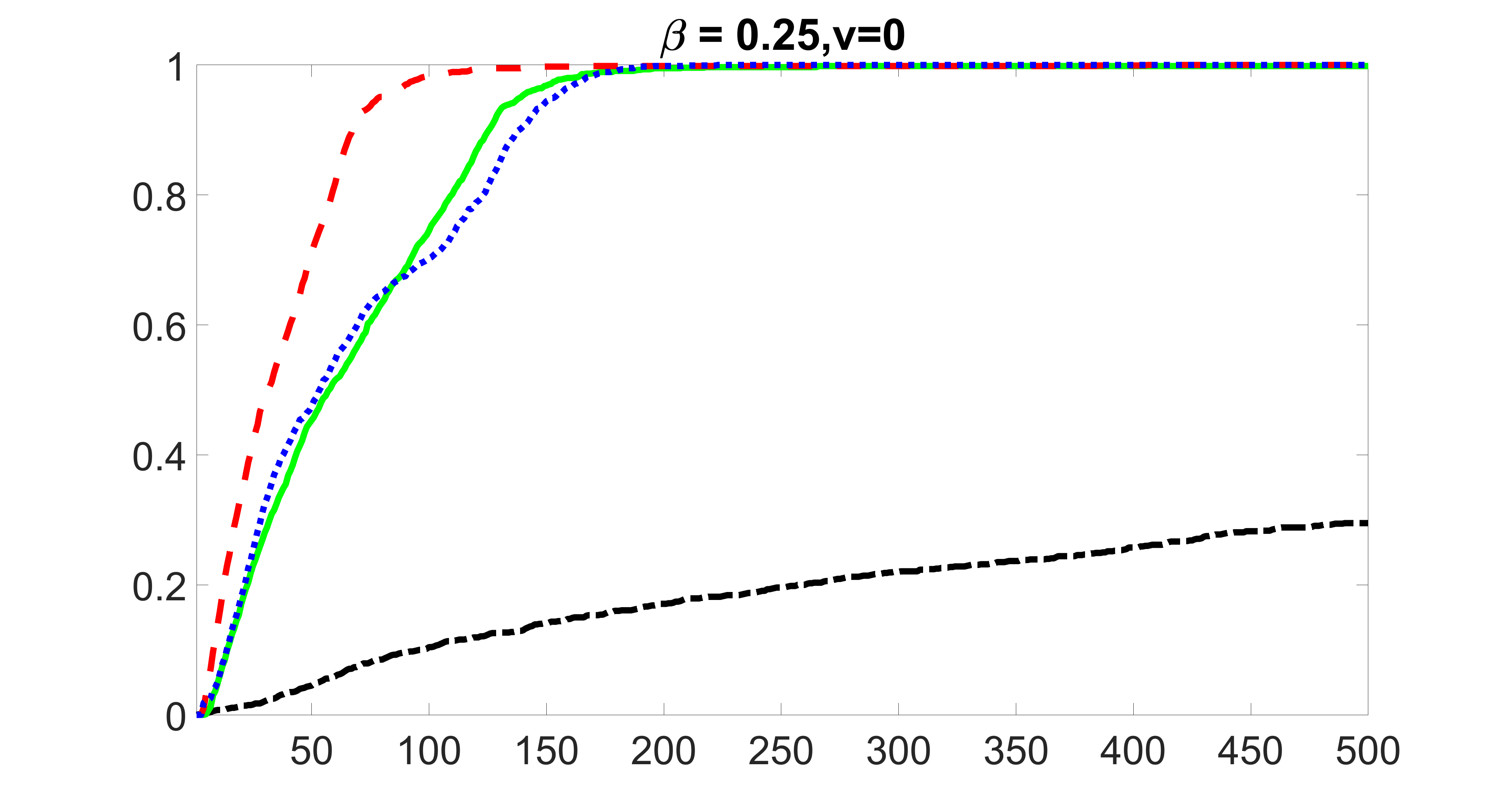}}
  \subcaptionbox{Overall coverage of $\mathcal{M}_1$}[0.45\linewidth]
 {\includegraphics[width=6cm,height=3.5cm]{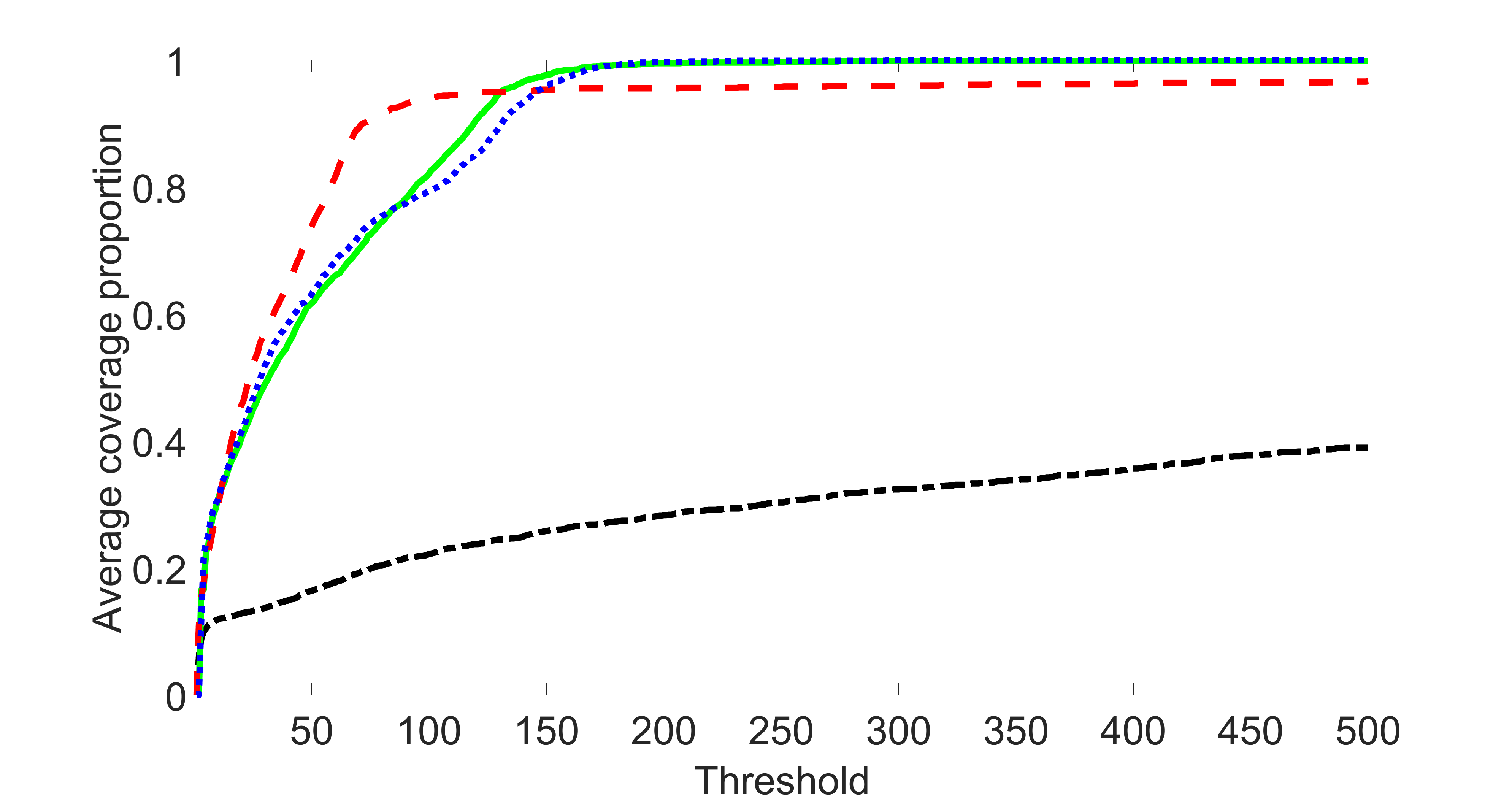}}
\caption{ Simulation results for the case $(n,s,K,\sigma) = (1000,5000,12,1)$: Panels (a) -- (g) plot the average coverage proportion for $X_l$, where $l \in \mathcal{M}_1 =  \{1,2,3,104,105, 106\} \cup \mathcal{P}_{LD}$. Panels (a) -- (c) correspond to strong outcome and weak exposure predictor, moderate outcome and moderate exposure predictor and weak outcome and strong exposure predictor; Panels (d) -- (g) correspond to strong, moderate, and weak predictors of outcome only. Panel (g) plots the average coverage proportion for the index set $\mathcal{P}_{LD}$. Panel (h) plots the average coverage proportion for the index set $\mathcal{M}_1$. The x-axis represents the size of $\widehat{\mathcal{M}} $, while
y-axis denotes the average proportion. The blue dot, green solid, red dashed and black dash dotted lines denote the blockwise joint screening, joint screening, outcome screening, and intersection screening methods, respectively.}
\label{sim3step1n1000sizesig12sigma1}
\end{figure}

\begin{figure}[htbp]
\captionsetup[subfigure]{justification=centering}
\centering
 \subcaptionbox{\footnotesize Confounder: strong \\ outcome, weak exposure}[0.45\linewidth]
 {\includegraphics[width=6cm,height=3.5cm]{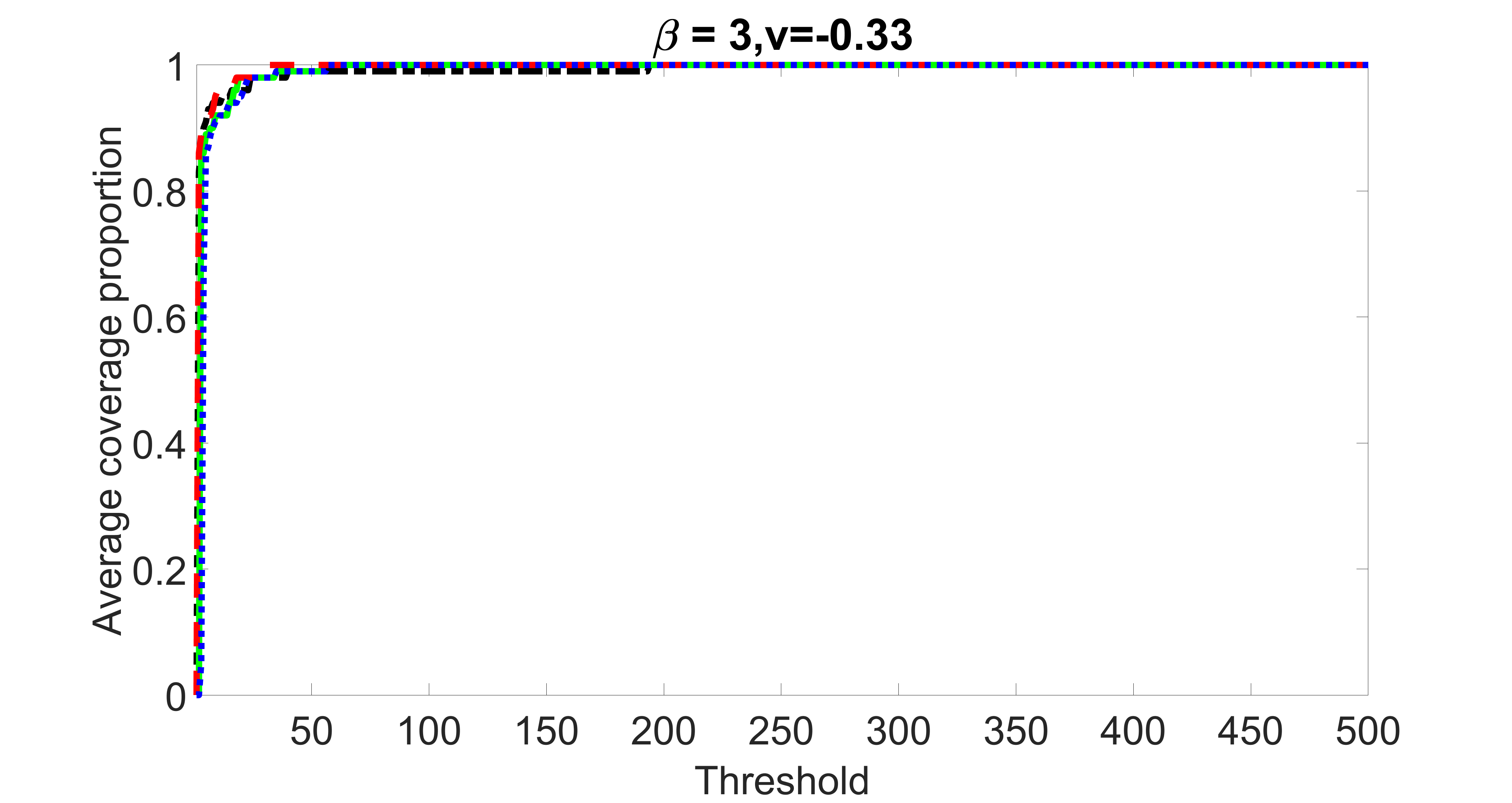}}
 \subcaptionbox{\footnotesize Confounder: medium \\ outcome, medium exposure}[0.45\linewidth]
 {\includegraphics[width=6cm,height=3.5cm]{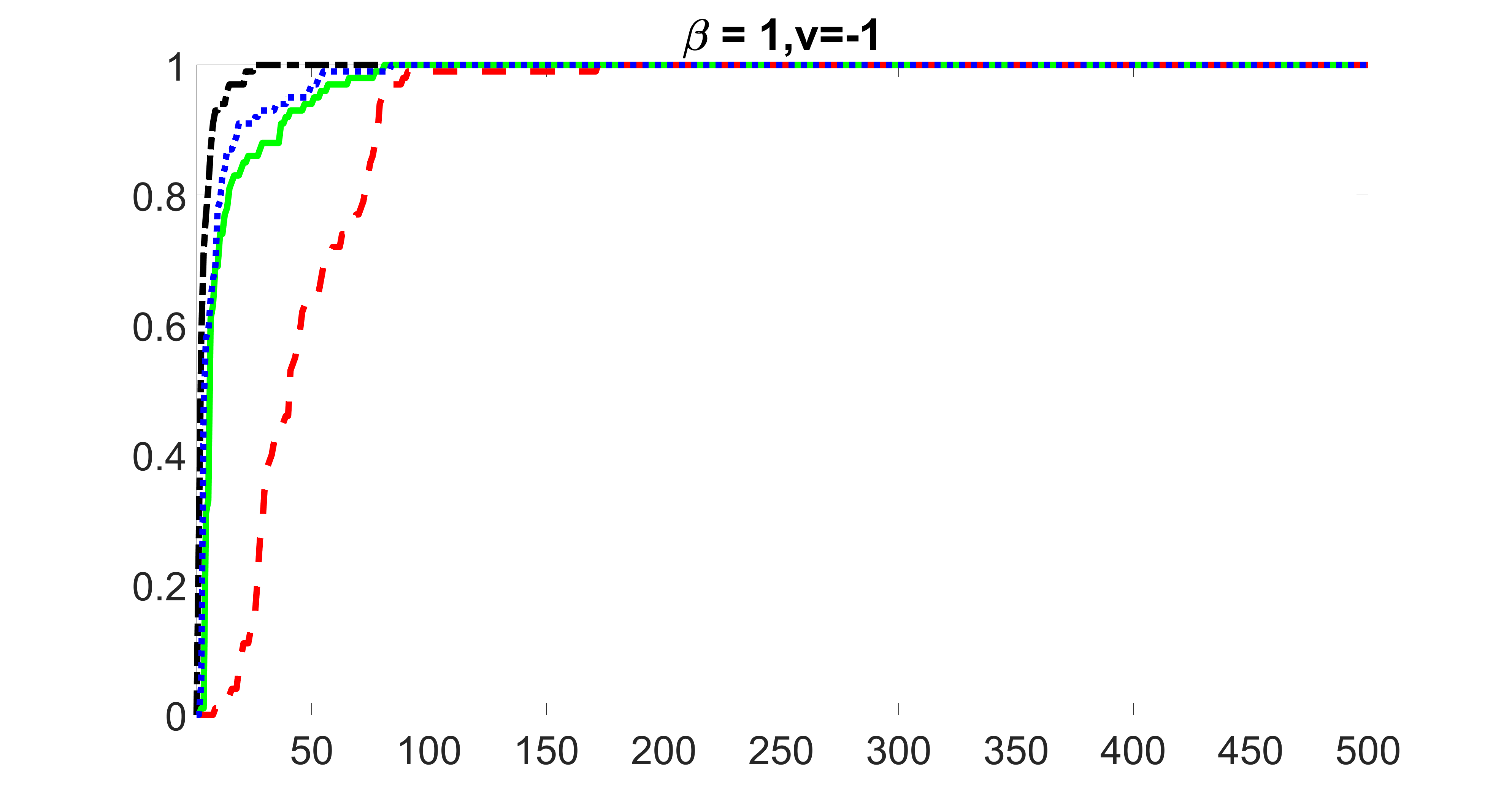}}
  \subcaptionbox{\footnotesize Confounder: weak \\ outcome, strong exposure}[0.45\linewidth]
 {\includegraphics[width=6cm,height=3.5cm]{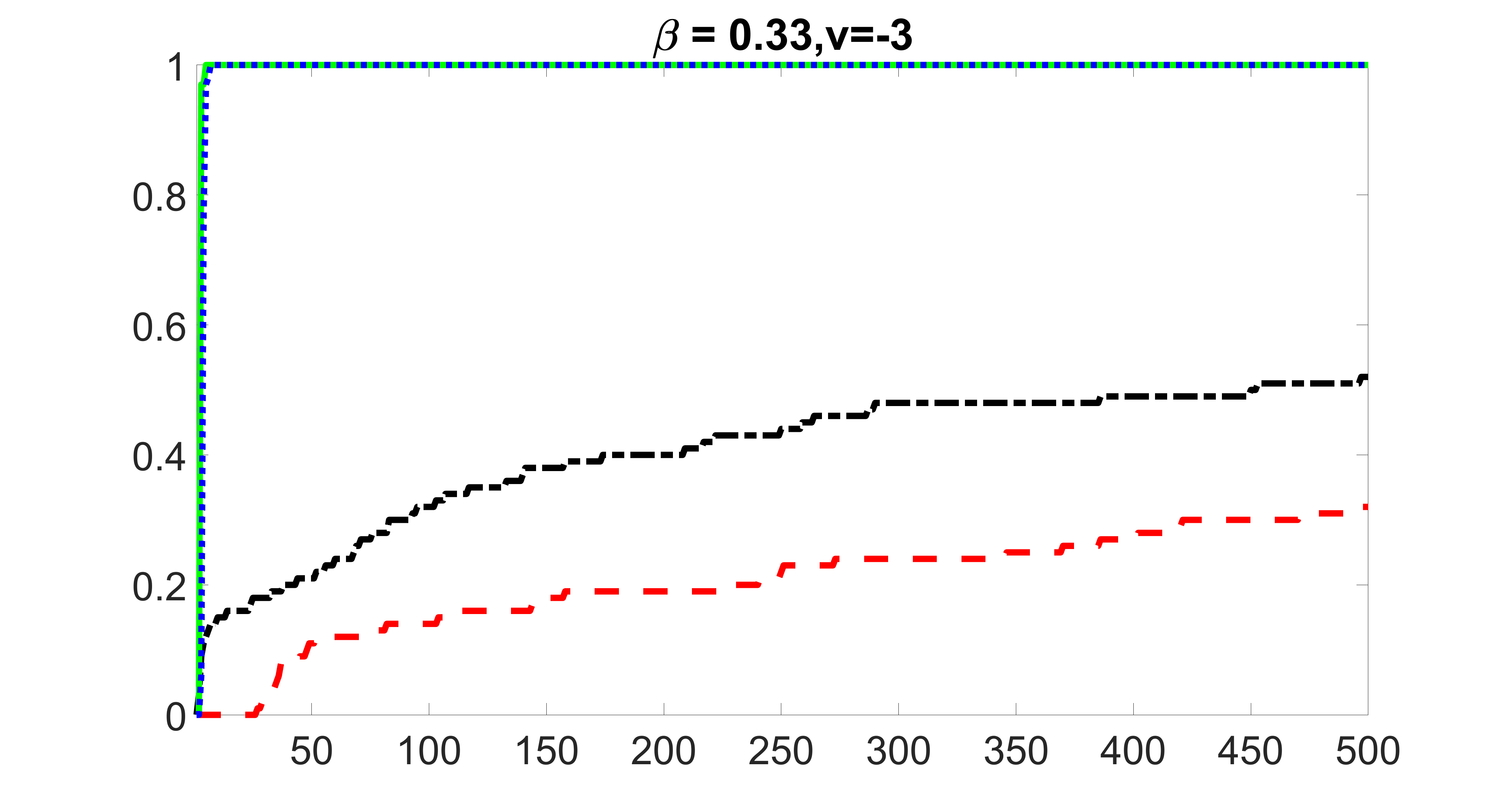}}
  \subcaptionbox{\footnotesize Precision: strong \\ outcome, zero exposure}[0.45\linewidth]
 {\includegraphics[width=6cm,height=3.5cm]{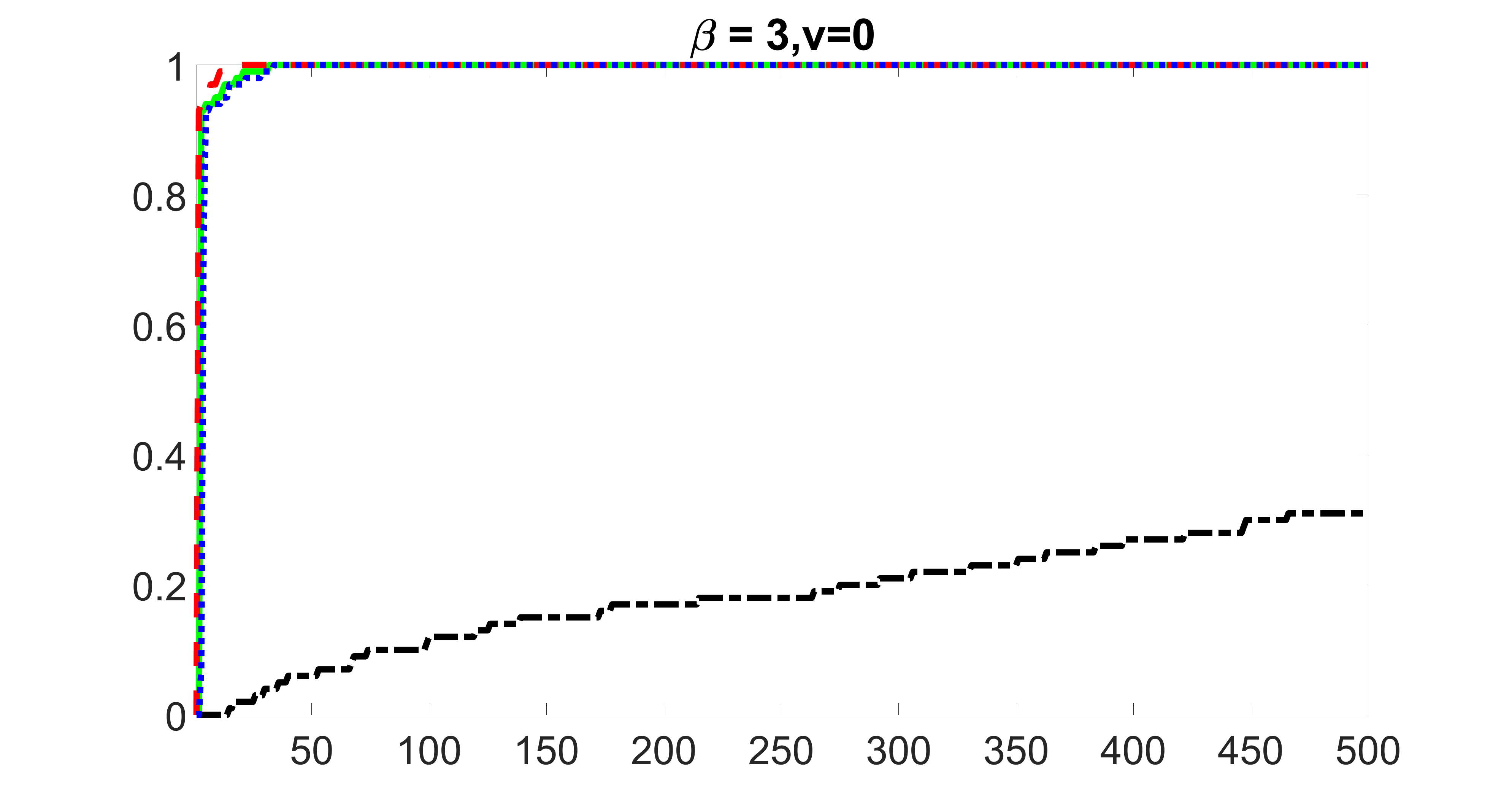}}
  \subcaptionbox{\footnotesize Precision: medium \\ outcome, zero exposure}[0.45\linewidth]
 {\includegraphics[width=6cm,height=3.5cm]{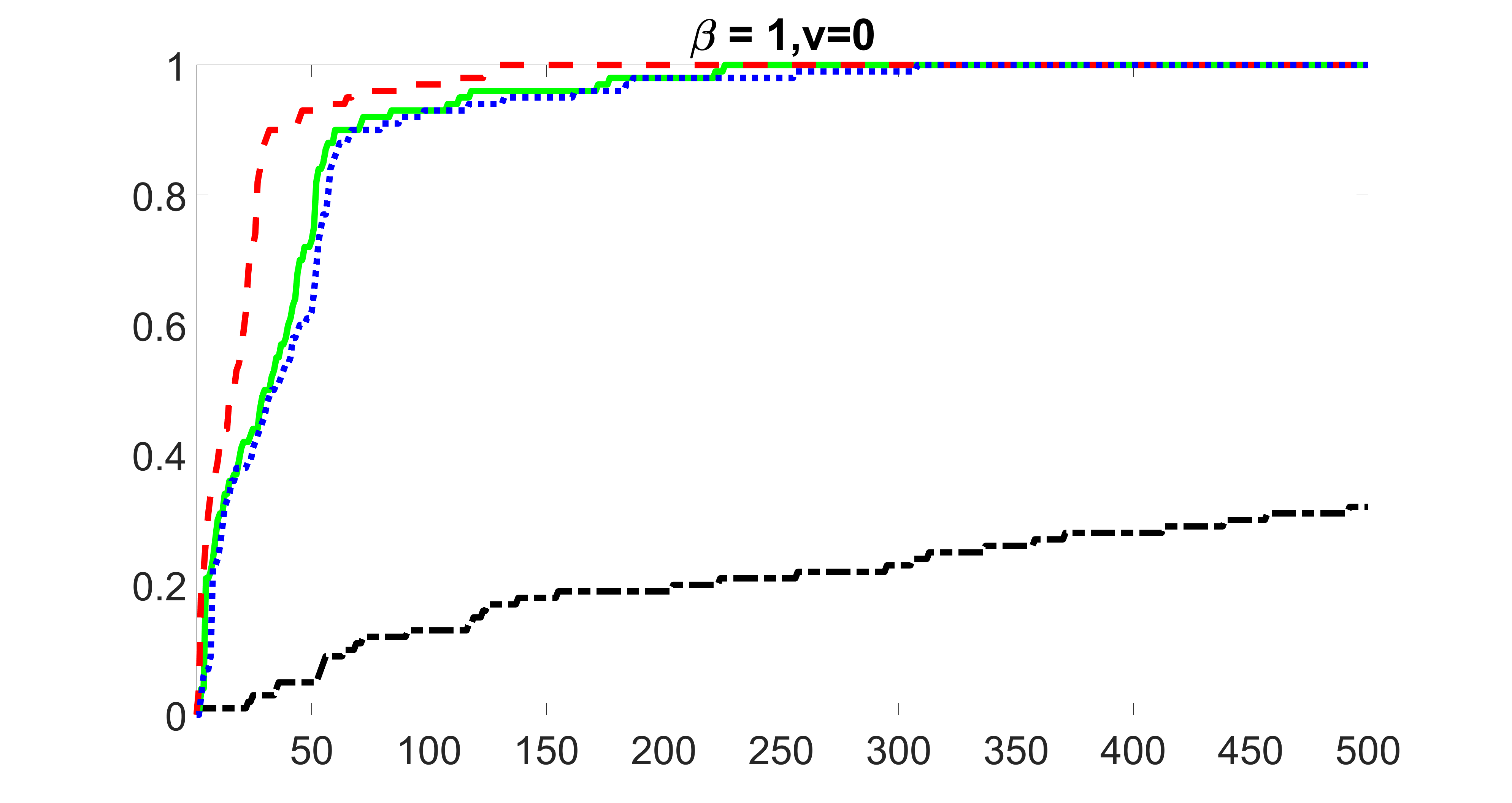}}
  \subcaptionbox{\footnotesize Precision: weak \\ outcome, zero exposure}[0.45\linewidth]
 {\includegraphics[width=6cm,height=3.5cm]{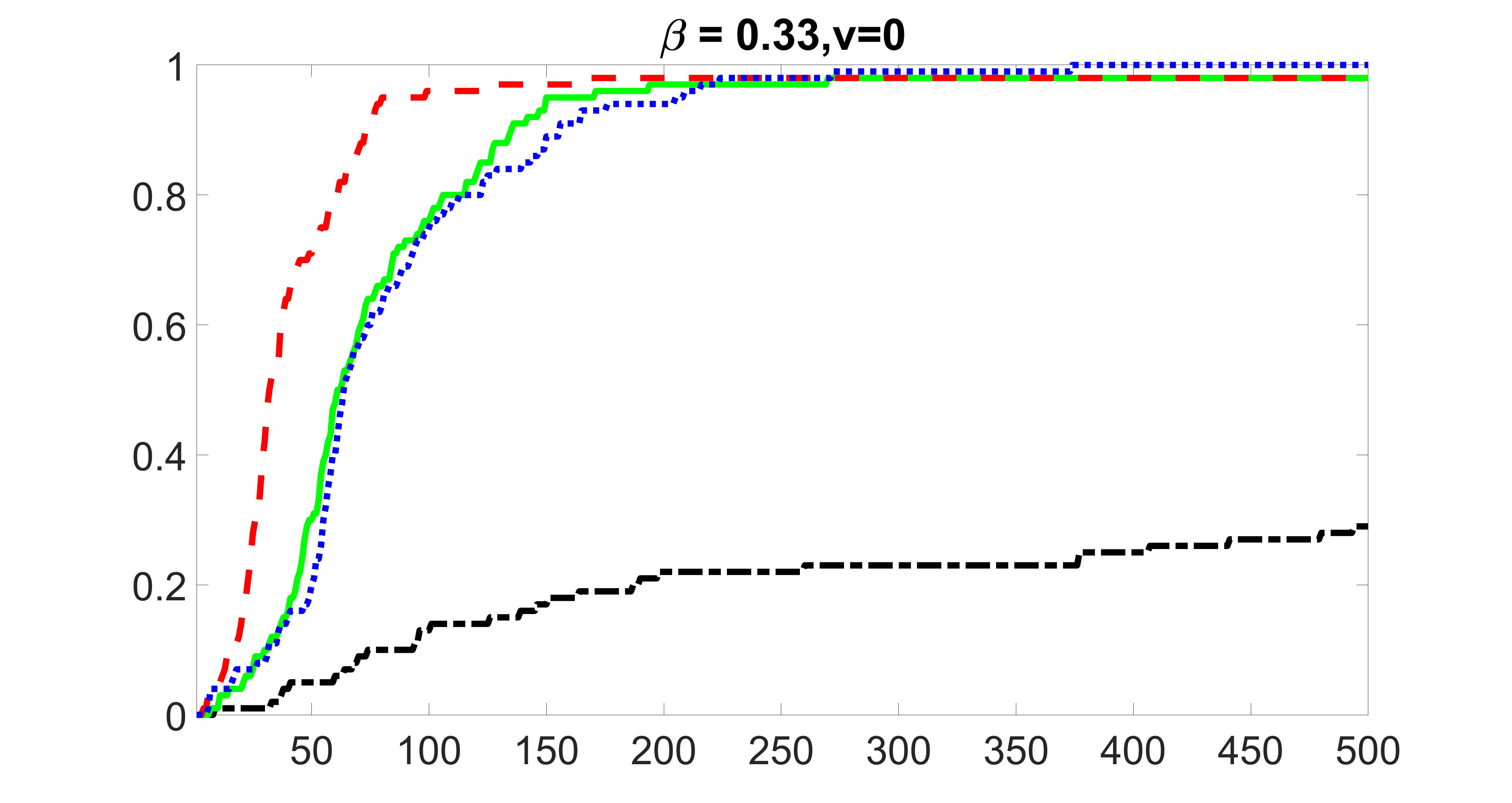}}
 \subcaptionbox{\footnotesize Precision: weaker \\ outcome, zero exposure}[0.45\linewidth]
 {\includegraphics[width=6cm,height=3.5cm]{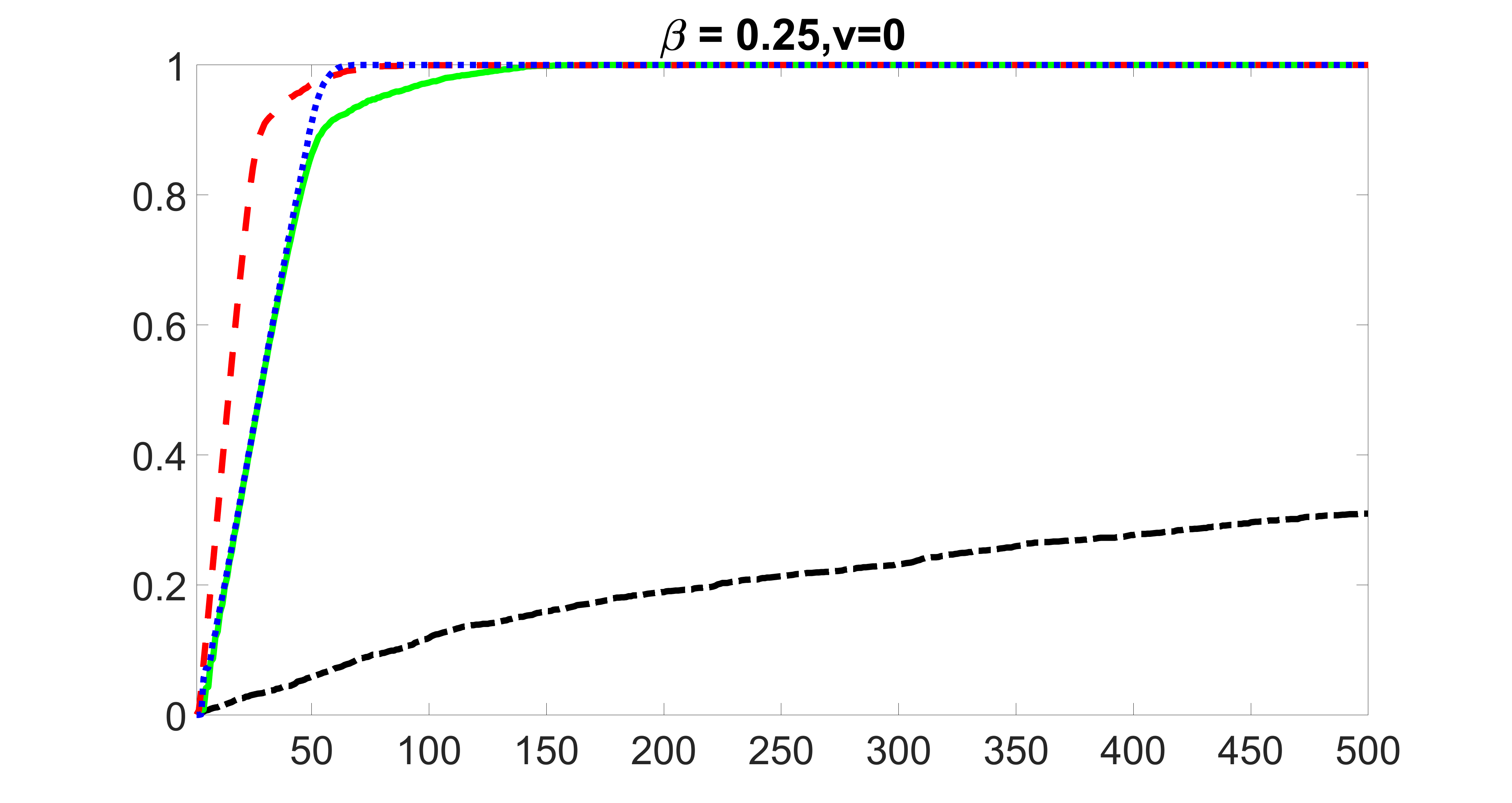}}
  \subcaptionbox{Overall coverage of $\mathcal{M}_1$}[0.45\linewidth]
 {\includegraphics[width=6cm,height=3.5cm]{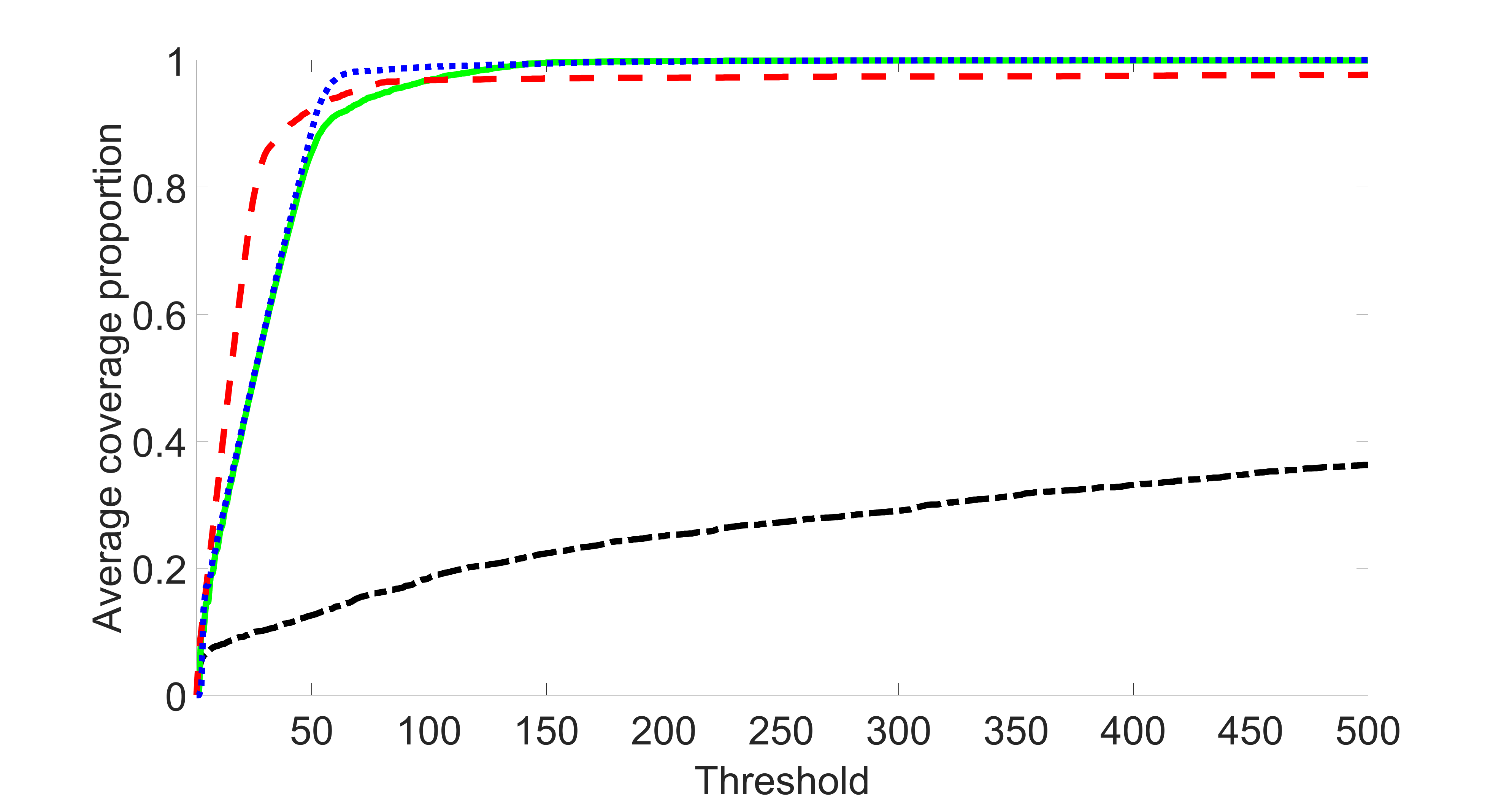}}
\caption{ Simulation results for the case $(n,s,K,\sigma) = (1000,5000,24,1)$: Panels (a) -- (g) plot the average coverage proportion for $X_l$, where $l \in \mathcal{M}_1 =  \{1,2,3,104,105, 106\} \cup \mathcal{P}_{LD}$. Panels (a) -- (c) correspond to strong outcome and weak exposure predictor, moderate outcome and moderate exposure predictor and weak outcome and strong exposure predictor; Panels (d) -- (g) correspond to strong, moderate, and weak predictors of outcome only. Panel (g) plots the average coverage proportion for the index set $\mathcal{P}_{LD}$. Panel (h) plots the average coverage proportion for the index set $\mathcal{M}_1$. The x-axis represents the size of $\widehat{\mathcal{M}} $, while
y-axis denotes the average proportion. The blue dot, green solid, red dashed and black dash dotted lines denote the blockwise joint screening, joint screening, outcome screening, and intersection screening methods, respectively.}
\label{sim3step1n1000sizesig24sigma1}
\end{figure}

\begin{figure}[htbp]
\captionsetup[subfigure]{justification=centering}
\centering
 \subcaptionbox{\footnotesize Confounder: strong \\ outcome, weak exposure}[0.45\linewidth]
 {\includegraphics[width=6cm,height=3.5cm]{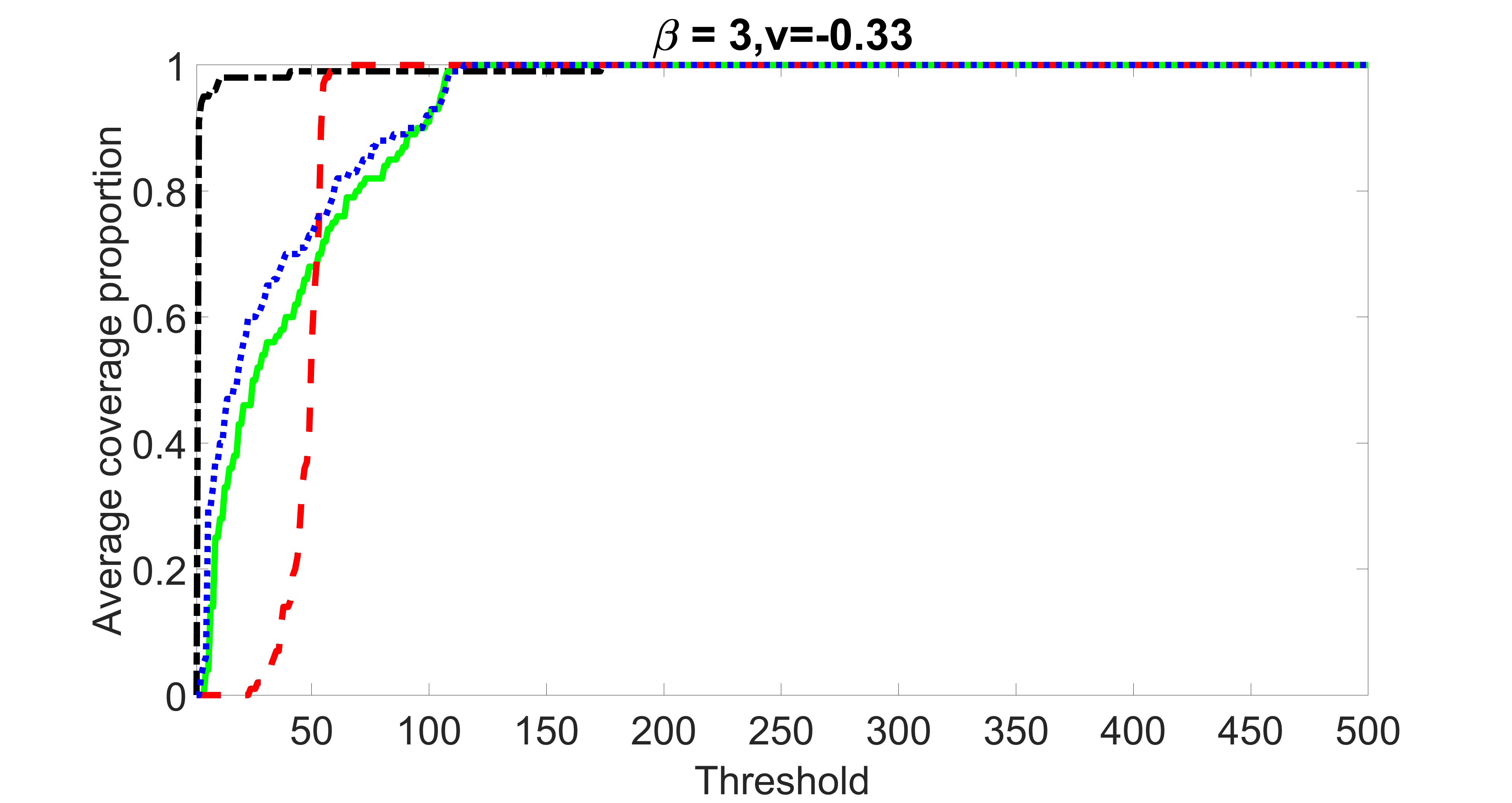}}
 \subcaptionbox{\footnotesize Confounder: medium \\ outcome, medium exposure}[0.45\linewidth]
 {\includegraphics[width=6cm,height=3.5cm]{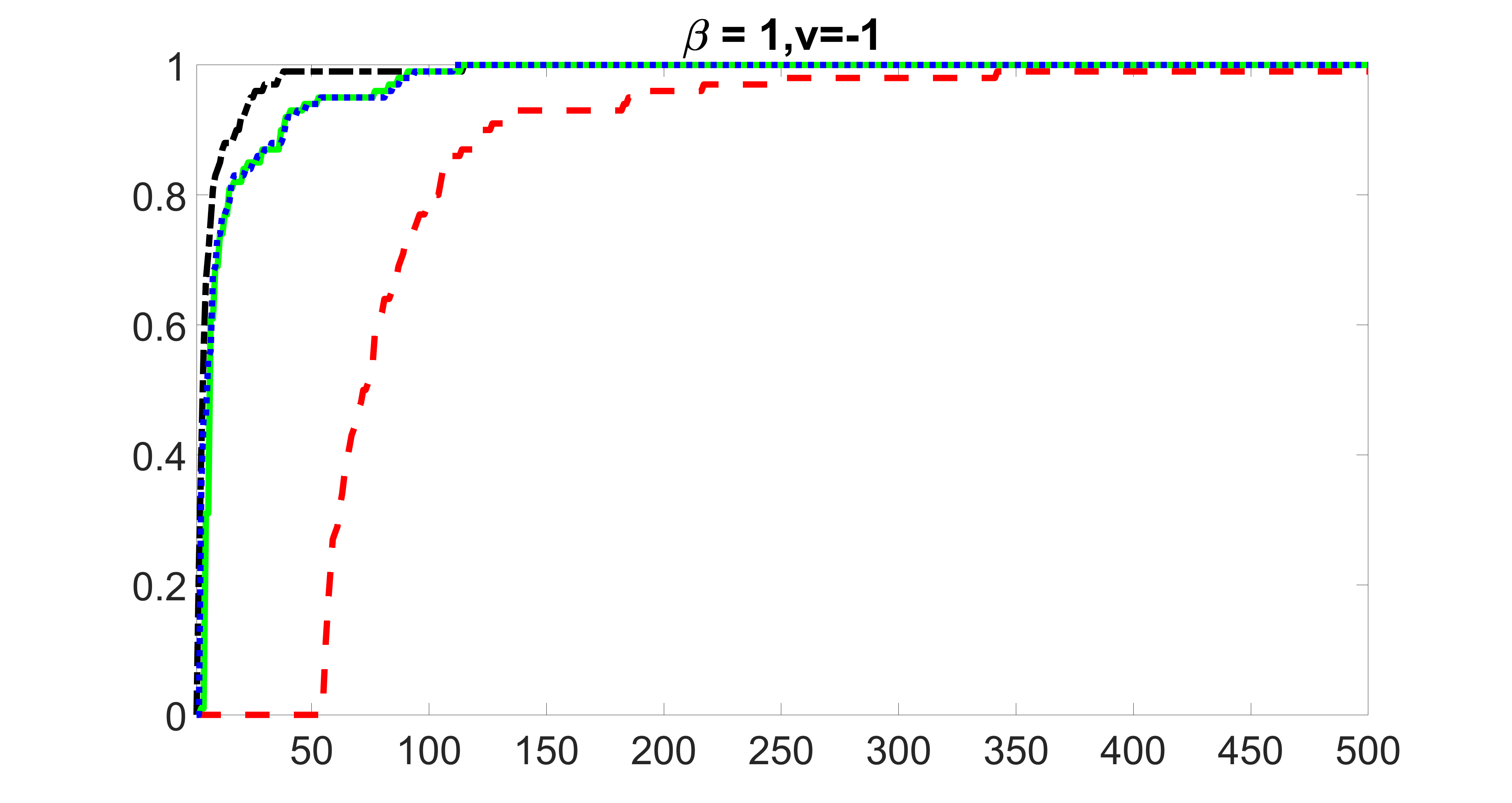}}
  \subcaptionbox{\footnotesize Confounder: weak \\ outcome, strong exposure}[0.45\linewidth]
 {\includegraphics[width=6cm,height=3.5cm]{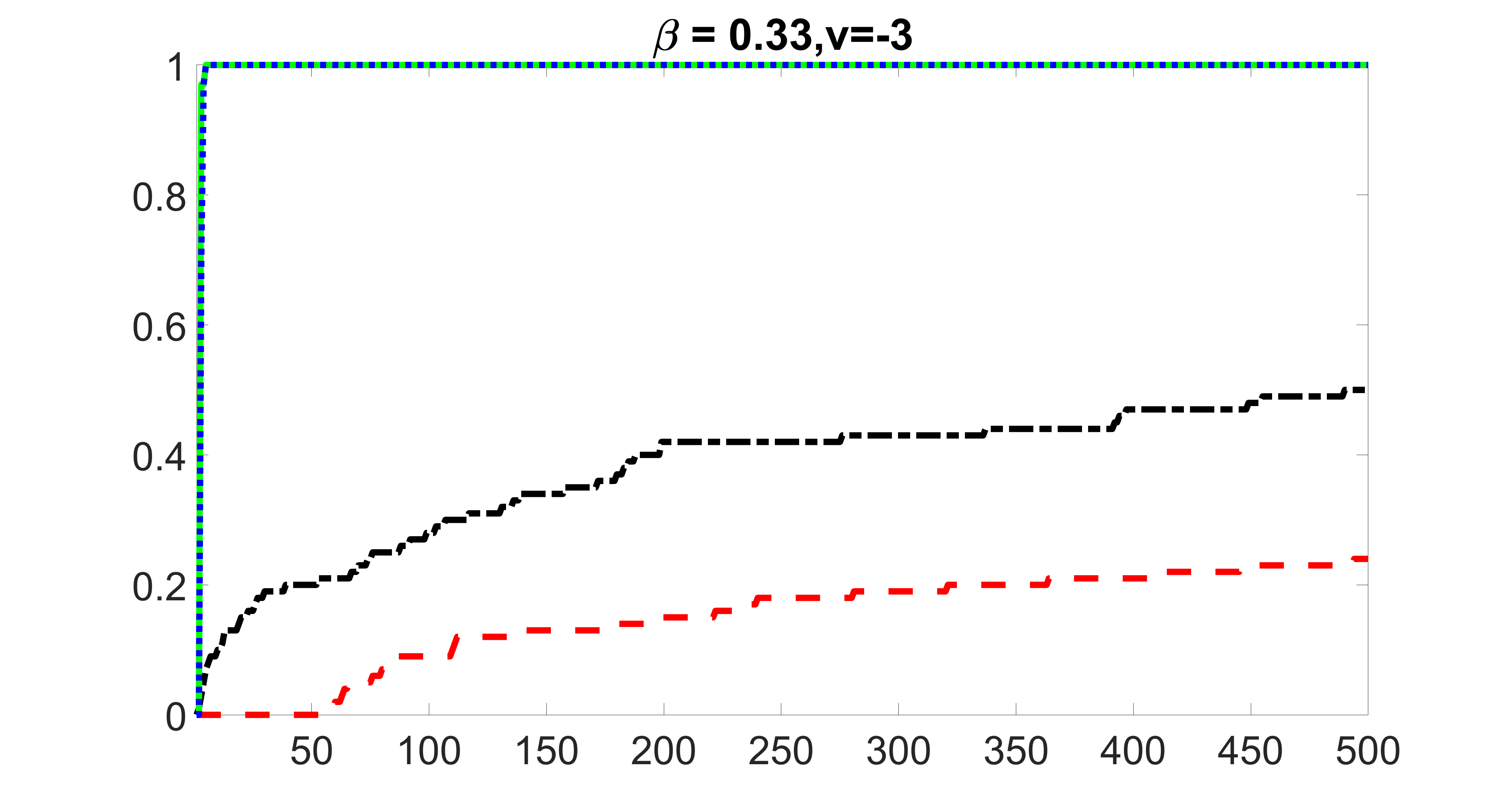}}
  \subcaptionbox{\footnotesize Precision: strong \\ outcome, zero exposure}[0.45\linewidth]
 {\includegraphics[width=6cm,height=3.5cm]{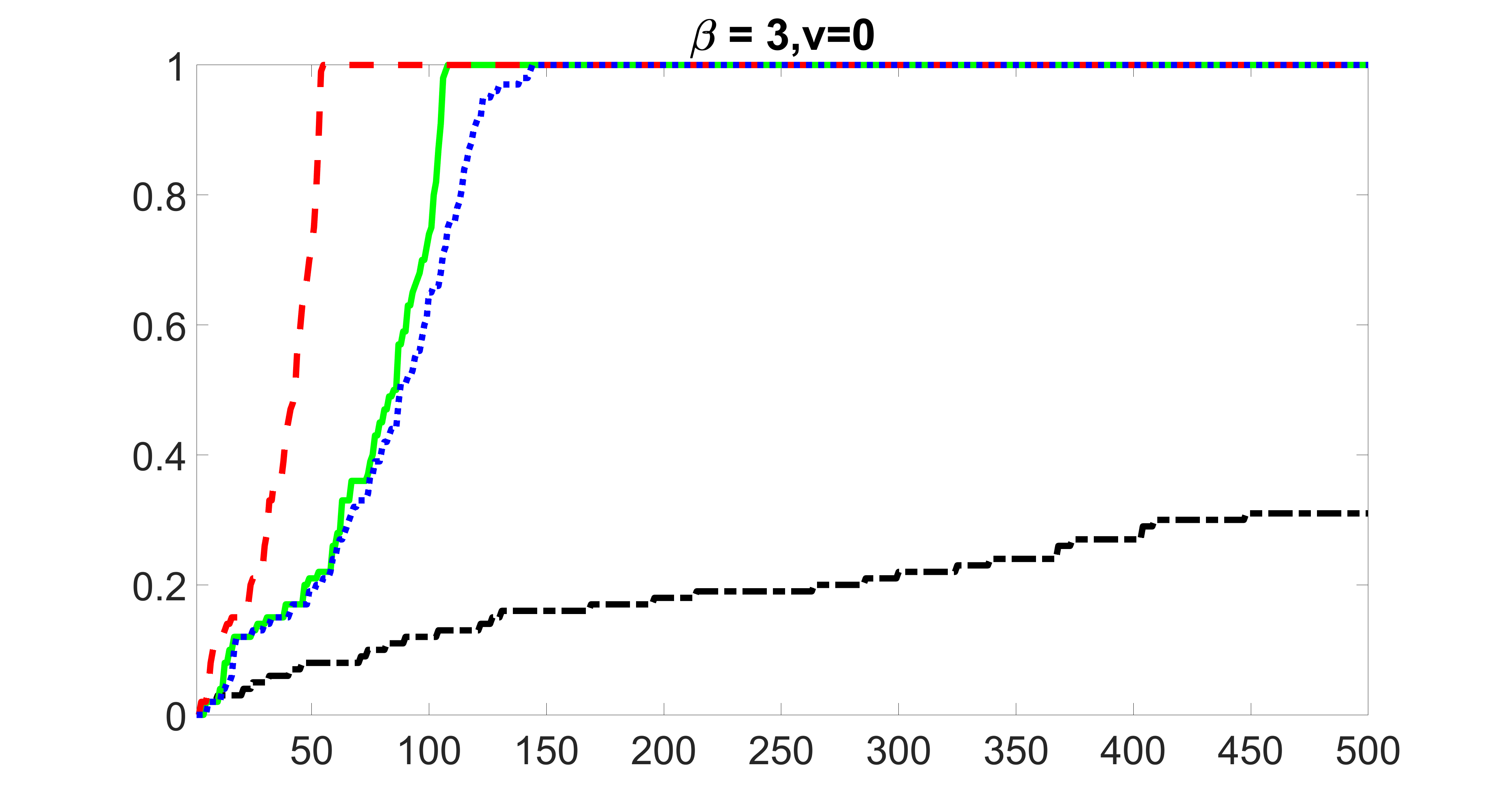}}
  \subcaptionbox{\footnotesize Precision: medium \\ outcome, zero exposure}[0.45\linewidth]
 {\includegraphics[width=6cm,height=3.5cm]{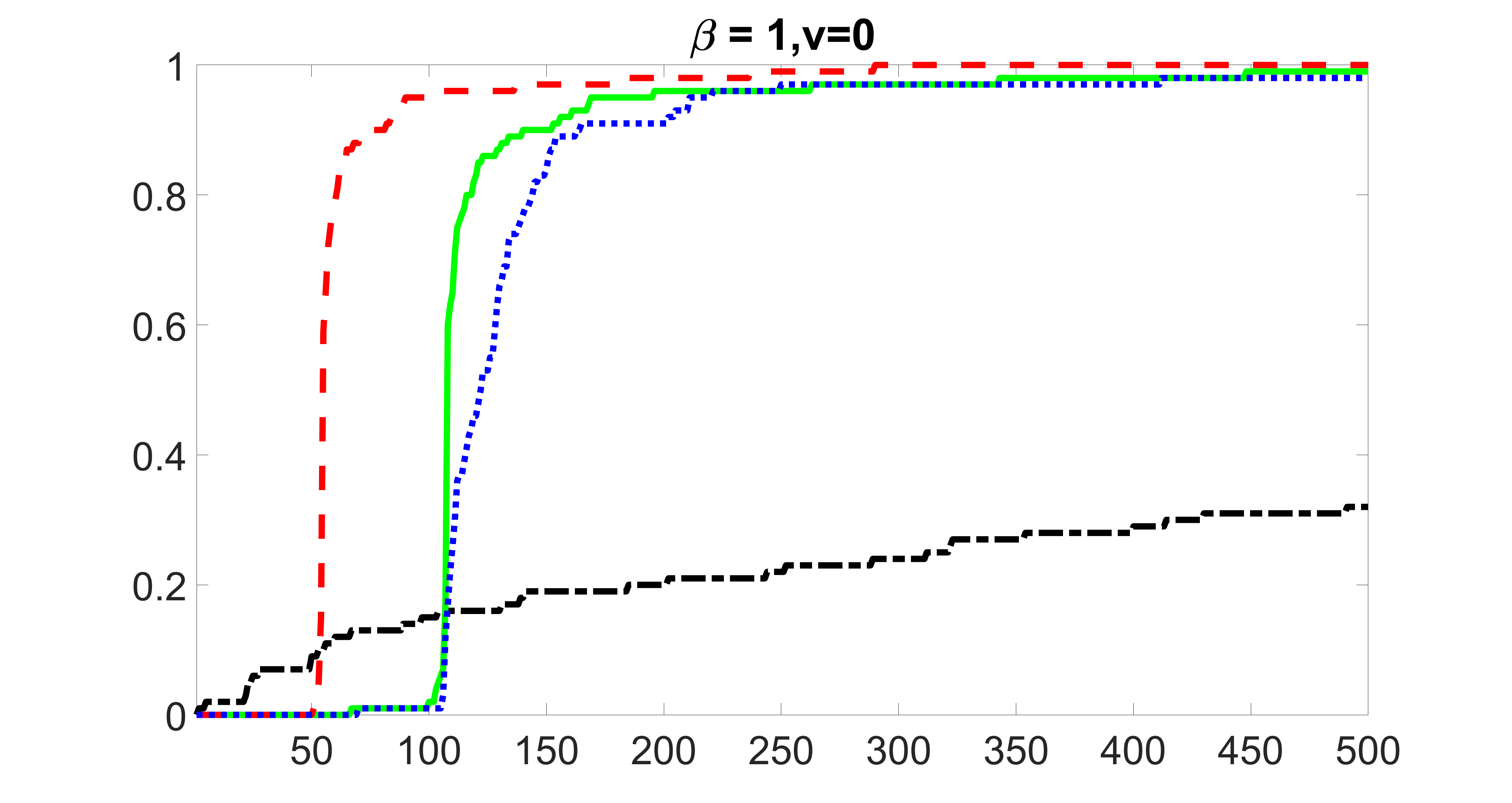}}
  \subcaptionbox{\footnotesize Precision: weak \\ outcome, zero exposure}[0.45\linewidth]
 {\includegraphics[width=6cm,height=3.5cm]{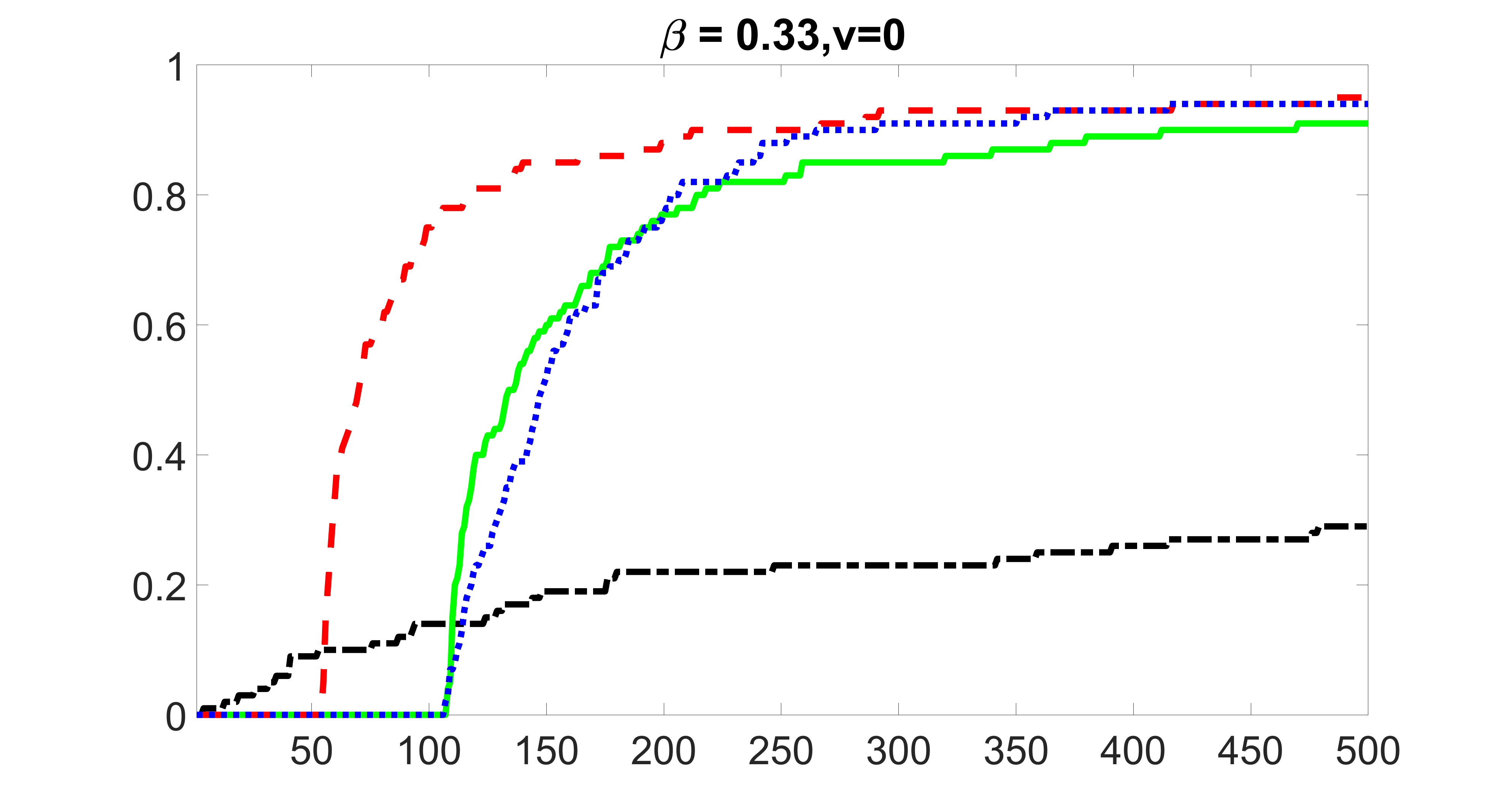}}
 \subcaptionbox{\footnotesize Precision: weaker \\ outcome, zero exposure}[0.45\linewidth]
 {\includegraphics[width=6cm,height=3.5cm]{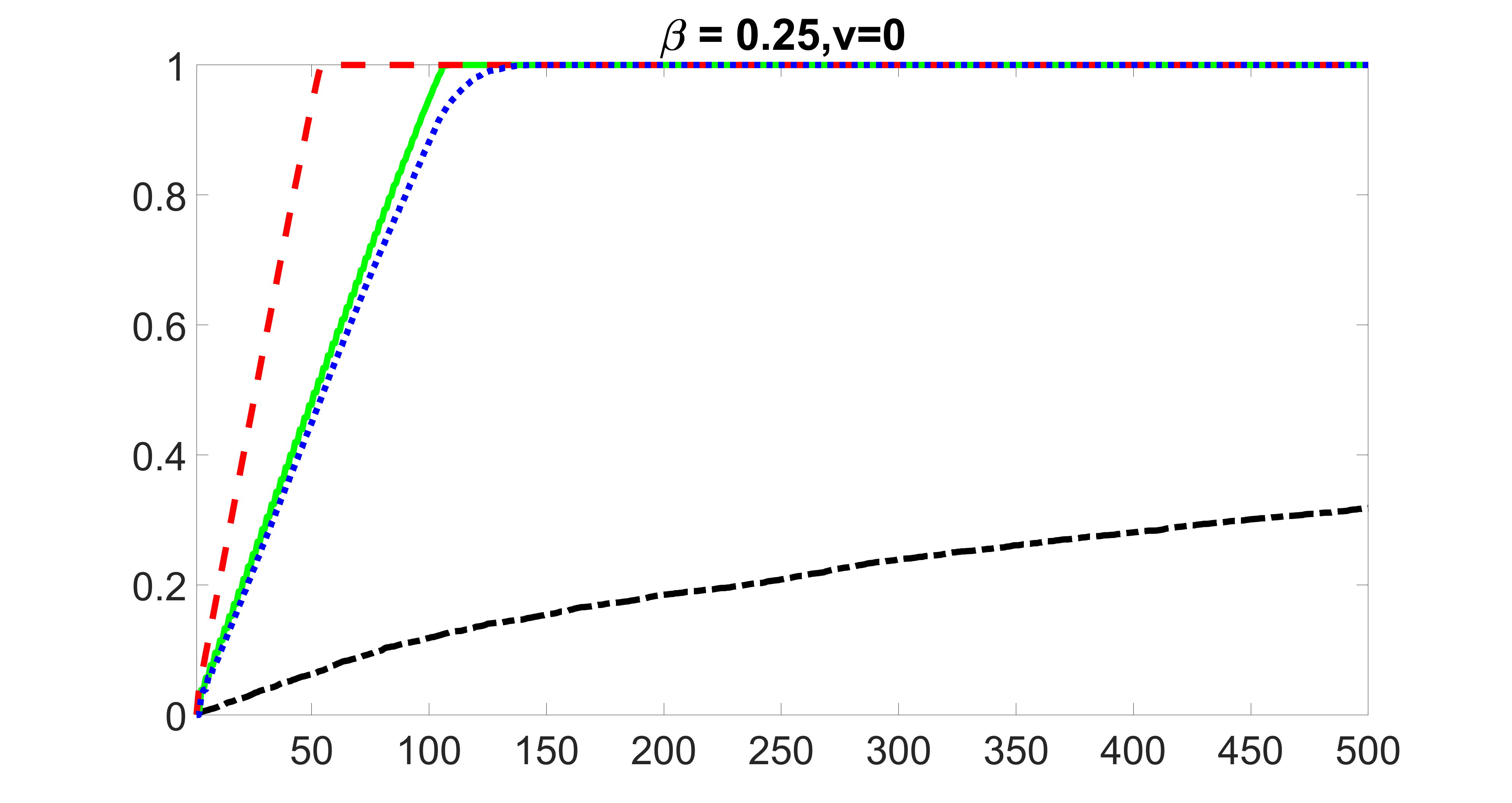}}
  \subcaptionbox{Overall coverage of $\mathcal{M}_1$}[0.45\linewidth]
 {\includegraphics[width=6cm,height=3.5cm]{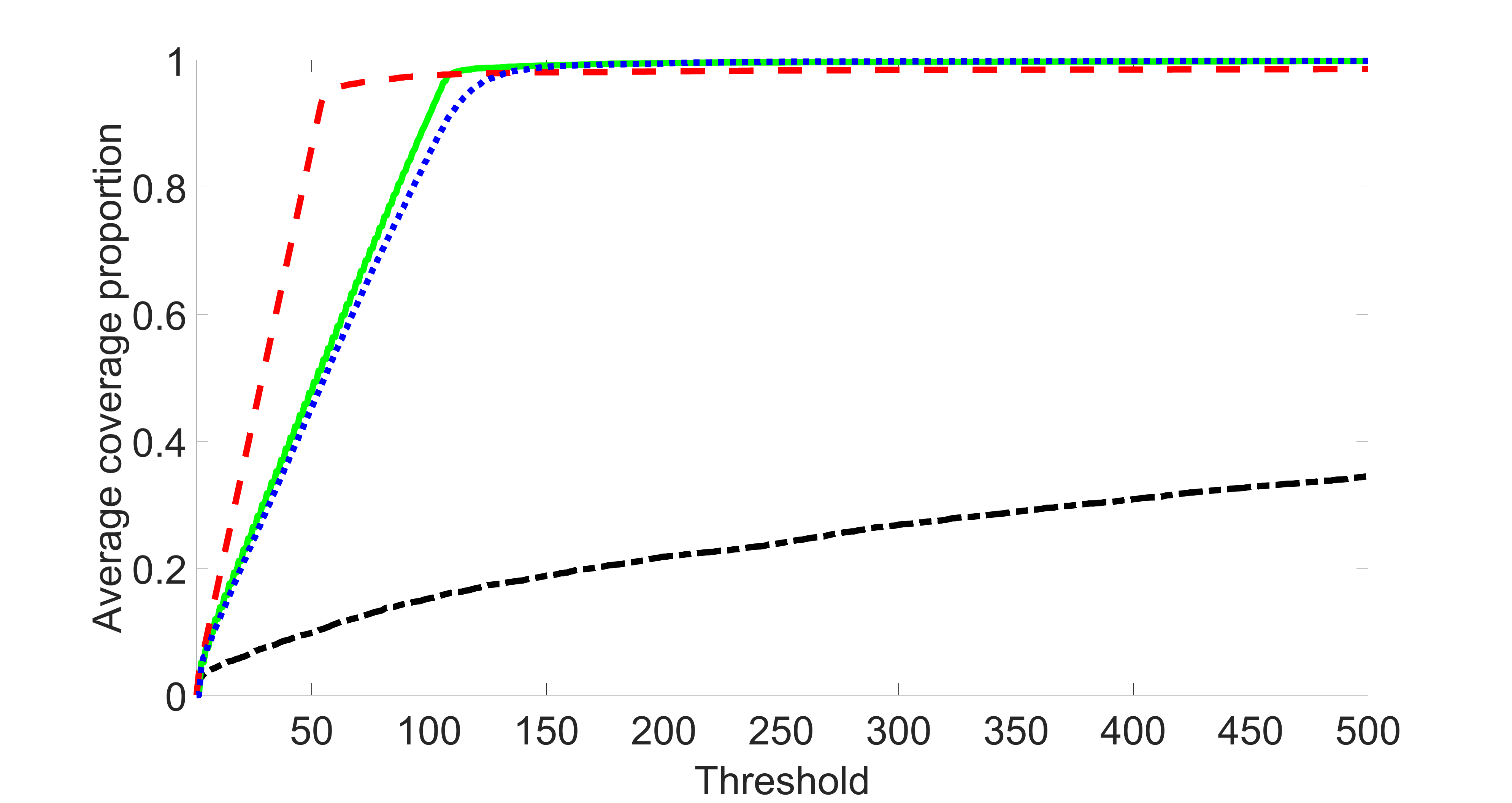}}
\caption{ Simulation results for the case $(n,s,K,\sigma) = (1000,5000,52,1)$: Panels (a) -- (g) plot the average coverage proportion for $X_l$, where $l \in \mathcal{M}_1 =  \{1,2,3,104,105, 106\} \cup \mathcal{P}_{LD}$. Panels (a) -- (c) correspond to strong outcome and weak exposure predictor, moderate outcome and moderate exposure predictor and weak outcome and strong exposure predictor; Panels (d) -- (g) correspond to strong, moderate, and weak predictors of outcome only. Panel (g) plots the average coverage proportion for the index set $\mathcal{P}_{LD}$. Panel (h) plots the average coverage proportion for the index set $\mathcal{M}_1$. The x-axis represents the size of $\widehat{\mathcal{M}} $, while
y-axis denotes the average proportion. The blue dot, green solid, red dashed and black dash dotted lines denote the blockwise joint screening, joint screening, outcome screening, and intersection screening methods, respectively.}
\label{sim3step1n1000sizesig52sigma1}
\end{figure}

\begin{table}[htbp]
\centering
\caption{Simulation results for $ \sigma=1 $: the average MSEs for $\bm\beta$ and $ {\bm B}$, and their associated standard errors in the parentheses are reported. The left panel summarizes the results from the joint screening method; the right panel summarizes the results from the blockwise joint screening method. The results are based on 100 Monte Carlo repetitions. }
\begin{tabular}{ rrr | rrrrr }
 &  &  &   & &\\
Proposed & MSE $\bm\beta$ & MSE ${\bm{B}}$ & Proposed (block) &MSE $\bm\beta$ & MSE ${\bm{B}}$  \\
\hline
n=200,K=2&1.423(0.096)&0.785(0.009)&n=200,K=2&1.390(0.090)&0.793(0.010)\\
n=500,K=2&0.831(0.069)&0.726(0.008)&n=500,K=2&0.892(0.082)&0.734(0.009)\\
n=1000,K=2&0.591(0.050)&0.676(0.008)&n=1000,K=2&0.488(0.028)&0.666(0.006)\\*
\hline
n=200,K=4&1.667(0.096)&0.815(0.011)&n=200,K=4&1.548(0.088)&0.805(0.010)\\*
n=500,K=4&1.059(0.082)&0.751(0.011)&n=500,K=4&1.094(0.090)&0.758(0.012)\\
n=1000,K=4&0.606(0.057)&0.671(0.008)&n=1000,K=4&0.555(0.045)&0.678(0.008)\\
\hline
n=200,K=6&1.955(0.101)&0.826(0.010)&n=200,K=6&1.701(0.084)&0.816(0.009)\\*
n=500,K=6&1.155(0.085)&0.749(0.011)&n=500,K=6&1.107(0.089)&0.752(0.011)\\
n=1000,K=6&0.578(0.051)&0.674(0.008)&n=1000,K=6&0.551(0.047)&0.672(0.008)\\*
\hline
n=200,K=12&2.466(0.096)&0.890(0.039)&n=200,K=12&2.223(0.129)&0.838(0.011)\\*
n=500,K=12&1.024(0.082)&0.735(0.010)&n=500,K=12&0.927(0.077)&0.727(0.008)\\*
n=1000,K=12&0.570(0.046)&0.673(0.008)&n=1000,K=12&0.627(0.057)&0.681(0.009)\\
\hline
n=200,K=24&2.533(0.164)&0.847(0.014)&n=200,K=24&2.136(0.138)&0.821(0.010)\\*
n=500,K=24&1.065(0.080)&0.733(0.010)&n=500,K=24&1.119(0.088)&0.737(0.011)\\
n=1000,K=24&0.662(0.050)&0.669(0.008)&n=1000,K=24&0.677(0.056)&0.673(0.009)\\
\hline
n=200,K=52&14.650(0.815)&2.034(0.487)&n=200,K=52&13.693(0.728)&1.870(0.459)\\*
n=500,K=52&1.816(0.144)&0.775(0.019)&n=500,K=52&1.725(0.143)&0.762(0.019)\\*
n=1000,K=52&0.937(0.066)&0.684(0.010)&n=1000,K=52&0.861(0.056)&0.675(0.008)\\*
\end{tabular}
\label{sim1t2}
\end{table}

\newpage
\section{ Theoretical properties}
Starting from here, we denote $\dot{\bm\beta} = ( \dot\beta_1, \ldots, \dot\beta_s)^T$, $\dot{\bm B}$ and $\dot{\bm C}_l$ as the true values for $\bm\beta = (\beta_1, \ldots,\beta_s)^T$, ${\bm B}$ and ${\bm C}_l$ respectively. Furthermore, we denote $\dot{\mathcal{C}}$, $\dot{\mathcal{P}}$, and $\dot{\mathcal{I}}$ as the true index sets of ${\mathcal{C}}$, ${\mathcal{P}}$, and ${\mathcal{I}}$.
\label{Theoretical guarantees}

\subsection{Sure screening property}
In this subsection, we study  theoretical properties for our screening procedure. We let $\mathcal{M}_1 =  \{1 \leq l \leq s_n : \dot{\beta}_{l}^{*} \neq 0 \}= \dot{\mathcal{C}} \cup \dot{\mathcal{P}}$, where $\dot{\mathcal{C}} = \{1 \leq  l \leq s_n: \dot{\bm{C}_{l}} \neq \bm{0} \textrm{ and } \dot{\beta}_{l}^{*} \neq 0 \}$ and $\dot{\mathcal{P}} = \{1 \leq  l \leq s_n: \dot{\bm{C}_{l}} = \bm{0} \textrm{ and } \dot{\beta}_{l}^{*} \neq 0 \}$. Here $\dot{\beta}_{l}^{*}$ and $\dot{\bm{C}_{l}}$ are the true values for $\beta_{l}$ and $\bm{C}_{l}$, respectively, and $\dot{\bm{B}}$ is the true value of $\bm{B}$.

We have the following theorems, where the assumptions needed are included in Section \ref{assumption}. 

\begin{thm}
\label{thm1}
Under Assumptions (A0) -- (A3) and (A5), let $\gamma_{1,n} = \alpha D_1   n^{-\kappa}$ and $\gamma_{2,n} = \alpha D_1   (pq)^{1/2}$ $n^{-\kappa}$  with $0 < \alpha < 1$, then we have 
$P(\mathcal{M}_1  \subset \widehat{\mathcal{M}}_{}) \to 1 $ as $n \to \infty$.
\end{thm}

Since the screening procedure automatically includes all the significant covariates for small value of $\gamma_{1,n}$ and $\gamma_{2,n}$, it is necessary to consider the size of  $\widehat{\mathcal{M}}_{}$, which we quantify in Theorem \ref{thm2}.

\begin{thm}
\label{thm2}
Under Assumptions (A0) -- (A5), when $\gamma_{1,n} = \alpha D_1   n^{-\kappa}$ and $\gamma_{2,n} = \alpha D_1   (pq)^{1/2} n^{-\kappa}$ with $0 < \alpha < 1$, we have $P(|\widehat{\mathcal{M}}_{}| = O(n^{2 \kappa + \tau})) \to 1$ as $n \to \infty$. 
\end{thm}

\begin{corr}
\label{coro2}
Under Assumptions (A0) -- (A5), when $\gamma_{1,n} = \alpha D_1   n^{-\kappa}$ and $\gamma_{2,n} = \alpha D_1   (pq)^{1/2} n^{-\kappa}$ with $0 < \alpha < 1$, we have $P(|\widehat{\mathcal{M}}-\widehat{\mathcal{M}}_{1}^*| = O(n^{2\kappa +\tau})) \to 1$ as $n \to \infty$.
\end{corr}

Theorem \ref{thm1} shows that if $\gamma_{1,n}$ and $\gamma_{2,n}$ are chosen properly, our screening procedure will include all significant variables with a high probability.   
Theorem \ref{thm2} guarantees that the size of selected model from the screening procedure is only of a polynomial order of $n$ even though the original model size is of an exponential order of $n$. 
Therefore, the false selection rate of our screening procedure vanishes as   $n \to \infty$, while the size of $\widehat{\mathcal{M}}$ grows in a polynomial order of $n$, where the order depends on two constants $\kappa$ and $\tau$ defined in Section \ref{assumption}. Theorem \ref{thm1a} shows our blockwise screening procedure also enjoys the screening property. The proofs of these theorems are collected in Section \ref{proofthm}.

\begin{thm}
\label{thm1a}
Under Assumptions (A0) -- (A3), (A5), and further assume the $j$-th block size $|\mathcal{B}_j| = D_6 n^{\nu_1}$ for some constant $D_6>0$. Let $\gamma_{1,n} = \alpha D_1   n^{-\kappa}$, $\gamma_{2,n} = \alpha D_1   (pq)^{1/2}$ $n^{-\kappa}$  with $0 < \alpha < 1$, $\gamma_{3,n} = \alpha D_1 D_6^{-1}   n^{-\kappa-\nu_1}$,  and $\gamma_{4,n} = \alpha D_1  D_6^{-1}  (pq)^{1/2}$ $n^{-\kappa-\nu_1}$  with $0 < \alpha < 1$, then we have 
$P(\mathcal{M}_1  \subset \widehat{\mathcal{M}}^{block}_{}) \to 1 $ as $n \to \infty$.
\end{thm}

\subsection{Theory for two-step estimator}
In this section, we develop a unified theory for our two-step estimator. In particular, we derive a non-asymptotic bound for the final estimates. We first introduce some notation. 

Denote parameter $\bm{\theta} =\left\{ \bm\beta^{\T},  \mathrm{vec}^{\T} (\bm{B})\right\}^{\T} \in \mathbb{R}^{s+pq}$, where $\bm\beta \in \mathbb{R}^{s}$ and $\bm{B} \in \mathbb{R}^{p \times q}$. Using this notation, problem \eqref{min1} can be recasted as minimizing  $l(\bm{\theta}) + P(\bm{\theta})$, where 
$l(\bm{\theta}) = (2n)^{-1} \sum_{i=1}^n \left( Y_i - \right.$ $\left. \langle \bm{Z}_i, \bm{B} \rangle - \right. $ $\left. \sum_{l \in \widehat{\mathcal{M}}_{}} {X}_{il} \beta_l \right)^2$,  and $P(\bm{\theta}) = \lambda_{1}  \sum_{l \in \widehat{\mathcal{M}}_{}}  |\beta_l|  + \lambda_{2} || \bm{B}||_* $. In addition, we let $\dot{\bm\theta} = \{ \dot{\bm\beta}^{\T}, $ $\mathrm{vec}(\dot{\bm{B}})^{\T}\}^{\T}$ be the true value for $\bm\theta$, where $\dot{\bm\beta}$ and $\dot{\bm{B}}$ is the true values for $\bm\beta$ and $\bm{B}$, respectively. Let $\widehat{\bm\theta}_{\bm\lambda} = \{\widehat{\bm\beta}^{ \T},\mathrm{vec}(\widehat{\bm{B}}\}^{\T})^{\T}$ be the proposed estimator for 
$\bm\theta$, where $\widehat{\bm\beta}$ and $\widehat{\bm{B}}$ are the estimators obtained from \eqref{min1} for tuning parameters $\bm\lambda = (\lambda_1,\lambda_2)$.

We hereby give nonasymptotic error bound for the proposed two-step estimator $\widehat{\bm\theta}_{\bm\lambda}$:  
\begin{thm}(Nonasymptotic error bounds for two-step estimator)
\label{thm3}
Under Assumptions (A0) -- (A9), $2 \kappa + \tau < 1$ and $\kappa < 1/4$, and the condition that $ \mathcal{M}_1\subset \widehat{\mathcal{M}}$ with $|\widehat{\mathcal{M}}_{}| = O(n^{2 \kappa + \tau}) $, conditional on $\widehat{\mathcal{M}}$, there exists some positive constants $c_1, c_2, c_3, c_4$, $C_0$, $C_1$, $g_0$ and $g_1$, such that for $\lambda_1  \geq 2 \sigma_{0} [ 2 n^{-1} \{ \log (\log n) + C_0 (2 \kappa + \tau)\log n \}]^{1/2}$ and 
$\lambda_2 \geq 2 b s_2 \sigma_{0} [ 2 n^{-1} \{ 3 \log s_2 + \log (\log n) \}]^{1/2} + 4 n^{-1/2} $ $\sigma_{\epsilon} (p^{1/2} + q^{1/2})$,   with probability at least 
$1- c_1/ \log n - c_2 / (s_2 \log n) - c_3 \exp\{ -c_4 (p+q)\} - \exp(-n)  $, one has
$$\left\| \widehat{\bm\theta}_{\bm\lambda} - \dot{\bm\theta} \right\|_2^2 \leq C_0 \max\left\{C_1 \lambda_1^2 n^{2 \kappa + \tau},\lambda_2^2 r \right\} \iota^{-2}.$$
\end{thm}
The bound in Theorem \ref{thm3} implies that the convergence rate of the proposed estimator $\widehat{\bm\theta}_{\lambda}$ is $O(\max\{$ $n^{2\kappa+\tau-1},n^{1-2\tau} \} )$.  
Here $\iota$ is a positive constant defined in Assumption (A6) in Section \ref{assumption}, and $r$ is the rank of $\dot{\bm{B}}$. The convergence rate is controlled by $\kappa$ and $\tau$, where $\kappa$ controls the exponential rates of model complexity that can diverge and $\tau$ controls the rate of largest eigenvalue of population covariance matrix that can grow. The proof of the theorem is deferred to section \ref{proofthm}.

\section{Assumptions for main theorems}\label{assumption}
In this section, we state the assumptions for the main theorems. We first make the following assumptions, which are needed for Theorems \ref{thm1} and \ref{thm2}.\\

(A0) The covariates $X_i$ are independent and identically distributed (i.i.d) with mean zero and covariance $\Sigma_x$. The random error $\epsilon_i$ are i.i.d with mean zero and variance $\sigma_{\epsilon}^2$. Define $\sigma_l^2 = (\Sigma_x)_{ll}$. The vectorized error matrices $\mathrm{vec}(E_i)$ are i.i.d with mean zero and covariance $\Sigma_e$. There exists a constant $\sigma_x > 0$ such that $ \left\|\Sigma_{x} \right\|_{\infty} \leq \sigma_x $.  Moreover, $x_i$ is independent of $E_i = (E_{i,jk})$ and $\epsilon_i$. \\

(A1) There exist some constants $D_1 > 0 $ and $b> 0 $, and $0 < \kappa < 1/2$ such that
\begin{eqnarray*}
\min_{l \in \mathcal{M}_1} \left| cov \left(\sum_{l^\prime \in \mathcal{M}_1} x_{i l^\prime}  \dot{\beta_{l^\prime}^{*}}, x_{il} \right) \right|
&\geq& D_1 n^{-\kappa},\\
\min_{l \in \mathcal{M}_2} \left\| cov \left( \sum_{l^\prime \in \mathcal{M}_2} x_{il^\prime} * \dot{\bm{C}_{l^\prime}}, x_{il}\right) \right\|_{op}
&\geq& D_1 (pq)^{1/2} n^{-\kappa},
\end{eqnarray*}
and $\max \left\{\max_{l \in \mathcal{M}_2} \left\|  \dot{\bm{C}_{l}} \right\|_{\infty}, \max_{l \in \mathcal{M}_2} \left\|  \dot{\bm{C}_{l}} \right\|_{op} ,\max_{l \in \mathcal{M}_2} \left| \langle  \dot{\bm{C}_{l}},  \dot{\bm{B}} \rangle \right|, \max_{l \in \mathcal{M}_1} | \dot{\beta}_{l}^{*} | \right\}< b$.\\

(A2) There exist positive constants $D_2$ and $D_3$ such that 
\[ 
\max \left\{ E[e^{D_2 x_{il}^2}],E[e^{D_2 E_{i,jk}^2}] , E[e^{D_2\langle \bm{E}_i, \dot{\bm{B}}\rangle^2}]\right\} \leq D_3
\]
for every $1 \leq l \leq s_n$, $1 \leq j \leq p$ and $1 \leq k \leq q$. 
Denote $\bm \epsilon = (\epsilon_1,\ldots,\epsilon_n)^T$ is a n-dimensional vector of zero-mean, 
there exists a constant $ \sigma_{0} > 0$ such that for any fixed $\|\bm v\|_2 = 1$, 
$P( \left| \langle \bm v, \bm \epsilon \rangle \right| \geq t ) \leq 2 \exp \left( - \frac{t^2}{2 \sigma_{0}^2} \right)$ for all $t>0$.\\

(A3) There exists a constant $D_4 > 0 $ such that $\log(s_n) = D_4 n^{\xi}$ for $\xi \in (0,1-2\kappa)$. \\

(A4) There exists constants $D_5 > 0 $  and $\tau > 0 $ such that $\lambda_{\max} (\Sigma_x) \leq D_5 n^\tau$.\\

(A5) $\log(pq) = o(n^{1- 2\kappa})$.\\

Before we state the assumptions for Theorem \ref{thm3}, we first introduce some notations. 

Denote $P(\bm{\theta}) = P_1(\bm{\beta})+P_2(\bm{B})$
 where $ P_1(\bm{\beta})= \lambda_{1}  \sum_{l \in \widehat{\mathcal{M}}_{}}  |\beta_l| $ and $P_2(\bm{B}) = \lambda_{2} || \bm{B}||_* $. In addition, let $r = \mathrm{rank}(\dot{\bm{B}})$, the true rank of matrix $\dot{\bm{B}} \in \mathbb{R}^{p \times q}$. Let us consider the class of matrices $ \Theta$ that have rank $r \leq \min\left\{ p,q \right\}$.  For any given matrix $\Theta$, we let $\mathrm{row}(\Theta) \subset \mathbb{R}^p$ and $\mathrm{col}(\Theta) \subset \mathbb{R}^q$ denote its row and column space, respectively. Let $U$ and $V$ be a given pair of $r$-dimensional subspace $U \subset \mathbb{R}^{p}$ and  $V \subset \mathbb{R}^{q}$, respectively. 

For a given $\bm\theta$ and pair $(U,V)$, we define the subspace 
$\Omega_1(\mathcal{M}_1)$,  $\overline{\Omega}_1(\mathcal{M}_1)$,  $\overline{\Omega}^{\perp}_1(\mathcal{M}_1)$, $\Omega_2(U,V)$, $\overline{\Omega}_2(U,V)$ and $\overline{\Omega}^{\perp}_2(U,V)$ as follows:
\begin{eqnarray*}
{\Omega}_1(\mathcal{M}_1)&=&\overline{\Omega}_1(\mathcal{M}_1) := \left\{ \bm\beta \in \mathbb{R}^{s} | \beta_j = 0 \textrm{ for all } j \not\in \mathcal{M}_1 \right\},\\
\overline{\Omega}_1^{\perp}(\mathcal{M}_1) &:=& \left\{ \bm\beta \in \mathbb{R}^{s} | \beta_j = 0 \textrm{ for all } j \in \mathcal{M}_1 \right\},\\
{\Omega}_2(U,V) &:=& \left\{ \Theta \in \mathbb{R}^{p \times q} |  \mathrm{row}(\Theta) \subset V, \textrm{ and } \mathrm{col}(\Theta) \subset U \right\},\\
\overline{\Omega}_2(U,V) &:=& \left\{ \Theta  \in \mathbb{R}^{p \times q} |   \mathrm{row}(\Theta) \subset V, \textrm{ or } \mathrm{col}(\Theta) \subset U \right\},\\
\overline{\Omega}_2^{\perp}(U,V) &:=& \left\{ \Theta  \in \mathbb{R}^{p \times q} |   \mathrm{row}(\Theta) \subset V^{\perp}, \textrm{ and } \mathrm{col}(\Theta) \subset U^{\perp} \right\}.
\end{eqnarray*}

Denote $\bm\Delta = \{ \bm\Delta_1^{\T},\mathrm{vec}(\bm\Delta_2)^{\T} \}^{\T} \in \mathbb{R}^{s+pq}$ with $\bm\Delta_1 \in \mathbb{R}^{s}$ and $\bm\Delta_2 \in \mathbb{R}^{p \times q}$. Then 
$\bm\Delta_{1,\overline{\Omega}_1} = \arg\min\limits_{\bm{v} \in \overline{\Omega}_1} \| \bm\Delta_{1} - \bm{v}\|_2$ and 
$\bm\Delta_{1,\overline{\Omega}_1^{\perp}} = \arg\min\limits_{\bm{v} \in \overline{\Omega}_1^{\perp}} \| \bm\Delta_{1} - \bm{v}\|_2$; 
$\bm\Delta_{2,\overline{\Omega}_2} = \arg\min\limits_{\bm{v} \in \overline{\Omega}_2} \| \bm\Delta_{2} - \bm{v}\|_F$ and 
$\bm\Delta_{2,\overline{\Omega}_2^{\perp}} = \arg\min\limits_{\bm{v} \in \overline{\Omega}_2^{\perp}} \| \bm\Delta_{2} - \bm{v}\|_F$. 
We write $\bm{X}_{comp} = (\bm{X},\bm{Z}_{new}) \in \mathbb{R}^{n \times (s+pq)}$ 
with $\bm{Z}_{new} = (\mathrm{vec}(\bm{Z}_1),$ $\ldots,\mathrm{vec}(\bm{Z}_n))^{\T} \in \mathbb{R}^{n \times pq}$ and let $X_{comp,i}$ represent the $i$-th column of $\bm{X}_{comp}^{\T}$ for $i=1,\ldots,n$. 
With loss of generality we assume that 
$\bm X$ has been column normalized, i.e. $\| \bm{x}_l\|_2 / \sqrt{n} = 1$, for all $l \in 1, \ldots,s $. 

We need the following assumptions:

(A6)
Define 
\begin{eqnarray*}
\iota := \min\limits_{
{\tiny
\begin{array}{c}
    \left|\bm\Delta_{1,\overline{\Omega}_1^{\perp}}\right|_1 \leq 3 \left|\bm\Delta_{1,\overline{\Omega}_1}\right|_1    \\
    \left\|\bm\Delta_{2,\overline{\Omega}_2^{\perp}}\right\|_{*} \leq 3 \left\|\bm\Delta_{2,\overline{\Omega}_2}\right\|_{*}
\end{array}}
}  \frac{1}{n} \sum_{i=1}^{n} \left\{ \frac{\left| \langle X_{comp,i}, \bm\Delta \rangle \right|^2}{\|\bm\Delta \|_2^2}\right\}, 
\end{eqnarray*}
and assume that $\iota$ is a positive constant.

(A7) Assume $\max\{p,q\} / \log(n) \to \infty$ and $ \max\{p,q\} = o(n^{1-2\tau})$ as $n \to \infty$ with $\tau < 1/2$.

(A8) The vectorized error matrices $\mathrm{vec}(\bm{E}_i)$ are i.i.d. $N(\bm{0},\bm\Sigma_e^2)$, where $\lambda_{\max}(\bm\Sigma_e) \leq C_U^2 < \infty$.

(A9) $\mathrm{rank}(\dot{\bm{B}}) = r < \min (p,q)$ holds.

\section{Auxiliary lemmas}
\label{auxlemma}
In this section, we include the auxiliary lemmas needed for the theorems and their proofs.
\begin{lem}
\label{lem1}
(Bernstein's inequality) Let $T_1,\ldots, T_n$ be independent random variable with zero mean such that $E(|T_i|^m) \leq m ! M^{m-2} v_i/2$, for every $m \geq 2$ (and all i) and some constant $M$ and $v_i$. Then
\[ 
P(|\sum_{i=1}^n T_i| > x) \leq 2 e^{-\frac{1}{2} \frac{x^2}{v+Mx}},
\]
for $v = \sum_{i=1}^n v_i$.
\end{lem}
This is Lemma 2.2.11 from \cite{van2000weak} and we omit the proof here.

\begin{lem}
\label{lem2}
Under Assumptions (A0), (A1) and (A2), for arbitrary $t>0$ and for every $l,l^\prime,j,k$, we have that 
\begin{eqnarray*}
P\left(( |\sum_{i=1}^n \{ x_{il} x_{il^\prime} -E(x_{il} x_{il^\prime}) \}| \geq t\right) &\leq& 2 \exp\left\{-\frac{t^2}{2(2n D_2^{-2} e^{D_2 \sigma_x} D_3 + t/D_2)}\right\},\\
P\left(\left|\sum_{i=1}^n (x_{il} E_{i,jk})\right| \geq t \right)  &\leq& 2 \exp\left\{-\frac{t^2}{2(2n D_2^{-2} D_3 + t/D_2)}\right\}, \\
P\left( \left|\sum_{i=1}^n \left( x_{il} \langle \bm{E}_i,  \dot{\bm{B}}\rangle \right)\right| \geq t \right) &\leq& 2 \exp\left\{-\frac{t^2}{2(2n D_2^{-2} D_3 + t/D_2)}\right\},\\
P\left(|\sum_{i=1}^n (x_{il} \epsilon_{i})| \geq t\right)  &\leq& 2 \exp\left\{-\frac{t^2}{2(2n D_2^{-2} D_3 + t/D_2)}\right\}.
\end{eqnarray*}
\end{lem}
\begin{proof}[Proof of Lemma \ref{lem2}:]
Note that the last part of Assumption (A2) actually implies that, 
there exist positive constants $D^\prime_2$ and $D^\prime_3$, such that $E[e^{D^\prime_2 \epsilon_i^2}] \leq D^\prime_3 $ by applying Theorem 3.1 from \cite{rivasplata2012subgaussian}. Therefore, it can be unified into the first part of Assumption (A2) which implies that  $$\max \left\{ E[e^{D_2 x_{il}^2}],E[e^{D_2 E_{i,jk}^2}] , E[e^{D_2 \langle \bm{E}_i, \dot{\bm{B}}\rangle^2}], E[e^{D_2 \epsilon_i^2}] \right\} \leq D_3 $$
for every $1 \leq l \leq s_n$, $1 \leq j \leq p$ and $1 \leq k \leq q$.
Therefore, by Assumptions (A0) and (A2) and Jensen's inequality, we have 
\begin{eqnarray*}
E \left[ e^{D_2 | x_{il} x_{il^\prime} - E \left(x_{il} x_{il^\prime}\right)|}\right] 
\leq E \left[ e^{D_2  | x_{il} x_{il^\prime} | + D_2|E \left(x_{il} x_{il^\prime}\right)|}\right]
= e^{D_2|E \left(x_{il} x_{il^\prime}\right)|} E \left[ e^{D_2  | x_{il} x_{il^\prime} | }\right]\\
\leq e^{D_2 \sigma_x} E\left[ e^{D_2 \frac{x_{il}^2 + x_{il^\prime}^2}{2}}\right]
\leq e^{D_2 \sigma_x}  \left[ E\left\{ e^{D_2 x_{il}^2 }\right\} E\left\{ e^{D_2 x_{il^\prime}^2}\right\} \right]^{1/2} 
\leq e^{D_2 \sigma_x} D_3. 
 \end{eqnarray*}
 Then for every $m \geq 2$, one has
 \[
 E\left[ | x_{il} x_{il^\prime} -E(x_{il} x_{il^\prime}) |^m \right] 
 \leq \frac{m !}{D_2^m} E\left[ e^{D_2  | x_{il} x_{il^\prime} -E(x_{il} x_{il^\prime}) |}\right] 
 \leq \frac{m !}{D_2^m} e^{D_2 \sigma_x} D_3. 
 \]
 It follows from Lemma \ref{lem1} that
 \[
P( |\sum_{i=1}^n \{ x_{il} x_{il^\prime} -E(x_{il} x_{il^\prime}) \}| \geq t) \leq 2 \exp\left\{-\frac{t^2}{2(2n D_2^{-2} e^{D_2 \sigma_x} D_3 + t/D_2)}\right\}.
\]
Similarly we obtain
 \begin{eqnarray*}
E \left[ e^{D_2 \left| x_{il} E_{i,jk}\right|}\right] 
\leq E \left[ e^{   D_2 \frac{x_{il}^2 + E_{i,jk}^2}{2} }\right]
\leq \left[ E\left\{ e^{D_2 x_{il}^2 }\right\} E\left\{ e^{D_2 E_{i,jk}^2}\right\} \right]^{1/2} 
\leq D_3. 
 \end{eqnarray*}
 Then for every $m \geq 2$, one has
 \[
 E\left[ \left| x_{il} E_{i,jk} \right|^m \right] 
 \leq \frac{m !}{D_2^m} E \left[ e^{ \left| x_{il} E_{i,jk}\right|}\right]  
 \leq \frac{m !}{D_2^m} D_3. 
 \]
 It follows from Lemma \ref{lem1} that
 \[
P\left(\left|\sum_{i=1}^n (x_{il} E_{i,jk})\right| \geq t \right)  \leq 2 \exp\left\{-\frac{t^2}{2(2n D_2^{-2} D_3 + t/D_2)}\right\}.
\]

Similarly we have
 \begin{eqnarray*}
E \left[ e^{D_2 \left| x_{il} \langle \bm{E}_i, \dot{\bm{B}}\rangle \right|}\right] 
&\leq& E \left[ e^{   D_2 \frac{x_{il}^2 + \langle \bm{E}_i, \dot{\bm{B}}\rangle^2}{2} }\right]
\leq \left[ E\left\{ e^{D_2 x_{il}^2 }\right\} E\left\{ e^{D_2 \langle \bm{E}_i, \dot{\bm{B}}\rangle^2}\right\} \right]^{1/2} 
\leq D_3, \\  
E \left[ e^{D_2 \left| x_{il} \epsilon_{i}\right|}\right] 
&\leq& E \left[ e^{   D_2 \frac{x_{il}^2 + \epsilon_{i}^2}{2} }\right]
\leq \left[ E\left\{ e^{D_2 x_{il}^2 }\right\} E\left\{ e^{D_2 \epsilon_{i}^2}\right\} \right]^{1/2} 
\leq D_3. 
 \end{eqnarray*}
 Then following the proof of showing second inequality, one has
 \begin{eqnarray*}
P\left( \left|\sum_{i=1}^n \left( x_{il} \langle  \bm{E}_i,  \dot{\bm{B}}\rangle \right)\right| \geq t \right) &\leq& 2 \exp\left\{-\frac{t^2}{2(2n D_2^{-2} D_3 + t/D_2)}\right\},\\
P\left(|\sum_{i=1}^n (x_{il} \epsilon_{i})| \geq t\right)  &\leq& 2 \exp\left\{-\frac{t^2}{2(2n D_2^{-2} D_3 + t/D_2)}\right\}.
 \end{eqnarray*}
\end{proof}
The following lemma is a standard result called Gaussian comparison inequality \citep{anderson1955integral}. 
\begin{lem}
\label{lem3}
Let $X$ and $Y$ be zero-mean vector Gaussian random vectors with covariance matrix $\Sigma_X$ and $\Sigma_Y$ respectively. 
If  $\Sigma_X - \Sigma_Y$ is positive semi-definite, then for any convex symmetric set $C$, $P(X \in C) \leq P(Y \in C)$. 
\end{lem}

\section{Proof of theorems}
\label{proofthm}
\begin{proof}[Proof of Theorem \ref{thm1}:]
We can write
\begin{eqnarray*}
&& P\left\{ \mathcal{M}_1 \subset \left( \widehat{\mathcal{M}}_{1}^{*} \cup \widehat{\mathcal{M}}_{2} \right) \right\} \\
&=& P\left\{ \cap_{l \in \mathcal{M}_1} \left( \widehat{\mathcal{M}}_{1}^{*} \cup \widehat{\mathcal{M}}_{2} \right) \right\}  \\
&=&  1 - P\left\{ \cup_{l \in \mathcal{M}_1} \left( \widehat{\mathcal{M}}_{1}^{*} \cup \widehat{\mathcal{M}}_{2} \right)^c \right\}  \\
&\geq& 1 - \sum_{l \in \mathcal{M}_1} P \left( \widehat{\mathcal{M}}_{1}^{*c} \cap \widehat{\mathcal{M}}_{2}^c \right)  \\
&=& 1 - \sum_{l \in \mathcal{M}_1} P\left(  |\widehat{{\beta}_{l}^{M}}|  \leq \gamma_{1,n}, \|\widehat{{\bm{C}}}_{l}^M\|_{op} \leq \gamma_{2,n}\right) \\
 &\geq& 1 - \sum_{l \in \mathcal{M}_1 \cap \mathcal{M}_2} P\left( \|\widehat{{\bm{C}}}_{l}^M\|_{op} \leq \gamma_{2,n} \right)
 - \sum_{l \in \mathcal{M}_1 \cap \mathcal{M}_2^c} P\left(|\widehat{{\beta}_{l}^{M}}|  \leq \gamma_{1,n}\right).
\end{eqnarray*}

Firstly, recall that $\dot{\bm{C}}_{l}^M = cov (\sum_{l^\prime \in \mathcal{M}_2} x_{i l^\prime} * \dot{\bm{C}}_{l^\prime}, x_{il})$ i.e.
$\dot{C}_{l,jk}^M = cov (\sum_{l^\prime \in \mathcal{M}_2} x_{i l^\prime} \dot{C}_{l^\prime,jk},$ $x_{il}) = n^{-1} \sum_{i=1}^n E(x_{il} Z_{i,jk})$. For every $1 \leq j \leq p$, $1 \leq k \leq q$ and $1 \leq l \leq s_n$, we have
\[
\widehat{C}_{l,jk}^M -\dot{C}_{l,jk}^M = n^{-1} \sum_{i=1}^n \left[ x_{il} Z_{i,jk} - E(x_{il} Z_{i,jk})\right].
\]
It follows from Assumptions (A0), (A1), (A2) and Lemma 2 that for any $t>0$, one has
\begin{eqnarray*}
&&P\left(\left|\widehat{C}_{l,jk}^M -\dot{C}_{l,jk}^M \right| \geq t\right) 
= P\left( \left|  \sum_{i=1}^n [ x_{il} Z_{i,jk} - E(x_{il} Z_{i,jk}) \right| \geq nt \right) \\
&=& P \left( \left|\sum_{l^\prime \in \mathcal{M}_2} \sum_{i=1}^n \left[x_{il} x_{il^\prime} - E(x_{il} x_{il^\prime}) \right] \dot{C}_{l^\prime,jk} + \sum_{i=1}^n x_{il} E_{i,jk}\right| \geq nt\right)\\
&\leq& \sum_{l^\prime  \in \mathcal{M}_2} P\left(\left| \sum_{i=1}^n \left[x_{il} x_{il^\prime} - E(x_{il} x_{il^\prime}) \right] \right| \geq \frac{nt}{b(s_2 + 1)}\right) 
+ P\left( \left| \sum_{i=1}^n x_{il} E_{i,jk}\right| \geq \frac{nt}{s_2 + 1} \right) \\
&\leq& 2 s_2 \exp \left[ - \frac{n t^2 b^{-2} (s_1 + 1)^{-2}}{2 \{ 2 D_2^{-2} e^{D_2 \sigma_x} D_3 + D_2^{-1} b^{-1} (s_1 + 1)^{-1} t \} }\right]\\
&& + 2 \exp \left[ - \frac{n t^2 (s_2 + 1)^{-2}}{2 \{ 2 D_2^{-2} D_3 + D_2^{-1} (s_2 + 1)^{-1} t\} }\right].
 \end{eqnarray*}
Therefore, for every $l \in \mathcal{M}_2$, we have
 \begin{eqnarray*}
&& P\left( \| \widehat{{\bm{C}}}_{l}^M \|_{op} \leq \gamma_{2,n} \right) 
\leq P\left( \| \widehat{{\bm{C}}}_{l}^M - \dot{\bm{C}}_{l}^M\|_{op} \geq D_1(pq)^{1/2}n^{-\kappa} - \gamma_{2,n} \right) \\
&\leq&  P\left( \| \widehat{{\bm{C}}}_{l}^M - \dot{\bm{C}}_{l}^M\|_{F} \geq (pq)^{1/2} (1-\alpha) D_1 n^{-\kappa} \right) \\
&=&  P\left( \sum_{j,k} \left| \widehat{C}_{l,jk}^M - \dot{C}_{l,jk}^M \right|^2 \geq (pq)^{1/2} (1-\alpha) D_1 n^{-\kappa} \right) \\ 
&\leq& \sum_{j,k} P\left(  \left| \widehat{C}_{l,jk}^M - \dot{C}_{l,jk}^M \right|^2 \geq   \left\{(1-\alpha) D_1 n^{-\kappa} \right\}^2 \right) \\ 
&\leq& \sum_{j,k} P\left(  \left| \widehat{C}_{l,jk}^M - \dot{C}_{l,jk}^M \right| \geq    (1-\alpha) D_1 n^{-\kappa}  \right)\\
&\leq& 2pq \left( s_2 \exp \left[ - \frac{n^{1-2 \kappa} \left\{ (1-\alpha)D_1 b^{-1} (s_2 + 1)^{-1} \right\}^2 }{2 \{2 D_2^{-2} e^{D_2 \sigma_x} D_3 + D_2^{-1} b^{-1} (s_2 + 1)^{-1}  (1-\alpha) D_1 n^{-\kappa}\}}\right] \right. \\
&&\left. +  \exp \left[ - \frac{n^{1-2\kappa} \left\{ (1-\alpha) D_1  (s_2 + 1)^{-1} \right\}^2}{2 \{2 D_2^{-2} D_3 + D_2^{-1} (s_2 + 1)^{-1} (1-\alpha) D_1 n^{-\kappa}\}}\right]\right)\\
&\leq& 2pq \left( s_2 \exp \left[ - \frac{n^{1-2 \kappa} \left\{ (1-\alpha)D_1 b^{-1} (s_2 + 1)^{-1} \right\}^2 }{2 \{2 D_2^{-2} e^{D_2 \sigma_x} D_3 + D_2^{-1} b^{-1} (s_2 + 1)^{-1}  (1-\alpha) D_1 \}}\right] \right. \\
&&\left. +  \exp \left[ - \frac{n^{1-2\kappa} \left\{ (1-\alpha) D_1  (s_2 + 1)^{-1} \right\}^2}{2 \{2 D_2^{-2} D_3 + D_2^{-1} (s_2 + 1)^{-1} (1-\alpha) D_1\}}\right]\right)\\
 \end{eqnarray*}
 Let 
\begin{eqnarray*}
d_0 &=& \min \left[  \frac{ \left\{ (1-\alpha)D_1 b^{-1} (s_2 + 1)^{-1} \right\}^2 }{2 \{2 D_2^{-2} e^{D_2 \sigma_x} D_3 + D_2^{-1} b^{-1} (s_2 + 1)^{-1}  (1-\alpha) D_1 \}}, \right.\\
&& \left. \frac{ \left\{ (1-\alpha) D_1  (s_2 + 1)^{-1} \right\}^2}{2 \{2 D_2^{-2} D_3 + D_2^{-1} (s_2 + 1)^{-1} (1-\alpha) D_1\}} \right],
\end{eqnarray*} 
We have for every $l \in \mathcal{M}_2$, 
\begin{eqnarray}
P\left( \| \widehat{{\bm{C}}}_{l}^M \|_{op} \leq \gamma_{2,n} \right) &\leq& 2pq(s_2 + 1)\exp(-d_0 n^{1-2\kappa}), \label{thm1eq1}
\end{eqnarray}

Let us consider $P\left( |\widehat{{\beta}_{l}^{M}}|  \leq \gamma_{1,n} \right)$. 
recall that, $\dot{\beta}_{l}^M = \dot{\beta}_{l}^{*M} + \langle\dot{\bm{C}}_{l}^{M}$, $\dot{\bm{B}} \rangle $, $\dot{\beta}_{l}^{*M} = cov (\sum_{l^\prime \in \mathcal{M}_1} $ $x_{i l^\prime}  \dot{\beta}_{l^\prime}, x_{il})$ and 
$\dot{\beta}_{l}^M =n^{-1} \sum_{i=1}^n E(x_{il} Y_{i})$. For every $1 \leq l \leq s_n$, we have
\[
\widehat{\beta}_{l}^M -\dot{\beta}_{l}^M = n^{-1} \sum_{i=1}^n \left\{ x_{il} Y_{i} - E(x_{il} Y_{i})\right\}.
\]
It follows from Assumptions (A0), (A1), (A2) and Lemma 2 that for any $t>0$, we have
\begin{eqnarray*}
&&P\left(\left|\widehat{\beta}_{l}^M -\dot{\beta}_{l}^M  \right| \geq t\right) 
= P\left[ \left|\sum_{i=1}^n \left\{ x_{il} Y_{i} - E(x_{il} Y_{i})\right\} \right| \geq nt \right] \\
&=& P \left[ \left|\sum_{l^\prime \in \mathcal{M}_1} \sum_{i=1}^n \left\{ x_{il} x_{il^\prime} - E\left( x_{il} x_{il^\prime} \right) \right\} \dot{\beta}_{l^\prime}^{*M} + 
\sum_{l^\prime \in \mathcal{M}_2} \sum_{i=1}^n \left\{ x_{il} x_{il^\prime} - E\left( x_{il} x_{il^\prime} \right) \right\}  \langle \dot{\bm{C}}_{l^\prime}^{M}, \dot{\bm{B}}\rangle + \right. \right.\\
&+&\left. \left. \sum_{i=1}^n \left\{ x_{il} \langle \bm{E}_{i}, \dot{\bm{B}} \rangle - E( x_{il})E\left(  \langle \bm{E}_{i}, \dot{\bm{B}} \rangle  \right) \right\}  + \sum_{i=1}^n x_{il} \epsilon_{i}\right| \geq n t\right]\\ 
&\leq& P \left[   \sum_{l^{\prime}  \in \mathcal{M}_1}  \left|\sum_{i=1}^n \left\{ x_{il} x_{il^\prime} - E\left( x_{il} x_{il^\prime} \right)  \right\} \right|  b +
   \sum_{l^{\prime}  \in \mathcal{M}_2}  \left|\sum_{i=1}^n \left\{ x_{il} x_{il^\prime} - E\left( x_{il} x_{il^\prime} \right)  \right\} \right|  b \right. \\
&& + \left.  \left| \sum_{i=1}^n x_{il} \langle \bm{E}_i, \dot{\bm{B}}\rangle \right| 
+  \left| \sum_{i=1}^n x_{il} \epsilon_{i} \right| \geq nt \right]\\
&\leq& \sum_{l^\prime  \in \mathcal{M}_1} P\left[\left| \sum_{i=1}^n \left\{  x_{il} x_{il^\prime} - E\left(  x_{il} x_{il^\prime} \right)  \right\} \right| \geq \frac{nt}{b(s_1 + s_2 + 2)}\right] \\
&& + \sum_{l^\prime  \in \mathcal{M}_2} P\left[\left| \sum_{i=1}^n \left\{  x_{il} x_{il^\prime} - E\left(  x_{il} x_{il^\prime} \right)  \right\} \right| \geq \frac{nt}{b(s_1 + s_2 + 2)}\right]  \\
&& + P\left( \left| \sum_{i=1}^n x_{il} \langle \bm{E}_i, \dot{\bm{B}}\rangle \right|  \geq \frac{nt}{s_1 + s_2 + 2} \right)  + P\left( \left| \sum_{i=1}^n x_{il} \epsilon_{i} \right| \geq \frac{nt}{s_1 + s_2 + 2} \right) \\
&\leq& 2 (s_1 + s_2) \exp \left[ - \frac{n t^2 (2b)^{-2} (s_1 + s_2 + 2)^{-2}}{2\left\{2 D_2^{-2} e^{D_2 \sigma_x} D_3 + D_2^{-1} (2b)^{-1} (s_1 + s_2 + 2)^{-1} t\right\}}\right]
\\
&& + 4 \exp \left[ - \frac{n t^2 (s_1 + s_2 + 2)^{-2}}{2 \left\{2 D_2^{-2} D_3 + D_2^{-1} (s_1 + s_2 + 2)^{-1} t\right\}}\right].
 \end{eqnarray*}

For $l \in \mathcal{M}_1 \cap \mathcal{M}^c_2$,  we have $\langle \dot{\bm{C}}^M_{l}, \dot{\bm{B}} \rangle = 0$, under Assumption (A1) and previous deduction, we have 
\begin{eqnarray*}
&& P\left( | \widehat{\beta}^M_{l} | \leq \gamma_{1,n} \right) 
=  P\left(  - | \widehat{\beta}^M_{l} | \geq - \gamma_{1,n} \right)
\leq P\left(  | \dot{\beta}_{l}^{*M} | - | \widehat{\beta}^M_{l} | \geq D_1   n^{-\kappa} - \gamma_{1,n} \right)  \\
&=&  P\left(  | \dot{\beta}_{l}^{*M}| - |\langle \dot{\bm{C}}^M_{l}, \dot{\bm{B}} \rangle  | - | \widehat{\beta}^M_{l} | \geq (1-\alpha)D_1   n^{-\kappa} \right)\\
&\leq&  P\left(  | \dot{\beta}_{l}^M  | - | \widehat{\beta}^M_{l} | \geq (1-\alpha)D_1   n^{-\kappa} \right)
\\
&\leq&  
P\left(  |  \dot{\beta}_{l}^M - \widehat{\beta}^M_{l}  | \geq (1-\alpha)D_1   n^{-\kappa} \right)\\
&\leq& 2 (s_1 + s_2) \exp \left[ - \frac{n^{1-2\kappa} (1-\alpha)^2 D_1^2  (2b)^{-2} (s_1 + s_2 + 2)^{-2}}{2\left\{2 D_2^{-2} e^{D_2 \sigma_x} D_3 + D_2^{-1} (2b)^{-1} (s_1 + s_2 + 2)^{-1} (1-\alpha)D_1   n^{-\kappa} \right\}}\right]
\\
&& + 4 \exp \left[ - \frac{n^{1-2\kappa} (1-\alpha)^2 D_1^2 (s_1 + s_2 + 2)^{-2}}{2 \left\{2 D_2^{-2} D_3 + D_2^{-1} (s_1 + s_2 + 2)^{-1} (1-\alpha) D_1 pq^{1/2} n^{-\kappa} \right\}}\right]\\
&\leq& 2 (s_1 + s_2) \exp \left[ - \frac{n^{1-2\kappa} (1-\alpha)^2 D_1^2  (2b)^{-2} (s_1 + s_2 + 2)^{-2}}{2\left\{2 D_2^{-2} e^{D_2 \sigma_x} D_3 + D_2^{-1} (2b)^{-1} (s_1 + s_2 + 2)^{-1} (1-\alpha)D_1   \right\}}\right]
\\
&& + 4 \exp \left[ - \frac{n^{1-2\kappa} (1-\alpha)^2 D_1^2  (s_1 + s_2 + 2)^{-2}}{2 \left\{2 D_2^{-2} D_3 + D_2^{-1} (s_1 + s_2 + 2)^{-1} (1-\alpha) D_1  \right\}}\right]\\
 \end{eqnarray*}
Let
\begin{eqnarray*}
d_1 &=& \min \left[ \frac{(1-\alpha)^2 D_1^2 (2b)^{-2} (s_1 + s_2 + 2)^{-2}}{2\left\{2 D_2^{-2} e^{D_2 \sigma_x} D_3 + D_2^{-1} (2b)^{-1} (s_1 + s_2 + 2)^{-1} (1-\alpha)D_1  \right\}}, \right.\\
&&\left. \frac{(1-\alpha)^2 D_1^2 (s_1 + s_2 + 2)^{-2}}{2 \left\{2 D_2^{-2} D_3 + D_2^{-1} (s_1 + s_2 + 2)^{-1} (1-\alpha) D_1  \right\}} \right],
\end{eqnarray*}
We have for each $l \in \mathcal{M}_1 \cap \mathcal{M}_2^c$, 
\begin{eqnarray}
P\left( |\widehat{{\beta}_{l}^{M}}|  \leq \gamma_{2,n} \right) \leq 2(s_1 +s_2 + 2) \exp \left( -d_1 n^{1-2\kappa}   \right) . \label{thm1eq4}
\end{eqnarray}
In sum, by Assumption (A5), and \eqref{thm1eq1} and \eqref{thm1eq4}, we have
\begin{eqnarray*}
&& P\left\{ \mathcal{M}_1 \subset \left( \widehat{\mathcal{M}}_{1}^{*} \cup \widehat{\mathcal{M}}_{2} \right) \right\} \\
&\geq& 1 - \sum_{l \in \mathcal{M}_1 \cap \mathcal{M}_2} P\left( \|\widehat{{\bm{C}}}_{l}^M\|_{op} \leq \gamma_{2,n} \right)
 - \sum_{l \in \mathcal{M}_1 \cap \mathcal{M}_2^c} P\left(|\widehat{{\beta}_{l}^{M}}|  \leq \gamma_{1,n} \right)\\
&\geq& 1 - 2pqs_2(s_2 + 1)\exp(-d_0 n^{1-2\kappa}) -  2s_1 (s_1 +s_2 + 2) \exp \left( -d_1 n^{1-2\kappa}   \right)\\
&\geq& 1 - d_0^\prime pq \exp \left( -d_1^\prime n^{1-2\kappa} \right) \to  1,  \quad \textrm{ as } n \to \infty,
\end{eqnarray*}
for some positive constants $d_0^\prime$ and $d_1^\prime$. 
Therefore, $P(\mathcal{M}_1  \subset \widehat{\mathcal{M}}_{}) \to  1$ as $n \to \infty$. 
\end{proof}

\begin{proof}[Proof of Theorem \ref{thm2}]
The proof consists of two steps. In step 1, we will show that
$P(\widehat{\mathcal{M}}_{} \subset \mathcal{M}^0) \to 1$, where $\mathcal{M}^0 = \mathcal{M}^0_1 \cup \mathcal{M}^0_2$,  $\mathcal{M}^0_1 = \left\{1 \leq l \leq s_n: |\dot{\beta}_{l}^M| \geq \gamma_{1,n}/2 \right\}$ and $\mathcal{M}^0_2 = \left\{1 \leq l \leq s_n: \right.$ $\left. \left\|\dot{\bm{C}}^M_{l}\right\|_{op} \geq \gamma_{2,n}/2 \right\}.$ Recall that 
$\widehat{\mathcal{M}}_{} = \widehat{\mathcal{M}}_{1}^{*} \cup \widehat{\mathcal{M}}_{2}
= \left\{1 \leq l \leq s_n: |\widehat{\beta}^M_{l}| \geq\right.$ $\left.\gamma_{1,n} \right\}$ $\cup \left\{1 \leq l \leq s_n: \right.$  $\left. \left\| \widehat{{\bm{C}}}_{l}^M \right\|_{op} \geq \gamma_{2,n} \right\}$. Let  $\gamma_{1,n} = \alpha D_1   n^{-\kappa}$ and $\gamma_{2,n} = \alpha D_1   (pq)^{1/2} n^{-\kappa}$  with $0 < \alpha < 1$, we have 
\begin{eqnarray*}
&& P( \widehat{\mathcal{M}}_{} \subset \mathcal{M}^0_1 \cup \mathcal{M}^0_2) \\
&\geq& P\left[   \cap_{1 \leq l \leq s_n} \left\{|\widehat{{\beta}_{l}^{M}} - \dot{\beta}_{l}^M|  \leq \frac{\gamma_{1,n}}{2}\right\} \cap_{1 \leq l \leq s_n}  \left\{  ||\widehat{{\bm{C}}}_{l}^M - \dot{\bm{C}}_{l}^M||_{op} \leq  \frac{\gamma_{2,n}}{2} \right\}\right]\\
&=& 1 - P \left[ \cup_{1 \leq l \leq s_n}  \{|\widehat{{\beta}_{l}^{M}} - \dot{\beta}_{l}^M|  \geq \frac{\gamma_{1,n}}{2}\} \cup_{1 \leq l \leq s_n}  \{||\widehat{{\bm{C}}}_{l}^M - \dot{\bm{C}}_{l}^M||_{op} \geq  \frac{\gamma_{2,n}}{2} \}   \right]\\
&\geq& 1 - \sum_{1 \leq l \leq s_n}  \left\{ P\left(  |\widehat{{\beta}_{l}^{M}} - \dot{\beta}_{l}^M|  \geq \frac{\gamma_{1,n}}{2} \right) + P\left(  ||\widehat{{\bm{C}}}_{l}^M - \dot{\bm{C}}_{l}^M||_{op} \geq  \frac{\gamma_{2,n}}{2}  \right)   \right\} \\
&\geq& 1 -   \sum_{1 \leq l \leq s_n}  P\left(  |\widehat{{\beta}_{l}^{M}} - \dot{\beta}_{l}^M|  \geq \frac{\gamma_{1,n}}{2} \right)  -   \sum_{1 \leq l \leq s_n}   P\left(  ||\widehat{{\bm{C}}}_{l}^M - \dot{\bm{C}}_{l}^M||_{F} \geq  \frac{\gamma_{2,n}}{2}  \right)   \\
&\geq& 1 - \sum_{1 \leq l \leq s_n}  P\left(  |\widehat{{\beta}_{l}^{M}} - \dot{\beta}_{l}^M|  \geq \alpha D_1 n^{-\kappa} /2 \right) \\
&& -   \sum_{1 \leq l \leq s_n} \sum_{j,k}  P\left(  | \widehat{C}_{l,jk}^M - \dot{C}_{l,jk}^M | \geq  \alpha D_1 n^{-\kappa}/2  \right)\\
&\geq& 1 - 2 s_n\left\{ (s_1 + s_2) \exp \left[ - \frac{ \alpha^2 D_1^2  (4b)^{-2} (s_1 + s_2 + 2)^{-2} n^{1-2\kappa} }{2\left\{2 D_2^{-2} e^{D_2 \sigma_x} D_3 + D_2^{-1} (4b)^{-1} (s_1 + s_2 + 2)^{-1} \alpha D_1 n^{-\kappa} \right\}}\right] \right.
\\
&& \left. +  2 \exp \left[ - \frac{\alpha^2 D_1^2  2^{-2} (s_1 + s_2 + 2)^{-2} n^{1-2\kappa}}{2 \left\{2 D_2^{-2} D_3 + D_2^{-1} (s_1 + s_2 + 2)^{-1} \alpha D_1 n^{-\kappa} \right\}}\right] \right\}\\
&& -2 s_n pq \left\{ s_2 \exp \left[ - \frac{\alpha^2 D_1^2 2^{-2} b^{-2} (s_1 + 1)^{-2} n^{1-2\kappa}}{2 \{ 2 D_2^{-2} e^{D_2 \sigma_x} D_3 + D_2^{-1} b^{-1} (s_1 + 1)^{-1} 2^{-1}\alpha D_1 n^{-\kappa}\}}\right] \right.\\
&&  \left. + \exp \left[ - \frac{ \alpha^2 D_1^2 2^{-2} (s_2 + 1)^{-2} n^{1-2\kappa}}{2 \{ 2 D_2^{-2} D_3 + D_2^{-1} (s_2 + 1)^{-1} 2^{-1}\alpha D_1 n^{-\kappa}\}}\right]\right\} \\
&\geq& 1 - 2 s_n\left\{ (s_1 + s_2) \exp \left[ - \frac{ \alpha^2 D_1^2  (4b)^{-2} (s_1 + s_2 + 2)^{-2} n^{1-2\kappa} }{2\left\{2 D_2^{-2} e^{D_2 \sigma_x} D_3 + D_2^{-1} (4b)^{-1} (s_1 + s_2 + 2)^{-1} \alpha D_1 \right\}}\right] \right.
\\
&& \left. +  2 \exp \left[ - \frac{\alpha^2 D_1^2  2^{-2} (s_1 + s_2 + 2)^{-2} n^{1-2\kappa}}{2 \left\{2 D_2^{-2} D_3 + D_2^{-1} (s_1 + s_2 + 2)^{-1} \alpha D_1 \right\}}\right] \right\}\\
&& -2 s_n pq \left\{ s_2 \exp \left[ - \frac{\alpha^2 D_1^2 2^{-2} b^{-2} (s_1 + 1)^{-2} n^{1-2\kappa}}{2 \{ 2 D_2^{-2} e^{D_2 \sigma_x} D_3 + D_2^{-1} b^{-1} (s_1 + 1)^{-1} 2^{-1}\alpha D_1 \}}\right] \right.\\
&&  \left. + \exp \left[ - \frac{ \alpha^2 D_1^2 2^{-2} (s_2 + 1)^{-2} n^{1-2\kappa}}{2 \{ 2 D_2^{-2} D_3 + D_2^{-1} (s_2 + 1)^{-1} 2^{-1}\alpha D_1 \}}\right]\right\} \\
&=& 1 - 2 \exp(D_4 n^{\xi})\left\{ (s_1 + s_2) \exp \left[ - \frac{ \alpha^2 D_1^2  (4b)^{-2} (s_1 + s_2 + 2)^{-2} n^{1-2\kappa} }{2\left\{2 D_2^{-2} e^{D_2 \sigma_x} D_3 + D_2^{-1} (4b)^{-1} (s_1 + s_2 + 2)^{-1} \alpha D_1 \right\}}\right] \right.
\\
&& \left. +  2 \exp \left[ - \frac{\alpha^2 D_1^2  2^{-2} (s_1 + s_2 + 2)^{-2} n^{1-2\kappa}}{2 \left\{2 D_2^{-2} D_3 + D_2^{-1} (s_1 + s_2 + 2)^{-1} \alpha D_1 \right\}}\right] \right\}\\
&& -2 pq \exp(D_4 n^{\xi}) \left\{ s_2 \exp \left[ - \frac{\alpha^2 D_1^2 2^{-2} b^{-2} (s_1 + 1)^{-2} n^{1-2\kappa}}{2 \{ 2 D_2^{-2} e^{D_2 \sigma_x} D_3 + D_2^{-1} b^{-1} (s_1 + 1)^{-1} 2^{-1}\alpha D_1 \}}\right] \right.\\
&&  \left. + \exp \left[ - \frac{ \alpha^2 D_1^2 2^{-2} (s_2 + 1)^{-2} n^{1-2\kappa}}{2 \{ 2 D_2^{-2} D_3 + D_2^{-1} (s_2 + 1)^{-1} 2^{-1}\alpha D_1 \}}\right]\right\} \\
\end{eqnarray*}
By Assumptions (A3) and (A5), we have 
\[
P( \widehat{\mathcal{M}}_{} \subset \mathcal{M}^0_1 \cup \mathcal{M}^0_2) \geq 1 - d_2 \exp( - d_3 n^{1-2 \kappa}), 
\]
for some constants $d_2$ and $d_3 > 0$. Therefore, we have $P( \widehat{\mathcal{M}}_{} \subset \mathcal{M}^0) \to 1$ as $n \to \infty$.\\

In step 2, we will show that $|\mathcal{M}^0| = O(n^{2 \kappa+\tau})$.
As $ \left| \mathcal{M}^0 \right| = \left| \mathcal{M}^0_1 \cup \mathcal{M}^0_2 \right| \leq\left| \mathcal{M}^0_1 \right| +  \left| \mathcal{M}^0_2 \right| $, we only need to show that both  $\left| \mathcal{M}^0_1 \right|$ and $\left| \mathcal{M}^0_2 \right|$ are $O(n^{2 \kappa+\tau})$. 

Define $\mathcal{M}^1_1 = \left\{ 1 \leq l \leq s_n:  \left| \dot{\beta}_{l}^M \right|^2 \geq \gamma_{1,n}^2/4 \right\}$ and $\mathcal{M}^0_1 \subset \mathcal{M}^1_1$. By the definition of $\mathcal{M}^1_1$, we have
\begin{eqnarray*}
\left| \mathcal{M}^1_1 \right| \gamma_{1,n}^2/4 &\leq& \sum_{l=1}^{s_n}  \left| \dot{\beta}_{l}^M \right |^2 
=  \sum_{l=1}^{s_n}  \left( E \left[x_{il} Y_{i}\right] \right)^2 = \left\| E\left[ \bm{x}_i * Y_i \right]\right\|^2.
\end{eqnarray*}
Define $\dot{\bm\beta}^{*} = (\dot{\beta}_{l}^{*},\ldots, \dot{\beta}_{s_n}^{*})^{\T}$ and $\dot{\bm{c}} = (\langle \dot{\bm{C}}_{l}, \dot{\bm{B}}\rangle,\ldots, \langle \dot{\bm{C}}_{s_n}, \dot{\bm{B}}\rangle )^{\T}$, we can write 
\begin{eqnarray*}
Y_i &=& \bm{x}_i^{\T} \left( \dot{\bm\beta}^{*} + \dot{\bm{c}}\right) + \langle \bm{E}_i,\dot{\bm{B}}\rangle + \epsilon_i. \\
\end{eqnarray*}
Multiplying $\bm{x}_i$ on both sides and taking expectations yield $E \left[\bm{x}_{i} * Y_{i}\right] = \Sigma_x \left( \dot{\bm\beta}^{*} + \dot{\bm{c}} \right)$. Therefore, we have 
\begin{eqnarray*}
\left| \mathcal{M}^2_1 \right| \gamma_{2,n}^2/4 &\leq& \left\| \Sigma_x \left( \dot{\bm\beta}^{*} + \dot{\bm{c}} \right) \right\|^2
\leq \lambda_{max}(\Sigma_x)  \left( \dot{\bm\beta}^{*} + \dot{\bm{c}} \right)^{\T}  \left( \dot{\bm\beta}^{*} + \dot{\bm{c}} \right) 
 \leq  4 b^2 \lambda_{max}(\Sigma_x).
\end{eqnarray*}
By Assumption (A4), we have $\left| \mathcal{M}^1_1 \right| \leq 4 b^2  \lambda_{max}(\Sigma_x) \gamma_{1,n}^{-2} = O(n^{2\kappa + \tau})$. 
This implies that $\left| \mathcal{M}^0_{1}\right| \leq \left| \mathcal{M}^1_{1}\right| \leq  O(n^{2\kappa + \tau})$. 

Define $\mathcal{M}^1_2 = \left\{ 1 \leq l \leq s_n:  \left\| \dot{\bm{C}}_{l}^M \right\|_{F}^2 \geq \gamma_{2,n}^2/4 \right\}$. As $\left\| \dot{\bm{C}}_{l}^M \right\|_{op} \leq \left\| \dot{\bm{C}}_{l}^M \right\|_{F}$, we have $\mathcal{M}^0_2 \subset \mathcal{M}^1_1$. By the definition of $\mathcal{M}^1_2$, we have
\begin{eqnarray*}
\left| \mathcal{M}^1_2 \right| \gamma_{2,n}^2/4 &\leq& \sum_{l=1}^{s_n}  \left\| \dot{\bm{C}}_{l}^M \right\|^2_{F} \\
&=& \sum_{j,k} \sum_{l=1}^{s_n} (\dot{C}_{l,jk}^M)^2 =  \sum_{j,k} \sum_{l=1}^{s_n}  \left( E \left[x_{il} Z_{i,jk}\right] \right)^2
= \sum_{j,k} \left\| E\left[ \bm{x}_i * Z_{i,jk} \right] \right\|^2.
\end{eqnarray*}
Define $\dot{\bm{C}}_{jk} = (\dot{C}_{1,jk},\ldots, \dot{C}_{s_n,jk})^{\T}$, we can write $Z_{i,jk} = \bm{x}_i^{\T} \dot{\bm{C}}_{jk} + E_{i,jk}$. Multiplying $\bm{x}_i$ on both sides and taking expectations yield $E \left[\bm{x}_{i} * Z_{i,jk}\right] = \Sigma_x \dot{\bm{C}}_{jk}$. 
\begin{eqnarray*}
\left| \mathcal{M}^1_2 \right| \gamma_{1,n}^2/4 &\leq& \sum_{j,k} \left\| \Sigma_x \dot{\bm{C}}_{jk} \right\|^2
\leq \lambda_{max}(\Sigma_x) \sum_{j,k} \dot{\bm{C}}_{jk}^{\T} \dot{\bm{C}}_{jk} 
 \leq  pq b^2  \lambda_{max}(\Sigma_x).
\end{eqnarray*}
By Assumption (A4), we have $\left| \mathcal{M}^1_2 \right| \leq 4  pq b^2 \lambda_{max}(\Sigma_x) \gamma_{2,n}^{-2} = O(n^{2\kappa + \tau})$.

Combining results from two steps above leads to $P\{ |\widehat{\mathcal{M}}_{}| = O(n^{2 \kappa + \tau}) \} \geq P(\widehat{\mathcal{M}}_{} \subset \mathcal{M}^0) \to 1$.
\end{proof}

\begin{proof}[Proof of Theorem \ref{thm1a}:]
We can write
\begin{eqnarray*}
&& P\left\{ \mathcal{M}_1 \subset \left( \widehat{\mathcal{M}}_{1}^{*} \cup \widehat{\mathcal{M}}_{2} \cup \widehat{\mathcal{M}}_{1}^{block,*} \cup \widehat{\mathcal{M}}_{2}^{block} \right) \right\} \\
&=& P\left\{ \cap_{l \in \mathcal{M}_1} \left( \widehat{\mathcal{M}}_{1}^{*} \cup \widehat{\mathcal{M}}_{2}  \cup \widehat{\mathcal{M}}_{1}^{block,*} \cup \widehat{\mathcal{M}}_{2}^{block}  \right)\right\}  \\
&=&  1 - P\left\{ \cup_{l \in \mathcal{M}_1} \left( \widehat{\mathcal{M}}_{1}^{*} \cup \widehat{\mathcal{M}}_{2} \cup \widehat{\mathcal{M}}_{1}^{block,*} \cup \widehat{\mathcal{M}}_{2}^{block}\right)^c \right\}  \\
&\geq& 1 - \sum_{l \in \mathcal{M}_1} P \left\{ \widehat{\mathcal{M}}_{1}^{*c} \cap \widehat{\mathcal{M}}_{2}^c \cap (\widehat{\mathcal{M}}_{1}^{block,*})^c \cap (\widehat{\mathcal{M}}_{2}^{block })^c\right\}  \\
    &=& 1 - \sum_{l \in \mathcal{M}_1} P\left(  |\widehat{{\beta}_{l}^{M}}|  \leq \gamma_{1,n}, \|\widehat{{\bm{C}}}_{l}^M\|_{op} \leq \gamma_{2,n}, \widehat{{\beta}_{l}^{block,M}} \leq \gamma_{3,n},\widehat{{C}_{l}^{block,M}} \leq  \gamma_{4,n}  \right) \\
 &\geq& 1 - \sum_{l \in \mathcal{M}_1 \cap \mathcal{M}_2} P\left( \|\widehat{{\bm{C}}}_{l}^M\|_{op} \leq \gamma_{2,n},\widehat{{C}_{l}^{block,M}} \leq  \gamma_{4,n}  \right)
 - \sum_{l \in \mathcal{M}_1 \cap \mathcal{M}_2^c} P\left(|\widehat{{\beta}_{l}^{M}}|  \leq \gamma_{1,n}, \widehat{{\beta}_{l}^{block,M}} \leq \gamma_{3,n}\right)\\
  &\geq& 1 - \sum_{l \in \mathcal{M}_1 \cap \mathcal{M}_2} P\left( \|\widehat{{\bm{C}}}_{l}^M\|_{op} \leq \gamma_{2,n}  \right) 
  - \sum_{l \in \mathcal{M}_1 \cap \mathcal{M}_2} P\left( \widehat{{C}_{l}^{block,M}} \leq  \gamma_{4,n}  \right) \\
  &&
 - \sum_{l \in \mathcal{M}_1 \cap \mathcal{M}_2^c} P\left(|\widehat{{\beta}_{l}^{M}}|  \leq \gamma_{1,n}\right) 
 - \sum_{l \in \mathcal{M}_1 \cap \mathcal{M}_2^c} P\left( \widehat{{\beta}_{l}^{block,M}} \leq \gamma_{3,n}\right)
\end{eqnarray*}

Firstly, recall that $\dot{\bm{C}}_{l}^M = cov (\sum_{l^\prime \in \mathcal{M}_2} x_{i l^\prime} * \dot{\bm{C}}_{l^\prime}, x_{il})$ i.e.
$\dot{C}_{l,jk}^M = cov (\sum_{l^\prime \in \mathcal{M}_2} x_{i l^\prime} \dot{C}_{l^\prime,jk},$ $x_{il}) = n^{-1} \sum_{i=1}^n E(x_{il} Z_{i,jk})$. For every $1 \leq j \leq p$, $1 \leq k \leq q$ and $1 \leq l \leq s_n$, we have
\[
\widehat{C}_{l,jk}^M -\dot{C}_{l,jk}^M = n^{-1} \sum_{i=1}^n \left[ x_{il} Z_{i,jk} - E(x_{il} Z_{i,jk})\right].
\]
It follows from Assumptions (A0), (A1), (A2) and Lemma 2 that for any $t>0$, one has
\begin{eqnarray*}
&&P\left(\left|\widehat{C}_{l,jk}^M -\dot{C}_{l,jk}^M \right| \geq t\right) 
= P\left( \left|  \sum_{i=1}^n [ x_{il} Z_{i,jk} - E(x_{il} Z_{i,jk}) \right| \geq nt \right) \\
&=& P \left( \left|\sum_{l^\prime \in \mathcal{M}_2} \sum_{i=1}^n \left[x_{il} x_{il^\prime} - E(x_{il} x_{il^\prime}) \right] \dot{C}_{l^\prime,jk} + \sum_{i=1}^n x_{il} E_{i,jk}\right| \geq nt\right)\\
&\leq& \sum_{l^\prime  \in \mathcal{M}_2} P\left(\left| \sum_{i=1}^n \left[x_{il} x_{il^\prime} - E(x_{il} x_{il^\prime}) \right] \right| \geq \frac{nt}{b(s_2 + 1)}\right) 
+ P\left( \left| \sum_{i=1}^n x_{il} E_{i,jk}\right| \geq \frac{nt}{s_2 + 1} \right) \\
&\leq& 2 s_2 \exp \left[ - \frac{n t^2 b^{-2} (s_1 + 1)^{-2}}{2 \{ 2 D_2^{-2} e^{D_2 \sigma_x} D_3 + D_2^{-1} b^{-1} (s_1 + 1)^{-1} t \} }\right]\\
&& + 2 \exp \left[ - \frac{n t^2 (s_2 + 1)^{-2}}{2 \{ 2 D_2^{-2} D_3 + D_2^{-1} (s_2 + 1)^{-1} t\} }\right].
 \end{eqnarray*}
Therefore, for every $l \in \mathcal{M}_2$, we have
 \begin{eqnarray*}
&& P\left( \| \widehat{{\bm{C}}}_{l}^M \|_{op} \leq \gamma_{2,n} \right) 
\leq P\left( \| \widehat{{\bm{C}}}_{l}^M - \dot{\bm{C}}_{l}^M\|_{op} \geq D_1(pq)^{1/2}n^{-\kappa} - \gamma_{2,n} \right) \\
&\leq&  P\left( \| \widehat{{\bm{C}}}_{l}^M - \dot{\bm{C}}_{l}^M\|_{F} \geq (pq)^{1/2} (1-\alpha) D_1 n^{-\kappa} \right) \\
&=&  P\left( \sum_{j,k} \left| \widehat{C}_{l,jk}^M - \dot{C}_{l,jk}^M \right|^2 \geq (pq)^{1/2} (1-\alpha) D_1 n^{-\kappa} \right) \\ 
&\leq& \sum_{j,k} P\left(  \left| \widehat{C}_{l,jk}^M - \dot{C}_{l,jk}^M \right|^2 \geq   \left\{(1-\alpha) D_1 n^{-\kappa} \right\}^2 \right) \\ 
&\leq& \sum_{j,k} P\left(  \left| \widehat{C}_{l,jk}^M - \dot{C}_{l,jk}^M \right| \geq    (1-\alpha) D_1 n^{-\kappa}  \right)\\
&\leq& 2pq \left( s_2 \exp \left[ - \frac{n^{1-2 \kappa} \left\{ (1-\alpha)D_1 b^{-1} (s_2 + 1)^{-1} \right\}^2 }{2 \{2 D_2^{-2} e^{D_2 \sigma_x} D_3 + D_2^{-1} b^{-1} (s_2 + 1)^{-1}  (1-\alpha) D_1 n^{-\kappa}\}}\right] \right. \\
&&\left. +  \exp \left[ - \frac{n^{1-2\kappa} \left\{ (1-\alpha) D_1  (s_2 + 1)^{-1} \right\}^2}{2 \{2 D_2^{-2} D_3 + D_2^{-1} (s_2 + 1)^{-1} (1-\alpha) D_1 n^{-\kappa}\}}\right]\right)\\
&\leq& 2pq \left( s_2 \exp \left[ - \frac{n^{1-2 \kappa} \left\{ (1-\alpha)D_1 b^{-1} (s_2 + 1)^{-1} \right\}^2 }{2 \{2 D_2^{-2} e^{D_2 \sigma_x} D_3 + D_2^{-1} b^{-1} (s_2 + 1)^{-1}  (1-\alpha) D_1 \}}\right] \right. \\
&&\left. +  \exp \left[ - \frac{n^{1-2\kappa} \left\{ (1-\alpha) D_1  (s_2 + 1)^{-1} \right\}^2}{2 \{2 D_2^{-2} D_3 + D_2^{-1} (s_2 + 1)^{-1} (1-\alpha) D_1\}}\right]\right)\\
 \end{eqnarray*}
 Let 
\begin{eqnarray*}
d_0 &=& \min \left[  \frac{ \left\{ (1-\alpha)D_1 b^{-1} (s_2 + 1)^{-1} \right\}^2 }{2 \{2 D_2^{-2} e^{D_2 \sigma_x} D_3 + D_2^{-1} b^{-1} (s_2 + 1)^{-1}  (1-\alpha) D_1 \}}, \right.\\
&& \left. \frac{ \left\{ (1-\alpha) D_1  (s_2 + 1)^{-1} \right\}^2}{2 \{2 D_2^{-2} D_3 + D_2^{-1} (s_2 + 1)^{-1} (1-\alpha) D_1\}} \right],
\end{eqnarray*} 
We have for every $l \in \mathcal{M}_2$, 
\begin{eqnarray}
P\left( \widehat{{C}_{l}^{block,M}} \leq  \gamma_{4,n}  \right) \leq P\left( \| \widehat{{\bm{C}}}_{l}^M \|_{op} \leq \gamma_{2,n} \right) &\leq& 2pq(s_2 + 1)\exp(-d_0 n^{1-2\kappa}), \label{thm1aeq1}
\end{eqnarray}

Let us consider $P\left( |\widehat{{\beta}_{l}^{M}}|  \leq \gamma_{1,n} \right)$. 
recall that, $\dot{\beta}_{l}^M = \dot{\beta}_{l}^{*M} + \langle\dot{\bm{C}}_{l}^{M}$, $\dot{\bm{B}} \rangle $, $\dot{\beta}_{l}^{*M} = cov (\sum_{l^\prime \in \mathcal{M}_1} $ $x_{i l^\prime}  \dot{\beta}_{l^\prime}, x_{il})$ and 
$\dot{\beta}_{l}^M =n^{-1} \sum_{i=1}^n E(x_{il} Y_{i})$. For every $1 \leq l \leq s_n$, we have
\[
\widehat{\beta}_{l}^M -\dot{\beta}_{l}^M = n^{-1} \sum_{i=1}^n \left\{ x_{il} Y_{i} - E(x_{il} Y_{i})\right\}.
\]
It follows from Assumptions (A0), (A1), (A2) and Lemma 2 that for any $t>0$, we have
\begin{eqnarray*}
&&P\left(\left|\widehat{\beta}_{l}^M -\dot{\beta}_{l}^M  \right| \geq t\right) 
= P\left[ \left|\sum_{i=1}^n \left\{ x_{il} Y_{i} - E(x_{il} Y_{i})\right\} \right| \geq nt \right] \\
&=& P \left[ \left|\sum_{l^\prime \in \mathcal{M}_1} \sum_{i=1}^n \left\{ x_{il} x_{il^\prime} - E\left( x_{il} x_{il^\prime} \right) \right\} \dot{\beta}_{l^\prime}^{*M} + 
\sum_{l^\prime \in \mathcal{M}_2} \sum_{i=1}^n \left\{ x_{il} x_{il^\prime} - E\left( x_{il} x_{il^\prime} \right) \right\}  \langle \dot{\bm{C}}_{l^\prime}^{M}, \dot{\bm{B}}\rangle + \right. \right.\\
&+&\left. \left. \sum_{i=1}^n \left\{ x_{il} \langle \bm{E}_{i}, \dot{\bm{B}} \rangle - E( x_{il})E\left(  \langle \bm{E}_{i}, \dot{\bm{B}} \rangle  \right) \right\}  + \sum_{i=1}^n x_{il} \epsilon_{i}\right| \geq n t\right]\\ 
&\leq& P \left[   \sum_{l^{\prime}  \in \mathcal{M}_1}  \left|\sum_{i=1}^n \left\{ x_{il} x_{il^\prime} - E\left( x_{il} x_{il^\prime} \right)  \right\} \right|  b +
   \sum_{l^{\prime}  \in \mathcal{M}_2}  \left|\sum_{i=1}^n \left\{ x_{il} x_{il^\prime} - E\left( x_{il} x_{il^\prime} \right)  \right\} \right|  b \right. \\
&& + \left.  \left| \sum_{i=1}^n x_{il} \langle \bm{E}_i, \dot{\bm{B}}\rangle \right| 
+  \left| \sum_{i=1}^n x_{il} \epsilon_{i} \right| \geq nt \right]\\
&\leq& \sum_{l^\prime  \in \mathcal{M}_1} P\left[\left| \sum_{i=1}^n \left\{  x_{il} x_{il^\prime} - E\left(  x_{il} x_{il^\prime} \right)  \right\} \right| \geq \frac{nt}{b(s_1 + s_2 + 2)}\right] \\
&& + \sum_{l^\prime  \in \mathcal{M}_2} P\left[\left| \sum_{i=1}^n \left\{  x_{il} x_{il^\prime} - E\left(  x_{il} x_{il^\prime} \right)  \right\} \right| \geq \frac{nt}{b(s_1 + s_2 + 2)}\right]  \\
&& + P\left( \left| \sum_{i=1}^n x_{il} \langle \bm{E}_i, \dot{\bm{B}}\rangle \right|  \geq \frac{nt}{s_1 + s_2 + 2} \right)  + P\left( \left| \sum_{i=1}^n x_{il} \epsilon_{i} \right| \geq \frac{nt}{s_1 + s_2 + 2} \right) \\
&\leq& 2 (s_1 + s_2) \exp \left[ - \frac{n t^2 (2b)^{-2} (s_1 + s_2 + 2)^{-2}}{2\left\{2 D_2^{-2} e^{D_2 \sigma_x} D_3 + D_2^{-1} (2b)^{-1} (s_1 + s_2 + 2)^{-1} t\right\}}\right]
\\
&& + 4 \exp \left[ - \frac{n t^2 (s_1 + s_2 + 2)^{-2}}{2 \left\{2 D_2^{-2} D_3 + D_2^{-1} (s_1 + s_2 + 2)^{-1} t\right\}}\right].
 \end{eqnarray*}

For $l \in \mathcal{M}_1 \cap \mathcal{M}^c_2$,  we have $\langle \dot{\bm{C}}^M_{l}, \dot{\bm{B}} \rangle = 0$, under Assumption (A1) and previous deduction, we have 
\begin{eqnarray*}
&& P\left( | \widehat{\beta}^M_{l} | \leq \gamma_{1,n} \right) 
=  P\left(  - | \widehat{\beta}^M_{l} | \geq - \gamma_{1,n} \right)
\leq P\left(  | \dot{\beta}_{l}^{*M} | - | \widehat{\beta}^M_{l} | \geq D_1   n^{-\kappa} - \gamma_{1,n} \right)  \\
&=&  P\left(  | \dot{\beta}_{l}^{*M}| - |\langle \dot{\bm{C}}^M_{l}, \dot{\bm{B}} \rangle  | - | \widehat{\beta}^M_{l} | \geq (1-\alpha)D_1   n^{-\kappa} \right)\\ 
&\leq&  P\left(  | \dot{\beta}_{l}^M  | - | \widehat{\beta}^M_{l} | \geq (1-\alpha)D_1   n^{-\kappa} \right)
\\
&\leq&  
P\left(  |  \dot{\beta}_{l}^M - \widehat{\beta}^M_{l}  | \geq (1-\alpha)D_1   n^{-\kappa} \right)\\
&\leq& 2 (s_1 + s_2) \exp \left[ - \frac{n^{1-2\kappa} (1-\alpha)^2 D_1^2  (2b)^{-2} (s_1 + s_2 + 2)^{-2}}{2\left\{2 D_2^{-2} e^{D_2 \sigma_x} D_3 + D_2^{-1} (2b)^{-1} (s_1 + s_2 + 2)^{-1} (1-\alpha)D_1   n^{-\kappa} \right\}}\right]
\\
&& + 4 \exp \left[ - \frac{n^{1-2\kappa} (1-\alpha)^2 D_1^2 (s_1 + s_2 + 2)^{-2}}{2 \left\{2 D_2^{-2} D_3 + D_2^{-1} (s_1 + s_2 + 2)^{-1} (1-\alpha) D_1 pq^{1/2} n^{-\kappa} \right\}}\right]\\
&\leq& 2 (s_1 + s_2) \exp \left[ - \frac{n^{1-2\kappa} (1-\alpha)^2 D_1^2  (2b)^{-2} (s_1 + s_2 + 2)^{-2}}{2\left\{2 D_2^{-2} e^{D_2 \sigma_x} D_3 + D_2^{-1} (2b)^{-1} (s_1 + s_2 + 2)^{-1} (1-\alpha)D_1   \right\}}\right]
\\
&& + 4 \exp \left[ - \frac{n^{1-2\kappa} (1-\alpha)^2 D_1^2  (s_1 + s_2 + 2)^{-2}}{2 \left\{2 D_2^{-2} D_3 + D_2^{-1} (s_1 + s_2 + 2)^{-1} (1-\alpha) D_1  \right\}}\right]\\
 \end{eqnarray*}
Let
\begin{eqnarray*}
d_1 &=& \min \left[ \frac{(1-\alpha)^2 D_1^2 (2b)^{-2} (s_1 + s_2 + 2)^{-2}}{2\left\{2 D_2^{-2} e^{D_2 \sigma_x} D_3 + D_2^{-1} (2b)^{-1} (s_1 + s_2 + 2)^{-1} (1-\alpha)D_1  \right\}}, \right.\\
&&\left. \frac{(1-\alpha)^2 D_1^2 (s_1 + s_2 + 2)^{-2}}{2 \left\{2 D_2^{-2} D_3 + D_2^{-1} (s_1 + s_2 + 2)^{-1} (1-\alpha) D_1  \right\}} \right],
\end{eqnarray*}
We have for each $l \in \mathcal{M}_1 \cap \mathcal{M}_2^c$, 
\begin{eqnarray}
P\left( \widehat{{\beta}_{l}^{block,M}} \leq \gamma_{3,n}\right) \leq P\left( |\widehat{{\beta}_{l}^{M}}|  \leq \gamma_{1,n} \right) \leq 2(s_1 +s_2 + 2) \exp \left( -d_1 n^{1-2\kappa}   \right) . \label{thm1aeq4}
\end{eqnarray}
In sum, by Assumption (A5), and \eqref{thm1aeq1} and \eqref{thm1aeq4}, we have
\begin{eqnarray*}
&& P\left\{ \mathcal{M}_1 \subset \left( \widehat{\mathcal{M}}_{1}^{*} \cup \widehat{\mathcal{M}}_{2} \cup \widehat{\mathcal{M}}_{1}^{block,*} \cup \widehat{\mathcal{M}}_{2}^{block} \right) \right\} \\
&\geq& 1 - 2\sum_{l \in \mathcal{M}_1 \cap \mathcal{M}_2} P\left( \|\widehat{{\bm{C}}}_{l}^M\|_{op} \leq \gamma_{2,n} \right)
 - 2\sum_{l \in \mathcal{M}_1 \cap \mathcal{M}_2^c} P\left(|\widehat{{\beta}_{l}^{M}}|  \leq \gamma_{1,n} \right)\\
&\geq& 1 - 4pqs_2(s_2 + 1)\exp(-d_0 n^{1-2\kappa}) -  4s_1 (s_1 +s_2 + 2) \exp \left( -d_1 n^{1-2\kappa}   \right)\\
&\geq& 1 - d_0^\prime pq \exp \left( -d_1^\prime n^{1-2\kappa} \right) \to  1,  \quad \textrm{ as } n \to \infty,
\end{eqnarray*}
for some positive constants $d_0^\prime$ and $d_1^\prime$. 
Therefore, $P(\mathcal{M}_1  \subset \widehat{\mathcal{M}}_{}) \to  1$ as $n \to \infty$. 
\end{proof}

\begin{proof}[Proof of Theorem \ref{thm3}:]

Denote $\bm\theta^{\widehat{\mathcal{M}}} = (\bm\beta^{\widehat{\mathcal{M}}, \T},\mathrm{vec}(\bm{B})^{\T})^{\T}$, where $\bm\beta^{\widehat{\mathcal{M}}} = \{ \beta_l \}_{l \in \widehat{\mathcal{M}}} \in \mathbb{R}^{|\widehat{\mathcal{M}}|}$ and $\bm{B} \in \mathbb{R}^{p \times q}$. 
In addition, for a given pair $\bm\lambda = (\lambda_1,\lambda_2)$, we let $\dot{\bm\theta}^{\widehat{\mathcal{M}}} = (\dot{\bm\beta}^{\widehat{\mathcal{M}}, \T},\mathrm{vec}(\dot{\bm{B}})^{\T})^{\T}$ be the true value for $\bm\theta^{\widehat{\mathcal{M}}}$ with $\dot{\bm\beta}^{\widehat{\mathcal{M}}}$ and $\dot{\bm{B}}$ being true values for $\bm\beta^{\widehat{\mathcal{M}}}$ and $\bm{B}$ respectively, and $\widehat{\bm\theta}_{\bm\lambda}^{\widehat{\mathcal{M}}} = (\widehat{\bm\beta}^{\widehat{\mathcal{M}}, \T},\mathrm{vec}(\widehat{\bm{B}})^{\T})^{\T}$ be the proposed estimator for 
$\bm\theta^{\widehat{\mathcal{M}}}$ with $\widehat{\bm\beta}^{\widehat{\mathcal{M}}}$ and $\widehat{\bm{B}}$ being the estimators for $\bm\beta^{\widehat{\mathcal{M}}}$ and $\bm{B}$ respectively. 
Furthermore, we let $r = \mathrm{rank}(\dot{\bm{B}})$, the true rank of matrix $\dot{\bm{B}} \in \mathbb{R}^{p \times q}$. Let us consider the class of matrices $ \Theta$ that have rank $r \leq \min\left\{ p,q \right\}$.  For any given matrix $\Theta$, we let $\mathrm{row}(\Theta) \subset \mathbb{R}^p$ and $\mathrm{col}(\Theta) \subset \mathbb{R}^q$ denote its row and column space, respectively. Let $U$ and $V$ be a given pair of $r$-dimensional subspace $U \subset \mathbb{R}^{p}$ and  $V \subset \mathbb{R}^{q}$, respectively.

For a given $\bm\theta^{\widehat{\mathcal{M}}}$ and pair $(U,V)$, we define the subspace 
$\Omega_1(\widehat{\mathcal{M}})$, $\Omega_2(U,V)$, $\overline{\Omega}_1(\widehat{\mathcal{M}})$, $\overline{\Omega}_2(U,V)$,  $\overline{\Omega}^{\perp}_1(\widehat{\mathcal{M}})$ and $\overline{\Omega}^{\perp}_2(U,V)$ as followed:
\begin{eqnarray*}
{\Omega}_1(\widehat{\mathcal{M}}) &=&\overline{\Omega}_1(\widehat{\mathcal{M}}) :=  \left\{ \bm\beta^{\widehat{\mathcal{M}}} = \{ \beta_l \}_{l \in \widehat{\mathcal{M}}}\in \mathbb{R}^{|\widehat{\mathcal{M}}|} | \beta_l = 0  \textrm{ for all } l \not\in \mathcal{M}_1 \right\},\\
\overline{\Omega}_1^{\perp}(\widehat{\mathcal{M}}) &:=& \left\{ \bm\beta^{\widehat{\mathcal{M}}} = \{ \beta_l \}_{l \in \widehat{\mathcal{M}}}\in \mathbb{R}^{|\widehat{\mathcal{M}}|} | \beta_l = 0  \textrm{ for all } l \in \mathcal{M}_1 \right\},\\
{\Omega}_2(U,V) &:=& \left\{ \Theta \in \mathbb{R}^{p \times q} |  \mathrm{row}(\Theta) \subset V, \textrm{ and } \mathrm{col}(\Theta) \subset U \right\},\\
\overline{\Omega}_2(U,V) &:=& \left\{ \Theta  \in \mathbb{R}^{p \times q} |   \mathrm{row}(\Theta) \subset V, \textrm{ or } \mathrm{col}(\Theta) \subset U \right\},\\
\overline{\Omega}_2^{\perp}(U,V) &:=& \left\{ \Theta  \in \mathbb{R}^{p \times q} |   \mathrm{row}(\Theta) \subset V^{\perp}, \textrm{ and } \mathrm{col}(\Theta) \subset U^{\perp} \right\},
\end{eqnarray*}
where $\overline{\Omega}_1^{\perp}(\widehat{\mathcal{M}})$ and $\overline{\Omega}_2^{\perp}(U,V)$ are the orthogonal complements for $\overline{\Omega}_1(\widehat{\mathcal{M}})$ and $\overline{\Omega}_2(U,V)$ respectively. 
For simplicity, we will use
${\Omega}_1$, $\overline{\Omega}_1$ and 
$\overline{\Omega}_1^{\perp}$ to denote
${\Omega}_1(\widehat{\mathcal{M}})$, $\overline{\Omega}_1(\widehat{\mathcal{M}})$ and 
$\overline{\Omega}_1^{\perp}(\widehat{\mathcal{M}})$, respectively; and use ${\Omega}_2$, $\overline{\Omega}_2$ and 
$\overline{\Omega}_2^{\perp}$ to denote
${\Omega}_2(U,V)$, $\overline{\Omega}_2(U,V)$ and $\overline{\Omega}_2^{\perp}(U,V)$.
It is easy to see that both $P_1$ and $P_2$ satisfy the condition that they are decomposable with respect to the subspace pair $(\Omega_1, \overline{\Omega}_1^{\perp})$ and $(\Omega_2, \overline{\Omega}_2^{\perp})$, respectively. 
Therefore the regularizer terms $P_1$ and $P_2$ satisfies condition  (G1) of  \cite{negahban2012unified}.

Here we define the function $F : \mathbb{R}^{|\widehat{\mathcal{M}}|+pq} \to \mathbb{R}$ by
\begin{eqnarray*}
F(\bm{\Delta}) &:=& l(\dot{\bm\theta}^{\widehat{\mathcal{M}}} + \bm{\Delta}) - l(\dot{\bm\theta}^{\widehat{\mathcal{M}}}) + \lambda_1 \left\{ P_1(\dot{\bm\beta}^{\widehat{\mathcal{M}}} + \bm{\Delta}_1) - P_1(\dot{\bm\beta}^{\widehat{\mathcal{M}}}) \right\} \\
&& + \lambda_2 \left\{ P_2( \dot{\bm{B}} + \bm{\Delta}_2) - P_2(\dot{\bm{B}}) \right\},
\end{eqnarray*}
where, $\bm\Delta = \{ \bm\Delta_1^{\T},\mathrm{vec}(\bm\Delta_2)^{\T} \}^{\T} \in \mathbb{R}^{|\widehat{\mathcal{M}}|+pq}$ with $\bm\Delta_1 \in \mathbb{R}^{|\widehat{\mathcal{M}}|}$ and $\bm\Delta_2 \in \mathbb{R}^{p \times q}$. 
Next, We will derive a lower bound for $F(\bm{\Delta})$. 

Before we formally prove the result, we first introduce the concept of subspace compatibility constant. 
For a subspace $\Omega$, the subspace compatibility constant with respect to the pair $(P,\| \cdot\|)$ is given by
\[ 
\psi(\Omega)  := \sup_{u \in \Omega \backslash \{ 0\}} \frac{P(u)}{\| u\|}.
\] 
Therefore, we have
\begin{eqnarray*}
\psi_1(\overline{\Omega}_1) &=& \sup_{\bm\beta^{\widehat{\mathcal{M}}} \in \overline{\Omega}_1 \backslash \{ 0\}} \frac{P_1(\bm\beta^{\widehat{\mathcal{M}}})}{\| \bm\beta^{\widehat{\mathcal{M}}}\|} = \frac{\sum_{l \in \widehat{\mathcal{M}}} | \beta_l | }{\sqrt{ \sum_{l \in \widehat{\mathcal{M}}} {\beta_l^2}}} \leq  \frac{\sqrt{ \sum_{l \in \widehat{\mathcal{M}}} | \beta_l |^2   \sum_{l \in \widehat{M}} 1^2} }{\sqrt{ \sum_{l \in \widehat{\mathcal{M}}} {\beta_l^2}}} \leq \sqrt{|\widehat{\mathcal{M}}|};\\
\psi_2(\overline{\Omega}_2)  &=& \sup_{\bm{B} \in \overline{\Omega}_2 \backslash \{ 0\}} \frac{P_2(\bm{B})}{\| \bm{B}\|} = \frac{ \left\| \bm{B} \right\|_* }{\| \bm{B}\|_F} \leq \frac{ \sqrt{r} \left\| \bm{B} \right\|_F }{\| \bm{B}\|_F}  = \sqrt{r},
\end{eqnarray*}
where the last step of first inequality is obtained by applying Cauchy-Schwarz inequality.
Furthermore, we have
\begin{eqnarray*}
F(\bm{\Delta}) &=& l(\dot{\bm\theta}^{\widehat{\mathcal{M}}} + \bm{\Delta}) - l(\dot{\bm\theta}^{\widehat{\mathcal{M}}}) + \lambda_1 \left\{ P_1(\dot{\bm\beta}^{\widehat{\mathcal{M}}} + \bm{\Delta}_1) - P_1(\dot{\bm\beta}^{\widehat{\mathcal{M}}}) \right\}\\
&& + \lambda_2 \left\{ P_2(\dot{\bm{B}} + \bm{\Delta}_2) - P_2(\dot{\bm{B}}) \right\}\\
 &\geq& \langle \nabla l(\dot{\bm\theta}^{\widehat{\mathcal{M}}}), \bm{\Delta} \rangle + \iota \left\| \bm{\Delta} \right\|^2 + 
 \lambda_1 \left\{ P_1(\dot{\bm\beta}^{\widehat{\mathcal{M}}} + \bm{\Delta}_1) - P_1(\dot{\bm\beta}^{\widehat{\mathcal{M}}}) \right\}\\
&& + \lambda_2 \left\{ P_2(\dot{\bm{B}} + \bm{\Delta}_2) - P_2(\dot{\bm{B}}) \right\}\\ 
  &\geq& \langle \nabla l(\dot{\bm\theta}^{\widehat{\mathcal{M}}}), \bm{\Delta} \rangle + \iota \left\| \bm{\Delta} \right\|^2 + 
 \lambda_1 \left[ P_1( \bm{\Delta}_{1,\overline{\Omega}_1^{\perp}}) - P_1(\bm{\Delta}_{1,\overline{\Omega}_1}) - 2 P_1\{(\dot{\bm\beta}^{\widehat{\mathcal{M}}})_{\overline{\Omega}_1^{\perp}}\} \right]\\
 &&+ \lambda_2 \left[ P_2( \bm{\Delta}_{2,\overline{\Omega}_2^{\perp}}) - P_2(\bm{\Delta}_{2,\overline{\Omega}_2}) - 2 P_2\{(\dot{\bm{B}})_{\overline{\Omega}_2^{\perp}}\} \right],
\end{eqnarray*} 
where the first inequality follows from Assumption (A6) and the second inequality follows from Lemma 3 in \cite{negahban2012unified}. 
By the Cauchy-Schwarz inequality applied to $P_k$ and its dual $P_k^*$, $k=1,2$,  where $P_k^*$ is defined as the dual norm of $P_k$ such that $P_k^*(\bm\theta) = \sup_{P_k(\bm\eta) \leq 1} \langle \bm\theta, \bm\eta \rangle$, we have 
$| \langle l(\dot{\bm\theta}^{\widehat{\mathcal{M}}}), \bm{\Delta} \rangle  |
=  |\langle \nabla l(\dot{\bm\beta}^{\widehat{\mathcal{M}}}), \bm{\Delta}_1 \rangle|  + |\langle\nabla l(\dot{\bm{B}}), \bm{\Delta}_2 \rangle |
\leq P_1^* \left\{ \nabla l(\dot{\bm\beta}^{\widehat{\mathcal{M}}}) \right\} P_1\left( \bm{\Delta}_1  \right) + P_2^* \left\{ \nabla l(\dot{\bm{B}}) \right\} P_2\left( \bm{\Delta}_2 \right).$
If $\lambda_k \geq 2 P_k^*(\cdot)$ holds, where
 $ P_1^*(\cdot) =  \left\| \nabla l(\dot{\bm\beta}^{\widehat{\mathcal{M}}}) \right\|_{\infty} $ and 
$ P_2^*(\cdot) =  \left\| \nabla l(\dot{\bm{B}}) \right\|_{op}$ \citep{negahban2012unified, kong2019l2rm}, 
one has
$| \langle \nabla l(\cdot), \bm{\Delta}_k \rangle  | \leq   \frac{1}{2} \lambda_k P_k(\bm{\Delta}_k) 
\leq \frac{1}{2} \lambda_k \left\{  P_k( \bm{\Delta}_{k,\overline{\Omega}_k^{\perp}}) + P_k(\bm{\Delta}_{k,\overline{\Omega}_k}) \right\}. $
Therefore, we have
\begin{eqnarray*}
F(\bm{\Delta}) 
&\geq& \langle \nabla l(\dot{\bm\theta}^{\widehat{\mathcal{M}}}), \bm{\Delta} \rangle + \iota \left\| \bm{\Delta} \right\|^2 + 
 \lambda_1 \left[ P_1( \bm{\Delta}_{1,\overline{\Omega}_1^{\perp}}) - P_1(\bm{\Delta}_{1,\overline{\Omega}_1}) - 2 P_1\{(\dot{\bm{\beta}}^{\widehat{\mathcal{M}}})_{\overline{\Omega}_1^{\perp}}\} \right]\\
 &&+ \lambda_2 \left[ P_2( \bm{\Delta}_{2,\overline{\Omega}_2^{\perp}}) - P_2(\bm{\Delta}_{2,\overline{\Omega}_2}) - 2 P_2\{\mathrm{vec}(\dot{\bm{B}})_{\overline{\Omega}_2^{\perp}}\} \right]\\
 &=& \langle \nabla l(\dot{\bm\beta}^{\widehat{\mathcal{M}}}), \bm{\Delta}_1 \rangle + \langle\nabla l(\dot{\bm{B}}), \bm{\Delta}_2 \rangle + \iota \left\| \bm{\Delta} \right\|^2 \\
 &&+ \lambda_1 \left[ P_1( \bm{\Delta}_{1,\overline{\Omega}_1^{\perp}}) - P_1(\bm{\Delta}_{1,\overline{\Omega}_1}) - 2 P_1\left\{(\dot{\bm{\beta}}^{\widehat{\mathcal{M}}})_{\overline{\Omega}_1^{\perp}}\right\} \right]\\
 &&+ \lambda_2 \left[ P_2( \bm{\Delta}_{2,\overline{\Omega}_2^{\perp}}) - P_2(\bm{\Delta}_{2,\overline{\Omega}_2}) - 2 P_2\left\{(\dot{\bm{B}})_{\overline{\Omega}_2^{\perp}}\right\} \right] \\
 &\geq& - \frac{\lambda_1}{2} \left\{ P_1( \bm{\Delta}_{1,\overline{\Omega}_1^{\perp}}) + P_1(\bm{\Delta}_{1,\overline{\Omega}_1}) \right\} 
 - \frac{\lambda_2}{2} \left\{ P_2( \bm{\Delta}_{2,\overline{\Omega}_2^{\perp}}) + P_2(\bm{\Delta}_{2,\overline{\Omega}_2}) \right\}  + \iota \left\| \bm{\Delta} \right\|^2 \\
&& + \lambda_1 \left[ P_1( \bm{\Delta}_{1,\overline{\Omega}_1^{\perp}}) - P_1(\bm{\Delta}_{1,\overline{\Omega}_1}) - 2 P_1\left\{(\dot{\bm{\beta}}^{\widehat{\mathcal{M}}})_{\overline{\Omega}_1^{\perp}}\right\} \right]\\
&&+ \lambda_2 \left[ P_2( \bm{\Delta}_{2,\overline{\Omega}_2^{\perp}}) - P_2(\bm{\Delta}_{2,\overline{\Omega}_2}) - 2 P_2\left\{(\dot{\bm{B}})_{\overline{\Omega}_2^{\perp}}\right\} \right] \\
&=&\iota \left\| \bm{\Delta} \right\|^2 + \lambda_1 \left[ \frac{1}{2} P_1( \bm{\Delta}_{1,\overline{\Omega}_1^{\perp}}) - \frac{3}{2}  P_1(\bm{\Delta}_{1,\overline{\Omega}_1}) - 2 P_1\left\{(\dot{\bm{\beta}}^{\widehat{\mathcal{M}}})_{\overline{\Omega}_1^{\perp}}\right\} \right] \\
&&+\lambda_2 \left[ \frac{1}{2} P_2( \bm{\Delta}_{2,\overline{\Omega}_2^{\perp}}) -\frac{3}{2}  P_2(\bm{\Delta}_{2,\overline{\Omega}_2}) - 2 P_2\left\{(\dot{\bm{B}})_{\overline{\Omega}_2^{\perp}}\right\} \right].\\
\end{eqnarray*} 
By the subspace compatibility, we have $ P_k(\bm{\Delta}_{k,\overline{\Omega}_k}) \leq \psi_k(\overline{\Omega}_k) \| \bm{\Delta}_{k,\overline{\Omega}_k}\|$, for $k=1,2$. Substituting them into the previous inequality, and noticing that $P_1\left\{(\dot{\bm{\beta}}^{\widehat{\mathcal{M}}})_{\overline{\Omega}_1^{\perp}}\right\} =  P_2\left\{(\dot{\bm{B}})_{\overline{\Omega}_2^{\perp}}\right\}$  $= 0 $, we obtain that 
\begin{eqnarray*}
F(\bm{\Delta}) 
&\geq& \iota \left\| \bm{\Delta} \right\|^2 -\sum_{k \in \{ 1,2\}}   \frac{3 \lambda_k}{2} \psi_k(\overline{\Omega}_k) \| \bm{\Delta}_{k,\overline{\Omega}_k}\|  \\
&\geq& \iota \left\| \bm{\Delta} \right\|^2 -\sum_{k \in \{ 1,2\}}   \frac{3 \lambda_k}{2} \psi_k(\overline{\Omega}_k) \| \bm{\Delta}_{k}\|  \\
&\geq& \iota \left\| \bm{\Delta} \right\|^2 - 3 \max_{k \in \{ 1,2\}}\{  \lambda_k \psi_k(\overline{\Omega}_k) \} \left\| \bm{\Delta} \right\|.
\end{eqnarray*}
The right hand side is a quadratic form of $\bm{\Delta}$. Therefore, as long as $ \left\| \bm{\Delta} \right\|^2 > \frac{9}{4 \iota^2}  \max^2_{k \in \{ 1,2\}}\{ $  $\lambda_k \psi_k(\overline{\Omega}_k) \} $, one has $F(\bm{\Delta}) > 0$. Therefore, by Lemma 4 in \cite{negahban2012unified}, 
we can establish that
\begin{eqnarray*}
\left\| \widehat{\bm\theta}^{\widehat{\mathcal{M}}}_{\bm\lambda} - \dot{\bm\theta}^{\widehat{\mathcal{M}}} \right\|_2^2 \leq C \max\left\{\lambda_1^2 |\widehat{\mathcal{M}}|,\lambda_2^2 r \right\} \iota^{-2}, 
\end{eqnarray*} 
for some constant $C > 0$. 
It is easy to see that there exists some constant $C_0 > 0$ such that
\begin{eqnarray*}
\left\| \widehat{\bm\theta}_{\bm\lambda} - \dot{\bm\theta} \right\|_2^2 \leq C_0 \max\left\{ C_1 n^{2 \kappa + \tau} \lambda_1^2 ,\lambda_2^2 r \right\} \iota^{-2}.
\end{eqnarray*} 

Now we calculate  $ P_1^*(\cdot) $ and $ P_2^*(\cdot) $. 
According to \cite{kong2019l2rm,negahban2012unified}, we have that
$ P_1^*(\cdot) =  \left\| \nabla l(\dot{\bm\beta}^{\widehat{\mathcal{M}}}) \right\|_{\infty} $ and 
$ P_2^*(\cdot) =  \left\| \nabla l(\dot{\bm{B}}) \right\|_{op}$, where 
$
\nabla l(\dot{\bm\beta}^{\widehat{\mathcal{M}}}) = - n^{-1}\sum_{i=1}^{n}  \epsilon_i X_i^{\widehat{\mathcal{M}}} $ and 
$\nabla l(\dot{\bm{B}}) = - n^{-1}\sum_{i=1}^{n} \epsilon_i *  \bm{Z}_i
$.

We first calculate $  \left\| \nabla l(\dot{\bm\beta}^{\widehat{\mathcal{M}}}) \right\|_{\infty} $. 
Denoting $\bm\epsilon = (\epsilon_1,\ldots,\epsilon_n)^{\T}$, $\bm X^{\widehat{\mathcal{M}}} = (X_1^{\widehat{\mathcal{M}},\T},\ldots,X_n^{\widehat{\mathcal{M}},\T})^{\T}$,
 where $X_i^{\widehat{\mathcal{M}}} = (x_{ij})^{\T}_{j \in \widehat{\mathcal{M}}} \in \mathbb{R}^{|\widehat{\mathcal{M}}|}$ for $i= 1,\ldots,n$. 
we let $\bm x_l^{\widehat{\mathcal{M}}}$, where $l = 1,\ldots, |\widehat{\mathcal{M}}|$ represent the $l$-th column of  $\bm X^{\widehat{\mathcal{M}}}$. 
Since $\bm X$ column normalized, i.e. $\| \bm{x}_l\|_2 / \sqrt{n} = 1$, for all $l \in 1, \ldots,s $, by Assumption (A2), there exists a constant $ \sigma_{0} > 0$ such that 
$P( \left| \langle \bm x_l^{\widehat{\mathcal{M}}}, \bm \epsilon \rangle / n \right| \geq t ) \leq 2 \exp \left( - \frac{n t^2}{2 \sigma_{0}^2} \right)$. 
Applying union bound, we have
$P \left( \left\| - n^{-1}\sum_{i=1}^{n}  \epsilon_i X_i^{\widehat{\mathcal{M}}} \right\|_\infty  \geq t \right) = 
P\left( \left\| \bm X^{\widehat{\mathcal{M}},\T}  \bm\epsilon / n \right\|_\infty  \geq t \right) 
= P\left( \mathrm{sup}_{l \in \widehat{\mathcal{M}}}
\left| \langle \bm x_l^{\widehat{\mathcal{M}}}, \bm \epsilon \rangle /n 
 \right| \geq t \right)
\leq 2 \exp \left(- \frac{n t^2}{2 \sigma_{0}^2} + \log |\widehat{\mathcal{M}}| \right) \leq
2 \exp \left( \right.$ $\left. - \frac{n t^2}{2 \sigma_{0}^2}   + C_1 (2 \kappa + \tau)\log n  \right)$.
By choosing $t^2 = 2  n^{-1} \sigma_{0}^2$  $\{ \log (\log n) + C_1 (2 \kappa + \tau)\log n \}$, we can see that when $\lambda_1 \geq 
 2 \sigma_{0} [2 n^{-1} \{ \log (\log n) +  C_1 (2 \kappa + \tau)\log n \}]^{1/2}$, 
then there exist a positive constant $c_1 > 0$ such that the choice of $\lambda_1$ holds with probability at least $1 - c_1  (\log n)^{-1}$. 

Secondly, we calculate $ \left\| \nabla l(\dot{\bm{B}}) \right\|_{op}$: 
\begin{eqnarray*}
&& \left\| \nabla l(\dot{\bm{B}}) \right\|_{op} = \left\|   - n^{-1} \sum_{i=1}^{n} \epsilon_i *  \bm{Z}_i \right\|_{op}
 =\left\| n^{-1} \sum_{i=1}^{n} \epsilon_i * \left( \sum_{l \in \mathcal{M}_2} X_{il}* \dot{\bm{C}}_{l} + \bm{E}_{i}\right) \right\|_{op}\\
 &=& \left\|  n^{-1} \sum_{i=1}^{n} \epsilon_i * \left( \sum_{l \in \mathcal{M}_2} X_{il}* \dot{\bm{C}}_{l} \right)+  n^{-1} \sum_{i=1}^{n} \epsilon_i * \bm{E}_{i} \right\|_{op}\\
& =&  \left\|  n^{-1} \sum_{l \in \mathcal{M}_2} \langle \bm x_l, \bm\epsilon \rangle * \dot{\bm{C}}_{l} + n^{-1} \sum_{i=1}^{n} \epsilon_i * \bm{E}_{i} \right\|_{op}\\
&\leq&  n^{-1} \sum_{l \in \mathcal{M}_2}   \left| \langle \bm x_l, \bm\epsilon \rangle \right| \left\|  \dot{\bm{C}}_{l} \right\|_{op} +   \left\|n^{-1} \sum_{i=1}^{n} \epsilon_i * \bm{E}_{i} \right\|_{op}.
\end{eqnarray*}
Under condition (A1),  
$
n^{-1} \sum_{l \in \mathcal{M}_2}   \left| \langle \bm x_l, \bm\epsilon \rangle \right| \left\|  \dot{\bm{C}}_{l} \right\|_{op} \leq  
n^{-1}  b \sum_{l \in \mathcal{M}_2}   \left| \langle \bm x_l, \bm\epsilon \rangle \right|$. 
Therefore, we have
\begin{eqnarray*}
&&P(n^{-1} b \sum_{l \in \mathcal{M}_2}   \left| \langle \bm x_l, \bm\epsilon \rangle \right| \geq t )
= P( \sum_{l \in \mathcal{M}_2}   \left| \langle \bm x_l, \bm\epsilon \rangle /n \right| \geq t/b )
\\
&\leq& P\left[ \cup_{l \in \mathcal{M}_2}  \left\{ \left| \langle \bm x_l, \bm\epsilon \rangle /n \right| \geq \frac{t}{b s_2} \right\} \right]
\leq \sum_{l \in \mathcal{M}_2}  P\left( \left| \langle \bm x_l, \bm\epsilon \rangle /n \right| \geq \frac{t}{b s_2}\right)\\
&\leq& \sum_{l \in \mathcal{M}_2}  P\left( \sup_{l \in \mathcal{M}_2} \left| \langle \bm x_l, \bm\epsilon \rangle /n \right| \geq \frac{t}{b s_2}\right)
\leq 2 \exp \left( - \frac{n t^2}{2 b^2 s_2^2 \sigma_{0}^2} + 2 \log s_2 \right).
\end{eqnarray*}
Therefore, by choosing 
$t= b s_2 \sigma_{0}$ $[ 2 n^{-1} \{ 3 \log s_2 + \log (\log n) \}]^{1/2}$, for any choice of $t_1 \geq t$, we guarantee that $t_1 \geq n^{-1} b \sum_{l \in \mathcal{M}_2}   \left| \langle \bm x_l, \bm\epsilon \rangle \right| $ therefore  $t_1 \geq n^{-1} \sum_{l \in \mathcal{M}_2}   \left| \langle \bm x_l, \bm\epsilon \rangle \right| \left\|  \dot{\bm{C}}_{l} \right\|_{op}$ is valid with probability at least $1 - c_2 / (s_2 \log n)$ for some positive constant $c_2$.

On the other hand, let $\bm{W}_i$ be a $p \times q$ random matrix with each entry i.i.d. standard normal. By Assumption (A8) and Lemma \ref{lem3}, conditioning on $\epsilon_i$ we have
\begin{eqnarray*}
P( \left\| \sum_{i=1}^{n} \epsilon_i * \bm{E}_{i} \right\|_{op} \geq t ) 
= P( \sup_{\| \bm A \|_* \leq 1 }\langle \bm A, \sum_{i=1}^{n} \epsilon_i * \bm{E}_{i} \rangle \geq t ) 
\leq  P( \sup_{\| \bm A \|_* \leq 1 }\langle \bm A, \sum_{i=1}^{n} \epsilon_i * \bm{W}_{i} \rangle \geq \frac{t}{C_U} ),
\end{eqnarray*}
since $\Sigma_e \preceq C_u^2 I_{pq \times pq}$. 

As $ \sup_{\| \bm A \|_* \leq 1 }\langle \bm A, \sum_{i=1}^{n} \epsilon_i * \bm{W}_{i} \rangle 
= \left\| \sum_{i=1}^{n} \epsilon_i * \bm{W}_{i}  \right\|_\infty$
conditioning on $\bm{W}_{i}$, each entry of the matrix 
$\sum_{i=1}^{n} \epsilon_i * \bm{W}_{i} $ 
is i.i.d. $N(0, \| \bm\epsilon \|_{op}^2)$. 
Since $\frac{\| \bm\epsilon \|_{op}^2}{\sigma_{\epsilon}^2}$ is a $\chi^2$ random variable with $n$ degrees of freedom, one has
\[
P(\frac{\| \bm\epsilon \|_{op}^2}{n \sigma_{\epsilon}^2} \geq 4) \leq \exp(-n) 
\]
using the tail bound of $\chi^2$ presented by the corollary of Lemma 1 from \cite{laurent2000adaptive}. Combing with the standard random matrix theory, we know that 
$\| n^{-1/2} \sum_{i=1}^{n} \epsilon_i * $ $\bm{W}_{i}  \|_{op} \leq 2 n^{-1/2} \sigma_{\epsilon} (p^{1/2} + q^{1/2})$
with probability at least $1 - c_3 \exp\{ -c_4 (p+q)\} - \exp(-n)$ where $c_3$ and $c_4$ are some positive constants. 
Combining with what we got from previous step, we have
$ \sum_{l \in \mathcal{M}_2}   \left| \langle \bm x_l, \bm\epsilon \rangle \right| \left\|  \dot{\bm{C}}_{l} \right\|_{op} + \left\| \sum_{i=1}^{n} \epsilon_i * \bm{E}_{i} \right\|_{op} \leq  b s_2 \sigma_{0}[ 2 n^{-1} \{ 3 \log s_2 + \log (\log n) \}]^{1/2} +  2 n^{-1/2} \sigma_{\epsilon} (p^{1/2} + q^{1/2})$
holds with probability at least 
$1- c_2 / (s_2 \log n) - c_3 \exp\{ -c_4 (p+q)\} - \exp(-n) $. Therefore, the choice of $\lambda_2 \geq 2 b s_2 \sigma_{0}[ 2 n^{-1} \{ 3 \log s_2 + \log (\log n) \}]^{1/2} +  4 n^{-1/2} \sigma_{\epsilon} (p^{1/2} + q^{1/2})$ holds with probability at least 
$1- c_2 / (s_2 \log n) - c_3 \exp\{ -c_4 (p+q)\} - \exp(-n) $. 

As a result, the event that both $\lambda_1$ and $\lambda_2$ satisfy the above inequalities holds with  probability at least 
$1-  c_1/ \log n - c_2 / (s_2 \log n) - c_3 \exp\{ -c_4 (p+q)\} - \exp(-n) $ for some postive constants $c_1$, $c_2$, $c_3$ and $c_4$.

In sum, there exists some positive constants $c_1, c_2, c_3, c_4$, with probability at least 
$1- c_1/ \log n - c_2 / (s_2 \log n) - c_3 \exp\{ -c_4 (p+q)\} $,  
one has
$$\left\| \widehat{\bm\theta}_{\bm\lambda} - \dot{\bm\theta} \right\|_2^2 \leq C_0 \max\left\{C_1 \lambda_1^2 n^{2 \kappa + \tau},\lambda_2^2 r \right\} \iota^{-2},$$
for some constants $C_0, C_1 > 0$. This completes the proof.

\end{proof}

\bibliographystyle{chicago}
\bibliography{refsg}

\end{document}